\pdfoutput=1
\documentclass[a4paper,11pt,twoside]{Thesis_Template/ThesisStyle}

\usepackage{amsmath,amssymb,dsfont,amsthm,mathtools}             % AMS Math
\usepackage{xifthen}
\usepackage{subfigure,rotating}
\usepackage[latin1]{inputenc}
\usepackage[T1]{fontenc}
 \usepackage[left=1.3in,right=0.9in,top=1in,bottom=1in,includefoot,includehead,headheight=13.6pt]{Thesis_Template/geometry} %{geometry}% 

\usepackage[nottoc, notlof, notlot]{tocbibind}
\usepackage{minitoc}
\usepackage{aecompl}

\usepackage[intoc,norefpage]{Thesis_Template/nomencl_remy} %%% je l'ai modifié pour avoir une section au lieu d'un chapitre
\renewcommand{\nomname}{List of abbreviations}
\makenomenclature

\usepackage{ifpdf}
\ifpdf
  \usepackage{graphicx}
  \usepackage[a4paper,pagebackref,hyperindex=true]{hyperref}
\else
  \usepackage{graphicx}
  \usepackage[a4paper,dvipdfm,pagebackref,hyperindex=true]{hyperref}
\fi

\graphicspath{{.}{images/}}

\renewcommand*{\backref}[1]{}
\renewcommand*{\backrefalt}[4]{%
\ifcase #1 %
(Not cited.)%
\or
(Cited on page~#2.)%
\else
(Cited on pages~#2.)%
\fi}

\usepackage{color}
\definecolor{linkcol}{rgb}{0,0,0.4} 
\definecolor{citecol}{rgb}{0.5,0,0} 

\hypersetup
{
bookmarksopen=true,
pdftitle="Non-linear Dependences in Finance",
pdfauthor="Remy CHICHEPORTICHE", 
pdfsubject="PhD Thesis: Non-linear Dependences in Finance", %subject of the document
pdftoolbar=false,   %toolbar hidden
pdfmenubar=true,    %menubar shown
pdfhighlight=/O,    %effect of clicking on a link
colorlinks=true,    %couleurs sur les liens hypertextes
pdfpagemode=None,   %aucun mode de page
pdfpagelayout=TwoPageRight, %ouverture en double page
pdffitwindow=true,  %pages ouvertes entierement dans toute la fenetre
linkcolor=linkcol,  %couleur des liens hypertextes internes
citecolor=citecol,  %couleur des liens pour les citations
urlcolor=linkcol    %couleur des liens pour les url
}

\theoremstyle{definition}
\newtheorem*{lemma}{Lemma}
\newtheorem*{theorem}{Theorem}
\newtheorem*{property}{Property}
\newtheorem*{definition}{Definition}

\def\e     {\mathrm{e}}
\def\const {\mathrm{const}}

\def\diag  {\operatorname{diag}}
\def\sign  {\operatorname{sign}}
\def\rank  {\operatorname{rank}}
\def\med   {\operatorname{med}}
\def\tr    {\operatorname{Tr}}

\def\argmin{\operatornamewithlimits{arg\,min}}
\def\ft    {\mathrm{FT}} %  \newcommand{\ft }[1]{\mathrm{FT}[#1]} %

\newcommand\fd  [1][d]{   f^{ \scriptscriptstyle (#1)}} %
\newcommand\rod [1][d]{\rho^{ \scriptscriptstyle (#1)}} %\newcommand{\rod }{\rho^{ \scriptscriptstyle ({d})}}
\newcommand\road[1][d]{\zeta^{\scriptscriptstyle (#1)}} %\newcommand{\road}{\zeta^{\scriptscriptstyle ({d})}}
\def\rob{\rod[1]}                                       %\newcommand{\rob }{\rho^{ \scriptscriptstyle (1)}}
\def\ros{\rod[\text{s}]}                                %\newcommand{\ros }{\rho^{ \scriptscriptstyle (\text{s})}}
\def\roa{\road[1]}                                      %\newcommand{\roa }{\zeta^{\scriptscriptstyle (1)}}
\def\roc{\road[2]}                                      %\newcommand{\roc }{\zeta^{\scriptscriptstyle (2)}}

\def\rhoP  {\rho^{ \scriptscriptstyle (\text{P})}} % Pearson's  rho %{\rod[\text{P}]}
\def\rhoS  {\rho^{ \scriptscriptstyle (\text{S})}} % Spearman's rho %{\rod[\text{P}]}
\def\tauK  {\tau^{ \scriptscriptstyle (\text{K})}} % Kendall's tau
\def\betaB {\beta^{\scriptscriptstyle (\text{B})}} % Blomqvist's beta
\def\rhoK  {\rho^{ \scriptscriptstyle (\text{K})}} % 
\def\rhoB  {\rho^{ \scriptscriptstyle (\text{B})}} % 

\def\tUU{\tau^{\scriptscriptstyle \text{UU}}}
\def\tLL{\tau^{\scriptscriptstyle \text{LL}}}
\def\tUL{\tau^{\scriptscriptstyle \text{UL}}}
\def\tLU{\tau^{\scriptscriptstyle \text{LU}}}

\renewcommand\ij{i\!j}
\renewcommand\imath{\mathrm{i}}

  \newcommand{\bra   }[1]{\langle #1 \vert}
  \newcommand{\ket   }[1]{\vert   #1 \rangle}
  \newcommand{\braket}[2]{\langle #1 \vert #2 \rangle}
  \newcommand{\vev   }[1]{\langle #1 \rangle}
  
  \newcommand{\esp   }[1]{\mathds{E}[#1]}
  \newcommand{\var   }[1]{\mathds{V}[#1]}
  \newcommand{\pr    }[1]{\mathds{P}[#1]}
  \newcommand{\1     }[1]{\mathds{1}_{\{#1\}}}
  \newcommand{\Esp   }[1]{\mathds{E}\!\left[#1\right]}
  \newcommand{\Var   }[1]{\mathds{V}\!\left[#1\right]}
\renewcommand{\Pr    }[1]{\mathds{P}\!\left[#1\right]}
  \newcommand{\Exp   }[1]{      \exp\!\left(#1\right)}
  \newcommand{\Erf  }[1]{\operatorname{erf}\!\left(#1\right)}
\renewcommand{\d     }[1]{\mathrm{d}#1} % pour les intégrales
  \newcommand{\pdf }[1][]{\mathcal{P}_{\negthickspace\ifthenelse{\isempty{#1}}{}{#1}}}
  \newcommand{\cdf }[1][]{\mathcal{P}_{\!\negthickspace{\scriptscriptstyle <}\ifthenelse{\isempty{#1}}{}{,#1}}}
  \newcommand{\qdf }[1][]{\mathcal{P}_{\!\negthickspace{\scriptscriptstyle <}\ifthenelse{\isempty{#1}}{}{,#1}}^{-1}} %{\mathcal{Q}_{#1}}
  \newcommand{\tcdf}[1][]{\mathcal{P}_{\!\negthickspace{\scriptscriptstyle >}\ifthenelse{\isempty{#1}}{}{,#1}}}
  \newcommand{\tqdf}[1][]{\mathcal{P}_{\!\negthickspace{\scriptscriptstyle >}\ifthenelse{\isempty{#1}}{}{,#1}}^{-1}} %{\mathcal{Q}_{#1}}
  \newcommand{\cop  }[2][]{C_{#1}#2}
  
  \newcommand{\copN }[2][]{\mathcal{C}_{#1}#2} %
  \newcommand{\CopN }[2][]{\mathcal{C}_{#1}\!\left(#2\right)} %

\newcommand{\ribracket}{] \mspace{-3 mu} ]}
\newcommand{\libracket}{[\mspace{-3 mu} [}

 \def \be  {\begin{equation}}
 \def \ee  {\end{equation}}

  \newcommand{\vect}[1]{\boldsymbol{#1}} % essayer aussi avec \mathbf
  \newcommand{\mat }[1]{\mathrm{#1}}
\renewcommand{\top}{\dagger}

\def\Wei{\beta}                      %\mat{X}
\def\ret{x}                          %r
\def\Ret{\mat{\MakeUppercase{\ret}}} %\mat{R}

\usepackage[square,comma]{Thesis_Template/natbib_remy}        % Put References at the end of each chapter
\usepackage{fancyhdr}                    % Fancy Header and Footer
\let\headruleORIG\headrule
\renewcommand{\headrule}{\color{black} \headruleORIG}

\usepackage{colortbl}
\arrayrulecolor{black}

\fancypagestyle{plain}{
  \fancyhead{}
  \fancyfoot{}
  
}

\makeatletter

\def\cleardoublepage{\clearpage\if@twoside \ifodd\c@page\else%
  \hbox{}%
  \thispagestyle{empty}%              % Empty header styles
  \newpage%
  \if@twocolumn\hbox{}\newpage\fi\fi\fi}

\makeatother
 
\begin{document}

\newgeometry{left=3.5cm,right=3cm,top=1.75cm,bottom=1.5cm} % from version 5 and later
%\thispagestyle{empty}
%\vspace{-3\baselineskip}
\begin{titlepage}
\begin{center}

\includegraphics[trim=-15 0 0 0,clip]{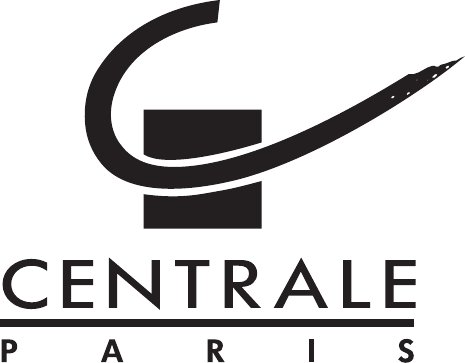} \\
\vspace*{.8cm}
\noindent {\Large \textbf{ECOLE CENTRALE\\ DES ARTS ET MANUFACTURES}} \\
%\vspace*{0.3cm}
%\noindent {\large <<~\textbf{ECOLE CENTRALE PARIS}~>>} \\
\vfill%\vspace*{0.5cm}
\noindent \fontsize{25}{30} \textbf{T H E S E} \\
\vfill
\noindent \normalsize pr\'esent\'ee par \\
\vspace*{0.2cm}
\noindent \Large R\'emy \textsc{Chicheportiche} \\
\vspace*{0.8cm}
\noindent \normalsize {pour l'obtention du} \\
\vspace*{0.2cm}
\noindent \Large \textsc{grade de Docteur} \\
\vfill
%\noindent \large of the University of Nice - Sophia Antipolis \\
\noindent \large
\begin{tabular}{rl}
      \textbf{Specialit\'e :}& {Math\'ematiques appliqu\'ees}	    \\
      \textbf{Laboratoire :} & {Math\'ematiques Appliqu\'ees aux Syst\`emes (MAS) --- EA~4037}	    \\
      \textbf{Sujet :}       & \textbf{\textsc{D\'ependances non lin\'eaires en finance}}
\end{tabular}

\end{center}
\vfill
\noindent \normalsize {soutenue le \textbf{27 juin 2013} devant un jury compos\'e de:} \\[.5cm]
%\noindent \normalsize {devant un jury compos\'e de:} \\[.5cm]
%\begin{center}
\noindent \normalsize 
\begin{tabular}{ll@{\quad -\quad }l}
      \textit{Directeurs :}     & Fr\'ed\'eric \textsc{Abergel}	        & Ecole Centrale Paris\\
                                & Anirban \textsc{Chakraborti}          & Ecole Centrale Paris\\\\
      \textit{Rapporteurs :}    & Yannick \textsc{Malevergne}           & Ecole de Management de Lyon\\
                                & Gr\'egory \textsc{Schehr}  		    & Universit\'e de Paris-Sud\\
      \textit{Examinateurs :}   & Jean-Philippe \textsc{Bouchaud}       & Capital Fund Management\\
                                & Jean-David \textsc{Fermanian}         & CREST\\
                               %& Dominique \textsc{Gu\'egan} ?  		& Universit\'e de Paris 1\\
                               %& Sidney \textsc{Redner} ?              & Boston University\\
                               %& Matteo \textsc{Marsili} ?             & ICTP Trieste\\\\
%      \textit{Invit\'e :}		& Fran\c cois \textsc{Quittard-Pinon} ? & Ecole de Management de Lyon
\end{tabular}
%\end{center}
\vfill
\flushright{\texttt{2013-ECAP-0042}}
%\fancyfoot[RE,RO,F]{2013-ECAP-0042}
\end{titlepage}
\sloppy

\titlepage

%\fancyfoot{}
\restoregeometry
\newgeometry{left=3.5cm,right=3cm,top=1.75cm,bottom=3.5cm} % from version 5 and later

\begin{titlepage}
\begin{center}

\includegraphics[trim=-15 0 0 0,clip]{logo_nb_centraleparis-vectoriel} \\
\vspace*{.8cm}
\noindent {\Large \textbf{ECOLE CENTRALE\\ DES ARTS ET MANUFACTURES}} \\
\vspace*{0.3cm}
\noindent \textbf{PARIS} \\
\vfill%\vspace*{0.5cm}
\noindent \fontsize{25}{30} \textbf{P H D\ \ T H E S I S} \\
\vspace*{1cm}
\noindent \normalsize {to obtain the title of} \\
\vspace*{0.5cm}
\noindent \Large \textsc{PhD of Science} \\
%\vspace*{0.2cm}
%\noindent \large of the University of Nice - Sophia Antipolis \\
\noindent \Large {Specialty : {Applied Mathematics}}\\
\vfill%\vspace*{0.4cm}
\noindent {\LARGE \textbf{{Non-linear~Dependences in~Finance}}} \\
\vfill%\vspace*{0.8cm}
\noindent \normalsize Defended by \\
\vspace*{0.2cm}
\noindent \Large R\'emy \textsc{Chicheportiche} \\
\vspace*{0.3cm}
\noindent \normalsize on June 27 2013 \\
%\vspace*{0.8cm}
%\noindent \large Thesis Advisor: Fr\'ed\'eric \textsc{Abergel} \\
%\vspace*{0.2cm}
%\noindent \large prepared at INRIA Sophia Antipolis, \textsc{Asclepios} Team
\end{center}
\vfill
%\noindent \normalsize \textbf{Jury :} \\
%\begin{center}
\noindent \normalsize 
\begin{tabular}{ll@{\quad -\quad }l}
      \textit{Advisors :}	    & Fr\'ed\'eric \textsc{Abergel}	        & Ecole Centrale Paris\\
                                & Anirban \textsc{Chakraborti}          & Ecole Centrale Paris\\\\
      \textit{Reviewers :}	    & Yannick \textsc{Malevergne}           & Ecole de Management de Lyon\\
                                & Gr\'egory \textsc{Schehr}  		    & Universit\'e de Paris-Sud\\
      \textit{Examinators :}    & Jean-Philippe \textsc{Bouchaud}       & Capital Fund Management\\
                                & Jean-David \textsc{Fermanian}         & CREST\\
                               %& Dominique \textsc{Gu\'egan} ?  		& Universit\'e de Paris 1\\
                               %& Sidney \textsc{Redner} ?              & Boston University\\
                               %& Matteo \textsc{Marsili} ?             & ICTP Trieste\\\\
%      \textit{Invited :}		& Fran\c cois \textsc{Quittard-Pinon} ? & Ecole de Management de Lyon
\end{tabular}
%\end{center}

\end{titlepage}
\sloppy

\titlepage

\restoregeometry

\dominitoc

\pagenumbering{roman}

\cleardoublepage
%\chapter*{} % pour formater cette page comme un debut de chapitre (pas de header)
\newgeometry{left=1.3in,right=0.9in,top=.85in,bottom=.95in} % from version 5 and later

%\begin{vcenterpage}
\noindent\rule[2pt]{\textwidth}{0.5pt}
\begin{center}
{\large\textbf{D\'ependances non-lin\'eaires en finance\\}}
\end{center}
{\large\textbf{R\'esum\'e:}}
La th\`ese est compos\'ee de trois parties. 
La partie I introduit les outils math\'ematiques et statistiques appropri\'es pour l'\'etude des d\'ependances, 
ainsi que des tests statistiques d'ad\'equation pour des distributions de probabilit\'e empiriques. 
Je propose deux extensions des tests usuels lorsque de la d\'ependance est pr\'esente dans les donn\'ees, 
et lorsque la distribution des observations a des queues larges.
Le contenu financier de la th\`ese commence \`a la partie II. 
J'y pr\'esente mes travaux concernant les d\'ependances transversales entre les s\'eries chronologiques 
de rendements journaliers d'actions, c'est \`a dire les forces instantan\'ees qui relient plusieurs actions entre elles
 et les fait se comporter collectivement plutôt qu'individuellement. 
 Une calibration d'un nouveau mod\`ele \`a facteurs est pr\'esent\'ee ici, 
 avec une comparaison \`a des mesures sur des donn\'ees r\'eelles.
Finalement, la partie III \'etudie les d\'ependances temporelles dans des s\'eries chronologiques individuelles, 
en utilisant les m\^emes outils et mesures de corr\'elations. 
Nous proposons ici deux contributions \`a l'\'etude du ``volatility clustering'', de son origine et de sa description: 
l'une est une g\'en\'eralisation du m\'ecanisme de r\'etro-action ARCH dans lequel les rendements sont auto-excitants, 
et l'autre est une description plus originale des auto-d\'ependances en termes de copule. 
Cette derni\`ere peut \^etre formul\'ee sans mod\`ele et n'est pas sp\'ecifique aux donn\'ees financi\`eres. 
En fait, je montre ici aussi comment les concepts de r\'ecurrences, records, r\'epliques et temps d'attente, 
qui caract\'erisent la dynamique dans les s\'eries chronologiques, peuvent \^etre \'ecrits dans la cadre unifi\'e des copules.
\\
{\large\textbf{Mots-cl\'es:}}
d\'ependances statistiques, copules, tests d'ad\'equation, processus stochastiques, s\'eries temporelles, models financiers, volatility clustering
\\
\noindent\rule[2pt]{\textwidth}{0.5pt}
%\end{vcenterpage}

\vfill

%\begin{vcenterpage}
\noindent\rule[2pt]{\textwidth}{0.5pt}
\begin{center}
{\large\textbf{Non-linear dependences in finance\\}}
\end{center}
{\large\textbf{Abstract:}}
The thesis is composed of three parts. 
Part I introduces the mathematical and statistical tools that are relevant for the study of dependences, 
as well as statistical tests of Goodness-of-fit for empirical probability distributions. 
I propose two extensions of usual tests when dependence is present in the sample data 
and when observations have a fat-tailed distribution.
The financial content of the thesis starts in Part II. 
I present there my studies regarding the ``cross-sectional'' dependences among the time series of daily stock returns, 
i.e.\ the instantaneous forces that link several stocks together and 
make them behave somewhat collectively rather than purely independently. 
A calibration of a new factor model is presented here, together with a comparison to measurements on real data.
Finally, Part III investigates the temporal dependences of single time series, 
using the same tools and measures of correlation. 
I propose two contributions to the study of the origin and description of ``volatility clustering'': 
one is a generalization of the ARCH-like feedback construction where the returns are self-exciting, 
and the other one is a more original description of self-dependences in terms of copulas. 
The latter can be formulated model-free and is not specific to financial time series. 
In fact, I also show here how concepts like recurrences, records, aftershocks and waiting times, 
that characterize the dynamics in a time series can be written in the unifying framework of the copula.
\\
\noindent\rule[2pt]{\textwidth}{0.5pt}
%\end{vcenterpage}

\restoregeometry

\cleardoublepage

% \section*{}
% Last thing to do ...

\section*{Remerciements / Acknowledgments}

Mes premiers remerciements vont \'evidemment \`a Jean-Philippe Bouchaud,
qui m'a offert l'opportunit\'e de m'engager sur la voie du doctorat \`a un moment o\`u j'en avais abandonn\'e l'id\'ee.
Je lui suis infiniment reconnaissant pour ses encouragements r\'ep\'et\'es, 
la profonde confiance qu'il m'a accord\'ee et pour son invraisemblable disponibilit\'e.

Je remercie Fr\'ed\'eric Abergel, directeur de la Chaire de finance quantitative de l'\'Ecole Centrale, 
 pour avoir accept\'e la supervision acad\'emique de mon doctorat, 
 et mis \`a ma disposition un environnement intellectuel et mat\'eriel favorable.
 Merci \`a Anirban Chakraborti pour sa co-supervision et son enthousiasme \`a me proposer de nouveaux projets.
 
J'adresse un remerciement particulier \`a mon jury de th\`ese:
\`a Yannick Maleverge et Gr\'egory Schehr pour avoir imm\'ediatement accept\'e la t\^ache de relecteur,
et \`a Jean-David Fermanian qui a port\'e un int\'er\^et \`a mon travail.
%et \`a X pour s'\^etre d\'eplac\'e.....

Je veux aussi exprimer ma gratitude aux dirigeants de Capital Fund Management pour s'\^etre impliqu\'es dans une convention CIFRE
et m'avoir procur\'e des conditions de travail des plus confortables,
ainsi qu'\`a tous les collaborateurs avec qui j'ai eu la chance d'interagir.
Je garderai longtemps le souvenir de la ``salle des stagiaires'' et de ses occupants
plus ou moins \'eph\'em\`eres, \`a commencer par Romain Allez et Bence T\'oth qui y ont s\'ejourn\'e avec moi le plus longtemps.

Mes passages \`a l'\'Ecole ont \'et\'e l'occasion d'\'echanges (scientifiques ou non) avec les 
membres de la chaire et du laboratoire MAS, en particulier avec Damien Challet et les autres doctorants:
 Alexandre Richard, Aymen Jedidi, Ban Zheng, Fabrizio Pomponio, Nicolas Huth, avec qui j'ai le plus interagi.
 Qu'ils soient ici cordialement salu\'es, de m\^eme que mes anciens %et fid\`eles 
 camarades
 Michael Kastoryano, Nicolas Cantale, Liliana Foletti, Stefanie Stantcheva et David Salfati.

Je n'aurais sans doute pas d'autre occasion d'exprimer ma gratitude 
\`a toutes les figures qui ont marqu\'e mon \'education.
Qu'il me soit donc permis de rendre ici un hommage tardif \`a
Mesdames Nelly Piguet, Regula Krattenmacher et Larissa Shargorodsky,
Messieurs Didier Deshusses, Christian Charvin, Frederic Hermann, Elie Prigent et Bernard Delez,
Monsieur Paco Carbonell, Ma\^itres Raymond Hyvernaud, Georges L\'eger, Jean-Marc Cagnet et Yanaki Altanov.
% Jean-Paul Nicolai

Je d\'edicace cette th\`ese \`a ma grande famille, proche et lointaine, qu'il est superflu de remercier.

\paragraph{}
\begin{flushright}Paris, juin 2013\end{flushright}%

\cleardoublepage
%\chapter*{} % pour formater cette page comme un debut de chapitre (pas de header)

%\nobibliography*

\defcitealias{chicheportiche2011goodness}{
    R.~Chicheportiche and J.-P.~Bouchaud, Goodness-of-fit tests with dependent observations, \emph{J.~Stat.~Mech.} \textbf{9}, P09003, 2011
}
\defcitealias{chicheportiche2012joint}{
    R.~Chicheportiche and J.-P.~Bouchaud, The joint distribution of stock returns is not elliptical, \emph{Int.~J.~Theo.~Appl.~Fin.} \textbf{15}(3), 1250019, 2012
}
\defcitealias{chicheportiche2012weighted}{
    R.~Chicheportiche and J.-P.~Bouchaud, Weighted Kolmogorov-Smirnov test: Accounting for the tails, \emph{Phys.~Rev.~E} \textbf{4}(86), 1115, 2012
}

\defcitealias{chicheportiche2012fine}{
    R.~Chicheportiche and J.-P.~Bouchaud, The fine structure of volatility feedback, arXiv preprint \mbox{qfin.ST/1206.2153}, Aug.~2012%, submitted to \emph{Eur.~Phys.~J.~B}
}
\defcitealias{chicheportiche2013recurrences}{
    R.~Chicheportiche and A.~Chakraborti, A model-free characterization of recurrences in stationary time series, arXiv preprint \mbox{physics.data-an/1302.3704}, Feb.~2013%, submitted to \emph{Phys.~Rev.~Lett.}
}
\defcitealias{remi}{
    R.~Chicheportiche and J.-P.~Bouchaud, A minimal factor model for non-linear dependences in stock returns
}
\defcitealias{chicheportiche2013fpt}{
    R.~Chicheportiche and J.-P.~Bouchaud, Some applications of first passage ideas in finance, to appear in \emph{First-Passage Phenomena and Their Applications} (R.~Metzler, G.~Oshanin and S.~Redner Eds.), World Scientific Publishers, May~2013
}

\defcitealias{chicheportiche2013timeseries}{
    R.~Chicheportiche and A.~Chakraborti, Non-linear and multi-points dependences in stationary time series: a copula approach
}
\defcitealias{chicheportiche2013multivariate}{
    A.~Richard, R.~Chicheportiche and J.-P.~Bouchaud, The sup of the 2D pinned Brownian sheet
}
\defcitealias{Pierre_inprep}{
    P.~Blanc, R.~Chicheportiche and J.-P.~Bouchaud, The fine structure of volatility II: intraday and overnight effects
}
\defcitealias{gould2013power}{
    M.~Gould, R.~Chicheportiche and J.-P.~Bouchaud, Re-assessing the tail distribution of financial returns
}
\defcitealias{chicheportiche2013ironing}{
    R.~Chicheportiche and J.-P.~Bouchaud, Ironing the copula surface
}

\section*{List of publications and submitted articles}\label{sec:publications}
\addcontentsline{toc}{section}{List of publications and submitted articles}

\subsubsection*{Published articles:}    
 \begin{list}{}
 {\leftmargin=3.5em \itemindent=-1em}
     \item[\cite{chicheportiche2011goodness}] \citetalias{chicheportiche2011goodness}
     \item[\cite{chicheportiche2012joint}]    \citetalias{chicheportiche2012joint}
     \item[\cite{chicheportiche2012weighted}] \citetalias{chicheportiche2012weighted}
 \end{list}

\paragraph{Pre-prints and submitted articles:}
 \begin{list}{}
 {\leftmargin=3.5em \itemindent=-1em}
     \item[\cite{chicheportiche2012fine}]         \citetalias{chicheportiche2012fine}
     \item[\cite{chicheportiche2013recurrences}]  \citetalias{chicheportiche2013recurrences}
     \item[\cite{remi}]                           \citetalias{remi}
     \item[\cite{chicheportiche2013fpt}]          \citetalias{chicheportiche2013fpt}
 \end{list}

\paragraph{Work in progress, working papers and collaborations:}
 \begin{list}{}
 {\leftmargin=3.5em \itemindent=-1em}
     \item[\cite{chicheportiche2013timeseries}]   \citetalias{chicheportiche2013timeseries}
     \item[\cite{Pierre_inprep}]                  \citetalias{Pierre_inprep}
     \item[\cite{chicheportiche2013multivariate}] \citetalias{chicheportiche2013multivariate}
     \item[\cite{chicheportiche2013ironing}]      \citetalias{chicheportiche2013ironing}
     \item[\cite{gould2013power}]                 \citetalias{gould2013power}
 \end{list}

%\begin{enumerate}
%\item \bibentry{chicheportiche2011goodness}
%\item \bibentry{chicheportiche2012joint}
%\item \bibentry{chicheportiche2012weighted}
%\item \bibentry{tilak2012study}
%\end{enumerate}

%\begin{enumerate}
%\item \bibentry{remi}
%\item \bibentry{chicheportiche2012fine}
%\item \bibentry{Pierre_inprep}
%\item \bibentry{chicheportiche2013recurrences}
%\item \bibentry{chicheportiche2013multivariate}
%\end{enumerate}

\cleardoublepage
\nomenclature{RHS}{Right-hand side}%
\nomenclature{LHS}{ Left-hand side}%
\printnomenclature[2.5cm]

\tableofcontents

\mainmatter

\chapter{Introduction}\label{chap:intro}
\minitoc
\clearpage\section{Introduction en fran\c cais}

\subsection{Les march\'es financiers comme syst\`emes complexes}
L'économie pr\'esente des traits caract\'eristiques des syst\`emes complexes:
un grand nombre de variables (individus, entreprises, contrats, maturit\'es, etc.),
des agents h\'et\'erog\`enes (d'o\`u une difficult\'e \`a isoler des groupes),
de l'information partielle et de l'al\'ea, 
de l'auto-organisation et des ph\'enom\`enes spontan\'es (interactions fortes, non-perturbativit\'e),
une dynamique de r\'eseau (entres banques, pays, contreparties, entreprises, etc.) et des interactions \`a longue port\'ee, 
un quasi-continuum d'\'echelles de temps, de la non-stationnarit\'e, des trajectoires avec m\'emoire 
et de l'exposition \`a des forces ext\'erieures (r\'egulations, catastrophes, etc.).

Cette r\'ealit\'e contraste fortement avec les hypoth\`eses des th\'eories \'economique et financi\`ere classiques,
comme l'homog\'en\'eit\'e (l'existence d'un agent repr\'esentatif), une rationalit\'e parfaite,
une capacit\'e de pr\'evision \`a horizon infini, un monde \`a l'\'equilibre dont les perturbations sont instantan\'ement ajust\'ees, 
l'absence de m\'emoire (toute l'information pass\'ee est refl\'et\'ee dans l'\'etat actuel du monde), \ldots

Pour traiter ces enjeux sous un angle plus empirique, des communaut\'es de chercheurs avec des points de vue 
diff\'erents ont progressivement investi le champ des sciences \'economiques et sociales.
Parmi elles, on trouve des sociologues quantitatifs, des \'econophysiciens, des statisticiens (physiciens statistiques ou biostatisticiens),
des math\'ematiciens appliqu\'es.
En effet, beaucoup des traits list\'es ci-dessus ont d\'ej\`a \'et\'e \'etudi\'es en divers champs des sciences naturelles,
o\`u des m\'ethodes ont \'et\'e  mises au point pour prendre en compte l'al\'ea, les interactions de r\'eseaux, les comportements collectifs, la persistance et la r\'eversion, etc.
\cite{anderson1988economy,arthur1997economy,blume2006economy,challet2005minority,galam2008sociophysics,kirman2010economic,abergel2013econophysics,bouchaud2013crises}

Parmi tous les champs de l'activit\'e socio-\'economique, la finance de march\'e pourrait bien \^etre le plus proche 
de ce que l'on peut raisonnablement \'etudier comme un syst\`eme physique:
du point de vue th\'eorique, le nombre de degr\'es de libert\'e est quelque peu r\'eduit par les r\`egles des places de march\'e 
(m\'ecanismes d'actions, carnets d'ordres, r\'egulations, contrats standards, divulgation de l'information) alors que le nombre d'acteurs et d'actions 
reste assez large pour autoriser des moyennages dans des ensembles statistiques; 
et du point de vue empirique, la finance de march\'e g\'en\`ere des quantit\'es astronomiques de donn\'ees 
rendant possible la ``r\'ep\'etition d'exp\'eriences'' sous hypoth\`eses, et le traitement des mesures
avec des outils de statistique appropri\'es.
En fait, il est \'etabli que certaines propri\'et\'es des march\'es financiers ressemblent \`a
celles de syst\`emes \'etudi\'es en physique statistique:
intermittence dans les \'ecoulements turbulents, bruit de craquement de mat\'eriaux en r\'eponse \`a un changement de conditions ext\'erieures, 
lignes de fracture dans les solides, interactions \`a longue port\'ee ou de champ moyen dans les syst\`emes de spins sur r\'eseau,
transitions de phases et criticalit\'e spontan\'ee, pour en citer quelques-unes.

\subsection{Champ d'application de cette th\`ese}
Cette th\`ese \'etudie les d\'ependances statistiques au sens g\'en\'eral,
aussi bien ``transversalement'' (entre variables al\'eatoires) que temporellement (r\'ealisations successives d'une variable).
Tous les r\'esultats th\'eoriques s'appliquent donc en principe \`a tout champ o\`u plusieurs variables interagissent, ou
o\`u des processus stationnaires ont de la m\'emoire.
Plus sp\'ecifiquement, l'\'etude est restreinte \`a des processus discrets avec des intervalles de temps \'equidistants,
ce qui signifie que les s\'eries chronologiques correspondantes sont \'echantillonn\'ees et n'ont pas besoin d'\^etre ``marqu\'ees'' temporellement.

En ce qui concerne le contenu appliqu\'e, le focus est sur les donn\'ees financi\`eres, 
plus pr\'ecis\'ement les s\'eries chronologiques de prix (ou rendements) d'actions \`a la fr\'equence relativement basse du jour.
Il s'agit principalement des prix des actions les plus liquides n\'egoci\'ees sur les march\'es \'etats-uniens,
et dont les historiques peuvent \^etre t\'el\'echarg\'es librement de plusieurs sources.

\subsubsection*{Propri\'et\'es statistiques, d\'ependances et dynamique}
Les caract\'eristiques statistiques d'un ensemble de variables al\'eatoires sont de plusieurs ordres: 
propri\'et\'es distributionnelles, corr\'elations, et dynamique.

En certains endroits de cette th\`ese, les propri\'et\'es distributionnelles sont \`a l'\'etude:
la distribution de probabilit\'e empirique de mesures r\'ep\'et\'ees est typiquement la premi\`ere chose 
qu'un mod\`ele doit reproduire. Concr\`etement, les quatre premiers moments d'une distribution 
sont tr\`es informatifs quant au ph\'enom\`ene \`a l'\oe uvre, mais une perspective d'ensemble pr\'ecise n'est 
atteignable que si des outils puissants sont disponibles pour comparer les fonctions de r\'epartition.
D\'evelopper de tels outils est un des buts de ce travail de doctorat.

Les corr\'elations et la dynamique peuvent \^etre \'etudi\'ees comme deux faces d'une m\^eme pi\`ece,
comme des d\'ependances spatiales et temporelles. 
Elles peuvent \^etre vues respectivement comme une manifestation de propagation horizontale d'information, 
ou d'incorporation de l'information pass\'ee dans les r\'ealisations pr\'esentes.
Ces d\'ependances sont le sujet principal de cette th\`ese, avec des applications en finance
qui pourraient avoir d'importantes cons\'equences en gestion du risque.

\subsubsection*{Non-lin\'earit\'es}
Pour des raisons historiques et de commodit\'e math\'ematique, les corr\'elations lin\'eaires
sont souvent consid\'er\'ees par les professionnels (et quelques fois aussi dans le monde acad\'emique et l'enseignement)
comme l'alpha et l'om\'ega des d\'ependances.
Cette description est satisfaisante tant que l'on se trouve dans un r\'egime o\`u de petites 
perturbations ont de petites cons\'equences, mais l'exp\'erience enseigne que ce n'est pas toujours le cas.
Souvent, des effets de seuil, de la latence et de la m\'emoire (hyst\'er\`ese), des effects collectifs (avalanches),
du chaos (sensibilit\'e extr\^eme aux conditions initiales), etc.\ peuvent g\'en\'erer des comportements anormaux,
et il est alors n\'ecessaire de faire attention aux non-lin\'earit\'es dans les d\'ependances, 
en \'etudiant par exemple les probabilit\'es d'\'ev\`enements de queue joints, 
les corr\'elations conditionnelles, corr\'elations quadratiques, etc.
Dans des situations extr\^emes, les corr\'elations lin\'eaires peuvent m\^eme \^etre trompeuses 
pour la compr\'ehension des interactions sous-jacentes, et comme l'a \'ecrit D.~Helbing, 
``un mod\`ele non-lin\'eaire simple peut \^etre capable d'expliquer des ph\'enom\`enes,
que m\^eme des mod\`eles lin\'eaires compliqu\'es peuvent \'echouer \`a reproduire.
Les mod\`eles non-lin\'eaires pourraient apporter un \'eclairage nouveau sur les ph\'enom\`enes sociaux. 
Ils pourraient m\^eme conduire \`a un changement de paradigme dans la mani\`ere dont nous interpr\'etons la soci\'et\'e.''~\cite{helbing2012social}

\subsection{Plan g\'en\'eral de la th\`ese, et principales contributions originales}
La th\`ese est compos\'ee de trois parties. 
La partie \ref{part:partI} introduit les outils math\'ematiques et statistiques appropri\'es pour l'\'etude des d\'ependances.
La plupart de ce qui y est pr\'esent\'e est connu, mais je passe en revue les d\'efinitions et propri\'et\'es principales.
C'est aussi dans cette partie que j'introduis les notations que j'ai essay\'e de maintenir coh\'erentes tout au long de la th\`ese.
J'ai pass\'e beaucoup de temps \`a travailler sur des tests statistiques d'ad\'equation (``Goodness of Fit'', en anglais) 
pour des distributions de probabilit\'e empiriques. 
Cela \'etait originellement motiv\'e par le besoin de tester les mod\`eles bivari\'es finaux contre des hypoth\`eses nulles. 
Ce ``side project'' initial s'est av\'er\'e \^etre li\'e \`a un probl\`eme non-r\'esolu et connu pour \^etre difficile.
Je n'apporte pas de solution d\'efinitive, mais discute en annexe de possibles directions pour
des travaux futurs, qui ont \'emerg\'e au cours de sessions de remue-m\'eninges avec Jean-Philippe Bouchaud et Alexandre Richard.
Pour autant, le sujet est fascinant, \'etant \`a la fronti\`ere de la th\'eorie des probabilit\'es (processus stochastiques)
et la physique (particules quantiques dans des potentiels), et j'ai \'et\'e int\'eress\'e de travailler sur d'autres extensions
des tests d'ad\'equation usuels.
Ce qui \'etait cens\'e n'\^etre que le d\'eveloppement d'une bo\^ite \`a outils pour un usage imm\'ediat
a finalement donn\'e lieu \`a de tr\`es int\'eressants projets th\'eoriques qui se sont concr\'etis\'es
par la publication de deux articles et d'un chapitre de livre.

Le contenu financier de la th\`ese commence \`a la partie \ref{part:partII}.
J'y pr\'esente mes travaux concernant les d\'ependances transversales entre les s\'eries chronologiques 
de rendements journaliers d'actions, c'est \`a dire les forces instantan\'ees qui relient plusieurs actions entre elles
et les fait se comporter collectivement plut\^ot qu'individuellement.
Ce vaste sujet avait d\'ej\`a \'et\'e approch\'e dans mon travail de master \cite{chicheportiche2010master},
o\`u j'ai effectu\'e une \'etude empirique des d\'ependances non-lin\'eaires, 
et ai d\'ecrit une construction hi\'erarchique pour les mod\'eliser.
Une partie de ce travail a \'et\'e publi\'ee avec des d\'eveloppements cons\'ecutifs
dans un journal de finance, et a servi de pr\'emisse pour le mod\`ele \`a facteur 
sur lequel j'ai finalement abouti au cours du doctorat.
Une calibration de ce mod\`ele est pr\'esent\'ee ici, avec une comparaison aux donn\'ees r\'eelles.

Finalement, la partie \ref{part:partIII} \'etudie les d\'ependances temporelles dans des s\'eries chronologiques individuelles,
en utilisant les m\^emes outils et mesures de corr\'elations.
Il est bien connu que les rendements d'actions ont une tr\`es faible auto-corr\'elation lin\'eaire, 
ce qui se manifeste dans les alternances al\'eatoires de signes (rendements positifs ou n\'egatifs),
mais leurs amplitudes ont une longue m\'emoire. 
L'effet r\'esultant de ``volatility clustering'' est une des plus anciennes \'enigmes en finance,
et nous proposons ici deux contributions \`a l'\'etude de son origine et de sa description:
l'une est une g\'en\'eralisation du m\'ecanisme de r\'etro-action ARCH dans lequel les rendements sont auto-excitants,
et l'autre est une description plus originale des auto-d\'ependances en termes de copule.
Cette derni\`ere peut \^etre formul\'ee sans mod\`ele et n'est pas sp\'ecifique aux donn\'ees financi\`eres.
En fait, je montre ici aussi, apr\`es une longue collaboration avec Anirban Chakraborti, comment les concepts
de r\'ecurrences, records, r\'epliques et temps d'attente, qui caract\'erisent la dynamique dans les s\'eries chronologiques,
peuvent \^etre \'ecrits dans la cadre unifi\'e des copules.

\vfill
%\paragraph{}
La th\`ese est r\'edig\'ee au pluriel, car la plupart de son contenu est le fruit de collaborations 
(\`a tout le moins avec mes directeurs), et parce que j'y inclus beaucoup de mat\'eriel publi\'e.
Je renvoie aux articles list\'es en page~\pageref{sec:publications} et aux remerciements qu'ils contiennent
pour l'identit\'e de mes co-auteurs et collaborateurs sur chaque sujet.

\clearpage\makeatletter
\newenvironment{Numeroter}[1][1]{%
    \renewcommand{\theenumi}{#1.\arabic{enumi}}
    \begin{enumerate}{}
   %{\leftmargin=3.5em \itemindent=-1em}
}{  \end{enumerate}
    \renewcommand{\theenumi}{\arabic{enumi}}
}
\makeatother

\section{Introduction in English}

\subsection{The financial markets as complex systems}
The economy presents features that are characteristic of complex systems:
a large number of variables (agents, firms, contracts, maturities, etc.), 
heterogeneous agents (difficulty to isolate groups), 
partial information and randomness, 
self-organization and emergent phenomena (strong interactions, non-perturbativity), 
network dynamics (among banks, countries, counter-parties, companies, etc.) and long-ranged interactions,
quasi-continuum of time scales, 
non-stationarity, history dependence and sensitivity to external forces (regulations, catastrophes, etc.).

This reality is in sharp contrast with the hypotheses of classical economic and financial theory, like
homogeneity (representative agent), perfect rationality, infinite horizon forecasting capabilities, 
steady-state world where perturbations of the equilibrium are instantaneously adjusted,
absence of memory (all past information is reflected in the current state of the world), \ldots

To address these issues with a more empirical perspective, 
different communities with diverse point of views have progressively invested the social and economic sciences.
These include quantitative sociologists, econophysicists, statisticians (statistical physicists, biostatisticians),
applied mathematicians and others.
Indeed, many of the features listed above have already been encountered in various areas of natural sciences, 
where methods have been designed to account for randomness, network interactions, collective behaviors, persistence and reversion, etc.
\cite{anderson1988economy,arthur1997economy,blume2006economy,challet2005minority,galam2008sociophysics,kirman2010economic,abergel2013econophysics,bouchaud2013crises}

Among all the fields of socio-economic activity, market finance may however be the closest to what
can be reasonably studied like a physical system: 
on the theoretical side, the number of degrees of freedom is somewhat reduced by the rules of the exchanges 
(auctions, order books, regulations, standard contracts, information disclosure) 
while the number of ``players'' and ``actions'' is still very large allowing for averaging in statistical ensembles;
and on the empirical side, it produces huge amounts of data making it possible to ``repeat experiments'' under assumptions,
and treat the measurements with appropriate statistical tools.
In fact there are evidences that some properties of financial markets resemble those of systems studied in statistical physics:
intermittency in turbulent flows, crackling noise of materials responding to changing external conditions, lines of fractures in solids, 
long-ranged or mean-field interactions in Ising spins systems, phase transitions and self-organized criticality, to name a few.

\subsection{Field covered in the thesis}
This thesis studies statistical dependences in a general sense, 
both cross-sectionally (between random variables) and temporally (successive realizations of one variable).
All theoretical results thus apply in principle to any field where many variables are interacting or 
where stationary processes have memory. %\textcolor{red}{Examples ?}.
More specifically, the study is restricted to discrete processes with equidistant intervals,
meaning that the corresponding series are sampled and need not be ``time-stamped''.

As of the applied content, the focus is on financial data, 
more precisely time series of stock prices (or returns) at the rather low frequency of the day. 
Mainly, the prices of the most liquid stocks traded in US markets are used, whose histories can be downloaded freely from different sources.

%\textcolor{red}{Suite du Mast\`ere Sp\'ecialis\'e \cite{chicheportiche2010master}.}

\subsubsection*{Statistical properties, dependences and dynamics}
The statistical characteristics of a set of random variables are of many kinds:
distributional properties, correlations and dynamics.

In some places of this thesis, distributional properties are under the spotlights.
As a matter of principle, the empirical probability distribution function
of repeated measurement of a variable is the first thing a model should reproduce.
Concretely, the first four moments of a distribution are much informative
about the underlying phenomenon, but a precise global picture is only achieved
if powerful tools are available to compare distribution functions.
Developing such tools is one of the goals of this PhD work.

Correlations and dynamics can be studied as two faces of the same coin: 
they are spatial (or ``cross-sectional'') dependences and temporal dependences,
and can be regarded as a manifestation of horizontal information propagation, 
or incorporation of past information into current realizations, respectively.
These dependences are the main subject of the present thesis, 
with applications in finance that could have important consequences in risk management.

\subsubsection*{Non-linearities}
For historical reasons and mathematical convenience, linear correlations are often seen by practitioners 
(and sometimes also in academia and teaching) as the ``be-all and end-all'' of dependences.
This works fine as long as small perturbations have small effects, 
but experience teaches that this is not always the case. 
Often, threshold effects, latency and memory (hysteresis), collective effects (avalanches), chaos (sensitivity to initial conditions)
generate abnormal behavior, 
and it is necessary to care for non-linearities in dependences by studying, for example, tail joint probabilities, conditional correlations, quadratic correlations, etc.
In extreme situations, linear correlations may even be misleading the understanding of the underlying interactions, 
and, as D.~Helbing wrote, ``a simple nonlinear model may explain phenomena, which even complicated linear models may fail to reproduce.
Nonlinear models are expected to shed new light on [...] social phenomena. 
They may even lead to a paradigm shift in the way we interpret society.''~\cite{helbing2012social}

%Causality is also not revealed by usual measures of dependences, 

\subsection{Detailed outline of the thesis and main original contributions}

The thesis is composed of three parts.
Part~\ref{part:partI} introduces the mathematical and statistical tools that are relevant for the study of dependences.
Most of what is presented there is common knowledge, 
but I review the main definitions and properties for completeness and later usage. 
It is also the place where I introduce the notations, which I tried (but probably not succeeded completely) to keep coherent along the thesis.
I have spent quite some time working on statistical tests of Goodness-of-fit for empirical probability distributions.
This was originally motivated by the need to test the final multivariate models, or at least the bivariate marginals, against null hypotheses.
This initial ``side project'' turned out to be a tough one and no definite solution is provided here
--- I still give in appendix an account of the literature and of some possible directions for future endeavor, 
that came out during brainstorming sessions with Jean-Philippe Bouchaud and Alexandre Richard.
Yet the subject was fascinating, being at the frontier of probability theory (stochastic processes) and plain physics (quantum particles in potentials),
and I was interested in working on other extensions of the usual GoF tests.
What was supposed to be only the development of a toolkit for immediate use 
gave rise to very interesting theoretical projects and in fact led to the publication of two articles and a textbook chapter.

The financial content of the thesis starts in Part~\ref{part:partII}. I present there
 my studies regarding the ``cross-sectional'' dependences among the time series
of daily stock returns, i.e.\ the instantaneous forces that link several stocks together
and make them behave somewhat collectively rather than purely independently.
Some work was already done in my Master's thesis \cite{chicheportiche2010master},
where I performed an empirical study of non-linear dependences, 
and described a hierarchical construction to model them. 
Part of that work was published with subsequent developments in a financial journal,
and served as a premise for the final factor model that I have come with.
A calibration of this model is presented here, together with a comparison to measurements on real data.

Finally, Part~\ref{part:partIII} investigates the temporal dependences of single time series, 
using the same tools and measures of correlation.
As is well known, stock returns have a vanishing linear autocorrelation that manifests itself in random sign alternations, 
but their magnitudes have a long memory. 
The resulting ``volatility clustering'' effect is one of the oldest puzzles in finance,
and we propose here two contributions to the study of its origin and description:
one is a generalization of the ARCH-like feedback construction where the returns are self-exciting,
and the other one is a more original description of self-dependences in terms of copulas.
The latter can be formulated model-free and is not specific to financial time series.
In fact, I also show here, after a longstanding collaboration with Anirban Chakraborti and undergrad students at Ecole Centrale,
how concepts like recurrences, records, aftershocks and waiting times, that characterize the dynamics in a time series
can be written in the unifying framework of the copula.

After this brief non-technical overview of the landscape, 
I give below a more detailed outline with an abstract of each chapter.
I quit now the first person and switch to the `we' as most of what follows is the 
result of collaborations (at the very least with my advisors), and because I include much material of published work.
I refer to the articles listed in page~\pageref{sec:publications} and acknowledgments therein 
for the identity of my co-authors on each topic.

\paragraph{Part I} \label{abstracts:begin}
\begin{Numeroter}[I]
\setcounter{enumi}{1}
\item There exist many measures of dependence between random variables,
      some related to joint amplitudes, others related to joint occurrence probabilities.
      We recall the definition and properties of some of them, and introduce the unifying framework of the copula,
      an object that embeds all the linear and non-linear dependences that are invariant under a rescaling of the marginals,
      i.e.\ those that count joint occurrences.
      We then rewrite the important measures of dependences in terms of the copula, 
      and suggest ways to efficiently visualize and compare them.
      An interesting point of comparison is the elliptical copula, that includes the Gaussian and Student cases.
      We define this class of multivariate distributions, and summarize its main features.
      
      The content of this chapter needs not be read straight from the beginning, 
      but can be accessed ``on demand'' when referred to, later in the text.
      The original material is limited to one theorem characterizing the asymptotic 
      tail dependence coefficient of a Student pair.
\item 
%\defcitealias{krapivsky1996life}{[P.~L.~Krapivsky and S.~Redner, \emph{Am.~J.~Phys.~}\textbf{64}, 546 (1996)]}
%\defcitealias{turban1992anisotropic}{[L.~Turban, \emph{J.~Phys.~A~}\textbf{25}, 127 (1992)]}
Usual goodness-of-fit (GoF) tests are designed for \emph{independent} samples,
and are not suited to investigate \emph{tail regions}, because of the universal limit 
properties of the cumulative distribution functions.
We extend the range of applicability of the GoF tests to these two cases.
\textbf{Section~\ref{sect:weightedKS}}:   
Accurate tests for the extreme tails of empirical distributions is a very important issue, 
relevant in many contexts, including geophysics, insurance, and finance. 
% In usual GoF tests, these tails barely contribute to the test statistic because of the universal limit 
% properties of the cumulative distribution functions.
We have derived exact asymptotic results for a generalization of the large-sample Kolmogorov-Smirnov test, 
well suited to testing a distribution with constant weight at any point of the domain, 
and in particular provide an improved resolution in the tail regions. 
In passing, we have rederived and made more precise the approximate limit solutions found originally in unrelated fields.
%first in \citetalias{turban1992anisotropic} and later in \citetalias{krapivsky1996life}.
\textbf{Section~\ref{sec:GoF}}:   
We revisit the Kolmogorov-Smirnov and Cram\'er-von~Mises goodness-of-fit tests 
and propose a generalization to identically distributed, but dependent univariate random variables.
We show that the dependence leads to a reduction of the ``effective'' number of independent observations. 
The generalized GoF tests are not distribution-free but rather depend on all the lagged bivariate copulas.
Hence, a precise formulation of the test must rely on a modeled or estimated copula.

\end{Numeroter}

\paragraph{Part II} 
\begin{Numeroter}[II]
\setcounter{enumi}{3}
% \item
% This chapter is a short review of the properties of sample correlations.
% We introduce the usual estimators and discuss the noise and the signal:
% Random Matrix Theory (RMT) aims at characterizing pure noise in the spectrum of large sample covariance matrices, 
% while factor models --- be they flat of hierarchical --- provide a simple and intuitive characterization of the underlying mechanisms driving the joint behavior.
% The Principal Components Analysis (PCA) is a description in terms of factors, that relies on the spectral decomposition of the covariance matrix.
% Modeling of dependences beyond the linear correlations can be achieved using other methods such as network interactions, graphs, structural models, etc.
\item
Using a large set of daily US and Japanese stock returns, we test in detail the relevance of Student models,
and of more general elliptical models, for describing the joint distribution of returns. 
We find that while Student copulas provide a good approximation for strongly correlated pairs of stocks, 
systematic discrepancies appear as the linear correlation between stocks decreases, that rule out \emph{all} elliptical models.  
Intuitively, the failure of elliptical models can be traced to the inadequacy of the assumption of a single volatility mode for all
stocks. 
We suggest several ideas of methodological interest to efficiently visualize and compare different copulas. 
We identify the rescaled difference with the Gaussian copula and the central value of the copula as strongly discriminating observables. 
We insist on the need to shun away from formal choices of copulas with no financial interpretation.
\item
We propose a non-hierarchical multi-factor model for the joint description of daily stock returns.
The model intends to better reproduce linear and non-linear dependences of empirical time series.
The non-Gaussianity of the factors is an important ingredient of the model, 
that allows for the empirically observed anomalous copula at the medial point and other stylized facts 
(quadratic covariances, copula diagonals).
A spectral analysis of the factor series suggests that a structure is present in the \emph{volatilities},
with a dominant mode clearly affecting all linear factors and residuals.
The model embedding this feature is calibrated on US stocks over several periods,
and reproduces qualitatively all the non-trivial empirical observations.
A systematic out-of-sample prediction over a long period confirms quantitatively the power of the model
and of the proposed estimation methodology, both for linear and non-linear properties.
\end{Numeroter}

\paragraph{Part III} 
\begin{Numeroter}[III]
\setcounter{enumi}{6}
\item
We attempt to unveil the fine structure of volatility feedback effects in the context of general quadratic autoregressive (QARCH) models,
which assume that today's volatility can be expressed as a general quadratic form of the past daily returns. 
The standard ARCH or GARCH framework is recovered when the quadratic kernel is diagonal. 
The calibration of these models on US stock returns reveals several unexpected features. 
The off-diagonal (non ARCH) coefficients of the quadratic kernel are found to be highly significant both In-Sample and Out-of-Sample, 
but all these coefficients turn out to be one order of magnitude smaller than the diagonal elements. 
This confirms that daily returns play a special role in the volatility feedback mechanism, as postulated by ARCH models. 
The feedback kernel exhibits a surprisingly complex structure, incompatible with models proposed so far in the literature.
Its spectral properties suggest the existence of volatility-neutral patterns of past returns. 
The diagonal part of the quadratic kernel is found to decay as a power-law of the lag, in line with the long-memory of volatility. 
Finally, QARCH models suggest some violations of Time Reversal Symmetry in financial time series, 
which are indeed observed empirically, although of much smaller amplitude than predicted. 
We speculate that a faithful volatility model should include both ARCH feedback effects and a stochastic component. 
\item
A discrete time series is a collection of successive realizations of a random variable,
every realization being conditional on the previous state.
Seen collectively and {\it ex ante}, it can also be seen as one realization of a random vector
with (directed) dependence.
We show how the copula introduced in Part~\ref{part:partI} in the context of cross-sectional dependences
can also describe appropriately the non-linear temporal dependences in time series:
in this context, we call them ``self-copulas''. 
\item
We introduce a specific, log-normal model for these self-copulas, for which a number of analytical results are derived. 
An application to financial time series is provided. 
As is well known, the dependence is to be long-ranged in this case, a finding that we confirm using self-copulas. 
As a consequence, the acceptance rates for GoF tests are substantially higher than if the returns were iid random variables.
\end{Numeroter}

\paragraph{Appendices} 
\renewcommand{\theenumi}{\Alph{enumi}}
\begin{enumerate}
\item Empirical counterparts of the copula (estimated together with the marginals on a sample of size $T$) 
      have good properties only on discrete values in $[0,1]$:
      for finite $T$, the naive copula estimates are biased.
      We provide an approximate bias-correction mechanism.      
\item Gaussian sheets (the generalization of Gaussian process indexed by two ``times'') arise
      as limits of empirical multivariate cumulative distribution functions, when the sample size tends to infinity.
      Some properties of such sheets are of importance for distribution testing, 
      since the Kolmogorov-Smirnov test statistic is related to the law of the supremum of their absolute value.
      We study some time changes and discretization schemes that may allow to reach a better knowledge of this law.
\item Two contributions of methodological interest are made for models with a single common linear factor.
      In \ref{sec:APXrank1} the estimator of optimal weights is found perturbatively in the limit of large dimension.
      In \ref{chap:APXnested} we reverse-engineer the constitution of the common factor, 
      and embed the one-factor model in a hierarchical nested structure.
\item We provide here technical details for the perturbative expansion of the pseudo-elliptical copula with log-normal scale
      around the independence copula, when all dependence parameters are small.
\item Two appendices related to the continuous-time limit of the QARCH construction studied in Chapter~\ref{part:partIII}.\ref{chap:QARCH}.
      In \ref{QARCHapx:A} we compute the exact spectrum of the Borland-Bouchaud model of Ref.~\cite{borland2005multi}, 
      for three different kernel functions.
      In \ref{QARCHapx:B} we relate the power-law behavior in the volatility correlation function, 
      to the coefficient of the power-law decay in the kernel function of the FIGARCH model.
\end{enumerate}
\renewcommand{\theenumi}{\arabic{enumi}}
\label{abstracts:end}

\part{Mathematical and statistical tools}\label{part:partI}
\chapter{Characterizing the statistical dependence}\label{chap:statdep}
\minitoc

%http://books.google.fr/books?id=-058B6kg32sC&pg=PA23&lpg=PA23&dq=blomqvist+coefficient&source=bl&ots=dhQ85eKlqt&sig=6lkmzSQQjWwDcf3xeJEmbDvcfqE&hl=fr&sa=X&ei=LHejUOTMGua50QXuyoHQDg&ved=0CEEQ6AEwBA#v=onepage&q=blomqvist%20coefficient&f=false

\section{Bivariate measures of dependence}\label{sec:bi_meas_dep}

In this first section, we recall several bivariate measures of statistical dependence between two random variables.
Of course, as we discuss in the next section, many-points dependences among $N$ variables do not reduce 
to the pairwise dependences in the general case.
But in the perspective of empirically measuring the dependences on a dataset, 
at least four issues motivate the present restriction to 2-points coefficients:
(i) the number of triplets, quadruplets, etc.\ that should be considered when measuring coefficients of dependence 
involving 3-points, 4-points, etc.\ gets rapidly huge when considering even as few as some dozens of variables,
whereas the number of undirected 2-points dependence coefficients is ``only'' $N(N-1)/2$;
(ii) error estimates on many-points coefficients are typically larger than on 2-points measures;
(iii) in some classes of usual multivariate probability distributions, 
some information about the full structure can be inferred (or excluded) from the knowledge of the bivariate marginals;
(iv) collective behavior and ``effective'' many-points dependences can be uncovered by
collecting 2-points coefficients into a matrix or tensor and performing algebraic and spectral analysis.

Let $\vect{X}=(X_1,X_2,\ldots,X_N)^{\top}$ be a random vector with 
joint probability distribution function (pdf)\nomenclature{PDF, pdf}{Probability distribution function} $\pdf[\vect{X}](\vect{x})$ and 
joint cumulative  distribution function (cdf)\nomenclature{CDF, cdf}{Cumulative distribution function}  $\cdf[\vect{X}](\vect{x})=\int \pdf[\vect{X}](\vect{z})\mathds{1}\left\{\vect{z}\leq \vect{x}\right\}\d\vect{z}$.
Each component $X_i$ is a random variable with marginal pdf $\pdf[i](x_i)$ and cdf $\cdf[i](x_i)$.
Throughout the text (unless otherwise stated), we assume $\cdf[\vect{X}]$ is continuous and $\qdf[i]$ are well-defined inverse marginal cdfs.

As the coefficients defined in this section are all measures of pairwise dependence,
we define the pair $\vect{X_{\ij}}=(X_i,X_j)^{\top}$.
Similarly, all the transformations of $\vect{X}$ defined later in the text can
be restricted to the 2-dimensional space describing the pair.

\subsection{Usual correlation coefficients}

\subsubsection*{Bravais-Pearson's correlation coefficient}
The Bravais-Pearson's correlation coefficient of $\vect{X_{\ij}}$ is the usual normalized linear covariance
\begin{equation}\label{eq:pearson_2d}
	\rhoP[\vect{X_{\ij}}]=\frac{\esp{X_iX_j}-\esp{X_i}\esp{X_j}}{\sqrt{\var{X_i}\var{X_j}}}
\end{equation}
Often in what follows, it will be simply denoted $\rho_{\ij}$.
This correlation measure involves the probabilities \emph{and} amplitudes of the possible realizations of $\vect{X_{\ij}}$.
It can also be understood as a measure of asymmetry between the sum and the difference of the normalized variables:
\begin{equation}\label{eq:pearson_2d_asym}\tag{\ref{eq:pearson_2d}$'$}
	\rhoP[\vect{X_{\ij}}]=\frac{\Var{\tfrac{X_i}{\sqrt{\var{X_i}}}+\tfrac{X_j}{\sqrt{\var{X_j}}}}-\Var{\tfrac{X_i}{\sqrt{\var{X_i}}}-\tfrac{X_j}{\sqrt{\var{X_j}}}}}{\Var{\tfrac{X_i}{\sqrt{\var{X_i}}}+\tfrac{X_j}{\sqrt{\var{X_j}}}}+\Var{\tfrac{X_i}{\sqrt{\var{X_i}}}-\tfrac{X_j}{\sqrt{\var{X_j}}}}}
\end{equation}
(the expression at the denominator is just an explicit way of writing the number $2$).

\subsubsection*{Spearman's rho}
Let $\vect{X^{\scriptscriptstyle (i)}}\stackrel{L}{=}(0,\ldots,X_i,\ldots,0)$ be distributed like the $i$-th marginal
and independent of $\vect{X}$.%\newline
Then the vector $\displaystyle\sum_{i=1}^N\vect{X^{\scriptscriptstyle (i)}}$ has independent components with the same marginals as $\vect{X}$.
The difference $\widetilde{\vect{X}}=\vect{X}-\displaystyle\sum_{i=1}^N\vect{X^{\scriptscriptstyle (i)}}$ 
is centered and has symmetric univariate marginals by construction.
It allows to characterize the dependence in $\vect{X}$ by counting how often both components of $\vect{X_{\ij}}$ are \emph{simultaneously} below
the realizations of the independent couple with same marginals. Spearman's rho is defined as
\begin{equation}\label{eq:spearman_2d}
	\rhoS[\vect{X_{\ij}}]=3\left(\Pr{\widetilde{X}_i \widetilde{X}_j>0}-\Pr{\widetilde{X}_i \widetilde{X}_j<0}\right)
	                      %      =3\,\rhoP[\sign(\vect{\widetilde{X}_{\ij}})]
\end{equation}
and thus involves only the signs of $\widetilde{\vect{X}}$:
\begin{equation}\tag{\ref{eq:spearman_2d}$'$}
	\rhoS[\vect{X_{\ij}}]%=3\left(\Pr{\widetilde{X}_i \widetilde{X}_j>0}-\Pr{\widetilde{X}_i \widetilde{X}_j<0}\right)
	                            =3\,\rhoP[\sign(\vect{\widetilde{X}_{\ij}})]
\end{equation}

%Contrarily to Kendall's tau and Blomqvist's beta, the right-hand side of \eqref{eq:rhoS_rhoP} is not a function of 
%a global transformation of $\vect{X_{\ij}}$, but of the \emph{two} transformed vectors $\vect{X^{\scriptscriptstyle (i)}}$ and $\vect{X^{\scriptscriptstyle (j)}}$.
%

\subsubsection*{Kendall's tau}
Let $\vect{X}'\stackrel{L}{=}\vect{X}$ be a random vector with same joint distribution as $\vect{X}$:
compared to the previous case, $\vect{X}'$ has not only the same univariate marginals as $\vect{X}$, but also its dependence structure.
Then $\widetilde{\vect{X}}=\vect{X}-\vect{X}'$ is 
the difference between two independent and identically distributed random vectors. It is centered and symmetric by construction.
%(one easily checks that $\esp{\widetilde{X}_i}=F_{\widetilde{X}_i}^{-1}(\tfrac12)=0, \forall i$). 
Kendall's tau is defined as
\begin{equation}\label{eq:kendall_2d}
	\tauK[\vect{X_{\ij}}]=\Pr{\widetilde{X}_i \widetilde{X}_j>0}-\Pr{\widetilde{X}_i \widetilde{X}_j<0}
	                     =\rhoP[\sign(\vect{\widetilde{X}_{\ij}})]
\end{equation}
and it measures concordances between realizations of identically distributed couples.
%and also has an interpretation in terms of ranks (see appendix)
%It is straightforward to show the following relationship:
%\begin{equation}
%	\tauK_{\vect{X_{\ij}}}=\rhoP_{\sign(\vect{\widetilde{X}_{\ij}})}
%\end{equation}

\subsubsection*{Blomqvist's beta}
Denote now $\widetilde{\vect{X}}=\vect{X}-\med\vect{X}$. Blomqvist's beta coefficient of dependence is defined as
\begin{align}\label{eq:blomqvist_2d}
	\betaB[\vect{X_{\ij}}]&=\Pr{\widetilde{X}_i \widetilde{X}_j>0}-\Pr{\widetilde{X}_i \widetilde{X}_j<0}
                           =\rhoP[\sign(\vect{\widetilde{X}_{\ij}})].
\end{align}
If all $X_i$'s are centered around their median value (in particular if symmetric),
 then Blomqvist's beta is equal to the Pearson's correlation of the signs:
\[
	\betaB_{\ij}=\rhoP[\sign(\vect{X_{\ij}})].%=\rod[0]_{ \ij}.
\]
% or, maybe more meaningfully, 
% \begin{equation}\label{eq:blomqvist_2d_odds}\tag{\ref{eq:blomqvist_2d}$'$}
	% \betaB_{\ij}=
% \end{equation}
% where 
% $%\begin{equation}
    % o_{\ij}=\frac{}{}
% $%\end{equation}
% the odds ratio (in that precise sense) of the difference and the sum.
% This last formulation provides at least a straightforward interpretation of the sign of $\rhoP$: 
% since $o_\ij=2/[\rhoP_\ij-1]$, it says that a positive correlation will results in 
% the odds of the difference vs.\ the sum to be lower than one, and \it{vice versa}

\subsection{Non-linear correlation coefficients}
Beyond the standard correlation coefficient, one can characterize the dependence structure 
through the correlation between functions $g(\vect{X_{\ij}})=(g_i(X_i),g_j(X_j))$ of the random variables.
In particular, we define the following coefficients:
\begin{align}
	\rod_{ \ij}&=\rhoP[g(\vect{X_{\ij}})], \quad\text{when}\quad g_i(x)=g_j(x)=\sign(x)|x|^d\\%\frac{\esp{S_i S_j |X_i X_j|^d}-\esp{S_i|X_i|^d}\esp{S_j|X_j|^d}}{\sqrt{\var{S_i|X_i|^d}\,\var{S_j|X_j|^d}}}\\
	\road_{\ij}&=\rhoP[g(\vect{X_{\ij}})], \quad\text{when}\quad g_i(x)=g_j(x)=|x|^d%\frac{\esp{|X_iX_j|^d}-\esp{|X_i|^d}\esp{|X_i|^d}}{\sqrt{\var{|X_i|^d}\,\var{|X_j|^d}}},
\end{align}
provided the variances in the denominators are finite. 
The case $\rod[1]$ corresponds to the usual linear correlation coefficient, 
for which we will use the standard notation $\rho$, 
whereas $\road[2]$ is the correlation of the squared returns, 
that would appear in the Gamma-risk of a $\Delta$-hedged option portfolio. 
However, high values of $d$ are expected to be very noisy in the presence of power-law tails 
(as is the case for financial returns) and one should seek for lower order moments, 
such as $\road[1]$ which also captures the correlation between the {\it amplitudes} (or the volatility) of price moves, 
or even $\ros \equiv \rod[0]\label{eq:ros}$ that measures the correlation of the signs.
Note that when $g_i(x)=\cdf[i](x)$, one recovers Spearman's rho:
\begin{equation}\label{eq:spearman_2d_ranks}
	\rhoS_{\ij}=\rhoP[g(\vect{X_{\ij}})], \quad\text{when}\quad g_i(x)=\cdf[i](x).\phantom{\quad g_j(x)=}
\end{equation}

\begin{proof}[Proof of Eq.~\eqref{eq:spearman_2d_ranks}]
    The LHS is decomposed onto four contributions:
    \begin{align*}
        \pr{X_i>X^{\scriptscriptstyle (i)}_i, X_j>X^{\scriptscriptstyle (j)}_j}&=
        \pr{U_i>U^{\scriptscriptstyle (i)}_i, U_j>U^{\scriptscriptstyle (j)}_j}\\
	&=\iint_0^1 \pr{u_i>U^{\scriptscriptstyle (i)}_i, u_j>U^{\scriptscriptstyle (j)}_j}\,\d{C_{\vect{X_{\ij}}}(u_i,u_j)}\\
	&=\iint_0^1 u_iu_j\,\d{C_{\vect{X_{\ij}}}(u_i,u_j)}\qquad\text{since } U^{\scriptscriptstyle (i)}_i\perp U^{\scriptscriptstyle (j)}_j\\
    \pr{X_i<X^{\scriptscriptstyle (i)}_i, X_j<X^{\scriptscriptstyle (j)}_j}&=\iint_0^1 (1\!-\!u_i)(1\!-\!u_j)\,\d{C_{\vect{X_{\ij}}}(u_i,u_j)}\\
    \pr{X_i>X^{\scriptscriptstyle (i)}_i, X_j<X^{\scriptscriptstyle (j)}_j}&=\iint_0^1 u_i\,(1\!-\!u_j)\,\d{C_{\vect{X_{\ij}}}(u_i,u_j)}\\
    \pr{X_i<X^{\scriptscriptstyle (i)}_i, X_j>X^{\scriptscriptstyle (j)}_j}&=\iint_0^1 (1\!-\!u_i)\,u_j\,\d{C_{\vect{X_{\ij}}}(u_i,u_j)}.
    \end{align*}
    Define then $U_i=\cdf[i](X_i)\sim \mathcal{U}[0,1]$;
    since $\esp{U_i}=\frac12$, the RHS is
    \[
        \rhoP[\vect{U_{\ij}}]=\iint_0^1(u_i\!-\!\tfrac12)(u_j\!-\!\tfrac12)\,\d{C_{\vect{X_{\ij}}}(u_i,u_j)}.
    \]
    The result follows, noting that 
    \[
        u_i\,u_j+(1\!-\!u_i)(1\!-\!u_j)-u_i(1\!-\!u_j)-(1\!-\!u_i)u_j=4(u_i\!-\!\tfrac{1}{2})(u_j\!-\!\tfrac{1}{2})
    \]
    and $\var{U_i}=\var{U_j}=\int_0^1 (u-\tfrac{1}{2})^2\,\d{u}=\tfrac{1}{12}$.
\end{proof}

Although the motivation of the definition \eqref{eq:spearman_2d} is quite intuitive (compare dependent and independent couples with same distributions), 
the definition of Spearman's rho as Pearson's correlation of the ranks \mbox{$U_i=\cdf[i](X_i)$} makes
 it clear that it is invariant under any continuous stretching or shrinking of the axis, 
and that it involves the dependence structure regardless of the marginals.

\subsection{Tail dependence}

Another characterization of dependence, of great importance for risk management purposes, 
is the so-called \emph{coefficient of tail dependence} which measures the joint probability of extreme events. 
More precisely, the upper tail dependence is defined as \cite{embrechts02correlation,malevergne2006extreme}:
\begin{equation}\label{eq:def_tauUU}
	\tUU_{\ij}(p)=\Pr{X_i>\qdf[i](p) \left\lvert\, X_j>\qdf[j](p) \right.},
\end{equation}
where $\cdf[k]$ is the cumulative distribution function (cdf) of $X_k$, 
and $p$ a certain probability level. In spite of its seemingly asymmetric definition, it is easy to show that 
$\tUU$ is in fact symmetric in $X_i\leftrightarrow{}X_j$. % when all univariate marginals are identical.
When $p \to 1$, $\tUU_*\equiv\tUU(1)$ measures the probability that $X_i$ takes a very 
large positive value knowing that $X_j$ is also very large, and defines the asymptotic tail dependence.
Random variables can be strongly dependent from the point of view of linear correlations, 
while being nearly independent in the extremes. 
For example, bivariate Gaussian variables are such that $\tUU_*=0$ for any value of $\rho < 1$. 
The lower tail dependence $\tLL$ is defined similarly:
\begin{equation}\label{eq:def_tauLL}
	\tLL_{\ij}(p)=\Pr{X_i<\qdf[i](1-p) \left\lvert\, X_j<\qdf[j](1-p) \right.},
\end{equation}
and is equal to $\tUU(p)$ for symmetric bivariate distributions. 
One can also define mixed tail dependences:
\begin{align}\label{eq:def_tauUL}
	\tLU_{\ij}(p)&=\Pr{X_i<\qdf[i](1-p) \left\lvert\, X_j>\qdf[j](  p) \right.},\\
	\tUL_{\ij}(p)&=\Pr{X_i>\qdf[i](  p) \left\lvert\, X_j<\qdf[j](1-p) \right.},
\end{align}
with obvious interpretations.

\section{Copulas}

\subsection{Definition and properties}
A copula is a $N$-variate cdf $\cop(\vect u)$ on the hypercube $[0,1]^N$, with uniform univariate marginals.

\begin{theorem}[\cite{sklar1959fonctions}]
A unique copula can be associated with any multivariate cdf $\cdf[\vect{X}]$ by uniformizing its marginals:
\begin{align}\label{eq:def_cop_mult}
	\cop[\vect{X}](u_1,\ldots,u_N)%&=\Pr{\bigcap_{i=1}^N \cdf[i](X_i)\leq u_i}\\
                                  %&=\Esp{\prod_{i=1}^N \1{X_i\leq \qdf[i]{u_i}}}\\
                                   &=\cdf[\vect{X}]\left(\qdf[1](u_1),\ldots, \qdf[N](u_N)\right)
\end{align}
\end{theorem}
According to this view, the copula of a random vector is the joint cdf of the marginal \emph{ranks} $U_i=\cdf[i](X_i)$.
It hence encodes all the dependence between the individual random variables that is invariant under increasing transformations.
Since the marginals of $U_i$ are uniform by construction, the copula only captures their degree of ``entanglement''. 

Of course, the converse of Sklar's theorem is not true: 
infinitely many multivariate distributions can have the same copula !
In fact, it is in principle possible to build joint distributions with given copula and \emph{any} univariate (continuous) marginals,
for example a bivariate Gaussian copula with chi-2 marginals.
Of course one sees that although possible, such exotic constructions barely make sense, 
as one would intuitively expect that random variables with positive support do not have a copula with the reversal symmetry $u_i\leftrightarrow 1-u_i$.

%The survival copula $\bar{C}_{\vect{X}}(\vect{u})=\Pr{\displaystyle\bigcap_{i=1}^N \cdf[i](X_i)> u_i}$ is often defined.

\paragraph{Construction}
A copula can be constructed in three manners:
\begin{enumerate}
\item From the very definition, by explicitly specifying a function $\cop(\vect u)$ satisfying all required properties. 
      Examples include
\begin{itemize}\label{pg:cop_examples}
\item the independence (product) copula: $\Pi(\vect{u})=u_1\,u_2\,\ldots\,u_N$;%\prod_{i=1}^N u_i$
\item the upper Fr\'echet-Hoeffding bound copula $M(\vect{u})=\min\{u_1,\ldots,u_N\}$;
\item the lower Fr\'echet-Hoeffding bound $W(\vect{u})=\max\{u_1+\ldots+u_N-(N\!-\!1),0\}$\newline which is a copula only for $N=2$;
\item Archimedean copulas: $\cop(\vect{u})=\phi^{-1}\left(\phi(u_1)+\ldots+\phi(u_N)\right)$,\\ where $\phi$ (the generator function) is $N-2$ times continuously differentiable, 
and such that $\phi(1)=0,\quad\lim\limits_{u\to 0}\phi(u)=\infty$ and $\phi^{{\scriptscriptstyle (N\!-\!2)}}$ is decreasing convex 
    (more on this is Sect.~\ref{sec:archCop}, page~\pageref{sec:archCop}).
\end{itemize}

\item Applying Sklar's theorem to usual classes of parametric multivariate distributions. 
      Examples include the Gaussian copula
       \[
            \cop[\text{G}](\vect{u})=\Phi_{\Sigma}(\Phi^{-1}(u_1),\ldots,\Phi^{-1}(u_N))
       \]
       where $\Sigma$ is a symmetric positive definite matrix.
       And the more general Elliptically Contoured copulas, that we introduce in the coming Sect.~\ref{sec:elliptical_model}.
\item Implicitly defined as the dependence function of a random vector described through a structural model.
      Example: factor models.
\end{enumerate}

\paragraph{Fr\'echet-Hoeffding bounds}
Copulas are multivariate cumulative distribution functions, 
and hence take values between 0 and 1,
but there are in fact stronger constraints accounting for the properties of probability distributions.
It is obvious that a univariate cdf tends to 0 at the extreme low values of its domain, 
and to 1 at the extreme high values, reflecting roughly the facts that
the probability of an empty set is 0 and the probability of the whole universe is 1.
Much in this spirit, the extreme values of the copulas are ``pinned'' to 
$\cop(0,\ldots,0)=0$ and $\cop(1,\ldots,1)=1$.
More precisely, if \emph{any} component is 0, then the copula is 0, too.
Furthermore, by reduction to the univariate marginals, if \emph{all but one} components are 1,
then the copula takes the value of the remaining component.
Beyond these properties related to the marginals, the copula is subject to bounds
imposed by the dependence structure: in short, 
a random pair cannot be more positively correlated than if it was built of twice the same variable,
and similarly it cannot be more negatively correlated than if it was composed of a variable and its own opposite.
This is more precisely stated by the following inequalities, for a copula of arbitrary dimension:
\begin{equation}\label{eq:frechet}
    W(\vect{u})\leq \cop(\vect u) \leq M(\vect{u}),
\end{equation}
where the definitions of the upper and lower bounds were given in the previous page.
We show in Fig.~\ref{fig:bornes_cop} how these bounds look like on the diagonal and anti-diagonal of a bivariate copula.

\subsection{Dependence coefficients and the copula}%\label{ssec:copulas}

% There are many other possible coefficients of dependence, such as Spearman's rho or Kendall's tau, that measure, 
% respectively, the correlation of ranks or the ``concordance'' probability (see e.g. \cite{numericalrecipes} for an introduction). 
%Both these measures are invariant under any increasing transformations. 
The copula can be restricted to subsets of the studied random vector.
For example, the bivariate copula of any random pair $(X_i,X_j)$ is %the joint cdf of $U_i=\cdf[i](X_i)$ and $U_j=\cdf[j](X_j)$:
\begin{equation}\label{eq:def_copula2}
	\cop(u_i,u_j)= \Pr{\cdf[i](X_i)\leq{}u_i \text{ and }  \cdf[j](X_j)\leq{}u_j}.
\end{equation}
It is then natural to ask if and how all the coefficients of bivariate dependence introduced earlier
are related to the description in terms of copula.
It turns out that, whereas the $\rod$'s and $\road$'s depend on the marginal distributions, %the tail dependences, 
Spearman's rho, Kendall's tau and Blomqvist's beta can be expressed in terms of the copula only:
\begin{align*}
    \rhoS &=12\int_{[0,1]^2}\!      u_i u_j \,\d{\cop(u_i,u_j)}-3 \\
    \tauK &= 4\int_{[0,1]^2}\! \cop(u_i,u_j)\,\d{\cop(u_i,u_j)}-1 \\
    \betaB&= 4\cop(\tfrac12,\tfrac12)-1.
\end{align*}
Similarly, the tail dependence coefficients can be expressed, for $u\lesssim 1$, as:
\begin{subequations}\label{eq:tau_C}
\begin{align}
	\tUU(u)&=\frac{1-2u+\cop(u,u)}{1-u}        & 
	\tLL(u)&=\frac{\cop(1\!-\!u,1\!-\!u)}{1-u} \\ 
	\tUL(u)&=\frac{1-u-\cop(u,1\!-\!u)}{1-u}   & 
	\tLU(u)&=\frac{1-u-\cop(1\!-\!u,u)}{1-u}   .
\end{align}
\end{subequations}

\subsection{Visual representations}%\label{ssec:copulas}
	\begin{figure}
		\center
        \subfigure[Diagonal]{     \includegraphics[scale=0.45,trim=0 0 30 30,clip]{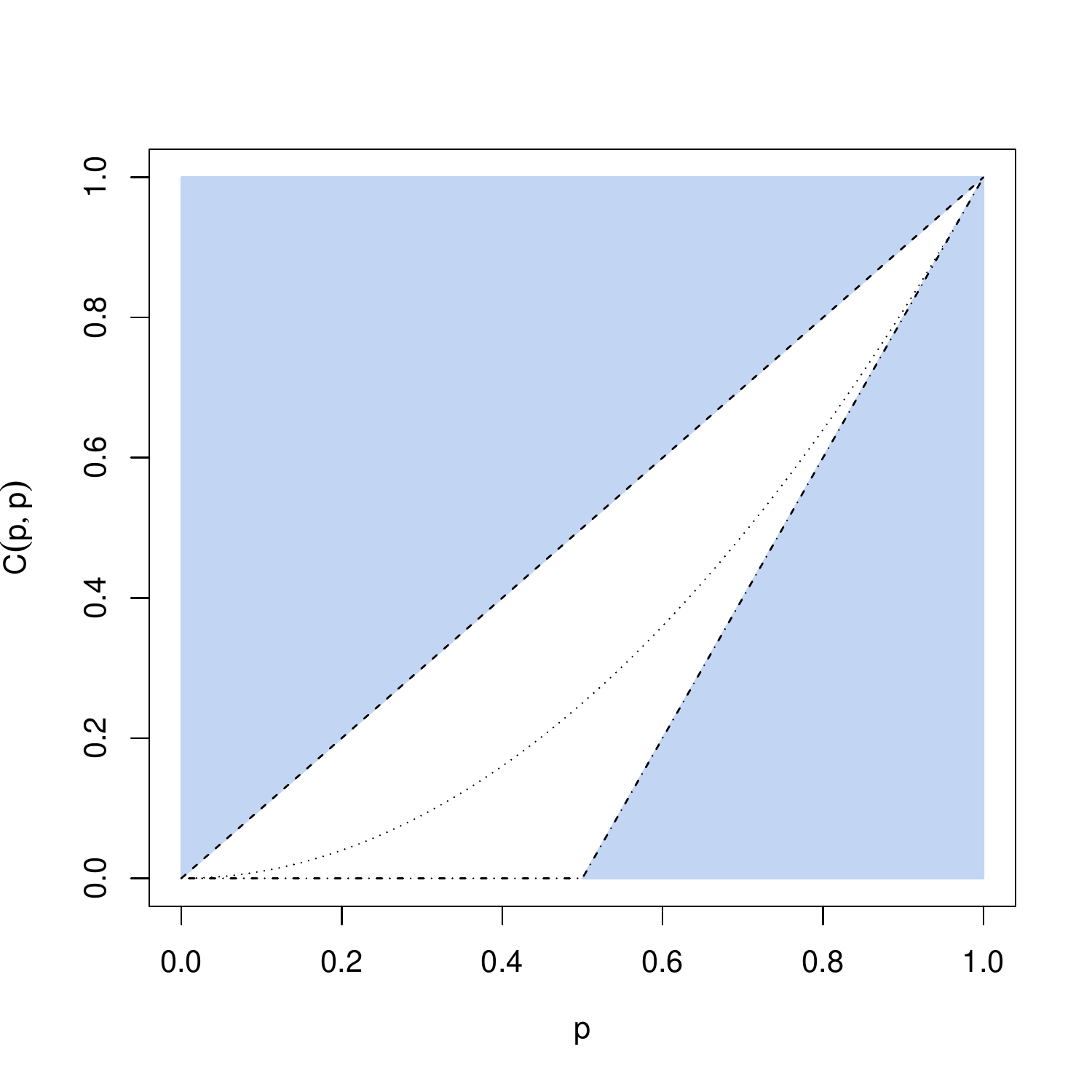}}\hfill
        \subfigure[Anti-diagonal]{\includegraphics[scale=0.45,trim=0 0 30 30,clip]{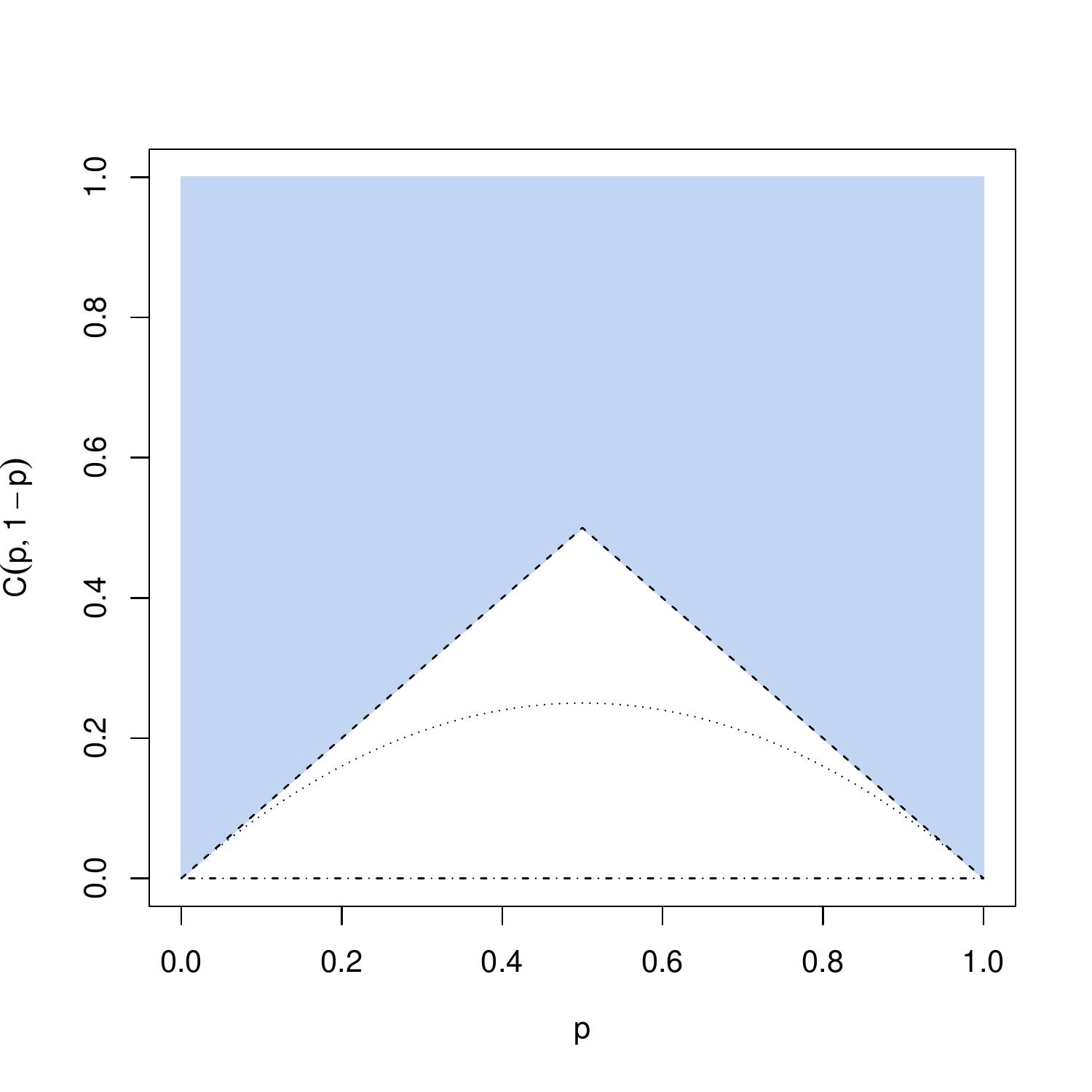}}
		\caption{Fr\'echet-Hoeffding bounds for the bivariate copula. 
        The allowed space is the triangle delimited by dashed and dash-dotted lines,
        corresponding to the upper bound $M(\vect{u})$ (maximal dependence) 
        and the lower bound $W(\vect{u})$ (maximal opposition), respectively.
		The product copula for independent variables is shown as a dotted line.
		}
		\label{fig:bornes_cop}
	\end{figure}
    
Copulas are not easy to visualize, first because they need to be represented as 3D plots of a surface in two dimensions 
and second because the bounds \eqref{eq:frechet}, within which $\cop(u,v)$ is constrained to live, 
compress the difference between two arbitrary copulas --- the situation is even worse in higher dimensions. 
Estimating copula densities, on the other hand, is even more difficult than estimating univariate densities, 
especially in the tails. 
We therefore propose to focus on the diagonal of the copula, $\cop(u,u)$ and the anti-diagonal, $\cop(u,1-u)$, 
that capture part of the information contained in the full copula, and can be represented as 1-dimensional plots. 
Furthermore, in order to emphasize the difference with the case of independent variables, 
it is expedient to consider the following quantities:
\begin{subequations}\label{eq:copdiff}
\begin{align}
	\label{eq:def_cop_diag}\frac{\cop(u,u)-u^2}{u\,(1\!-\!u)}&=\tUU(u)+\tLL(1-u)-1\\
	\label{eq:def_cop_adiag}\frac{\cop(u,1\!-\!u)-u\,(1\!-\!u)}{u\,(1\!-\!u)}&=1-\tUL(u)-\tLU(1-u),
\end{align}
\end{subequations}
where the normalization is chosen such that the tail correlations appear naturally. 
Another important reference point is the Gaussian copula $\cop[\text{G}](u,v)$, 
and we will consider below the normalized differences along the diagonal and the anti-diagonal:
\begin{equation}\label{eq:Delta}
	\Delta_{\scriptscriptstyle \text{d}}(u)=\frac{\cop(u,u)-\cop[\text{G}](u,u)}{u\,(1\!-\!u)}; \qquad \Delta_{\scriptscriptstyle \text{a}}(u)=\frac{\cop(u,1\!-\!u)-\cop[\text{G}](u,1\!-\!u)}{u\,(1\!-\!u)}.
\end{equation}

The incentive to focus on the diagonals of the copula is motivated in the first place by the above argument
that the whole copula is a particularly complex object that is hard to manipulate, estimate and visualize,
and that some means have to be found in order to investigate its subtleties while reducing the number of parameters. 
But it also has a profound meaning in terms of physical observables that can be empirically measured.
In particular, the diagonal quantity \eqref{eq:def_cop_diag} tends to the asymptotic tail dependence 
coefficients $\tUU_*$ when $u \to 1$ ($\tLL_*$ when $u \to 0$). 
Similarly, the anti-diagonal quantity \eqref{eq:def_cop_adiag} tends to $-\tUL_*$ as $u \to 1$ and to $-\tLU_*$ as $u \to 0$.
This still holds for the alternative definition in Eq.~\eqref{eq:Delta}:
 $\Delta_{\scriptscriptstyle \text{d,a}}(u)$ tend to the asymptotic tail dependence coefficients in the limits $u \to 1$ and $u \to 0$, 
owing to the fact that these coefficients are all zero in the Gaussian case.
As we show in Chapter~\ref{part:partIII}.\ref{chap:copulas_time}, 
the diagonal copula plays also an important role in the characterization of temporal dependences in time series:
observables such as recurrence times, sequences lengths, waiting times, etc.\ are fundamentally many-points objects,
and their statistics are characterized by the diagonal $n$-points copulas.

The center-point of the copula, $\cop(\tfrac12,\tfrac12)$, is particularly interesting: 
it is the probability that both variables are simultaneously below their respective median. 
For bivariate Gaussian variables, one can show that:
\be\label{eq:ell_cop_rho}
\cop[\text{G}](\tfrac12,\tfrac12) = \frac14 + \frac{1}{2 \pi} \arcsin \rho \equiv \frac14 \left(1 + \ros\right),
\ee
where $\ros$ is the sign correlation coefficient defined on page \pageref{eq:ros}. 
The trivial but remarkable fact is that the above expression holds for any elliptical models, 
that we define and discuss in the next section. 
This will provide a powerful test to discriminate between empirical data 
and a large class of copulas that have been proposed in the literature.

\paragraph{Generalized Blomqvist's beta}
We report for completeness a possible multivariate generalization of the bivariate Blomqvist's beta,
proposed in Ref.~\cite{schmid2007nonparametric}: 
\begin{equation}
	\beta_{\vect{X}}(\vect{u},\vect{v})=\frac{\cop[\vect{X}](\vect{u})-\Pi(\vect{u})+\overline{\cop[\vect{X}]{}}(\vect{v})-\overline\Pi(\vect{v})}{M(\vect{u})-\Pi(\vect{u})+\overline{M}(\vect{v})-\overline\Pi(\vect{v})},
\end{equation}
where $\overline{\cop[\vect{X}]{}}$ is the survival copula defined similarly to Eq.~\eqref{eq:def_cop_mult} as 
\begin{align}\label{eq:def_tcop_mult}
	\overline{\cop[\vect{X}]{}}(u_1,\ldots,u_N)&=\tcdf[\vect{X}]\left(\tqdf[1](1-u_1),\ldots, \tqdf[N](1-u_N)\right).
\end{align}
$\beta_{\vect{X}}$ has an immediate interpretation as a normalized distance between $\cop[\vect{X}]$ and $\Pi$:
it measures the effective global departure from independence normalized by the maximum departure from independence.
When $N=2$ and the copula is symmetric, it also nicely interpolates between several of the coefficients defined above, for example
\begin{align*}
    \vect{u}=\vect{v}=(\tfrac12,\tfrac12)   &\Longrightarrow    \beta_{\vect X}(\vect{u},\vect{v}) =4\cop[\vect{X}](\tfrac12,\tfrac12)-1    =\betaB_{\vect{X}}\\
    \vect{u}=(u,u),\vect{v}=(1,1)           &\Longrightarrow    \beta_{\vect X}(\vect{u},\vect{v}) =\frac{\cop[\vect{X}](u,u)-u^2}{u\,(1-u)}=\tUU(u)+\tLL(1-u)-1.
\end{align*}

\subsection{A word on Archimedean copulas}\label{sec:archCop}
In the universe of all possible copulas, a particular family has become increasingly popular in finance: that of 
``Archimedean copulas''. These copulas are defined as follows \cite{wuvaldez}:  
  \be
  \cop[\phi](u,v) \equiv  \phi^{-1} \left[\phi(u) + \phi(v)\right], \qquad u,v\in[0,1]
  \ee
  where $\phi(u): [0,1] \to [0,1]$ is a function such that $\phi(1)=0$ and $\phi^{-1}$ is decreasing and completely monotone. 
  For example, Frank copulas are such that $\phi_{\text{F}}(u) = \hbox{$\ln [\e^\theta - 1]$} - \hbox{$\ln [\e^{\theta u} - 1]$}$ where $\theta$ is a real parameter, or
  Gumbel copulas, such that $\phi_{\text{G}}(u) = (-\ln u)^\theta$, $0 < \theta \leq 1$. The asymptotic coefficient of tail dependence are all zero for
  Frank copulas, whereas $\tUU_*=2-2^\theta$ (and all other zero) for the Gumbel copulas. The case of general multivariate copulas is obtained 
  as natural generalization of the above definition.
  
  One can of course attempt to fit empirical copulas with a specific Archimedean copula. By playing enough with the function $\phi$, it is obvious that
  one will eventually reach a reasonable quality of fit. What we take issue with is the complete lack of intuitive interpretation or plausible 
  mechanisms to justify why the returns of two correlated assets should be described with Archimedean copulas. This is particularly clear after 
  recalling how two Archimedean random variables are generated: first, take two $\mathcal{U}[0,1]$ random variables $s,w$. 
  Set $w'=K^{-1}(w)$ with $K(x)=x-\phi(x)/\phi'(x)$. %
 %\footnote{\textcolor{red}{C'est curieux, c'est les memes equations que dans Optimal Trading with Linear Costs. Mais surement rien a voir.}}
  Now, set:
  \be
  u = \phi^{-1}\!\left[s\,\phi(w')\right]; \qquad v = \phi^{-1}\!\left[(1\!-\!s)\,\phi(w')\right],
  \ee
  and finally write $X_1= \qdf[1](u)$ and $X_2= \qdf[2](v)$ 
  to obtain the two Archimedean returns with the required marginals ${\mathcal P}_1$ and ${\mathcal P}_2$ \cite{genest1993statistical,wuvaldez}.
  Unless one finds a natural economic or micro-structural interpretation for such a labyrinthine construction, 
  we content that such models should be discarded {\it a priori}, for lack of plausibility.
  In the same spirit, one should really wonder why the financial industry has put so much faith in Gaussian copulas
  models to describe correlation between {\it positive} default times, that were then used to price CDO's and other credit derivatives. 
  We strongly advocate the need to build models bottom-up: mechanisms should come before any mathematical representation
  (see also \cite{mikosch2006copulas} and references therein for a critical view on the use of copulas).
 
%\textcolor{red}{Citer Genest, Remillard, MacKay, Fermanian, etc.}

%\begin{theorem}[Associativity] \cite{ling1964representation,stupvnanova2011associative}
%\end{theorem}
% proof in \cite{nelsen1998introduction}

\section{Assessing temporal dependences}

In the context of discrete time series, where $i$ is a time index and $j=i+\ell$ ($\ell>0$, say), 
the realization of the random variable $X_j$ occurs after that of $X_i$.
Nevertheless, before $X_i$ is realized, the {\it a priori} probability distribution of $X_j$ is not
\emph{conditional} but \emph{joint} to that of $X_j$, and it is possible to use all
the measures of dependence defined above in this context.

For example, $\rho_{\ij}$ is the coefficient of linear \emph{auto}-correlation at lag $\ell=j-i$:
it says in essence ``how much of $X_i$ is expected to be reproduced, $\ell$ time steps later, in the realization of $X_j$''.
When positive, it indicates persistence, while reversion is revealed by negative autocorrelation.

Blomqvist's beta (defined in Eq.~\eqref{eq:blomqvist_2d}, page~\pageref{eq:blomqvist_2d}) is somehow more intuitive, 
as it quantifies directly the difference in probability that the signs will persist or alternate.
 
When one wants to be more specific in the characterization of temporal dependences and study for example 
the persistence and reversion probabilities of one realization of arbitrary magnitude,
one needs to resort to the copula. 
Typically, persistence of extreme events ($p$-th quantile, where $p$ is close to 1) separated by a time lag $\ell$
is measured by the tail dependence coefficients  $\tUU_\ell(p)$ and $\tLL_\ell(p)$, 
while reversion of extreme events is measured by $\tUL_\ell(p)$ and $\tLU_\ell(p)$.
We will elaborate on the topic of copulas in time series in Chapter~\ref{chap:copulas_time} of Part~\ref{part:partIII}.

All these coefficients will typically be of decreasing amplitude as $\ell$ gets larger, 
reflecting the loss of memory as time goes on and randomness comes in.
There nevertheless exist systems where the $\ell$-dependence is not monotonous
but rather exhibits oscillations due to periodic activity 
(see for example the rising interest in electroencephalogram data analysis).

\section{Review of elliptical models}\label{sec:elliptical_model}

\subsection{General elliptical models}

An elliptical distribution is characterized by a location parameter $\vect\mu$,
a positive definite dispersion matrix $\Sigma$ and a function $\phi$ with positive support: 
let us denote this class by $E_N(\vect\mu,\Sigma,\phi)$.
Ellipticity of a random variable is a property of its characteristic function,
as explicited in the following definition.

\begin{definition}
  \begin{equation}
	\vect{X}\sim E_N(\vect\mu,\Sigma,\phi) 
    \quad\Longleftrightarrow\quad
    %\varphi_{\vect{X}-\vect\mu}(\vect{t})        =\phi(\vect{t}^{\top}\Sigma^{-1}\vect{t})
    \Esp{\Exp{\imath\vect{t}^{\top}(\vect{X}-\vect\mu)}}=\phi(\vect{t}^{\top}\Sigma^{-1}\vect{t})
  \end{equation}
  %where $\varphi_{\vect{Z}}(\vect{t})=\Esp{\e^{\imath\vect{t}^{\top}\vect{Z}}}$ is the characteristic function of $\vect{Z}$
\end{definition}
The quadratic form in the argument of $\phi$ means that the points of constant characteristic function 
draw an ellipsoid in the Fourier space, whence the name
(see \cite{cambanis1981theory} for the construction and properties of elliptically contoured multivariate distributions).
This definition in Fourier space is neither very convenient to build such variables,
nor comfortable for its interpretation. 
The following theorem decomposes the elliptical random vector in terms of a uniformly distributed $N$-dimensional sphere
and a random radial scale.

\begin{theorem}[\citet*{cambanis1981theory}]
If $\mat{A}$ is a $(N\times k)$ matrix with $k\leq N$, and $\vect{U}\sim\mathcal{U}(S^k)$, then
\begin{subequations}\label{eq:theorem_cambanis}
  \begin{equation}
    \vect{X}=\vect\mu+\sigma \mat{A}\vect{U}
    \quad\Longrightarrow\quad\exists \phi \text{ s.t.\ }
  	\vect{X}\sim E_N(\vect\mu,\mat{A}\mat{A}^{\top},\phi)
  \end{equation}
  The converse holds true and calls, for a positive definite $\Sigma$ of rank $k$, 
  \begin{equation}
  	\vect{X}\sim E_N(\vect\mu,\Sigma,\phi) %\quad\text{and}\quad \rank{\Sigma}=k
    \quad\Longrightarrow\quad \exists \sigma\text{ and } \mat{A} \text{ s.t.\ }\mat{A}\mat{A}^{\top}=\Sigma\quad\text{and}\quad
    \vect{X}=\vect\mu+\sigma \mat{A}\vect{U}
  \end{equation}
\end{subequations}
\end{theorem}

  The class of elliptical distributions comprises the Gaussian and Student distributions:
\begin{enumerate}%[1.]
	\item{The Gaussian distribution:} If the scaling variable has a chi-2 distribution with 
    $k=\rank{\Sigma}$ degrees of freedom, i.e.\ $\sigma\sim\chi^2(k)$ then the resulting $\vect{X}$ is normally distributed:
    \[
    \sigma\sim\chi^2(k)\quad\Longrightarrow\quad\sigma \mat{A}\vect{U}\sim\mathcal{N}(0,\Sigma)
    \]
	\item{The Student distribution:}  As is well known, the quotient of a Gaussian variable over a chi-2 variable is Student-distributed.
    Hence, 
    \[
    \sigma=\sqrt{\frac{\nu}{2}}\frac{\sigma_k}{\sigma_{\nu}}\quad\text{and}\quad\sigma_{\nu}\sim\chi^2(\nu) \quad\Longrightarrow\quad \sigma \mat{A}\vect{U}\sim t_{\nu}(0,\Sigma)
    \]
\end{enumerate}

As a consequence of the theorem above and the of the example 1, all
elliptical random variables $X_i$ can be simply generated by multiplying 
standardized Gaussian random variables $\epsilon_i$ with a common random 
(strictly positive) factor $\sigma$, drawn independently from the $\epsilon$'s 
from an arbitrary distribution: \cite{embrechts02correlation}
\begin{equation}\label{eq:repres_ellipt}
	\vect{X}=\vect\mu+\sigma\,\vect\epsilon
\end{equation}
where $\mu$ is a location parameter, 
and \mbox{$\vect\epsilon\sim\mathcal{N}(\mathbf{0},\Sigma)$}.

\begin{property}[stability]
Let $\vect{X}\sim E_N(\vect\mu_X,\Sigma,\phi_X)$ and $\vect{Y}\sim E_N(\vect\mu_Y,c\Sigma,\phi_Y)$ be independent random vectors.
Their sum is elliptical, too:
\[
a\vect{X}+b\vect{Y}\sim E_N(a\vect\mu_X+b\vect\mu_Y,\Sigma,\phi)
\quad\text{with}\quad
\phi(\vect{\cdot})=\phi_X(a^2\vect{\cdot})\phi_Y(b^2c\vect{\cdot})
\]
\end{property}

As an example, consider the difference of two independent and elliptical random vectors: $\vect{X}-\vect{X}'\sim E_N(\vect{0},\Sigma,\phi^2)$
Remark that the stability also holds whenever the random vectors $\vect{X}$ and $\vect{Y}$ are dependent only through their radial part $\sigma$, see \cite{hult2002multivariate}.

\begin{property}[bivariate marginals] 
Let $\vect{X}\sim E_N(\vect\mu,\Sigma,\phi)$. Then, for all pair $(i,j)$,
\[
    \vect{X_{\ij}}\sim E_2(\vect{\mu_{\ij}},\Sigma_{(\ij)},\phi), \quad\text{with}\quad\Sigma_{(\ij)}=\left(\begin{array}{cc}\Sigma_{\!i\!i}&\Sigma_{\ij}\\\Sigma_{\ij}&\Sigma_{\!j\!j}\end{array}\right)
\]
In words, the bivariate marginals of a multivariate elliptical random variable
are also elliptical, and inherit the parameters of the joint distribution.
\end{property}
This is clear from either representations \eqref{eq:theorem_cambanis} or \eqref{eq:repres_ellipt}.
This property has an important converse: 
if the bivariate marginals are not elliptical with the same radial characteristic function, 
then the joint probability density is not elliptical.
We will make use of this corollary in our empirical study of multivariate distribution for stock returns in Chapter~\ref{chap:IJTAF} of Part~\ref{part:partII}.

\subsection{Predicted dependence coefficients}
The measures of dependence defined in Section~\ref{sec:bi_meas_dep} can be explicitly computed for elliptical models.
Let $\vect{X}\sim E_N(\vect\mu,\Sigma,\phi)$ and $r_{\ij}={\Sigma_{\ij}}/{\sqrt{\Sigma_{i\!i}\Sigma_{j\!j}}}$.
Then 
\begin{subequations}
\begin{align}\label{eq:ell_rho1}
   \rob_{\ij}&=r_{\ij}\\\label{eq:coeffs_ell}
   \roa_{\ij}&=\frac{\fd[1]\cdot D(r_{\ij})-1}{\frac{\pi}{2}\fd[1]-1}\\
   \roc_{\ij}&=\frac{\fd[2]\cdot(1+2r_{\ij}^2)-1}{3 \fd[2]-1}\\\label{eq:ell_betaB}%\label{eq:C05_ell}
 \betaB_{\ij}=%&=\frac{2}{\pi}\arcsin r_{\ij}\\
\rod[0]_{\ij}&=\frac{2}{\pi}\arcsin r_{\ij}\\\label{eq:ell_tauK}
  \tauK_{\ij}&=\frac{2}{\pi}\arcsin r_{\ij}\qquad\text{proof in \cite{lindskog2001kendall}},
\end{align}
\end{subequations}
where $\fd=\esp{\sigma^{2d}}/\esp{\sigma^d}^2$ and $D(r)=\sqrt{1-r^2}+r\,\arcsin r$.
Some remarks regarding these equations are of importance:
\begin{itemize}
\item $\rob_{\ij}$, $\betaB_{\ij}$ and $\tauK_{\ij}$ do not depend on $\phi$ (i.e.\ on $\sigma$): 
they are invariant in the class of elliptical distributions 
with continuous marginals and given dispersion matrix $\Sigma$.
\item No such invariance exists for Spearman's rho $\rhoS$ (\cite{hult2002multivariate}) and for $\rod$ and $\roa$!
\item For conclusions \eqref{eq:ell_tauK} and \eqref{eq:ell_betaB}, $\vect{X}\sim E_N(\vect\mu,\Sigma,\phi)$ is a sufficient condition
but not a necessary one ! Every random vector $\vect{X}$ with elliptical \emph{copula} (see below) verifies \eqref{eq:ell_tauK} and \eqref{eq:ell_betaB}, 
and is hence said to follow a multivariate meta-elliptical distribution.
\item Relations \eqref{eq:ell_tauK} and \eqref{eq:ell_betaB} can be used to define alternative ways of measuring the correlation matrix elements $r_{\ij}$.
Indeed,
\[
	\rhoK\equiv\sin\left(\tfrac{\pi}{2}\tauK\right)\qquad\text{and}\qquad\rhoB\equiv\sin\left(\tfrac{\pi}{2}\betaB\right)
\]
provide very low-moments (though slightly negatively biased) estimators --- see \cite{lindskog2001kendall} for a study of $\rhoK$.
Said differently, it is always possible to define the effective correlation $\rhoB$ which, 
according to Eq.~\eqref{eq:ell_cop_rho} and the subsequent discussion on page~\pageref{eq:ell_cop_rho},
can also be written as 
\be\label{eq:def_rhoB}
    \rhoB=\cos\!\left(2\pi \cop(\tfrac12,\tfrac12)\right).
\ee
Now Eq.~\eqref{eq:ell_betaB} states together with Eq.~\eqref{eq:ell_rho1} that for all elliptical models, 
this effective correlation is just equal to the usual linear correlation $\rob$ !
This property provides a very convenient and simple testable prediction to check whether a given copula is 
compatible with an elliptical model or not, and will be at the heart of our empirical study in Chapters~\ref{chap:IJTAF} and~\ref{chap:multifact} of part~\ref{part:partII}.

\item $\fd[2]$ is related to the kurtosis of the $X$'s through the relation $\kappa=3(\fd[2]-1)$.
\end{itemize}

The calculation of the tail correlation coefficients depends on the specific form of the function $\phi$
or, said differently, on the distribution $\pdf[\sigma]$ of $\sigma$ in Eq.~\eqref{eq:repres_ellipt}, 
for which several choices are possible. We will focus in the following on two of them, 
corresponding to the Student model and the log-normal model.

\subsubsection*{The Gaussian ensemble} 

In the case of bivariate Gaussian variables of variance $\sigma_0^2$, i.e.\ when
\[
    \pdf[\sigma](\sigma)=\delta(\sigma-\sigma_0),
\]
all $\fd=1$ and all the coefficients of dependence depend on $\rho$ only ! 
For example, the quadratic correlation $\roc$ is given by:
\be
\roc = \rho^2;
\ee
whereas the correlation of absolute values $\roa$ is given by:
\be
\roa = \frac{D(\rho)-1}{\frac{\pi}{2}-1},
\ee
%The correlation of signs is given by \mbox{$\ros= \frac{2}{\pi} \arcsin \rho$}. 

For some other classes of distributions, the higher-order coefficients $\rod$ and $\road$ are explicit functions of the
coefficient of linear correlation. This is for example the case of Student variables (see Fig.~\ref{fig:dep}) 
and more generally for all elliptical distributions.

\subsubsection*{The Student ensemble} 

When the distribution of the square volatility is inverse Gamma, i.e.\ when
\[
    \pdf[\sigma^2](u)=\frac{1}{\Gamma(\frac{\nu}{2})}u^{-\frac{\nu}{2}-1}\,\e^{-\frac{1}{u}},
\]
the joint pdf of the returns turns out to have an explicit form \cite{demarta2005t,embrechts02correlation,malevergne2006extreme}:
\begin{equation}%\textstyle
   t_{\nu}(\vect{x})=\frac{1}{\sqrt{(\nu\pi)^N\det\Sigma}}\frac{\Gamma(\frac{\nu+N}{2})}
   {\Gamma(\frac{\nu}{2})}\left(1+\frac{\vect{x}^\dagger\Sigma^{-1}\vect{x}}{\nu}\right)^{-\frac{\nu+N}{2}}
\end{equation}
This is the multivariate Student distribution with $\nu$ degrees of freedom for $N$ random variables with dispersion matrix $\Sigma$.
Clearly, the marginal distribution of $t_{\nu}(\vect{x})$
is itself a Student distribution with a tail exponent equal to $\nu$, which is well known to describe satisfactorily the univariate pdf of high-frequency returns
(from a few minutes to a day or so), with an exponent in the range $[3,5]$ (see e.g.\ Refs.~\cite{bouchaud2003theory,cont2001empirical,fuentes2009universal,plerou1999scaling}).
The multivariate Student model is therefore a rather natural choice; its corresponding copula defines the Student copula. 
For $N=2$, it is entirely parameterized by $\nu$ and the correlation coefficient $\rho=\Sigma_{12}$ ($\Sigma_{11}=\Sigma_{22}=1$).

The moments of $\sigma$ are easily computed and lead to the following expressions for the coefficients $\fd$:
\begin{subequations}
\begin{align}
  \fd[1]&=\frac{2}{\nu-2}\left(\frac{\Gamma(\frac{\nu}{2})}{\Gamma(\frac{\nu-1}{2})}\right)^2&%\nu>2\\
  \fd[2]&=\frac{\nu-2}{\nu-4}%&\nu>4
\end{align}
\end{subequations}
when $\nu>2$ (resp. $\nu>4$).  Note that in the limit $\nu \to \infty$ at fixed $N$, the multivariate Student distribution boils down to a 
multivariate Gaussian. The shape of $\roa(\rho)$ and $\roc(\rho)$ for $\nu=5$ is given in Fig.~\ref{fig:dep}.

\begin{figure}[t!]
	\center
    \includegraphics[scale=0.45,trim=0 0 0 50,clip]{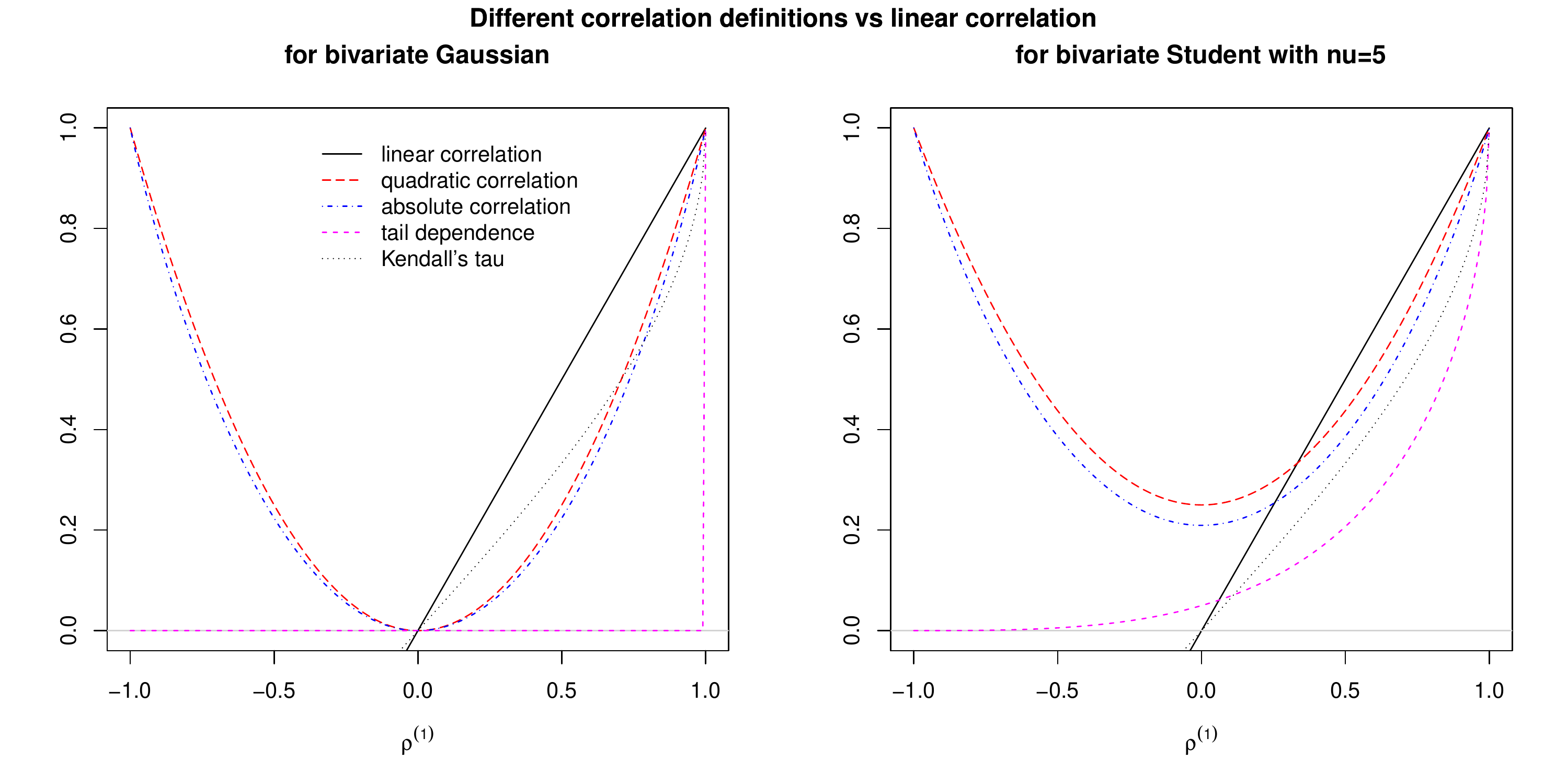}
	\caption{For all elliptical distributions, the different measures of dependence are functions of $\rho$.
	         Illustration for a pair with a bivariate Gaussian distribution (left), 
             and a bivariate Student distribution with $\nu=5$ (right).}
	\label{fig:dep}
\end{figure}

One can explicitly compute the coefficient of tail dependence for Student variables,
which only depends on $\nu$ and on the linear correlation coefficient $\rho$. 
By symmetry, one has $\tUU(p;\nu,\rho)=\tLL(p;\nu,\rho)=\tUL(p;\nu,-\rho)=\tLU(p;\nu,-\rho)$. 
When $p = 1 - \epsilon$ with $\epsilon \to 0$, the result is given by the following theorem:%
\footnote{
This result was simultaneously found by Manner and Segers \cite{manner2009tails} (in a somewhat more general context). 
We still sketch our proof because it follows a different route (uses the copula), 
and the final expression looks quite different, although of course numerically identical.
}

% one finds:
% \begin{equation}\label{eq:tau_puissance}
	% \tUU(p;\nu,\rho)=\tUU_*(\nu,\rho) + \beta(\nu,\rho)\cdot \epsilon^{\frac{2}{\nu}} + \mathcal{O}(\epsilon^{\frac{4}{\nu}}),
% \end{equation}
% where $\tUU_*(\nu,\rho)$ and $\beta(\nu,\rho)$ are coefficients given in \ref{apx:tau_puissance}, %Appendix~\ref{apx:tau_puissance}, 
% see also Figs.~\ref{fig:dep}~and~\ref{fig:tau_stud}. 

%\section{Proof of Eq.~\eqref{eq:tau_puissance}, page~\pageref{eq:tau_puissance}}\label{apx:tau_puissance}
\begin{theorem}[\cite{chicheportiche2012joint}]
Let $(X_1,X_2)$ follow a bivariate student distribution with $\nu$ degrees of freedom and correlation $\rho$,
and denote by $T_{\nu}(x)$ the univariate Student cdf with $\nu$ degrees of freedom.
The pre-asymptotic behavior of its tail dependence when $p\to 1$ is approximated by the following expansion in (rational) powers of $(1-p)$:
\begin{equation}\label{eq:tau_puissance}
	\tUU(p;\nu,\rho)=\tUU_*(\nu,\rho) + \beta(\nu,\rho)\cdot \epsilon^{\frac{2}{\nu}} + \mathcal{O}(\epsilon^{\frac{4}{\nu}}),
\end{equation}
with
\begin{align*}
	\tUU_*(\nu,\rho)&=2-2\:T_{\nu+1}(k(1))\\
	 \beta(\nu,\rho)&=\frac{\nu^{\frac{2}{\nu}+1}}{\frac{2}{\nu}+1}\,k(1)\,t_{\nu+1}\big(k(1)\big)\,L_{\nu}^{-\frac{2}{\nu}}\\\nonumber
	            %   &=\left(\frac{\Gamma(\frac{\nu}{2})\sqrt{\pi}}{\Gamma(\frac{\nu+1}{2})}\right)^{\frac{2}{\nu}}
		        % \frac{\nu^{\frac{2}{\nu}}}{\frac{2}{\nu}+1}
		        % \frac{\sqrt{\nu+1}\sqrt{1-\rho}}{\sqrt{1+\rho}}\cdot{}t_{\nu+1}\left(\frac{\sqrt{\nu+1}\sqrt{1-\rho}}{\sqrt{1+\rho}}\right)
\end{align*}
where $k(1)={\sqrt{(\nu+1)(1-\rho)}}/{\sqrt{1+\rho}}$
and   $L_{\nu}=\pi^{-\frac{1}{2}}\nu^{\frac{\nu}{2}}{\Gamma(\frac{\nu+1}{2})}/{\Gamma(\frac{\nu}{2})}$.
\end{theorem}

The following Lemma, is needed for the proof of the Theorem:
\begin{lemma}
Let $(X_1,X_2)$ and $T_{\nu}(x)$ be like in the Theorem above,
and $x_p\equiv{}T_{\nu}^{-1}(p)$. Define
%\begin{subequations}\label{eq:K_def1}
\begin{align*}
	K   &=\sqrt{\frac{\nu+1}{1-\rho^2}}\,\frac{X_1-\rho{}x_p}{\sqrt{\nu+X_2^2}}\\
	k(p)&=\frac{\sqrt{\nu+1}\sqrt{1-\rho}}{\sqrt{1+\rho}}\left[\nu\cdot{}x_p^{-2}+1\right]^{-\frac{1}{2}}
	     =\frac{k(1)}{\sqrt{1+\frac{\nu}{x_p^{2}}}}.
\end{align*}
%\end{subequations}
Then, 
\begin{equation}
	\pr{K\leq k(p)\mid X_2=x_p}=T_{\nu+1}\big(k(p)\big)
\end{equation}
\end{lemma}
\begin{proof}[Proof of the Lemma]
The proof proceeds straightforwardly by showing that $t_{\nu+1}(k)=\pdf[X_1|X_2](x_1|x_2)\frac{\partial{}x_1}{\partial{}k}$ when $x_2=x_p$.
The particular result for the limit case $p=1$ is stated in \cite{embrechts02correlation}.
\end{proof}

\begin{proof}[Proof of the Theorem]
Recall from Eq.~\eqref{eq:tau_C} that
\begin{equation}\label{eq:app1}
	\tUU(p)=2-\frac{\cop(p,p)-\cop(1,1)}{p-1}=2-\frac{1}{1-p}\int_p^1\frac{\d{\cop(p,p)}}{\d{p}}\,\d{p}
\end{equation}
One easily shows that
\begin{align*}
	\d{\cop(p,p)}&=2\:\pr{X_1\leq{}x_p\mid{}X_2=x_p}\,\d{p}=2\:T_{\nu+1}\big(k(p)\big)\d{p}%\\
	%\text{where}\quad k(p)&=\sqrt{\nu+1}\sqrt{\frac{1-\rho}{1+\rho}}\left(\nu\cdot x_p^{-2}+1\right)^{-\frac{1}{2}}
\end{align*}
where the second equality holds in virtue of the aforementioned lemma.\newline
Now, for $p$ close to $1$, $k(p)=k(1)\left(1-\frac{1}{2}\frac{\nu}{x_p^2}\right)+\mathcal{O}(x_p^{-4})$, and
\begin{equation*}
	\frac{\d{\cop(p,p)}}{\d{p}}=2\,T_{\nu+1}\big(k(p)\big)\approx{}2\,T_{\nu+1}\big(k(1)\big)-2\big(k(1)-k(p)\big)\cdot{}t_{\nu+1}\big(k(1)\big)
\end{equation*}
But since the Student distribution behaves as a power-law precisely in the region $p\lesssim 1$, 
we write $t_{\nu}(x)\approx L_{\nu}/x^{\nu+1}$ and immediately get 
$
	1-p=x_p^{-\nu}{L_{\nu}}/{\nu}%\qquad\text{with}\quad 
   %L_{\nu}=\frac{\Gamma(\frac{\nu+1}{2})}{\Gamma(\frac{\nu}{2})\sqrt{\pi}}\nu^{\frac{\nu}{2}}
$.
The result follows by collecting all the terms and performing the integration in \eqref{eq:app1}.
\end{proof}

The notable features of the above results are:
\begin{itemize}
\item For large $\nu$, the exponent is almost zero and the correction term is of order $\mathcal{O}(1)$, and the expansion ceases to hold.
This is particularly true for the Gaussian distribution ($\nu=\infty$) for which the behavior is radically different at the limit $p\to 1$
(where there's strictly no tail correlation) and at $p<1$ where a dependence subsists, see Fig.~\ref{fig:tau_stud1}.

\item The asymptotic tail dependence $\tUU_*(\nu,\rho)$ is strictly positive for all $\rho > -1$ and finite $\nu$, 
and tends to zero in the Gaussian limit $\nu \to \infty$ (see Fig.~\ref{fig:tau_stud1}). 
The intuitive interpretation is quite clear: large events are caused by large occasional bursts in volatility. 
Since the volatility is common to all assets, the fact that one return is positive extremely large is enough to infer that the volatility is itself large.
Therefore that there is a non-zero probability $\tUU_*(\nu,\rho)$ that another asset also has a large positive return 
(except if $\rho=-1$ since in that case the return can only be large and negative!). 
It is useful to note that  the value of $\tUU_*(\nu,\rho)$ does not depend on the explicit shape of $\pdf[\sigma]$ 
provided the asymptotic tail of $\pdf[\sigma]$ decays as $L(\sigma)/\sigma^{1+\nu}$, where $L(\sigma)$ is a slow function. 

\item The coefficient $\beta(\nu,\rho)$ is also positive, indicating that estimates of $\tUU_*(\nu,\rho)$ 
based on measures of $\tUU(p;\nu,\rho)$ at finite quantiles (e.g.\ $p=0.99$) are biased upwards. 
Note that the correction term is rapidly large because $\epsilon$ is raised to the power $2/\nu$. 
For example, when $\nu=4$ and $\rho=0.3$, $\beta\approx 0.263$ and one expects a first-order correction $0.026$ for $p=0.99$.
This is illustrated in Figs.~\ref{fig:tau_stud1},~\ref{fig:tau_stud2}.
The form of the correction term (in $\epsilon^{2/\nu}$) is again valid as soon as $\pdf[\sigma]$ decays asymptotically as $L(\sigma)/\sigma^{1+\nu}$.

\item Not only is the correction large, but the accuracy of the first order expansion is very bad, 
since the ratio of the neglected term to the first correction is itself $\sim \epsilon^{2/\nu}$. 
The region where the first term is accurate is therefore exponentially small in $\nu$ --- see Fig.~\ref{fig:tau_stud2}. 
\end{itemize}

\begin{figure}[t!]
\center
    \subfigure[$\tUU$ vs $\nu$ at several thresholds $p$ for bivariate Student variables with $\rho=0.3$. Note that $\tUU_* \to 0$ when $\nu \to \infty$, but rapidly 
    grows when $p \neq 1$.]{\label{fig:tau_stud1}\includegraphics[scale=0.45]{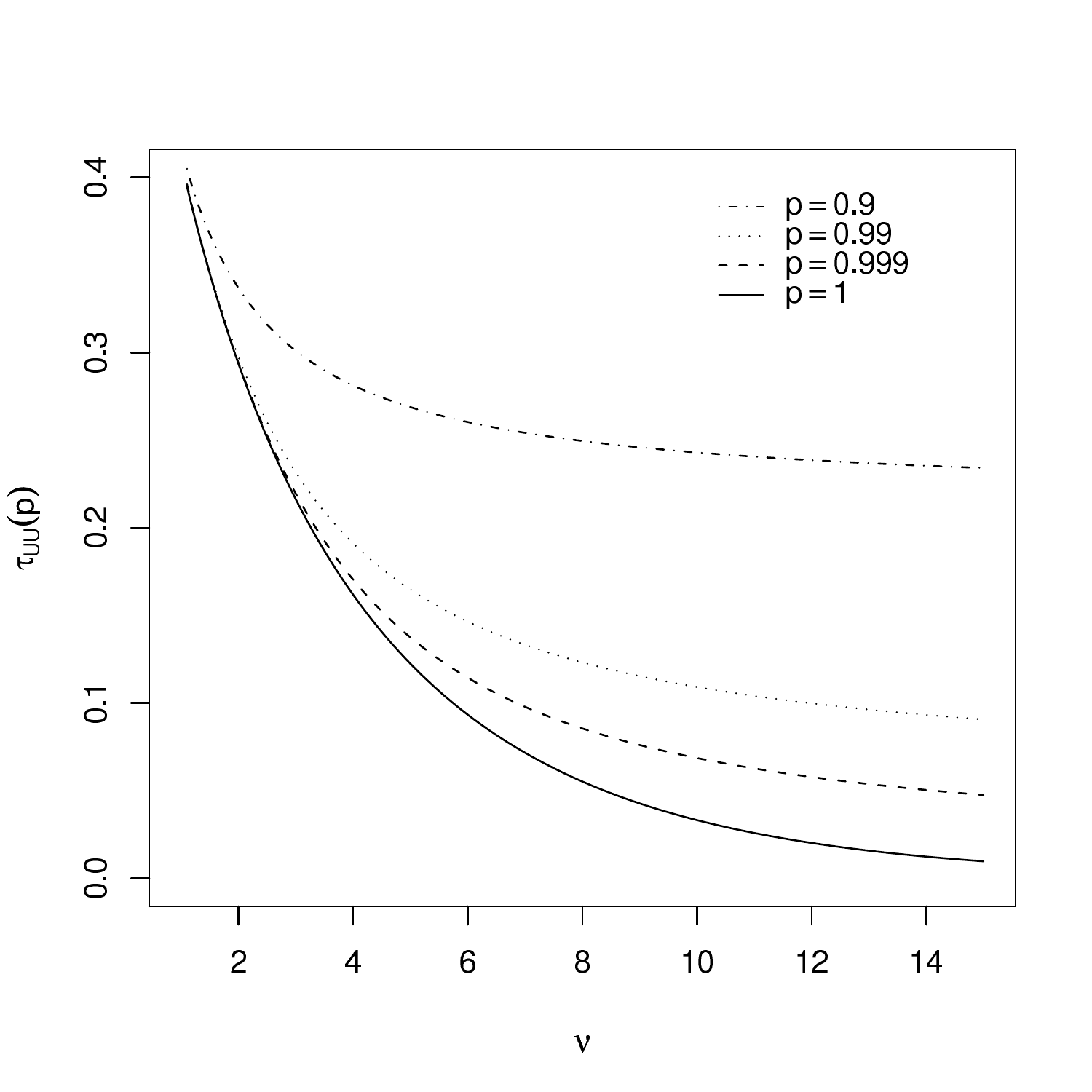}}\hfill
    \subfigure[$\tUU$ vs $p$ for bivariate Student variables with $\rho=0.3$ and $\nu=5$: exact curve (plain), first order power-law expansion (dashed) and simulated series with different lengths $T$ 
    (symbols)]{\label{fig:tau_stud2}\includegraphics[scale=0.45,trim=0 0 420 0,clip]{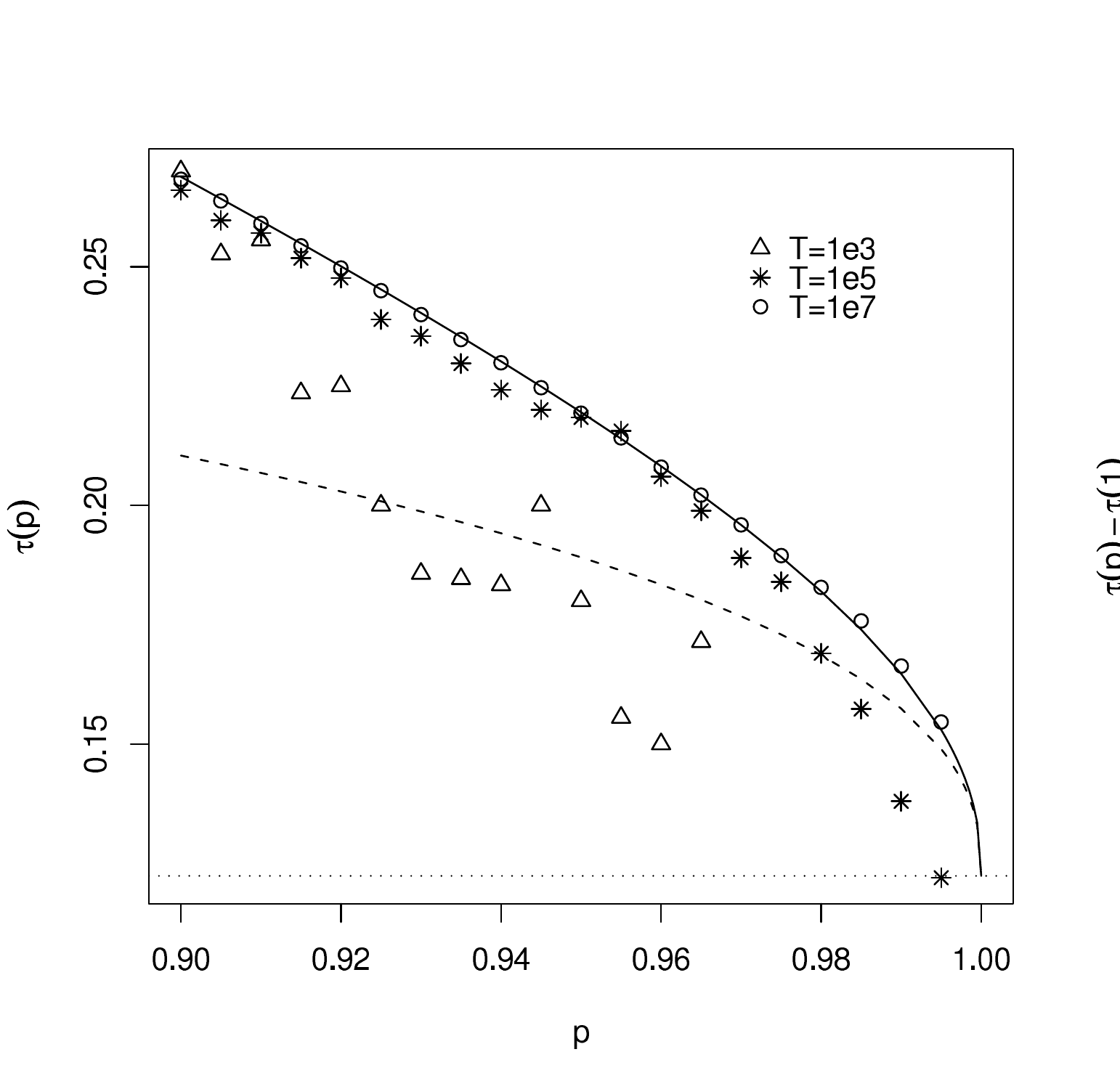}}
    \caption{}	
    \label{fig:tau_stud}
\end{figure}

Finally, we plot in Fig.~\ref{fig:dev_gauss_stud} the rescaled difference between the Student copula and the Gaussian copula 
both on the diagonal and on the anti-diagonal,
for several values of the linear correlation coefficient $\rho$ and for $\nu=5$. 
One notices that the difference is zero for $p=\tfrac12$, as expected from the expression of $\cop^*(\rho)$ for general elliptical models. 
Away from $p=\tfrac12$ on the diagonal, the rescaled difference has a positive convexity and non-zero limits when $p \to 0$ and $p \to 1$,
corresponding to $\tLL_*$ and $\tUU_*$.

\begin{figure}[ht]
\center
    \includegraphics[scale=0.45]{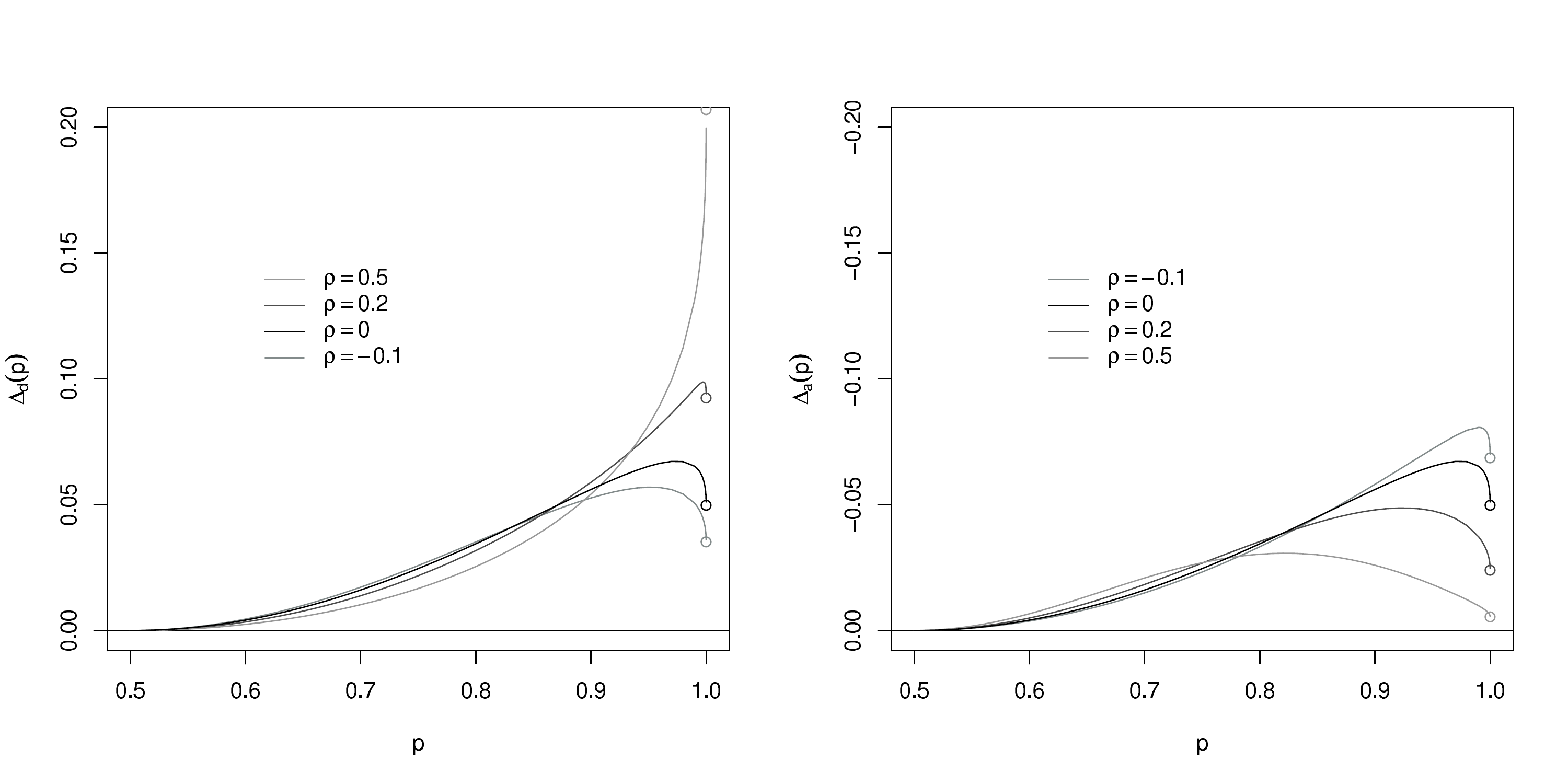}
    \caption{Diagonal and anti-diagonal of the Student copula: 
    the quantities $\Delta_{\scriptscriptstyle \text{d}}(p)$ and $\Delta_{\scriptscriptstyle \text{a}}(p)$ defined in Eq.~\eqref{eq:Delta} are plotted versus $p$ for several values of $\rho$ and fixed $\nu=5$.
    (The curves are identical in the range $p\in[0,0.5]$ due to the symmetry $p\leftrightarrow 1\!-\!p$).}	
    \label{fig:dev_gauss_stud}
\end{figure}

\subsubsection*{The log-normal ensemble} 

If we now choose $\sigma=\sigma_0\, \e^\xi$ with $\xi \sim \mathcal{N}(0,s)$, 
%(as suggested by advocates of multi-fractal models \cite{calvetfisher,lux,bacrymuzy}), % pas pertinent, ici on est en cross-sectionnel, pas en temporel
the resulting multivariate model defines the log-normal ensemble and the log-normal copula. 
The factors $\fd$ are immediately found to be: $\fd=\e^{d^2s^2}$ with no restrictions on $d$.
The Gaussian case now corresponds to the limit $s \equiv 0$.

Although the inverse Gamma and the log-normal distributions are very different, it is well known that the tail region of the log-normal is in practice
very hard to distinguish from a power-law. In fact, one may rewrite the pdf of the log-normal distribution as:
\begin{equation}
	\widehat \sigma^{-1} \e^{-\frac{\ln^2 \widehat \sigma}{2 s^2}} = \widehat \sigma^{-(1+\nu_{\text{eff}})} 
    \quad\text{with}\quad
    \nu_{\text{eff}} = \frac{\ln \widehat \sigma}{2 s^2},
\end{equation}
where we have introduced $\widehat \sigma = \sigma/\sigma_0$. The above equation means that for large $\widehat \sigma$ a log-normal is like a power-law with a 
slowly varying exponent. If the tail region corresponds to $\widehat \sigma \in [4,8]$ (say), the effective exponent $\nu$ lies in the range $[0.7/s^2,1.05/s^2]$.
Another way to obtain an approximate dictionary between the Student model and the log-normal model is to fix $\nu_{\text{eff}}(s)$ such that the coefficient
$\fd[2]$ (say) of the two models coincide. Numerically, this leads to $\nu_{\text{eff}}(s) \approx 2 + 0.5/s^2$, leading to $s \approx 0.4$ for $\nu=5$. 
In any case, our main point here is that from an empirical point of view, one does not expect striking differences between a Student model and a log-normal model --- even 
though from a strict mathematical point of view, the models are very different. In particular, the asymptotic tail dependence coefficients $\tUU_*$ or $\tLL_*$ are always zero
for the log-normal model (but $\tLL(p)$ or $\tUU(p)$ converge exceedingly slowly towards zero as $p \to 1$).

\subsection{Pseudo-elliptical generalization}\label{ssec:pseudo-ell}

In the previous elliptical description of the returns, all stocks are subject to 
the exact same stochastic volatility, what leads to non-linear dependences like tail effects and residual higher-order correlations even for $\rho=0$ (see Fig.~\ref{fig:dep}).
In order to be able to fine-tune somewhat this dependence due to the common volatility, a simple generalization is to let each stock be 
influenced by an idiosyncratic volatility, thus allowing for a more subtle structure of dependence. 
More specifically, we write\footnote{{
In vectorial form, we collect the individual (yet dependent) stochastic volatilities $\sigma_i$ 
on the diagonal of a matrix $\mat{D}$, 
and $\vect{X}=\mat{D}\vect{\epsilon}$ can be decomposed as $\mat{X}=R\mat{D}\mat{A}\vect{U}$
where $\vect{U}$ is a random vector uniformly distributed on the unit hypersphere, 
the radial component $R$ is a chi-2 random variable independent of $U$ with $\rank(\mat{A})$ degrees of freedom, 
and $\mat{A}$ is a matrix with appropriate dimensions such that $\mat{A}\mat{A}^{\dagger}=\Sigma$.
This description can be contrasted with the one proposed in Ref.~\cite{kring2009multi} under the term ``Multi-tail Generalized Elliptical Distribution'':
$\vect{X}=R(\vect{u})\mat{A}\vect{U}$ with $R$ now depending on the unit vector $\vect{u}=\mat{A}\vect{U}/||\mat{A}\vect{U}||$.
In other words, the description in $\eqref{eq:pseudo_ell_coeff}$ provides a different radial amplitude for each component, whereas \cite{kring2009multi}
characterizes a direction-dependent radial part identical for every component.
The latter allows for a richer phenomenology than the former, but lacks financial intuition.}}: 
\begin{equation}\label{eq:pseudo_ell_coeff}
	X_i=\sigma_i\cdot\epsilon_i\qquad{}i=1\ldots{}N
\end{equation}
where the Gaussian residuals $\epsilon_i$ have the same joint distribution as before
(in particular with linear correlations $\esp{\epsilon_i\epsilon_j}=r_{\ij}$), and are independent of the $\sigma_i$, 
but we now generalize the definition of the ratios $\fd$.
The random vector $\vect\sigma$ has a non trivial dependence structure which is partly described by the mixed $d$-moments
\begin{equation}\label{eq:fd_general}
\fd_{\ij}=\frac{\esp{\sigma_i^d\sigma_j^d}}{\esp{\sigma_i^d}\,\esp{\sigma_j^d}}.
\end{equation}
When the $\sigma_i$'s are dependent but identically distributed, 
the diagonal elements $\fd_{i\!i}$ do not depend on $i$.
Within this setting, the generalization of coefficients \eqref{eq:coeffs_ell} can be straightforwardly calculated:
\begin{subequations}
\begin{align}
\rob_{\ij}&=r_{\ij}\frac{\fd[1]_{\ij}}{\fd[1]_{i\!i}}\\
\roa_{\ij}&=\frac{\fd[1]_{\ij} \cdot D(r_{\ij})-1}{\tfrac{\pi}{2}\fd[1]_{i\!i}-1}\\
\roc_{\ij}&=\frac{\fd[2]_{\ij} \cdot(1+2r_{\ij}^2)-1}{3\,\fd[2]_{i\!i}-1}
\end{align}
\end{subequations}
Importantly, Kendall's tau and Blomqvist's beta remain invariant in this class of Pseudo-elliptical distributions !\label{page:invariantbeta}

As an explicit example, we consider the natural generalization of the log-normal 
model and write $\sigma_i = \sigma_{0i}\, \e^{\xi_i}$, 
with\footnote{A further generalization that allows for stock dependent ``vol of vol'' $s_i$ is also possible.} 
$\xi_i\stackrel{\text{\tiny id}}{\sim}\mathcal{N}(0,s)$, and some correlation structure of the $\xi$'s: $\esp{\xi_i \xi_j}=s^2 c_{\ij}$. 
One then finds:
\begin{equation}\label{eq:fd_lognorm}
\fd_{\ij}\equiv\fd(c_{\ij})=\e^{d^2 s^2 c_{\ij}}.
\end{equation}

Using the generic notation $c$ for the correlation of the log-volatilities, 
and $r$ for the correlation of the residuals (now different from $\rho$) we find:
\begin{subequations}\label{eq:coeffs_pseudo_ell}
\begin{align}
	\rho(r,c)&=\frac{\fd[1](c)}{\fd[1](1)}\cdot{}r\\
	\roc(r,c)&=\frac{\fd[2](c)\,(1+2r^2)-1}{3\fd[2](1)-1}\\
	\roa(r,c)&=\frac{\fd[1](c)\,D(r)-1}{\frac{\pi}{2}\fd[1](1)-1}%\\
	%\cop^*(r,c)&=\frac14+\frac{1}{2\pi}\arcsin r\quad\forall{}c
\end{align}
\end{subequations}
%These formulas are straightforwardly generalized for arbitrary description of the volatilities, using the coefficients $\fd$ defined in
%\eqref{eq:fd_general} instead of the explicit expressions given by \eqref{eq:fd_lognorm}.
When $c$ is fixed (e.g.\ for the elliptical case $c=1$), $\rho$ and $r$ are proportional, and 
all measures of non-linear dependences can be expressed as a function of $\rho$. But this ceases to be true as soon
as there exists some non trivial structure $c$ in the volatilities. In that case, $c$ and $r$ are ``hidden'' underlying variables, that
can only be reconstructed from the knowledge of $\rho$, $\roc$, $\cop(\tfrac{1}{2},\tfrac{1}{2})$, assuming of course that the model is accurate.

Notice that the result on $\cop(\frac12,\frac12)$ is totally independent of the structure of the volatilities\footnote{ 
This property holds even when $\sigma_i$ depends on the sign of $\epsilon_i$, which might be useful to model the leverage
effect that leads to some asymmetry between positive and negative tails, as the data suggests.}.
Indeed, what is relevant for the copula at the central point is not the amplitude of the returns, 
but rather the number of $+/-$ events, 
which is unaffected by any multiplicative scale factor as long as the median of the univariate marginals is nil. 
An important consequence of this result is that for all elliptical or pseudo-elliptical model,
$\rho=0$ implies that \mbox{$\cop^*(\rho=0)=\frac14$}.

\section{Conclusion}
We have suggested several ideas of methodological interest to efficiently visualize and  compare different copulas. 
We recommend in particular the rescaled difference with the Gaussian copula along the diagonal $\Delta_{\scriptscriptstyle \text{d,a}}(p)$ and 
the central value of the copula $\cop(\tfrac12,\tfrac12)$ as strongly discriminating observables. 
We have studied the dependence of these quantities, as well as other non-linear correlation coefficients like 
the higher-order correlations $\road$, 
with the linear correlation coefficient $\rho$.

The case of elliptical copulas is emphasized, with the lognormal and Student models as examples.
We provide an original characterization of the pre-asymptotic tail dependence coefficient when $p\to 1$
as a development in rational powers of $(1\!-\!p)$.

An explicit prediction of elliptical models is empirically tested on multivariate financial data 
in Chapter~\ref{chap:IJTAF} of part~\ref{part:partII}.

\chapter{Goodness-of-fit testing}\label{chap:GOF}
\minitoc

The problem of testing whether a null-hypothesis theoretical probability distribution
is compatible with the empirical probability distribution of a sample of observations
is known as goodness-of-fit (GoF)\nomenclature{GoF}{Goodness of fit} testing 
and is ubiquitous in all fields of science and engineering.
Goodness-of-Fit tests are designed to assess quantitatively whether a sample of $N$ observations
can statistically be seen as a collection of $N$ realizations of a given probability law, 
or whether two such samples are drawn from the same hypothetical distribution.

The best known theoretical result is due to Kolmogorov and Smirnov (KS)\nomenclature{KS}{Kolmogorov-Smirnov} \cite{kolmogorov1933sulla,smirnov1948table}, 
and has led to the eponymous statistical test for an \emph{univariate} sample of \emph{independent} draws.
The major strength of this test lies in the fact that the asymptotic distribution of its test statistic 
is completely independent of the null-hypothesis cdf.

Several specific extensions have been studied (and/or are still under scrutiny), including: 
different choices of distance measures \cite{darling1957kolmogorov}, 
multivariate samples \cite{fasano1987multidimensional,cabana1994goodness,cabana1997transformed,Fermanian2005119}, 
investigation of different parts of the distribution domain \cite{anderson1952asymptotic,deheuvels2005weighted,noe1968calculation,chicheportiche2012weighted},
dependence in the successive draws \cite{chicheportiche2011goodness},  etc. 

This class of problems has a particular appeal for physicists since the works of Doob \cite{doob1949heuristic} and Khmaladze \cite{khmaladze1982martingale},
who showed how GoF testing is related to stochastic processes. 
%Finding the law of a test often amounts to treating a Fokker-Planck problem,
%which in turn maps into a Schr\"odinger equation for a particle in a certain potential confined by walls. 
Finding the law of a test amounts to computing a survival probability in a dissipative system. 
In a Markovian setting, this is often achieved by treating a Fokker-Planck problem,
which in turn maps into a Schr\"odinger equation for a particle in a certain potential confined by walls.

\subsection{Empirical cumulative distribution and its fluctuations}\label{sec:GoFintro}

Let $\vect{X}$ be a random vector of $N$ independent and identically distributed variables, 
with marginal cumulative distribution function (cdf) $F$. 
One realization of $\vect{X}$ consists of a time series $\{x_1,\ldots,x_n, \ldots, x_N\}$ that exhibits no persistence (see 
Sect.~\ref{sec:GoF} when some non trivial dependence is present).
The empirical cumulative distribution function 
\begin{equation}\label{eq:F_N}
    F_N(x)=\frac{1}{N}\sum_{n=1}^N\1{X_n\leq x}
\end{equation}
converges to the true CDF $F$ as the sample size $N$ tends to infinity.
For finite $N$, the expected value and fluctuations of $F_N(x)$ are
\begin{align*}
	\esp{F_N(x)}&=F(x),\\
    \mathrm{Cov}(F_N(x),F_N(x'))&=\frac{1}{N}\left[F(\min(x,x'))-F(x)F(x')\right].
\end{align*}
The rescaled empirical CDF
\begin{equation}\label{eq:Y_N}
    Y_N(u)=\sqrt{N}\,\left[F_N(F^{-1}(u))-u\right]
\end{equation}
measures, for a given $u\in[0,1]$, the difference between the empirically determined cdf 
of the $X$'s and the theoretical one, evaluated at the $u$-th quantile.
It does not shrink to zero as $N\to\infty$, and is therefore the quantity on which any statistic for GoF testing is built.

\subsubsection*{Limit properties}
One now defines the process $Y(u)$ as the limit of $Y_N(u)$ when $N \to \infty$.
According to the Central Limit Theorem (CLT)\nomenclature{CLT}{Central limit theorem}, 
it is Gaussian and its covariance function is given by:
\begin{equation}\label{eq:Itheo}
	{I}(u,v)=\min(u,v)-uv,
\end{equation}
which characterizes the so-called Brownian bridge, i.e.\  a Brownian motion $Y(u)$  such that \mbox{$Y(u\!=\!0)=Y(u\!=\!1)=0$}. 

Interestingly, $F$ does not appear in Eq.~(\ref{eq:Itheo}) anymore, 
so the law of any functional of the limit process $Y$ is independent of the law of the underlying finite size sample.
This property is important for the design of \emph{universal} GoF tests.

\subsubsection*{Norms over processes}
In order to measure a limit distance between distributions, a norm $||.||$ over the space of continuous bridges needs to be chosen. 
Typical such norms are the norm-2 (or `Cram\'er-von~Mises' distance)
\[
	||Y||_2=\int_0^1Y(u)^2\d{u},
\]
as the bridge is always integrable, or the norm-sup (also called the Kolmogorov distance)
\[
	||Y||_\infty=\sup_{u\in[0,1]}|Y(u)|,
\]
as the bridge always reaches an extremal value.

Unfortunately, both these norms mechanically overweight the core values $u \approx 1/2$ and disfavor the tails $u \approx 0,1$:
since the variance of $Y(u)$ is zero at both extremes and maximal in the central value, the major contribution to $||Y||$ indeed
comes from the central region. To alleviate this effect, in particular when the GoF test is intended to investigate 
a specific region of the domain, it is preferable to introduce additional weights and study $||Y\sqrt{\psi}||$ rather than $||Y||$ itself.
Anderson and Darling show in Ref.~\cite{anderson1952asymptotic} that the solution to the problem with the Cram\'er-von~Mises norm and arbitrary weights $\psi$
is obtained by spectral decomposition of the covariance kernel, and use of Mercer's theorem. 
They design an eponymous test \cite{darling1957kolmogorov} with the specific choice of weights $\psi(u)=1/{I(u,u)}$
equal to the inverse variance, which equi-weights all quantiles of the distribution to be tested.
We analyze here the case of the same weights, but with the Kolmogorov distance.

\newcommand\FF[3]{%
  {\vphantom{#2}}#1#2#3%
}

\section{Weighted {K}olmogorov-{S}mirnov tests}\label{sect:weightedKS}
%%%%%%%%%%%%%%%%%%%%%%%%%%%%%%%%%%%%%%%%%%%%%%%%%%%%%%%%%%%%%%%%%%%%%%%%
%\subsection{Introduction and motivation}

So again $Y(u)$ is a Brownian bridge, i.e.\ a centered Gaussian process on $u\in[0,1]$ with covariance function
${I}(u,v)$ given in Eq.~(\ref{eq:Itheo}).
In particular, $Y(0)=Y(1)=0$ with probability equal to 1, no matter how distant $F$ is from the sample cdf around the core values.
In order to put more emphasis on specific regions of the domain, we weight the Brownian bridge as follows: 
for given $a\in]0,1[$ and $b\in [a,1[$, we define
\begin{equation}\label{eq:weightedY}
	\tilde{Y}(u)=Y(u)\cdot\left\{\begin{array}{cl}\sqrt{\psi(u)}&,\,a\leq u\leq b\\0&,\,\text{otherwise.}\end{array}\right.
\end{equation}
We will characterize the law of the supremum $K(a,b)\equiv\sup_{u\in[a,b]}\left|\tilde{Y}(u)\right|$:
\[
	\mathcal{P}_{\!{\scriptscriptstyle <}}(k|a,b)\equiv\pr{K(a,b)\leq k}
                                                =\Pr{|\tilde{Y}(u)|\leq k, \forall u\in[a,b]}.
\]

\subsection{The equi-weighted {B}rownian bridge: {K}olmogorov-{S}mirnov}

In the case of a constant weight, corresponding to the classical KS test, the probability 
\mbox{$\mathcal{P}_{\!{\scriptscriptstyle <}}(k;0,1)$} is well defined and has the well known KS form \cite{kolmogorov1933sulla}:
\begin{equation}\label{eq:KSprob}
    \mathcal{P}_{\!{\scriptscriptstyle <}}(k;0,1) = 1 - 2 \sum_{n=1}^{\infty} (-1)^{n-1} \e^{-2 n^2 k^2},
\end{equation}
which, as expected, grows from $0$ to $1$ as $k$ increases, see Fig.~\ref{fig:KS} on page \pageref{fig:KS}. 
The value $k^*$ such that this probability is $95 \%$ is $k^* \approx 1.358$ \cite{smirnov1948table}. 
This can be interpreted as follows: if, for a data set of size $N$, 
the maximum value of $|Y_N(u)|$ is larger than $\approx 1.358$, 
then the hypothesis that the proposed distribution is a ``good fit'' can be rejected with $95 \%$ confidence.

\subsubsection*{Diffusion in a cage with fixed walls}

The Brownian bridge $Y$ is nothing else than a Brownian motion with imposed terminal condition, and can be written as
$Y(u)=X(u)-u\,X(1)$ where $X$ is a Brownian motion.
The test law $\mathcal{P}_{\!{\scriptscriptstyle <}}(k|a,b)=\Pr{-k\leq Y(u)\leq k, \forall u\in[a,b]}$ 
is in fact the survival probability of $Y$ in a cage with absorbing walls, 
and can be determined by counting the number of Brownian paths $X(u)$ that
go from 0 to 0 without ever hitting the barriers.
More precisely, the survival probability of the Brownian bridge in the stripe $[-k,k]$ can be computed as $f_1(0;k)/f_1(0;\infty)$,
where $f_u(y;k)$ is the transition kernel of the Brownian motion within the allowed region, 
and it satisfies the simple Fokker-Planck\nomenclature{FP}{Fokker-Planck} equation 
\[
    \left\{
    \begin{array}{rl}
        \partial_uf_u(y;k)&=\displaystyle\frac{1}{2}\partial_y^2f_u(y;k)\\
        f_u(\pm k;k)&=0
    \end{array}
    \right. ,\quad\forall u\in[0,1].
\]
By spectral decomposition of the Laplace operator, the solution is found to be
\[
    f_u(y;k)=\frac{1}{k}\sum_{n\in\mathds{Z}}\e^{-E_n u}\,\cos\!\left(\sqrt{2E_n}\, y\right), \quad\text{where}\quad E_n=\frac{1}{2}\left(\frac{(2n\!-\!1)\pi}{2k}\right)^2
\]
and the free propagator in the limit $k\to\infty$ is the usual
\[
    f_u(y;\infty)=\frac{1}{\sqrt{2\pi u}}\,\e^{-\frac{y^2}{2u}},
\]
so that the survival probability of the constrained Brownian bridge is
\begin{equation}\label{eq:survival_prob_constr_brown_bridge}
    \mathcal{P}_{\!{\scriptscriptstyle <}}(k;0,1)=\frac{f_1(0;k)}{f_1(0;\infty)}=\frac{\sqrt{2\pi}}{k}\sum_{n\in\mathds{Z}}\Exp{-\frac{(2n\!-\!1)^2\pi^2}{8k^2}}.
\end{equation}
Although it looks different from the historical solution Eq.~\eqref{eq:KSprob}, the two expressions can be shown to be exactly identical. %
          First write the summand in Eq.~\eqref{eq:survival_prob_constr_brown_bridge} as the Fourier transform of $\e^{-2k^2x^2}$.
          Then, adding a small imaginary part to $x$ allows to perform the summation of the series over $n$,
          what makes poles $[1-\e^{\pm\imath \pi x}]^{-1}$ appear in the integrand.
          Finally use Cauchy's residue theorem to perform the integral, obtaining a sum over all residues evaluated at integer values of $x$.

The computation of Kolmogorov's distribution $\mathcal{P}_{\!{\scriptscriptstyle <}}(k;0,1)$ performed above is way easier than the canonical ones \cite{anderson1952asymptotic}.

\subsubsection*{Diffusion in a cage with moving walls}
An appropriate change of variable and time 
\[
    W(t)=(1+t)\, Y\!\left(\frac{t}{1+t}\right),\quad t\in[0,\infty[
\]
leads to the problem of a Brownian diffusion 
inside a box with walls \emph{moving at constant velocity}. Since the walls expand as $\sim t$ faster than the 
diffusive particle can move ($\sim\sqrt{t}$), the survival probability converges to a positive value, 
which turns out to be given by the usual Kolmogorov distribution \eqref{eq:KSprob} \cite{krapivsky1996life}.

A way of addressing this problem of a diffusing particle in an expanding cage was 
suggested in Refs.~\cite{bray2007survival2,bray2007survival1}.
The time-dependent boundary conditions of the usual (forward) Fokker-Planck equation
make it difficult to solve as such, and a nice way out is to consider instead 
the backward Fokker-Planck equation for the transition density, 
with the initial position of the absorbing wall as an additional independent variable !
This eliminates the time dependence in the boundary condition at the cost of introducing the initial wall parameter.

\begin{figure}[!h!]
    \includegraphics[scale=.6,trim=0 0 0 25,clip]{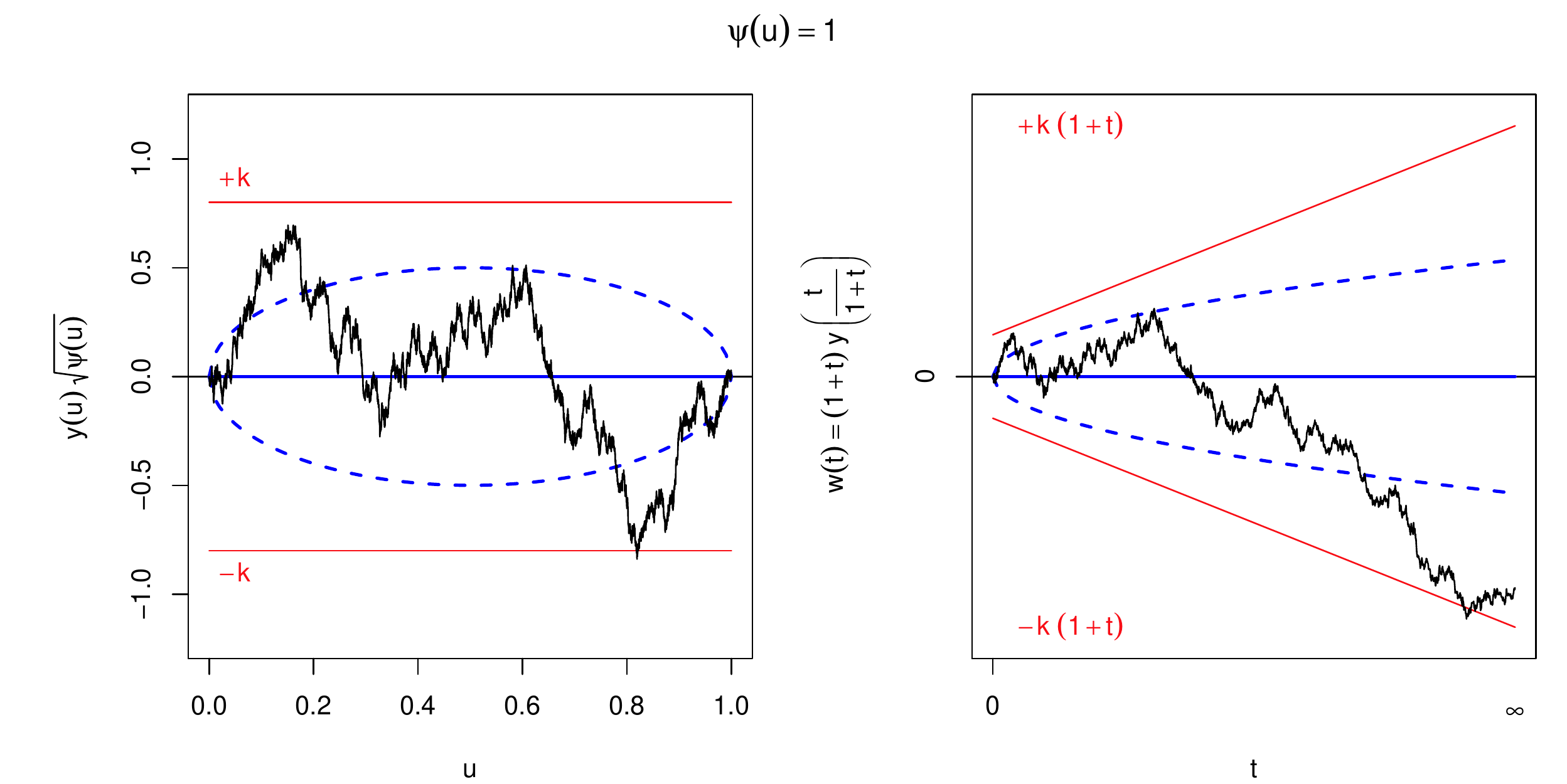}
    \caption{The equi-weighted Brownian bridge, $\psi(u)=1$.
    The time-changed rescaled process lives in a geometry with boundaries receding at constant speed.}\label{fig:BraySmith}
\end{figure}

\subsection[The variance-weighted {B}rownian bridge: accounting for the tails]{The variance-weighted {B}rownian bridge:\\ accounting for the tails}
The classical KS test suffers from an important flaw: the test is only weakly sensitive to the quality of the fit in the tails of the tested
distribution \cite{mason1983modified,mason1992correction}, when it is often these tail events (corresponding to centennial floods, devastating earthquakes, financial crashes, etc.) 
that one is most concerned with (see, e.g., Ref.~\cite{clauset2009power}). 

A simple and elegant GoF test for the tails \emph{only} can be designed starting with digital weights in the form
$\psi(u;a)=\1{u\geq a}$ or $\psi(u;b)=\1{u\leq b}$ for upper and lower tail, respectively.
The corresponding test laws can be read off Eq.~(5.9) in Ref.~\cite{anderson1952asymptotic}.%
\footnote{The quantity $M$ appearing there is the volume under the normal bivariate surface between specific bounds,
and it takes a very convenient form in the unilateral cases $\tfrac12\leq a \leq u \leq 1$ and $0\leq u \leq b\leq\frac12$.
Mind the missing $j$ exponentiating the alternating $(-1)$ factor.}
Investigation of both tails is attained with
$\psi(u;q)=\1{u\leq 1-q}+\1{u\geq q}$ (where $q>\tfrac12$).

Here we rather focus on a GoF test for a univariate sample of size $N\gg 1$, with the Kolmogorov distance but equi-weighted quantiles, 
which is equally sensitive to \emph{all regions} of the distribution.%
\footnote{Other choices of $\psi$ generally result in much harder problems.} 
We unify two earlier attempts at finding asymptotic solutions, one by Anderson and Darling in 1952 \cite{anderson1952asymptotic} 
and a more recent, seemingly unrelated one that deals with ``life and death of a particle in an expanding cage'' 
by Krapivsky and Redner \cite{krapivsky1996life,redner2007guide}. 
We present here the exact asymptotic solution of the corresponding stochastic problem, 
and deduce from it the precise formulation of the GoF test, which is of a fundamentally different nature than the KS test.

So in order to zoom on the tiny differences in the tails of the Brownian bridge, we weight it as explained in the introduction,
with its variance
\[
	\psi(u)=\frac{1}{u(1-u)}.
\]
Solutions for the distributions of such variance-weighted Kolmogorov-Smirnov statistics were studied by No\'e,
leading to the laws of the one-sided \cite{noe1968calculation} and two-sided \cite{noe1972calculation} finite sample tests.
They were later generalized and tabulated numerically by Niederhausen \cite{niederhausen1981sheffer,niederhausen1981tables,wilcox1989percentage}.
However, although exact and appropriate for small samples, these solutions rely on recursive relations and are not in closed form.
We instead come up with an analytic closed-form solution for large samples that relies on an elegant analogy from statistical physics.

\subsubsection*{Diffusion in a cage with moving walls}

Define the time change $t=\frac{u}{1-u}$. The variable $W(t)=(1+t)\,Y\!\left(\frac{t}{1+t}\right)$  
is then a Brownian motion (Wiener process) on $[\tfrac{a}{1-a},\tfrac{b}{1-b}]$, since one can check that:
\[
	{\rm Cov}\big(W(t),W(t')\big)=\min(t,t').
\]
$\mathcal{P}_{\!{\scriptscriptstyle <}}(k|a,b)$ can be now written as
\[
	\mathcal{P}_{\!{\scriptscriptstyle <}}(k|a,b)=\Pr{|W(t)|\leq k\sqrt{t},\forall t\in[\tfrac{a}{1-a},\tfrac{b}{1-b}]}.
\]

The problem with initial time $\frac{a}{1-a}=0$ and horizon time $\frac{b}{1-b}=T$ 
has been treated by Krapivsky and Redner in Ref.~\cite{krapivsky1996life}
as the survival probability $S(T;k=\sqrt{\frac{A}{2D}})$ of a Brownian particle 
diffusing with constant $D$ in a cage with walls expanding as $\sqrt{At}$. 
Their result is that for large $T$, 
\[
S(T;k)\equiv\mathcal{P}_{\!{\scriptscriptstyle <}}(k|0,\tfrac{T}{1+T})\propto T^{-\theta(k)}.
\]
They obtain analytical expressions for $\theta(k)$ in both limits \mbox{$k\to 0$} and \mbox{$k\to\infty$}.
The limit solutions of the very same differential problem were found earlier by Turban for the critical behavior of the directed self-avoiding walk in parabolic 
geometries \cite{turban1992anisotropic}.

We take here a slightly different route, suggested by Anderson and Darling in Ref.~\cite{anderson1952asymptotic} 
but where the authors did not come to a conclusion. Our contributions are: 
(i) we treat the general case $a>0$ for \emph{any} $k$;
(ii) we explicitly compute the $k$-dependence of both the exponent \emph{and} the prefactor of the power-law decay; and
(iii) we provide the link with the theory of GoF tests and 
compute the pre-asymptotic distribution of the weighted Kolmogorov-Smirnov test statistic for large sample sizes $N\to\infty$, 
i.e.\ when $]a,b[\to]0,1[$ .

\begin{figure}[h]
    \includegraphics[scale=.6,trim=0 0 0 25,clip]{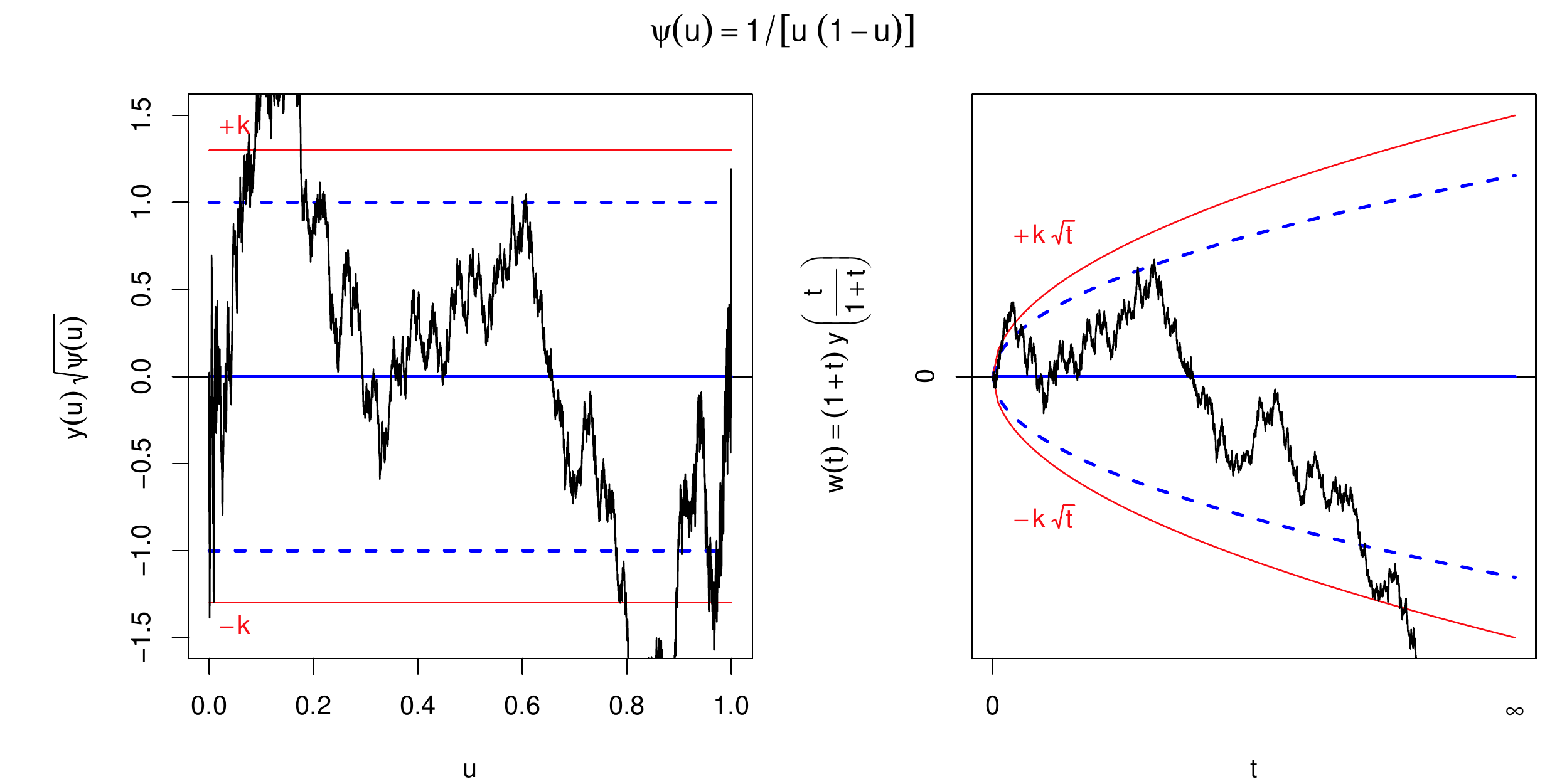}
    \caption{The variance-weighted Brownian bridge, $\psi(u)=1/[u(1-u)]$.
             The time-changed rescaled process lives in a geometry with boundaries receding as $\sim\sqrt{t}$.}\label{fig:KrapRedner}
\end{figure}

\subsubsection*{Mean-reversion in a cage with fixed walls}

Introducing now the new time change $\tau=\ln\sqrt{\frac{1-a}{a}\,t}$, the variable $Z(\tau)=W(t)/\sqrt{t}$ is a stationary 
Ornstein-Uhlenbeck (OU)\nomenclature{OU}{Ornstein-Uhlenbeck} process on $[0,T]$ where 
\begin{equation}\label{eq:Tchange}
    T=\ln\sqrt{\frac{b\,(1-a)}{a\,(1-b)}},
\end{equation}
and
\[
	{\rm Cov}\big(Z(\tau),Z(\tau')\big)=\e^{-|\tau-\tau'|}.
\]
Its dynamics is described by the stochastic differential equation (SDE)\nomenclature{SDE}{Stochastic differential equation}
\begin{equation}
	\d{Z(T)}=-Z(T)\d{T}+\sqrt{2}\,\d{B(T)},
\end{equation}
with $B(T)$ an independent Wiener process.
The initial condition for $T=0$ (corresponding to $b=a$) is $Z(0)=Y(a)/\sqrt{\var{Y(a)}}$,
a random Gaussian variable of zero mean and unit variance. The distribution $\mathcal{P}_{\!{\scriptscriptstyle <}}(k|a,b)$ can now be understood as 
the unconditional survival probability of a mean-reverting particle in a cage with fixed absorbing walls, see Fig.~\ref{fig:KrapRedner}:
\begin{align*}
	%S(T;k) =
    \mathcal{P}_{\!{\scriptscriptstyle <}}(k|T)
	      &=\Pr{-k\leq Z(\tau)\leq k, \forall \tau\in[0,T]}\\
	      &=\int_{-k}^{k}f_T(z;k)\,\d{z},
\end{align*}
where 
\[
    f_T(z;k)\,\d{z}=\mathds{P}\big[Z(T)\in[z,z+\d{z}[ \left| \{Z(\tau)\}_{\tau<T}\right.\big]
\]
is the density probability of the particle being at $z$
at time $T$, when walls are in $\pm k$. Its dependence on $k$, although not explicit on the right hand side, is due to the boundary condition 
associated with the absorbing walls (it will be dropped in the following for the sake of readability)%
\footnote{In particular, $\mathcal{P}_{\!{\scriptscriptstyle <}}(k|0)=\Erf{\tfrac{k}{\sqrt{2}}}$.}.

The Fokker-Planck equation governing the evolution of the density $f_T(z)$ reads
\[
	\partial_{\tau}f_{\tau}(z)=\partial_z\left[z\,f_{\tau}(z)\right]+\partial_z^2\left[f_{\tau}(z)\right],\quad 0<\tau\leq T.
\]
Calling $\mathcal{H}_{\text{FP}}$ the second order differential operator $-\left[\mathds{1}+z\partial_z+\partial_z^2\right]$, 
the full problem thus amounts to finding the general solution of
\[
    \bigg\{
	\begin{array}{rcl}
	-\partial_{\tau}f_\tau(z)&=&\mathcal{H}_{\text{FP}}(z)f_\tau(z)\\
	f_{\tau}(\pm k)&=&0, \forall \tau\in[0,T]
	\end{array}
    \bigg. .
\]
We have explicitly introduced a minus sign since we expect that the density decays with time in an absorption problem.
Because of the term $z\partial_z$, $\mathcal{H}_{\text{FP}}$ is not hermitian and thus cannot be diagonalized.
However, as is well known, one can define $f_{\tau}(z)=\e^{-\frac{z^2}{4}}\phi_{\tau}(z)$ and the Fokker-Planck equation becomes
\[
    \bigg\{
	\begin{array}{rcl}
	-\partial_{\tau}\phi_\tau(z)&=&\left[-\partial_z^2+\frac{1}{4}z^2-\frac12\mathds{1}\right]\phi_{\tau}(z)\\
	\phi_{\tau}(\pm k)&=&0, \forall \tau\in[0,T]
	\end{array}
    \bigg. ,
\]
and its Green's function, i.e.\ the (separable) solution \emph{conditionally on the initial position} $(z_{\text{i}},T_{\text{i}})$, 
is the superposition of all modes
\[
	G_{\phi}(z,T\mid z_{\text{i}},T_{\text{i}})=\sum_{\nu}\e^{-\theta_{\nu}(T-T_{\text{i}})}\widehat\varphi_{\nu}(z)\widehat\varphi_{\nu}(z_{\text{i}}),
\]
where $\widehat\varphi_{\nu}$ are the normalized solutions of the stationary Schr\"odinger equation
\[
    \Bigg\{
	\begin{array}{rcl}
	\left[-\partial_z^2+\frac{1}{4}z^2\right]\varphi_{\nu}(z)&=&\left(\theta_{\nu}+\frac12\right)\varphi_{\nu}(z)\\
	\varphi_{\nu}(\pm k)&=&0
	\end{array}
    \Bigg.,
\]
each decaying with its own energy $\theta_\nu$, 
where $\nu$ labels the different solutions with increasing eigenvalues,
and the set of eigenfunctions $\{\widehat\varphi_{\nu}\}$ defines an orthonormal basis of the Hilbert space on which $\mathcal{H}_{\text{S}}(z)=\left[-\partial_z^2+\frac{1}{4}z^2\right]$ acts.
In particular, 
\begin{equation}\label{eq:orthonormality}
	\sum_{\nu}\widehat\varphi_{\nu}(z)\widehat\varphi_{\nu}(z')=\delta(z-z'),
\end{equation}
so that indeed $G_\phi(z,T_{\text{i}}\mid z_{\text{i}},T_{\text{i}})=\delta(z-z_{\text{i}})$, and the general solution writes
\begin{align*}
	f_T(z_T;k)%&=\int_{-k}^{k}\frac{\e^{-\frac{z_T^2}{4}}}{\e^{-\frac{z_{\text{i}}^2}{4}}}G_{\phi}(z_T,T\mid z_{\text{i}},T_{\text{i}})\,f_0(z_{\text{i}})\,\d{z_{\text{i}}}\\\nonumber
		       &=\int_{-k}^{k}\e^{\frac{z_{\text{i}}^2-z_T^2}{4}}G_{\phi}(z_T,T\mid z_{\text{i}},T_{\text{i}})\,f_0(z_{\text{i}})\,\d{z_{\text{i}}},
\end{align*}
where $T_{\text{i}}=0$, which corresponds to the case $b=a$ in Eq.~(\ref{eq:weightedY}), and $f_0$ is the distribution of the initial value $z_{\text{i}}$
which is here, as noted above, Gaussian with unit variance.

$\mathcal{H}_{\text{S}}$ figures out an harmonic oscillator of mass $\frac12$ and frequency $\omega=\frac{1}{\sqrt{2}}$ within 
an infinitely deep well of width $2k$: its eigenfunctions are parabolic cylinder functions \cite{mei1983harmonic,gradshteyn1980table}
\begin{align*}
    y_+(\theta;z)&=\phantom{z\,}\e^{-\frac{z^2}{4}}\,\FF{_1}{F}{_1}\!\left(-\tfrac{  \theta}{2},\tfrac{1}{2},\tfrac{z^2}{2}\right)\\
    y_-(\theta;z)&=         z\, \e^{-\frac{z^2}{4}}\,\FF{_1}{F}{_1}\!\left( \tfrac{1-\theta}{2},\tfrac{3}{2},\tfrac{z^2}{2}\right)
\end{align*}
properly normalized. 
The only acceptable solutions for a given problem are the linear combinations of $y_+$ and $y_-$ 
which satisfy orthonormality (\ref{eq:orthonormality})  and the boundary conditions:
for periodic boundary conditions, only the integer values of $\theta$ would be allowed, 
whereas with our Dirichlet boundaries \mbox{$|\widehat\varphi_{\nu}(k)|=-|\widehat\varphi_{\nu}(-k)|=0$}, real non-integer eigenvalues $\theta$ are allowed.%
\footnote{A similar problem with a \emph{one-sided} barrier leads to a continuous spectrum; 
this case has been studied originally in Ref.~\cite{mei1983harmonic} and more recently in Ref.~\cite{lladser2000domain} 
(it is shown that there exists a quasi-stationary distribution for any $\theta$)
and generalized in Ref.~\cite{aalen2004survival}.
} 
For instance, the fundamental level $\nu=0$ is expected to be the symmetric solution 
\mbox{$ \widehat\varphi_0(z)\propto y_+(\theta_0;z)$}
with $\theta_0$ the smallest possible %non-odd 
value compatible with the boundary condition:
\begin{equation}
	\theta_0(k)=\inf_{\theta>0}\big\{\theta:y_+(\theta;k)=0\big\}.
\end{equation}
In what follows, it will be more convenient to make the $k$-dependence explicit, and
a hat will denote the solution with the normalization relevant to our problem, 
namely 
%$\widehat{\varphi}_{\nu}(z;k)=\varphi_{\nu}(z;k)/||\varphi_{\nu}||_k$,
 $\widehat{\varphi}_{0  }(z;k)=y_+(\theta_0(k);z)/||y_+||_k$,
with the norm
\[
   %||\varphi_{\nu}||_k^2 \equiv\int_{-k}^k\varphi_{\nu}(z;k)^2\,\d{z}
    ||y_+||_k^2 \equiv\int_{-k}^ky_+(\theta_0(k);z)^2\,\d{z},
\]
so that
$%\[
    \int_{-k}^k\widehat{\varphi}_{\nu}(z;k)^2\,\d{z}=1.
$%\]

\subsubsection*{Asymptotic survival rate}

Denoting by $\Delta_{\nu}(k)\equiv[\theta_{\nu}(k)-\theta_0(k)]$ the gap between the excited levels and the fundamental, 
the higher energy modes $\widehat\varphi_\nu$ cease to contribute to the Green's function when $\Delta_{\nu}T\gg 1$, 
and their contributions to the above sum die out exponentially as $T$ grows.
Eventually, only the lowest energy mode $\theta_0(k)$ remains, and the solution tends to 
\[
	f_{T}(z;k)=A(k)\,\e^{-\frac{z^2}{4}}\,\widehat\varphi_{0}(z;k)\,\e^{-\theta_{0}(k)T},
\]
when $T\gg (\Delta_1)^{-1}$, with 
\begin{equation}\label{eq:A_tilde}
	A(k)=\int_{-k}^k\e^{\frac{z_{\text{i}}^2}{4}}\widehat\varphi_0(z_{\text{i}};k)f_0(z_{\text{i}})\,\d{z_{\text{i}}}.
\end{equation}

Let us come back to the initial problem of the weighted Brownian bridge reaching its extremal value in $[a,b]$.
If we are interested in the limit case where $a$ is arbitrarily close to $0$ and $b$ close to $1$, then $T \to \infty$ and
the solution is thus given by
\begin{align*}
	%S(T;k)  
    \mathcal{P}_{\!{\scriptscriptstyle <}}(k|T)&=A(k)\,\e^{-\theta_{0}(k)T}\int_{-k}^{k}\e^{-\frac{z^2}{4}}\widehat\varphi_{0}(z;k)\,\d{z}\\
            &=\widetilde{A}(k)\, \e^{-\theta_{0}(k)T},  
\end{align*}
with $\widetilde{A}(k) \equiv \sqrt{2\pi}A(k)^2$.

We now compute explicitly the limit behavior of both $\theta_0(k)$ and $\widetilde{A}(k)$.
\paragraph*{$\boxed{k\to \infty}$} 
As $k$ goes to infinity, the absorption rate $\theta_0(k)$ is expected to converge toward $0$: intuitively, an infinitely 
far barrier will not absorb anything. At the same time, $\mathcal{P}_{\!{\scriptscriptstyle <}}(k|T)$ must tend to 1 in that limit. 
So $\widetilde{A}(k)$ necessarily tends to unity. Indeed,
\begin{align}
	       \theta_0(k)&\xrightarrow{k\to\infty} \sqrt{\frac{2}{\pi}}k\,\e^{-\frac{k^2}{2}} \to 0,\\\nonumber
	  \widetilde{A}(k)&\xrightarrow{k\to\infty}\left(\int_{-\infty}^{\infty}\widehat{\varphi}_{0}(z;\infty)^2\,\d{z}\right)^2 = 1.
\end{align}
In principle, we see from Eq.~(\ref{eq:A_tilde}) that corrections to the latter arise both (and jointly) from 
the functional relative difference of the solution 
\mbox{$\epsilon(z;k)=y_+(\theta_0(k);z)/y_+(0;z)-1$}, 
and from the finite integration limits ($\pm k$ instead of $\pm \infty$).
However, it turns out that the correction of the first kind is of second order in $\epsilon$.%
\footnote{From Eq.~(\ref{eq:A_tilde}) we have, when \mbox{$k\to\infty$},
\[
    A(k)=(2\pi)^{-1/2}\frac{\int_{-k}^k\e^{-z^2/2}[1+\epsilon(z;k)]\d{z}}{\sqrt{\int_{-k}^k\e^{-z^2/2}[1+\epsilon(z;k)]^2\d{z}}}.
\]
The result follows by keeping only the dominant terms in the expansion in powers of $\epsilon(z;k)$.
A similar computation for the asymptotic analysis by expanding the wave function in $\theta$ was performed in Ref.~\cite{krapivsky1996life}.
Alternatively, algebraic arguments allow to understand that, to first order in the energy correction $\theta_0(k)-\theta_0(\infty)$, 
the perturbation of the wave function is orthogonal to $\widehat{\varphi}_{0}(z;\infty)$.} 
The correction to $A(k)$ 
is thus dominated by the finite integration limits $\pm k$, so that% pre-asymptotically:
\begin{equation}
    \widetilde{A}(k\to\infty)\approx \Erf{\frac{k}{\sqrt{2}}}^2.
\end{equation}
\paragraph*{$\boxed{k\to 0}$}
For small $k$, the system behaves like a free particle in a sharp and infinitely deep well, since the quadratic potential is almost flat around 0.
The fundamental mode becomes then 
\[
	\widehat\varphi_0(z;k\to 0)=\frac{1}{\sqrt{k}}\cos\!\left(\frac{\pi z}{2k}\right),
\]
and consequently
\begin{align}
	       \theta_0(k)&\xrightarrow{k\to 0}\frac{\pi^2}{4k^2}-\frac12,\\\nonumber
	  \widetilde{A}(k)&\xrightarrow{k\to 0}\left(\int_{-k}^{k}\frac{\e^{-\frac{z^2}{4}}}{(2\pi)^{\frac{1}{4}}}
                                                              \frac{1}{\sqrt{k}}\cos\!\left(\frac{\pi z}{2k}\right)\d{z}\right)^2\\
                        &\approx\frac{1}{\sqrt{2\pi}\,k}
	  \left(\frac{4k}{\pi}\right)^2=\frac{16}{\pi^2\sqrt{2\pi}}k.
\end{align}

We show in Fig.~\ref{fig:A_theta} the functions $\theta_0(k)$ and $\widetilde{A}(k)$ computed numerically from the exact solution,
together with their asymptotic analytic expressions. In intermediate values of $k$ (roughly between 0.5 and 3) these limit
expressions fail to reproduce the exact solution.
\begin{figure}[p] %[!h!]
	\center
	\includegraphics[scale=0.75,trim=0 0 -25 0,clip]{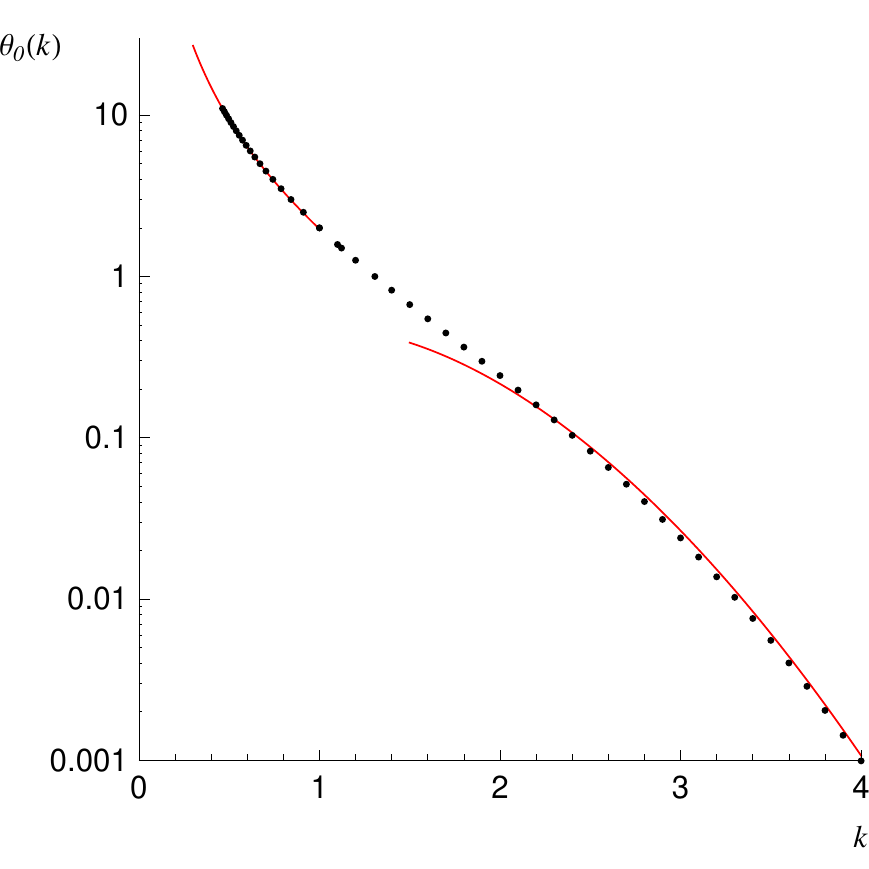}
	\includegraphics[scale=0.75,trim=0 0 -25 0,clip]{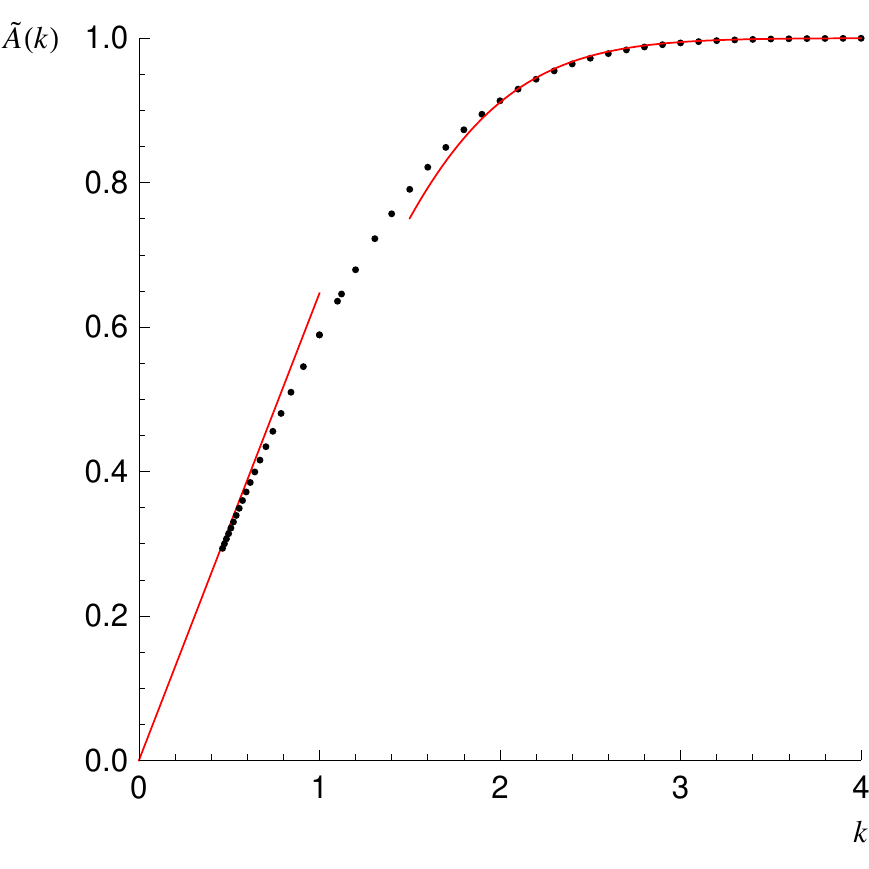}
	\caption{\textbf{Left:  } Dependence of the exponent       $\theta_0$ on $k$; similar to Fig.~2 in Ref.~\cite{krapivsky1996life}, but in lin-log scale; see in particular Eqs.~(9b) and (12) there.
	         \textbf{Right: } Dependence of the prefactor $\widetilde{A}$ on $k$. The red solid lines illustrate the analytical behavior in the
	         limiting cases $k\to 0$ and $k\to\infty$.}
	\label{fig:A_theta}
%\end{figure}
%\begin{figure}[!h!]
	\center
	\includegraphics[scale=0.85,trim=0 0 -25 0,clip]{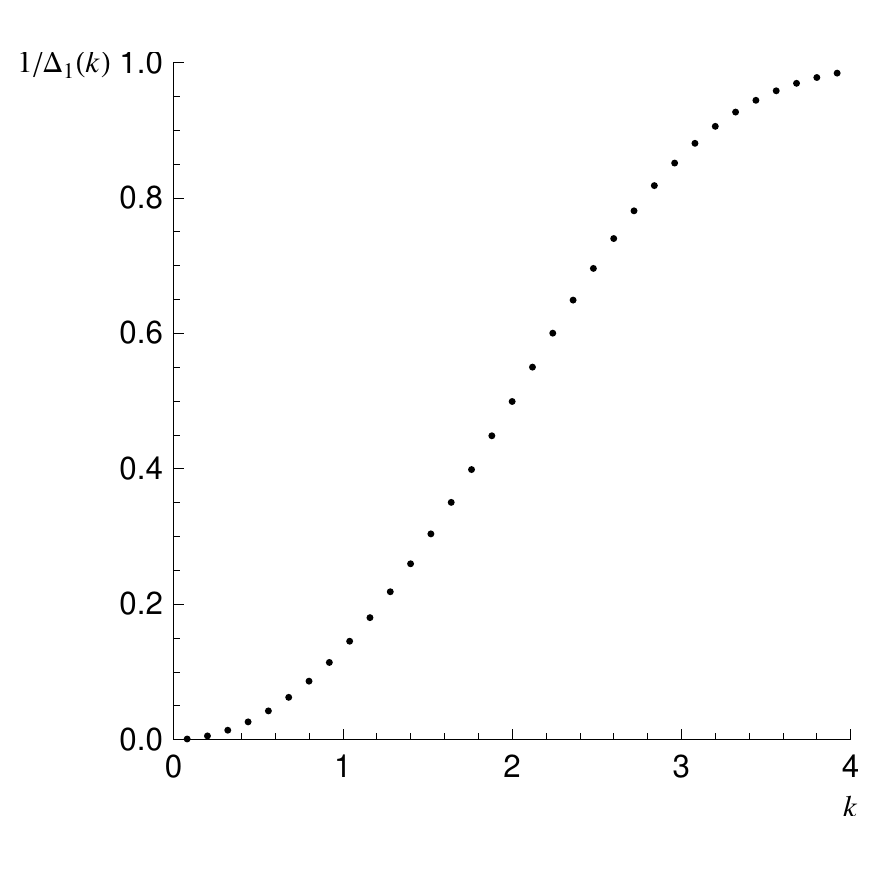}
	\caption{$1/\Delta_1(k)$ saturates to 1, so that the condition \mbox{$N\gg\exp[1/\Delta_1(k)]$} is virtually always satisfied.}
	\label{fig:delta1}
\end{figure}

\subsubsection*{Higher modes and validity of the asymptotic ($N\gg 1$) solution}
Higher modes $\nu>0$ with energy gaps $\Delta_\nu\lesssim 1/T$ must in principle be kept in the pre-asymptotic computation.
This, however, is irrelevant in practice since the gap \mbox{$\theta_1-\theta_0$} is never small.
Indeed, $\widehat\varphi_1(z;k)$ is proportional to the asymmetric solution $y_-(\theta_1(k);z)$
and its energy
\[
    \theta_1(k)=\inf_{\theta>\theta_0(k)}\big\{\theta:y_-(\theta;k)=0\big\}
\]
is found numerically to be very close to \mbox{$1+4\theta_0(k)$}. 
In particular, $\Delta_1>1$ (as we illustrate in Fig.~\ref{fig:delta1}) 
and thus $T \Delta_1 \gg 1$ will always be satisfied in cases of interest.

\subsection{Back to {GoF} testing and conclusion}

Let us now come back to GoF testing. 
In order to convert the above calculations into a meaningful test, one must specify values of $a$ and $b$. 
The natural choice would be $a=1/N$, corresponding to the min of the sample series since 
$F(\min x_n)\approx F_N(\min x_n)=\frac{1}{N}$. 
Eq.~(\ref{eq:Tchange}) above motivates a slightly different value of $a=1/(N\!+\!1)$ and $b=1-a$, such that 
the relevant value of $T$ is given correspondingly by 
\[
	T=\ln\sqrt{\frac{b\,(1-a)}{a\,(1-b)}} = \ln N.%,\qquad N \gg 1.
\]
This leads to our central result for the cdf of the weighted maximal Kolmogorov distance $K(\tfrac{1}{N\!+\!1},\tfrac{N}{N\!+\!1})$ 
under the hypothesis that the tested and the true distributions coincide:
\begin{equation}\label{eq:final_res}
	\boxed{S(N;k)=\mathcal{P}_{\!{\scriptscriptstyle <}}(k|\ln N)=\widetilde{A}(k)N^{-\theta_0(k)}},
\end{equation}
which is valid whenever $N \gg 1$ since, as we discussed above, the energy gap $\Delta_1$ is greater than unity.

The final cumulative distribution function (the test law) is depicted in Fig.~\ref{fig:S} for different values of the sample size $N$:
As $N$ grows toward infinity, the curve is shifted to the right, and eventually $S(\infty;k)$ is zero for any $k$.
The red solid lines illustrate the analytical behavior in the limiting cases $k\to 0$ and $k\to\infty$.
Contrarily to the standard KS case (Fig.~\ref{fig:KS}), this distribution \emph{still depends on $N$}. 
In particular, the threshold value $k^*$ corresponding to a $95 \%$ confidence level increases with $N$. 
Since for large $N$, $k^* \gg 1$ one can use the asymptotic expansion above, which soon becomes quite accurate, as shown in Fig.~\ref{fig:S}. 
This leads to:
\[
\theta_0(k^*) \approx - \frac{\ln 0.95}{\ln N} \approx \sqrt{\frac{2}{\pi}}k^*\,\e^{-\frac{k^{*2}}{2}},
\]
which gives $k^* \approx 3.439, 3.529, 3.597, 3.651$ for, respectively, $N=10^3,10^4,10^5,10^6$. 
For exponentially large $N$ and to logarithmic accuracy, one has: $k^* \sim \sqrt{2 \ln (\ln N)}$. 
This variation is very slow, but one sees that as a matter of principle, the ``acceptable'' maximal 
value of the weighted distance is much larger (for large $N$) than in the KS case.

%%% CODE R 
%%% fon <- function(k,N) {sqrt(2/pi)*k*exp(-k^2/2)+log(.95)/log(N)}
%%% round(sapply(10^(3:6),function(N) uniroot(fon,c(3,4),N=N)$root),3)

\begin{figure}
	\center
	\subfigure[Kolmogorov asymptotic distribution.]{\includegraphics[scale=0.75]{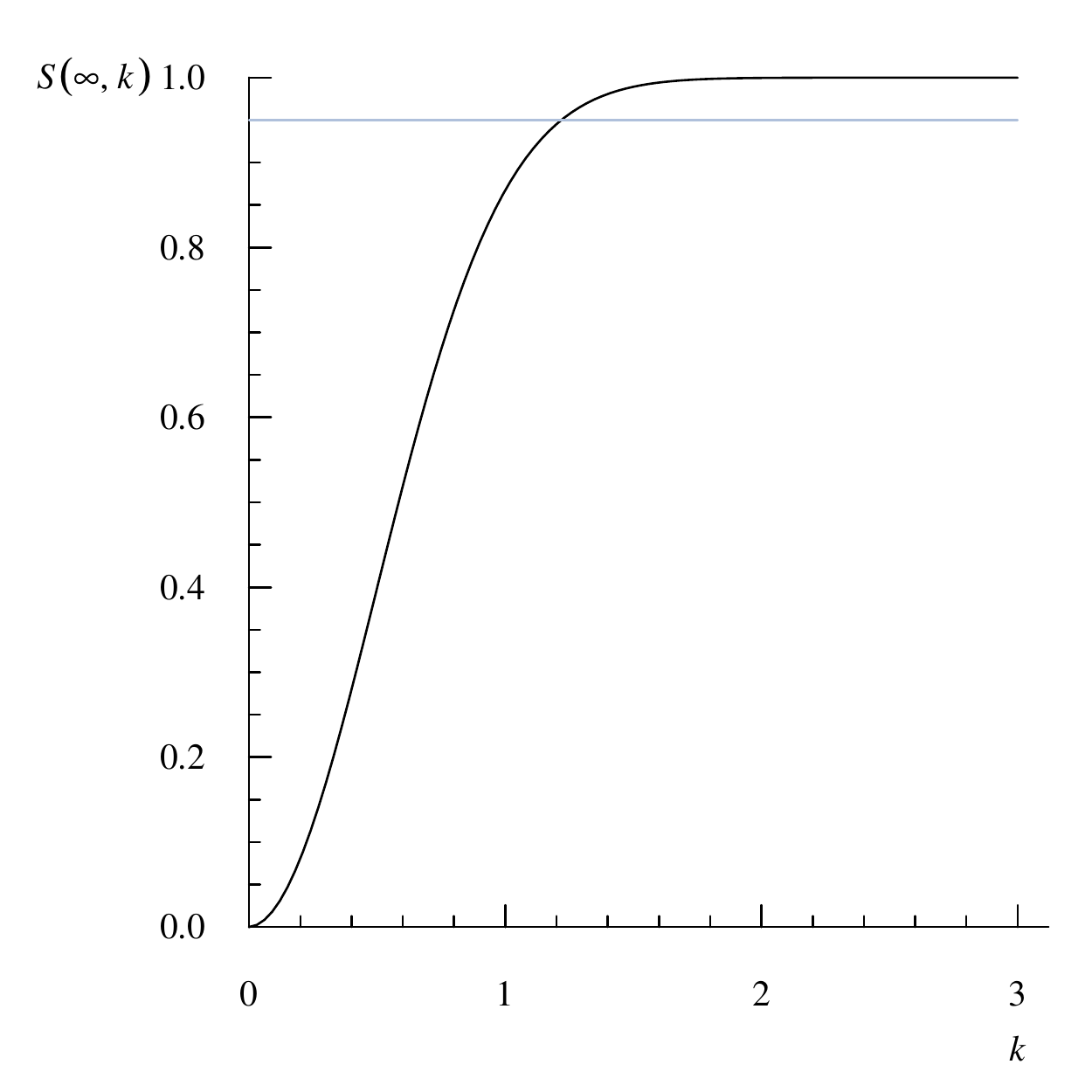}\label{fig:KS}}
	\subfigure[Dependence of $S(N;k)$ on $k$ for $N=10^3,10^4,10^5,10^6$ (from left to right).]{\includegraphics[scale=1.,trim=0 10 -25 0,clip]{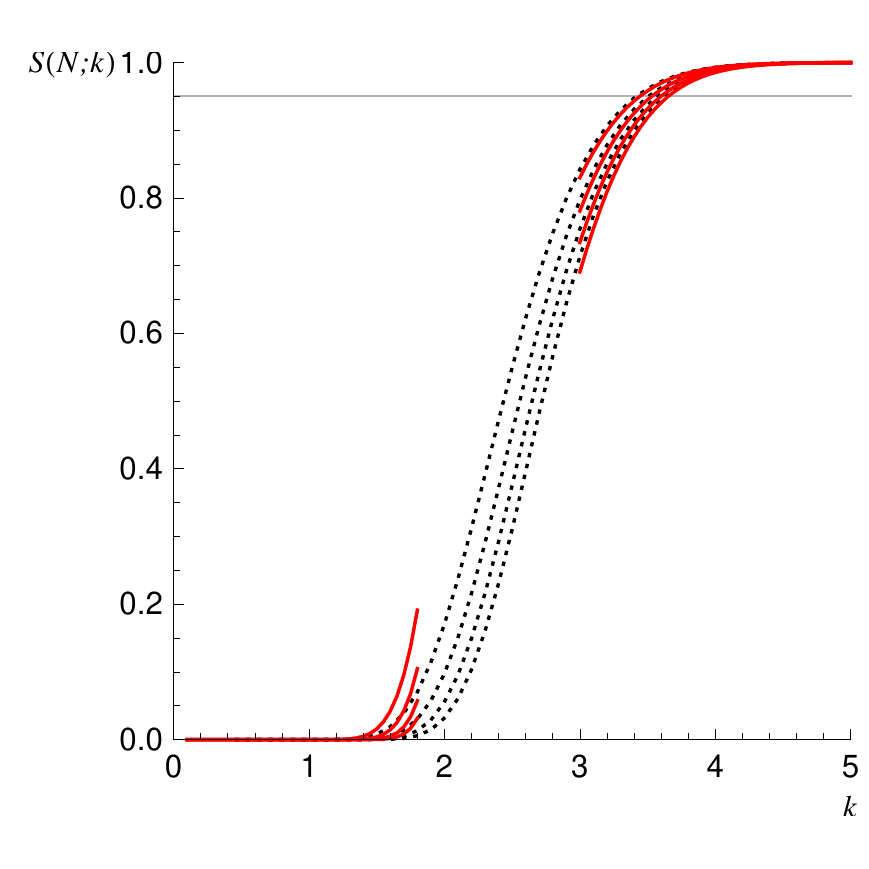}\label{fig:S}}
	\caption{The horizontal grey line corresponds to a $95 \%$ confidence level.}
	
\end{figure}

In conclusion, we believe that accurate GoF tests for the extreme tails of empirical distributions is a very important issue, relevant in many 
contexts. We have derived exact asymptotic results for a generalization of the Kolmogorov-Smirnov test, 
well suited to testing the whole domain up to these extreme tails. 
Our final results are summarized in Eq.~(\ref{eq:final_res}) and Fig.~\ref{fig:S}.
In passing, we have rederived and made more precise the result of 
Krapivsky and Redner \cite{krapivsky1996life} concerning the survival probability of a diffusive particle in an expanding cage. It would be 
interesting to exhibit other choices of weight functions that lead to soluble survival probabilities. It would also be interesting to extend the 
present results to multivariate distributions, and to dependent observations, along the lines of Ref.~\cite{chicheportiche2011goodness}.

%We want to thank Sid Redner for a useful discussion and for his inspiring work, 
%and Lo\"ic Turban for bringing Ref.~\cite{turban1992anisotropic} to our attention.

%%%%%%%%%%%%%%%%%%%%%%%%%%%%%%%%%%%%%%%%%%%%%%%%%%%%%%%%%%%%%%%%%%%%%%%%

\section{Goodness-of-fit tests for a sample of dependent draws}\label{sec:GoF}
%\subsection{Introduction}

Usual GoF tests, and even the generalization that we just worked out in the previous section of this chapter,
are designed for samples of identically distributed \emph{and independent} draws.
It however so happens that in certain fields (physics, finance, geology, etc.) 
the random variables under scrutiny have some dependence, be it \emph{spatial} or \emph{temporal} (in which case one speaks of ``memory'').
Applying naively GoF tests to such samples leads to a pessimistic evaluation of the null hypothesis of the test,
and eventually causes a too high rejection rate, as will be made more precise below.
It is thus natural to ask how these dependences effectively affect the law of the test statistic.
A minimal departure from independence is to allow for homogeneous dependences between draws $i$ and $j$,
for which the two-points copula is translationally invariant.
In this chapter, we adopt the point of view of temporal dependences because it has immediate 
applications in financial time series (see later in Chapter \ref{part:partIII}.\ref{chap:cop_fin}),
but the theory applies as well to observables equally spaced on the line, 
and whose dependences rely on their respective distances.

Whereas the unconditional law of the variable may well be unique and independent of time, 
the conditional probability distribution of an observation \emph{following} a previous observation 
exhibits specific patterns, and in particular long-memory, even when the linear correlation is short-ranged or trivial. 
Examples of such phenomena can be encountered in fluid mechanics (the velocity of a turbulent fluid) 
and finance (stock returns have small auto-correlations but exhibit strong volatility clustering, a form of heteroskedasticity).
The long-memory nature of the underlying processes makes it inappropriate to use standard GoF tests in these cases. 
Still, the determination of the unconditional distribution of returns is a classic problem in quantitative finance, with obvious 
applications to risk control, portfolio optimization or derivative pricing. 
Correspondingly, the distribution of stock returns (in particular the behavior of its tails) 
has been the subject of numerous empirical and theoretical papers (see e.g. \cite{plerou1999scaling,dragulescu2002probability} 
and for reviews \cite{bouchaud2003theory,malevergne2006extreme} and references therein). 
Clearly, precise statements are only possible if meaningful GoF tests are available \cite{weiss1978modification}.

In the rest of this section, we study theoretically how to account for general dependence in GoF tests: 
we first describe the statistical properties of the empirical cdf of a non-iid vector of observations of finite size,
as well as measures of its difference with an hypothesized cdf. 
We then study the limit properties of this difference and the asymptotic distributions of two norms.
Examples and financial applications are relegated to Chapter~\ref{chap:cop_fin} in part~\ref{part:partIII}.

%%%%%%%%%%%%%%%%%%%%%%%%%%%%%%%%%%%%%%%%%%%%%%%%%%%%%%%%%%%%%%%%%%%%%%%%

\subsection{Empirical cumulative distribution and its fluctuations}
Contrarily to section~\ref{sec:GoFintro}, we now consider
 $\vect{X}$ to be a random vector with $N$ identically distributed but dependent variables, with marginal cdf $F$. 
One realization of $\vect{X}$ consists of a time series $\{x_1,\ldots,x_n, \ldots, x_N\}$ that exhibits some sort of persistence.
Due to the persistence, the fluctuations of the empirical CDF \eqref{eq:F_N} depend on the joint distribution of all couples $(X_n,X_m)$,
so that the covariance of the rescaled empirical CDF $Y_N(u)$ defined in Eq.~\eqref{eq:Y_N} is now
\begin{equation}\label{eq:CovYbar}
	{\rm Cov}(Y_N(u),Y_N(v))=\min(u,v)-uv + \frac{1}{N} \sum_{n, m \neq n}^N \big(\cop[nm](u,v)-uv\big)
\end{equation}
and makes the pairwise copulas $\cop[nm]$ explicitly appear.
It is more convenient to introduce
\begin{equation}\label{eq:Psi}
	\Psi_N(u,v)= \frac{1}{N} \sum_{n, m \neq n}^N \Delta_{nm}(u,v),\quad\text{where}\quad\Delta_{nm}(u,v)= \frac{\cop[nm](u,v)-uv}{\min(u,v)-uv},
\end{equation}
which measures the relative departure from the independent case $\Psi_N(u,v) \equiv 0$, 
corresponding to the pairwise copula being independence $\cop[nm](u,v)=uv$. 
Note that decorrelated but dependent variables may lead to a non zero value of $\Psi_N$, 
since the whole pairwise copula enters the formula and not only the linear correlation coefficients.
When the denominator is zero, the fraction should be understood in a limit sense;
we recall in particular that \cite{chicheportiche2012joint}
\begin{equation}
	\Delta_{nm}(u,u) = \frac{\cop[nm](u,u)-u^2}{u\,(1-u)}=\tUU_{nm}(u)+\tLL_{nm}(1-u)-1
\end{equation}
tends to the upper/lower tail dependence coefficients $\tUU_{nm}(1)$ when $u$ tends to $1$, and to $\tLL_{nm}(1)$ when $u$ tends to $0$.
Intuitively, the presence of $\Psi_N(u,v)$ in the covariance of $Y_N$ above leads to a reduction of the 
number of effectively independent variables, 
but a more precise statement requires some further assumptions that we detail below. 

In the following, we will restrict to the case of {strong-}stationary random vectors, 
for which the copula $\cop[nm]$ only depends on the lag $t=m-n${, i.e. $\cop[t]\equiv \cop[n,n+t]$}.   
The average of $\Delta_{nm}$ over $n,m$ can be turned into an average over $t$:
\begin{equation}\label{eq:avgDelta_t_homo}
	\Psi_N(u,v)=\sum_{t=1}^{N-1} (1-\frac{t}{N})\big(\Delta_{t}(u,v)+\Delta_{-t}(u,v)\big)
\end{equation}
with $\Delta_t(u,v) = \Delta_{n,n+t}(u,v)$ and $\Delta_{-t}(u,v) = \Delta_{n+t,n}(u,v)$. 
Note that in general $\Delta_t(u,v) \neq \Delta_{-t}(u,v)$, but clearly $\Delta_{t}(u,v)=\Delta_{-t}(v,u)$, which implies that
$\Psi_N(u,v)$ is symmetric in $u \leftrightarrow v$. 

We will assume in the following that the dependence encoded by $\Delta_t(u,v)$ has a limited range in time, 
or at least that it decreases sufficiently fast for the above sum to converge when $N \to \infty$. 
If the characteristic time scale for this dependence is $T$, we assume in effect that $T \ll N$. 
In the example worked out in Section~\ref{sec:example} below, one finds:
\[
\Delta_{t}(u,v) = f\!\left(\frac{t}{T}\right) \frac{A(u,v)}{I(u,v)}, \qquad I(u,v) \equiv \min(u,v)-uv
\]
where $f(\cdot)$ is a certain function. If $f(r)$ decays faster than $r^{-1}$, one finds (in the limit $T \gg 1$):
\[
\Psi_\infty(u,v) = \lim_{N \to \infty} \Psi_N(u,v) = T\, \frac{A(u,v)+A(v,u)}{I(u,v)} \int_0^\infty \d{r} f(r),
\]
with corrections at least of the order of $T/N$ when $N \gg T$.

\subsection{Limit properties}
We now define the process $\tilde{y}(u)$ as the limit of $Y_N(u)$ when $N \to \infty$.
For a given $u$, it represents the asymptotics of the difference 
between the empirically determined cdf of the underlying $X$'s and the theoretical one, at the $u$-th quantile.
According to the Central Limit Theorem under weak dependences, 
it is Gaussian as long as the strong mixing coefficients, 
\[%\fl
	\alpha_{\text{SM}}(t)=\sup_{\tau}\sup_{A,B}\left\{\big|\mathds{P}(A\cap B)-\mathds{P}(A)\mathds{P}(B)\big|: A\in \sigma(\{Z_n(u)\}_{n\leq \tau}),B\in \sigma(\{Z_n(v)\}_{n\geq \tau+t})\right\}
\]
associated to the sequence $\{Z_n(u)\}=\{\1{F(X_n)\leq u}-u\}$, vanish at least as fast as $\mathcal{O}(t^{-5})$%
\footnote{This condition means that the occurrence of any two realizations of the underlying variable 
can be seen as independent for sufficiently long time between the realizations.
Since the copula induces a measure of probability on the Borel sets,
it amounts in essence to checking that $|\cop[t](u,v)-uv|$ converges quickly towards 0.
See Refs.~\cite{bradley2007introduction,chen2010nonlinearity,beare2010copulas} for definitions of 
$\alpha-$, $\beta-$, $\rho-$mixing coefficients and sufficient conditions on copulas for geometric mixing (fast exponential decay)
in the context of copula-based stationary Markov chains.}.
We will assume this condition to hold in the following. 
For example, this condition is met if the function $f(r)$ defined above decays exponentially, or if $f(r \geq 1)=0$.

The covariance of the process $\tilde{y}(u)$ is given by:
\begin{equation}\label{eq:Htheo}
	H(u,v)=I(u,v)\,\left[1+\Psi_{\infty}(u,v)\right]
\end{equation}
and characterizes a Gaussian bridge since $\var{\tilde{y}(0)}=\var{\tilde{y}(1)}=0$, 
or equivalently $\pr{\tilde{y}(0)=y}=\pr{\tilde{y}(1)=y}=\delta(y)$.
Indeed, $I(u,v)=\min(u,v)-uv$ is the covariance function of the Brownian bridge, 
and $\Psi_{\infty}(u,v)$ is a non-constant scaling term.

By Mercer's theorem, the covariance $H(u,v)$ can be decomposed on its eigenvectors
and $\tilde{y}(u)$ can correspondingly be written as an infinite sum of Gaussian variables:
\begin{equation}\label{eq:y_sum_z}
	\tilde{y}(u)=\sum_{j=1}^{\infty}U_j(u)\sqrt{\lambda_j} \, z_j
\end{equation}
where $z_j$ are independent centered, unit-variance Gaussian variables, 
and the functions $U_j$ and the numbers $\lambda_j$ are solutions to the eigenvalue problem:
\begin{equation}
	\int_0^1 H(u,v) U_i(v)\,\d{v}=\lambda_i \, U_i(u)\quad\textrm{with}\quad\int_0^1 U_i(u)U_j(u)\,\d{u}= \delta_{ij}.
\end{equation}

In practice, for every given problem, the covariance function in Equation~(\ref{eq:Htheo}) has a specific shape, 
since $\Psi_\infty(u,v)$ is copula-dependent. Therefore, contrarily to the case of independent random variables,
the GoF tests will not be characterized by universal (problem independent) distributions.

\subsection{Law of the norm-2 (Cram\'er-von~Mises)}
\nomenclature{CvM, CM}{Cram\'er-von~Mises}
The norm-2 of the limit process is the integral of $\tilde{y}^2$ over the whole domain:
\begin{subequations}
\begin{equation}
	C\!M=\int_0^1\tilde{y}(u)^2\,\d{u}.
\end{equation}
In the representation (\ref{eq:y_sum_z}), it has a simple expression: 
\begin{equation}\label{eq:CM_sum_lambda}
	C\!M=\sum\limits_{j=1}^{\infty}\lambda_j z_j^2.
\end{equation}
\end{subequations}
and its law is thus the law of an infinite sum of squared independent Gaussian variables 
weighted by the eigenvalues of the covariance function. 
Diagonalizing $H$ is thus sufficient to find the distribution of $C\!M$, 
in the form of the Fourier transform of the characteristic function
\begin{equation}\label{eq:charctCM}
	\phi(t)=\Esp{\e^{\imath t\, C\!M}}
	       =\prod_{j}\left(1-2\imath t\lambda_j\right)^{-\frac12}.
\end{equation}
The hard task consists in finding the infinite spectrum of $H$ (or some approximations, if necessary).

Ordering the eigenvalues by decreasing amplitude, 
Equation~(\ref{eq:CM_sum_lambda}) makes explicit the decomposition of $C\!M$ over contributions of decreasing importance 
so that, at a wanted level of precision, only the most relevant terms can be retained.
In particular, if the top eigenvalue dominates all the others, we get the chi-square law with a single degree of freedom:
\begin{equation}
	\pr{C\!M\leq k}=\operatorname{erf}\sqrt{\frac{k}{\lambda_0}}.
\end{equation}

Even if the spectrum cannot easily be determined but $H(u,v)$ is known, 
all the moments of the distribution can be computed exactly. For example:
\begin{subequations}\label{eq:momentsCvM}
\begin{align}
	\esp{C\!M}&=\tr H= \int_0^1H(u,u)\,\d{u},\\
	\var{C\!M}&\equiv 2\,\tr H^2=2\iint_0^1H(u,v)^2\,\d{u}\,\d{v}.
\end{align}
\end{subequations}

\subsection{Law of the supremum (Kolmogorov-Smirnov)}
The supremum of the difference between the empirical cdf of the sample and the target cdf under the null-hypothesis
has been used originally by Kolmogorov and Smirnov as the measure of distance.
The variable 
\begin{equation}
	K\!S=\sup_{u\in[0,1]}|\tilde{y}(u)|
\end{equation}
describes the limit behavior of the GoF statistic. 
In the case where $1 + \Psi_\infty(u,v)$ can be factorized as $\sqrt{\psi(u)}\sqrt{\psi(v)}$, 
the procedure for obtaining the limiting distribution was worked out in \cite{anderson1952asymptotic}, 
and leads to a problem of a diffusive particle in an expanding cage, for which some results are known. 
There is however no general method to obtain the distribution of $K\!S$ for an arbitrary covariance function $H$. 

Nevertheless, if $H$ has a dominant mode, the relation (\ref{eq:y_sum_z}) becomes approximately: 
$
	\tilde{y}(u)=U_0(u)\sqrt{\lambda_0}\,z_0\equiv\kappa_0(u)z_0
$, and 
\begin{equation}
	K\!S=\sqrt{\lambda_0}\,|z_0|\sup_{u\in[0,1]}|U_0(u)|\equiv\kappa_0(u_0^*)|z_0|.
\end{equation}
The cumulative distribution function is then simply
\begin{equation}\label{eq:KS_dep_final}
	\pr{K\!S\leq k}=\Erf{\frac{k}{\sqrt{2}\,\kappa_0(u_0^*)}},\qquad k\geq 0.
\end{equation}
%where $u^*=\arg\sup|U_0^H(u)|$. 
This approximation is however not expected to work for small values of $k$, since in this case $z_0$ must be small, 
and the subsequent modes are not negligible compared to the first one.
A perturbative correction --- working also for large $k$ only --- can be found when the second eigenvalue is small,
or more precisely when $\tilde{y}(u)=\kappa_0(u)z_0+\kappa(u)z_1$ with $\epsilon=\kappa/\kappa_0 \ll 1$.
The first thing to do is find the new supremum
\begin{equation}
	u^*=\arg\sup(\tilde{y}(u)^2)=u_0^*+\frac{\kappa'(u_0^*)}{|\kappa_0''(u_0^*)|}\frac{z_1}{z_0}.
\end{equation}
Notice that it is dependent upon $z_0,z_1$ so that $K\!S$ is no longer exactly the absolute value of a Gaussian.
However it can be shown (after lengthy but straightforward calculation) that, to second order in $\epsilon$, 
$\tilde{y}(u^*)$ \emph{remains  Gaussian}, albeit with a new width
\begin{equation}\label{eq:kappa_star}
	%\kappa_0^2\,\frac{\kappa_0^2-\kappa'^2\kappa_0/|\kappa_0''|+\kappa^2}{\kappa_0^2-\kappa'^2\kappa_0/|\kappa_0''|}>\kappa_0^2
	\kappa^* \approx \sqrt{\kappa_0^2+\kappa^2}>\kappa_0,
\end{equation}
where all the functions are evaluated at $u_0^*$. 
In fact, this approximation works also with more than two modes, provided
\begin{equation}\label{eq:kappa_u_perturb}
	\kappa(u)^2\equiv\sum_{j\neq 0}\lambda_jU_j(u)^2 \ll \kappa_0(u)^2=\lambda_0U_0(u)^2,
\end{equation}
in which case:
\begin{equation}\label{eq:kappa_star_sum}
	\kappa^*\approx\sqrt{\sum_j\lambda_jU_j(u_0^*)^2},
\end{equation}
where the sum runs from 0 to the farthest mode still satisfying the inequality in Eq.~\eqref{eq:kappa_u_perturb}.

\subsection{Conclusion}
We have introduced a framework for the study of statistical tests of Goodness-of-Fit 
with dependent samples, that heavily relies on the notion of bivariate copulas. 

In summary, GoF testing on persistent series cannot be universal as is the case for i.i.d.\ variables, 
but requires a careful estimation of the self-copula at all lags.
Correct asymptotic laws for the test statistics can be found 
as long as dependences are short ranged, i.e.\ $T\ll N$. 
These laws depend on the spectral properties of the covariance kernel.
The resulting distribution of the Cram\'er-von~Mises statistic is given in Eq.~\eqref{eq:charctCM} (in Fourier form),
and the law of the  Kolmogorov-Smirnov test can be read off Eq.~\eqref{eq:KS_dep_final} when a single eigenmode is dominant, 
or substituting Eq.~\eqref{eq:kappa_star_sum} to take into account the contribution of subsequent modes.

An application of the modified Goodness-of-Fit tests developed in this chapter
is provided in Chapter~\ref{chap:cop_fin} of part~\ref{part:partII}, where we test the distribution of stock returns.

We conclude with two remarks of methodological interest. 

1)~The method presented for dealing with self-dependences while using statistical tests of Goodness-of-Fit
is computationally intensive in the sense that it requires to estimate empirically the self-copula for
all lags over the entire unit square. 
In the non-parametric setup, discretization of the space must be chosen so as to provide a good approximation 
of the continuous distance measures while at the same time not cause too heavy computations.
Considering that fact, it is often more appropriate to use the Cram\'er-von~Mises-like test rather 
than the Kolmogorov-Smirnov-like, as numerical error on the evaluation of the integral will typically 
be much smaller than on the evaluation of the supremum on a grid,
more so when the grid size is only about $\frac{1}{M}\approx \frac{1}{100}$.

2)~The case with long-ranged dependence $T \gg N \gg 1$ cannot be treated in the framework presented here.
First because the Central Limit Theorem does not hold in that case, 
and finding the limit law of the statistic may require more advanced mathematics.
But even pre-asymptotically, summing the lags over the available data up to $t\approx N$ means that 
a lot of noise is included in the determination of $\Psi_N(u,v)$ (see Equation~\ref{eq:Psi}).
This, in turn, is likely to cause the empirically determined kernel $H(u,v)$ not to be positive definite. 
One way of addressing this issue is to follow a semi-parametric procedure: 
the copula $\cop[t]$ is still estimated non-parametrically, but the kernel $H$ sums the lagged copulas $\cop[t]$ 
only up to a scale where the linear correlations and leverage correlations vanish, 
and only one long-ranged dependence mode remains. 
This last contribution can be fitted by an analytical form, that can then be summed up to its own scale, or even to infinity.

\part{Cross-sectional dependences}\label{part:partII}
%\chapter[\textcolor{red}{Correlation matrices: structure and noise}]{Correlation matrices:\\ structure and noise}\label{chap:correl}
%\minitoc
%\include{p2c1-correlation_matrices}
\chapter{The joint distribution of stock returns is not elliptical}\label{chap:IJTAF}
\minitoc

%{red
\section{Introduction}

The most important input of portfolio risk analysis and portfolio optimization 
is the correlation matrix of the different assets. 
In order to diversify away the risk, one must avoid allocating on bundles of correlated assets, 
an information in principle contained in the correlation matrix. 
The seemingly simple mean-variance Markowitz program is however well known to be full of thorns. 
In particular, the empirical determination of large correlation matrices turns out to be difficult, 
and some astute ``cleaning'' schemes must be devised before using it for constructing optimal allocations \cite{ledoit2004well}. 
This topic has been the focus of intense research in the recent years, some inspired by Random Matrix Theory 
(for a review see \cite{elkaroui2009,potters2009financial}) or clustering ideas \cite{marsili2004,tola2008cluster,tumminello2007hierarchically}. 

There are however many situations of practical interest where the (linear) correlation matrix is 
inadequate or insufficient \cite{bouchaud2003theory,malevergne2006extreme,mashal2002beyond}. 
For example, one could be more interested in minimizing the probability of large negative returns 
rather than the variance of the portfolio. 
Another example is the Gamma-risk of option portfolios, where the correlation of the squared-returns of the underlyings is needed. 
Credit derivatives are bets on the probability of simultaneous default of companies or of individuals; 
again, an estimate of the correlation of tail events is required (but probably very hard to ascertain empirically) 
\cite{frey2001modelling}, and for a recent interesting review \cite{brigo}. 

Apart from the case of multivariate Gaussian variables, the description of non-linear dependence 
is not reducible to the linear correlation matrix. The general problem of parameterizing the full 
joint probability distribution of $N$ random variables can be ``factorized'' into the specification 
of all the marginals on the one hand, and of the dependence structure (called the `copula') of 
$N$ standardized variables with uniform distribution in $[0,1]$, on the other hand. 
The nearly trivial statement that all multivariate distributions can be represented in that way is 
called Sklar's Theorem \cite{embrechts02correlation,sklar1959fonctions}. 
Following a typical pattern of mathematical finance, the introduction of copulas ten years ago 
has been followed by a calibration spree, with academics and financial engineers frantically looking 
for copulas to best represent their pet multivariate problem. 
But instead of trying to understand the economic or financial mechanisms that lead to some particular 
dependence between assets and construct copulas that encode these mechanisms, the methodology has been
--- as is sadly often the case --- to brute force calibrate on data copulas straight out from statistics 
handbooks \cite{durrleman2000copula,fermanian2005some,fischer2009empirical,fortin2002tail,patton2001estimation}.
The ``best'' copula is then decided from some quality-of-fit criterion, irrespective of whether the 
copula makes any intuitive sense (some examples are given below). 
This is reminiscent of the `local volatility models' for option markets \cite{dupire1994pricing}: 
although the model makes no intuitive sense and cannot describe the actual dynamics of the underlying asset, 
it is versatile enough to allow the calibration of almost any option smile (see \cite{hagan2002woodward}). 
This explains why this model is heavily used in the financial industry. Unfortunately, a blind calibration of 
some unwarranted model (even when the fit is perfect) is a recipe for disaster. 
If the underlying reality is not captured by the model, it will most likely derail in rough times --- a particularly 
bad feature for risk management! Another way to express our point is to use a Bayesian language: 
there are families of models for which the `prior' likelihood is clearly extremely small 
--- we discuss below the case of Archimedean copulas. 
Statistical tests are not enough --- intuition and plausibility are mandatory. 

The aim of this chapter is to study in depth the family of elliptical copulas, in particular Student copulas, 
that have been much used in a financial context and indeed have a simple intuitive interpretation 
\cite{embrechts02correlation,frahm2003elliptical,hult2002multivariate,luo2009t,malevergne2006extreme,shaw2007copula}. 
We investigate in detail whether or not such copulas can faithfully represent the
joint distribution of the returns of US stocks. (We have also studied other markets as well). 
We unveil clear, systematic discrepancies between our empirical data and the corresponding predictions of elliptical models. 
These deviations are qualitatively the same for different tranches of market capitalizations, 
different time periods and different markets. 
Based on the financial interpretation of elliptical copulas, we argue that such discrepancies are actually expected, 
and propose {the ingredients of} a generalization of elliptical copulas 
that should capture adequately non-linear dependences in stock markets. 
The full study of this generalized model, together with a calibration procedure and stability tests, 
is provided in the next chapter.

%%%%%%%%%%%%%%%%%%%%%%%%%%%%%%%%%%%%%%%%%%%%%%%%%%%%%%%%%%%%%%%%%%%%%%%%
\section{Empirical study of the dependence of stock pairs}\label{sec:empirics}
We carefully compare our comprehensive empirical data on stock markets 
with the predictions of elliptical models, with special emphasis on Student copulas and log-normal copulas.
We conclude that elliptical copulas fail to describe the full dependence structure of equity markets. 

\subsection{Methodology: what we do and what we do not do}

As this chapter's title stresses, the object of this study is to dismiss the elliptical copula as description
of the multivariate dependence structure of stock returns.
Concretely, the null-hypothesis ``$\text{H}_0$: the joint distribution is elliptical''
admits as corollary ``all pairwise bivariate marginal copulas are elliptical and differ only by their linear correlation coefficient ''.
In other words, all pairs with the same linear correlation $\rho$ are supposed to have identical values of non-linear dependence measures, 
and this value is predicted by the elliptical model. Focusing the empirical study on the pairwise measures of dependence, we will reject $\text{H}_0$ by
showing that their average value over all pairs with a given $\rho$ is different from the value predicted by elliptical models.

Our methodology differs from usual hypothesis testing using statistical tools and goodness of fit tests, as can be encountered for example in \cite{malevergne2003testing} 
for testing the Gaussian copula hypothesis on financial assets. Indeed, the results of such tests are often not valid because financial time series are \emph{persistent} (although almost not linearly autocorrelated),
so that successive realizations cannot be seen as independent draws of an underlying distribution, see the recent discussion on this issue in \cite{chicheportiche2011goodness}.

\subsection{Data set and time dependence}

The dataset we considered is composed of daily returns of 1500 stocks labeled in USD in the period 1995--2009 (15 full years). 
We have cut the full period in three sub-periods (1995--1999, 2000--2004 and 2005--2009), and also the stock universe into three
pools (large caps, mid-caps and small caps). We have furthermore extended our analysis to the Japanese stock markets. Qualitatively,
our main conclusions are robust and do not depend neither on the period, nor on the capitalization or the market. Some results, on the
other hand, do depend on the time period, but we never found any strong dependence on the capitalization. 
All measures of dependence are calculated pairwise with non-parametric estimators, and
using all trading dates shared by both equities in the pair.

We first show the evolution of the linear correlation coefficients (Fig.~\ref{fig:rho_evol}) and the upper and lower tail dependence coefficients for $p=0.95$ (Fig.~\ref{fig:tau_evol}) as a function of
time for the large-cap stocks. One notes that (a) the average linear correlation ${\overline \rho}(t)$ fluctuates quite a bit, from around $0.1$ in the mid-nineties to around $0.5$ during 
the 2008 crisis; (b) the distribution of correlation coefficients shifts nearly rigidly around the moving average value ${\overline \rho}(t)$; (c) there is a marked, time 
dependent asymmetry between the average upper and lower tail dependence coefficients; (d) overall, the tail dependence tracks the behavior of ${\overline \rho}(t)$. More
precisely, we show the time behavior of the tail dependence assuming either a Gaussian underlying copula or a Student ($\nu=5$) copula with the same average correlation ${\overline \rho}(t)$.
Note that the former model works quite well when ${\overline \rho}(t)$ is small, whereas the Student model fares better when ${\overline \rho}(t)$ is large. We will repeatedly come back to this point below.

\begin{figure}[p]
\center
\includegraphics[scale=0.42,trim=0 0 -30 0,clip]{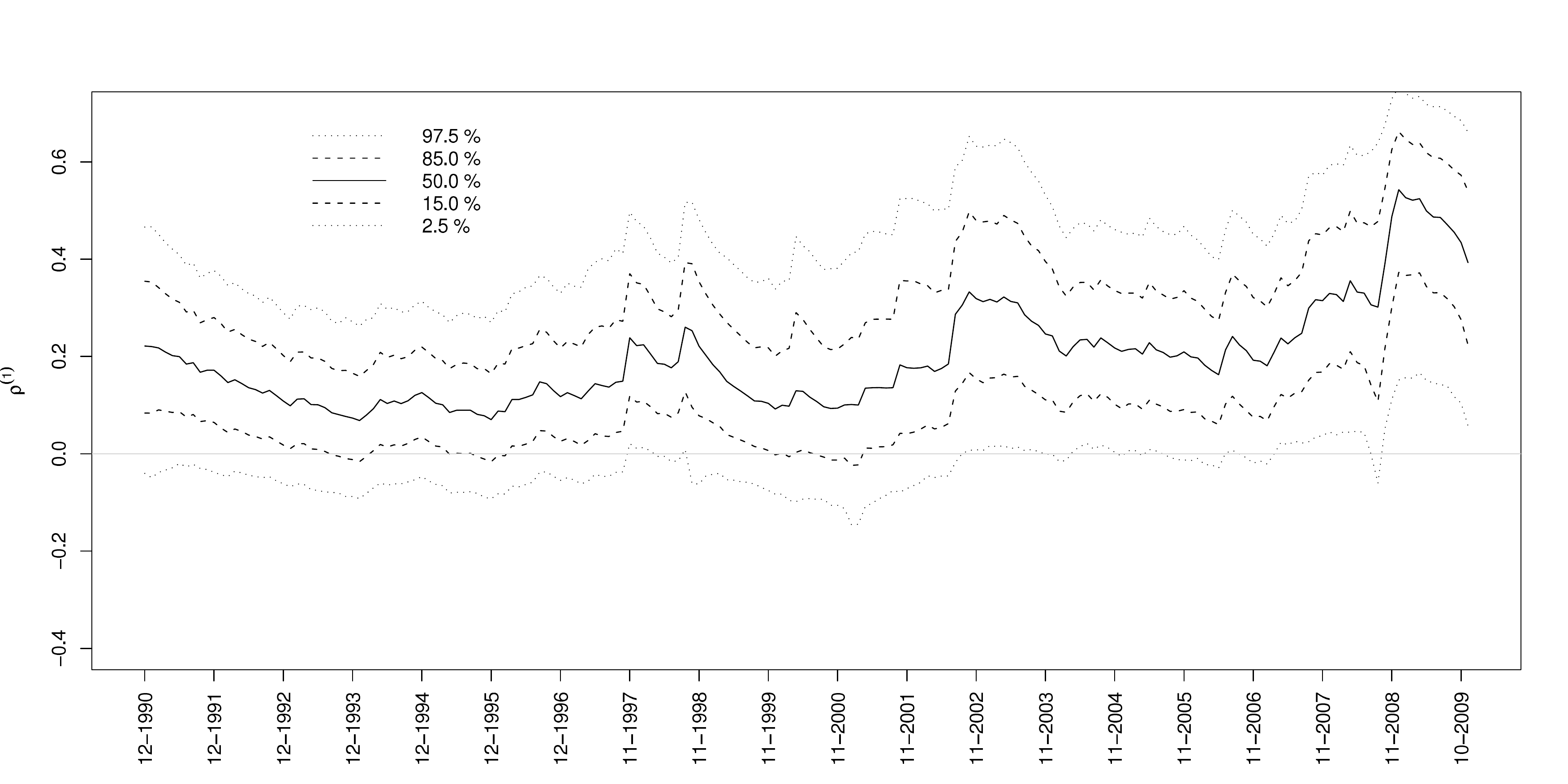}
\caption{Time evolution of the coefficient of linear correlation,
         computed for the stocks in the S\&P500 index, with an exponentially moving average of 125 days
         from January 1990 to December 2009.
         Five quantiles of the $\rho$ distribution are plotted according to the legend, showing that this distribution moves quasi-rigidly around its median value.}
\label{fig:rho_evol}
\end{figure}

\begin{figure}[p]
\center
%\vfill
\includegraphics[scale=0.42,trim=-32 0 -30 0,clip]{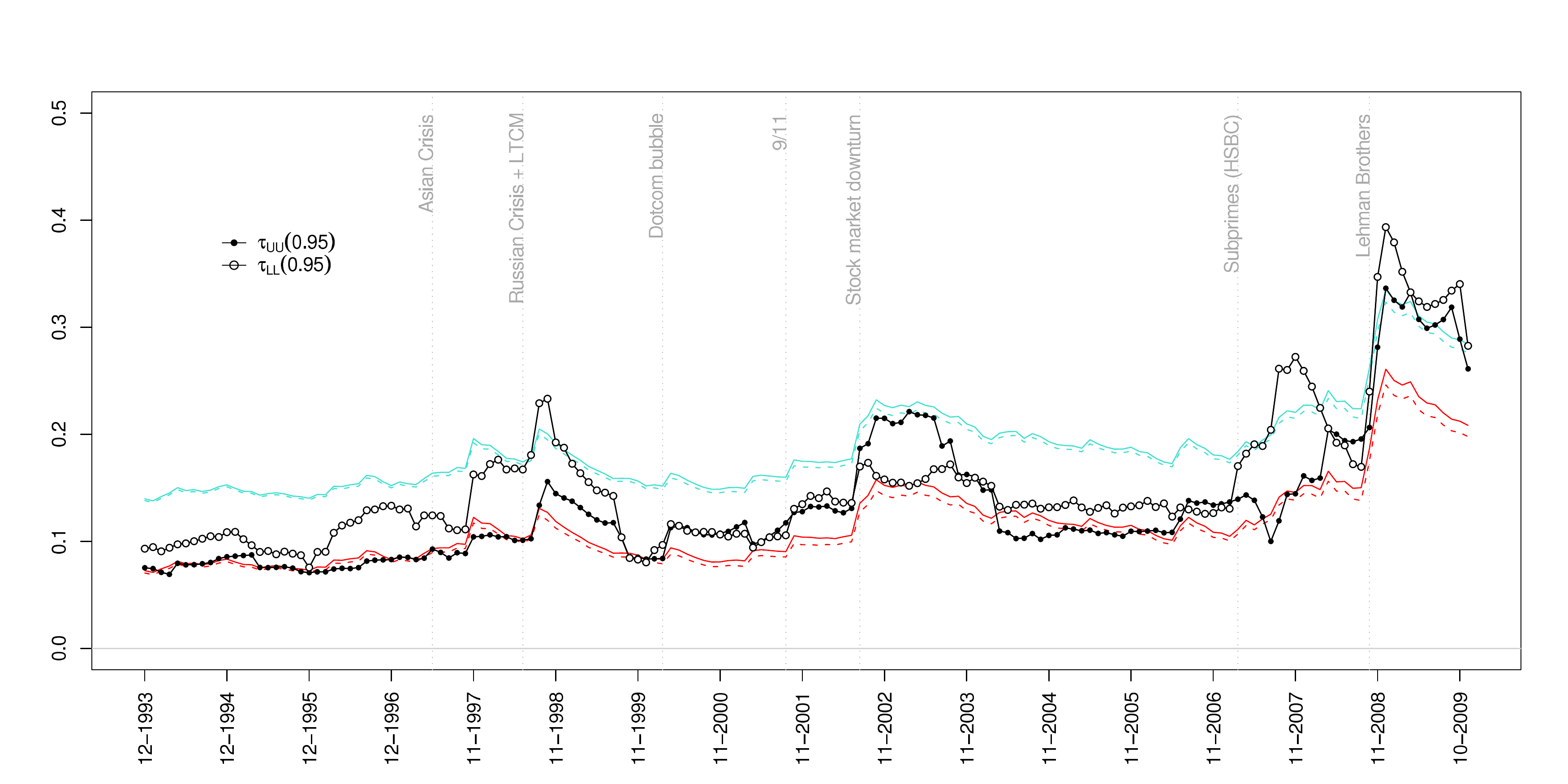}
\caption{Time evolution of the coefficient of tail dependence at $p=0.95$,
         computed for the stocks in the S\&P500 index. 
         A sliding flat window of 250 days moves monthly (25 days) from January 1990 to December 2009.
         We depict in turquoise and red the predictions for Student pair with $\nu=5$ and $\nu=\infty$ respectively, 
         based on the evolution of the mean value $\bar\rho(t)$ of the cross-sectional linear correlation (dashed) or averaged over the full instantaneous distribution of $\rho(t)$ (plain), 
         with little difference.}
\label{fig:tau_evol}
\end{figure}

We computed the quadratic $\roc$ and absolute $\roa$ correlation coefficients for each pair in the pool,
as well as the whole rescaled diagonal and anti-diagonal of the empirical copulas 
(this includes the coefficients of tail dependence\footnote{
The empirical relative bias for $p=0.95$ is $\lesssim{}1\%$ even for series with only $10^3$ points --- 
typically daily returns over 4 years.}, see Part~\ref{part:partI}).
We then average these observables over all pairs with a given linear coefficient $\rho$, within bins of varying width in order to take account of the frequency of observations in each bin. 
We show in Fig.~\ref{fig:all_vs_rho} $\roa$ and $\tUU(0.95), \tLL(0.95)$ as a function of $\rho$ for all stocks in the period 2000--2004, 
together with the prediction of the Student copula model for various values of $\nu$, 
including the Gaussian case $\nu = \infty$. 
For both quantities, we see that the empirical curves systematically cross the Student predictions, 
looking more Gaussian for small $\rho$'s and compatible with Student $\nu \approx 6$ for large $\rho$'s, 
echoing the effect noticed in Fig.~\ref{fig:tau_evol} above. 
The same effect would appear if one compared with the log-normal copula: 
elliptical models imply a residual dependence due to the common volatility, 
\emph{even when the linear correlation goes to zero}. 
This property is not observed empirically, since we rather see that higher-order correlations almost vanish together with $\rho$. 
The assumption that a common volatility factor affects all stocks is therefore too strong,
although it seems to make sense for pairs of stocks with a sufficiently large linear correlation. 
This result is in fact quite intuitive and we will expand on this idea in Sect.~\ref{ssec:pseudo_ell_log} below. 

The above discrepancies with the predictions of elliptical models are qualitatively the same for all periods, market caps and is also found for the Japan data set.
Besides, elliptical models predict a symmetry between upper and lower tails, whereas the data suggest that the tail dependence is asymmetric. Although most of 
the time the lower tail dependence is stronger, there are periods (such as 2002) when the asymmetry is reversed.

\begin{figure}[t!h]\center
    \subfigure[\mbox{$\roa$       vs $\rho$}]{\label{fig:abscor_rho_data}\includegraphics[scale=0.5]{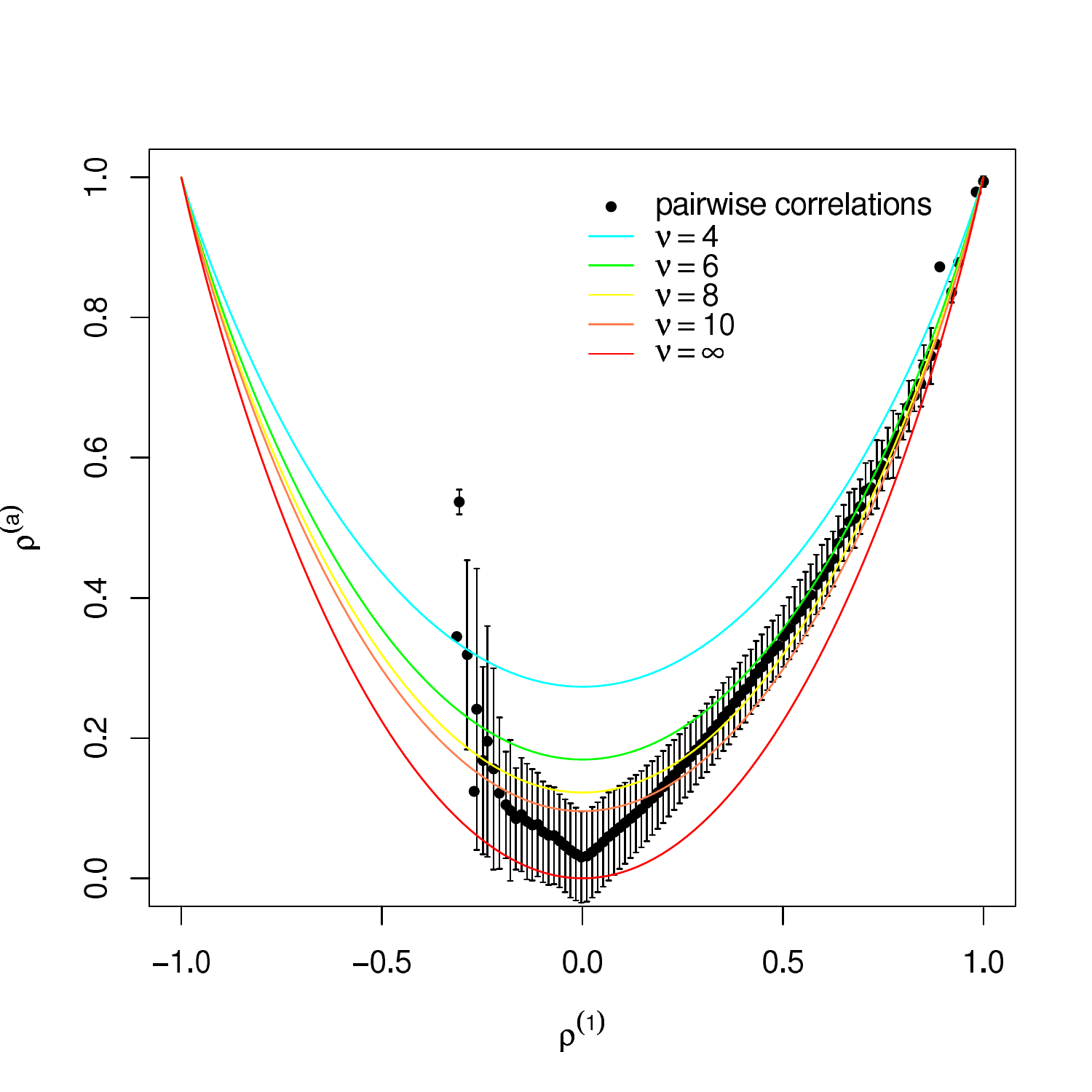}}\hfill
    \subfigure[\mbox{$\tau(0.95)$ vs $\rho$}]{\label{fig:tau_rho_data   }\includegraphics[scale=0.5]{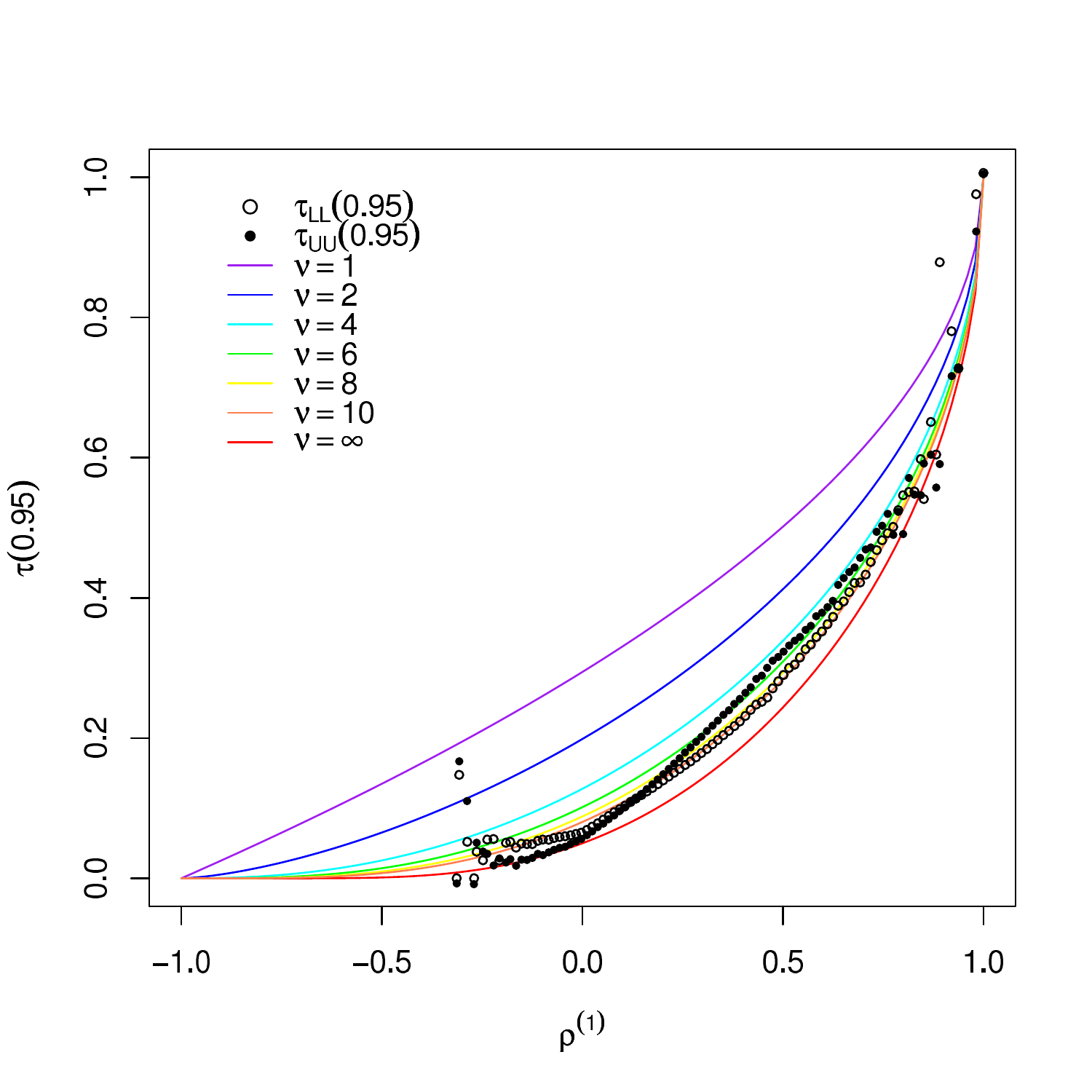}}
    \caption{Empirical (2000--2004) and elliptical. The fact that the empirical points ``cross'' the theoretical lines
  corresponding to elliptical models reveals that there is a dependence structure in the volatility, more complex
  than what is suggested by a single common stochastic scale factor.}	
  \label{fig:all_vs_rho}
\end{figure}

\subsection{Copula diagonals}

As argued in Chapter~\ref{part:partI}.\ref{chap:statdep} above, it is convenient to visualize the rescaled difference $\Delta_{\scriptscriptstyle \text{d,a}}(p)$ 
between the empirical copula and the Gaussian copula, along the diagonal $(p,p)$ and the anti-diagonal $(p,1\!-\!p)$, see Eqs.~\eqref{eq:Delta} page~\pageref{eq:Delta}. 
The central point $p=\frac12$ is special: $\Delta_{\scriptscriptstyle \text{d,a}}(\frac12)=\betaB-\frac{2}{\pi}\arcsin\rod[1]$ 
is zero for all pseudo-elliptical models. 

We show in Fig.~\ref{fig:cop_dev} the diagonal and anti-diagonal copulas for all pairs of stocks in the period 2000--2004, 
and for various values of $\rho$ (similar plots for other periods are available in Ref.~\cite{chicheportiche2010master}). 
We also show the prediction of the Student $\nu=5$ model and of a Frank copula model with Student $\nu=5$ marginals and the appropriate value of $\rho$. 
What is very striking in the data is that $\Delta_{\scriptscriptstyle \text{d}}(p)$ is concave for small $\rho$'s and becomes convex for large $\rho$'s, 
whereas the Student copula diagonal is {\it always convex}. 
This trend is observed for all periods, and all caps, and is qualitatively very similar in Japan as well for all periods between 1991 and 2009. 
We again find that the Student copula is a reasonable representation of the data only for large enough $\rho$. 
The Frank copula is always a very bad approximation -- see the wrong curvature along the diagonal and the inaccurate behavior along the anti-diagonal.

Let us now turn to the central point of the copula, $\cop^*=\cop(\frac12,\frac12)$. 
We plot in Fig.~\ref{fig:C05} the quantity $-\cos\left(2\pi \cop^*\right)-\rho$ as a function of $\rho$, 
which should be zero for all elliptical models, as discussed in pages~\pageref{eq:ell_betaB} and~\pageref{page:invariantbeta}.
The data here include the 284 equities constantly member of the S\&P500 index in the period {2000--2009}. 

We again find a clear systematic discrepancy, that becomes stronger for smaller $\rho$'s: 
the empirical value of $\cop^*$ is too large compared with the elliptical prediction. 
In particular, for stocks with zero linear correlation ($\rho = 0$), we find $\cop^* > \frac14$, 
i.e.\ even when two stocks are uncorrelated, the probability that they move in the same direction\footnote{
A more correct statement is that both stocks have returns below their median with probability larger than $\frac14$. 
However, the median of the distributions are very close to zero, justifying our slight abuse of language.} is larger than $\frac14$. 
The bias shown in Fig.~\ref{fig:C05} is again found for all periods and all market caps, and for Japanese stocks as well. 
Only the amplitude of the effect is seen to change across the different data sets.

Statistical errors in each bin are difficult to estimate since pairs containing the same asset are mechanically correlated. 
In order to ascertain the significance of the previous finding, we have compared the empirical value of $\cop^*(\rho)$ with the result of a numerical
simulation, where we generate time series of Student ($\nu=5$) returns using the empirical correlation matrix. 
In this case, the expected result is that all pairs with equal correlation have the same bivariate copula and thus the same $\cop^*$ 
so that the dispersion of the results gives an estimate of measurement noise. 
We find that, as expected, $\cop^*(\rho)$ is compatible with the Student prediction, at variance with empirical results.
%The vertical dispersion of the scatter plot measures how much the numerous pairs in each bin of $\rho$ are dissimilar in term of bivariate copulas.
We also find that the dispersion of the empirical points is significantly larger than that of the simulated elliptical pairs with identical linear correlations,
suggesting that all pairs {\it cannot be described by the same bivariate copula} (and definitively not an elliptical copula, as argued above).

All these observations, and in particular the last one, clearly indicate that Student copulas, or any elliptical copulas, are inadequate to represent the full 
dependence structure in stock markets.%\footnote{Clearly, if the bivariate copula is not elliptical, the multivariate copula is not elliptical either.}. 
 Because this class of copulas has a transparent interpretation, we in fact know why this is the case: assuming a common random volatility 
factor for all stocks is oversimplified. This hypothesis is indeed only plausible for sufficiently correlated stocks, in agreement with the set of observations we made above. As a
first step to relax this hypothesis, we now turn to the pseudo-elliptical log-normal model, that allows further insights but still has unrecoverable failures.

\begin{figure}[tp]
	\center\hfill
  	\subfigure[]{\label{fig:C05a}\includegraphics[scale=0.33,clip]{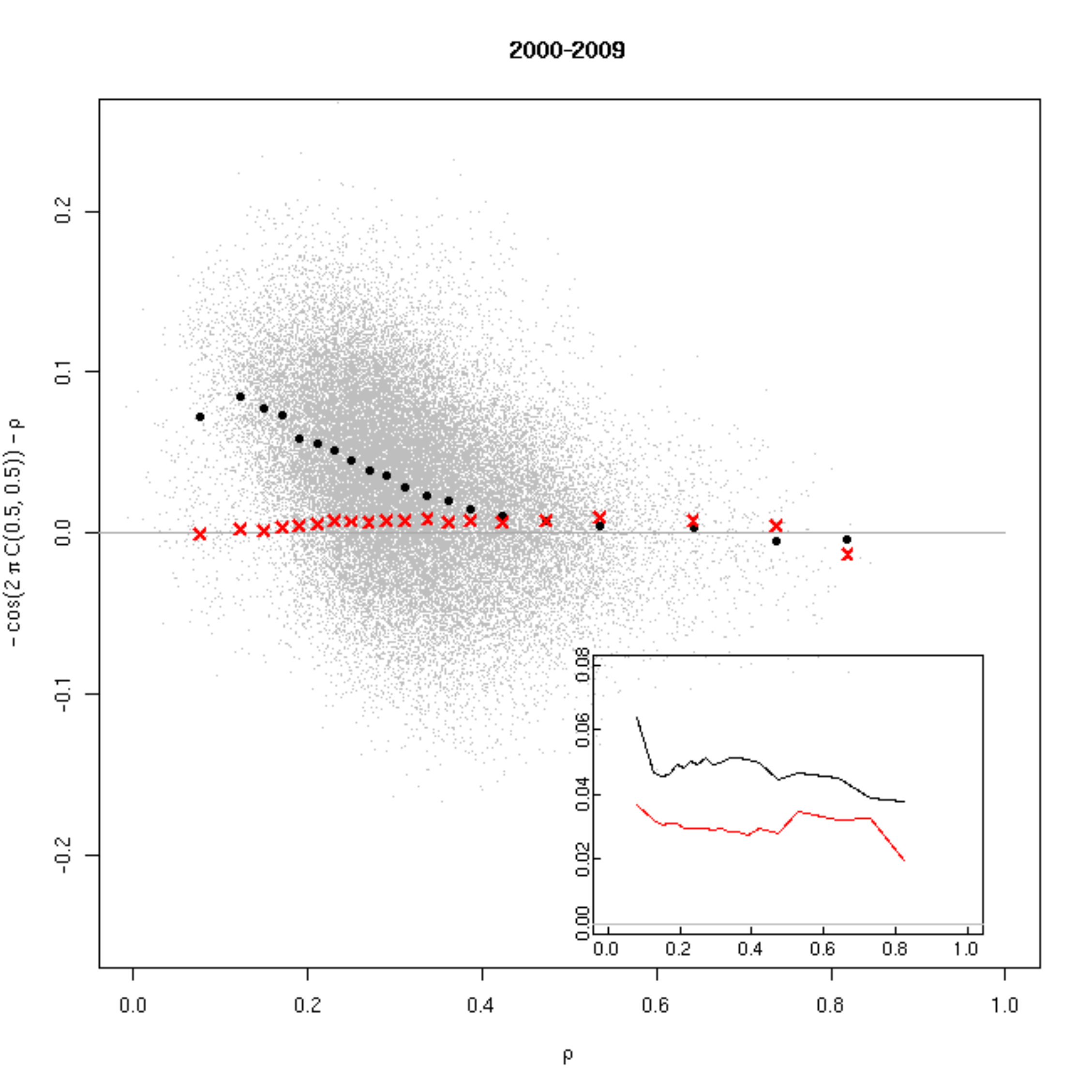}}\hfill
  	\subfigure[]{\label{fig:C05b}\includegraphics[scale=0.33,clip]{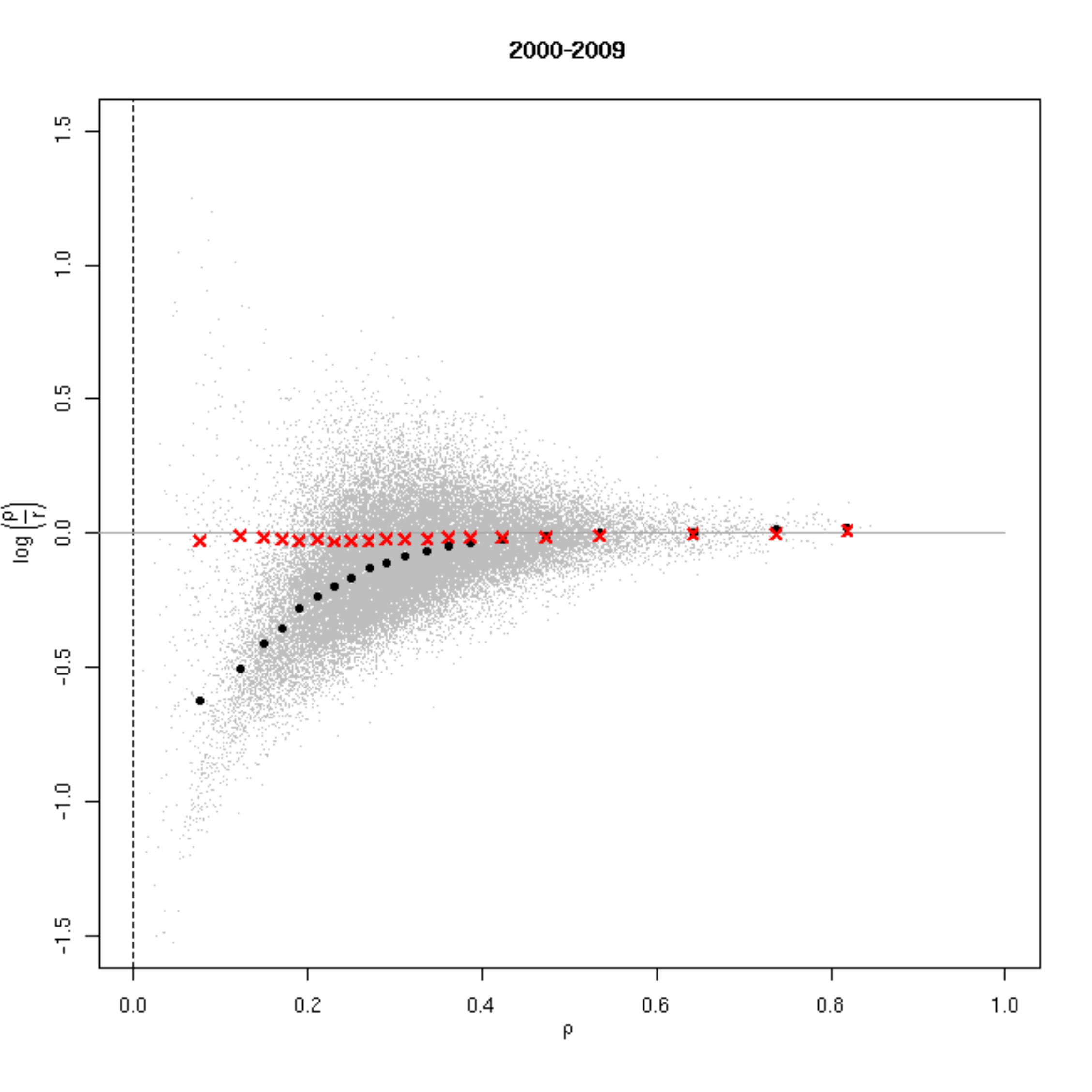}}\hfill
	\caption{{\textbf{Left:} $-(\cos\left(2\pi \cop^*\right)+\rho)$ versus the coefficient $\rho$ of linear correlation.
	The grey cloud is a scatter plot of all pairs of stocks.
	Black dots correspond to the average value of empirical measurements within a correlation bin,
	and red crosses correspond to the numerically generated elliptical (Student, $\nu=5$) data, compatible with the 
	expected value zero. 
	Statistical error are not shown, but their amplitudes is of the order of the fluctuations of the red crosses, 
	so that empirics and elliptical prediction don't overlap for $\rho\lesssim 0.4$.
	This representation allows to visualize very clearly the systematic discrepancies for small values of $\rho$, beyond the average value.
	Inset: the 1\,s.d. dispersion of the scattered points inside the bins. 
	\textbf{Right:} 
	\mbox{$z=\ln \left(\rho/|\cos\left(2\pi \cop^*\right)|\right)$} as a function of $\rho$, with the same data and the same conventions. 
	This shows more	directly that the volatility correlations of weakly correlated stocks is overestimated by elliptical models, whereas strongly 
	correlated stocks are compatible with the elliptical assumption of a common volatility factor.}}
	\label{fig:C05}
\end{figure}

\begin{figure}[!p!]
  \center
  \includegraphics[scale=0.45,clip]{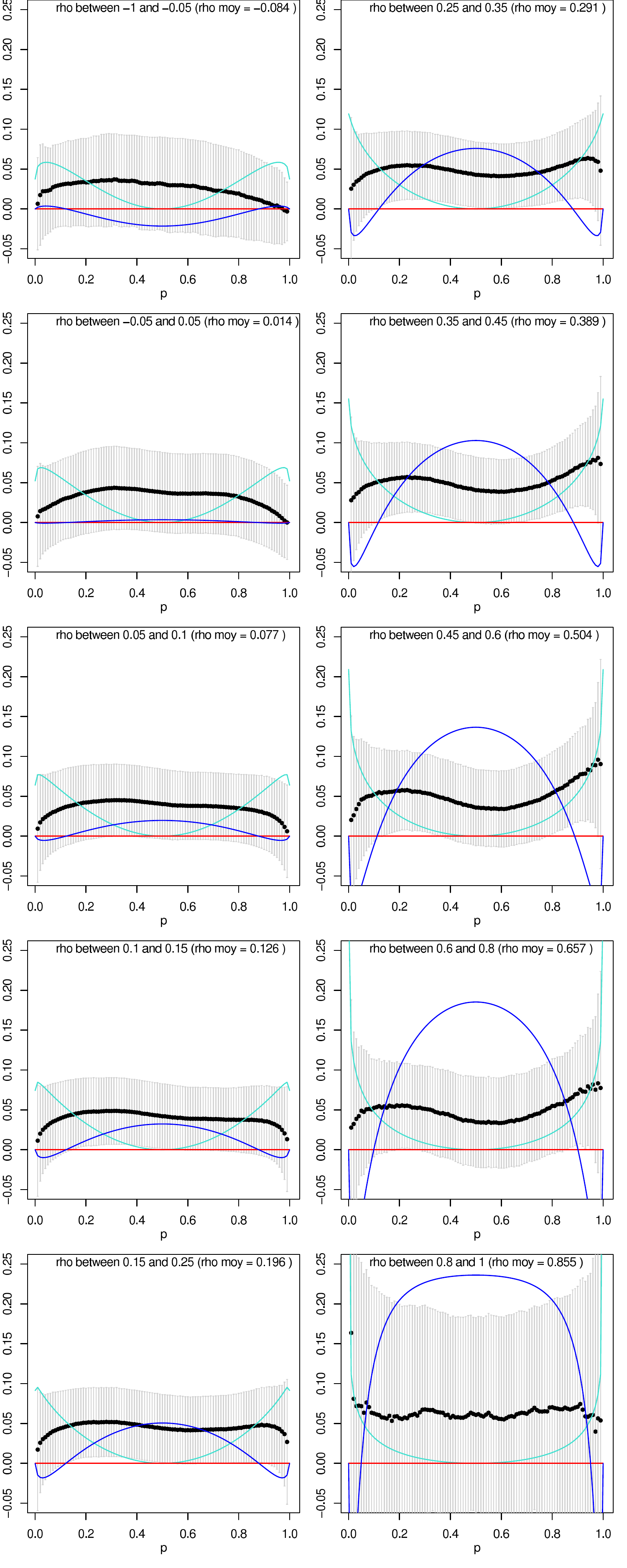}\hfill
  \includegraphics[scale=0.45,clip]{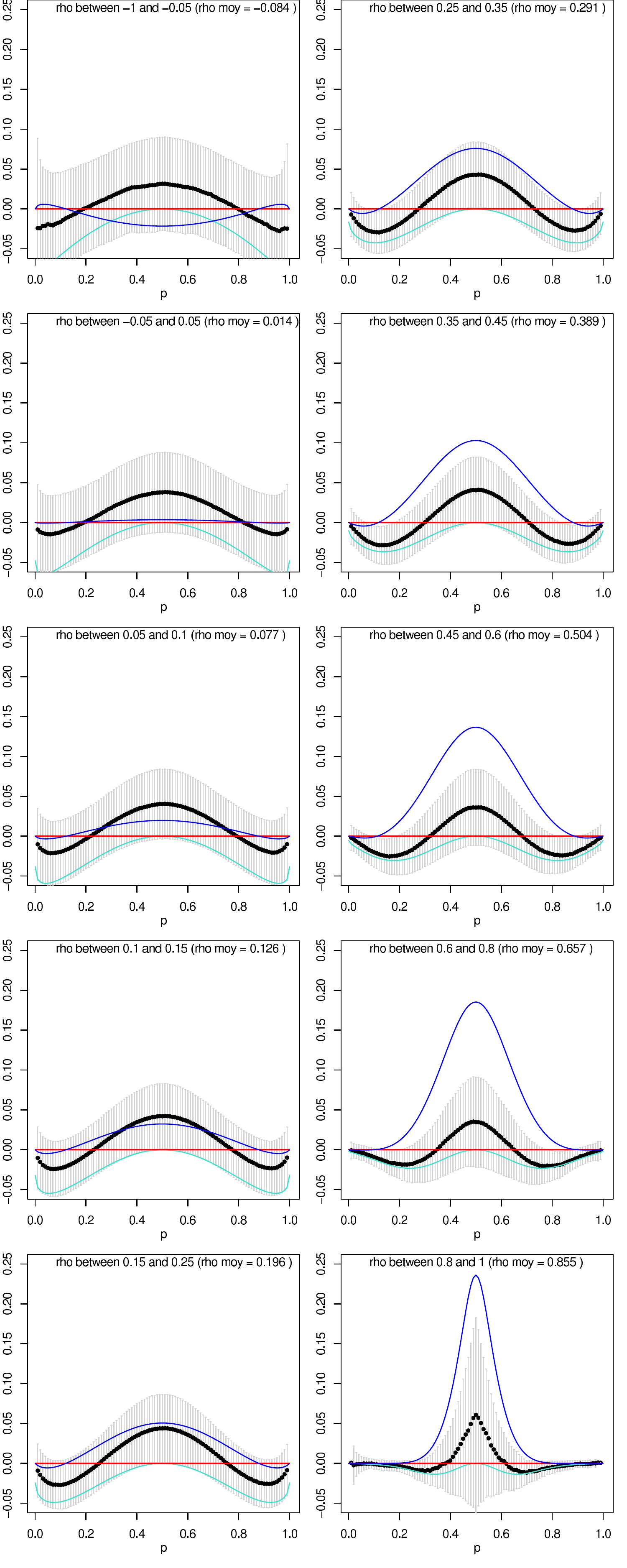}
  \caption{Diagonal (left) and anti-diagonal (right) of the copula vs $p$.
        Real data (2000-2004) averaged in 10 bins of $\rho$. The vertical lines corresponds to 1-sigma dispersion of the results (i.e.\ not 
        the statistical error bar).
	The turquoise line corresponds to an elliptical model (Student with $\nu=5$):
	it predicts too low a value at the center point and too large tail dependences (left and right limits).
	The blue line depicts the behavior of Frank's copula.
	It is worth noticing that the tail dependence coefficient is always measured empirically at some $p<1$, for example $p=0.95$.
	In this case, the disagreement between data and an elliptical model may be accidentally small.
	It is important to distinguish the interpretation of the limit coefficient from that of a penultimate value [6].}%{\cite{coles1999dependence}}.}
  \label{fig:cop_dev}
\end{figure}

\subsection{Interpretation}

The incentive to focus on bivariate measures comes from the theoretical 
property that all the marginals, including bivariate, of a multivariate elliptical distribution are themselves elliptical.
In turn, a motivated statement that the pairwise distributions are not elliptical is enough to claim the non-ellipticity of the joint multivariate distribution.
This is the basic line of argumentation of the present paper.

In a financial context, elliptical models are basically stochastic volatility models with arbitrary dynamics such that the 
ergodic distribution is $\pdf[\sigma]$ (since we model only single-time distributions, the time ordering is irrelevant here).
The important assumption, however, is that the random amplitude factor $\sigma$ is \emph{equal for all stocks} (or all assets in 
a more general context). In other words, one assumes that the mechanisms leading to increased levels of volatility 
affect all individual stocks identically. There is a unique market volatility factor. Of course, this is a very restrictive 
assumption since one should expect {\it a priori} other, sector specific, sources of volatility. This, we believe, is the 
main reason for the discrepancies between elliptical models and the empirical data reported below.

\subsection{Pseudo-elliptical log-normal model}\label{ssec:pseudo_ell_log}

As we just showed, the fact that the central value of the copula $\cop^*$ does not obey the relation 
$-\cos \left(2\pi \cop^*\right) = \rho$ rules out all elliptical models. 
One possible way out is to consider a model where random volatilities are stock dependent, as explained in Sect.~\ref{ssec:pseudo-ell}. 
Choosing for simplicity a log-normal model of correlated volatilities and inserting \eqref{eq:fd_lognorm} into \eqref{eq:coeffs_pseudo_ell}, the new prediction is:
\be
%\rho = - \e^{s^2(\gamma - 1)} \cos\left(2\pi \cop^*\right),
\rho = - \e^{s^2(c - 1)} \cos\left(2\pi \cop^*\right),
\ee
%where $\gamma$ is the correlation of the log-volatilities. 
where $c$ is the correlation of the log-volatilities. 
This suggests to plot \mbox{$z=\ln \left(\rho/|\cos\left(2\pi \cop^*\right)|\right)$} as a function of $\rho$, as shown in Fig.~\ref{fig:C05b}. 
For a purely elliptical model, $z$ should be identically zero, corresponding to perfectly correlated volatilities ($c=1$).
%($\gamma=1$). 
What we observe in Fig.~\ref{fig:C05b}, on the other hand, 
%is that $\gamma$ is indeed close to unity for large enough $\rho$'s, 
is that $c$ is indeed close to unity for large enough $\rho$'s, 
but systematically decreases as $\rho$ decreases; 
in other words, the volatilities of weakly correlated stocks are themselves weakly correlated. 
This is, again, in line with the conclusion we reached above. 

However, this pseudo-elliptical model still predicts that $\cop^* = \frac14$ for stocks with zero linear correlations, 
in disagreement with the data shown in Fig.~\ref{fig:C05a}. 
This translates, in Fig.~\ref{fig:C05b}, into a negative divergence of $z$ when $\rho \to 0$. 
This finding means that we should look for other types of constructions. 
How can one have at the same time $\rho = 0$ and $\cop^* > \frac14$ ? 
A toy-model, that serves as the basis for a much richer model that we will report elsewhere \cite{remi}, is the following. 
Consider two independent, symmetrically distributed random factors $\psi_1$ and $\psi_2$ with equal volatilities, 
and construct the returns of assets $1$ and $2$ as:
\be
X_{1,2} = \psi_1 \pm \psi_2.
\ee
Clearly, the linear correlation is zero. 
Using a cumulant expansion, one easily finds that to the kurtosis order, the central value of the copula is equal to:
\be
\cop^* = \frac14 + \frac{\kappa_2 - \kappa_1}{24 \pi} + \dots
\ee
Therefore, if the kurtosis $\kappa_1$ of the factor to which both stocks are positively exposed is smaller than the kurtosis $\kappa_2$ of the spread, 
one does indeed find $\cop^*(\rho=0) > \frac14$. 

We will investigate in Chapter~\ref{chap:multifact} an additive model such as the above one, 
with random volatilities affecting the factors $\psi_1$ and $\psi_2$, 
that generalize elliptical models in a way to capture the presence of several volatility modes.

%%%%%%%%%%%%%%%%%%%%%%%%%%%%%%%%%%%%%%%%%%%%%%%%%%%%%%%%%%%%%%%%%%%%%%%%
\section{Conclusion}\label{sec:conclusion}

The object of this chapter was to discuss the adequacy of Student copulas, or more generally elliptical copulas, 
to describe the multivariate distribution of stock returns. 
We have shown, using a very large data set, with the daily returns of 1500 US stocks over 15 years, 
that elliptical models fail to capture the detailed structure of stock dependences. 

In a nutshell, the main message elicited by our analysis is that Student copulas provide a good approximation 
to describe the joint distribution of strongly correlated pairs of stocks, but badly miss their target for weakly correlated stocks. 
We believe that the same results hold for a wider class of assets: 
it is plausible that highly correlated assets do indeed share the same risk factor. 
Intuitively, the failure of elliptical models to describe weakly correlated assets can be traced to the 
inadequacy of the assumption of a single ``market''  volatility mode.
We expect that exactly as for returns, several factors are needed to capture sectorial volatility modes and idiosyncratic modes as well. 
The precise way to encode this idea into a workable model that faithfully captures the non-linear 
dependence in stock markets is reported in the next chapter. 
We strongly believe that such a quest cannot be based on a formal construction of mathematically convenient models 
(such as Archimedean copulas that, in our opinion, cannot be relevant to describe asset returns). 
The way forward is to rely on intuition and plausibility and come up with models that make financial sense. 
This is, of course, not restricted to copula modeling, but applies to all quarters of quantitative finance.

\chapter[Volatility dependences: a factor model]{A minimal factor model\\ for volatility dependences\\ in stock returns}\label{chap:multifact}
\minitoc
% \def\Wei{\beta}                      %\mat{X}
% \def\ret{x}                          %r
% \def\Ret{\mat{\MakeUppercase{\ret}}} %\mat{R}
%%%%%%%%%%%%%%%%%%%%%%%%%%%%%%%%%%%%%%%%%%%%%%%%%%%%%%%%%%%%%%%%%%%%%%%%
\section{Introduction}
Dependences among financial assets or asset classes stand at the heart of 
modern portfolio selection theories.
Whatever the (concave) utility of an investor and its risk measure, 
diversification is profitable but optimal diversification is only reached 
if the underlying dependence structure is well understood.

\subsubsection*{Financial risk and portfolio selection: the linear covariances}
The standard Markowitz theory \cite{markowitz1952portfolio,markowitz1959portfolio,bouchaud2003theory} of optimal portfolio design 
aims at finding the optimal weights $w_i$ to attribute to each stock of a pool. 
It assumes that stock returns are correlated random variables $\ret_i$, and that
the optimizing agent has a ``mean-variance'' quadratic utility function in the form $U(\vect{w})=\esp{\vect{\ret}\cdot\vect{w}}-\mu\var{\vect{\ret}\cdot\vect{w}}$.
It hence relies on the linear covariance matrix $\rho=\esp{\vect{\ret}\vect{\ret}^\dagger}$
of the stock returns, and more importantly on its inverse $\rho^{-1}$.
Indeed, with no further constraints (budget, transaction costs, operational risk constraint, prohibition of short selling, etc.), 
the optimal weights are given by
\[
    \vect{w}_\rho^*\propto {\rho^{-1}\vect{g}}=\vect{g}+\mat{V}\left(\Lambda^{-1}-\mathds{1}\right)\mat{V}^\dagger\vect{g}
\]
where $\vect{g}$ is the vector of gain targets for the assets in the basket, 
and $\rho=\mat{V}\Lambda\mat{V}^\dagger$ is the spectral decomposition of the covariance matrix
with $\mat{V}$ being the square matrix of eigenvectors and $\Lambda$ the diagonal matrix of eigenvalues.
Empirical estimates of $\rho$ and its spectrum $\Lambda$ are typically very noisy, 
and cleaning schemes need to be applied before inversion if one wants to avoid 
artificially enhancing the weights of low-risk in-sample modes (i.e.\ with eigenvalue smaller than 1)
that turn into high-risk realized out-of-sample modes.

All this is fairly standard practice now, and several cleaning schemes have been designed, 
in view of modeling either the signal (parametric models, factor models, Principal Components Analysis),
 or the noise (RMT-based \cite{laloux1999noise, laloux2000random, ledoit2004well,potters2005financial,bartz2012directional}).
 
\subsubsection*{Non-linear dependences and risk}
However, it is now established that markets operate beyond the linear regime, that implicitly assumes gaussianity.

For one thing, individual stock returns are known to be non-Gaussian, 
and moments beyond the mean and variance have gained considerable interest (e.g.\ the excess kurtosis, or low-moment estimates thereof).
But more importantly, stock returns are \emph{jointly} not Gaussian: 
the structure of dependence between pairs of stocks is not compatible with the Gaussian copula,
and as a consequence the penalty in the utility function should be more subtle than just
the portfolio variance and include non-linear measures of risk (like tail events, quadratic correlations, etc.) in order to better fit the agent's risk aversion profile.
Only in a multivariate Gaussian setting do these non-linear dependences resume to the linear correlations.
%\begin{figure}[h]
%    \includegraphics[scale=0.5]{plot_correlations}
%\end{figure}

Non-linear dependences are also very important in the pricing and risk management of structured products and portfolios of derivatives.
For example, the payoff of a hedged option has a V-shape with linear asymptotes and quadratic core, see Fig.~\ref{fig:payoff_option}.
A portfolio of several such hedged options has thus a variance characterized by the absolute and quadratic correlations of the underlying stocks (gamma risk).
These correlations of amplitudes are noisier than linear correlations, 
whence the need for a reliable model of both linear and non-linear dependences.

\begin{figure}
    \center
    \includegraphics[scale=.6,trim=0 15 0 0,clip]{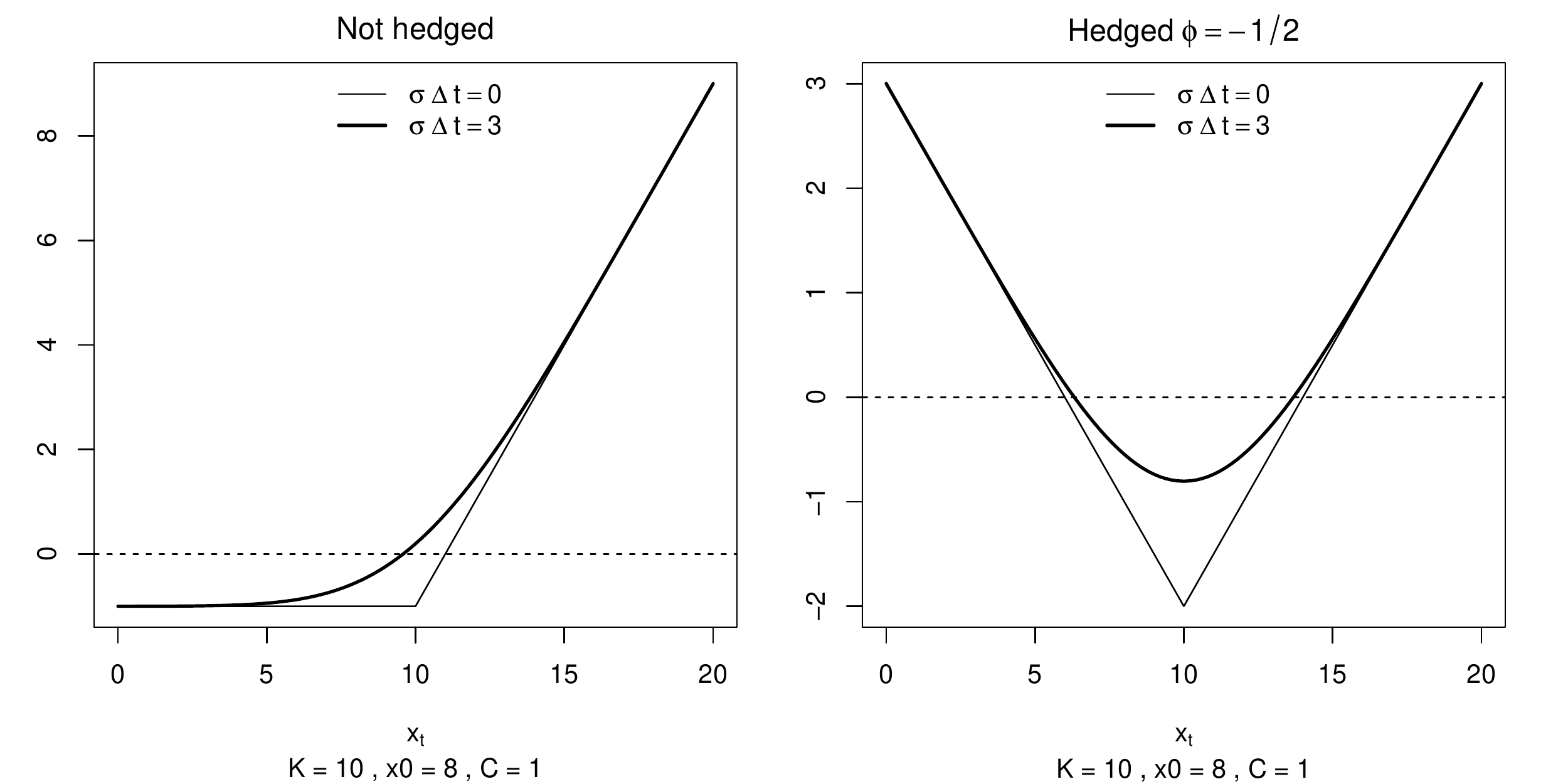}
    \caption{Expected payoff of an option as a function of the current price of the underlying stock.
             \textbf{Left:} unhedged; \textbf{Right:} hedged by short selling $\phi=1/2$ shares of the underlying.
             The illustration is for a call option of strike $x_K=10$, price $\mathcal{C}=1$ on a stock of initial price $x_0=8$ following a Bachelier diffusion with volatility $\sigma$.
             The thin line is the payoff at expiry ($\sigma\,\Delta t=0$) and the thick curve is the expected value of the payoff before expiry ($\sigma\,\Delta t=3$). }
    \label{fig:payoff_option}
\end{figure}

%\subsubsection*{Literature}
\paragraph{}
In Chapter~\ref{chap:IJTAF}, we showed that the daily returns of stocks are not
exposed to a unique mode of volatility affecting all individuals at once.
We thus ruled out all models with a single stochastic volatility $\sigma$, of the form
\begin{equation}\label{eq:ellipt}
    \ret_i=\sigma\,\epsilon_i,
\end{equation}
with jointly Gaussian (and correlated) residuals $\epsilon_i$'s.
This, we argued, was revealing a finer structure in the non-linear dependences,
and opened the way for a description taking into account several modes of volatility.
However, we also showed that any description in the form of individual volatilities 
\begin{equation}\label{eq:pseudo_ellipt}
    \ret_i=\sigma_i\,\epsilon_i,
\end{equation}
with possible dependences between the $\sigma_i$'s,
would not be able to explain successfully the departure of $\cop(\frac12,\frac12)$ from the elliptic prediction either.
We called models like Eq.~\eqref{eq:pseudo_ellipt} ``pseudo-elliptical'', 
with one typical example being the multiplicative decomposition of $\sigma_i$
onto the market volatility $\sigma$, a sectorial volatility $\hat\sigma_s$ (where stock $i$ belongs to sector $s$), 
and a residual volatility $\widetilde\sigma_i$:
\[
    \sigma_i=\sigma\hat\sigma_s\widetilde\sigma_i.
\]
It hence turned out that tuning only the radial amplitudes of a Gaussian vector
would not be enough to reproduce some non-linear dependences relying on the \emph{ranks} of the realizations in the data.
Instead, additive non-Gaussian factors are thought to be able to generate anomalous copula values,
because of the interplay of factor kurtosis and residual kurtosis,
as motivated by the toy model for $\cop(\frac12,\frac12)>0$ with uncorrelated variables,
presented in the previous chapter, page~\pageref{ssec:pseudo_ell_log}.

%%% DANS L'ARTICLE IL FAUDRA DEFINIR C

All pseudo-empirical models have a simple prediction linking the medial copula to the coefficient of linear correlation, 
see the discussion in page~\pageref{page:invariantbeta}:%\cite{chicheportiche2012joint}:
\begin{equation}
    \cop(\tfrac12,\tfrac12)=\frac{1}{4}+\frac{1}{2\pi}\arcsin\,\rho
\end{equation}
Said differently, the effective correlation%
\footnote{The superscript (B) stands for ``Blomqvist'', as $\rhoB$ is related to 
Blomqvist's beta coefficient. 
The properties of $\rhoB$ were discussed in Chapter~\ref{part:partI}.\ref{chap:statdep},
see in particular the definition~\eqref{eq:def_rhoB} and the comments in page~\pageref{eq:def_rhoB}.}%
\begin{equation}\label{eq:hatrho}
    \rhoB\equiv\cos\!\left(2\pi \cop(\tfrac12,\tfrac12)\right)
\end{equation}
is equal to $\rho$ for these models, whereas in general it is not.

\subsubsection*{Summary of empirical stylized facts to be addressed by the model}
More generally, any model for a joint description of stock returns must 
address the stylized facts:
\begin{itemize}
    \item Generate fat-tailed return series, and even non-Gaussian factors and residuals
    \item Reproduce the structure of linear correlations with a reduced number of factors. 
          In particular exhibit a market mode of linear dependences.
    \item Allow for a dependence between the volatilities of the residuals and the volatilities of the factors \cite{cizeau2001correlation,allez2011individual}.
    \item Reproduce the anomalous copula structure (diagonals, and in particular medial point), see Fig.~\ref{fig:emp_cop}.
          It was already noted in chapter~\ref{chap:IJTAF} that stock pairs with a high correlation are ``more elliptical'',
          and that in periods of high turmoil like the financial crisis (b) stock pairs are both more correlated and more elliptical,
          revealing a strong exposure to a common mode of volatility.
          \begin{figure}
          \center
          \subfigure[2000--2004]{\includegraphics[scale=.38,trim=0 0 0 25,clip]{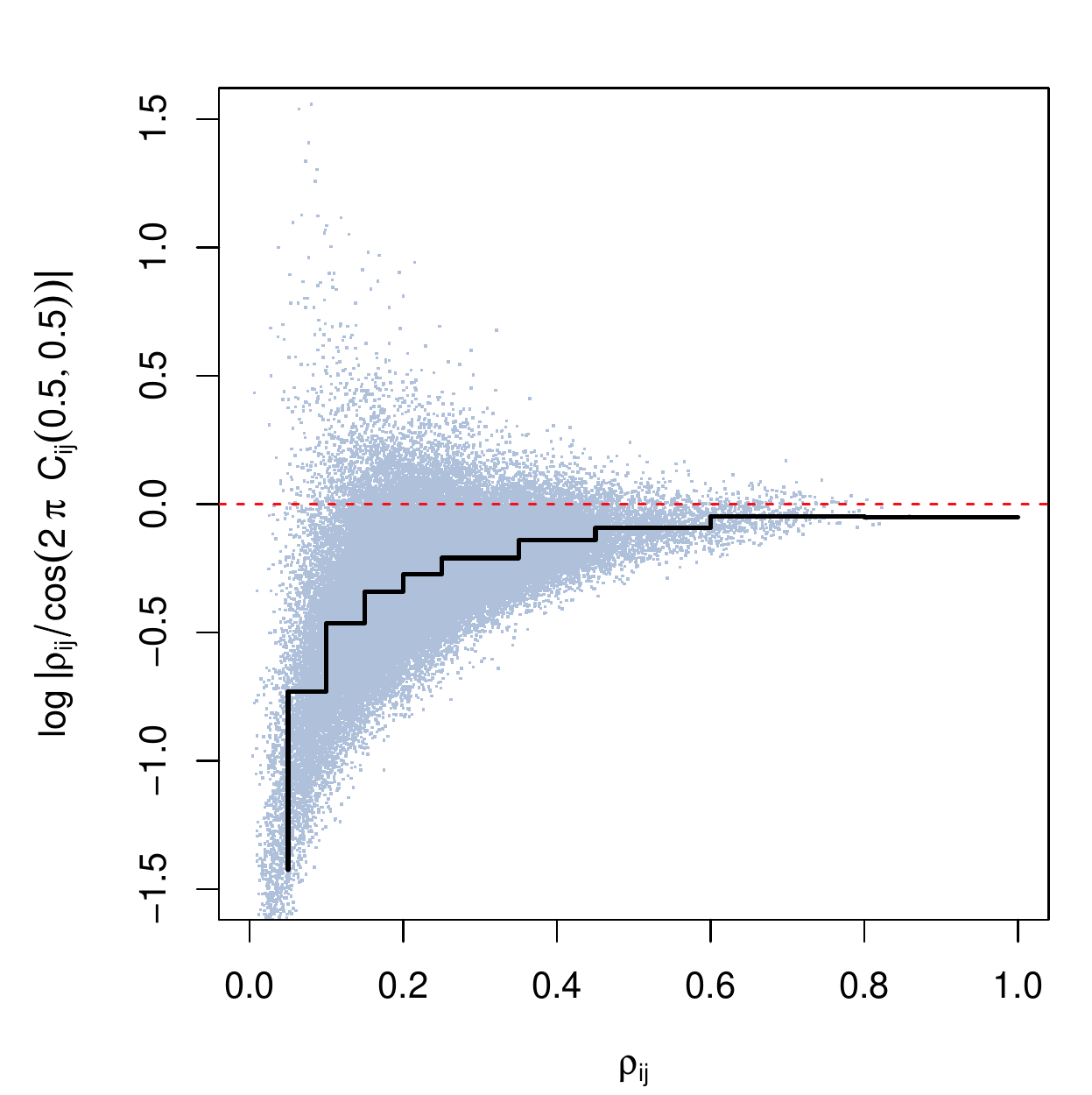}}
          \subfigure[2005--2009]{\includegraphics[scale=.38,trim=0 0 0 25,clip]{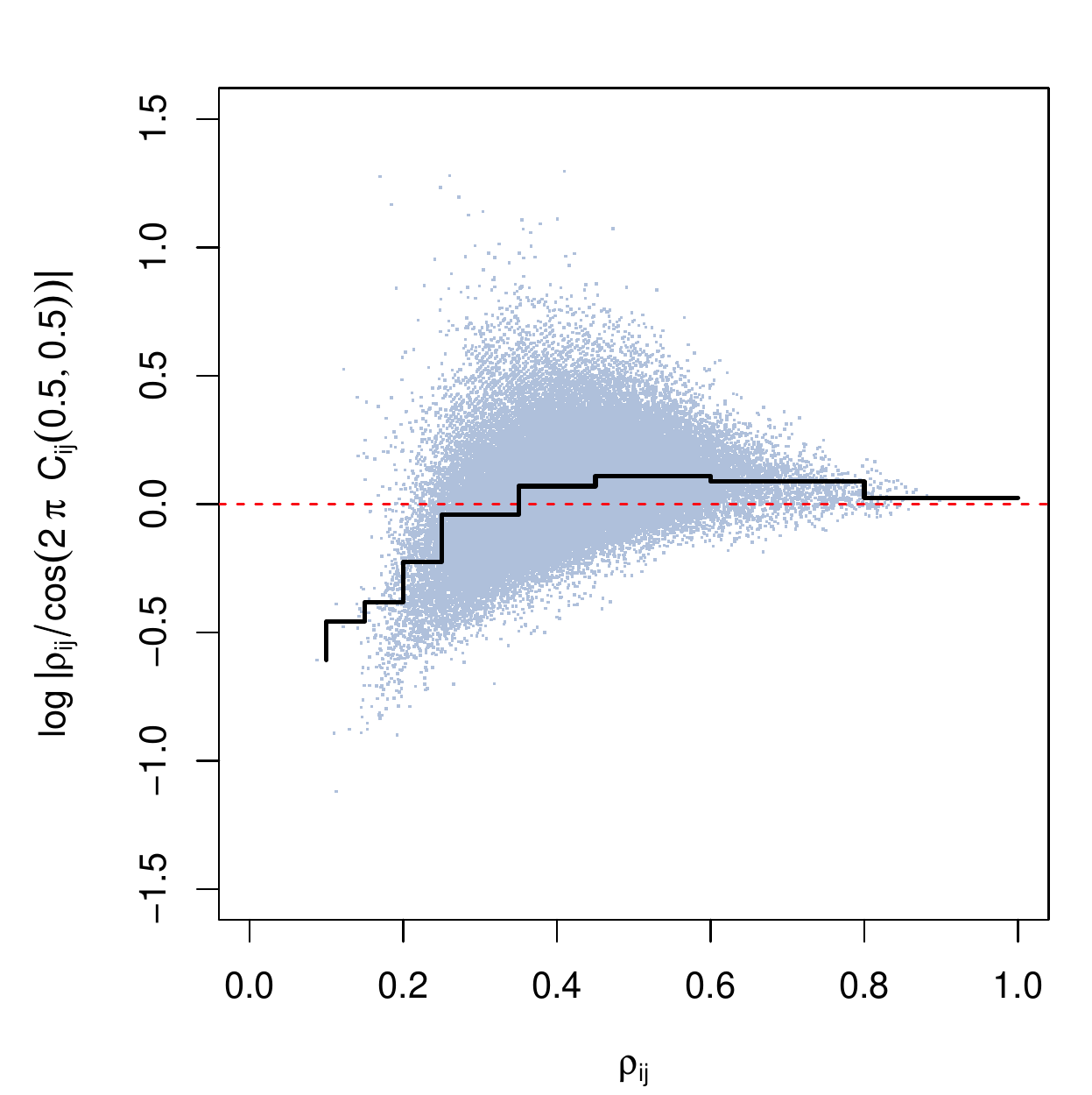}}
          \subfigure[2009--2012]{\includegraphics[scale=.38,trim=0 0 0 25,clip]{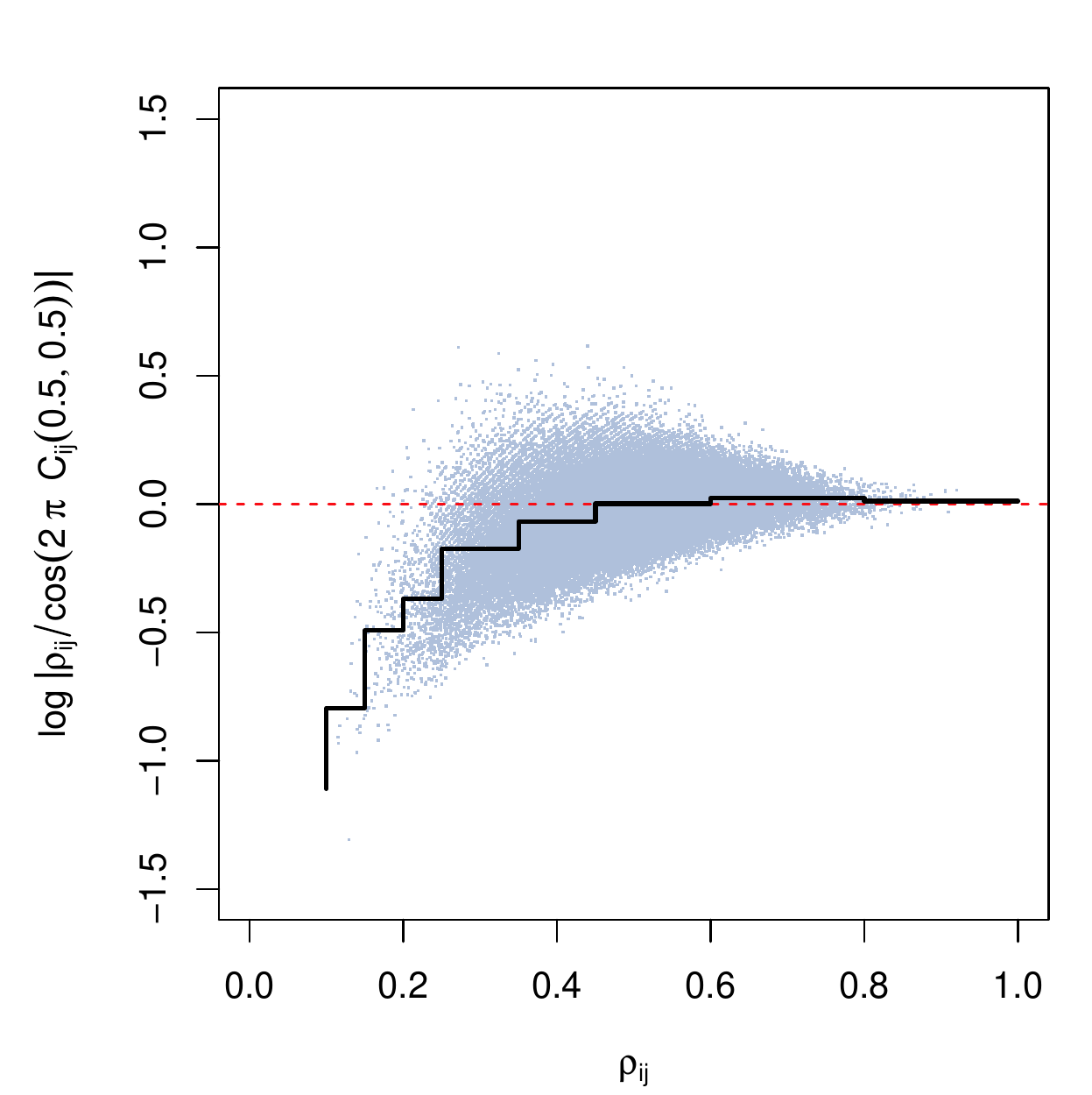}}
          \caption{Empirics: scatter plot of $\ln|\rho/\rhoB|$ vs $\rho$ for each stock pair, see Eq.~\eqref{eq:hatrho}. 
                   This figure quantifies the departure of the medial point of the copula $\cop(\tfrac12,\tfrac12)$
                   from the elliptical benchmark, for which the prediction is a straight horizontal line at 0 (dashed red).}
          \label{fig:emp_cop}
          \end{figure}
    \item Predict the structure of non-linear (typically quadratic) correlations with a reduced number of new parameters, 
          in order to clean the empirically measured (and much noisy) corresponding dependence coefficients.
\end{itemize}
We show below, that this can be achieved with a factor structure in the factor-volatilities themselves, 
with again a common mode and idiosyncratic volatility modes for the linear factors.

\subsubsection*{Data set}
We study daily close-to-close log-returns of stock prices of companies 
that are present in the S\&P500 index during the whole of the period studied.
We will be considering three periods for the empirical study and the model calibration: 
before the financial crisis (Jan, 2000 -- Dec, 2004);
during the financial crisis (Jan, 2005 -- Dec, 2009);
after  the financial crisis (Aug, 2009 -- Dec, 2012).
A longer dataset is used for the sliding windows procedure of In-sample/Out-of-sample testing 
in the last section: there we consider the ten years period 2000--2009.

It is enlightening to group the companies according to their sector of activity,
in order to see if patterns appear. 
When needed, we will make use of Bloomberg's classification as summarized in Tab.~\ref{tab:sectors}.
\begin{table}
    \center
    \begin{tabular}{lc||c|c|c||c|}
        Bloomberg sector            & Code          & 2000--04 & 2005--09& 2009--12 & 2000--09\\\hline\hline
       %Asset Backed Securities     & \# 1          &   0    &   0     &          &   0\\
       %Basic Materials             & \# 2          &   0    &   0     &          &   0\\
        Communications              & \# 3          &  33    &  25     &  29      &  18\\
        Consumer, Cyclical          & \# 4          &  60    &  49     &  33      &  40\\
        Consumer, Non-Cyclical      & \# 5          &  67    &  75     &  75      &  53\\
        Diversified                 & \# 6          &   0    &   0     &   1      &   0\\
        Energy                      & \# 7          &  19    &  21     &  34      &  15\\
        Financial                   & \# 8          &  57    &  55     &  75      &  37\\
       %Funds                       & \# 9          &   0    &   0     &   0      &   0\\
       %Government                  & \#10          &   0    &   0     &   0      &   0\\
        Industrial                  & \#11          &  51    &  50     &  50      &  42\\
       %Mortgage Securities         & \#12          &   0    &   0     &   0      &   0\\
        Technology                  & \#13          &  38    &  43     &  35      &  33\\
        Utilities                   & \#14          &  27    &  27     &  28      &  24\\\hline
        Total number of firms ($N$) &               & 352    & 345     & 360      & 262\\\hline
        Total number of days ($T$)  &               &1255    &1258     & 755      &2514
    \end{tabular}
    \caption{Economic sectors according to Bloomberg classification, 
             with corresponding number of individuals for each period.}
    \label{tab:sectors}
\end{table}

% Method: pairwise dependences (linear and quadratic correlations, copulas, etc.); 
% elliptical prediction known. Model predictions: analytic or Monte-Carlo.

\subsubsection*{Outline}
This chapter is made of four sections. 
In Section~2, we study the linear correlations of pairs of stocks and 
discuss the design and estimation of a factor model for their description.
The factors and residuals generated by the calibration of the model are
studied in Section~3, and motivate  the specification of the volatility content
of the model that we present in Section~4.
The resulting non-Gaussian model is calibrated, and Section~5 is dedicated to an Out-of-sample stability analysis,
validating its usefulness for the description of non-linear dependences.

\section{Linear factors}\label{sec:MFlin}

We study a simple one-level (not hierarchical) model for the joint description of the stock returns $\ret_i$ of $N$ firms, 
as a combination of $M$ shared factors $f_k$:
\begin{equation}\label{eq:MODEL}
    \ret_i=\sum_{k=1}^M \Wei_{ki}f_k+e_i.
\end{equation}
The weight $\Wei_{ki}$ parameterizes the linear exposure of stock $i$ to factor $k$.
At this stage, we do not yet specify the statistical properties of the factors $f_k$ and residuals $e_i$,
except that we impose their linear uncorrelatedness:
\begin{subequations}\label{eq:ortho_f_e}
\begin{align}
    \vev{f_kf_l}&=\delta_{kl}\\
    \vev{e_ie_j}&=\delta_{ij}\,\Big(1-\sum_l\Wei_{li}^2\Big)\\
    \vev{f_ke_j}&=0.
\end{align}
\end{subequations}
This way, the residuals can be understood as idiosyncrasies and all the linear dependence is 
supposed to be taken into account by the factors.
This, as we will shortly see, is already not trivial, 
for statistical factor models usually rely on the Principal Components Analysis (PCA)\nomenclature{PCA}{Principal components analysis}%
 decomposition which does not generate orthogonal residuals.

\subsection{Linear correlations}\label{sec:lincor_MF}
The predictions of the model in terms of covariances of the returns $\ret_i$ do not need additional assumptions,
and only depend on the matrix of linear weights $\Wei$: 
\begin{equation}\label{eq:lincor_predict}
    \rho_{ij}=\vev{\ret_i\ret_j}=\begin{cases}(\Wei^{\dagger}\Wei)_{ij}&, i\neq j\\1&,i=j\end{cases}
\end{equation}

When calibrating the model on real data, 
we want to find the $M$ most relevant {\emph{uncorrelated}} and {\emph{common}} unit-variance factors $\mat{F}$ $(T\times M)$,
    and the exposures $\Wei$ $(M\times N)$ of every stock to every factor, such that 
    \begin{equation}\label{eq:MODEL_mat}\tag{\ref{eq:MODEL}$'$}\Ret=\mat{F}\Wei+\mat{E},\end{equation}
    and at the same time, $\mat{E}$ be orthogonal $T\times(N\!-\!M)$.
Notice that the factors series $\mat{F}_{tk}$ are \emph{not} inputs of the estimation problem, 
but rather are obtained as a byproduct. 
This is at contrast with the usual econometric determination of elasticities $\Wei$ in linear regressions
of the explanatory (thus known) variables $\mat{F}$ onto the explained variables $\Ret$.

    Positing that the model \eqref{eq:MODEL_mat} is a faithful representation of real returns, 
we are able to uncover the factor structure and calibrate the corresponding weights $\Wei$.
Then, the estimate of the correlations \eqref{eq:lincor_predict} is a cleaned version of the brute force and noisy sample covariances $\frac{1}{T}\Ret^\dagger\Ret$.

\subsubsection*{A proxy: the Principal Components Analysis (PCA)} 
    Diagonalization of the sample correlation matrix yields
    \[
        \frac{1}{T}\Ret^\dagger \Ret = \mat{V}\Lambda \mat{V}^\dagger
    \]
    where $\Lambda$ is the diagonal matrix of eigenvalues, and the columns of $\mat{V}$ are the corresponding eigenvectors.
    Hence, there always exist (linearly orthogonal) factor series $\widetilde{\mat{F}}$ such that the return series $\Ret$ can be decomposed as
    \begin{equation}\label{eq:PCA_sol}
        \Ret=\widetilde{\mat{F}}\Lambda^{\frac12}\mat{V}^\dagger\qquad\text{where}\qquad\frac{1}{T}\widetilde{\mat{F}}^\dagger\widetilde{\mat{F}}=\mathds{1}_N.
    \end{equation}
    In order to reconciliate this decomposition in terms of statistical uncorrelated modes $\widetilde{\mat{F}}$
    with the factors $\mat{F}$ of the model, the PCA solution \eqref{eq:PCA_sol} needs to be identified with Eq.~\eqref{eq:MODEL_mat}.
    
    The factors should explain as much as possible of the returns covariances (thus of the portfolio variance), 
    leaving only idiosyncratic residual volatility to be explained by the $e_i$'s.
    Said differently, only those eigenvalues having a significant signal-to-noise ratio (compared f.ex.\ to an RMT benchmark)
    should be kept in the identification of the spectral decomposition with the factor model.
    This procedure is known as ``eigenvalue clipping'', and typically the most relevant eigenvalues are the largest ones
    and the smallest ones, on the right and the left of the noise bulk in the empirical spectrum.
    Since the smallest eigenvalues do not contribute much to the portfolio variance despite being statistically significant, 
    we will focus on the modes with large eigenvalues.
    Ordering the eigenvalues in decreasing order, and splitting the first $M$ 
    (subscript $_{M|}$) from the last $(N-M)$ (subscript $_{|N\!-\!M}$),
    it is straightforward to obtain the identification
    \begin{equation}\label{eq:X_PCA}
        {\Wei_{\text{PCA}}=\Lambda_{M|}^{\frac12}\mat{V}_{M|}^\dagger}.
    \end{equation}
    At this stage, the series of factors can be formally identified as the first $M$ spectral modes
    \begin{equation}\label{eq:id_f}
        \mat{F}=\widetilde{\mat{F}}_{M|}=(\Ret\mat{V}\Lambda^{-\frac12})_{M|}
    \end{equation}
    such that indeed $\frac{1}{T}\mat{F}^\dagger \mat{F}=\mathds{1}_M$.
    But {the residuals are \emph{not} orthogonal}:
    \begin{equation}\label{eq:id_e}
        \mat{E}=\widetilde{\mat{F}}_{|N\!-\!M}\Lambda_{|N\!-\!M}^{\frac12}\mat{V}_{|N\!-\!M}^\dagger\quad\text{s.t.}\quad \frac{1}{T}\mat{E}^\dagger \mat{E}=\mat{V}_{|N\!-\!M}\Lambda_{|N\!-\!M}\mat{V}_{|N\!-\!M}^\dagger
    \end{equation}
    and thus cannot be understood as idiosyncrasies of the returns series.

\subsubsection*{Factor weights}
The PCA can alternatively be thought of as the solution of 
\[
    \frac{1}{T}\Ret^\dagger \Ret=\Wei^\dagger \Wei\quad\text{with}\quad \Wei\Wei^\dagger\text{ diagonal}
\]
% \[
    % \Wei_{\text{PCA}}=\argmin_{\Wei\Wei^\dagger=\mathds{1}} \operatorname{dist}\!\left(\frac{1}{T}\Ret^\dagger \Ret,\Wei^\dagger \Wei\right)
% \]
where importantly $\Wei$ is full rank, and the $M$ modes with largest amplitude $(\Wei\Wei^\dagger)_{kk}$ are kept 
\emph{after} the equation is solved.
When one rather wants to optimize the contribution of \emph{only} $M$ modes, 
one can try to minimize a distance between the LHS and the RHS with a matrix of weights of rank $M$:
\[
    \argmin \left\|\frac{1}{T}\Ret^\dagger \Ret-\Wei^\dagger \Wei\right\|
\]
In this case an exact solution does not exist as there are $NM$ variables for $N(N+1)/2$ equations, 
yet one does not throw away information of subsequent modes, so that in the end the information
content of that solution is typically bigger than that of the first $M$ columns of $\Wei_{\text{PCA}}$.

On the top of this, the estimation of the $\Wei$'s can be further improved if one \emph{knows}
the standard deviation of the series of returns $\Ret$, and there only remains to estimate
the correlations (rather than the covariances).
After renormalization, the matrix $\frac{1}{T}\Ret^\dagger \Ret$ has only ones on the diagonal,
and since this is also what the model predicts according to Eq.~\eqref{eq:lincor_predict}, 
the optimal estimation can be performed on the off-diagonal content only:
\begin{equation}\label{eq:offdiag_content}
    \argmin \left\|\frac{1}{T}\Ret^\dagger \Ret-\Wei^\dagger \Wei\right\|_{\text{off-diag}}.
\end{equation}
By lowering the number of equations to $N(N-1)/2$ while keeping the same number of variables $NM$,
one improves the identifiability of the model parameters.

The underlying view of this estimation method as of the collective dynamics is that
 variances and correlations are better 
estimated separately than jointly through the covariance matrix.
As soon as one claims having a ``good'' estimator of the variances, 
one is better off estimating the covariances of the normalized time series,
rather than having estimated the covariances in the first place. 
%
%\textcolor{red}{
Notice that  orthogonality of the lines of $\Wei$ obtained with this method is not granted,
as opposed to the PCA, but what matters is rather uncorrelatedness of the factors.
%En particulier le mode dominant conserve une projection significative sur les autres (3.5\%--35\%).
%}

\subsubsection*{Linear model calibration on financial data}
\begin{figure}
    \center
    \subfigure[2000--2004]{\includegraphics[scale=.75,trim=25 25 10 25,clip]{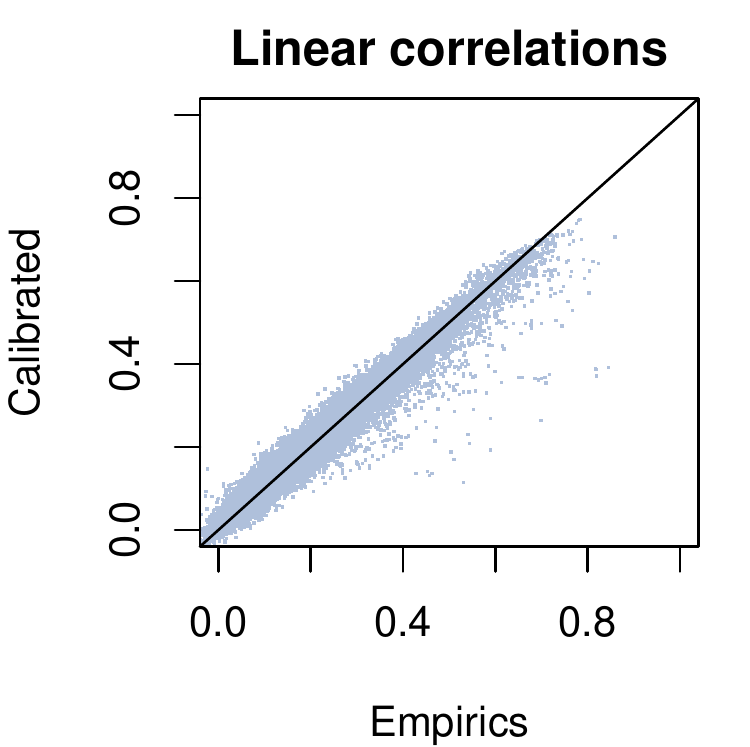}}
    \subfigure[2005--2009]{\includegraphics[scale=.75,trim=25 25 10 25,clip]{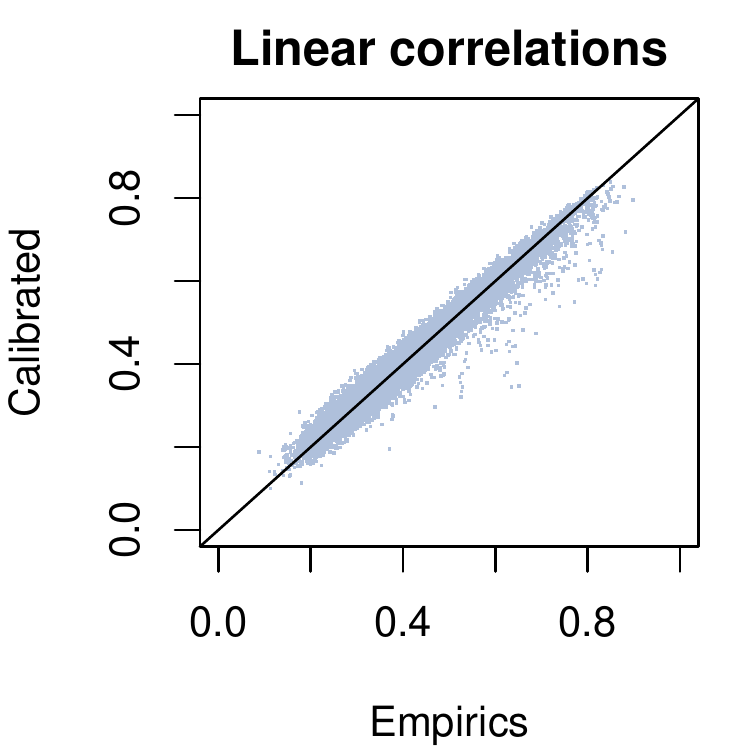}}
    \subfigure[2009--2012]{\includegraphics[scale=.75,trim=25 25 10 25,clip]{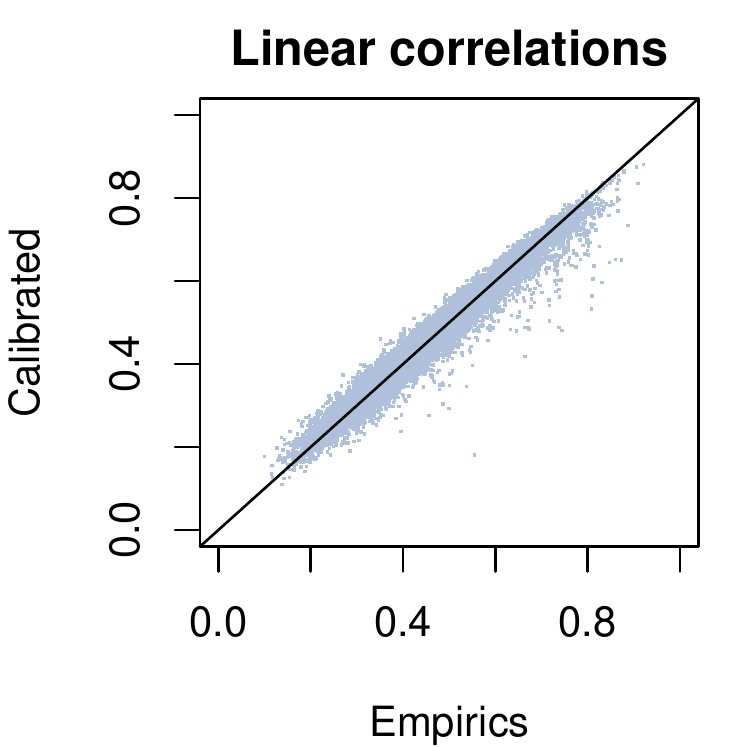}}
   %\subfigure[2000--2004]{\includegraphics[scale=.14,trim=0 0 0 100,clip]{/home/rchicheportiche/matrices/MultiFact/pict_2000-2004/cal_C1}}
   %\subfigure[2005--2009]{\includegraphics[scale=.14,trim=0 0 0 100,clip]{/home/rchicheportiche/matrices/MultiFact/pict_2005-2009/cal_C1}}
   %\subfigure[2009--2012]{\includegraphics[scale=.14,trim=0 0 0 100,clip]{/home/rchicheportiche/matrices/MultiFact/pict_2009-2012/cal_C1}}
    \caption{Linear correlations: calibrated ($M=10$) vs sample correlations.
             Outliers that stand far from the identity line (in particular for the period (a)) are strongly correlated pairs,
             typically in the same sector. Such pair correlations are poorly taken into account by factor models.}\label{fig:cal_lin}
\end{figure}
The linear weights $\Wei$ are calibrated on the three data sets by solving Eq.~\eqref{eq:offdiag_content} 
with $M=10$, corresponding to the number of Bloomberg sectors plus one. 
The resulting correlations $(\Wei^\dagger\Wei)_{\ij}$ are shown as a scatter plot versus the 
sample correlation $\frac{1}{T}(\Ret^\dagger\Ret)_{\ij}$ on Fig.~\ref{fig:cal_lin}.
Although the calibration seems to provide a ``good fit'' (whatever this might mean),
a precise characterization of the benefit of using the method discussed above over the PCA
is only achieved by a criterion of model selection, or by out-of-sample risk measurements.

We discuss this in detail in Sect.~\ref{sec:stability}, where we use Markowitz' portfolio selection framework
to compute the \emph{realized} risk of optimal baskets of stocks.
We show there that when using the proposed calibration method, the out-of-sample risk 
can be more than 25\% lower than PCA% 
\footnote{The PCA method is also known as eigenvalue clipping in the context of cleaning schemes for matrix inversion,
 and is one of the best generic cleaning scheme known so far, see Ref.~\cite{potters2009financial}}, 
with an in-sample risk almost unchanged, see Fig.~\ref{fig:ISOS_lin} below.

\subsubsection*{Orthogonal residuals}
As already said, the byproduct of the PCA identification \eqref{eq:id_f} is the matrix $\mat{E}$ of overlapping residuals given in Eq.~\eqref{eq:id_e}.
When the weights $\Wei$ are known, it is however possible to design a different identification scheme
that generates orthogonal residuals. 
Consider indeed the date-by-date regression
\begin{equation}\label{eq:dbdR}
    \Ret_{t\cdot}=\mat{F}_{t\cdot}\Wei+\mat{E}_{t\cdot}
\end{equation}
where the explained variables are the input returns, the explanatory variables are the (freshly estimated) weights $\Wei$
and the regression parameters to be estimated are the factors.
Then a GLS\nomenclature{GLS}{Generalized least squares} solution of the regression yields the wanted factors and residual series.
It is only {approximate} in the sense that $\frac{1}{T}\mat{E}^\dagger \mat{E}$ is only ``as close as can be'' to a diagonal matrix,
and $\frac{1}{T}\mat{F}^\dagger \mat{F}$ is only approximately $\mathds{1}_M$.

In the next section, we discuss in details the properties of the factors series $\mat{F}_{t\cdot}$ 
and residual series $\mat{E}_{t\cdot}$, obtained after calibration of the model on the datasets.

\section{Properties of the reconstructed factors and residuals}\label{sec:MFspectral}
After the calibration of the weights $\Wei$, we are able to perform 
a statistical analysis on the reconstructed series of factors and residuals.
Because we have designed a calibration procedure improving over the PCA,
the obtained factors are not pure eigenmodes of the correlation matrix.
Still, we can interpret them looking at their statistical signatures.
    In particular, the first factor, $f_1$ is clearly a market mode (the ``index'', a collective mode of fluctuations).
    Indeed, introducing the average daily return $\overline{\MakeUppercase{\ret}}=\frac{1}{N}\Ret\vect{1}$,
    we find that it has an almost perfect overlap with the first factor: $\text{Cor}(\mat{F}_{t1},\overline{\MakeUppercase{\ret}}_t)=99.9\%$ !
    
    %%% CODE R : cor(estimQ2mf[["F"]][,1],rowMeans(eta,na.rm=TRUE))
    %%%          plot(-eigen(cor(eta,use="p"),symmetric=TRUE)$vectors[,1],type="h",ylim=c(0,.1))

%%% ARTICLE    
%\textcolor{red}{Pour l'article: comparer $\mat{F}$ aux modes de PCA $\widetilde{\mat{F}}$ de Eq.~\eqref{eq:id_f} !
%Par example en utilisant l'overlap distance de Romain, i.e.\ le det de $\mat{F}^\dagger\widetilde{\mat{F}}$ ou on n'a garde que les $M$ premiers modes}

% \begin{table}
    % \center
    % \begin{tabular}{@{ $}c@{$ }||c|c|c||}
    % & Mean & Std-dev & Exc. kurt.\\\hline
    % f_1 & 0 & 1 &\\
    % f_2 & 0 & 1 &\\
    % f_3 & 0 & 1 &
    % \end{tabular}
    % \begin{tabular}{r||*{10}{c|}}
           % $k$& $1$ & $2$ & $3$ & $4$ & $5$ & $6$ & $7$ & $8$ & $9$ & ${10}$\\\hline\hline
    % Exc. Kurt.&  0  &  0  &  0  &  0  &  0  &  0  &  0  &  0  &  0  &  0  \\
    % IPR       &  0  &  0  &  0  &  0  &  0  &  0  &  0  &  0  &  0  &  0  
    % \end{tabular}
    % \begin{tabular}{r||*{3}{c|}|*{3}{c|}|*{3}{c|}|}
    %    IPR$_k$& $1$ & $2$ & $3$ & $1$ & $2$ & $3$ & $1$ & $2$ & $3$ \\\hline\hline
    % period1   &  0  &  0  &  0  &  0  &  0  &  0  &  0  &  0  &  0  \\
    % period2   &  0  &  0  &  0  &  0  &  0  &  0  &  0  &  0  &  0  \\
    % period3   &  0  &  0  &  0  &  0  &  0  &  0  &  0  &  0  &  0  \\
    % period4   &  0  &  0  &  0  &  0  &  0  &  0  &  0  &  0  &  0  \\
    % period5   &  0  &  0  &  0  &  0  &  0  &  0  &  0  &  0  &  0  
    % \end{tabular}
    % \caption{Inverse Participation Ratios (IPR) IPR$_\alpha=N_a\sum_{a=1}^{N_a}v_{\alpha}(a)^4$ of the first three eigenvectors of ...}
% \end{table}

The non-linear properties of the reconstructed factors and residuals can be investigated
through the higher-order correlations
\begin{subequations}\label{eq:allcorabs}
\begin{align}
    \label{eq:facfac}\frac{1}{p^2}\ln\frac{\vev{|\mat{F}_{tk}\mat{F}_{tl}|^p}}{\vev{|\mat{F}_{tk}|^p}\vev{|\mat{F}_{tl}|^p}}\\
    \label{eq:resres}\frac{1}{p^2}\ln\frac{\vev{|\mat{E}_{ti}\mat{E}_{tj}|^p}}{\vev{|\mat{E}_{ti}|^p}\vev{|\mat{E}_{tj}|^p}}\\
    \label{eq:facres}\frac{1}{p^2}\ln\frac{\vev{|\mat{F}_{tk}\mat{E}_{tj}|^p}}{\vev{|\mat{F}_{tk}|^p}\vev{|\mat{E}_{tj}|^p}}
\end{align}
\end{subequations}
for any value of $p>0$. As an example, we show in Fig.~\ref{fig:matelems} the off-diagonal matrix elements of the 
factor-factor correlations \eqref{eq:facfac} for a value of $M=10$ 
(again, the factors series are obtained from the procedure described in the previous section). 
On each figure corresponding to a value of $k$, the curves represent the values of Eq.~\eqref{eq:facfac}
for all $l\neq k$, as a function of $p$.
We observe that the curves depart from the Gaussian benchmark for which the correlation 
computed and shown would all be 0.
Anticipating over the forthcoming theoretical description,
we also see that the curves are not compatible with factors having a lognormal volatility,
in which case all the curves would be exactly horizontal.
The concavity of the curves is a signature of non-Gaussianity in log-volatilities, 
while their splitting (in particular as $p\to 0$) reveals a complex structure 
that we will uncover using a model in Section~\ref{sec:modeling_vol} below.

\begin{figure}
    \center
    \includegraphics[scale=.55,trim=  0 165  710 0,clip]{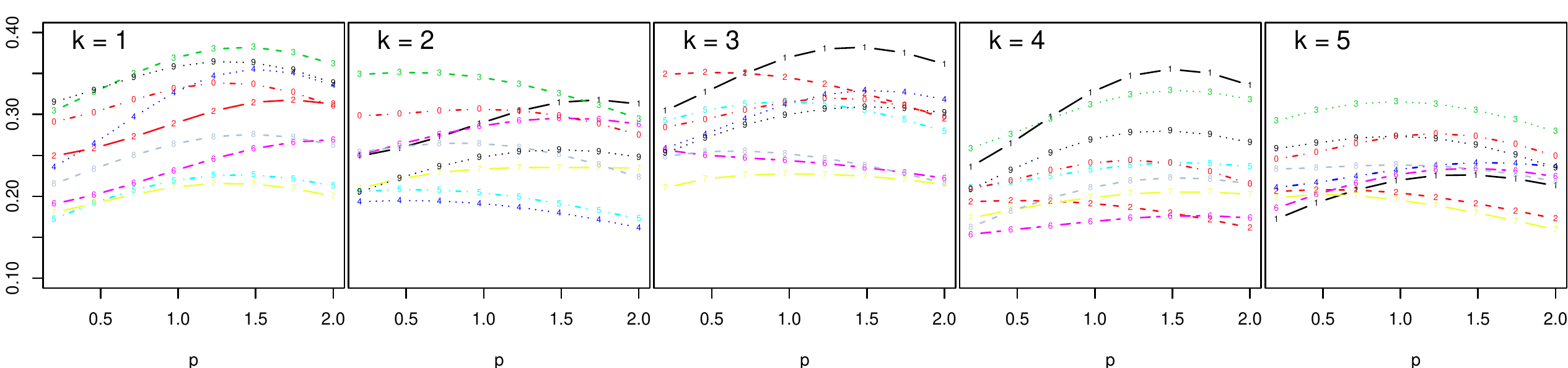}
    \includegraphics[scale=.55,trim=  0 165 1416 0,clip]{MultiFact/pict_2000-2009/emp_logcorabs_fact}%
    \includegraphics[scale=.55,trim=733 165    0 0,clip]{MultiFact/pict_2000-2009/emp_logcorabs_fact}
    \caption{Visual representation of the estimated factors-factors dependences, for $M=10$, on the period 2000--2009. 
    The correlation~\eqref{eq:facfac} is shown for every factor $k$ with all other factors $l\neq k$,
    as a function of the order $p$ of the absolute moment considered.}
    \label{fig:matelems}
\end{figure}

This ``naked eye'' analysis cannot be reproduced for large matrices like the 
factors-residuals and residuals-residuals correlations,
for which it turns much more convenient to use a spectral approach.
Spectral decomposition of the corresponding matrices of coefficients reveals 
eigenmodes of amplitude fluctuations.
Fig.~\ref{fig:spectra} illustrates the content of the spectra for
the factor-factor coefficients (top), the residual-residual coefficients (middle),
and the singular values of the rectangular matrix of the factors-residuals cross-correlations (bottom),
for the period 2000--2009.
The first three values of each spectrum are shown in dark, while the others are represented as a grey bulk.
Notice how the bulk widens for the res-res spectrum as the moment $p$ increases:
this is due to the high dimensionality of this matrix, and the large resulting noise.

\begin{figure}
    \center
    \subfigure[Singular values spectra (vertical axis), for several values of $p$ (horizontal axis).
               The three largest values are in black, while the remaining values are shown as a grey bulk.]{\label{fig:spectra}\includegraphics[scale=.75,angle=-90             ]{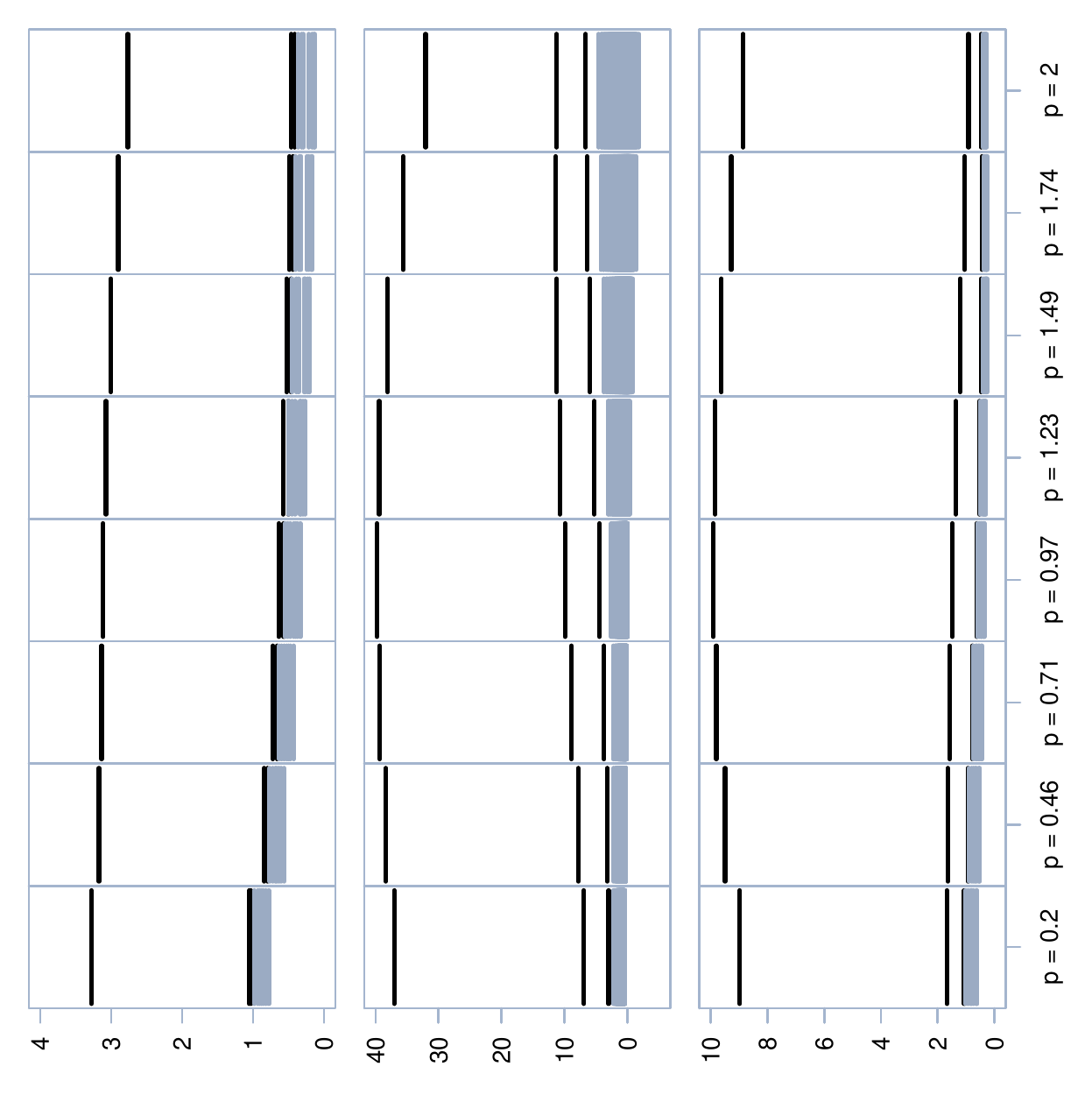}}
    \vfill
    \subfigure[Corresponding eigenvectors: first three, from top to bottom, averaged over $p$.            ]{\label{fig:evects} \includegraphics[scale=.68,trim=  0 0  600 0,clip]{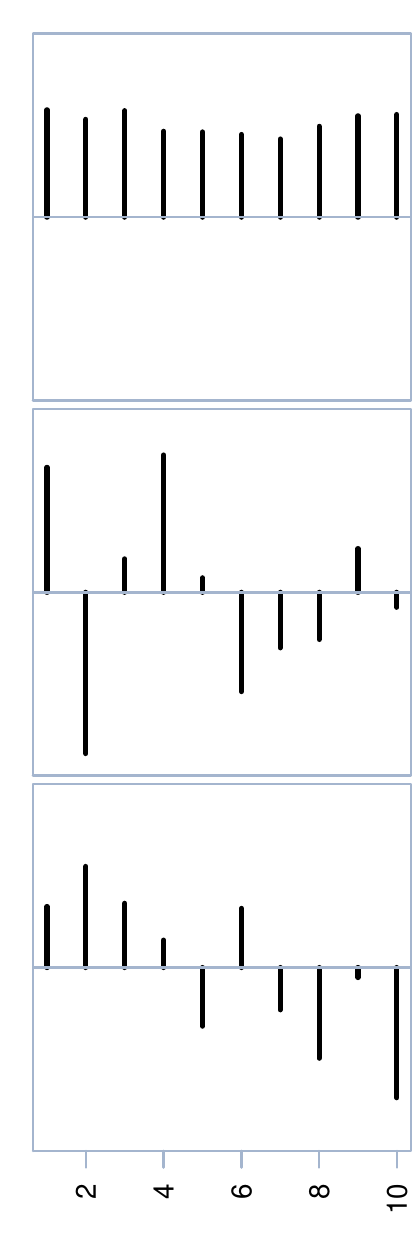}\hspace{.8cm}
                                                                                                                               \includegraphics[scale=.68,trim=120 0  480 0,clip]{MultiFact/pict_2000-2009/svd_vects}\hspace{.8cm}
                                                                                                                               \includegraphics[scale=.68,trim=360 0  120 0,clip]{MultiFact/pict_2000-2009/svd_vects}}
    \caption{$M=10$, 2000-2009. Spectral decomposition of the 
        factor-factor,     % correlations (Eq.~\eqref{eq:facfac}, \textbf{top}),
    residual-residual, and % correlations (Eq.~\eqref{eq:resres}, \textbf{middle}), and
      factor-residual correlations.} % (Eq.~\eqref{eq:facres}, \textbf{bottom}).}
\end{figure}

\begin{figure}
    \center
    \subfigure[]{\includegraphics[scale=.5,trim=  0 0  600 0,clip]{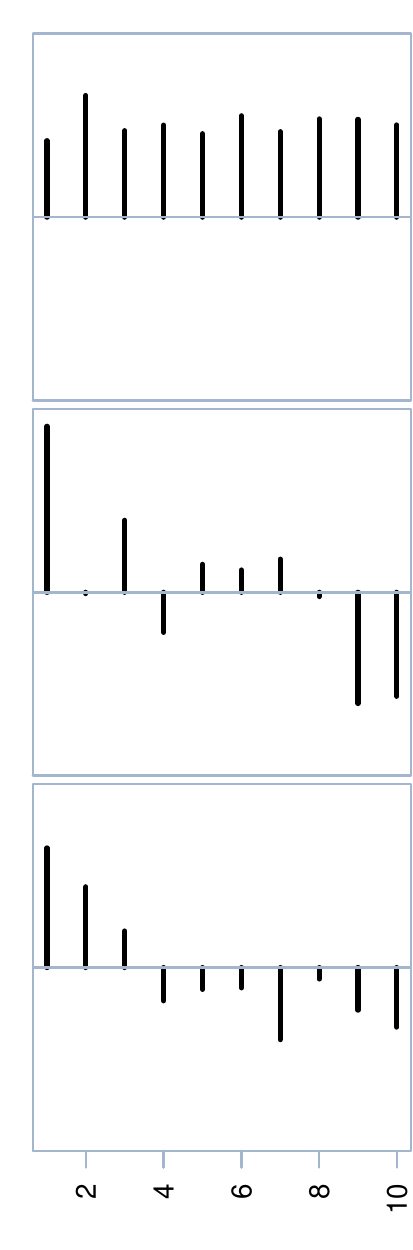}\label{fig:0004EVff}}\hspace{.5cm}
    \subfigure[]{\includegraphics[scale=.5,trim=240 0  360 0,clip]{MultiFact/pict_2000-2004/svd_vects}\label{fig:0004EVee}}\hspace{.5cm}
    \subfigure[]{\includegraphics[scale=.5,trim=360 0  240 0,clip]{MultiFact/pict_2000-2004/svd_vects}
                 \includegraphics[scale=.5,trim=600 0    0 0,clip]{MultiFact/pict_2000-2004/svd_vects}\label{fig:0004EVef}}
    \caption{$M=10$, 2000-2004.}\label{fig:2000-2004}
    \vfill
    \subfigure[]{\includegraphics[scale=.5,trim=  0 0 600 0,clip]{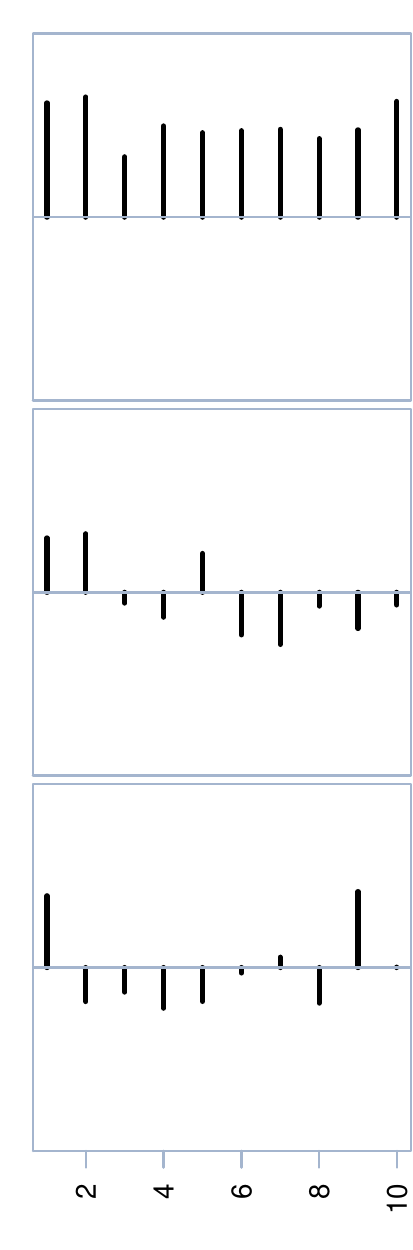}\label{fig:0509EVff}}\hspace{.5cm}
    \subfigure[]{\includegraphics[scale=.5,trim=240 0 360 0,clip]{MultiFact/pict_2005-2009/svd_vects}\label{fig:0509EVee}}\hspace{.5cm}
    \subfigure[]{\includegraphics[scale=.5,trim=360 0 240 0,clip]{MultiFact/pict_2005-2009/svd_vects}
                 \includegraphics[scale=.5,trim=600 0   0 0,clip]{MultiFact/pict_2005-2009/svd_vects}\label{fig:0509EVef}}
    \caption{$M=10$, 2005-2009}\label{fig:2005-2009}
    \vfill
    \subfigure[]{\includegraphics[scale=.5,trim=  0 0 600 0,clip]{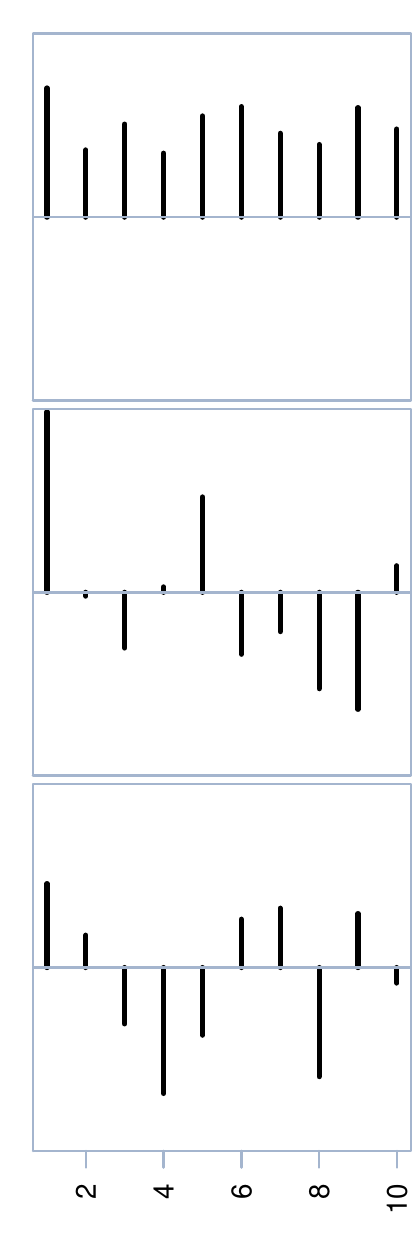}\label{fig:0912EVff}}\hspace{.5cm}
    \subfigure[]{\includegraphics[scale=.5,trim=240 0 360 0,clip]{MultiFact/pict_2009-2012/svd_vects}\label{fig:0912EVee}}\hspace{.5cm}
    \subfigure[]{\includegraphics[scale=.5,trim=360 0 240 0,clip]{MultiFact/pict_2009-2012/svd_vects}
                 \includegraphics[scale=.5,trim=600 0   0 0,clip]{MultiFact/pict_2009-2012/svd_vects}\label{fig:0912EVef}}
    \caption{$M=10$, 2009-2012}\label{fig:2009-2012}
\end{figure}

This model-free description highlights the presence of few dominant modes of \emph{amplitude} 
(not to be confused with the factors of \emph{linear} correlations), in particular a common mode.
This is confirmed by the study of the corresponding eigenvectors on Fig.~\ref{fig:evects}:
the first eigenvector of both the factor-factor amplitudes and the residual-residual amplitudes 
has all components with same sign, revealing a collective mode driving the amplitudes of the factors and the residuals.
The second eigenvector of res-res reveals ``finance against all'' behavior in the given period.
The third eigenvector of res-res reveals ``energy against all'' behavior (energy + utilities).
This mode is not observed on the fact-res correlations, thus indicating that the impact of 
energy on the volatility is only present in the residuals, and not shared by the linear factors.

Of course this structure is not stable over time, where it is expected and observed that 
the number and meaning of significant modes varies according to the studied period.
The spectral decomposition for the three sub-periods 2000--2004, 2005--2009, 2009--2012
is shown in Figs.~\ref{fig:2000-2004}, \ref{fig:2005-2009}, \ref{fig:2009-2012}, respectively.
We represent the eigenvectors of the factor-factor (a), residual-residual (b), and the left/right singular vectors of the factor-residual (c).
When applicable, we have averaged the components of the vectors over Bloomberg sector in order to see if a structure emerges (see Tab.~\ref{tab:sectors} for the labels).
As a general observation, the collective mode of amplitude fluctuations is stable over time, 
and strongly dominant over any other mode.
As can be seen in Fig.~\ref{fig:2005-2009}, it seems to be impacting/impacted by the financial sector more strongly 
in the 2005--2009 period (containing the financial crisis)
whereas in other periods it is almost uniformly spread over sectors.
This financial sector is also a stand-alone mode of fluctuations in the crisis period, 
as revealed by the structure of the second eigenvector in Figs.~\ref{fig:0509EVee} and \ref{fig:0509EVef}, 
where the stocks in that sector oscillate with opposing phase with respect to the other sectors.
In the pre-crisis period, the second relevant mode is rather composed of commodities. 
Indeed, the second eigenvector features the opposition of utilities (\#~14), energy (\#~7) and communications (\#~3) against the rest.
This already illustrates that the structure is not stable over time, 
what is confirmed by the post-crisis period where there is not even a clear signature
 for the existence of a significant second mode 
(the apparent structure in the second mode of Fig.~\ref{fig:0912EVef} is not significant 
because the sector \#~6 is populated with a single individual).

%\textcolor{red}{Transition}
    
\section{Modeling the volatility content}\label{sec:modeling_vol}
The univariate and multivariate pseudo-empirical properties of the reconstructed factors series 
motivate the construction of a model that integrates the following characteristics:
\begin{itemize}
    \item all factors amplitudes are driven by a common mode;
    \item all residuals amplitudes are driven by \emph{that same} common mode;
    \item factors and residuals series are non-Gaussian, and their amplitudes are not lognormal either.
\end{itemize}

The spectral analysis suggests the presence of two volatility drivers $\sigma_0$ and $\sigma_{0'}$ in the factors
and the residuals.%
\footnote{
It is very easy to add a possible third driver $\sigma_{0''}$ in the equation \eqref{eq:model_ej_0} for the residuals,
but its statistical significance is not granted and the estimation of the associated parameters could turn quite noisy.}
Concretely, we aggregate the volatility determinants multiplicatively as follows:
\begin{subequations}\label{eq:model_fkej_0}
\begin{align}
    % f_k&=\epsilon_k \,\sigma_0^{A_{k0}}\,\sigma_{0'}^{A_{k0'}}\,\sigma_k^{{A}_{kk}}\\\label{eq:model_ej_0}
    %%e_j&=\eta_j     \,\sigma_0^{B_{j0}}\,\sigma_{0'}^{B_{j0'}}\,\sigma_{0''}^{B_{j0''}}\,  \widetilde\sigma_j^{{B}_{jj}}
    % e_j&=\eta_j     \,\sigma_0^{B_{j0}}\,\sigma_{0'}^{B_{j0'}}\,\widetilde\sigma_j^{{B}_{jj}}.
    f_k&=\epsilon_k \,\exp(A_{k0} \omega_0)\,\exp(A_{k0'} \omega_{0'})\,\exp(A_{kk} \omega_k)\\\label{eq:model_ej_0}
    e_j&=\eta_j     \,\exp(B_{j0} \omega_0)\,\exp(B_{j0'} \omega_{0'})\,\exp(B_{jj} \widetilde\omega_j),
\end{align}
\end{subequations}
where the $\omega$'s are stochastic log-volatilities (all independent of each other and independent of the $\epsilon$'s and $\eta$'s),
and the parameters $A$'s and $B$'s weight the contribution of every volatility mode.
In particular, we expect $A_{k0}$ and $B_{k0}$ to reflect approximately the coordinates of the first eigenvectors 
of the ``log-abs'' correlations on Fig.~\ref{fig:2000-2004}--\ref{fig:2009-2012}(a) and Fig.~\ref{fig:2000-2004}--\ref{fig:2009-2012}(b) respectively,
and $A_{k0'}$ and $B_{k0'}$ the coordinates of the second eigenvectors.
%The tilde parameters encode the remaining volatility not explained by the common drivers.% (if any).
$A_{kk}$ and $B_{jj}$ are the standard-deviations of the remaining log-volatilities not explained by the common drivers (if any).

In the next subsection, we estimate a restricted model where only a single volatility driver is active.
\subsection{A dominant volatility mode}
The simplest improvement over the independent factors assumption, while keeping uncorrelatedness,
is to allow for a common fluctuation of amplitudes.
The multiplicative structure of Eqs.~\eqref{eq:model_fkej_0} is better rendered using the summation in exponential notation:
\begin{subequations}\label{eq:model0}
\begin{align}
    \label{eq:f_k_0}f_k&=\epsilon_k \exp(A_{k0} \omega_0 + A_{kk} \omega_k)\\
    \label{eq:e_j_0}e_j&=\eta_j     \exp(B_{j0} \omega_0 + B_{jj} \widetilde\omega_j)
\end{align}
\end{subequations}
with $\epsilon_k,\eta_j$ Gaussian with variance such that Eqs.~\eqref{eq:ortho_f_e} hold,
and we have dropped the tildes and hats when no confusion was possible.
It is important to keep the intuitive meaning of the log-volatility factors in mind:
\begin{align*}
    \omega_0 &= \text{dominant and common driver of log-volatility across factors and residuals} \\
    \omega_1 &= \text{log-volatility of the index (market linear mode $f_1$), net of $\omega_0$ contribution}
\end{align*}
and the subsequent $\omega_k, \widetilde\omega_k, $ ($k>1$) characterize the ``residual volatility''
not explained by the common driver $\omega_0$ in the amplitude of the factors $f_k$ and residuals $e_j$.

The model is completely characterized from a probabilistic point of view 
when the law of the log-volatilities is specified.
Rather than using their Probability Distribution Functions, we resort to the 
Moment Generating Function (MGF)\nomenclature{MGF, mgf}{Moment generating function} $M_{\omega}(p)\equiv\esp{\exp(p\omega)}$,
as we find it more appropriate to characterize the $\omega$'s by their moments, in the perspective of calibration.
The $p$-dependence of the curves in Fig.~\ref{fig:matelems} suggests that the non-Gaussianity in the log-volatilities 
is approximately homogeneous across the factors, and thus possibly due to the common volatility driver $\omega_0$ only,
while the residual volatilities $\omega_k$ and $\widetilde\omega_j$ can be taken as Gaussian.
For those, the MGF is $M_{\text{G}}(p)=\exp(p^2/2)$.
But in the general case, developing in cumulants, $M_{\omega_0}$ is the exponential of a polynomial.
Typically, with
    \[
    \esp{\omega_0  }=0\qquad
    \esp{\omega_0^2}=1\qquad
    \esp{\omega_0^3}=\zeta_0\qquad
    \esp{\omega_0^4}=3+\kappa_0\qquad
    \]
    one gets
    \begin{align}\label{eq:MGF}
        M_{\omega_0}(p) &=\Exp{\frac12 p^2+\frac{\zeta_0}{6}p^3+\frac{\kappa_0}{24}p^4}.
    \end{align}

\subsubsection*{Parameters estimation}
    The model defined by Eqs.~(\ref{eq:MODEL},\ref{eq:model0}) contains the following parameters:
    $MN$ linear weights $\Wei_{ki}$ (previously estimated, see section~\ref{sec:MFlin}), 
    $M$ coefficients $A_{k0}$ and $N$ coefficients $B_{j0}$ of exposure to the common volatility mode,
    the standard-deviations $A_{kk}$ and $B_{jj}$ of the residual log-volatilities, 
    and the skewness $\zeta_0$ and kurtosis $\kappa_0$ of the volatility mode $\omega_0$.
    So there are overall $NM+2(N+M)+2$ parameters, for a dataset of size $NT$. 
    More importantly the number of parameters is only marginally enhanced with respect to a typical linear factor model 
    (where only the $NM$ linear weights enter into account): only $2(N+M+1)$ new parameters intended to 
    improve the description of all pairwise dependences coefficients.

    In order to estimate the latter we express the model predictions for the observables defined in Eq.~\eqref{eq:allcorabs}.
    It is convenient to introduce the following ratio of MGF's:
    \begin{align}\label{eq:defPhi}
        \Phi_0(a,b) &=\frac{M_{\omega_0}\!(a\!+\! b)}{M_{\omega_0}\!(a)M_{\omega_0}\!(b)},
    \end{align}
    as well as the Gaussian equivalent $\Phi_{\text{G}}(a,b)=\exp(ab)$.
    In logarithmic form, $\phi_0(a,b;p)=\frac{1}{p^2}\ln\!\Phi_0(pa,pb)$ is a polynomial in $p$ when expanding in cumulants.
    Indeed, with Eq.~\eqref{eq:MGF}, 
            \[\phi_0(a,b;p)=ab + \frac{p}{2}\zeta_0(a^2b+ab^2)+\frac{p^2}{12}\kappa_0(2a^3b+3a^2b^2+2ab^3)\]
            and $\phi_{\text{G}}(a,b;p)=ab$ is independent of $p$.
    Then the theoretical prediction for the matrix elements can be computed analytically:
    \begin{subequations}
    \label{eq:fkfleiej}
    \begin{align}
    \label{eq:fkfl0}
        \frac{1}{p^2}\ln\frac{\Esp{|f_k|^p|f_l|^p}}{\Esp{|f_k|^p}\Esp{|f_l|^p}}&=
        \phi_0(A_{k0},A_{l0};p)+\Big(\gamma(p) +A_{kk}^2 \Big)\delta_{kl}\\
    \label{eq:fkel0}
        \frac{1}{p^2}\ln\frac{\Esp{|f_k|^p|e_i|^p}}{\Esp{|f_k|^p}\Esp{|e_i|^p}}&=
        \phi_0(A_{k0},B_{i0};p)\\
    \label{eq:ekel0}
        \frac{1}{p^2}\ln\frac{\Esp{|e_i|^p|e_j|^p}}{\Esp{|e_i|^p}\Esp{|e_j|^p}}&=
        \phi_0(B_{i0},B_{j0};p)+\Big(\gamma(p) +B_{ii}^2 \Big)\delta_{ij}
    \end{align}
    \end{subequations}
    where %the dots stand for the expansion \eqref{eq:defPhi}, and 
    \[
         \gamma(p)=\frac{1}{p^2}\ln\frac{\Esp{|\epsilon|^{2p}}}{\Esp{|\epsilon|^p}^2}
               =\frac{1}{p^2}\ln\!\left(\sqrt{\pi}\frac{\Gamma(\tfrac{1}{2}+p)}{\Gamma(\tfrac{1+p}{2})^2}\right)
    \]
    is the normalized $2p$-moment of the absolute value of Gaussian variables.
    
    Clearly, if even $\omega_0$ was Gaussian, and the diagonal elements were regular, 
    the matrices described in Eqs.~\eqref{eq:fkfl0} and \eqref{eq:ekel0} would be trivially of rank 1,
    and the identification of $A_{\cdot0}$ and $B_{\cdot0}$ with the first eigenvectors of the corresponding matrices
    would be straightforward.
    Non-Gaussianity and specificities on the diagonal perturb this identification,
    but the overall picture is essentially the same story, as we will shortly show with the calibration results.
    
    The model estimation procedure is as follows (the linear weights $\Wei$ are previously estimated).
    As discussed just above, there are $2(N+M+1)$ parameters to be estimated.
    Because the equations \eqref{eq:fkfleiej} are coupled through \eqref{eq:fkel0}, all parameters should in principle be estimated jointly.
    The corresponding optimization program would however be computer intensive, 
    and the stability of the solution would not be granted in such a large dimensional space.
    We proceed stepwise instead, by first estimating the parameters $A_{k0}, A_{kk}, \zeta_0, \kappa_0$
    using the fac-fac predictions~\eqref{eq:fkfl0}, and then estimate the remaining parameters $B_{i0}, B_{ii}$ from
    the res-res \emph{and} fac-res correlations for consistence.
    Note that there is an overall sign degeneracy, 
    as well as a sign indetermination for the idiosyncratic parameters $A_{kk}, B_{ii}$ (which we then arbitrarily take positive).
    \begin{enumerate}
        \item Estimate $A_{k0}$, $A_{kk}$ and the non-Gaussianity parameters from Eq.~\eqref{eq:fkfl0}:
        \begin{equation}\label{eq:costfct_A}
            \min\left\{\sum_p\sum_{k,l}\left(
            \frac{1}{p^2}\ln\frac{\vev{|\mat{F}_{tk}\mat{F}_{tl}|^p}}{\vev{|\mat{F}_{tk}|^p}\vev{|\mat{F}_{tl}|^p}}
            -
            \frac{1}{p^2}\ln\frac{\Esp{|f_k|^p|f_l|^p}}{\Esp{|f_k|^p}\Esp{|f_l|^p}}
            \right)^2\right\}
        \end{equation}
        The sum on $p$ runs over eight values between $p=0.2$ and $p=2$ (see x-axis of Fig.~\ref{fig:spectra}) 
        and is crucial here to the estimation of the non-Gaussianity parameters, 
        since the loss function is independent of $p$ for Gaussian variables.
        This amounts to performing a best (joint!) quadratic fit of the curves similar to Fig.~\ref{fig:matelems}, 
        for each period.
        \item Estimate $B_{i0}$ from Eq.~\eqref{eq:fkel0}:
        \begin{subequations}\label{eq:loss_B}
        \begin{equation}
            \min\left\{\sum_{k,i}\left(
            \frac{1}{p^2}\ln\frac{\vev{|\mat{F}_{tk}\mat{E}_{ti}|^p}}{\vev{|\mat{F}_{tk}|^p}\vev{|\mat{E}_{ti}|^p}}
            -
            \frac{1}{p^2}\ln\frac{\Esp{|f_k|^p|e_i|^p}}{\Esp{|f_k|^p}\Esp{|e_i|^p}}
            \right)^2\right\}
        \end{equation}
        or jointly with $B_{ii}$ from Eq.~\eqref{eq:ekel0}, as the vector solution of
        \begin{equation}
            \min\left\{\sum_{i,j}\left(
            \frac{1}{p^2}\ln\frac{\vev{|\mat{E}_{ti}\mat{E}_{tj}|^p}}{\vev{|\mat{E}_{ti}|^p}\vev{|\mat{E}_{tj}|^p}}
            -
            \frac{1}{p^2}\ln\frac{\Esp{|e_i|^p|e_j|^p}}{\Esp{|e_i|^p}\Esp{|e_j|^p}}
            \right)^2\right\}
        \end{equation}
        \end{subequations}
        (here it is too intensive to calculate the optimum in the $N$-dimensional 
        space for all values of $p$ so we take a single value, 
        typically $p=1$ if we intend to reproduce best absolute correlations, 
               or $p=2$ if we favor quadratic correlations).
        
    \end{enumerate}
    The convergence is ensured by starting close to the solution, namely taking as prior the first eigenvector of the corresponding matrices.
    The calibration results are given graphically on Figs.~\ref{fig:AB2000-2004}, \ref{fig:AB2005-2009}, \ref{fig:AB2009-2012},
    where we show, separately for each sub-period, the estimated parameters $A_{0k}$ and $B_{0j}$.
    As expected and discussed above, they are very close to the first eigenvector of the corresponding matrix of ``log-abs'' correlations shown in Figs.~\ref{fig:2000-2004}, \ref{fig:2005-2009} and~\ref{fig:2009-2012}.
    We also show the resulting parameters $A_{kk}$ and $B_{jj}$: notably, some factors $k$ seem to have all their volatility explained by the common driver $\omega_0$
    so that there is no left residual volatility.
\begin{figure}
    \center
    \subfigure[$A_{k0}$ and $A_{kk}$]{\includegraphics[scale=.5]{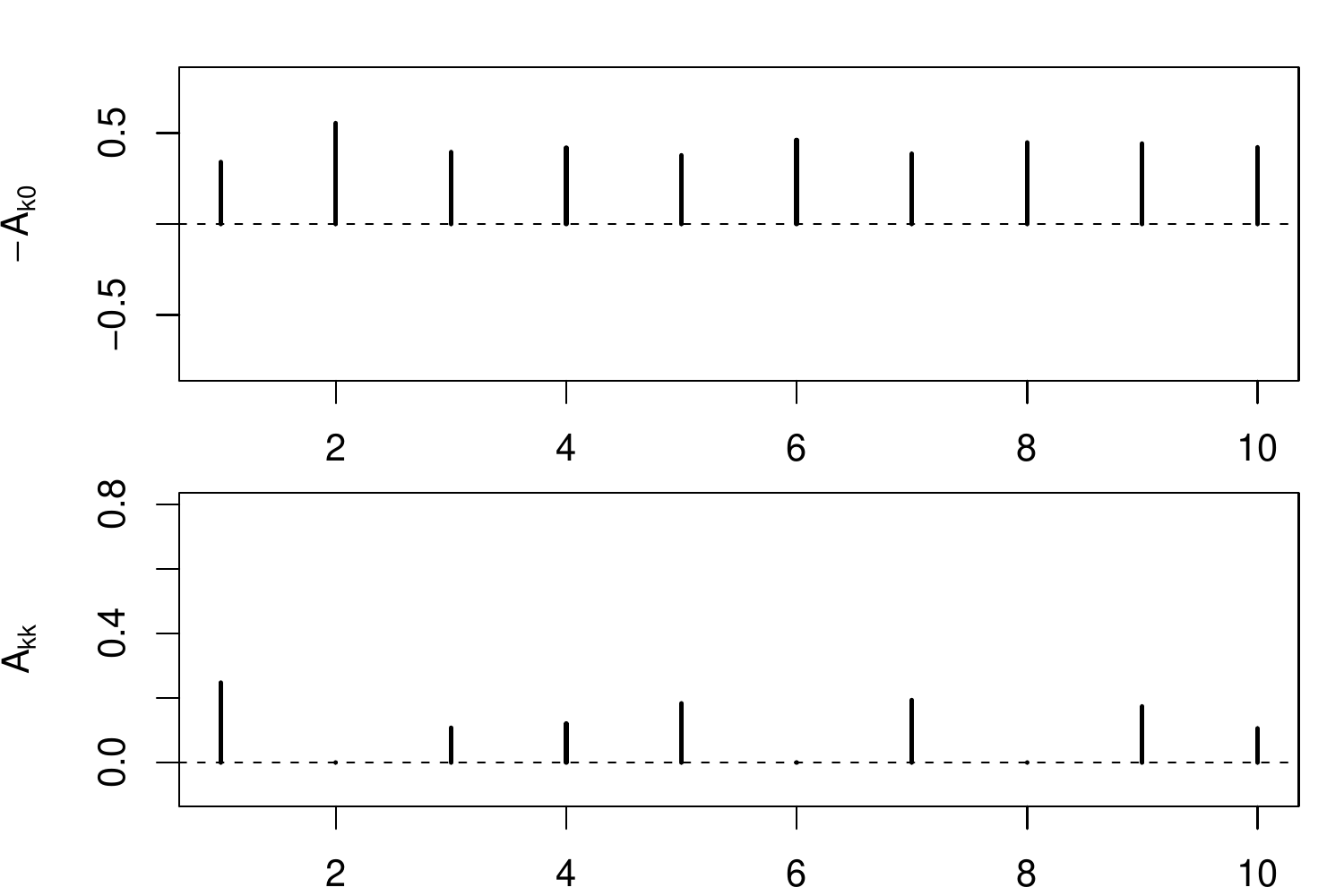}} %\hspace{.5cm}
    \subfigure[$B_{j0}$ and $B_{jj}$]{\includegraphics[scale=.5]{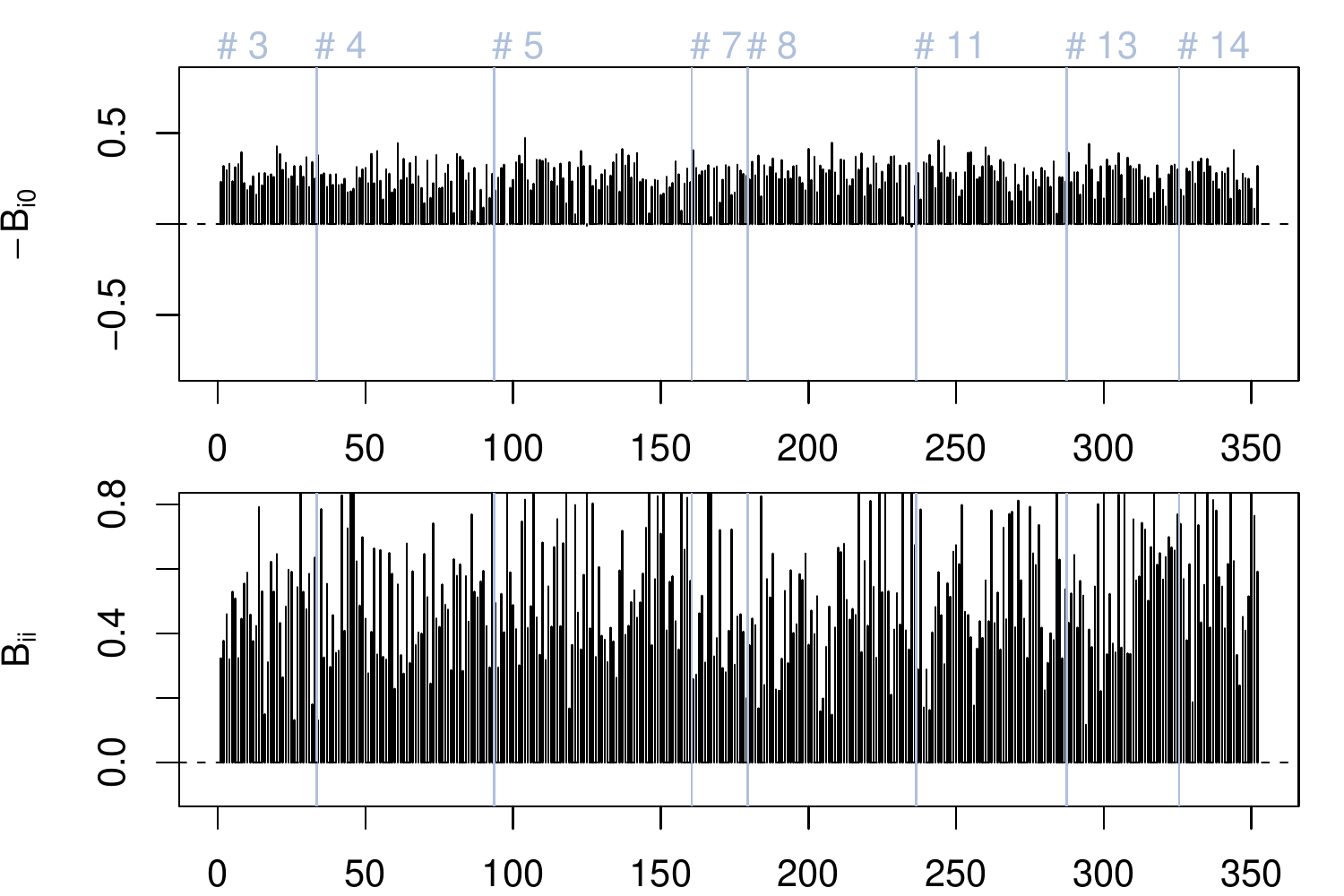}} %\hspace{.5cm}
    \caption{$M=10$, 2000-2004.}\label{fig:AB2000-2004}
    \subfigure[$A_{k0}$ and $A_{kk}$]{\includegraphics[scale=.5]{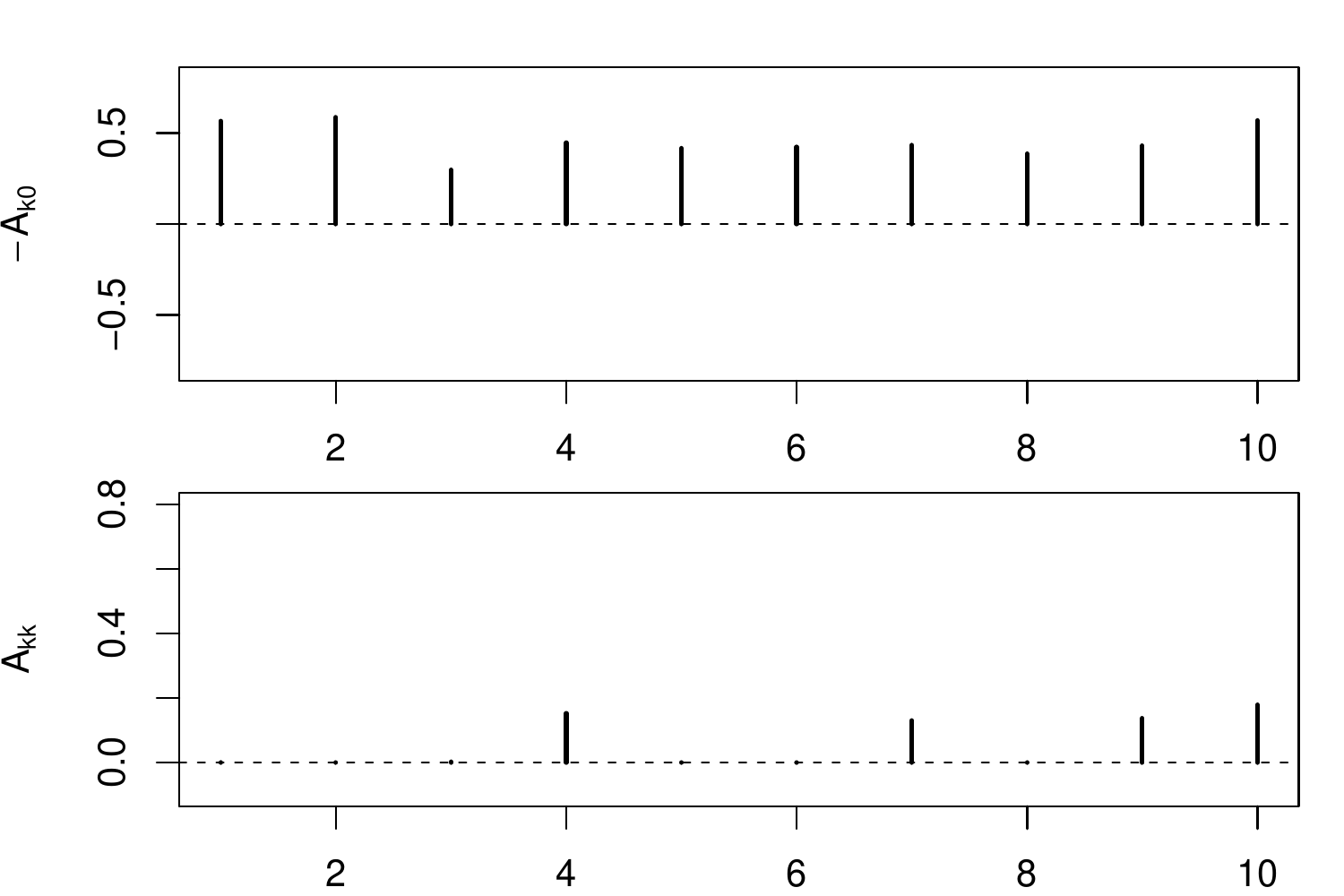}} %\hspace{.5cm}
    \subfigure[$B_{j0}$ and $B_{jj}$]{\includegraphics[scale=.5]{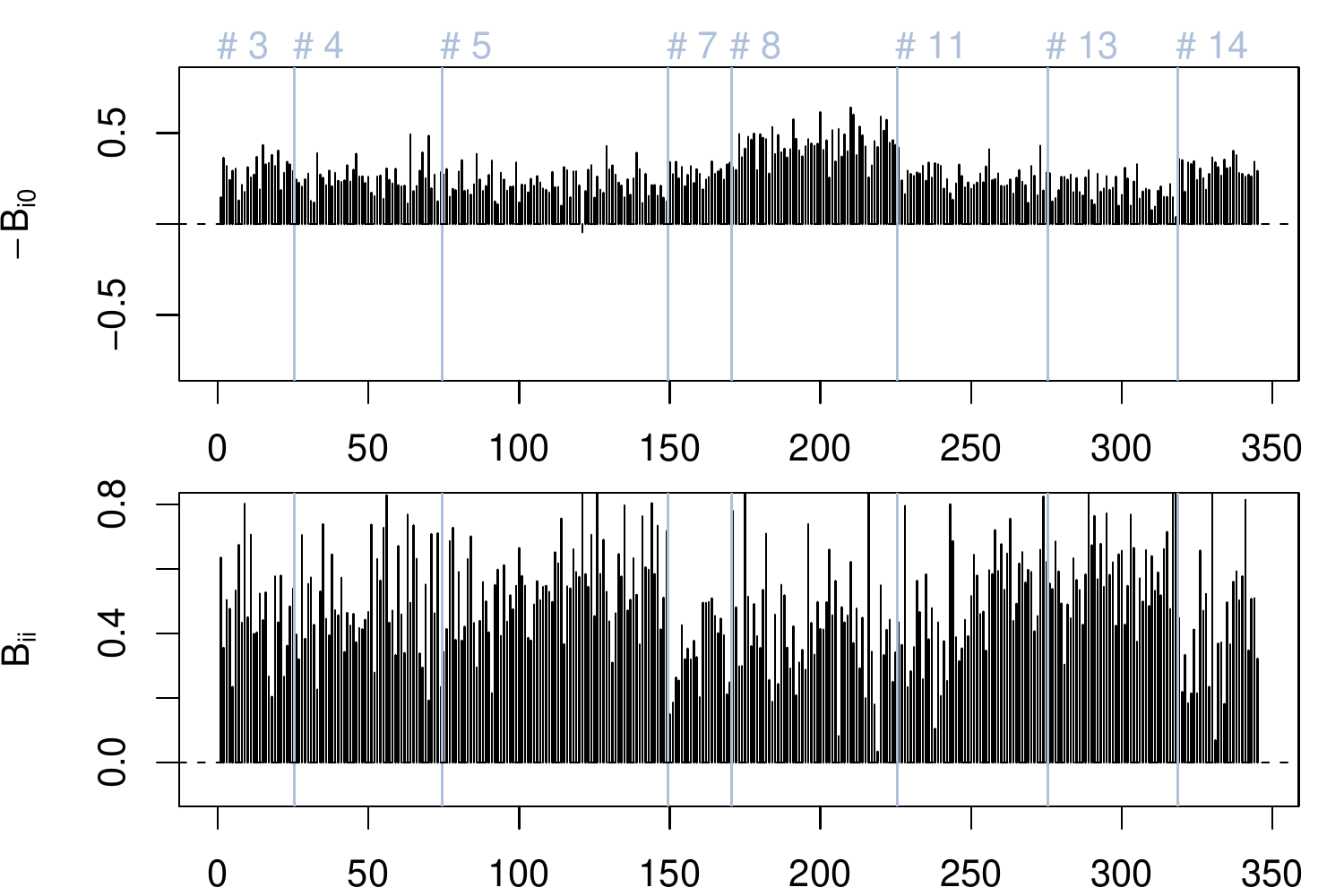}} %\hspace{.5cm}
    \caption{$M=10$, 2005-2009}\label{fig:AB2005-2009}
    \subfigure[$A_{k0}$ and $A_{kk}$]{\includegraphics[scale=.5]{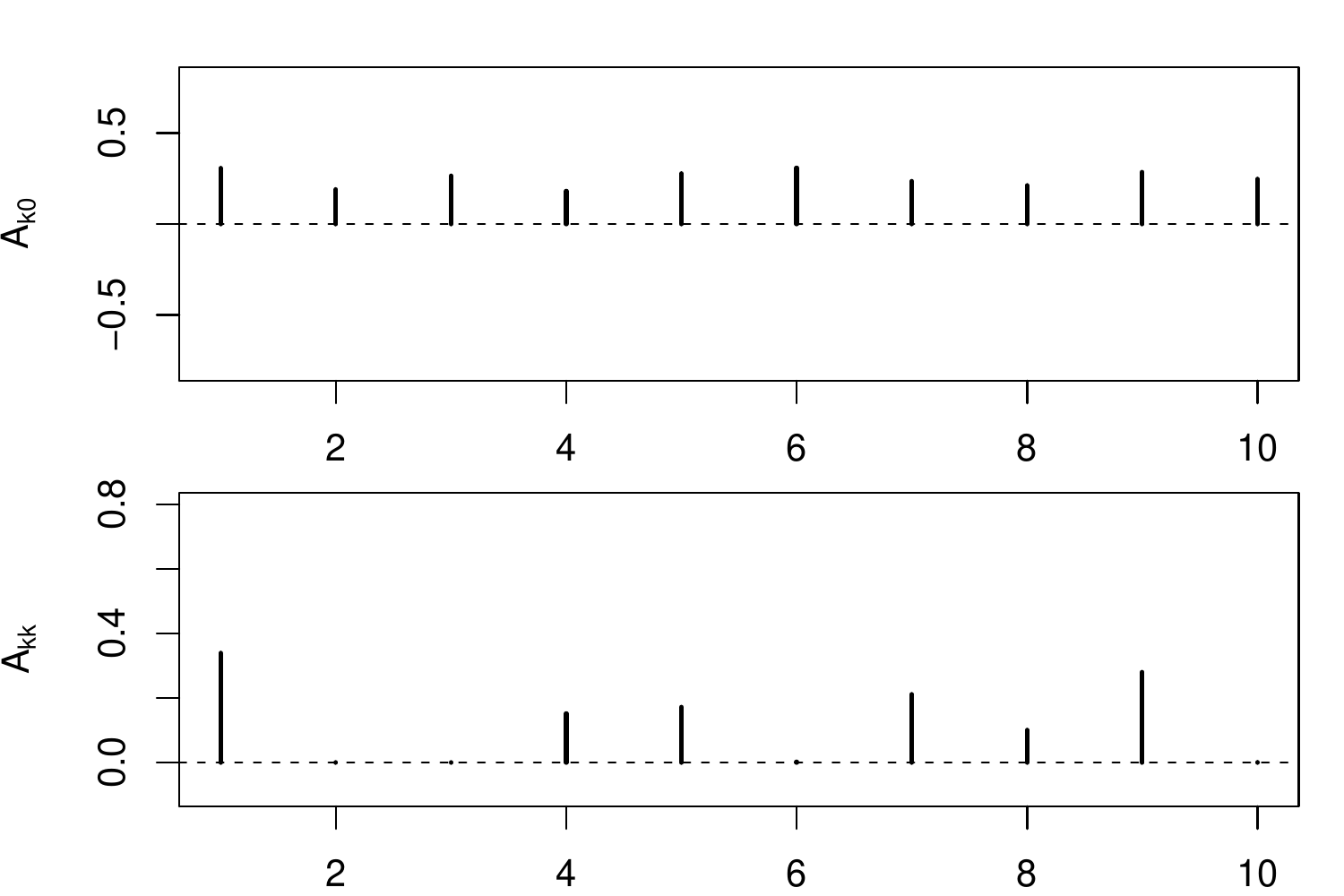}} %\hspace{.5cm}
    \subfigure[$B_{j0}$ and $B_{jj}$]{\includegraphics[scale=.5]{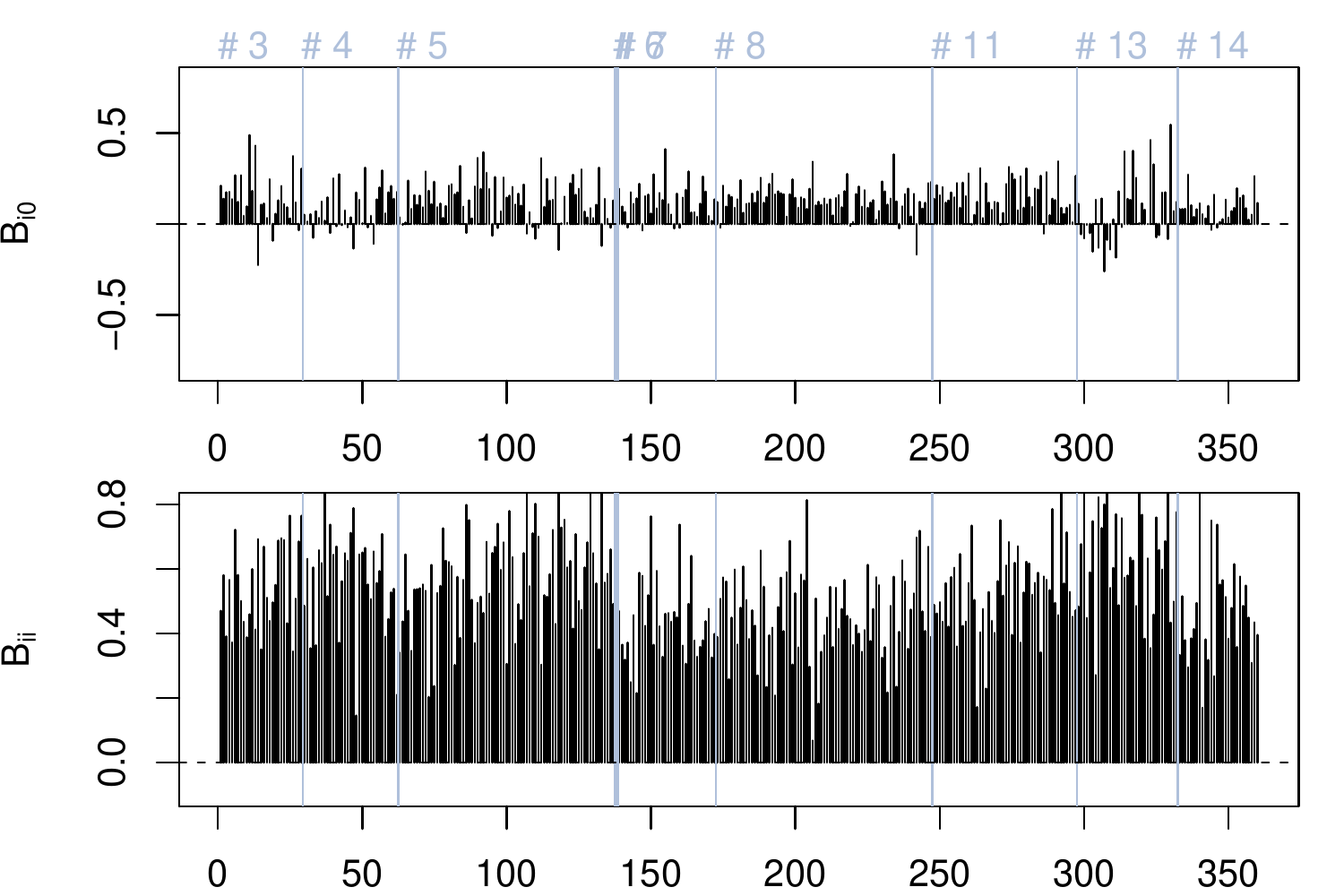}} %\hspace{.5cm}
    \caption{$M=10$, 2009-2012}\label{fig:AB2009-2012}
\end{figure}
    The estimated values of the non-gaussianity parameters of the log-volatilities are reported in Tab.~\ref{tab:ILE}.
    noticeably, the kurtosis of the common driver $\omega_0$ is found to be negative in every period:
    the log-volatility is less kurtic than a Gaussian.
    This was already revealed by the concavity of the curves in Fig.~\ref{fig:matelems}.

\begin{table}
    \center
    \begin{tabular}{c||c|c|c|}
                                       &2000--2004&2005--2009&2009--2012\\\hline\hline
                             $\zeta_0$ & -0.072   & -1.492   &  0.607   \\
                            $\kappa_0$ & -0.129   & -1.916   & -0.608   \\\hline
   %             $\vev{\omega_{t0}^3}$ &   &   & \\
   %           $\vev{\omega_{t0}^4}-3$ &   &   & \\\hline
   %$|{\rm Cor}(\omega_{t0 },F_{t1})|$ &  1.30~\% &  1.25~\% & 7.10~\%
    \end{tabular}
    \caption{$M=10$. Estimated non-Gaussianity parameters.}%, and reproduced index leverage effect.
    \label{tab:ILE}
\end{table}

\subsubsection*{Reconstructing the volatility factors}
We are interested in reconstructing the series of the common volatility mode $\omega_{t0}$
out of the model equations and the estimated parameters.
Similarly to what was done in Sect.~\ref{sec:MFlin} to recover the series of linear factors, 
we perform here two date-by-date regressions motivated by the Eqs.~\eqref{eq:model0}:   
\begin{align*}
    \ln|\mat{F}_{tk}|-\vev{\ln|\mat{F}_{tk}|}&=\omega_{t0} A_{k0} + A_{kk} \omega_{tk}\\
    \ln|\mat{E}_{tj}|-\vev{\ln|\mat{E}_{tj}|}&=\omega_{t0} B_{j0} + B_{jj} \widetilde\omega_{tj}
\end{align*}
Whereas the first regression is performed over only $M$ variables $A_{k0}$,
the second one is realized over the $N$ variables $B_{j0}$ and thus much less noisy
(we will always use the second determination in the following).
The overlap of the series of $\omega_0$ estimated with the two regressions is good, 
with a correlation coefficient between 0.55 and 0.75 depending on the period studied.
We show in Fig.~\ref{fig:firstfact} the series of first factor 
 straight from the linear program $\mat{F}_{t1}$, and also 
reconstructed from the procedure above as $\exp{(A_{10}\omega_{t0})}$, after estimation of the parameters.
    
\begin{figure}
    \center
    \subfigure[2000-2004]{\includegraphics[scale=.62,trim=230 0  200 0,clip]{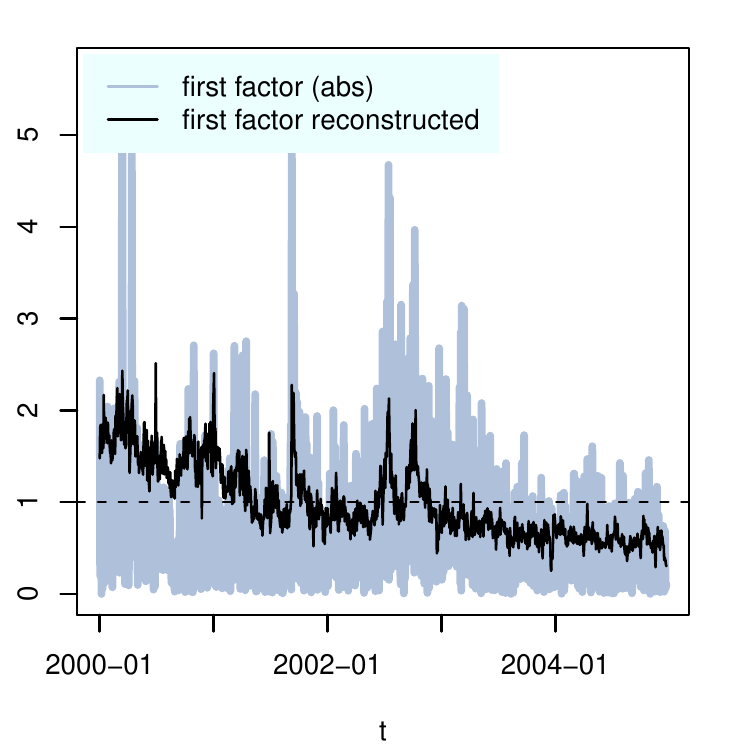}\label{fig:firstfact1mode}}
    \subfigure[2005-2009]{\includegraphics[scale=.62,trim=230 0  200 0,clip]{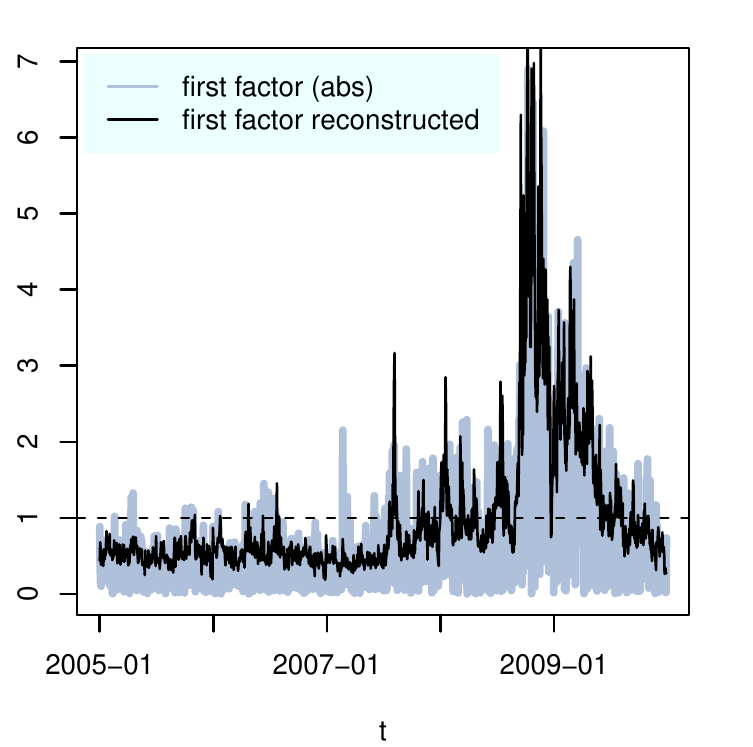}}
    \subfigure[2009-2012]{\includegraphics[scale=.62,trim=230 0  200 0,clip]{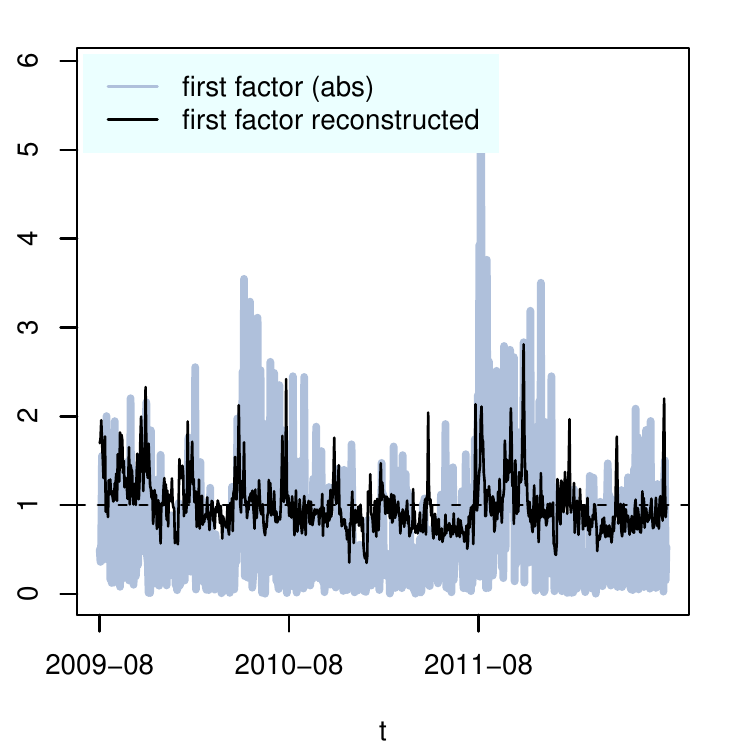}}
    \caption{Original and reconstructed first factor ($M=10$)}%: $\mat{F}_{t1}$ (bold grey), and $\exp{(A_{10}\omega_{t0})}$ (thin black)}
    \label{fig:firstfact}
\end{figure}

%%% ARTICLE
%\textcolor{red}{
%Pour l'article: Discussion de l'index leverage effect, a ajouter a Tab.~\ref{tab:ILE}.
%Comparer avec \cite{reigneron2011principal}.
%Pour l'instant les resultats sont mauvais: la correl $\vev{\omega,f_1}$ ne decroit pas en amplitude avec le lag ???
%}

\subsubsection*{Predicted non-linear dependences}
As a consistency check of the estimation procedure, we analyze
the model prediction with estimated parameters and compare them with 
empirical measurements of the same quantities.
Of particular interest are the quadratic correlations and the diagonal copulas,
whose anomalies observed in a previous study \cite{chicheportiche2012joint} motivated the present model.

The quadratic correlations can be computed from the model definition and write
\begin{align}\nonumber
\esp{\ret_i^2\ret_j^2}  &=\sum_{kl}\left(\Wei_{ki}^2\Wei_{lj}^2+2\Wei_{ki}\Wei_{kj}\Wei_{li}\Wei_{lj}\right)\Phi_0(A_{k0},A_{l0};2)\Big(\tfrac{1}{3}\cdot 3\cdot \Phi_{\text{G}}(A_{kk},A_{ll};2)\Big)^{\delta_{kl}} \\\nonumber
                        &+(1+2\delta_{ij})\Big(1-\sum_l\Wei_{li}^2\Big)\sum_{k}\Wei_{kj}^2\Phi_0(A_{k0},B_{i0};2)\\\nonumber
                        &+(1+2\delta_{ij})\Big(1-\sum_l\Wei_{lj}^2\Big)\sum_{k}\Wei_{ki}^2\Phi_0(A_{k0},B_{j0};2)\\\label{eq:quad_cor_model}
                        &+\Big(1-\sum_l\Wei_{li}^2\Big)\Big(1-\sum_l\Wei_{lj}^2\Big)\Phi_0(B_{i0},B_{j0};2)\Big(3\Phi_{\text{G}}(B_{ii},B_{jj};2)\Big)^{\delta_{ij}}.
\end{align}
When all parameters $A,B$ are zero, the prediction for Gaussian factors and residuals is retrieved: $\esp{\ret_i^2\ret_j^2}=1+2\esp{\ret_i\ret_j}^2$.
We illustrate in the left panel of Fig.~\ref{fig:cal_quad} a scatter plot of 
the left-hand side (calibrated) versus the right-hand side (empirical) of Eq.~\eqref{eq:quad_cor_model}, 
for all periods. They show a good agreement of model and sample quadratic correlations.
Furthermore, the middle and right panels of the same figure illustrate the fact that the
pairs of stock returns cannot be described by a bivariate Student distribution, 
for which a regular curve should be observed instead of the scattered cloud 
in the plane of quadratic vs linear correlations, 
a conclusion already drawn in chapter~\ref{part:partII}.\ref{chap:IJTAF} but which the present model highlights.%
\footnote{
Notice that the choice of $p$ in the estimation procedure of the parameters $B_{i0}$, $B_{i0'}$ and $B_{ii}$ is important here.
Although the model has identical predictions for any $p$, estimation biases and errors are in practice different for low moments $p\approx 0.2$ or high moments $p\approx 2$.
Obviously, best fits for the quadratic correlations are obtained with $p=2$ since in this case 
the same quantities appear in Eq.~\eqref{eq:quad_cor_model} and in the loss function \eqref{eq:loss_B}.
}
\begin{figure}
    \center
    \subfigure[2000--2004]{\includegraphics[scale=.38,trim=0 0 0 25,clip]{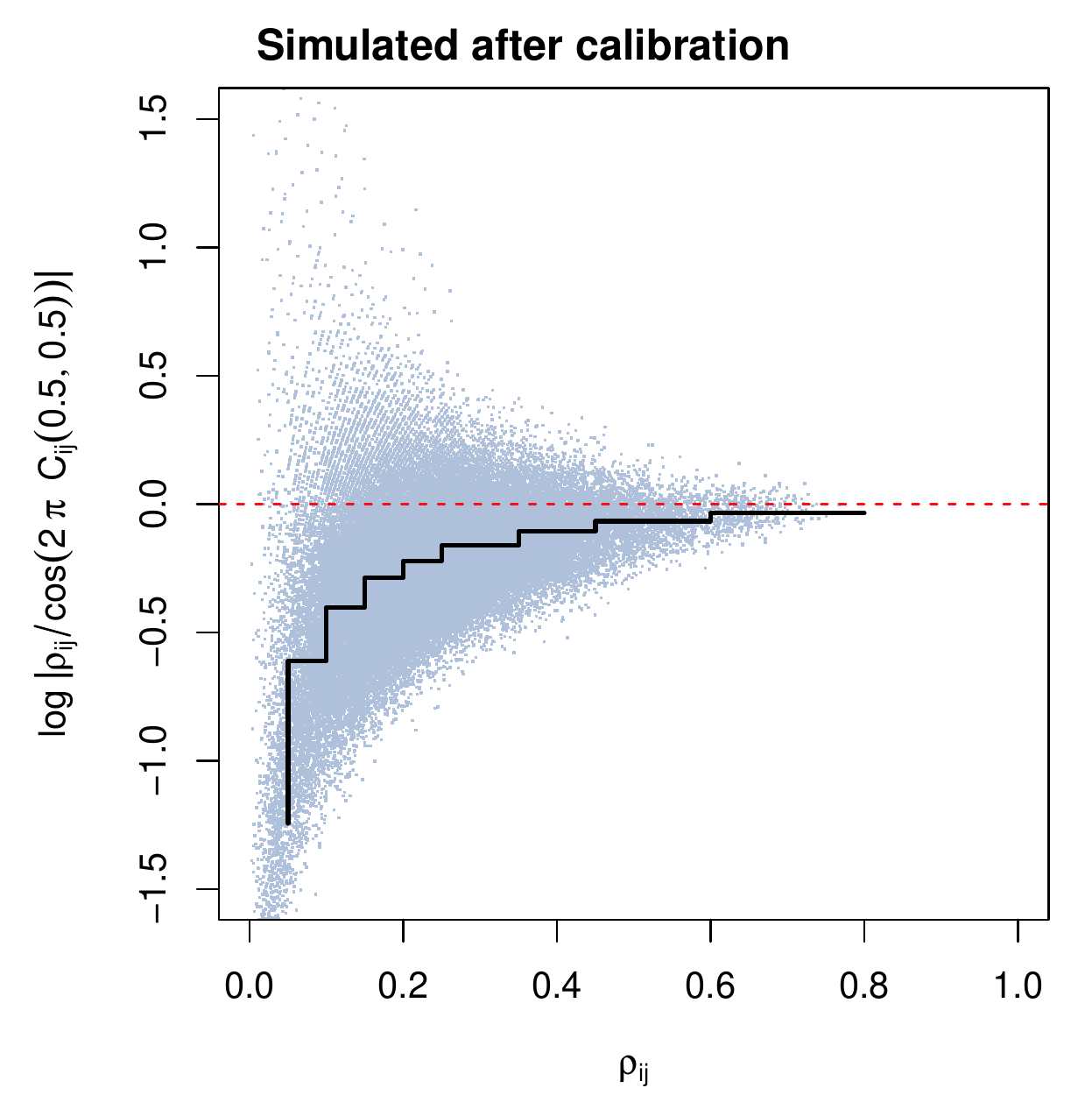}}
    \subfigure[2005--2009]{\includegraphics[scale=.38,trim=0 0 0 25,clip]{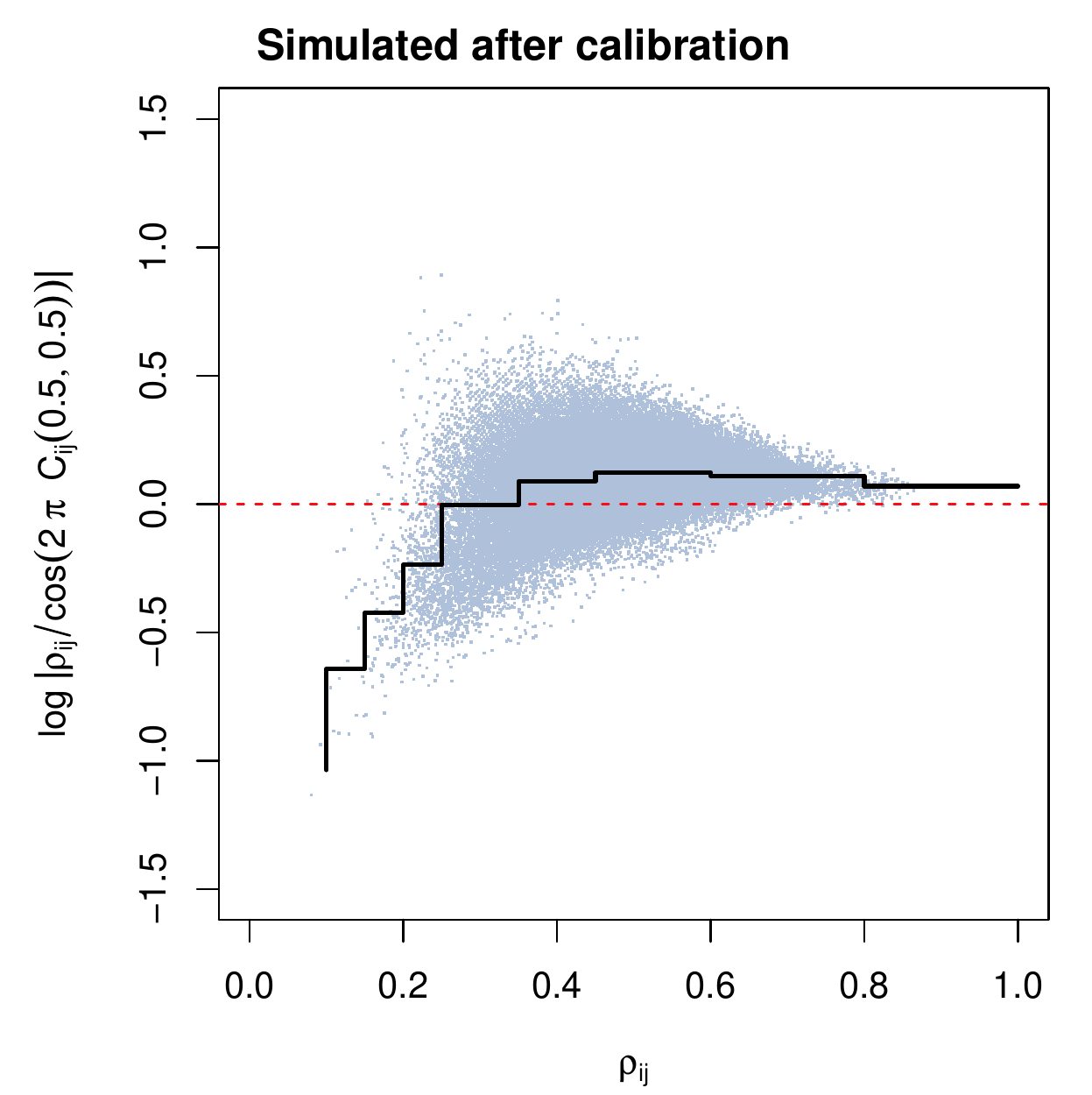}}
    \subfigure[2009--2012]{\includegraphics[scale=.38,trim=0 0 0 25,clip]{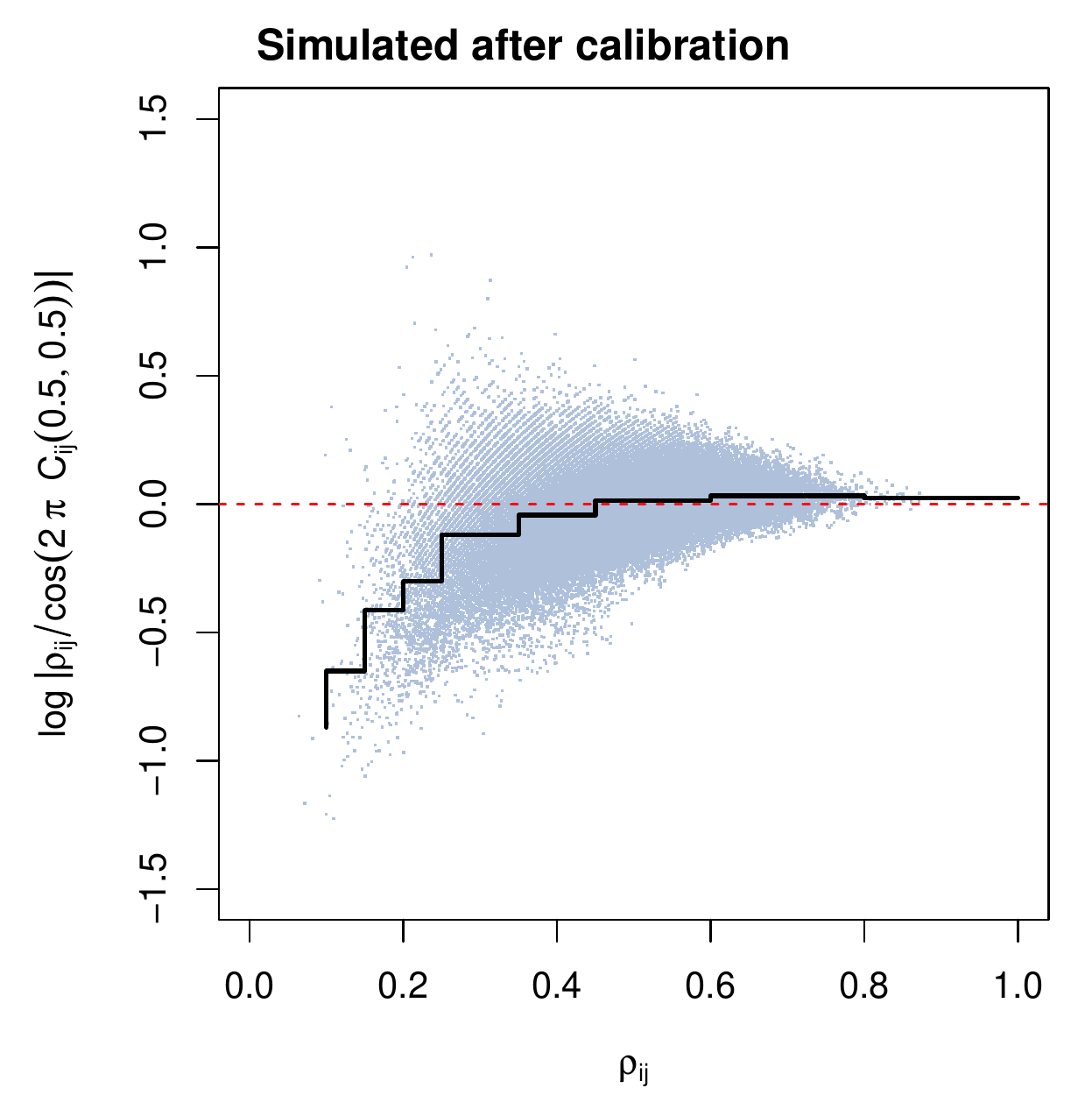}}
    \caption{Calibrated model. $\ln|\rho/\rhoB|$ vs $\rho$, see Eq.~\eqref{eq:hatrho}.}\label{fig:sim_cop}
\end{figure}

\begin{figure}[p]
    \center
    \subfigure[2000--2004]{\includegraphics[scale=.2,trim=750 50 0 50,clip]{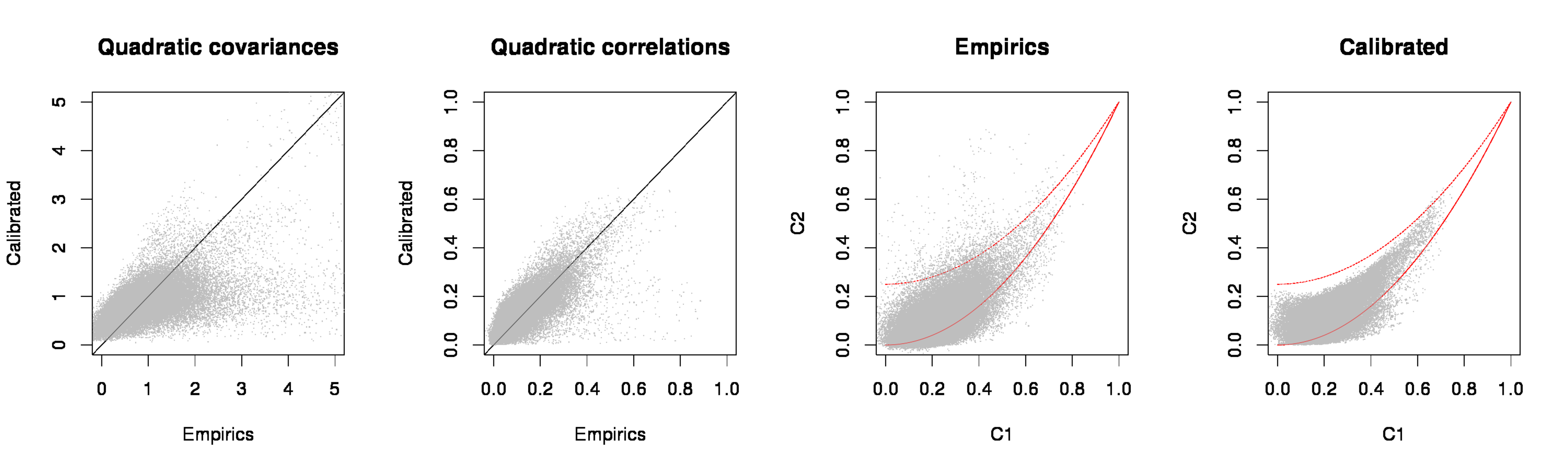}}\\
    \subfigure[2005--2009]{\includegraphics[scale=.2,trim=750 50 0 50,clip]{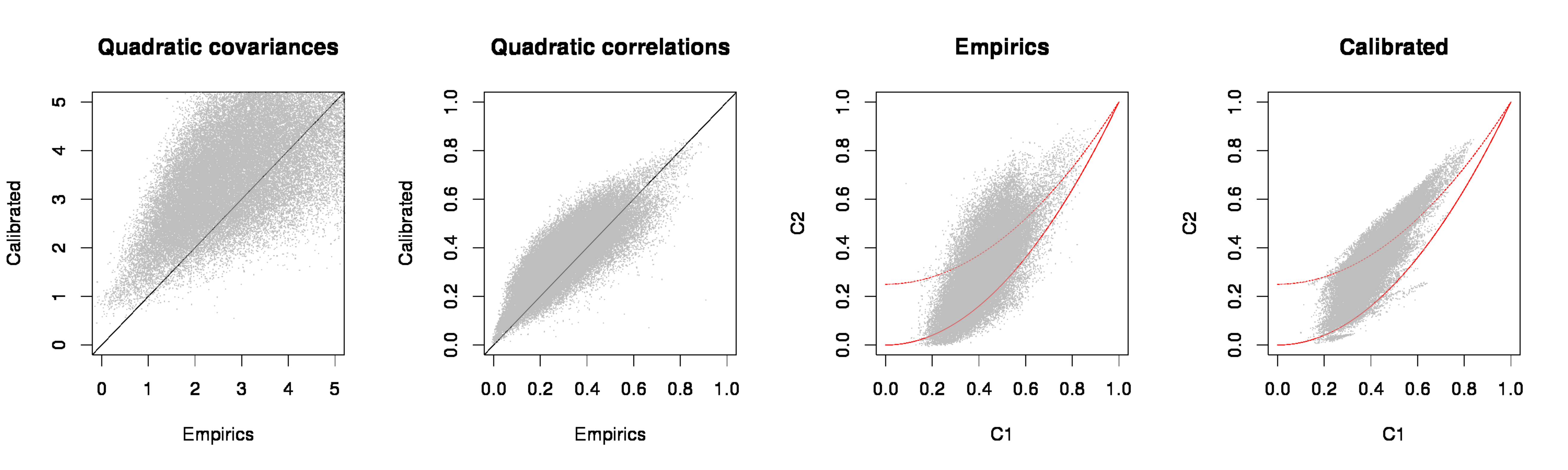}}\\
    \subfigure[2009--2012]{\includegraphics[scale=.2,trim=750 50 0 50,clip]{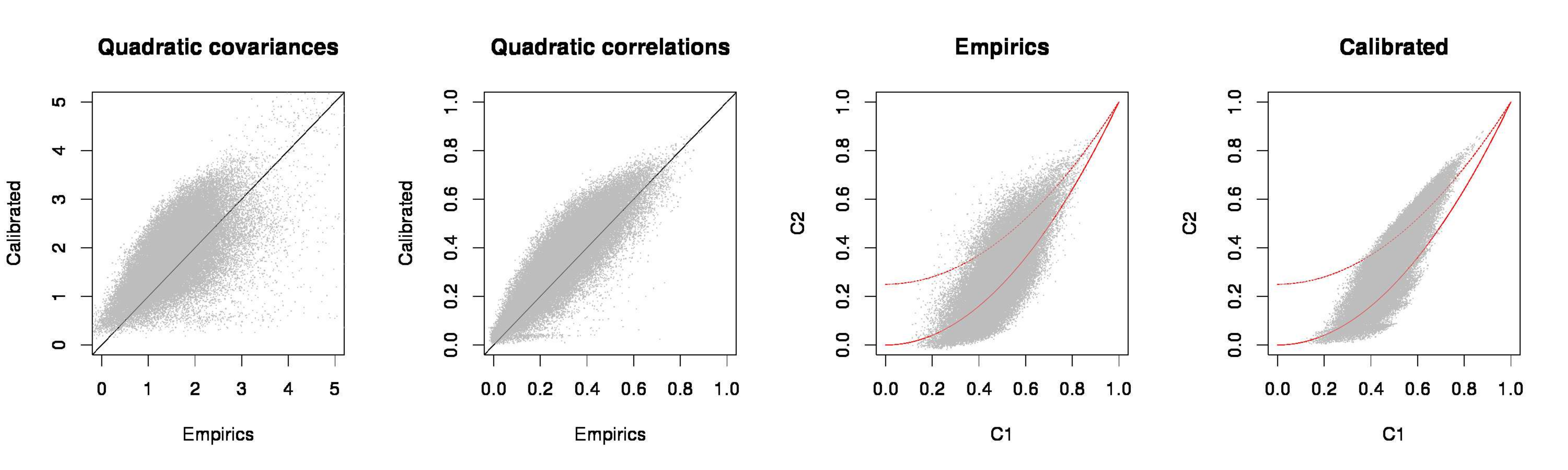}}
    \caption{\textbf{Left: }   calibrated vs sample quadratic correlations. 
             \textbf{Middle: } sample quadratic correlations vs sample linear correlations;
             \textbf{Right: }  calibrated quadratic correlations vs calibrated linear correlations.
             Two benchmark curves are added in red: the Gaussian case (solid) and the Student case with $\nu=5$ d.o.f.\ (dashed)\label{fig:cal_quad}}
\end{figure}
\nomenclature{d.o.f.}{Degree(s) of freedom}

The middle point $\cop(\frac12,\frac12)$ of the copula, which was shown on Fig.~\ref{fig:emp_cop} to be incompatible with any elliptical prediction,
is also better explained by the model.
Although an analytical expression linking  $\cop(\frac12,\frac12)$ to the model parameters is out of reach, 
it is possible to reproduce its predicted value by simulating long time series according to the model with estimated parameters.%
\footnote{The non-Gaussian series of log-volatility $\omega_{t0}$ is generated as independent realizations of a Beta distribution
whose coefficients are determined so that the first four moments match those of $\omega_0$.
This class of distributions allows for negative kurtosis.
It is known that the realizations of volatility exhibit strong persistence, 
a characteristic that our simulated series do not reproduce.
This however does not generate a bias in the obtained coefficients, but rather makes them ``too un-noisy''.
}
Fig.~\ref{fig:sim_cop} shows the obtained coefficients in the same form as the empirical measurements of Fig.~\ref{fig:emp_cop}, 
which they have to be compared with. 
The results are in very good agreement, and emphasize the capacity of our non-Gaussian factor model to cope with the
non-trivial behavior of the medial point of the copula.

This is confirmed and even generalized by the reproduction of the copula along the whole diagonals.
Figs.~\ref{fig:side_cop20002004}, \ref{fig:side_cop20052009} and~\ref{fig:side_cop20092012}
compare empirically measured and model-predicted values of the quantities
\[
    \Delta_{\scriptscriptstyle \text{d}}(p)=\frac{\cop(p,p)-\cop[\text{G}](p,p)}{p\,(1-p)}
    \quad\text{and}\quad
    \Delta_{\scriptscriptstyle \text{a}}(p)=\frac{\cop(p,1-p)-\cop[\text{G}](p,1-p)}{p\,(1-p)}
\]
versus $p$, for several values of the linear correlation and over the three periods.
A direct visual comparison reveals that the main non-trivial qualitative features of the empirical diagonal copulas
are reproduced by our model: the evolution of the concavity as $\rho$ changes, some of the tail behaviors, and the middle-point behavior as already discussed above.
One may note however that asymmetries $u\leftrightarrow 1-u$ are not reproduced by our fully symmetric model.
\begin{sidewaysfigure}%[p]
\center
\subfigure[Empirical]{ \includegraphics[scale=.5]{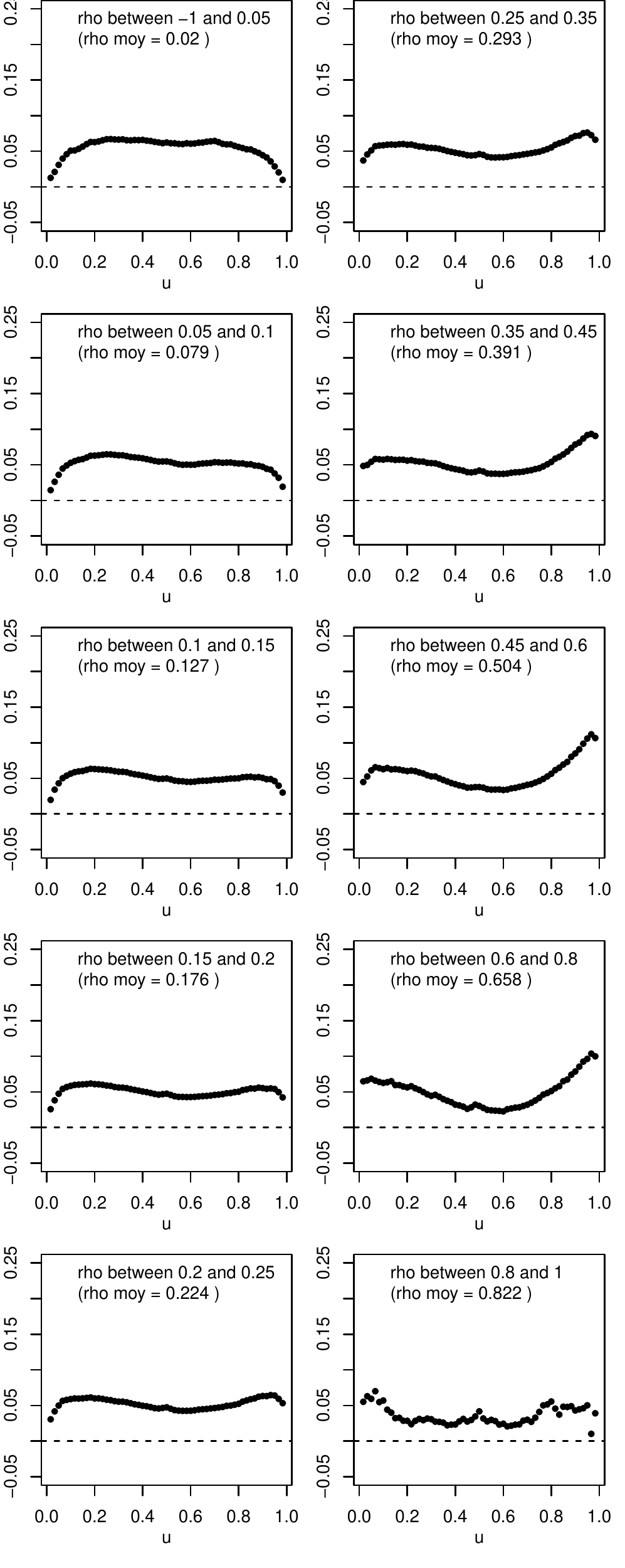}}
\subfigure[Calibrated]{\includegraphics[scale=.5]{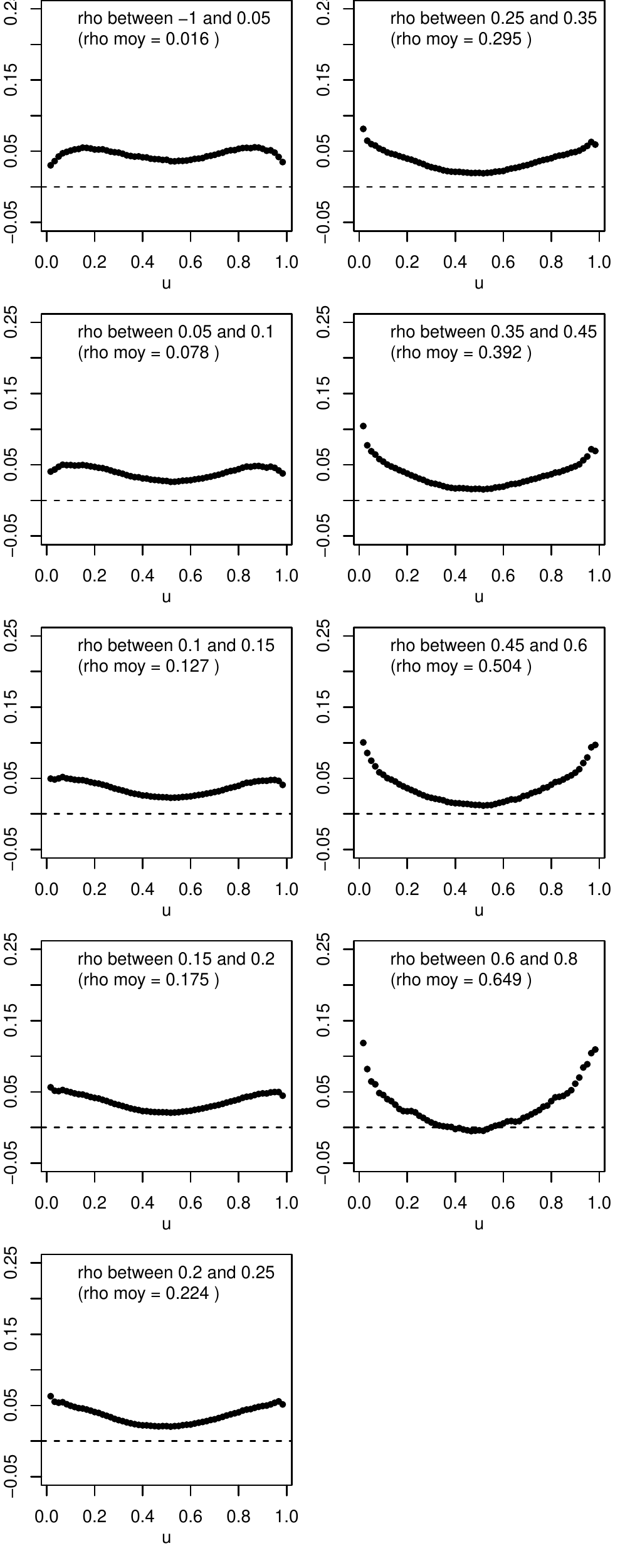}}
\hfill
\subfigure[Empirical]{ \includegraphics[scale=.5]{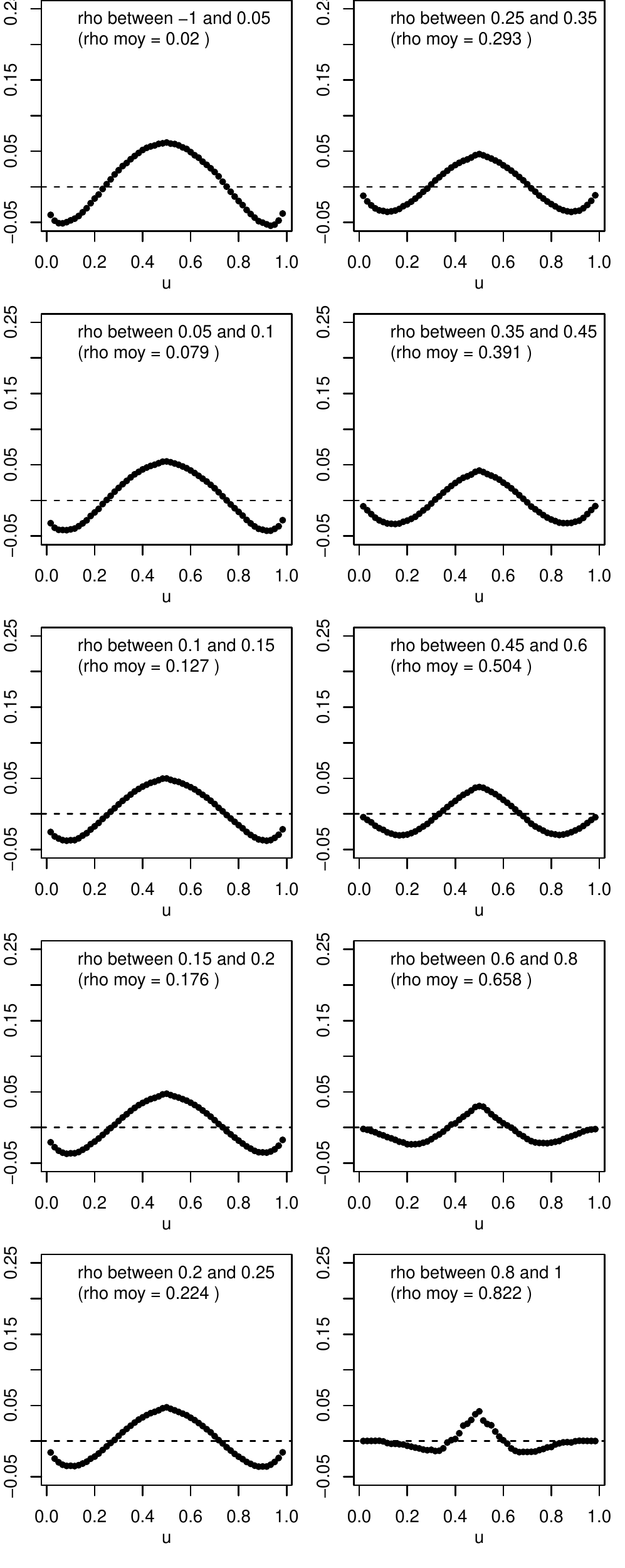}}
\subfigure[Calibrated]{\includegraphics[scale=.5]{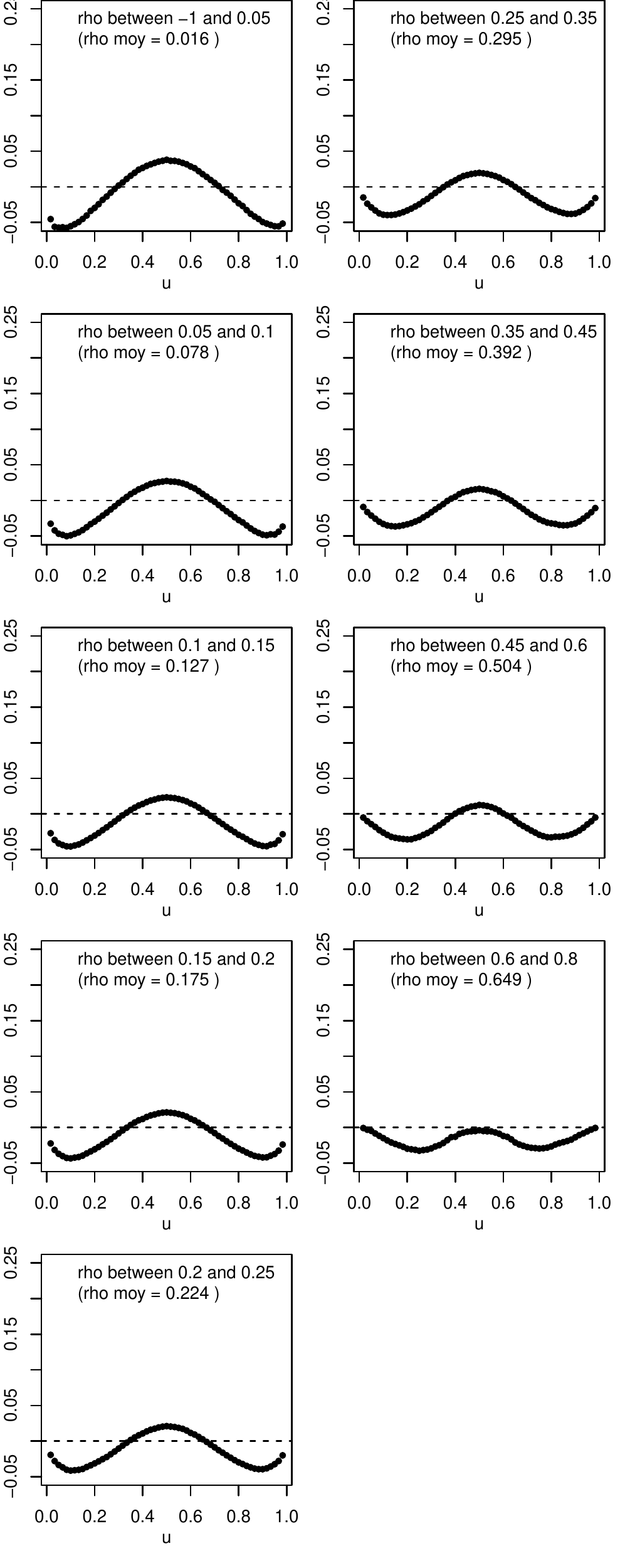}}
\caption{2000--2004, diagonal (left) and anti-diagonal (right)}
\label{fig:side_cop20002004}
\end{sidewaysfigure}
\begin{sidewaysfigure}
\center
\subfigure[Empirical]{ \includegraphics[scale=.5]{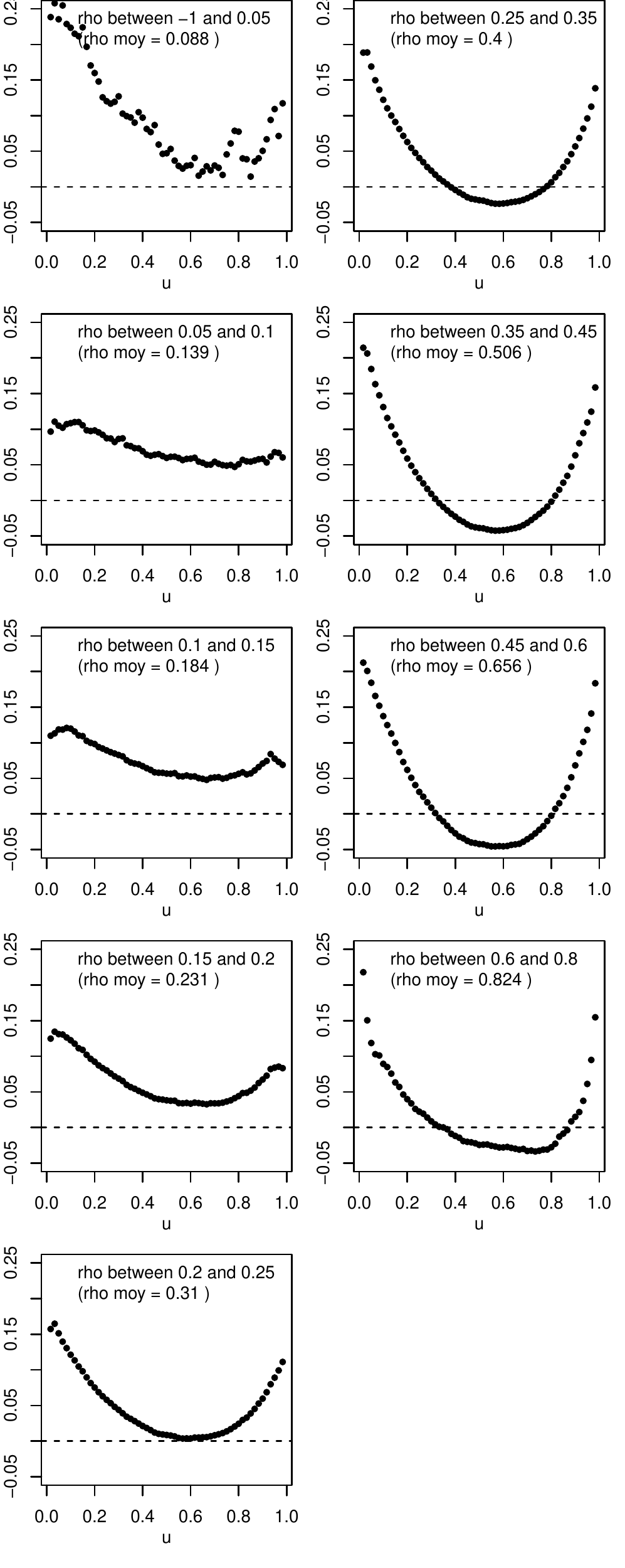}}
\subfigure[Calibrated]{\includegraphics[scale=.5]{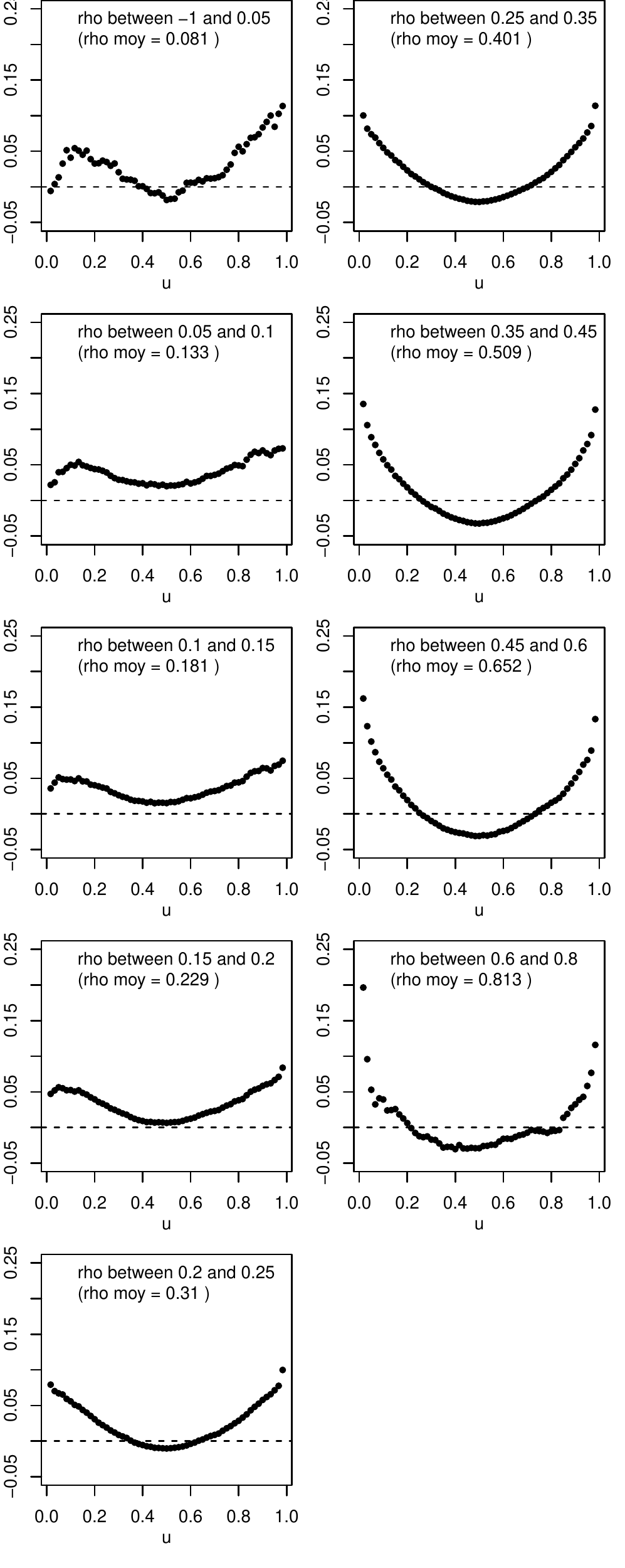}}
\hfill
\subfigure[Empirical]{ \includegraphics[scale=.5]{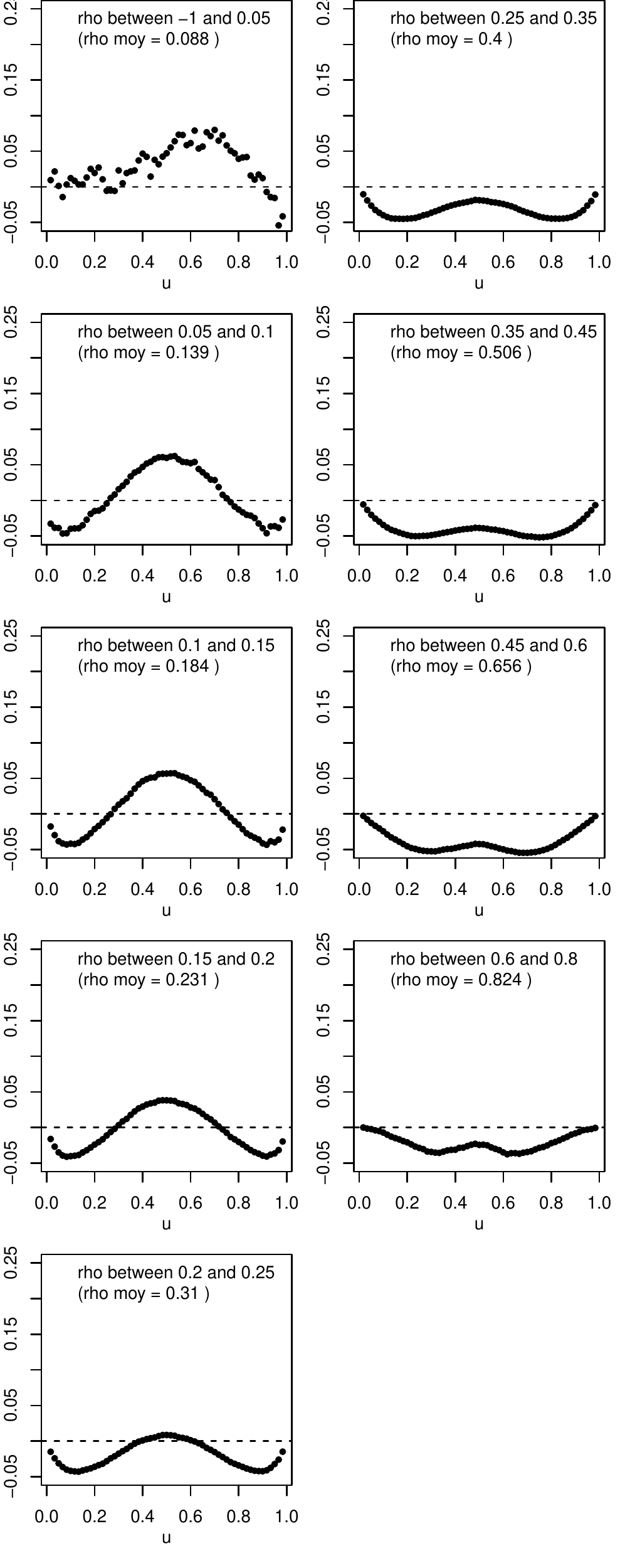}}
\subfigure[Calibrated]{\includegraphics[scale=.5]{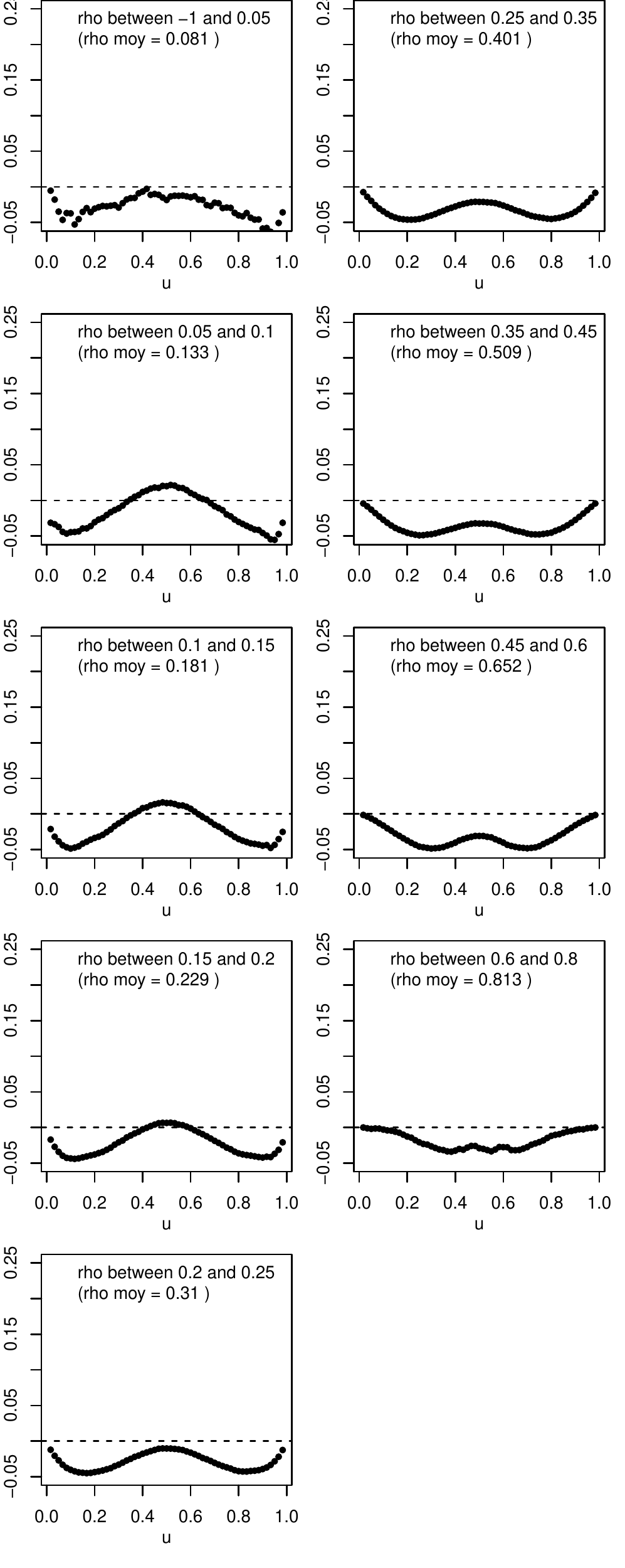}}
\caption{2005--2009, diagonal (left) and anti-diagonal (right)}
\label{fig:side_cop20052009}
\end{sidewaysfigure}
\begin{sidewaysfigure}
\center
\subfigure[Empirical]{ \includegraphics[scale=.5]{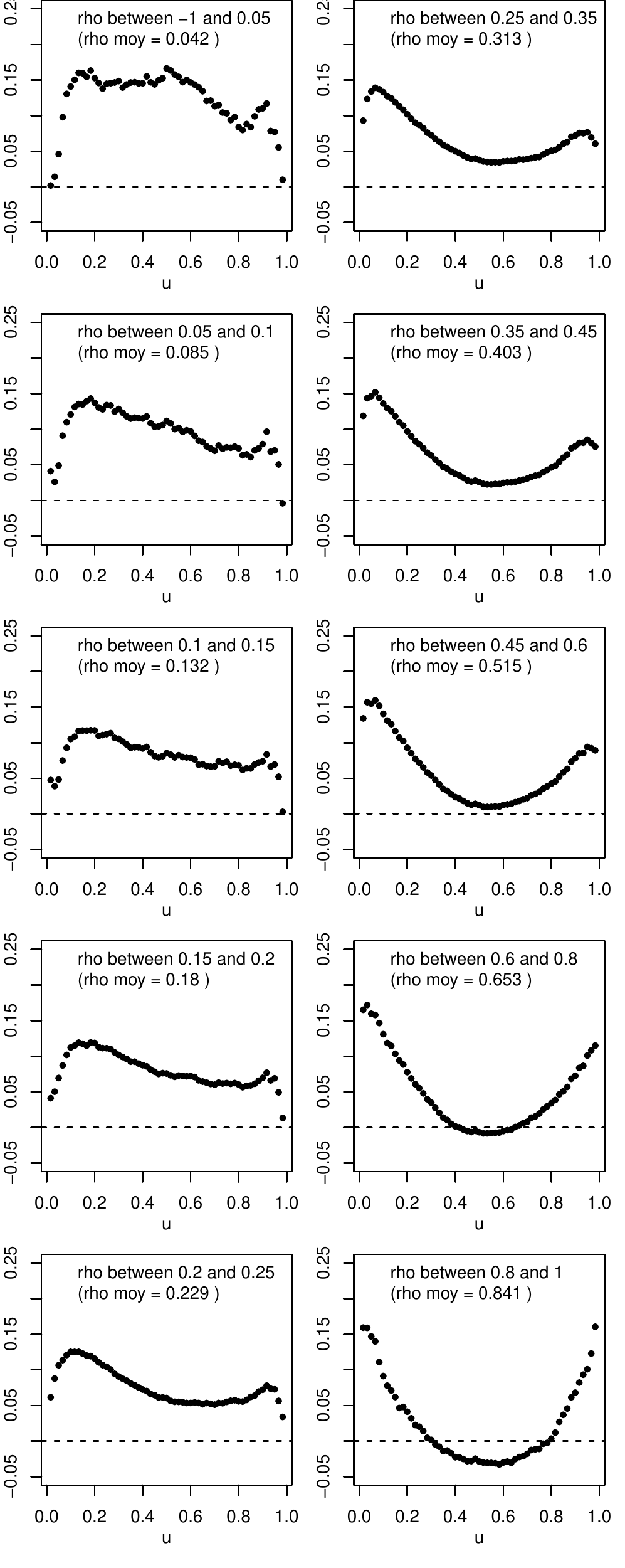}}
\subfigure[Calibrated]{\includegraphics[scale=.5]{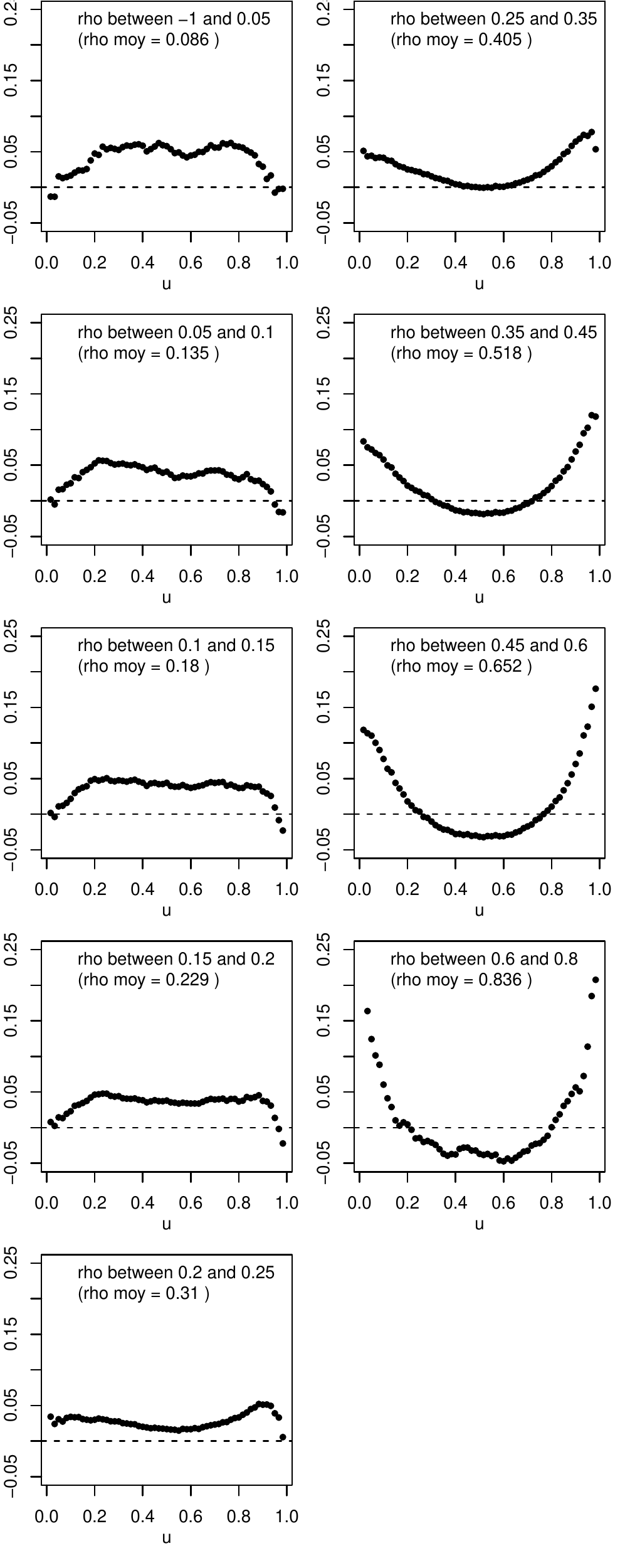}}
\hfill
\subfigure[Empirical]{ \includegraphics[scale=.5]{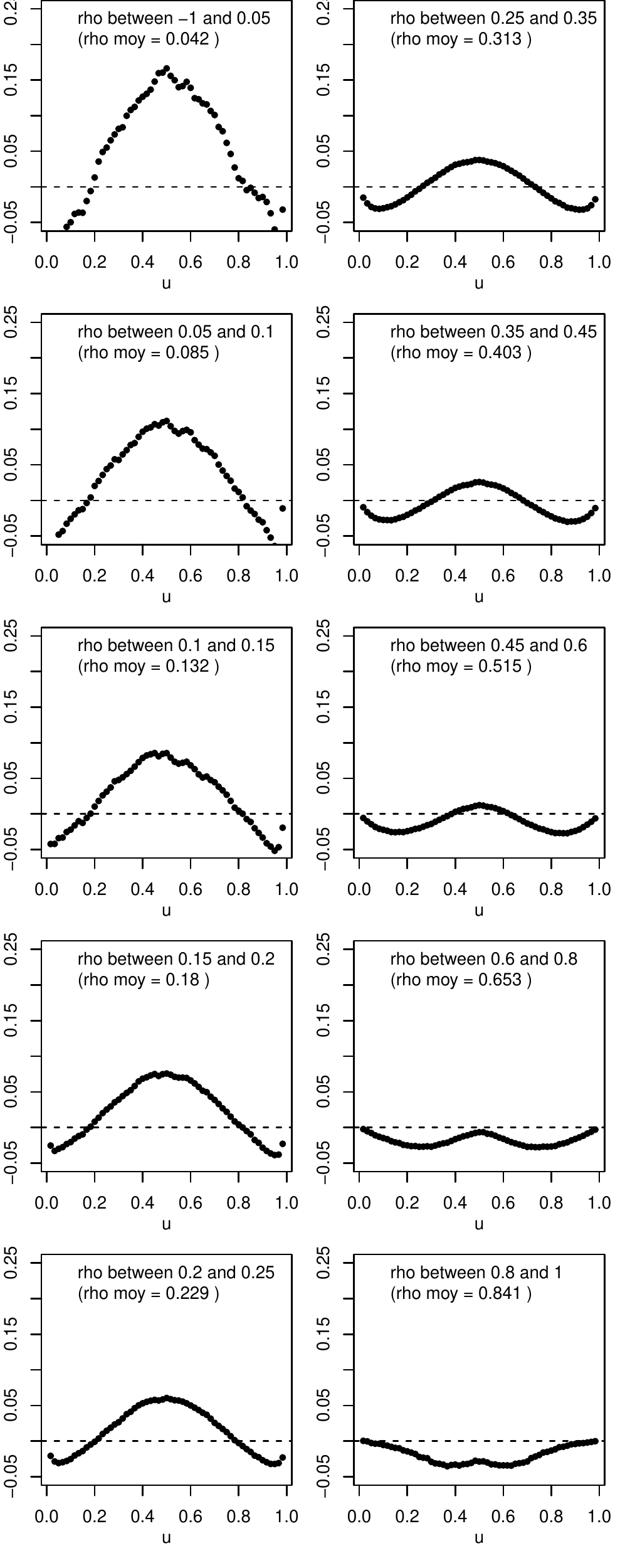}}
\subfigure[Calibrated]{\includegraphics[scale=.5]{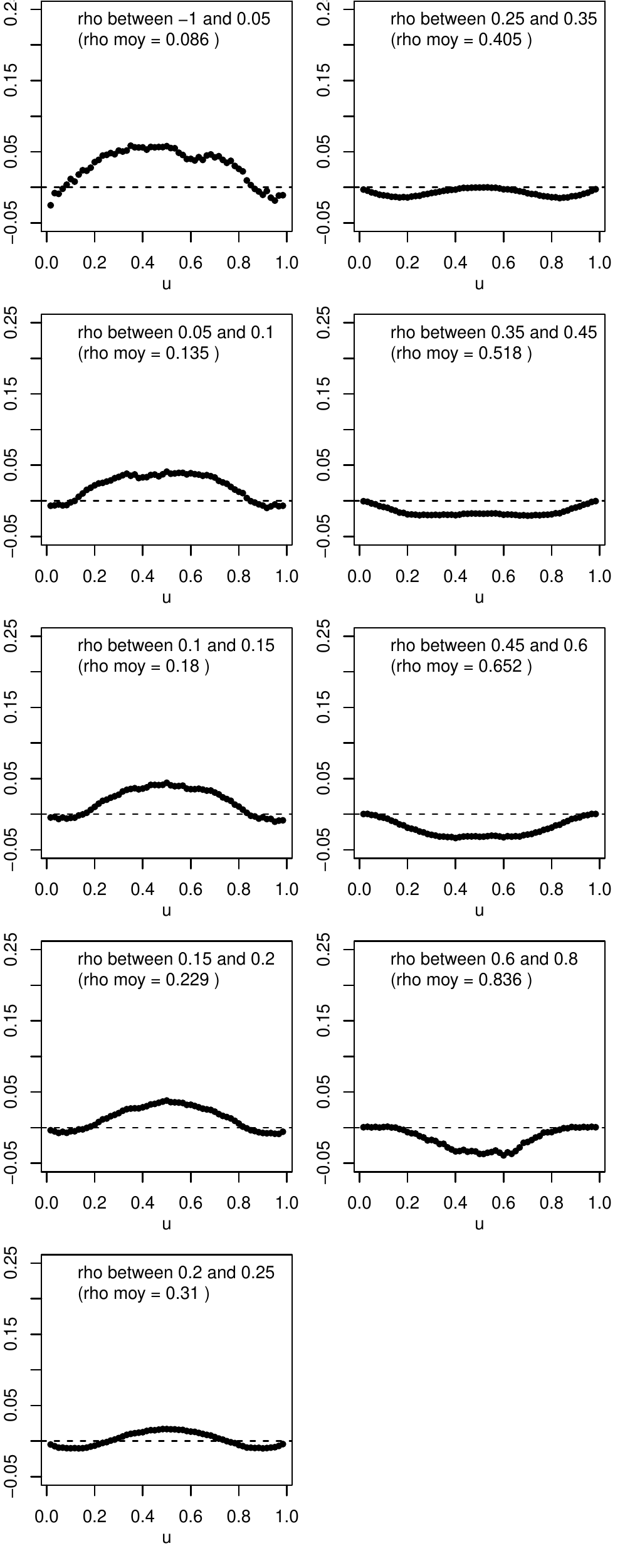}}
\caption{2009--2012, diagonal (left) and anti-diagonal (right)}
\label{fig:side_cop20092012}
\end{sidewaysfigure}

Before proceeding to the out-of-sample evaluation of the model in Sect.~\ref{sec:stability},
we briefly present how the model is improved (though marginally)
by introducing a second volatility driver $\omega_{0'}$, as suggested by the spectral analysis of Sect.~\ref{sec:MFspectral}.

\subsection{Another volatility driver}
The spectral analysis of the factor and residual absolute correlations has revealed
that there exists a small but significant second mode of volatility.
The model~\eqref{eq:model0} can be improved accordingly in order to account for this additional source of collective amplitude fluctuations:
\begin{align}\label{eq:anothervolatility}
    f_k&=\epsilon_k \exp(A_{k0} \omega_0 + A_{k0'} \omega_{0'} + A_{kk} \omega_k)\\
    e_j&=\eta_j     \exp(B_{j0} \omega_0 + B_{j0'} \omega_{0'} + B_{jj} \widetilde\omega_j),
\end{align}
and as a consequence the RHS of Eqs.~\eqref{eq:fkfleiej} get an additional term each, namely 
$\phi_{0'}(A_{k0'},A_{l0'};p)$, $\phi_{0'}(A_{k0'},B_{i0'};p)$ and $\phi_{0'}(B_{i0'},B_{j0'};p)$, respectively.

The whole estimation procedure runs identically.
However, for the determination of the parameters $A_{k0}$ and $A_{k0'}$, 
the reduced number of observations ($M(M-1)/2$ factor-factor correlations, times 8 values of $p$)
provides only a low resolution, and the minimization program does not succeed in distinguishing the two volatility drivers:
it outputs an hybrid where both $\omega_0$ and $\omega_{0'}$ contribute to the same mode
(this issue is not relevant in the determination of the parameters $B_{j0}$, $B_{j0'}$ 
thanks to the sufficient number of observations).
In order to break the degeneracy and ``orthogonalize'' the modes, we add an overlap term $\left(\sum_k A_{k0}A_{k0'}\right)^2$ 
in the cost function Eq.~\eqref{eq:costfct_A}.

As an example, we report in Fig.~\ref{fig:2modes0004} the results for the period 2000--2004.
As expected, the parameters $A_{k0}$ and $A_{k0'}$ are very close to the first two 
eigenvectors of the factor-factor ``log-abs'' correlation matrix (Fig.~\ref{fig:0004EVff}), 
and the parameters $B_{j0}$, $B_{j0'}$ look like the first two eigenvectors 
of the residual-residual  matrix (Fig.~\ref{fig:0004EVee}, averaged over sectors).
Clearly, taking this additional second volatility driver into consideration
improves the theoretical description of the returns.
We illustrate this on Fig.~\ref{fig:0004_2_O} where we show how $\omega_0$ et $\omega_{0'}$ contribute respectively
to the volatility of the market mode of linear correlation, $f_1$.
\begin{figure}
    \center
    \subfigure[$A_{k0}$ and $A_{k0'}$]{\label{fig:0004_2_A}\includegraphics[scale=.5,trim=  0 140 0 0,clip]{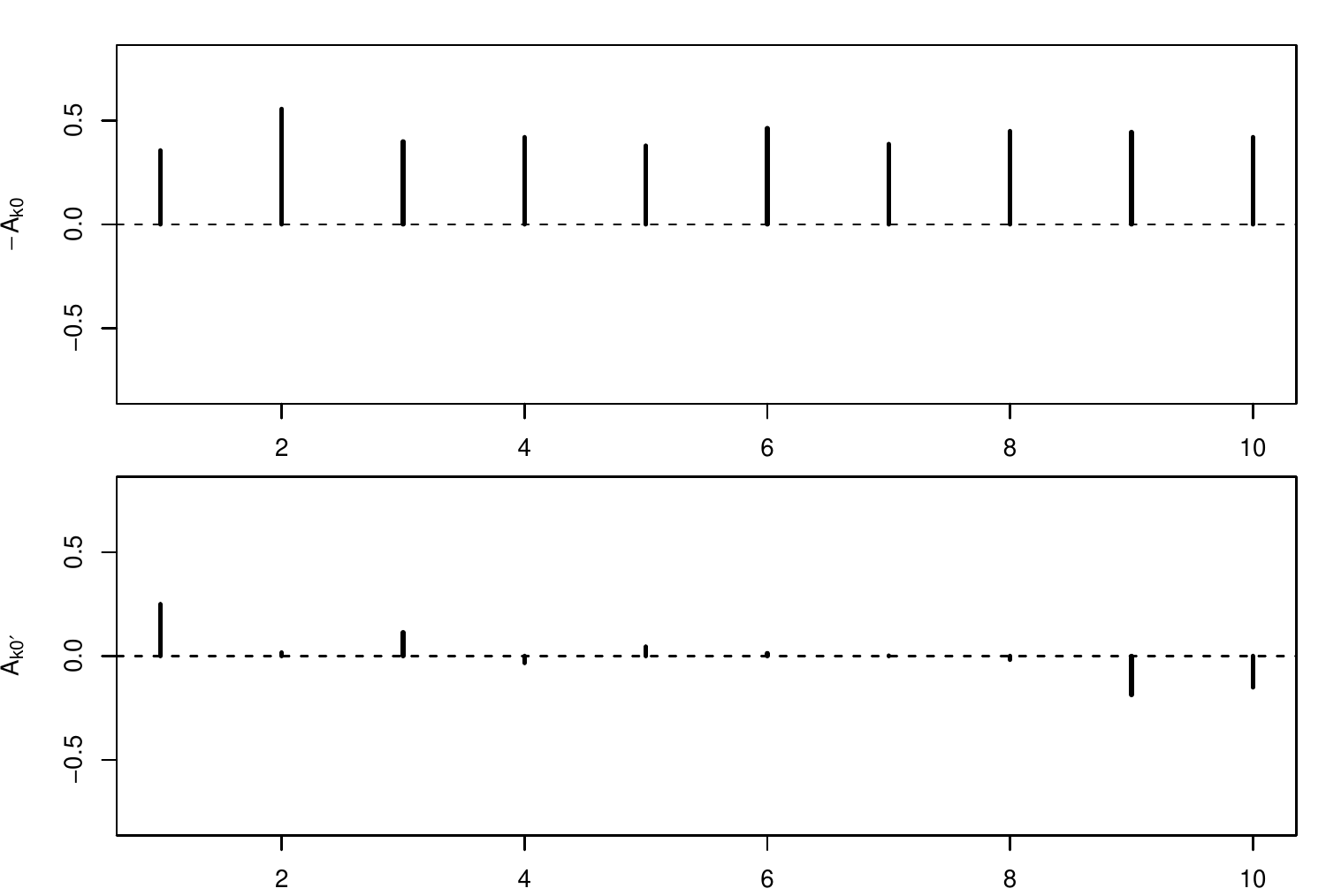}}%\hspace{.5cm}
    \subfigure[$B_{j0}$ and $B_{j0'}$]{\label{fig:0004_2_B}\includegraphics[scale=.5,trim=  0 140 0 0,clip]{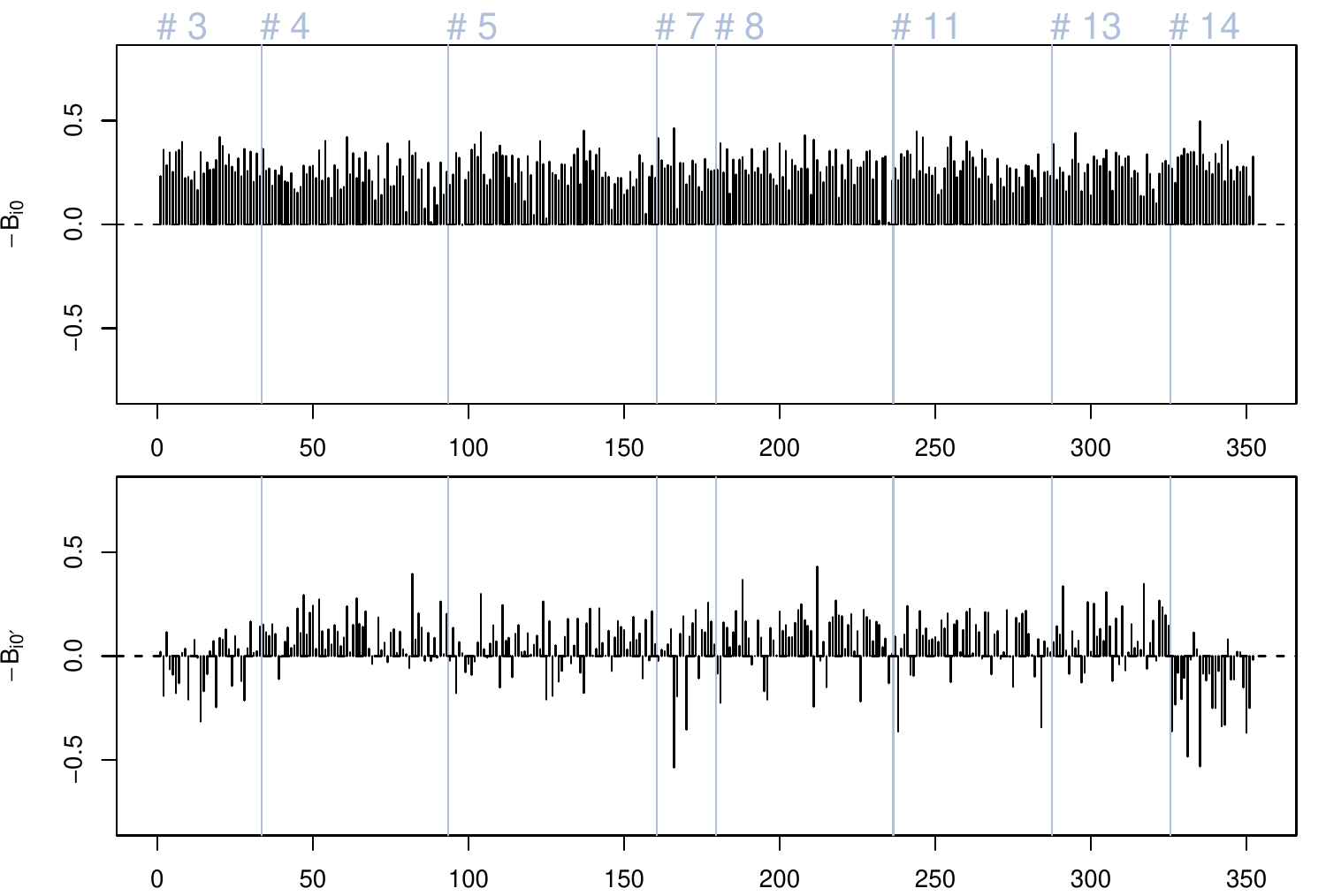}} 
    \vfill
    \subfigure[Index amplitudes reproduced with two volatility drivers $\omega_0$ and $\omega_{0'}$, to be compared with Fig.~\ref{fig:firstfact1mode}]{\label{fig:0004_2_O}\includegraphics[scale=1,trim=220 0 0 0,clip]{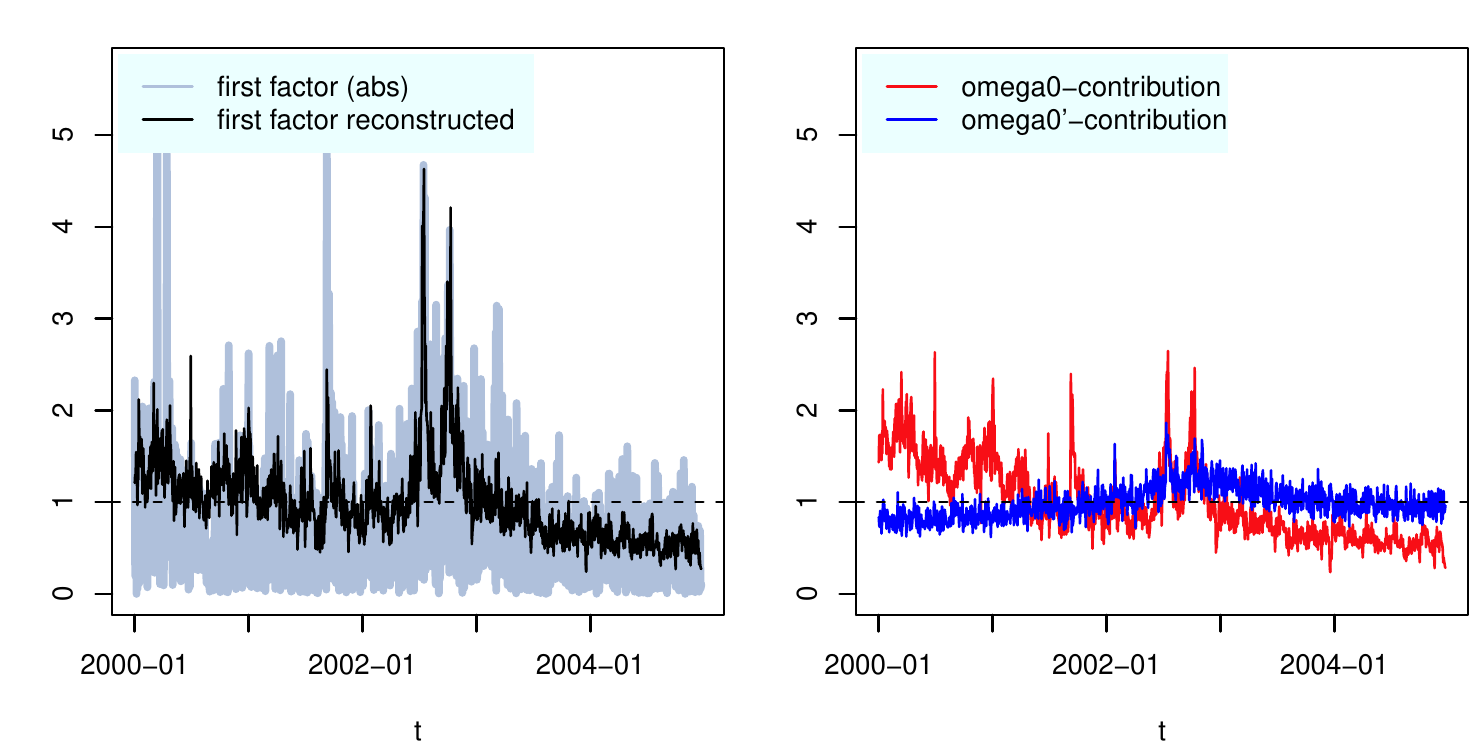}}
    \caption{$M=10$, 2000--2004}\label{fig:2modes0004}
\end{figure}

\section{Out-of-sample analysis}\label{sec:stability}
All the results presented above are ``in-sample'', 
in the sense that we have shown the predicted dependence coefficients with estimated parameters on a period
and compared them to the realized coefficients in that same period.
The ultimate test for a model that aims at describing joint financial returns (and more generally of any risk model), 
is to perform good ``out-of-sample'', i.e.\ use a model calibrated on a period to predict a (average) behavior in a \emph{subsequent} period.
Refer to Ref.~\cite{potters2009financial}!

      We consider a long period 2000--2009 on which we perform an In-sample/Out-of-sample analysis over
      sliding windows ($N=262$ returns series are kept, see Tab.~\ref{tab:sectors}).
      We rely on the procedure used in Ref.~\cite{potters2009financial}.
\begin{enumerate}
\item The model is calibrated in windows of $T^{\scriptscriptstyle \text{IS}}=2N=524$ days.
      An optimal portfolio is built and a corresponding risk measure is computed over the window used for estimation:
      this is the In-sample risk. 
      We consider below two kinds of risks corresponding to two different portfolios:
      (i)  the quadratic risk of a basket of returns, that will assess the quality of the \emph{linear} elements of the model; and
      (ii) the quadratic risk of a basket of (centered and normalized) absolute returns, that will assess the quality of the \emph{volatility} description of the model.
\item The same risk measures are computed Out-of-sample on a small period of $T^{\scriptscriptstyle \text{OS}}=59$ days (three months) following the estimation period.
\item The sliding lags are chosen so that the control samples are non-overlapping,
      i.e.\ at dates $\tau=T^{\scriptscriptstyle \text{IS}}+n\times T^{\scriptscriptstyle \text{OS}}+1$, $n=0,1,2,\ldots$.
      Sliding windows will be indexed with parenthesis notation `$(\tau)$', 
      in order to avoid confusion with regular time stamps $t$ of the running dates.
\end{enumerate}

\subsection{Linear correlations}
For a given covariance matrix $\rho$, optimal portfolio weights can be computed in the sense of Markowitz:
\begin{equation}\label{eq:markowitz_w}
                \vect{w}^*(\tau) =\frac{\rho^{-1}\vect{g}(\tau)}{\vect{g}(\tau)^\dagger\rho^{-1}\vect{g}(\tau)}
\end{equation}
where we consider an omniscient stationary predictor of returns
\begin{equation}\label{eq:markowitz_g}
g_i(\tau) =\frac{\Ret_{\tau i}}{\sqrt{\frac{1}{N}\sum_j \Ret_{\tau j}^2}}
\end{equation}
and a unit total gain $\mathcal{G}\equiv\vect{g}^\dagger\vect{w}^*=1$. 
This means that the in-sample/out-of-sample test procedure applied below is intended to measure only {risk}
and not the risk-return trade-off (Sharp ratio) as is usual e.g.\ when back-testing financial strategies.
Indeed what we ultimately want to conclude is whether our model of stock returns 
allows to have a better view of dependences and thus to better diversify away the risk 
(since we work with normalized returns, we are not concerned with individual variances but only care for dependences).

Quadratic risk, is essentially a measure of expected small fluctuations of the portfolio value:
\begin{equation}\label{eq:def_risk}
    \mathcal{R}^2(\tau)=\frac{1}{T'}\sum_{t'}\frac{1}{N}\sum_{i=1}^N\left[\Ret_{t'i}\frac{w_i^*(\tau)}{\sigma_i^{\scriptscriptstyle \text{IS}}(\tau)}\right]^2
\end{equation}
where, for later convenience, the returns are normalized by a rolling in-sample estimate of their volatility
$\sigma^{\scriptscriptstyle \text{IS}}(\tau)$ 
(although the returns have been normalized over the whole period, they may not be close to unit-variance in-sample
because of low-frequency regime switches in the volatility).
This risk is computed both 
         in-sample (in which case $T'=T^{\scriptscriptstyle \text{IS}}$ and $t'=\tau-T', \ldots, \tau-1 $) 
 and out of sample ($T'=T^{\scriptscriptstyle \text{OS}}$ and $t'=\tau+1 , \ldots, \tau+T'$),
 for different input correlation matrices in Eq.~\eqref{eq:markowitz_w}:
 \begin{itemize}
 \item{Empirical: }     the in-sample raw correlation matrix, $$\displaystyle \rod[1]_\text{Emp}(\tau)=\frac{1}{T^{\scriptscriptstyle \text{IS}}}\sum_{t'=\tau-T^{\scriptscriptstyle \text{IS}}}^{\tau-1}\Ret_{t'\cdot}\cdot\Ret_{t'\cdot};$$
 \item{Clipped: }       the PCA solution of Eq.~\eqref{eq:X_PCA}, $\rod[1]_\text{Clip}(\tau)=\Wei_\text{PCA}(\tau)^\dagger\Wei_\text{PCA}(\tau)$, for several values of the number of retained modes $M$ with largest eigenvalues;
 \item{MultiFactor: }   the improved solution of Eq.~\eqref{eq:offdiag_content}, $\rod[1]_\text{MF}(\tau)=\Wei(\tau)^\dagger\Wei(\tau)$, for several values of the number of factors $M$.
 \end{itemize}

 All these scenarios can furthermore be compared to the benchmark of a full-rank pure noise Wishart correlation matrix.
 In this case, Random Matrix Theory predicts the values of the average in-sample and out-of-sample risks,
 in the limit of large matrices with quality factor $q=N/T^{\scriptscriptstyle \text{IS}}$:
 \[
    \vev{\mathcal{R}^2_\text{RMT}}_\text{IS}=\mathcal{R}^2_\text{True} \cdot (1-q)
    \qquad\text{and}\qquad
    \vev{\mathcal{R}^2_\text{RMT}}_\text{OS}=\mathcal{R}^2_\text{True}  \,/\,(1-q).
 \]
 Moreover, the true risk (i.e.\ the value of $\mathcal{R}^2$ when the optimal weights are determined using the correlation matrix 
 of the process that generates the realized returns $\Ret_{\cdot i}/\sigma_i^{\scriptscriptstyle \text{IS}}$) 
 can be shown to be $\mathcal{R}^2_\text{True}=1$ with the definition \eqref{eq:def_risk}.
 
We show graphically the results of the testing procedure on Fig.~\ref{fig:ISOS_lin}:
In-sample and Out-of-sample average risks of every cleaning scheme are plotted versus the control parameter $\alpha=M/N$,
 where averages are performed over the sliding windows $(\tau)$.

When only a very reduced number of factors ($M\approx 1,2,3$) is kept, eigenvalue clipping performs better (although quite bad), 
and similarly when keeping also the very last modes: 
this is because the linear factors are only good when the eigenmodes are statistically significant, on the left and right of the RMT noise bulk.
In the limit $\alpha\to 1$ (i.e.\ $M=N$) all cleaning schemes collapse to the risk values associated with
the raw sample correlation matrix (black solid level line).
The RMT predictions are shown for reference as dashed grey lines.

In intermediate values instead
(where the Out-of-sample risk is minimal because of marginal gain in signal being higher than marginal risk increase due to added ``false positive information''),
the procedure worked out in Sect.~\ref{sec:lincor_MF} to better estimate the linear correlations 
provides an improved determination of average Out-of-sample risk.
In fact, Fig.~\ref{fig:ISOS_lin_gain} shows that the relative gain
\begin{equation}\label{eq:rel_gain}
    [\vev{\mathcal{R}^2_{\text{Clip}}}-\vev{\mathcal{R}^2_{\text{MF}}}]/[\vev{\mathcal{R}^2_{\text{Clip}}}-\mathcal{R}^2_{\text{True}}]%\approx 30\%
\end{equation}
 can reach up to 25\%, while not dramatically increasing over-fitting (the In-sample risk is only slightly artificially lowered).

\begin{figure}
    \center
    \includegraphics[scale=0.6]{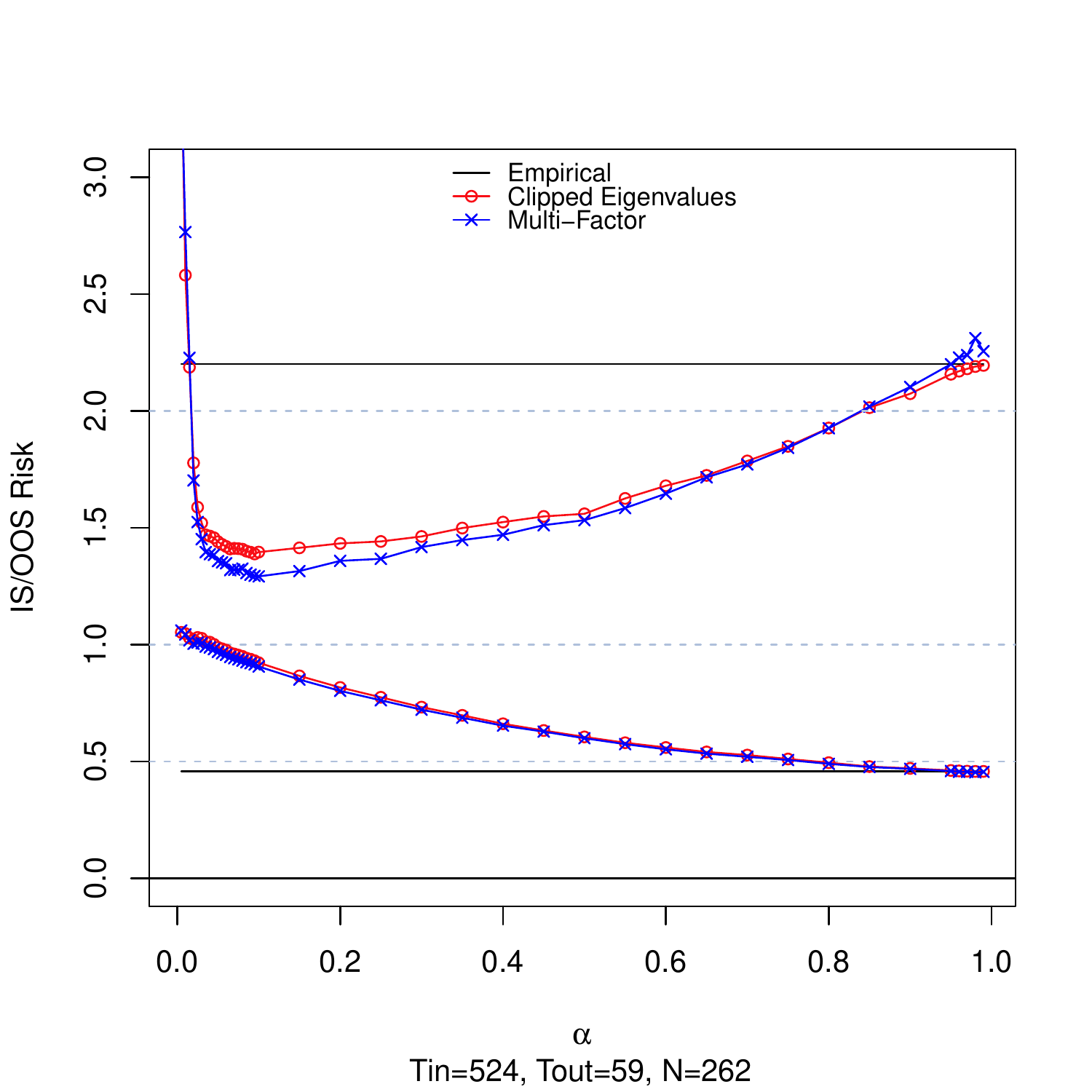}
    \caption{Linear correlations: In-sample risk (lower curves) and out-of-sample risk (upper curves) defined in Eq.~\eqref{eq:def_risk}
    and averaged over sliding windows in 2000--2009, 
    for two cleaning schemes: eigenvalue clipping (red circles) and calibrated multi-factor model (blue crosses), 
    both with $M=\alpha N$ linear factors.}
    \label{fig:ISOS_lin}
    \includegraphics[scale=0.6]{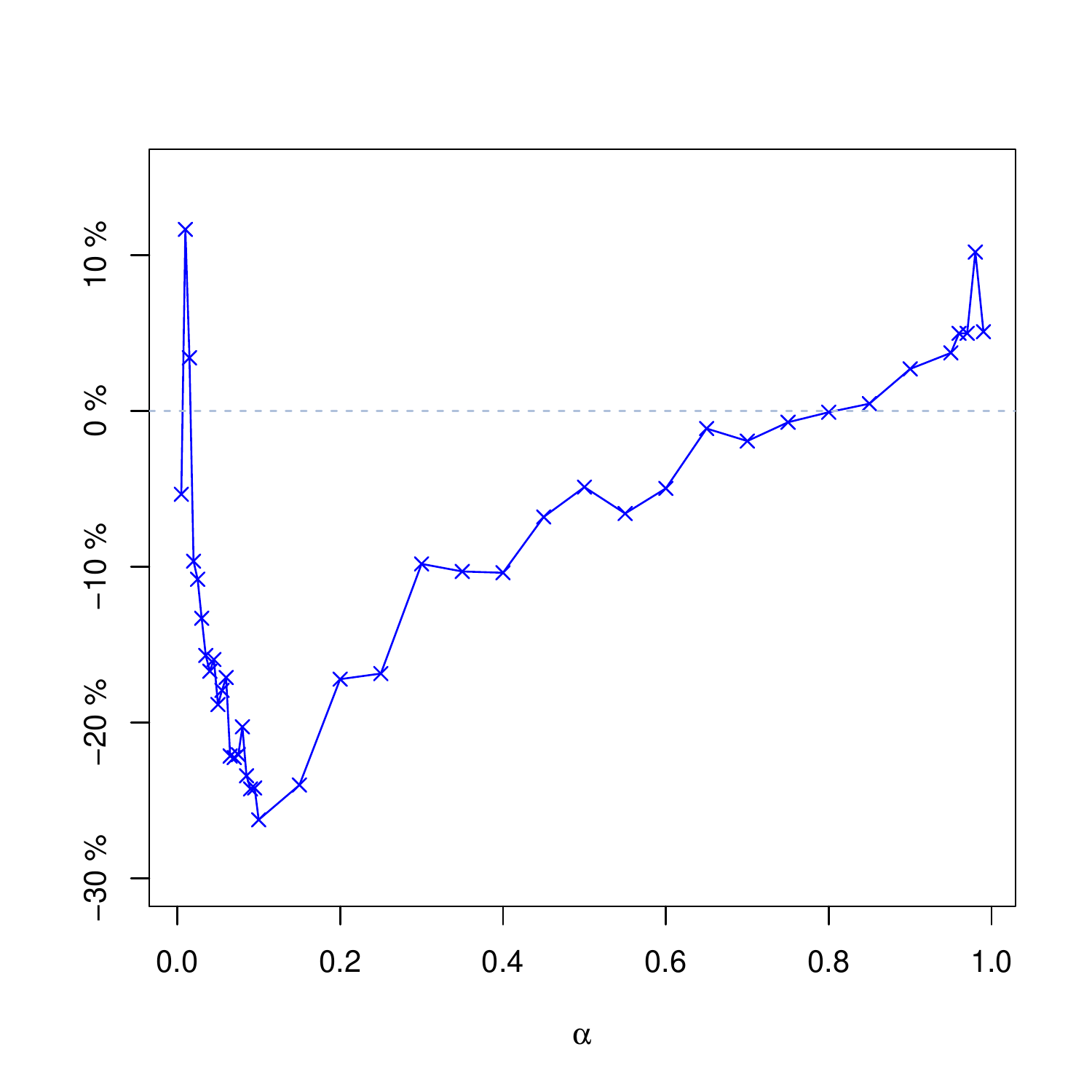}
    \caption{Linear correlations: Out-of-sample risk is lowered by more than $25\%$.
    The location of the minimum coincides with that of minimum risk ($\alpha\approx 0.1$).}
    \label{fig:ISOS_lin_gain}
\end{figure}

%%% ARTICLE
%\textcolor{red}{
%Another measure of risk, more in the spirit of econometrics: residual variance
%\begin{equation}
%    \mathcal{R}^2(\tau)=\frac{1}{T'}\sum_{t'}\left[\Ret_{t'i}-\sum_{j\neq i}\Ret_{t'j}\mat{W}_{ji}\right]^2
%\end{equation}
%ON GARDE CA POUR L'ARTICLE, TANT PIS}

\subsection{Absolute correlations}
We now turn to the core properties of the model: the non-linear dependences.
We have already shown that the model is able to reproduce, after calibration, several empirically observed quantities
like the copula,
and want now to perform an out-of-sample assessment of the volatility-driven dependence in the absolute correlations
$\vev{\mat{Y}_{ti}\mat{Y}_{tj}}_t$, where
\[
    \mat{Y}_{ti} = \frac{|\Ret_{ti}|-\vev{|\Ret_{ti}|}_t}{\sqrt{\left(|\Ret_{ti}|-\vev{|\Ret_{ti}|}_t\right)^2}}
\]
The definitions of the gain predictor $\vect{g}$ and the risk measure $\mathcal{R}^2(\tau)$ are identical
to Eqs.~\eqref{eq:markowitz_g} and \eqref{eq:def_risk} respectively, with $\mat{Y}$ in place of $\Ret$.
The different cleaning schemes considered are:
 \begin{itemize}
 \item{Empirical: }     the in-sample raw correlation matrix, $$\displaystyle \rod[\text{a}]_\text{Emp}(\tau)=\frac{1}{T^{\scriptscriptstyle \text{IS}}}\sum_{t'=\tau-T^{\scriptscriptstyle \text{IS}}}^{\tau-1}\mat{Y}_{t'\cdot}\cdot\mat{Y}_{t'\cdot};$$
 \item{Clipped: }       the PCA solution $\rod[\text{a}]_\text{Clip}(\tau)=\mat{V}_{M|}\Lambda_{M|}\mat{V}_{M|}^\dagger$, keeping the $M$ eigenmodes of $\rod[\text{a}]_\text{Emp}$ with largest eigenvalues;
 \item{Gaussian factors: }   the Gaussian prediction $\rod[\text{a}]_\text{MFG}(\tau)$ obtained as the sample absolute correlations of long time series simulated
                             according to the $M$-factors model \eqref{eq:anothervolatility} where all volatility parameters $A$ and $B$ are set to 0.
 \item{Multifactor (model):} the model prediction $\rod[\text{a}]_\text{MFnG}(\tau)$ obtained as the sample absolute correlations of long time series simulated
                             according to the $M$-factors model \eqref{eq:anothervolatility} calibrated with the volatility parameters (for each of the two drivers) turned on.
 \end{itemize}
Notice that the meaning of $M$ is not comparable in all cleaning schemes:
while for the ``clipped eigenvalues'' it corresponds to the number of relevant modes in the 
matrix of \emph{absolute correlations}, 
for the multi-factor models it instead counts the the number of \emph{linear} factors.
This can be seen immediately on Fig.~\ref{fig:ISOS_abs}, where the red curve corresponding to ``Clipping''
has the usual U-shape discussed above, while the blue curves corresponding to ``multi-factor'' saturate as $\alpha$ increases above $\approx 0.1$,
a threshold above which letting additional linear factors barely affects the volatility dependences.

More importantly, this figure shows that multi-factor models offer a better optimal Out-of-sample risk
together with less In-sample over-fitting.
(Notice that the RMT benchmark is not justified in the case of absolute returns, which are much skewed).
The role of volatility dependences is put forward by the better performance of the final non-Gaussian
multi-factor level over the Gaussian multi-factor cleaning scheme: Fig.~\ref{fig:ISOS_abs_gain} shows
that while the Gaussian model offers an Out-of-sample risk about 20--25\% lower than eigenvalue clipping,
the non-Gaussian model performs up to 50\% better, at the location of the minimal risk ($\alpha\approx 0.1$).
These figures are calculated using the risk difference defined in Eq.~\eqref{eq:rel_gain},
but with a reference $\mathcal{R}^2_\text{ref}=1.5$ 
(arbitrarily chosen between the highest average In-sample risk and the lowest Out-of-sample risk),
since the true risk for absolute returns is not 1.
The absolute numbers are not informative, but we also show in inset of Fig.~\ref{fig:ISOS_abs_gain}
the relative over-performance of the non-Gaussian multi-factor model over the Gaussian multi-factor model:
\[
        [\vev{\mathcal{R}^2_{\text{MFnG}}}-\vev{\mathcal{R}^2_{\text{MFG}}}]/[\vev{\mathcal{R}^2_{\text{MFnG}}}-\mathcal{R}^2_{\text{ref}}].
\]
The improvement is maximal again around $\alpha\approx 0.1$. % (if we discard the noisy ends).

\begin{figure}
    \center
    \includegraphics[scale=0.6]{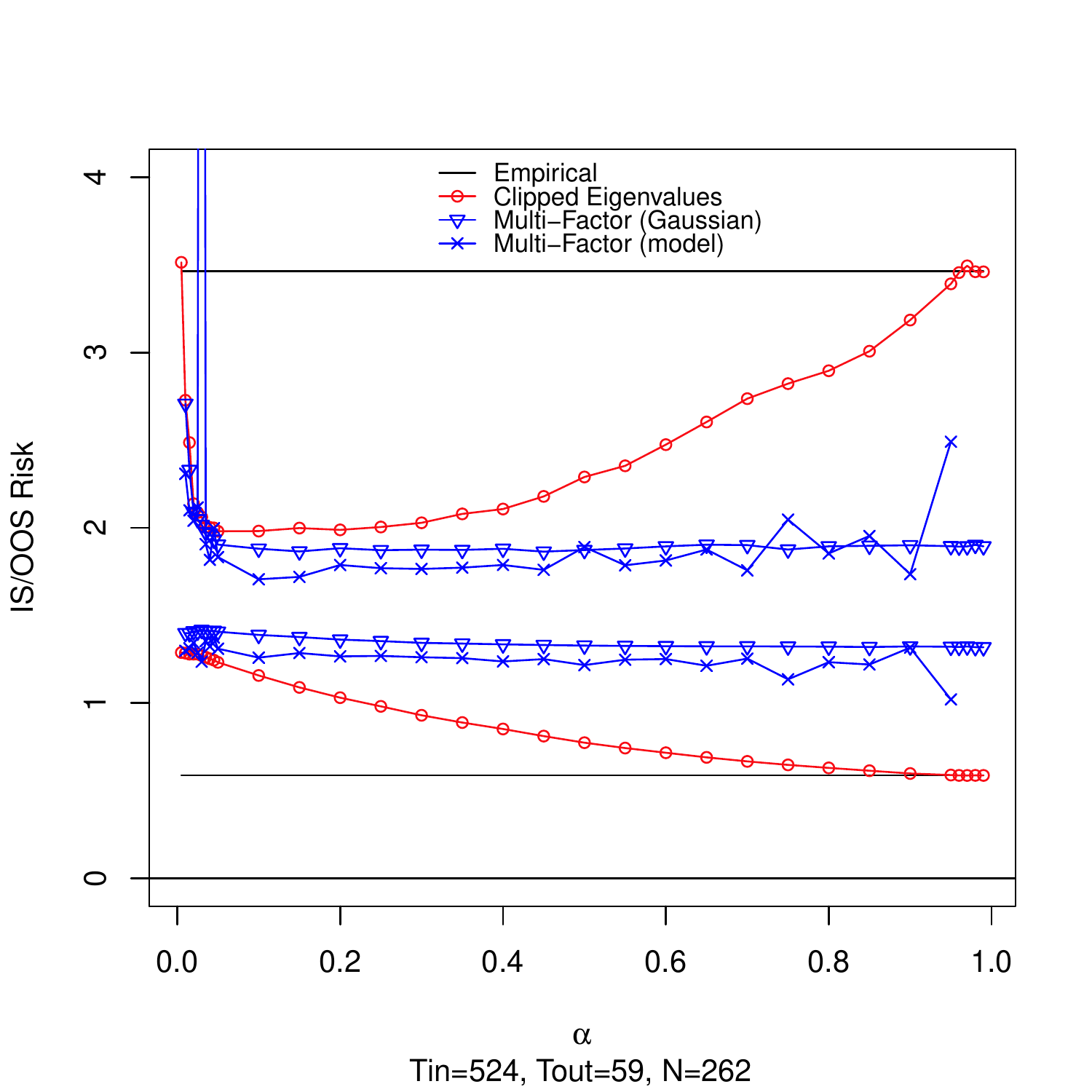}
    \caption{Absolute correlations: In-sample risk (lower curves) and out-of-sample risk (upper curves) defined in Eq.~\eqref{eq:def_risk}
    and averaged over sliding windows in 2000--2009, 
    for three cleaning schemes: eigenvalue clipping of $M=\alpha N$ modes of quadratic correlation (red circles),
    Gaussian multi-factor (blue triangles) and calibrated multi-factor model (blue crosses), 
    both with $M=\alpha N$ linear factors.}
    \label{fig:ISOS_abs}
    \includegraphics[scale=0.6]{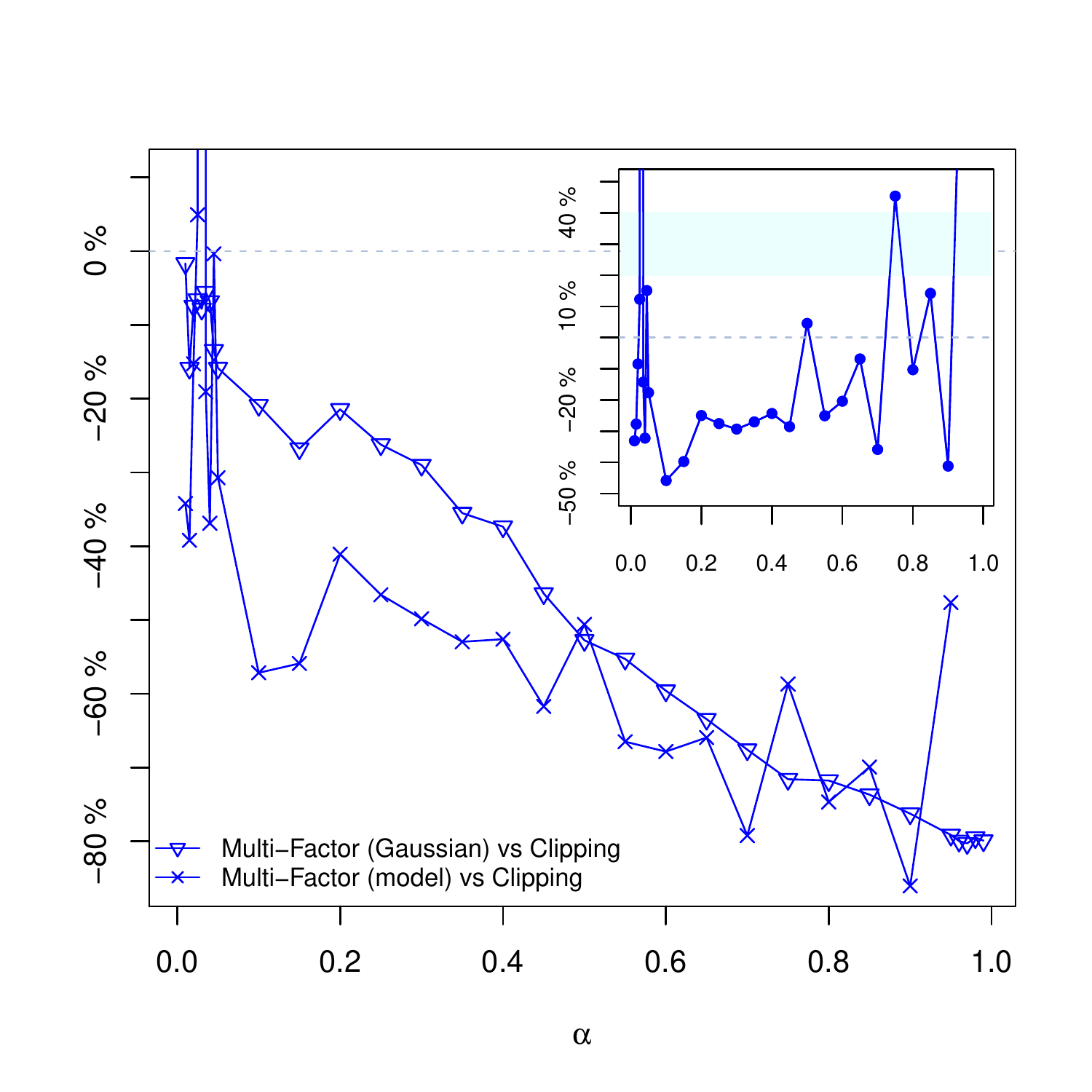}
    \caption{Absolute correlations: Out-of-sample risk is lowered by more than $50\%$ where the risk is minimal ($\alpha\approx 0.1$).
             The role of volatility dependences is put forward by the better performance of the non-Gaussian
             multi-factor level over the Gaussian multi-factor cleaning scheme: the relative risk difference is shown in the inset.}
    \label{fig:ISOS_abs_gain}
\end{figure}

%%% ARTICLE
 % \textcolor{red}{
    % Error estimates ?
    % sources of error: temporal dependence (only noisier or also biased ?)
    % Peter: nombre effectif d'observations tres petit a cause des dependences longue portee dans la vol.
    % Emmanuel: estimation dynamique (Kalman filtering) ?
% }

\subsection{How many factors to keep in the model ?}
The number $M$ of linear factors in the description \eqref{eq:MODEL} is an important input of the model.
The intuition that statistical factors are somewhat related to economic sectors does not stand
the identification of algebraic modes of fluctuations to sectorial or other macroeconomic factors,
beyond the first two or three modes.
Still, even if there is no one-to-one identification, 
the number of sectors can be regarded as a reasonable prior for $M$.
In our calibration, we have retained $M=10$ factors corresponding to the number of Bloomberg sectors plus one,
with satisfactory results at reproducing the main empirical stylized facts.

A more justified determination of $M$ is reached by considering the In-Sample / Out-of-sample risk 
measurement described just above. 
      The In-sample risk is always ``too good'' because it is measured using a model calibrated so as to match
      the particular realizations of the noise in the same period: this phenomenon is known as ``over-fitting''
      and is more pronounced the higher the number of parameters is.
      The Out-of-sample risk is always ``too bad'' because it is measured using a model that can never account 
      completely for the reality, and unexplained noise remains. 
      A too small number of parameters makes the model oversimplified and unrealistic, 
      and leaves room for much uncertainty that manifests itself in a high risk. 
      On the other hand, calibrating too many parameters will cause the model to over-fit the data of the estimation period, 
      and depart significantly from a ``typical'' realization of noise; when applying this over-fitted model out-of-sample,
      the realized risk can be very large as well.
      Therefore, there typically exists an optimal number of parameters,
      for which the model fits reasonably the data and is stable when applied out-of-sample.

Adjusting the trigger parameter $\alpha=M/N$ allowed us to find an optimal configuration
where the out-of-sample risk is minimized while the in-sample risk is not artificially lowered.
A unique value of $\alpha\approx 0.1$ is found to be optimal both for the quadratic risk and for the risk associated with absolute returns,
and thus suggests a number of $M\approx 24$ factors to be kept.
    
\section{Conclusions}
We have presented a factor model for stock returns aiming at reproducing observed non-linear dependences
like the quadratic correlations and the copula diagonals,
the latter being sensitive to the relative ranks of the joint realizations.
Simple generalizations of elliptical model, with idiosyncratic radial parts, were still unable 
to explain important features like the medial copula departing from a standard value.
Such properties can only be accounted for by an interplay of the kurtosis of factors and residual parts.

The main features of our model are the following:
(i) A linear factor structure, that is able to generate very accurately the linear correlations, 
up to pair effects that are not to be explained by any factor model;
(ii) independent residuals which allow for an interpretation in terms of indiosyncracies;
(iii) non-Gaussianity of the factors with specific kurtosis;
(iv) non-Gaussianity of the residuals with specific kurtosis;
(v) volatility correlations among the factors generated by common modes of log-volatilities, in particular a ``volatility market mode'';
(vi) volatility correlations between factors and residuals.

Precise predictions are formulated, and a calibration of all parameters is performed over different periods.
An original calibration scheme is provided, with powerful results both in the linear sector and in the volatility parameters.
The calibrated models exhibit dependences that show a very good overall fit to empirical observations, 
even for the non-trivial observables like the copula medial point.
The stability analysis performed in the last section illustrates the ability of the model not only to fit the data in-sample
but also to reproduce the typical behavior of out-of-sample realizations of the noise.
It points toward an optimal number of factors $M\approx 24$.

To conclude this part about the cross-sectional dependences and provide a link to the 
next part dedicated to temporal dependences, notice that the model studied in this chapter
has no temporal content.
In particular, the description of the log-volatilities does not exhibit the time-dependence
responsible for the so called ``volatiilty clustering'' effect that will be discussed in several places in the next part.
This however was not crucial here for the characterization of dependences among stocks.
In terms of inference and parameters estimation, it generates larger uncertainties, but no bias is expected as long as
the time series still have a self-averaging property.
    
% The non-Gaussianity of the logvols is not the major point of the model.
% Although it is clearly there, a precise estimation of the moments is very difficult,
% and anyways not much valuable. 
% What matters more is the different kurtosis of the linear factors, blabla

\part{Temporal dependences}\label{part:partIII}
\chapter{Volatility dynamics}\label{chap:QARCH}
\minitoc

\nocite{cont2010encyclopedia}
\nocite{berd2011lessons}
\nocite{engle1986handbook}

%%%%%%%%%%%%%%%%%%%%%%%%%%%%%%%%%%%%%%%%%%%%%%%%%%%%%%%%%%%%%%%%%%%%%%%%
\section{Introduction}

One of the most striking universal stylized facts of financial returns is the volatility clustering effect,
which was first reported by Mandelbrot as early as 1963 \cite{mandelbrot1963variation}. He noted that \ldots
{\it large changes tend to be followed by large changes, of either sign, and small changes tend to be followed by small changes.}
The first quantitative description of this effect was the ARCH\nomenclature{ARCH}{Auto-Regressive Conditional Heteroskedastic} 
model proposed by Engle in 1982 \cite{engle1982autoregressive}. 
It formalizes Mandelbrot's hunch in the simplest possible way, 
by postulating that returns $r_t$ are conditionally Gaussian random variables, 
with a time dependent volatility (RMS)\nomenclature{RMS}{Root mean square} $\sigma_t$ that evolves according to:
\be
\sigma^2_t = s^2 + g r_{t-1}^2.
\ee
In words, this equation means that the (squared) volatility today is equal to a baseline level $s^2$, 
plus a self-exciting term that describes the influence of yesterday's {\it perceived} volatility $r_{t-1}^2$ on today's activity, 
through a feedback parameter $g$. 
Note that this ARCH model was primarily thought of as an econometric model that needs to be calibrated on data, 
while a more ambitious goal would be to {\it derive} such a model from a more fundamental theory 
--- for example, based on behavioural reactions to perceived risk. 

It soon became apparent that the above model is far too simple to account for empirical data. 
For one thing, it is unable to account for the long memory nature of volatility fluctuations. 
It is also arbitrary in at least two ways:
\begin{itemize} 

\item First, there is no reason to limit the feedback effect to the previous day only. 
      The Generalized ARCH model (GARCH) \cite{bollerslev1986generalized}, which has become a classic in quantitative finance, 
      replaces $r_{t-1}^2$ by an exponential moving average of past squared returns. 
      Obviously, one can further replace the exponential moving average by any weighting kernel $k(\tau)\geq 0$, 
      leading to a large family of models such as:
      \be
	\sigma^2_t = s^2 + \sum_{\tau=1}^\infty k(\tau) r_{t-\tau}^2,
      \ee
      which includes all ARCH and GARCH models. 
      For example, ARCH($q$) corresponds to a kernel $k(\tau)$ that is strictly zero beyond $\tau=q$. 
      A slowly (power-law) decaying kernel $k(\tau)$ is indeed able to account for the long memory of volatility 
      --- this corresponds to the so-called FIGARCH model (for Fractionally Integrated GARCH) \cite{bollerslev1994arch}. 

\item Second, there is no {\it a priori} reason to single out the day as the only time scale to define the returns. 
      In principle, returns over different time scales could also feedback on the volatility today 
      \cite{muller1997volatilities,borland2005multi,lynch2003market}, 
      leading to another natural extension of the GARCH model as:
      \be
	\sigma^2_t = s^2 + \sum_{\ell} \sum_{\tau=1}^\infty g_\ell(\tau) R^{(\ell)2}_{t-\tau},
      \ee
      where $R^{(\ell)}_t$ is the cumulative, $\ell$ day return between $t-\ell$ and $t$. 
      The first model in that category is the HARCH model of the Olsen group \cite{muller1997volatilities}, 
      where the first ``H'' stands for Heterogeneous. 
      The authors had in mind that different traders are sensitive to and react to returns on different time scales.
      Although this behavioural interpretation was clearly expressed, 
      there has been no real attempt%
      \footnote{See however the very recent stochastic volatility model with heterogeneous time scales of \cite{delpini2012stochastic}.} 
      to formalize such an intuition beyond the hand-waving arguments given in \cite{borland2005multi}.
\end{itemize}

The common point to the zoo of generalizations of the initial ARCH model is that 
the current level of volatility $\sigma_t^2$ is expressed as a quadratic form of past realized returns. 
The most general model of this kind, called QARCH (for Quadratic ARCH), is due to Sentana \cite{sentana1995quadratic}, and reads:
\be\label{QARCHeq:quadraticARCH}
\sigma^2_t = s^2 + \sum_{\tau=1}^\infty L(\tau) \, r_{t-\tau} + \sum_{\tau,\tau'=1}^\infty K(\tau,\tau') \, r_{t-\tau}r_{t-\tau'},
\ee
where $L(\tau)$ and $K(\tau,\tau')=K(\tau',\tau)$ are some kernels that should satisfy technical conditions for $\sigma^2_t$ to be always positive 
(see below and \cite{sentana1995quadratic}).
The QARCH can be seen as a general discrete-time model for the dependence of $\sigma_t^2$ on all past returns $\left\{r_{t'}\right\}_{t'<t}$, truncated to second order. 
The linear contribution, which involves $L(\tau)$, captures a possible dependence of the volatility on the sign of the past returns. 
For example, negative past returns tend to induce larger volatility in the future 
--- this is the well-known leverage effect \cite{black1976studies,bouchaud2001leverage,bekaert2000asymmetric}, see also \cite{reigneron2011principal} and references therein.%
\footnote{(G)QARCH and alternative names such as Asymmetric (G)ARCH, Nonlinear (G)ARCH, Augmented ARCH, etc.\ 
often refer to this additional leverage (asymmetry) contribution,
whereas the important innovation of QARCH is in fact the possibility of off-diagonal terms in the kernel $K$.} 
The quadratic contribution, on the other hand, contains through the matrix $K(\tau,\tau')$ all ARCH models studied in the literature. 
For example, ARCH($q$), GARCH and FIGARCH models all correspond to a purely {\it diagonal} kernel, $K(\tau,\tau')=k(\tau) \delta_{\tau,\tau'}$
where $\delta_{\tau,\tau'}$ is Kronecker's delta. 

In view of the importance of ARCH modelling in finance, it is somewhat surprising that the general framework provided by QARCH has not been fully explored. 
Only versions with very short memories, corresponding to at most $2 \times 2$ matrices for $K$, seem to have been considered in the literature. 
In fact, Sentana's contribution is usually considered to be the introduction of the linear contribution in the GARCH framework, 
rather than unveiling the versatility of the quadratic structure of the model. 
The aim of this chapter is to explore in detail the QARCH framework, both from a theoretical and empirical point of view. 
Of particular interest is the empirical determination of the structure of the feedback kernel $K(\tau,\tau')$ for the daily returns of stocks, 
which we compare with several proposals in the literature, 
including the multi-scale model of \cite{borland2005multi} and the trend-induced volatility model of \cite{zumbach2010volatility}. 
Quite surprisingly, we find that while the off-diagonal elements of $K(\tau,\tau')$ are significant, 
they are at least an order of magnitude smaller than the diagonal elements $k(\tau) := K(\tau,\tau)$.  
The latter are found to decay very slowly with $\tau$, in agreement with previous discussions. 
Therefore, in a first approximation, the dominant feedback effect comes from the amplitude of {\it daily returns} only, 
with minor corrections coming from returns computed on large time spans, at variance with the assumption of the model 
put forward in \cite{borland2005multi}. We believe that this finding is unexpected and far from trivial. It is a strong constraint on any attempt 
to justify the ARCH feedback mechanism from a more fundamental point of view.

In parallel with ARCH modelling,  stochastic volatility models represent another strand of the literature that has vigorously grown in the last twenty years. 
Here again, a whole slew of models has emerged \cite{henry2008analysis}, with the Heston model \cite{heston1993closed} and the SABR model \cite{hagan2002managing} as the best known examples. 
These models assume that the volatility itself is a random process, governed either by a stochastic differential equation (in time) 
or an explicit cascade construction in the case of more recent multifractal models \cite{bacrymuzy,calvetfisher,lux}
(again initiated by Mandelbrot as early as 1974! \cite{mandelbrot1974intermittent}). 
There is however a fundamental difference between most of these stochastic volatility models and the ARCH framework: 
while the former are {\it time-reversal invariant} (TRI), the latter is explicitly {\it backward looking}. 
This, as we shall discuss below, implies that certain correlation functions are not TRI within QARCH models, 
but are TRI within stochastic volatility models. 
This leads to an empirically testable prediction; we report below that TRI is indeed violated in stock markets, as also documented in \cite{zumbach2009time}. 

The outline of this chapter is as follows. 
We first review in Section~\ref{QARCHsec:section2} some general analytical properties of QARCH models, 
in particular about the existence of low moments of the volatility. 
We then introduce in Section~\ref{QARCHsec:section3} several different sub-families of QARCH,
 that we try to motivate intuitively. 
The consideration of these sub-families follows from the necessity of reducing the dimensionality of the problem, 
but also from the hope of finding simple regularities that would suggest a plausible interpretation (behavioural or else) of the model, 
beyond merely best fit criteria. 
In Section~\ref{QARCHsect:emp_stocks}, we attempt to calibrate ``large'' QARCH models on individual stock returns, 
first without trying to impose any a priori structure on the kernel $K(\tau,\tau')$, 
and then specializing to the various sub-families mentioned above. 
The same analysis is done in Section~\ref{QARCHsect:emp_index} for the returns of the stock index.
We isolate in Section~\ref{QARCHsect:TRI_0} the discussion on the issue of TRI for stock returns, 
both from a theoretical/modeling and an empirical point of view.
We give our conclusions in Section~\ref{QARCHsect:Concl}, and relegate to Appendices~\ref{QARCHapx:A} and~\ref{QARCHapx:B} more technical issues. 

\section{General properties of QARCH models}\label{QARCHsec:section2}

Some general properties of QARCH models are discussed in Sentana's seminal paper \cite{sentana1995quadratic}. 
We review them here and derive some new results. 
The QARCH model for the return at time $t$, $r_t$, is such that:
\be
\ln p_{t} - \ln p_{t-1} = r_t = \sigma_t \xi_t,
\ee
where $p_t$ is the price at time $t$, $\sigma_t$ is given by the QARCH specification, Eq.~\eqref{QARCHeq:quadraticARCH} above, 
while the $\xi$'s are i.i.d.\ \nomenclature{i.i.d.}{Independent and identically distributed} random variables, of zero mean and variance equal to unity. 
While many papers take these $\xi$'s to be Gaussian, it is preferable to be agnostic about their univariate distribution. 
In fact, several studies including our own (see below), suggest that the $\xi$'s themselves have fat-tails: 
asset returns are {\it not} conditionally Gaussian and ``true jumps'' do occur.%
\footnote{There seems to be a slowly growing consensus on this point (see e.g.\ \cite{ait2009analyzing}): 
Gaussian processes with stochastic volatility cannot alone account for 
the discontinuities observed in market prices.}

In this section, we will focus on the following non-linear correlation functions 
(other correlations will be considered below, when we turn to empirical studies):
\begin{subequations}
\begin{align}
	           \mathcal{C}^{(2)}(\tau)   &\equiv\vev{\left(r^2_t-\vev{r_{t'}^2}_{t'}\right)             r^2_{t-\tau}}_t\\
	\widetilde{\mathcal{C}}^{(2)}(\tau)  &\equiv\vev{\left(\sigma^2_t-\vev{\sigma_{t'}^2}_{t'}\right)   r^2_{t-\tau}}_t\\
	           \mathcal{D} (\tau',\tau'')&\equiv\vev{\left((r^2_{t}-\vev{r_{t'}^2}_{t'}\right)          r_{t-\tau'}r_{t-\tau''}}_t\\
	\widetilde{\mathcal{D}}(\tau',\tau'')&\equiv\vev{\left((\sigma^2_{t}-\vev{\sigma_{t'}^2}_{t'}\right)r_{t-\tau'}r_{t-\tau''}}_t.
\end{align}
\end{subequations}
Here and below, we assume stationarity and correspondingly $\vev{\dots}_t$ refers to a sliding average over $t$. 
The following properties are worth noticing: by definition, 
$\mathcal{D}(\tau,\tau) \equiv \mathcal{C}^{(2)}(\tau)$ and 
$\widetilde{\mathcal{D}}(\tau,\tau) \equiv \widetilde{\mathcal{C}}^{(2)}(\tau)$.
Furthermore, whereas $\mathcal{C}^{(2)}(\tau) = \mathcal{C}^{(2)}(-\tau)$ by construction, 
the same is not true in general for $\widetilde{\mathcal{C}}^{(2)}(\tau)$. 
However, using the QARCH causal construction and the independence of the $\xi$'s, 
one can easily convince oneself that when $\tau > 0$, $\widetilde{\mathcal{C}}^{(2)}(\tau) \equiv \mathcal{C}^{(2)}(\tau)$. 
Similarly, for \mbox{$\tau'>\tau''>0$}, $\widetilde{\mathcal{D}}(\tau',\tau'') \equiv \mathcal{D}(\tau',\tau'')$, 
while in general, $\mathcal{D}(\tau',\tau'') \neq \mathcal{D}(-\tau',-\tau'') \equiv 0$. 

\subsection{Second moment of the volatility and stationarity}\label{QARCHssec:second_moment}

QARCH models only make sense if the expected volatility does not diverge to infinity. 
The criterion for stability is easy to establish if the $\xi$'s are IID and of zero mean, and reads:
\be
\tr K < 1.
\ee
In this case, the volatility is a stationary process such that $\vev{\sigma^2}\equiv\mathds{E}[\sigma^2]=s^2/(1 - \tr K)$: 
the feedback-induced increase of the volatility only involves the diagonal elements of $K$. 
Note also that the leverage kernel $L(\tau)$ does not appear in this equation. 
As an interesting example, we consider kernels with a power-law decaying diagonal: 
$K(\tau,\tau)=g\,\tau^{-\alpha}\1{\tau\leq q}$. 
For a given decay parameter $\alpha$, the amplitude $g$ must be smaller than a certain value $g_c(\alpha,q)$ for $\vev{\sigma^2}$ to be finite. 
Fig.~\ref{QARCHfig:sig2crit} shows the critical frontier $g_c(\alpha,q)$ for $q=1,32,256$ and $q \to \infty$. 
The critical frontier in the limit case $q=\infty$ is given by $g_c=1/\zeta(\alpha)$, 
where $\zeta(\alpha)$ is Riemann's zeta function (solid red). 
Note in particular that the model is always unstable when \mbox{$\alpha < 1$}, 
i.e.\ when the memory of past realized volatility decays too slowly.%
\footnote{In the context of fractionally integrated processes $I(d)$ , the condition $\alpha\leq 1$
is equivalent to the `difference parameter' $d=\alpha-1$ being positive.} 
At the other extreme, $q=1$, the constraint is well known to be $g=k(1)\leq 1$ (solid red).

Within a strict interpretation of the QARCH model, there are additional constraints on the kernels $K$ and $L$ 
that arise from the fact that $\sigma_t^2$ should be positive for any realization of price returns. 
This imposes that a) all the eigenvalues of $K$ should be non-negative, and b) that the following inequality holds:
\be\label{QARCHeq:nonneg}
\sum_{\tau,\tau'=1}^q L(\tau) K^{-1}(\tau,\tau') L(\tau') \leq 4\,s^2,
\ee
where $K^{-1}$ is the matrix inverse of $K$. 
However, these constraints might be too strong if one interprets the QARCH model as a generic expansion of $\sigma^2_t$ in powers of past returns,
truncated to second order \cite{sentana1995quadratic}. 
It could well be that higher order terms are stabilizing and lead to a meaningful, stable model beyond the limits quoted here.

\begin{figure}[p]
	\center
	\includegraphics[scale=0.7,trim=  0 0 864 0,clip]{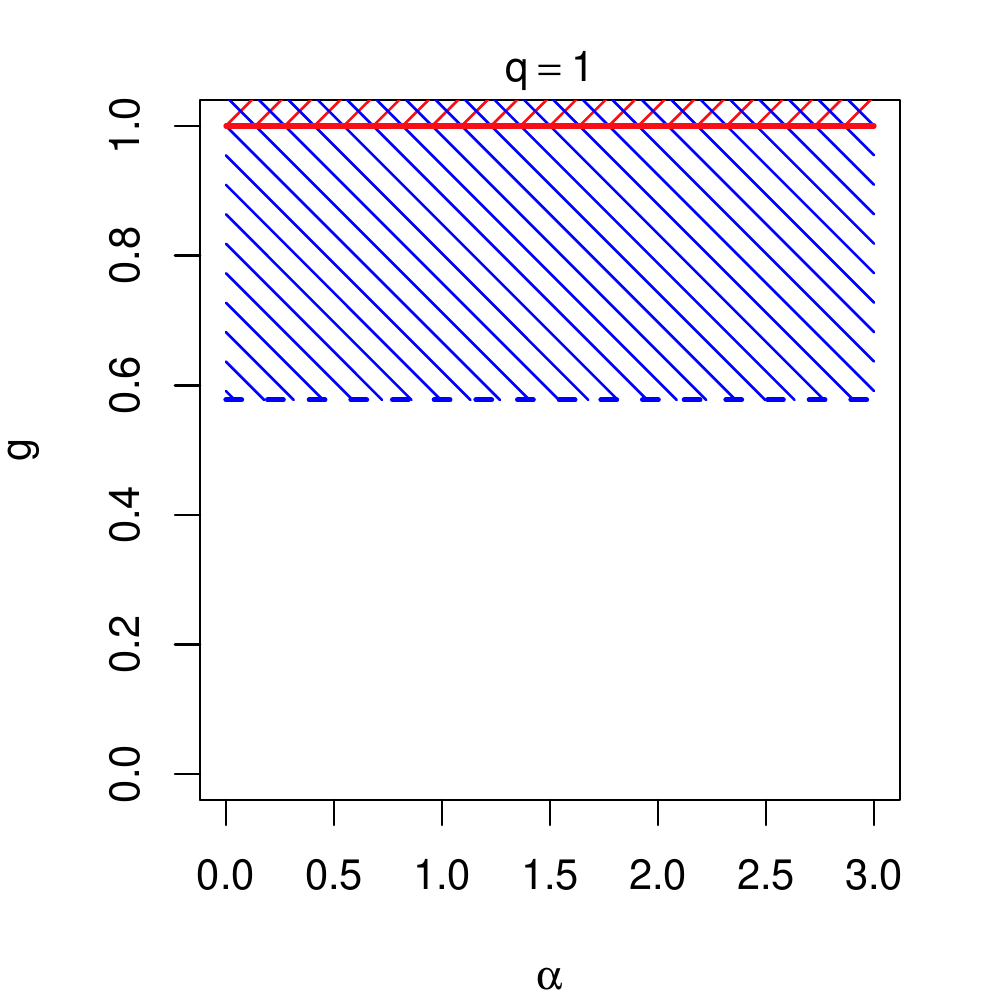}
	\includegraphics[scale=0.7,trim=287 0 576 0,clip]{QARCH/critfront24_allq}
	\includegraphics[scale=0.7,trim=575 0 288 0,clip]{QARCH/critfront24_allq}
	\includegraphics[scale=0.7,trim=863 0   0 0,clip]{QARCH/critfront24_allq}
	\caption{Allowed region in the parameter space for $K(\tau,\tau)=g\,\tau^{-\alpha}\1{\tau\leq q}$ and $L(\tau)=0$,
	         according to the finiteness of $\vev{\sigma^2}$ and $\vev{\sigma^4}$. 
             Divergence of $\vev{\sigma^2}$ is depicted by  $45^\circ$ (red)  hatching, while 
             divergence of $\vev{\sigma^4}$ is depicted by $-45^\circ$ (blue) hatching.
             In the wedge between the dashed blue and solid red lines, 
             \mbox{$\vev{\sigma^2} < \infty$} while $\vev{\sigma^4}$ diverges.}
	\label{QARCHfig:sig2crit}
\end{figure}

\subsection{Fourth moment of the volatility}\label{QARCHsec:fourth_moment}

The existence of higher moments of $\sigma$ can also be analyzed, leading to more and more cumbersome algebra 
\cite{he1999fourth,ling2002stationarity,teyssiere2010long}. 
In view of its importance, we have studied in detail the conditions for the existence of the fourth moment of $\sigma$, 
which allows one to characterize the excess kurtosis $\kappa$ of the returns, traditionally defined as:
\be
\kappa = \frac{\langle r^4 \rangle}{\langle r^2 \rangle^2} - 3 \equiv \frac{\langle \sigma^4 \rangle \langle \xi^4 \rangle}{\langle \sigma^2 \rangle^2} - 3.
\ee
In the general case, $\langle \sigma^4 \rangle$, $\mathcal{C}^{(2)}(\tau)$ and $\mathcal{D}(\tau',\tau'')$ 
are related by the following set of self-consistent equations:
\begin{subequations}
\begin{align}
	\vev{\sigma^4}-\vev{\sigma^2}^2 =&
	\vev{\sigma^2}^2\Big(\tr(K^2)-\tr(K^{\bullet 2})\Big)+
	   \sum_{\tau>0}K(\tau,\tau){\mathcal{C}}^{(2)}(\tau)\\\nonumber%
	+&2\!\sum_{0<\tau_2<\tau_1}\!K(\tau_1,\tau_2)\left[{\mathcal{D}}(\tau_1,\tau_2)-\!\sum_{0<\tau<\tau_2}\!K(\tau,\tau)\mathcal{D}(\tau_1\!-\!\tau,\tau_2\!-\!\tau)\right]\\
	{\mathcal{C}}^{(2)}(\tau>0)=&K(\tau,\tau)\Big(\vev{\sigma^4}\vev{\xi^4}-\vev{\sigma^2}^2\Big)+
	\sum_{\tau'\neq\tau}K(\tau',\tau')\mathcal{C}^{(2)}(\tau\!-\!\tau')\\\nonumber
	+&2\!\sum_{\tau'>\tau''=\tau+1}^q\!K(\tau',\tau'')\mathcal{D}(\tau'\!-\!\tau,\tau''\!-\!\tau)\\
	{\mathcal{D}}(\tau_1 >0,\tau_2>0)=&2K(\tau_1,\tau_2)\Big(\mathcal{C}^{(2)}(\tau_1\!-\!\tau_2)+\vev{\sigma^2}^2\Big)\\\nonumber
    +& \phantom{2}\,\sum_{0<\tau'<\tau_2}\,K(\tau',\tau')\mathcal{D}(\tau_1\!-\!\tau',\tau_2\!-\!\tau')\\\nonumber
	+&2\!\sum_{\tau'>\tau_2,\tau'\neq\tau_1}\!K(\tau',\tau_2)\mathcal{D}(\tau'\!-\!\tau_2,|\tau_1\!-\!\tau_2|).
\end{align}
\end{subequations}
where we assume for simplicity here that the leverage effect is absent, i.e.\ $L(\tau) \equiv 0$, 
and $K^{\bullet 2}$ means the square of $K$ in the Hadamard sense (i.e.\ element by element). 
For a QARCH with maximum horizon $q$, we have thus a set of $1+q+q(q-1)/2$ linear equations for $k(\tau\geq 0)$
that can be numerically solved for an arbitrary choice of the kernel $K$. 
These equations simplify somewhat in the case of a purely diagonal kernel $K(\tau,\tau')=k(\tau) \delta_{\tau,\tau'}$. 
One finds:
\begin{subequations}\label{QARCHeq:vol4}
\begin{align}
	\vev{\sigma^4}-\vev{\sigma^2}^2=&\sum_{\tau'>0}k({\tau'})\mathcal{C}^{(2)}(\tau')\\
	{\mathcal{C}}^{(2)}(\tau)=&
    %k({\tau})\Big(\vev{\sigma^4}\vev{\xi^4}-\vev{\sigma^2}^2\Big)+\sum_{\tau'\neq\tau>0}k({\tau'})\mathcal{C}^{(2)}(\tau-\tau')
    \sum_{\tau'>0}k({\tau'})\mathcal{C}^{(2)}(\tau'-\tau)
\end{align}
\end{subequations}
By substituting $\vev{\sigma^4}$, it is easy to explicit the linear system in matrix form $\nabla\,\mathcal{C}^{(2)}=S$ with
\begin{subequations}
\begin{align}\label{eq:nabla_elements}
	\nabla(\tau,\tau')&=\delta_{\tau{\tau'}}-\vev{\xi^4}k({\tau})k({\tau'})-\left[k({\tau-\tau'})+k({\tau+\tau'})\right]\\
	S(\tau)&= k({\tau})\vev{\sigma^2}^2\left(\vev{\xi^4}-1\right)
\end{align}
\end{subequations}
and the convention that $k({\tau})=0, \forall\tau\leq 0$.

Let us examine this in more detail for ARCH($q$). 
For simplicity, we assume here that $\xi$ is Gaussian ($\vev{\xi^4}=3$) and $s$ is chosen such that $\vev{\sigma^2}=1$. 
The condition on $k(\tau)$ for which $\vev{\sigma^4}$ diverges is given by $\det \nabla =0$, 
where $\nabla$ is the matrix whose entries are defined in Eq.~\ref{eq:nabla_elements}. 
For different $q$'s, this reads:
\begin{itemize}
\item{for $q=1$}, one recovers the well known result that ARCH(1) has a finite fourth moment only when \mbox{$k_1<1/\sqrt{3}$}.
\item{for $q=2$}, the stability line is given by $k_1 + k_2 = 1$, while the existence of a finite fourth moment is given by the condition \mbox{$k_1^2 < (1/3-k_2^2)(1-k_2)/(1+k_2)$}.
\item{for $q \to \infty$}, we again assume the $\tau$ dependence of $k(\tau)$ to be a power-law, $g\,\tau^{-\alpha}$ (corresponding to the FIGARCH model). 
                  The critical line for which the fourth moment diverges is shown in dashed blue in Fig.~\ref{QARCHfig:sig2crit}. 
                  After a careful extrapolation to $q = \infty$, we find that whenever \mbox{$1 < \alpha < \alpha_c \approx 1.376$}, 
		  the fourth moment exists as soon as the model is stationary, i.e.\ when \mbox{$g<1/\zeta(\alpha)<1/\zeta(\alpha_c)\approx 0.306$}. 
\end{itemize}
The last result is quite interesting and can be understood from Eq.~\eqref{QARCHeq:vol4}, which shows that to lowest order in $g$, one has:
\be\label{QARCHeq:sigma4pert}
\frac{\langle \sigma^4 \rangle}{\langle \sigma^2 \rangle^2} - 1 \approx (\vev{\xi^4}-1) g^2 \sum_{\tau > 0} \frac{1}{\tau^{2\alpha}}.
\ee
The above expression only diverges if \mbox{$\alpha < 1/2$}, but this is far in the forbidden region \mbox{$\alpha < 1$} where $\langle \sigma^2 \rangle$ 
itself diverges. Therefore, perhaps unexpectedly, a FIGARCH model with a long memory (i.e.\ \mbox{$\alpha < 1.376$}) cannot lead to a 
large kurtosis of the returns, {\it unless the $\xi$ variables have themselves fat tails.} We will come back to this important point below.
By the way, FIGARCH models with long memory and $1<\alpha<3/2$ are able to generate power-law correlations of the volatility $\mathcal{C}^{(2)}(\tau)\propto \tau^{-\beta}$,
with a slow decay characterized by an exponent $\beta=3-2\alpha$, as we demonstrate in Appendix~\ref{QARCHapx:B}.

\paragraph{}
As we alluded to in the introduction, ARCH($q$) models posit that today's volatility is only sensitive to past daily returns. This assumption can be relaxed in several natural ways,
each of which leading to a specific structure of the feedback kernel $K$. We will present these extensions in increasing order of complexity.  

\subsection{Returns over different time scales}

Let us define the $\ell$-day return between $t-\ell$ and $t$ as $R^{(\ell)}_t$, such that:
\be
R^{(\ell)}_t = \sum_{\tau=1}^{\ell} r_{t-\tau}; \qquad R^{(1)}_t = r_{t-1} = \ln p_{t-1} - \ln p_{t-2},
\ee
where $p_t$ is the price at time $t$. The simplest extension of ARCH($q$) is to allow all past 2-day returns to play a role as well, i.e.:
\be
\sigma_t^2 = s^2 + \sum_{\tau=0}^{q-1} \, g_1(\tau) [R^{(1)}_{t-\tau}]^2 + \sum_{\tau=0}^{q-2} \, g_2(\tau) [R^{(2)}_{t-\tau}]^2,
\ee
where $g_1(\tau)$ and $g_2(\tau)$ are coefficients, all together $2q-1$ of them. Upon identification with the QARCH kernel, one finds:
\begin{align}\nonumber
K(\tau,\tau)   &=  g_1(\tau-1) + g_2(\tau-1) + g_2(\tau-2),\\
K(\tau,\tau+1) &= g_2(\tau-1)\\\nonumber
K(\tau,\tau+\ell) &= 0 \, {\text{ for }}\, \ell \geq 2,
\end{align}
with the convention that $g_2(-1)=0$. The model can thus be re-interpreted in the following way: 
the square volatility is still a weighted sum of past {\it daily} squared returns, 
but there is an extra contribution that picks up the realized one-day covariance of returns. 
If $g_2(\tau) \geq 0$, the model means that the persistence of the same trend two days in a row leads to increased volatilities. 
A schematic representation of this model is given in Fig.~\ref{QARCHfig:matrices}(a).

One can naturally generalize again the above model to include 2-day, 3-day, $\ell$-day returns, 
with more coefficients $g_1(\tau), g_2(\tau), \dots, g_{\ell}(\tau)$, with a total of $\ell(2q+1-\ell)/2$ parameters. 
Obviously, when $\ell=q$, all possible time scales are exhausted, and the number of free parameters is $q(q+1)/2$, 
i.e.\ exactly the number of independent entries of the symmetric $q \times q$ feedback kernel $K$. 

\section{Some special families of QARCH models}\label{QARCHsec:section3}
\begin{figure}
	\center
	\includegraphics[scale=0.5]{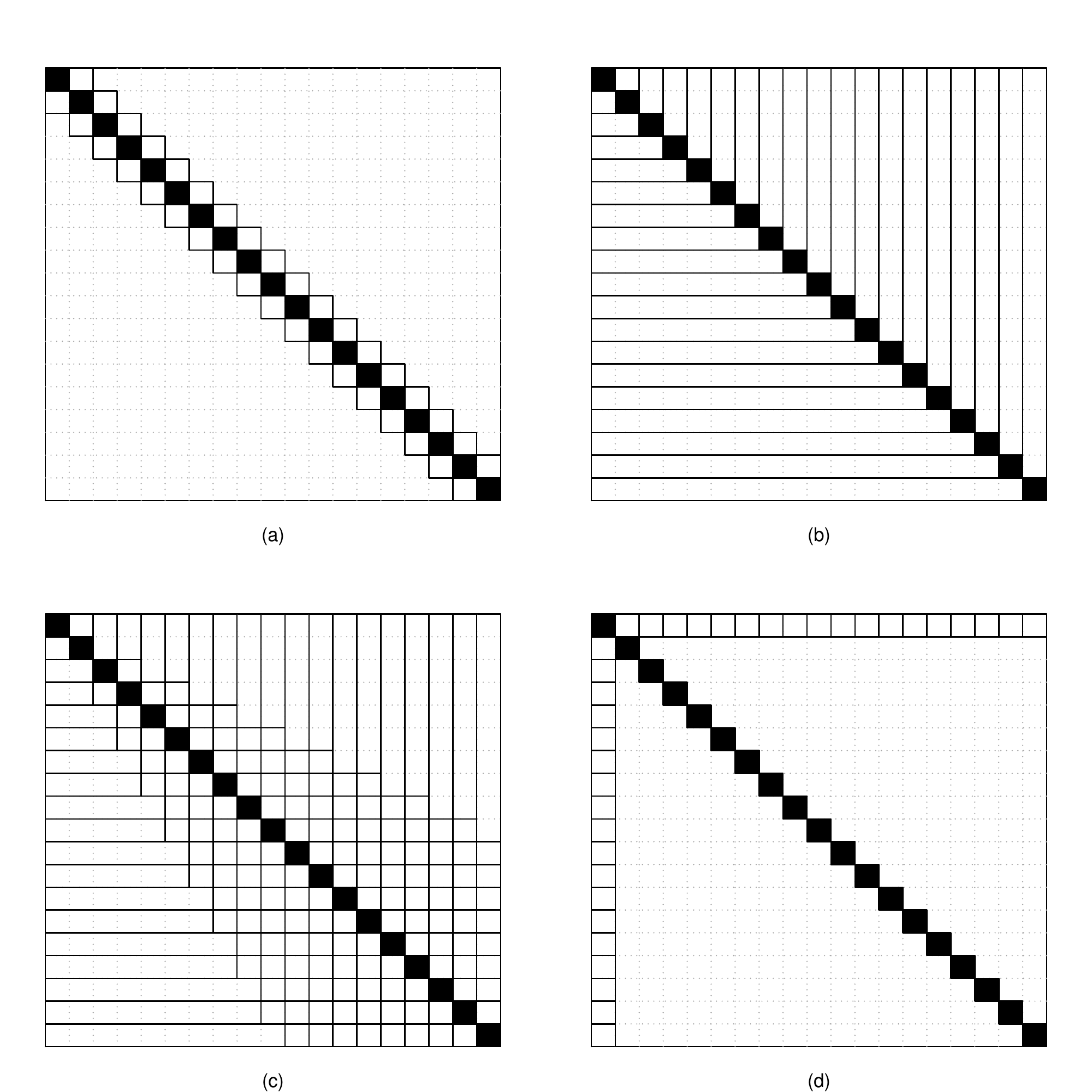}
	\caption{Matrix structures. \textbf{(a)} Overlapping two-scales; \textbf{(b)} Borland-Bouchaud multi-scale; 
	                            \textbf{(c)} Zumbach; \textbf{(d)} Long-Trend.}
	\label{QARCHfig:matrices}
\end{figure}

\subsection{Multi-scale, cumulative returns}

Another model, proposed in \cite{muller1997volatilities,borland2005multi}, is motivated by the idea that traders may be influenced not only by yesterday's return, but also by the change of 
price between today and 5-days ago (for weekly traders), or today and 20-days ago (for monthly traders), etc. In this case, the natural extension of the ARCH framework is to 
write:
\be
\sigma^2_t = s^2 + \sum_{\ell=1}^q g_{\text{BB}}(\ell) [R^{(\ell)}_{t}]^2,
\ee
where the index BB refers to the model put forward in \cite{borland2005multi}. The BB model requires a priori $q$ different parameters. 
It is simple to see that in this case, the kernel matrix can be expressed as:
\be
K_{\text{BB}}(\tau',\tau'') = G[\max(\tau',\tau'')], \quad {\mbox{with}} \quad G[\tau] = \sum_{\ell=\tau}^q g_{\text{BB}}(\ell).
\ee
The spectral properties of these matrices are investigated in detail in Appendix~\ref{QARCHapx:A}.
One can also consider a mixed model where both cumulative returns and daily returns play a role. 
This amounts to taking the off-diagonal elements of $K$ as prescribed by the above equation, 
but to specify the diagonal elements $K(\tau,\tau)$ completely independently from $G[\tau]$. 
This leads to a matrix structure schematically represented in Fig.~\ref{QARCHfig:matrices}(b), parameterized by $2q-1$ independent coefficients. 

\subsection{Zumbach's trend effect (ARTCH)}

Zumbach's model \cite{zumbach2010volatility} is another particular case in the class of models described by Eq.~\eqref{QARCHeq:quadraticARCH}.
It involves returns over different lengths of time and characterizes the effect of past trending aggregated returns on future volatility.
The quadratic part in the volatility prediction model is
\be
	\text{ARCH}+\sum_{\ell= 1}^{\lfloor{q/2\rfloor}} g_{\text{Z}}({\ell})\,R^{(\ell)}_{t}R^{(\ell)}_{t\!-\!\ell}
\ee
When relevant, only specific time scales like the day ($\ell=1$), the week ($\ell= 5$), the month ($\ell = 20$), etc.\ can be retained in the summation.
The off-diagonal elements of the kernel $K$ now take the following form:
\begin{equation}\label{QARCHeq:K.Z}
	K_{\text{Z}}(\tau',\tau'' > \tau')=\sum_{\ell = \max(\tau',\tfrac{\tau''}{2})}^{\min(\tau''\!-\!1,\lfloor{q/2\rfloor})}g_{\text{Z}}({\ell})
\end{equation}
Since it is upper triangular by construction, we symmetrize it as $\tfrac{1}{2}(K+K^{\dagger})$, and the diagonal is filled with the ARCH parameters.
This model contains $q+ \lfloor{q/2\rfloor}$ independent coefficients, and is schematically represented in Fig.~\ref{QARCHfig:matrices}(c). 

\subsection{A generalized trend effect}

In Zumbach's model, the trend component is defined by comparing returns computed over the same horizon $\ell$. 
This of course is not necessary. 
As an extreme alternative model, we consider a model where the volatility today is affected by 
the last return $r_{t-1}$ confirmation (or the negation) of a long trend.
In more formal terms, this writes:
\be
	\text{ARCH}+  r_{t-1} \times \sum_{\ell=1}^{q-1} g_{\text{LT}}({\ell})  r_{t-1-\ell}, 
\ee
where $g_{\text{LT}}(\ell)$ is the sequence of weights that defines the past ``long trend'' (hence the index LT). 
This now corresponds to a kernel $K$ with diagonal elements corresponding to the ARCH effects 
and a single non trivial row (and column) corresponding to the trend effect: $K(1,\tau>1) = g_{\text{LT}}(\tau-1)$.
This model has again $2q-1$ free parameters.

Of course, one can consider QARCH models that encode some, or all of the above mechanisms --- 
for example, a model that schematically reads $\text{ARCH} + \text{BB} + \text{LT}$ would require $3q-2$ parameters. 

\subsection{Spectral interpretation of the QARCH}\label{QARCHsec:spectral}

Another illuminating way to interpret QARCH models is to work  
in the diagonal basis of the $K$ matrix, where the quadratic term in Eq.~\eqref{QARCHeq:quadraticARCH} reads:
\begin{equation}\label{QARCHeq:eigen_decomp}
	 \sum_{\tau',\tau''=1}^{q}\left(\sum_{n}\lambda_nv_n(\tau')v_n(\tau'')\right)\,r_{t-\tau'}\,r_{t-\tau''}
	\equiv \sum_{n}\lambda_n\, \langle r | v_n \rangle_t^2
\end{equation}
with $(\lambda_n,v_n)$ being, respectively, the $n$-th eigenvalue and eigenvector of $K$, 
and $\langle r | v_n \rangle_t = \sum_{\tau=1}^{q}v_n(\tau)\,r_{t-\tau}$
 the projection of the pattern created by the last $q$ returns on the $n$-th eigenvector. 
One can therefore say that the square volatility $\sigma_t^2$ picks up contributions from various past returns eigenmodes. 
The modes associated to the largest eigenvalues $\lambda$ are those which have the largest contribution to volatility spikes. 

The ARCH($q$) model corresponds to a diagonal matrix; in this case the modes are trivially individual daily returns. 
Another trivial case is when $K$ is of rank one and its spectral decomposition is simply
\begin{equation}\label{QARCHeq:K.rank1}
	K(\tau',\tau'')=\lambda v(\tau')v(\tau'')
\end{equation}
where $\lambda=\tr(K)$ is the only non-null eigenvalue, 
and $v(\tau)=\sqrt{K(\tau,\tau)/\tr(K)}$ the eigenvector associated with this non-degenerate mode.
The corresponding contribution to the increase in volatility \eqref{QARCHeq:eigen_decomp} is therefore $\lambda \widehat R^2_t$,
where
\be
\widehat R_t = \vev{r|v}_t=\sum_{\tau=1}^qv(\tau)\,r_{t-\tau},
\ee
can be interpreted as an average return over the whole period, with a certain set of weights $v(\tau)$. 

The pure BB model (without extra ARCH contributions) can also be diagonalized analytically in the large $q$ limit for certain choices of the function $g_{\text{BB}}(\tau)$. 
We detail these calculations (which are mostly of theoretical interest) in Appendix~\ref{QARCHapx:A}.

\section{Empirical study: single stocks}\label{QARCHsect:emp_stocks}

We now turn to the calibration on real data of ``large'' QARCH models, i.e.\ models that take into account $q$ past returns with $q$ large (20 or more). 
The difficulty is that the full calibration of the matrix $K$ requires the determination of $q(q+1)/2$ parameters, which is already 
equal to $210$ when $q=20$! This is why the discussion of the previous section is important: imposing some {\it a priori} structure on the
matrix $K$ may help limiting the number of parameters, and gaining robustness and transparency. However, perhaps surprisingly, we will find that
none of the above model seem to have enough flexibility to reproduce the subtle structure of the empirically determined $K$ matrix.

\subsection{Dataset and methodology}

Equation~\eqref{QARCHeq:quadraticARCH} is a prediction model for the predicted variable 
$\sigma_t$ with explanatory variables past returns $r$ at all lags. The dataset we will use 
to calibrate the model is composed of daily stock prices (Open, Close, High, Low) 
for $N=280$ names present in the S\&P-500 index from Jan.\ 2000 to Dec.\ 2009 ($T=2515$ days), without 
interruption. The reference price for day $t$ is defined to be the close of that day $C_t$, and the return 
$r_t$ is given by $r_{t} = \ln C_{t} - \ln C_{t-1}$. The true volatility is of course unobservable; we 
replace it by the standard Rogers-Satchell (RS) estimator \cite{rogers1991estimating,floros2009modelling}: 
\be
	\widehat \sigma_t^2=\ln(H_t/O_t)\ln(H_t/C_t)+\ln(L_t/O_t)\ln(L_t/C_t).
\ee
As always in this kind of studies over extended periods of time, 
our dataset suffers from a selection bias since we have retained only 
those stock names that have remained alive over the whole period.

There are several further methodological points that we need to discuss right away:
\begin{itemize}
\item {\it Universality hypothesis}. We assume that the feedback 
matrix $K$ and the leverage kernel $L$ are identical for all stocks, once returns are standardized to get rid of 
the idiosyncratic average level of the volatility. This will allow us to use the whole data set (of size $N \times T$) to calibrate the model. 
Some dependence of $K$ and $L$ on global properties of firms (such as market cap, liquidity, etc.) may be
possible, and we leave this for a later study. However, we will see later that the universality hypothesis appears to 
be a reasonable first approximation.
\item {\it Removal of the market-wide volatility}. We anticipate that the volatility of a single 
stock has a market component that depends on the return of the index, and an idiosyncratic 
component that we attempt to account for with the returns of the stock itself. As a proxy for
the instantaneous market volatility, we take the cross-sectional average of the squared returns of individual
stocks, i.e.
\be
\Sigma_t = \sqrt{\frac{1}{N} \sum_{i=1}^N r_{i,t}^2 }
\ee
and redefine returns and volatilities as $r_t/\Sigma_t$ and $\widehat \sigma_t/\Sigma_t$. Finally, as 
announced above, the return time series are centered and standardized, and the RS volatility time series
are standardized such that $\langle \widehat \sigma_{i,t}^2 \rangle = 1$ for all $i$s. 
(This also gets rid of the multiplicative bias of the Rogers-Satchell estimator when used with non-Gaussian returns.)
\item {\it Calibration strategy}. The standard procedure used to calibrate ARCH models is maximum-likelihood, which 
relies on the choice of a family of distributions for the noise component $\xi$, often taken to be Student-t 
distributions (which include, in a limit, the Gaussian distribution). However, this method cannot be used 
directly in our case because there are far too many parameters and the numerical optimization of the log-likelihood 
is both extremely demanding in computer time and unreliable, as many different optima are expected in general. 
An alternative method, the Generalized Method of Moments (GMM\nomenclature{GMM}{Generalized method of moments}),
 is to determine the \mbox{$1+q+q(q+1)/2$} 
parameters of the model using empirically determined correlation functions that depend on $s^2$, $L(\tau)$ and $K(\tau,\tau')$.
This latter method is however sensitive to tail events and can lead to biases. 
We will therefore use a hybrid strategy, where a first estimate of these parameters, 
obtained using the GMM, serves as the starting point of a one-step likelihood maximization, 
which determines the set of most likely parameters in the vicinity of the GMM estimate 
(more details on this below). 
\item {\it Choice of the horizon $q$}. In principle, the value of the farthest relevant lag $q$ is an additional free parameter 
of the model, and should be estimated jointly with all the others. However, this would lead to a huge computational effort 
and to questionable conclusions. In fact, we will find that the diagonal elements $K(\tau,\tau)$ decay quite slowly with 
$\tau$ (in line with many previous studies) and can be accurately determined up to large lags using the GMM.
Off-diagonal elements, on the other hand, turn out to be much smaller and rather noisy. We will therefore restrict 
the horizon for these off-diagonal elements to $q=10$ (two weeks) or $q=20$ (four weeks). Longer horizons, although possibly
still significant, lead to very small out-of-sample extra predictability (but note that longer horizons {\it are} needed for the
diagonal elements of $K$).
\end{itemize}

\subsection{GMM estimation based on correlation functions}\label{QARCHsec:first_est}

On top of the already defined four-point correlation functions $\mathcal{C}^{(2)}(\tau)$ and $\mathcal{D}(\tau',\tau'')$ (and their 
corresponding ``tilde'' twins), we will introduce two- and three-point correlation functions that turn out to be useful (note that
the $r_t$s are assumed to have zero mean):
\begin{subequations}
\begin{align}
	           \mathcal{C}^{(1)}(\tau)        &\equiv\vev{r_t r_{t-\tau}}_t\\
	           \mathcal{C}^{(a)}(\tau)        &\equiv\vev{\left(r^2_t-\vev{r^2}\right)|r_{t-\tau}|}_t\\
	\widetilde{\mathcal{C}}^{(a)}(\tau)       &\equiv\vev{\left(\sigma^2_t-\vev{\sigma^2}\right)|r_{t-\tau}|}_t\\
	           \mathcal{L}      (\tau)        &\equiv\vev{\left(r^2_{t}-\vev{r^2}\right)r_{t-\tau}}_t\\
	\widetilde{\mathcal{L}}     (\tau)        &\equiv\vev{\left(\sigma^2_{t}-\vev{\sigma^2}\right)r_{t-\tau}}_t\\
	           \mathcal{L}^{(a)}(\tau)        &\equiv\vev{      |r_t| r_{t-\tau}}_t\\
	           \mathcal{D}^{(a)}(\tau',\tau'')&\equiv\vev{\left(|r_t|-\vev{|r|}\right)r_{t-\tau'}r_{t-\tau''}}.
\end{align}
\end{subequations}
The $\mathcal{C}^{(1)}(\tau)$ correlation function is by definition equal to $\vev{r_t^2}_t=1$ for $\tau=0$, 
and is usually considered to be zero for $\tau > 0$. 
However, as is well known, there are small anti-correlations of stock returns. 
On our data set, we find that these linear correlations are very noisy but significant, and can be fitted by:
\be
\mathcal{C}^{(1)}(\tau \geq 1) \approx - 0.04 \, \e^{-0.39 \tau},
\ee
corresponding to a decay time of \mbox{$\approx 2.5$} days.
The values of $\mathcal{C}^{(a)}$ characterize volatility correlations and are similar in spirit to $\mathcal{C}^{(2)}$, but they only involve 
third order moments of $r$, instead of fourth order moments, and are thus more robust to extreme events. 
The $\mathcal{L}$ correlations, on the other hand, characterize the leverage effect, 
i.e.\ the influence of the {\it sign} of past returns on future volatilities. 

These correlation functions allow us to define a well-posed problem of solving a system with \mbox{$1+q+\frac{q(q+1)}{2}$} unknowns
$\left(s^2, L(\tau), K(\tau',\tau'')\right)$ using the following \mbox{$1+q+q+\frac{q(q-1)}{2}$} equations that involve empirically 
measured correlation functions (in calligraphic letters), for \mbox{$1\leq\tau\leq q$} and \mbox{$1\leq\tau_2<\tau_1\leq q$}:
\begin{subequations}\label{QARCHeq:model_predictions}
\begin{align}\label{QARCHeq:2points}
	\vev{\sigma^2}&=s^2+\sum_{\tau',\tau''}K(\tau',\tau'')\mathcal{C}^{(1)}(\tau'\!-\!\tau'')\\
	\widetilde{\mathcal{L}}(\tau)&=\sum_{\tau'}L(\tau')\mathcal{C}^{(1)}(\tau\!-\!\tau')\\
	                         &+ \sum_{\tau'}K(\tau',\tau')\mathcal{L}(\tau\!-\!\tau')
	                         +2\sum_{\tau'\neq \tau}K(\tau',\tau)\mathcal{L}(\tau'\!-\!\tau)\\\label{QARCHeq:3pointsC}
	\widetilde{\mathcal{C}}^{(a)}(\tau)&\approx\sum_{\tau'}L(\tau')\mathcal{L}^{(a)}(\tau'\!-\!\tau)\\\nonumber
	&+\sum_{\tau'}K(\tau',\tau')\mathcal{C}^{(a)}(\tau-\tau')
	 +2\!\sum_{\tau'>\tau''>\tau>0}\!K(\tau',\tau'')\mathcal{D}^{(a)}(\tau'\!-\!\tau,\tau''\!-\!\tau)\\\label{QARCHeq:4points}
	\widetilde{\mathcal{D}}(\tau_1,\tau_2) &\approx L(\tau_2)\mathcal{L}(\tau_1\!-\!\tau_2)+L(\tau_1)\mathcal{L}(\tau_2\!-\!\tau_1)\\\nonumber
	                              +&2\sum_{\tau'>\tau_2}K(\tau',\tau_2)\left[\mathcal{D}(\tau_1\!-\!\tau_2,\tau'\!-\!\tau_2)
	                              + \mathcal{C}^{(1)}(\tau_1\!-\!\tau')-\mathcal{C}^{(1)}(\tau'\!-\!\tau_2)\mathcal{C}^{(1)}(\tau_1\!-\!\tau_2)\right]\\\nonumber
	                              +&\sum_{\tau'\leq\min(\tau_1,\tau_2)}\!K(\tau',\tau')\mathcal{D}(\tau_1\!-\!\tau',\tau_2\!-\!\tau'),
\end{align}
\end{subequations}
where all the sums only involve positive $\tau$s.
These equations are exact if all 3-point and 4-point correlations that involve $r$s at 3 (resp.\ 4) distinct times are strictly zero.
But since the linear correlations $\mathcal{C}^{(1)}(\tau>0)$ are very small, it is a safe approximation to neglect these higher order 
correlations. 

Note that the above equations still involve fourth order moments (the off-diagonal elements of $\mathcal{D}$),
that in turn lead to very noisy estimators of the off-diagonal of $K$.
In order to improve the accuracy of these estimators, 
we have cut-off large events by transforming the returns $r_t$ into $r_{\text{cut}} \tanh(r_t/r_{\text{cut}})$, 
which leaves small returns unchanged but caps large returns. 
We have chosen to truncate events beyond $3-\sigma$, i.e.\ $r_{\text{cut}}=3$. 
In any case, we will use the above equations in conjunction with maximum likelihood (for which no cut-off is used) to obtain more robust estimates of these off-diagonal elements. 

When solving the set of equations~\eqref{QARCHeq:model_predictions}, 
we find that the diagonal elements $K(\tau,\tau)$ are an order of magnitude larger than the 
corresponding off-diagonal elements $K(\tau,\tau'\neq \tau)$. 
This was not expected {\it a priori} and is in fact one of the central result of this study, 
and confirms that {\it daily returns} indeed play a special role in the volatility feedback mechanism, as postulated by ARCH models. 
Returns on different time scales, while significant, can be treated as a perturbation of the main ARCH effect.  
\begin{figure}
	\center
	\includegraphics[scale=0.6]{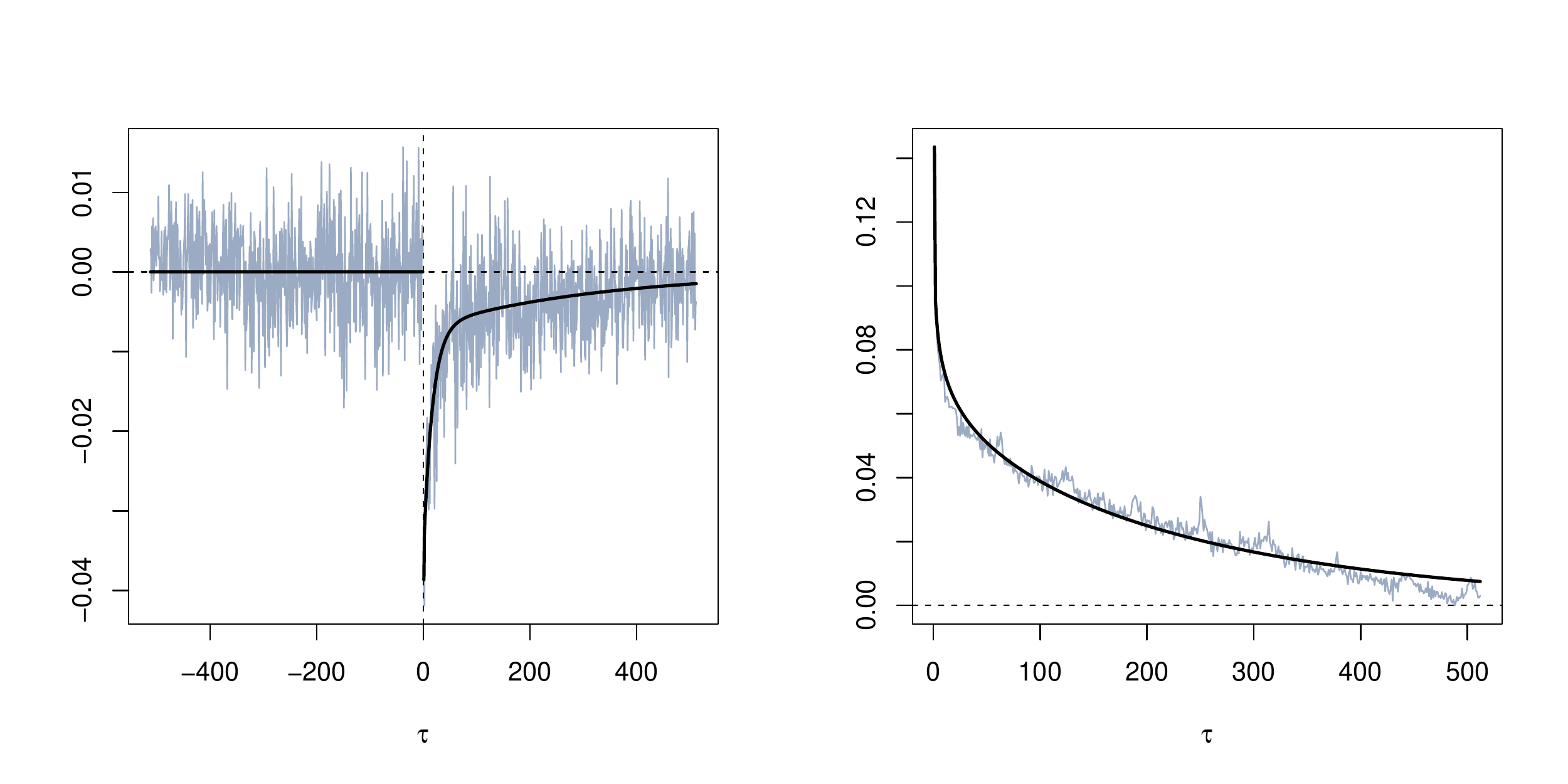}
	\caption{Two empirical correlation functions: the leverage $\mathcal{L}(\tau)$ and
	         the correlation of amplitudes $\widetilde{C}^{(a)}(\tau)$, together with their fits.
	         $\mathcal{L}(\tau)$ is fitted by the sum of two exponentials $-a\,\e^{-\tau/b} -c\,\e^{-\tau/d}$, with $a=0.007$, $b=327$ days, $c=0.029$, $d=17$ days; 
	         whereas $\widetilde{C}^{(a)}(\tau)$ is 
	         fitted by a power-law truncated by an exponential: $B \tau^{-\beta}\e^{-\tau/\tau_0}$, with $\beta=0.14, B=0.106, \tau_0=290$ days.}
	\label{QARCHfig:LCa}
%\end{figure}
%\begin{figure}[p]

\vfill
	\center
	\includegraphics[scale=0.66,trim=0 0 305 0,clip]{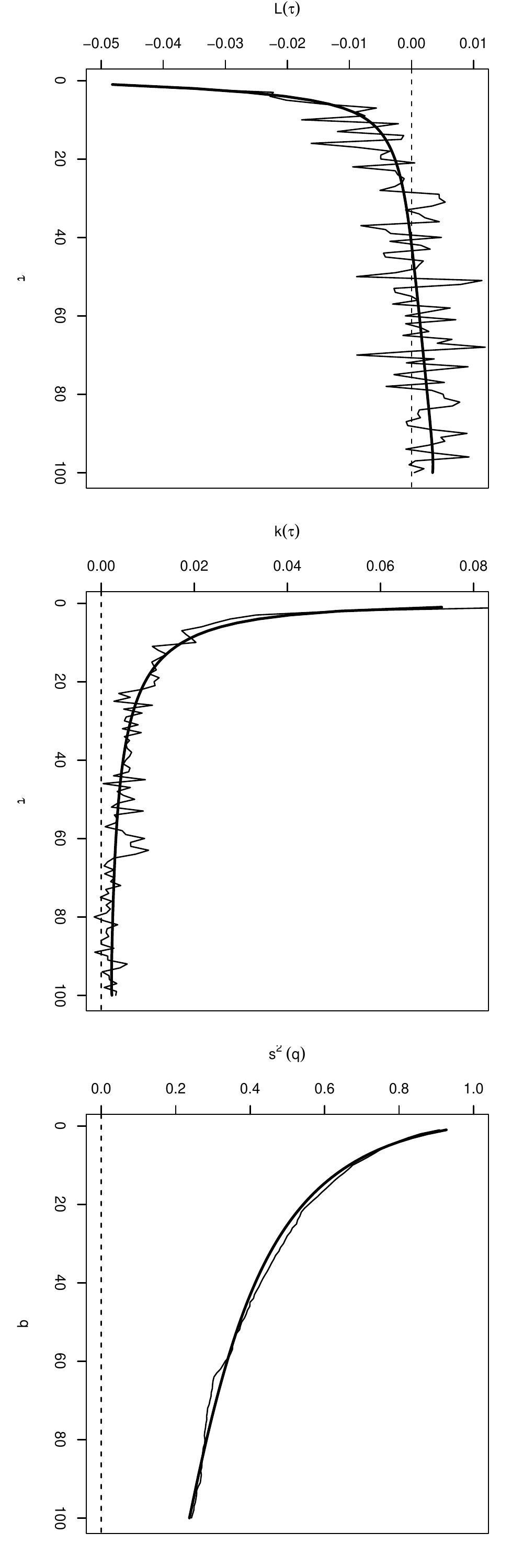}
	\caption{Calibration of the diagonal kernels for stocks, with $q=100$.
             \textbf{Left:} $L(\tau)$. \textbf{Right:} $K(\tau,\tau)$.
             The thick curves are obtained with fitted instead of raw input correlation functions, see Fig.~\ref{QARCHfig:LCa}.}
	\label{QARCHfig:L_stocks}
	% \includegraphics[scale=0.33]{/home/rchicheportiche/Redaction/paper/QARCH/s2_512_nb}
    % \caption{$s^2$ as a function of the farthest lag;
	% the red line is a fit according to the formula~\eqref{QARCHeq:fit_s2} (see text for parameter values).}
	% \label{QARCHfig:s2_stocks}
\end{figure}
\subsection{The diagonal kernels}

The above remark suggests to calibrate the model in two steps. 
We first neglect off-diagonal effects altogether, and determine the $2q+1$ parameters
$s^2, L(\tau)$ and $k(\tau)=K(\tau,\tau)$ through the following reduced set of equations:

\begin{subequations}\label{QARCHeq:model_simplified}
\begin{align}\label{QARCHeq:2points_simplified}
	\vev{\sigma^2}=1&=s^2+\sum_{\tau} k(\tau) \\
	\widetilde{\mathcal{L}}(\tau)&=L(\tau) + \sum_{\tau' \neq \tau}L(\tau')\mathcal{C}^{(1)}(\tau\!-\!\tau')
	                         + \sum_{\tau'} k(\tau')\mathcal{L}(\tau\!-\!\tau')
	                         \\
	                         \label{QARCHeq:3pointsC_simplified}
	\widetilde{\mathcal{C}}^{(a)}(\tau)&\approx\sum_{\tau'}L(\tau')\mathcal{L}^{(a)}(\tau'\!-\!\tau)+
	\sum_{\tau'}k(\tau')\mathcal{C}^{(a)}(\tau-\tau').
\end{align}
\end{subequations}

The input empirical correlation functions ${\mathcal{L}}(\tau)$ and $\widetilde{\mathcal{C}}^{(a)}(\tau)$ are plotted in Fig.~\ref{QARCHfig:LCa}, 
together with, respectively, a double-exponential fit and a truncated power-law fit (see legend for parameters values).
$\widetilde{\mathcal{L}}$ and $\mathcal{L}^{(a)}$ look very similar to $\mathcal{L}(\tau)$; 
note that all these functions are approximately zero for \mbox{$\tau < 0$}. 
The above equations are then solved using these analytical fits, 
which leads to the kernels $k(\tau)$ and $L(\tau)$ that we report in bold in Fig.~\ref{QARCHfig:L_stocks}. 
Using the raw data --- instead of the fits --- for all the correlation functions results in more noisy $L(\tau)$ and $k(\tau)$ (thin lines),
but still very close to the bold curves shown in Fig.~\ref{QARCHfig:L_stocks}. 
As expected, $L(\tau)$ is very close to $\widetilde{\mathcal{L}}(\tau)$: 
there is a weak, but significant leverage effect for individual stocks.

\begin{figure}
	\center
	\includegraphics[scale=0.425]{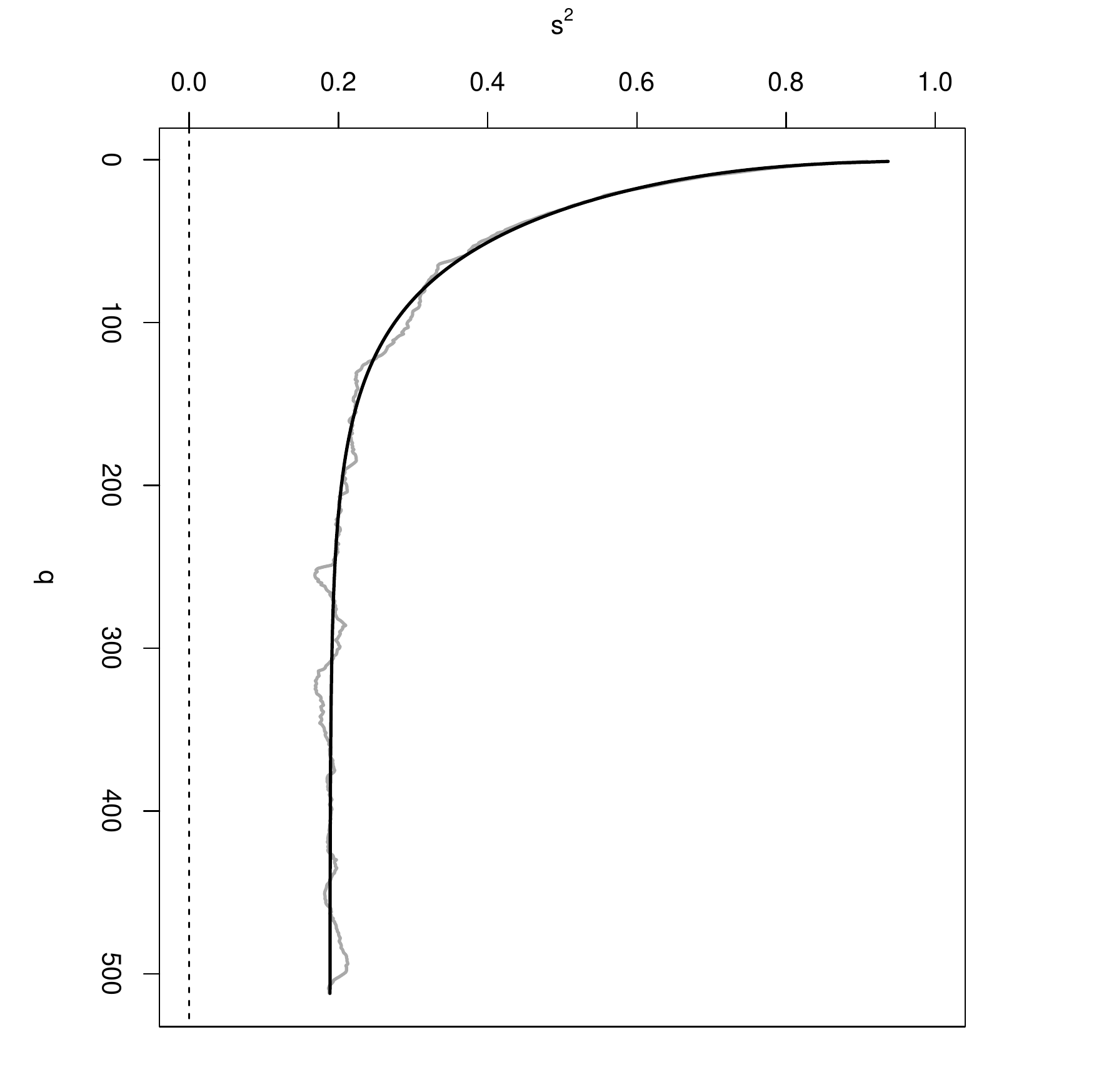}
    \caption{$s^2$ as a function of the farthest lag;
	the solid line is a fit according to the formula~\eqref{QARCHeq:fit_s2} (see text for parameter values).}
	\label{QARCHfig:s2_stocks}
\end{figure}

We then show in Fig.~\ref{QARCHfig:s2_stocks} a plot of $s^2(q) = 1 - \sum_{\tau=1}^q k(\tau)$ as a function of $q$. 
Including the feedback of the far away past  progressively decreases the value of the baseline level $s^2$.
In order to extrapolate to $q \to \infty$, we have found that the following fit is very accurate:
\be\label{QARCHeq:fit_s2}
	s^2(q) = s_\infty^2 + g\,\frac{q^{1-\alpha}}{\alpha-1} \,\e^{-q/q_0},
\ee
with $s_\infty^2 \approx 0.21, \alpha \approx 1.11, g \approx 0.081$ and $q_0 \approx 53$. 
Several comments are interesting at this stage:
\begin{itemize}
\item The asymptotic value $s_\infty^2$ is equal to $20 \%$ of the observed squared volatility, 
meaning that volatility feedback effects increase the volatility by a factor \mbox{$\approx 2.25$} 
(but remember that we have scaled out the market wide volatility). 
Such a strong amplification of the volatility resonates with Shiller's ``excess volatility puzzle'' 
and gives a quantitative estimate of the role of self-reflexivity and feedback effects in financial markets 
\cite{shiller1981stock,cutler1989moves,fair2002events,soros1994alchemy,bouchaud2011endogeneous,filimonov2012quantifying}.
\item The above fit corresponds to a power-law behavior, \mbox{$k(\tau) \approx g \tau^{-\alpha}$} for $\tau \ll q_0$, 
and an exponential decay for larger lags.
Therefore, a characteristic time scale of $q_0 \approx 3$ months appears, 
beyond which volatility feedback effects decay more rapidly. 
\item With a diagonal positive kernel $K$, the condition for positive definiteness \eqref{QARCHeq:nonneg} of the quadratic form reads $\sum_{\tau}L(\tau)^2/k(\tau)\leq4s^2$.
The estimated values of $L$ and $k$ yield a left-hand side equal to 0.595, while the right-hand side amounts to 0.823.
\item Using the results of Sect.~\ref{QARCHsec:fourth_moment}, one can compute the theoretical value of $\langle \sigma^4 \rangle$ 
that corresponds to the empirically determined $k(\tau)$. 
As expected from the fact that $g$ is small and $\alpha$ close to unity, one finds that the fluctuations of volatility induced by the long-memory feedback are weak: 
$\langle \sigma^4 \rangle = 1.156$ (see Eq.~\eqref{QARCHeq:sigma4pert} above).
Including the contribution of the leverage kernel $L(\tau)$ to $\langle \sigma^4 \rangle$ does not change much the final numerical value, 
that shifts from 1.156 to 1.161.
\item The smallness of $\vev{\sigma^4}-\vev{\sigma^2}^2$ shows that most of the kurtosis of the returns $r_t = \sigma_t \xi_t$ must come from the noise $\xi_t$, 
which cannot be taken as Gaussian. Using the diagonal ARCH model with the kernels determined as above to predict $\sigma_t$, 
one can study the distribution of $\xi_t = r_t/\sigma_t$ and find the most likely Student-t distribution that accounts for it. 
We find that the optimal number of degrees of freedom is \mbox{$\nu \approx 6.4$}, and the resulting fit is shown in Fig.~\ref{QARCHfig:mu_xi}. 
Note that while the body and `near-tails' of the distribution are well reproduced by the Student-t, 
the far-tails are still fatter than expected. 
This is in agreement with the commonly accepted tail index of \mbox{$\nu_\text{tail} \approx 4$}, significantly smaller than $6.4$. 
\end{itemize}

 \begin{figure}
 	\center
 	\includegraphics[trim=500 0 0 0,clip,scale=0.425]{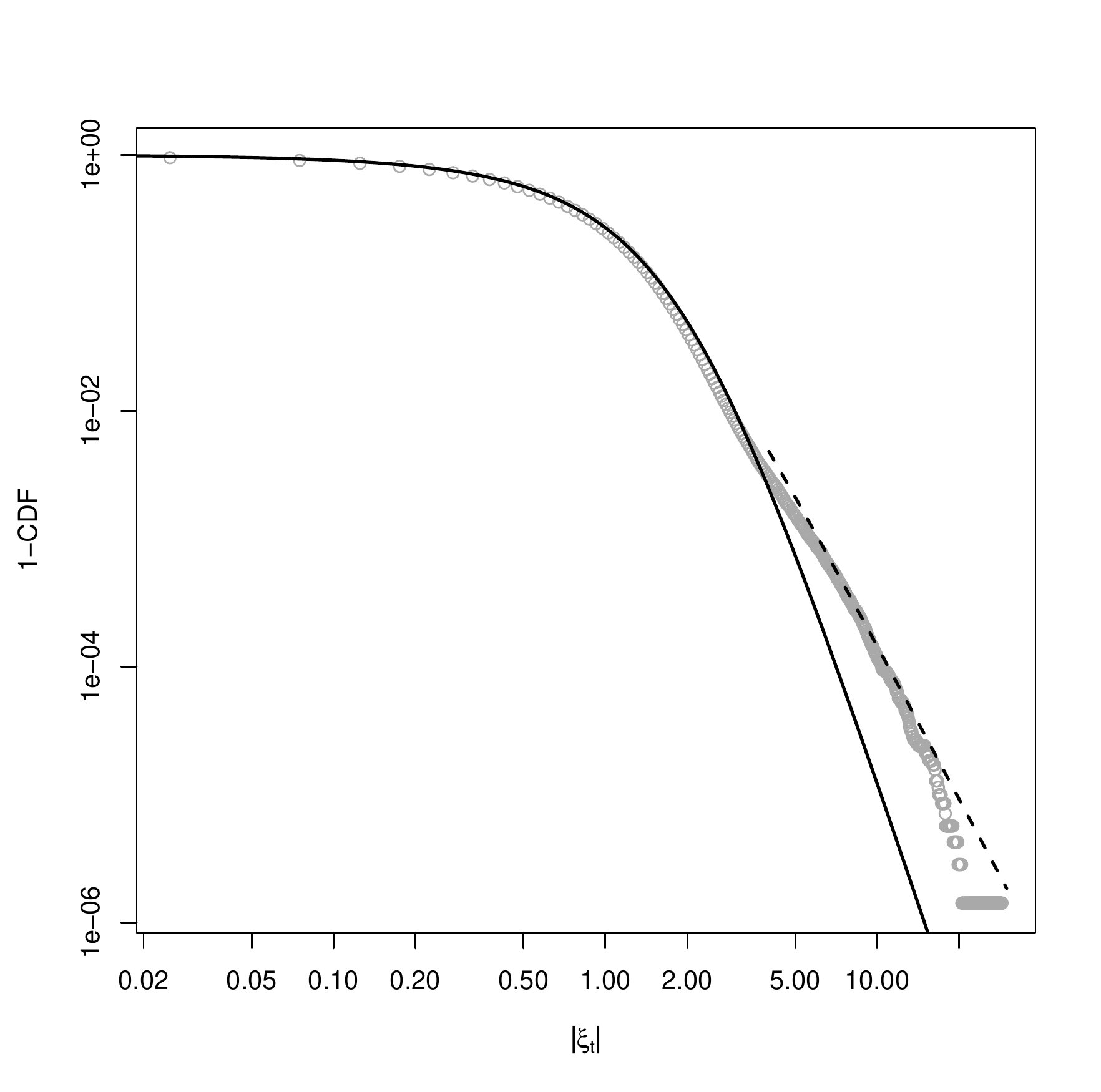}
 	\caption{Cumulative distribution function of the residuals $\xi_t=r_t/\sigma_t$. 
 	The plain line is the Student distribution with best (maximum-likelihood) fitting tail parameter $\nu=6.4$. 
 	Far tails suggest a fatter distribution with a smaller value of $\nu_\text{tail} \approx 4$ (dashed).}
 	\label{QARCHfig:mu_xi}
 \end{figure}

Assuming $\xi_t$ to be a Student-t random variable with $\nu=6.4$ degrees of freedom, 
we have re-estimated $k(\tau)$ and $L(\tau)$ using maximum-likelihood (see below). 
The final results are more noisy, but close to the above ones after fitting. Our  
strategy is thus to fix both $k(\tau)$ and $L(\tau)$, and only focus on the off-diagonal elements of $K$ henceforth. 

\subsection{The off-diagonal kernel, GMM \& maximum likelihood estimation}

We can now go back to Eq.~\eqref{QARCHeq:4points} that allows one to solve for $K(\tau',\tau'' > \tau')$, once $k(\tau)$ and $L(\tau)$ are 
known. As announced above, we choose $q=20$ for the time being. Because $\mathcal{D}$ involves the fourth moment of the returns, this procedure
is not very stable, even with a lot of data pooled together, and even after the truncation of large returns. Maximum likelihood estimates 
would be more adapted here, but the dimensionality of the problem prevents a brute force determination of the $q(q-1)/2$ parameters. 
In order to gain some robustness, we use the following strategy. The Student log-likelihood per point ${\mathcal{I}}$, with $\nu$ degrees of freedom, is given by:%
\footnote{In the following we do {\it not} truncate the large returns, but completely neglect the weak linear correlations $\mathcal{C}^{(1)}(\tau)$ that are present for small lags, and
that should in principle be taken into account in the likelihood estimator.}
\be\label{QARCHeq:loglike}
{\mathcal{I}}_\nu(L,K,\{r_t\})=\frac{1}{2T}\sum_t\left[\nu\,\ln{a_t^2}-(\nu+1)\,\ln(a_t^2+r_t^2)\right],\quad a_t^2 \equiv(\nu-2)\,\sigma_t^2,
\ee
where $r_t\equiv \sigma_t\xi_t$ and $\sigma_t^2$ is given by the QARCH model expression, Eq.~\eqref{QARCHeq:quadraticARCH}, 
and in this section $K$ is a notation for the off-diagonal content only. 
We fix $\nu=6.4$ and determine numerically the gradient 
$\partial {\mathcal{I}}_\nu$
and the Hessian $\partial \partial {\mathcal{I}}_\nu$ of ${\mathcal{I}}_\nu$ as a function of {\it all} the $q(q-1)/2$ off-diagonal coefficients 
$K(\tau',\tau'' > \tau')$, computed either around the GMM estimates of these parameters, or around the ARCH point where all these coefficients are zero. 
Note that $\partial {\mathcal{I}}_\nu$ is a vector with $q(q-1)/2$ entries and $\partial \partial {\mathcal{I}}_\nu$ is a $q(q-1)/2 \times q(q-1)/2$
matrix. It so happens that the eigenvalues of the Hessian are all found to be {\it negative}, i.e.\ the starting point is in the vicinity of a local
maximum. This allows one to find easily the values of the $q(q-1)/2$ parameters that maximize the value of ${\mathcal{I}}_\nu$; they are (symbolically) given by:
\be
K^* = K_0 - \left(\,\overline{\partial \partial {\mathcal{I}}_\nu}\,\right)^{-1} \cdot\overline{ \partial {\mathcal{I}}_\nu},
\ee
where $K_0$ is the chosen starting point --- either the GMM estimate $K_0=K_{\text{GMM}}$ based on Eq.~\eqref{QARCHeq:4points}, 
or simply $K_0=0$ if one starts from a diagonal ARCH model --- and the overline on top of $\partial\partial\mathcal{I},\partial\mathcal{I}$ indicates averaging over stocks. 
This one step procedure is only approximate but can be iterated; 
it however assumes that the maximum is close to the chosen initial point, 
and would not work if some eigenvalues of the Hessian become positive. 
In our case, both starting points ($K_0=0$ or $K_0=K_{\text{GMM}}$) lead to nearly exactly the same solution; 
furthermore the Hessian recomputed at the solution point is very close to the initial Hessian, 
indicating that the likelihood is a locally quadratic function of the parameters, 
and the gradient evaluated at the solution point is very close to zero in all directions,
confirming that a local maximum as been reached.

Before exposing our results, we briefly go through a digression to discuss the bias and error 
on the estimated parameters $K^*$ as well as on the resulting maximal average likelihood $\overline{\mathcal{I}}^*$.
The likelihood $\mathcal{I}$, its gradient $\partial {\mathcal{I}}$ and Hessian matrix $\partial \partial {\mathcal{I}}$
are generic functions of the set of parameters $K$ to be estimated, and of the dataset, of size $n$.
As the number $n$ of observations tends to infinity, the covariance matrix of the ML \nomenclature{ML}{Maximum likelihood} estimator of the parameters is well known to be $(nI)^{-1}$,
where $I$ is the Fisher Information matrix
\[
	I=\Esp{-\partial \partial {\mathcal{I}}}\approx -\overline{\partial \partial {\mathcal{I}}}(K^*)
	                                       %\approx -\partial \partial {\mathcal{I}}(K^{\infty}),
\]
while the asymptotic bias scales as $n^{-1}$ and is thus much smaller than the above error ($\sim n^{-1/2}$). 
As a consequence, ML estimates of $K$ exceeding $\pm \diag(-n\,\overline{\partial\partial\mathcal{I}}^*)^{-1/2}$ will be deemed significant.
The {\it bias} on the average in-sample (IS) value of the maximum likelihood itself can be computed to first order in $\frac{1}{n}$, 
and is very generally found to be $+M/2n$, where $M$ is the number of parameters to be determined.
Similarly, the bias on the average out-of-sample value of the maximum likelihood is $-M/2n$.%
\footnote{These corrections to the likelihoods 
lead to the (per point) Akaike Information Criterion \cite{akaike1974new} \mbox{$AIC=-2(\mathcal{I}-M/n)$}, 
that trades off the log-likelihood and the number of parameters. AIC is used for model selection purposes mainly. 
When comparing parametric models with the same number of parameters, AIC is not more powerful than the likelihood.}
Since each sampling of our data set will contain $n = T\cdot N/2\approx 350,000$ observations, 
differences of likelihood smaller than $M/2n \approx 3 \cdot 10^{-4}$ are 
insignificant when $M=190$ (corresponding to all off-diagonal elements when $q=20$). 
This number is $\approx 5$ times smaller when one considers the restricted models introduced above (which contain $\approx 40$ parameters). 

The most likely off-diagonal coefficients of $K^*$ are found to be highly significant (see Tab.~\ref{QARCHtab:IS-OOS}): 
the IS increase of likelihood from the purely diagonal ARCH($q$) model is $\Delta \mathcal{I} \approx 10^{-3}$ per point. 
This is confirmed by an out-of-sample (OS) experiment, where we determine $K^*$ on half the pool of stocks 
and use it to predict the volatility on the other half (whence the above factor ${1}/{2}$ in the numerical estimation of $n$). 
The experiment is performed over $N_{\text{samp}}=150$ random pool samplings. 
The average OS likelihood is very significantly better for the full off-diagonal kernel $K^*$ than for the diagonal ARCH($q$), 
itself being better than the GMM estimate $K_{\text{GMM}}$ based on Eq.~\eqref{QARCHeq:4points}, and probably subject to biases due to the
truncation procedure. Note that the full off-diagonal kernel $K^*$ has many more parameters than the diagonal ARCH($q$); it therefore 
starts with a handicap out-of-sample since the bias on the OS likelihood is, as recalled above, $-M/2n$.

However, as announced above, the off-diagonal elements of $K^*$ are a factor ten smaller than the corresponding diagonal values. 
We give a heat-map representation of the matrix $K^*$ in Fig.~\ref{QARCHfig:heatmap}. Various surprises are immediately apparent. 

First, while the off-diagonal elements are mostly positive for small $\tau',\tau''$, 
clear negative streaks appear for intermediate and large $\tau$s. 
This is unexpected, since one would have naively guessed that any trend 
(i.e.\ positive realized correlations between returns) should {\it increase} future volatilities. 
Here we see that some quadratic combinations of past returns contribute negatively to the volatility. 
This will show up in the spectral properties of $K^*$ (see Sect.~\ref{QARCHsec:spectral2} below).

The second surprise is that there does not seem to be any obvious structure of the matrix, 
that would be reminiscent of one of the simple models represented in Fig.~\ref{QARCHfig:matrices}.
This means that the fine structure of volatility feedback effects is much more subtle than anticipated. 

 \begin{figure}
 	\center
 	\includegraphics*[scale=0.35,angle=-90]{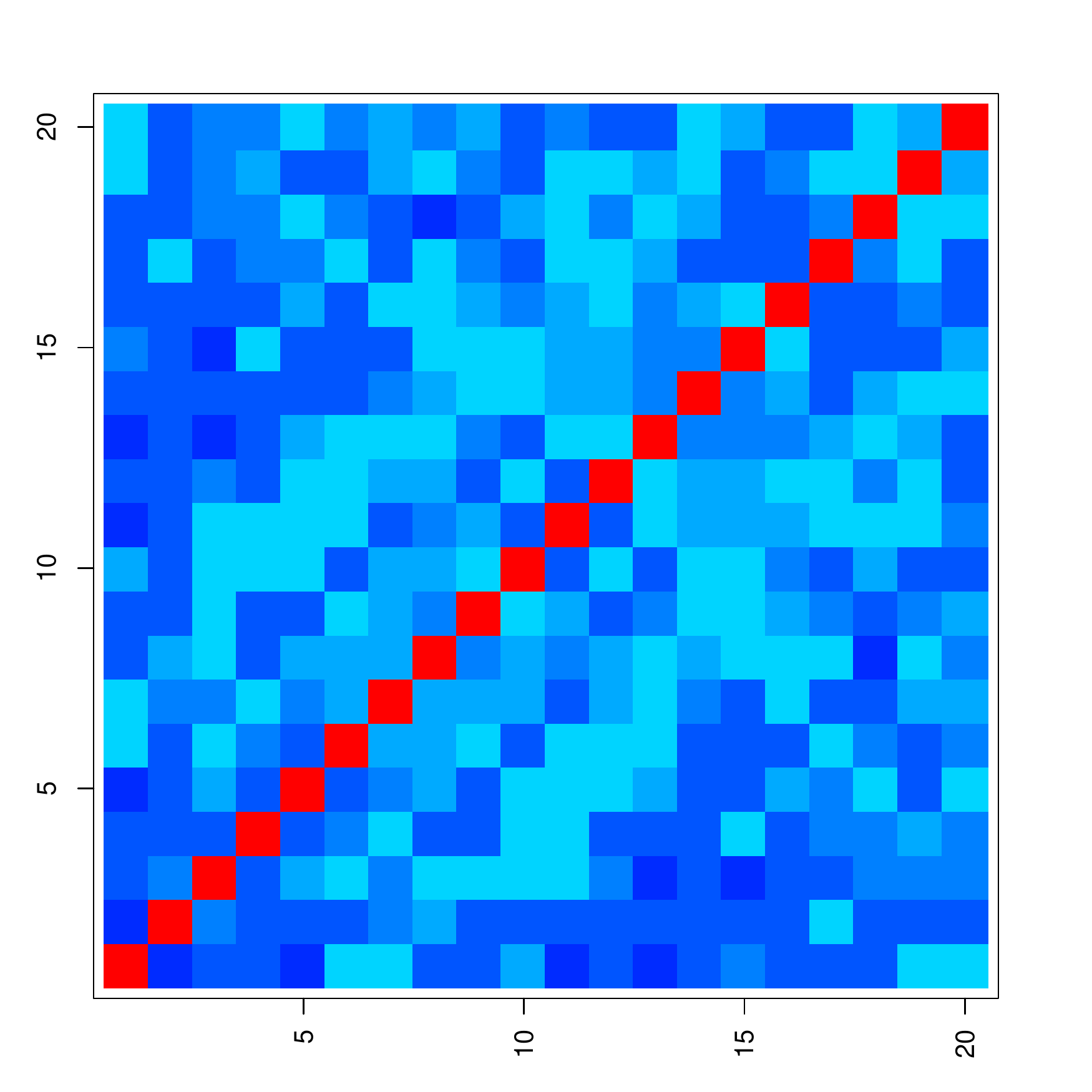}
 	\includegraphics*[scale=0.35,angle=-90]{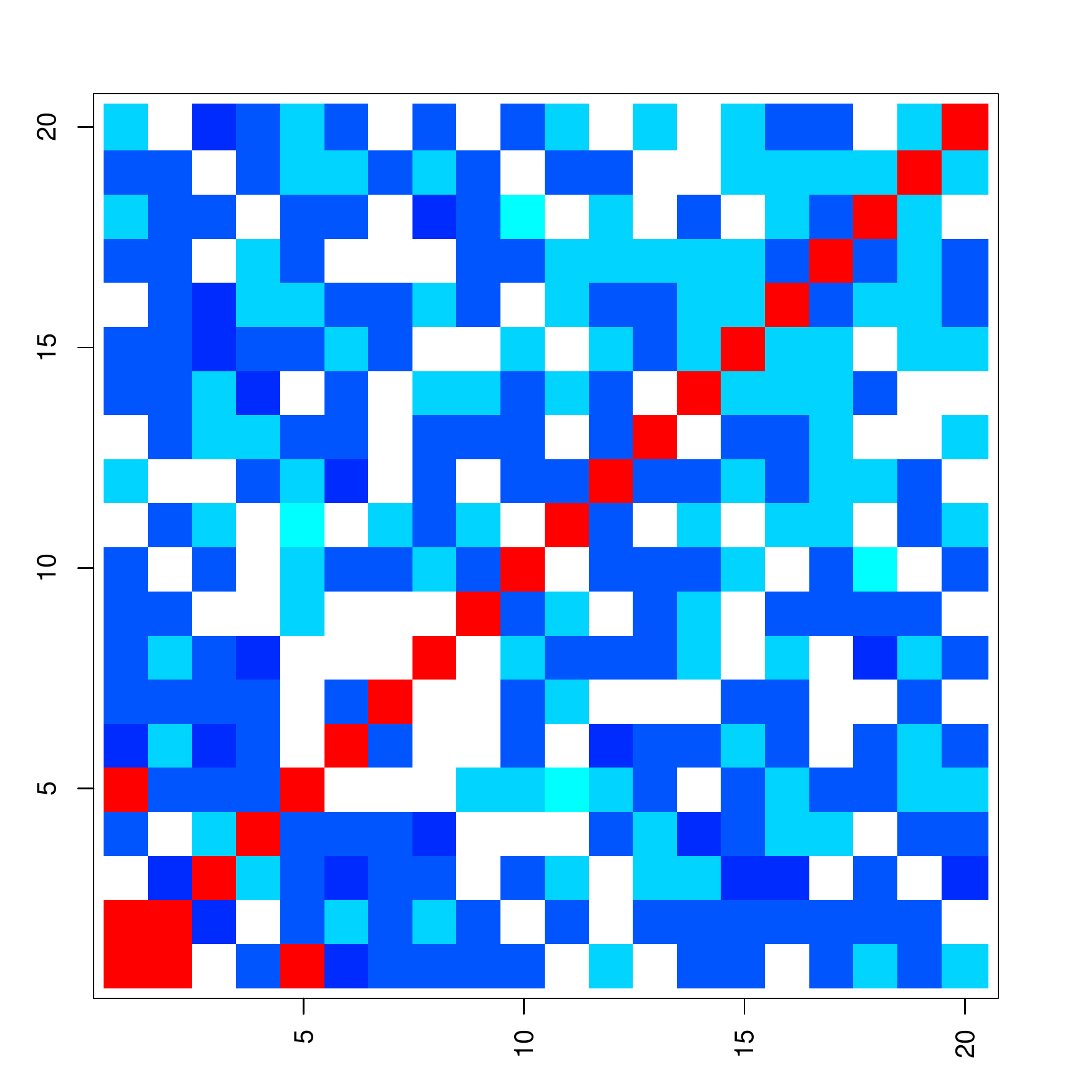}
 	\includegraphics*[scale=0.35,angle=-90]{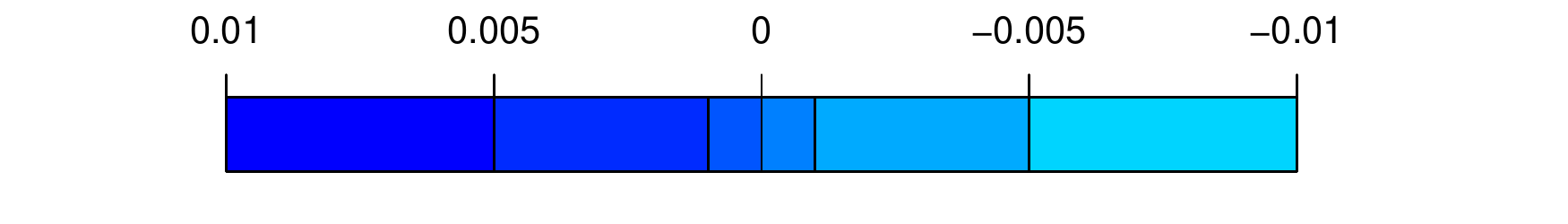}
 	\caption{Heatmap of the unconstrained model.
 	\textbf{Left: } GMM estimation; \textbf{Right: } ML estimation (white spots correspond to values smaller than their corresponding error margin). Red spots correspond to values 
 	larger than $0.01$. Note the 
 	negative streaks at large lags, and the significant off-diagonal entry for $\tau=1,\tau'=2$ days. }
 	\label{QARCHfig:heatmap}
 \end{figure}

We have nevertheless implemented a {\it restricted} maximum-likelihood estimation 
that imposes the structure of one of the models considered in Sect.~\ref{QARCHsec:section3}. 
We find that all these models are equally ``bad'' --- 
although they lead to a significant increase of likelihood compared to the pure ARCH case, both IS and OS, 
they are all superseded by the unconstrained model shown in Fig.~\ref{QARCHfig:heatmap}, again both IS and OS (see Tab.~\ref{QARCHtab:IS-OOS}). 
The best OS model is ``Long-trend'', with a kernel $g_{\text{LT}}(\ell)$ shown in Fig.~\ref{QARCHfig:structure_fcts}, 
together with the functions $g_2(\ell)$, $g_{\text{BB}}(\ell)$, $g_{\text{Z}}(\ell)$. 
While $g_{\text{LT}}(\ell)$ looks roughly like an exponential with memory time $10$ days, the two-day return kernel $g_2(\ell)$ reveals intriguing oscillations. 
Two-day returns influence the volatility quite differently from one day returns!
On the other hand, we do not find any convincing sign of the multi-scale ``BB'' structure postulated in \cite{borland2005multi}. 

\begin{figure}
	\center
	\includegraphics[scale=0.45]{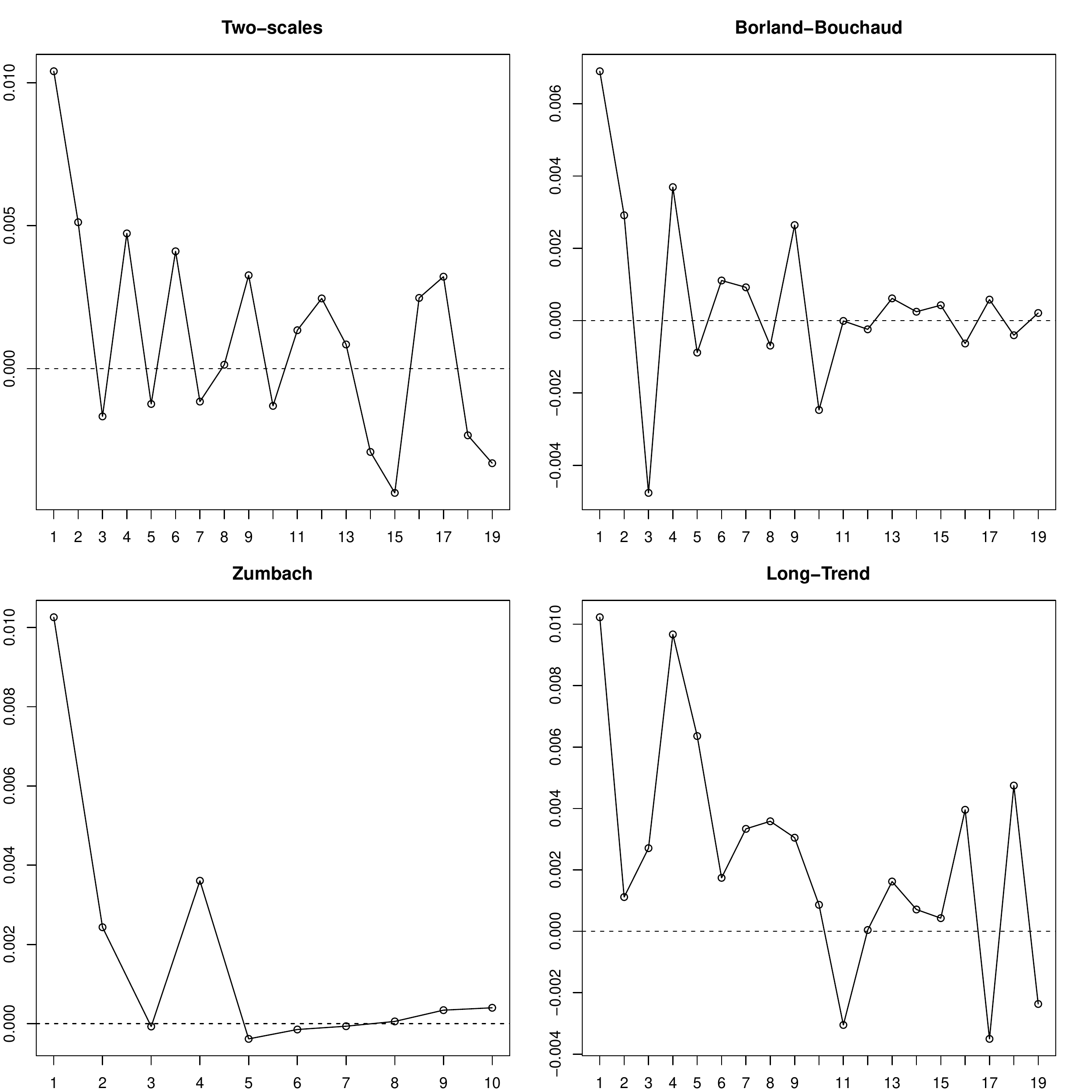}
	\caption{Plot of the empirically determined kernels $g_2(\ell)$, $g_{\text{BB}}(\ell)$, $g_{\text{Z}}(\ell)$ and $g_{\text{LT}}(\ell)$ for the 
	restricted models of Sect.~\ref{QARCHsec:section3}.}
	\label{QARCHfig:structure_fcts}
\end{figure}

Note that the structure shown in Fig.~\ref{QARCHfig:structure_fcts} is found to be stable when $q$ is changed. 
It would be interesting to subdivide the pool of stocks in different categories 
(for example, small caps/large caps) or in different sub-periods, and study how the off-diagonal structure of $K$ is affected.
However, we note that the dispersion of likelihood over different samplings of the pool of stocks is only $50\%$ larger than 
the ``true'' dispersion, due only to random samplings of a fixed QARCH model with parameters calibrated to the data (see caption of Tab.~\ref{QARCHtab:IS-OOS}). 
This validates, at least as a first approximation, the assumption of homogeneity among all the stock series that have been averaged over.

%Finally, we have revisited the most likely value of the Student parameter $\nu$ with now the full matrix $K^*$, plus diagonal terms up to 
%$q=100$, and found again $\nu=6.4$. This shows that our procedure is consistent from that point of view.

To conclude this empirical part, we have performed several ex-post checks to be sure that our assumptions and preliminary estimations are justified.
First, we have revisited the most likely value of the Student parameter $\nu$ 
(tail index of the distribution of the residuals $\xi(t)=r(t)/\sigma_{\text{QARCH}}(t)$)
with now the full matrix $K^*$, plus diagonal terms up to $q=100$, and found again $\nu=6.4$. 
This shows that our procedure is consistent from that point of view.
Second, we have computed the quadratic correlation of the residuals $\xi_t$, which are assumed in the model to be IID random variables with, in 
particular, no variance autocorrelation: $\vev{\xi_t^2\xi_{t-\tau}^2}-1=0$ for $\tau \neq 0$. 
Empirically, we observe a negative correlation of weak magnitude exponentially decaying with time. 
This additional dependence of the amplitude of the residuals, together with the excess fat tails of their probability distribution,
is probably a manifestation of the truly exogenous events occurring in financial markets that have different statistical properties \cite{joulin2008stock} 
and not captured by the endogenous feedback mechanism. Finally, about the universality hypothesis, we discuss in the caption of Tab.~\ref{QARCHtab:IS-OOS} how the assumption of homogeneous stocks
is validated by comparing the cross-sectional dispersion of the likelihoods obtained empirically and on simulated series.

\begin{table}
\begin{center}
\begin{tabular}{c||cc|c|}
  &GMM        &ARCH(20)    &ARCH+ML\\\hline\hline
IS&$-1.31533$ &$-1.31503$ &$\mathbf{-1.31405}$\\
OS&$-1.32003$ &$-1.31971$ &$\mathbf{-1.31914}$\\\hline
\end{tabular}

\vspace{2em}

\begin{tabular}{c|cccccc|}
ARCH+   &   BB       &  Z        &  2s       &  LT       &  2s+Z    &   2s+LT   \\\hline\hline
IS      &$-1.31486 $ &$-1.31490$ &$-1.31490$ &$-1.31487$ &$-1.31488$&$-1.31489$ \\
OS      &$-1.31960 $ &$-1.31957$ &$-1.31962$ &$-1.31957$ &$-1.31956$&$-1.31957$ \\\hline
\end{tabular}
\caption{Log-likelihoods, according to Eq.~\eqref{QARCHeq:loglike}. 
In-sample and out-of-sample likelihoods are computed as follows: 
for each of $N_{\text{samp}}=150$ iterations, half of the stock names are randomly chosen for the calibration of $K,L$ 
and the likelihood is computed with the obtained $K^*,L^*$ on each series of the same sample (`In-sample' likelihoods).
Then, the likelihood is again computed with the same parameters but on the series of the other sample (`Out-of-sample' likelihoods).
While the former quantify how much the estimated model succeeds in reproducing the given sample, 
the latter measure the reliability of the model on \emph{other} similar datasets. 
In order to quantify the validity of the model in an absolute way, the likelihood can be compared with the ``true'' value,  
obtained with simulated data (since an analytical treatment is out of reach). 
The average likelihood per point $\overline{\mathcal{I}}^*(r_t)$ with $r_t$ simulated as a QARCH with parameters $K^*,L^*$, 
and $\nu=6.4$ is equal to $-1.34019$, which is 1.5\% away from the empirical values\protect\footnotemark.
The likelihoods reported in the table are averages over all samplings, and the corresponding 1-s.d.\ \nomenclature{s.d.}{Standard deviation}
dispersion is found to be $\approx 3\cdot 10^{-3}$ in all cases, 
to be compared to  $2\cdot 10^{-3}$ for random samplings of a fixed QARCH model with the same parameters.
}
\label{QARCHtab:IS-OOS}
\end{center}
\end{table}
\footnotetext{Note that the true likelihood is not necessarily larger than the realized one under a misspecified model.}

\subsection{Spectral properties of the empirical kernel $K$}
\label{QARCHsec:spectral2}
\begin{sidewaysfigure}
	\center
    \subfigure[The difference between the ranked eigenvalues of the estimated kernels and those of $K_{\text{ARCH}}$ 
	           as a function of the latter (again ranked).
	           The dashed oblique line has slope $-1$ and separates positive eigenvalues from negative ones.]{
	\includegraphics[scale=0.6]{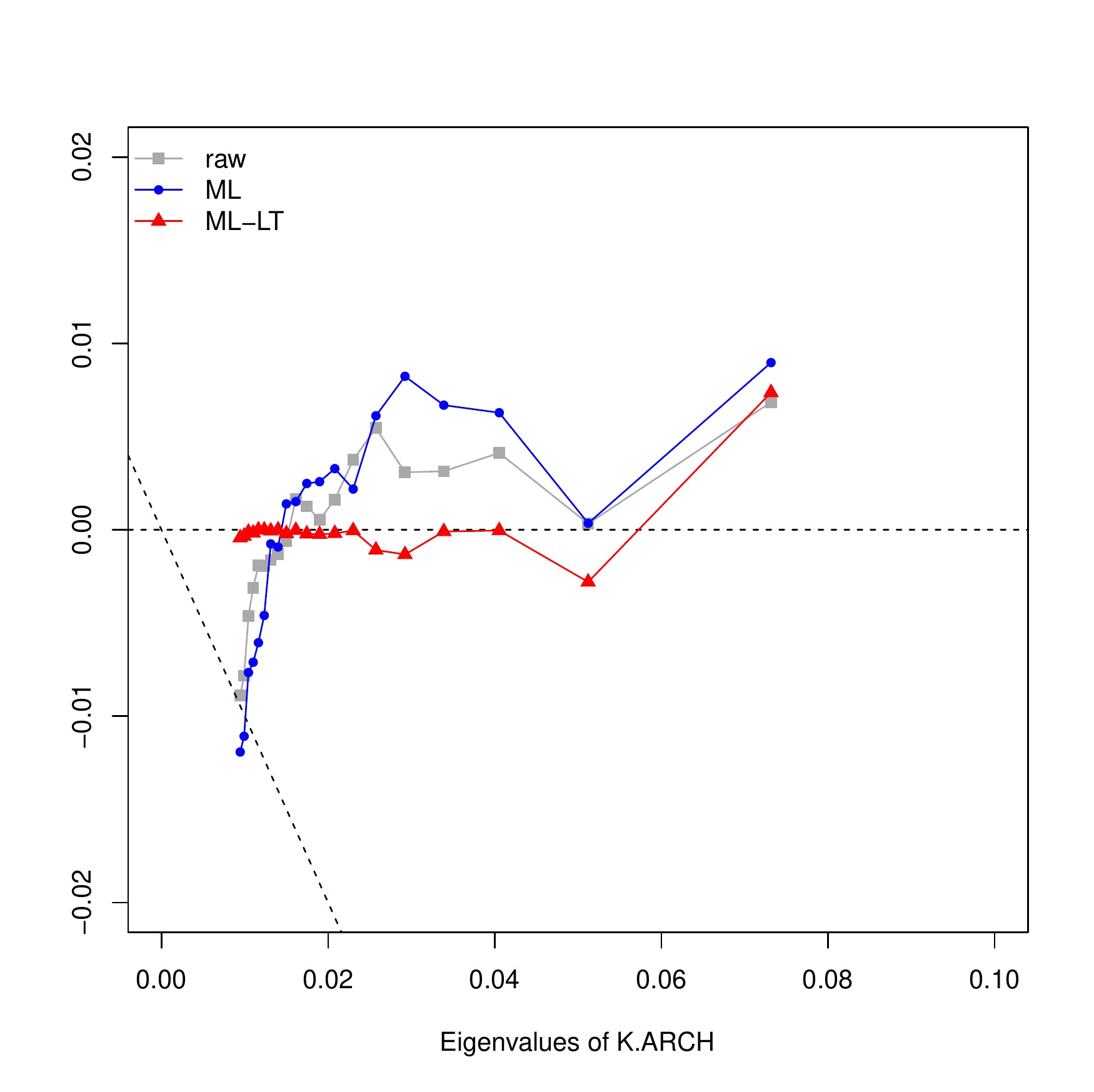} %,trim=-200 0 -200 0,clip
    }
    \hfill
    \subfigure[Structure of the first three and last three eigenvectors. 
	           Whatever the estimation method, the first eigenvector has a non trivial structure, with mostly positive components, 
               indicating a genuine departure from the diagonal ARCH benchmark, for which we would find a single peak at $\tau=1$.]{
	\includegraphics[scale=0.4]{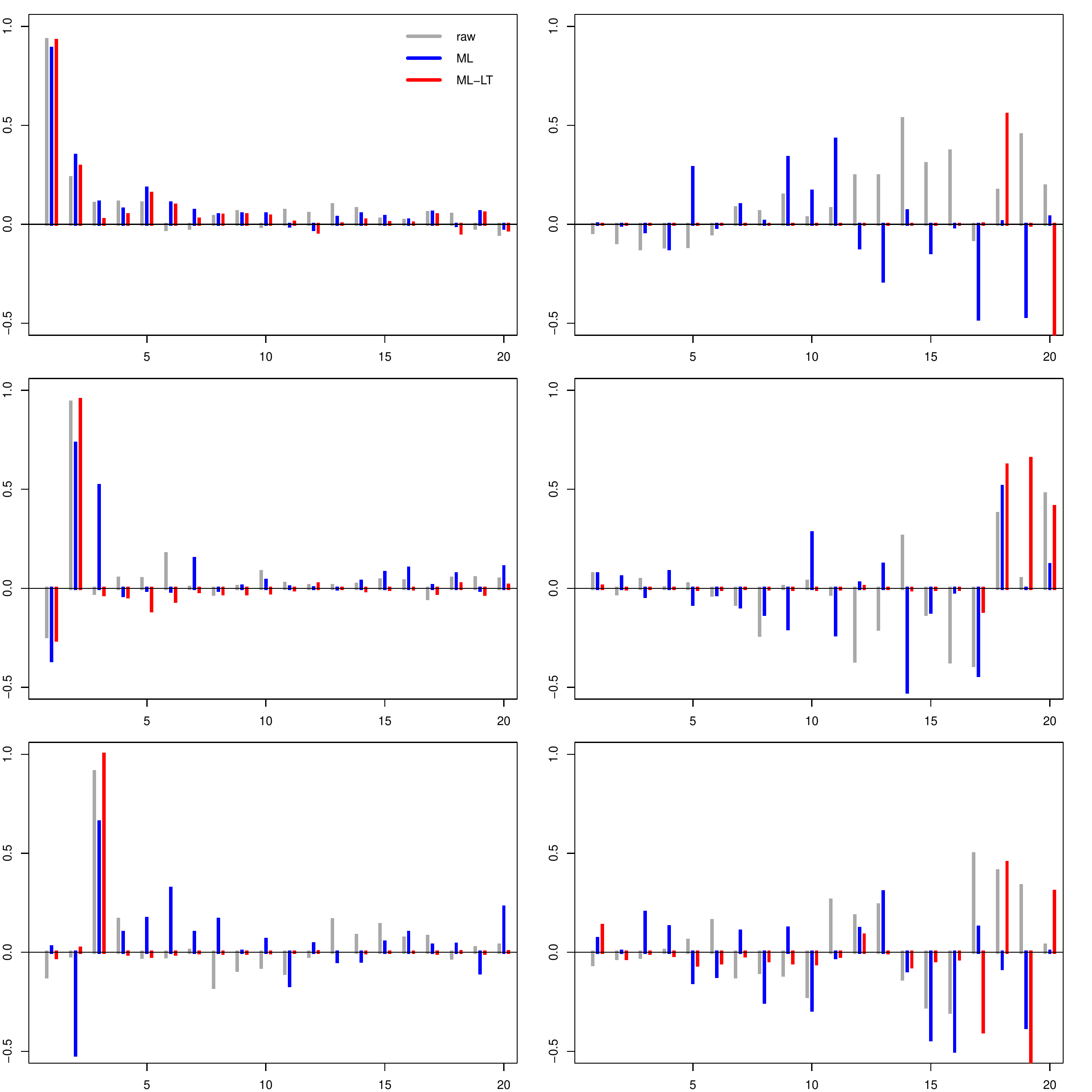} %,trim=-100 0 -100 0,clip
    }
	\caption{Spectral decomposition of the feedback kernel $K$, for the GMM, ML and ML+LT estimates.}
	\label{QARCHfig:eigen_stocks}
\end{sidewaysfigure}

As discussed in Sect.~\ref{QARCHsec:spectral}, another way to decipher the structure of $K$ is to look at its eigenvalues and eigenvectors. 
We show in Fig.~\ref{QARCHfig:eigen_stocks} the eigenvalues of $K^*$ as a function of the eigenvalues of the purely diagonal ARCH model. 
We see that a) the largest eigenvalue is clearly shifted upwards by the off-diagonal elements; 
            b) the structure of the top eigenvector is non-trivial, and has positive contributions at all lags (up to noise);
            c) the unconstrained estimations --- both GMM and ML --- lead to 6 very small eigenvalues (perhaps even slightly negative) 
               that all constrained models fail to reproduce.
               
The positiveness of all eigenvalues is not granted a priori, 
because nothing in the calibration procedure imposes the positivity of the matrix $K^*$. 
Although we would naively expect that past excitation could only lead to an amplification of future volatility
(i.e.\ that only strictly positive modes should appear in the feedback kernel), we observe that quasi-neutral modes do occur. 
This is clearly related to the negative streaks noted above at large lags, 
but we have no intuitive interpretation for this effect at this stage.

\section{Empirical study: stock index}\label{QARCHsect:emp_index}
 \begin{figure}
 	\center
 	\includegraphics*[scale=0.6,angle=-90]{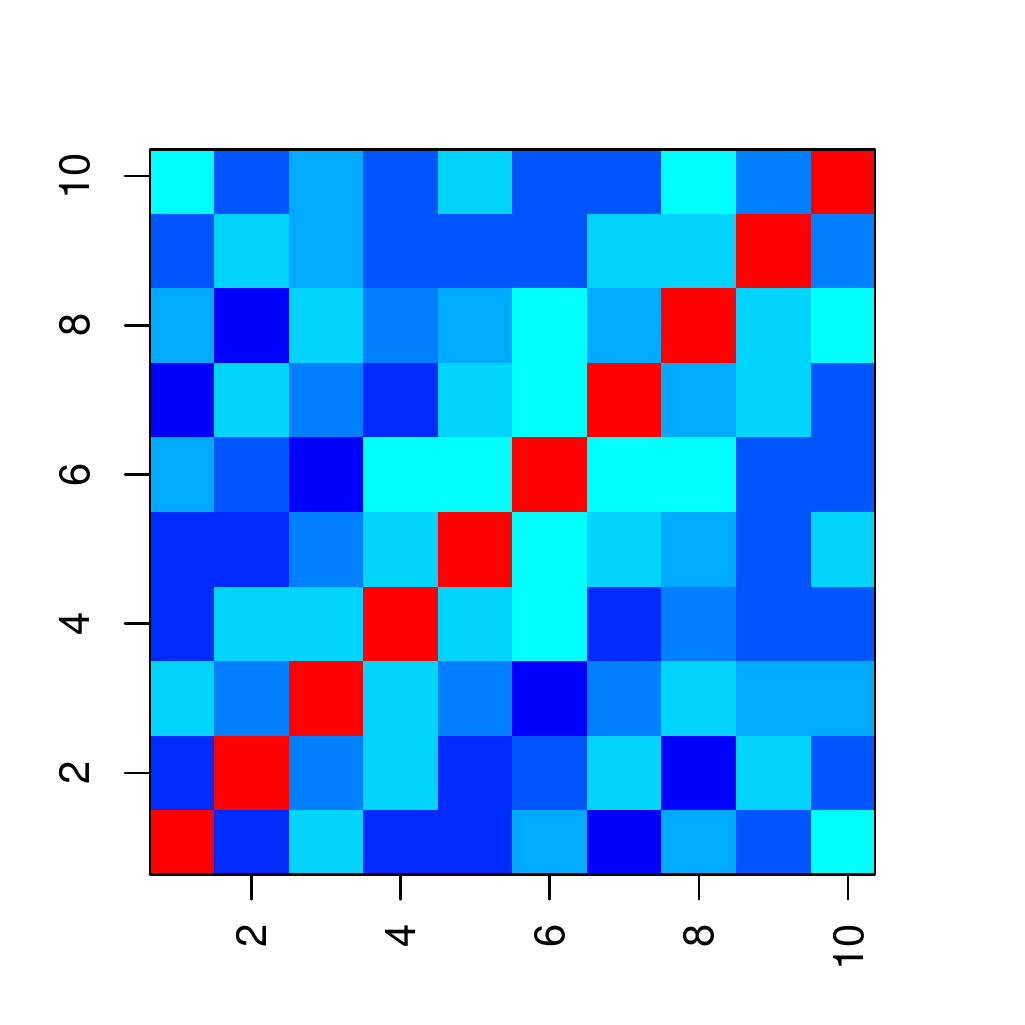}
 	\includegraphics*[scale=0.6,angle=-90]{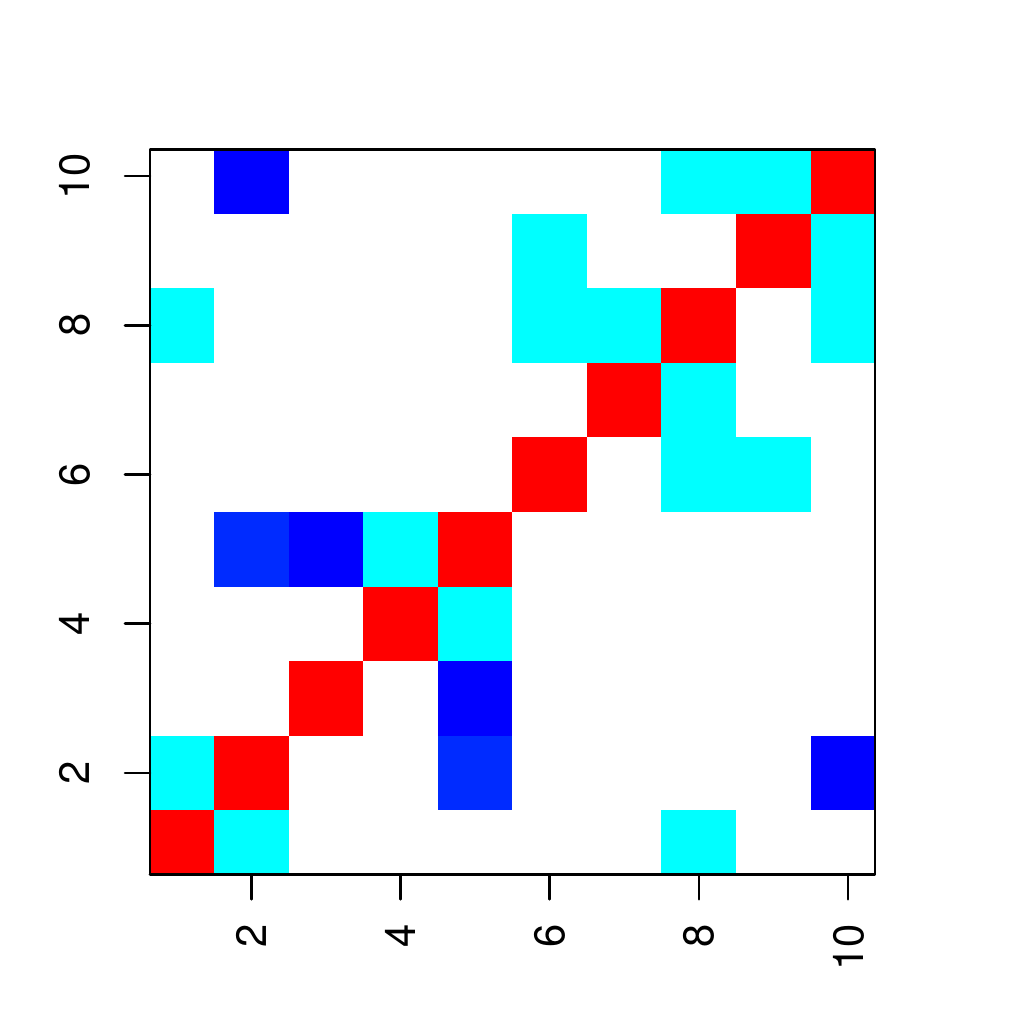}
 	\includegraphics*[scale=0.35,angle=-90]{QARCH/colkeys}
 	\caption{Heatmap of the $q=10$ kernel for the index volatility.
 	\textbf{Left: }the GMM estimation; \textbf{Right: }the ML estimation with GMM prior 
 	(again, we have checked that the ARCH prior leads to very close results).}
 	\label{QARCHfig:heatmap_idx}
 \end{figure}

We complete our analysis by a study of the returns of the  S\&P-500 index in the QARCH framework.
We use a long series of more than 60 years, from Oct.\ 1948 to Sept.\ 2011 ($15\,837$ days).

The computation of the correlation functions and the GMM calibration of a long ARCH(512) yields a $s^2(q)$ 
that can be fitted with the formula~\eqref{QARCHeq:fit_s2} and the following parameters: 
$s_{\infty}^2\approx0.20, \alpha\approx 1.28, g\approx 0.162$ and $q_0\approx262$.
Contrarily to single stocks, the One-step Maximum Likelihood estimation failed with $q=20$,
as the gradient and Hessian matrix evaluated at the arrival point are, respectively, large and not negative definite. 
Although the starting point appears to be close to a local maximum (the Hessian matrix is negative definite), 
the one-step procedure does not lead to that maximum.

In order to control better the Maximum Likelihood estimation, we lower the dimensionality of the parameter space 
and estimate a QARCH(10), although still with a long memory diagonal ($q=50$).
Here the procedure turns out to be legitimate, and the resulting kernel $K$ is depicted in Fig.~\ref{QARCHfig:heatmap_idx} (right). 
Interestingly, the off-diagonal content in the QARCH model for index returns is mostly not significant 
(again, white regions correspond to values not exceeding their theoretical uncertainty) 
apart from a handful of negative values around $\tau=8-10$ and one row/column at $\tau=5$. 
The contribution of the latter to the QARCH feedback can be made explicit as
\[
	2\,r_{t-5}\sum_{\tau<5}K(5,\tau)r_{t-\tau}
\]
and describes a trend effect between the daily return last week $r_{t-5}$ and the (weighted) cumulated return since then 
$\sum_{\tau<5}K(5,\tau)r_{t-\tau}$. 
It would be interesting to know whether this finding is supported by some intuitive feature of the trading activity on the S\&P-500 index. 
Note that, again, none of the ``simple'' structures discussed in Sect.~\ref{QARCHsec:section3} is able to account for the structure of $K^*$
(compare Fig.~\ref{QARCHfig:heatmap_idx} with Fig.~\ref{QARCHfig:matrices}).

The spectral analysis reveals a large eigenvalue much above the ARCH prediction, and a couple of very small eigenvalues,
similarly to what was found for the stock data.
However, the eigenvectors associated with them exhibit different patterns: 
the first eigenvector does not reveal the expected collective behavior, but rather a dominant $\tau=1$ component, 
with a significant $\tau=2$ component of opposite sign.
The other modes do not show a clear signature and are hard to interpret.

The procedure for computing in-sample/out-of-sample likelihoods is similar to the case of the stock data, 
but the definitions of the universes differs somewhat since we only have a single time series at our disposal. 
Instead of randomly selecting half of the series, we select half of the dates (in block, to avoid breaking the time dependences) 
to define the in-sample universe $\Omega$,
on which the correlation functions are computed and both estimation methods (GMM, and one-step maximum likelihood) are applied.
Then we evaluate the likelihoods of the calibrated kernels, first on $\Omega$ to obtain the `in-sample' likelihoods, 
and then on the complement of $\Omega$ to get the `out-of-sample' likelihoods.
We iterate $N_{\text{samp}}=150$ times and draw a random subset of dates every time, then average the likelihoods, 
that we report in the Table~\ref{QARCHtab:IS-OOS.idx}. 
The 1-s.d.\ systematic dispersion of the samplings is now $\approx 7 \cdot 10^{-3}$.
Because of the fewer number of observations in the index data compared to the stock data, 
corrections for the bias induced by the number of parameters $M$ become relevant. 
Adjusting the out-of-sample likelihood by subtracting the bias $-M/2n \approx 3 \cdot 10^{-3}$ (with $n\approx T/2=7.5\cdot 10^3$ and $M=q(q-1)/2=45$), 
brings the ARCH+ML result to a level competitive with ARCH (but not obviously better), and certainly better than the GMM estimate.

\begin{table}
\begin{center}
\begin{tabular}{c||cc|c|}
  &GMM        &ARCH(10)   &ARCH+ML     \\\hline\hline	%&ARCH+ML+veto
IS&$-1.16750$ &$-1.16666$ &${-1.16522}$\\		%&-1.16487
OS&$-1.16972$ &$-1.16704$ &${-1.17079}$\\\hline		%&-1.17066
\end{tabular}
\caption{Average log-likelihoods, according to Eq.~\eqref{QARCHeq:loglike}, for the stock index.
}
\label{QARCHtab:IS-OOS.idx}
\end{center}
\end{table}

\section{Time-reversal invariance}
\label{QARCHsect:TRI_0}

As noticed in the introduction, QARCH models violate, by construction, time-reversal invariance. 
Still, the correlation of the squared returns ${\mathcal{C}}^{(2)}(\tau)$ is trivially invariant 
under time-reversal, i.e.\  ${\mathcal{C}}^{(2)}(\tau)={\mathcal{C}}^{(2)}(-\tau)$. 
However, the correlation of the true squared volatility with past squared returns $\widetilde{\mathcal{C}}^{(2)}(\tau)$ 
is in general not (see \cite{pomeau1985symetrie,zumbach2009time} for a general discussion). 
A measure of TRI \nomenclature{TRI}{Time-reversal invariance} violations is therefore provided by the integrated difference $\Delta(\tau)$:
\be\label{QARCHeq:tra}
\Delta(\tau) = \sum_{\tau'=1}^\tau \left[\widetilde{\mathcal{C}}^{(2)}(\tau')- \widetilde{\mathcal{C}}^{(2)}(-\tau')\right].
\ee
The empirical determination of $\widetilde{\mathcal{C}}^{(2)}(\tau)$ and $\Delta(\tau)$ for stock returns is shown in Fig.~\ref{QARCHfig:TRI_emp}. 
Although less strong than for simulated data (see Fig.~\ref{QARCHfig:TRI}), we indeed find a clear signal of TRI violation for stock returns, 
in agreement with a related study by Zumbach~\cite{zumbach2009time}.
We compare in Fig.~\ref{QARCHfig:TRI} the quantity $\Delta(\tau)$ obtained from a {\it bona fide} numerical simulation of the model, 
with previously estimated parameters. 
Note that any measurement noise on the volatility $\sigma_t^2$ tends to reduce the TRI violations, 
but we have performed the numerical simulation in a way to reproduce this measurement noise as faithfully as possible.

However, the alert reader should worry that the existence of asymmetric leverage correlations ${\mathcal{L}}(\tau >0) \neq 0$ 
between past returns and future volatilities is in itself a TRI-violating mechanism, which has nothing to do with the ARCH feedback mechanism. 
In order to ascertain that the effect we observe is not a spurious consequence of the leverage effect, 
we have also computed the contribution of ${\mathcal{L}}(\tau)$ to $\Delta(\tau)$, which reads to lowest order and schematically:
\be
\Delta(\tau) = \sum\limits_{\tau'=1}^\tau L(\tau')\left[\mathcal{L}(\tau'\!-\!\tau)- \mathcal{L}(\tau'\!+\!\tau)\right] + K \,\,\text{contributions}.
\ee
The first term on the right-hand side is plotted in the inset of Fig.~\ref{QARCHfig:TRI_emp}, 
and is found to have a {\it negative} sign, and an amplitude much smaller than $\Delta(\tau)$ itself. 
It is therefore quite clear that the TRI-violation reported here is genuinely associated to the ARCH mechanism and not to the leverage effect, 
a conclusion that concurs with that of \cite{zumbach2009time}.

Still, the smallness of the empirical asymmetry compared with the simulation results suggests that the ARCH mechanism is ``too deterministic''.
It indeed seems reasonable to think that the baseline volatility $s^2$ has no reason to be constant, but may contain an extra random contribution.
Writing
\[
	\sigma^2(t)=\sigma_{\text{A}}^2(t) + \omega(t); \qquad \vev{\omega}=0; \qquad\vev{\omega(t)\omega(t-\tau)}\equiv\mathcal{C}_{\omega}(\tau)=\mathcal{C}_{\omega}(-\tau)
\]
with $\omega_t$ a noise contribution and $\sigma_{\text{A}}$ the ARCH volatility%
\footnote{For the sake of clarity we consider here a diagonal ARCH framework, 
          but the argument is straightforwardly generalized for a complete QARCH.}
(i.e.\ deterministic when conditioned on past returns), then the asymmetry is found to be given by:
\[
	\Delta(\tau)=\Delta_{\text{A}}(\tau)-\sum_{\tau''=1}^\tau\sum_{\tau'=1}^qk(\tau')\left[\mathcal{C}_{\omega}(\tau'-\tau'')-\mathcal{C}_{\omega}(\tau'+\tau'')\right].
\]
If one assumes that the correlation function $\mathcal{C}_{\omega}$ is positive and decays with time, 
the extra contribution to the asymmetry is negative, and reduces the observed TRI.
This conclusion speaks in favor of a mixed approach to volatility modeling, 
bringing together elements of autoregressive QARCH models with those of stochastic volatility models.
It would in fact be quite surprising that (although unobservable) the volatility should be a purely deterministic function of past returns.
Although the behavioral interpretation of the above construction is not clear at this stage, 
the uncertainty on the baseline volatility level $s^2$ could come, for example, 
from true exogenous factors that mix in with the volatility feedback component described by the QARCH framework.

\begin{figure}
	\center
	\includegraphics[scale=0.4]{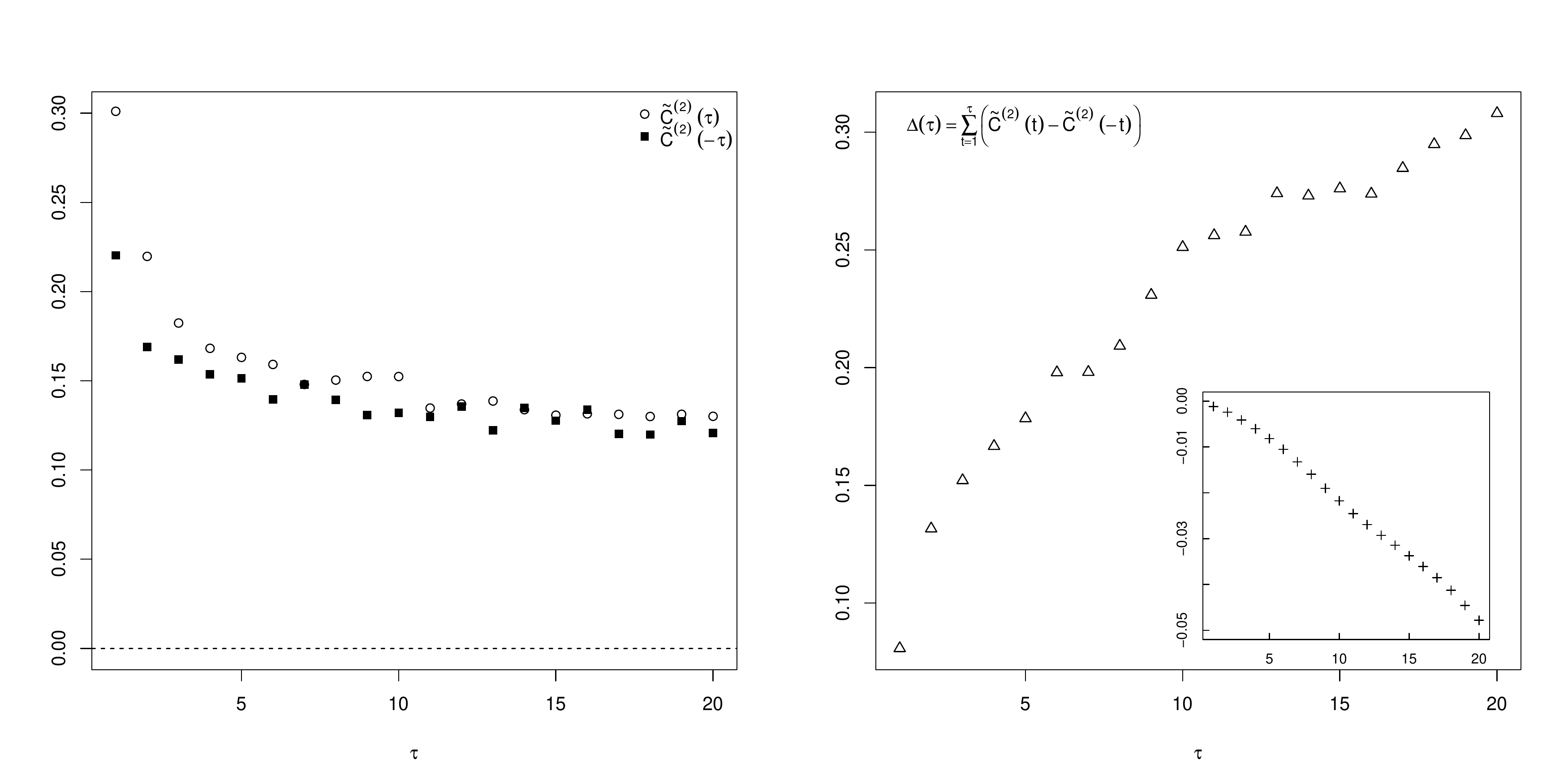}
	\caption{Measure of time-reversal asymmetry (Eq.~\ref{QARCHeq:tra}) for the stock data.
	         \textbf{Inset: } The contribution to $\Delta(\tau)$ stemming from the leverage, i.e.\ the quantity 
	         $\sum\limits_{\tau'=1}^\tau L(\tau')\left[\mathcal{L}(\tau'\!-\!\tau)- \mathcal{L}(\tau'\!+\!\tau)\right]$.
	         Note that this contribution is negative, and an order of magnitude smaller than $\Delta(\tau)$ itself.}
	\label{QARCHfig:TRI_emp}
%\end{figure}
%\begin{figure}
%	\center
	\includegraphics[scale=0.4]{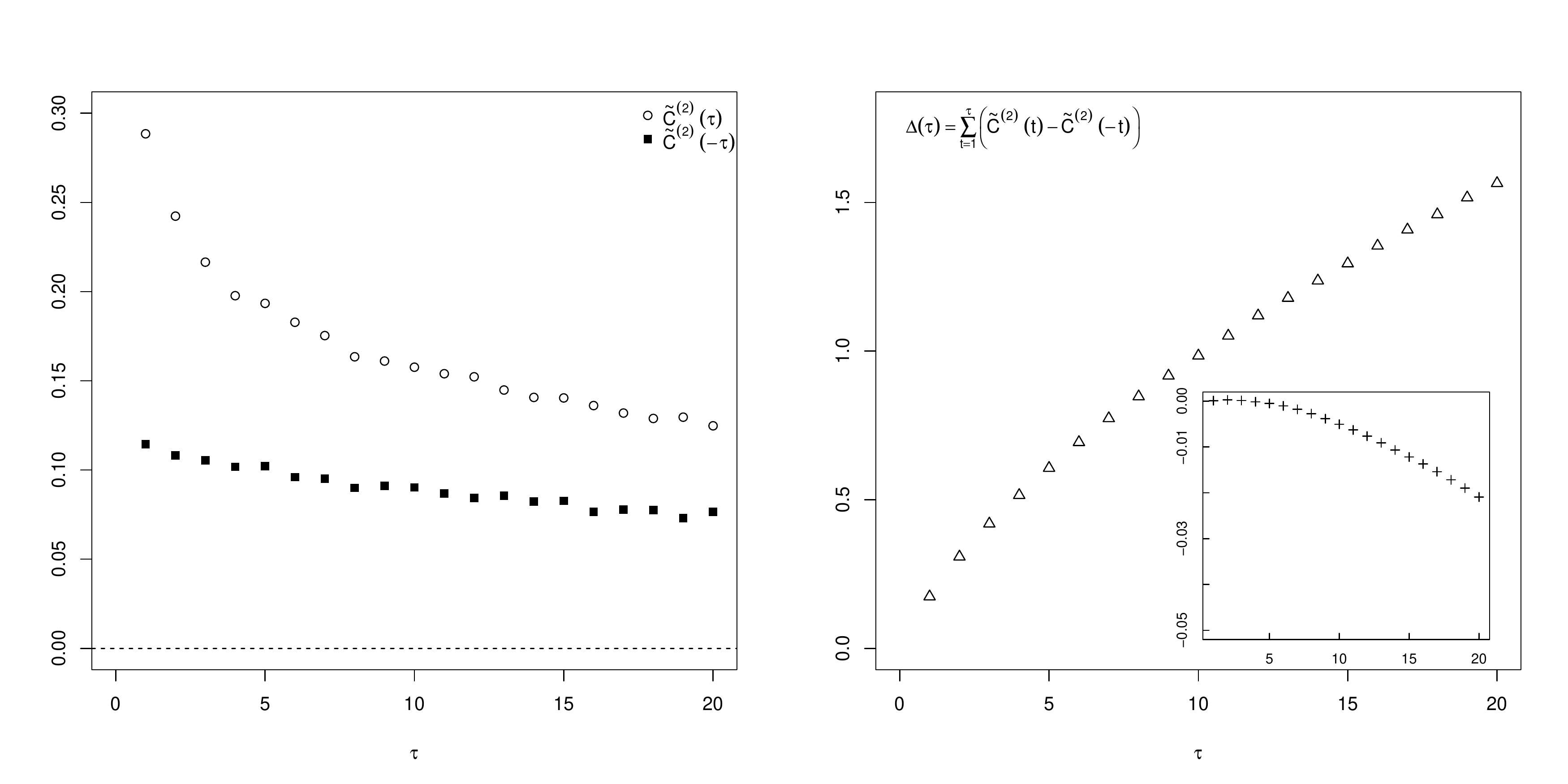}
	\caption{Measure of time-reversal asymmetry (Eq.~\ref{QARCHeq:tra}) for a simulated ARCH model with Student ($\nu=4$) residuals on the 5 minute scale.
	The parameters of the simulation are the estimated kernel $k^*(\tau)$ and $L^*(\tau)$ for stocks, with $q=20$.
	At each date, 100 intraday prices are simulated (corresponding to the number of 5-minutes bins inside 8 hours) with the {\it same} $\sigma_t^2$ 
	given by the QARCH model. The volatility is then computed using Rogers-Satchell's estimator, exactly as for empirical data.}
	\label{QARCHfig:TRI}
\end{figure}

\section{Conclusion, extensions}\label{QARCHsect:Concl}

We have revisited the QARCH model, which postulates that the volatility today can be expressed as a general quadratic form of the past daily returns $r_t$. 
The standard ARCH or GARCH framework is recovered when the quadratic kernel is diagonal, 
which means that only past squared daily returns feedback on today's volatility. 
This is a very restrictive {\it a priori} assumption, and the aim of the present study was
to unveil the possible influence of other quadratic combinations of past returns, such as, for example, square weekly returns. 
We have defined and studied several sub-families of QARCH models that make intuitive sense, 
and with a restricted number of free parameters. 

The calibration of these models on US stock returns has revealed several features that we did not anticipate. 
First, although the off-diagonal (non ARCH) coefficients of the quadratic kernel are found to be highly significant 
both in-sample and out-of-sample, they are one order of magnitude smaller than the diagonal elements. 
This confirms that daily returns indeed play a special role in the volatility feedback mechanism, as postulated by ARCH models. 
Returns on different time scales can be thought of as a perturbation of the main ARCH effect. 
The second surprise is that the structure of the quadratic kernel does not abide with any of the simpler QARCH models proposed so-far in the literature. 
The fine structure of volatility feedback is much more subtle than anticipated. 
In particular, neither the model proposed in \cite{borland2005multi} (where returns over several horizons play a special role), 
nor the trend model of Zumbach  in \cite{zumbach2010volatility} are supported by the data. 
The third surprise is that some off-diagonal coefficients of the kernel are found to be negative for large lags, 
meaning that some quadratic combinations of past returns contribute negatively to the volatility. 
This also shows up in the spectral properties of the kernel, which is found to have very small eigenvalues, 
suggesting the existence of unexpected volatility-neutral patterns of past returns.  

As for the diagonal part of the quadratic kernel, our results are fully in line with all past studies: 
the influence of the past squared-return $r_{t-\tau}^2$ on today's volatility $\sigma_t^2$ 
decays as a power-law $g\, \tau^{-\alpha}$ with the lag $\tau$, at least up to $\tau \approx 2$ months, 
with an exponent $\alpha$ close to unity ($\alpha \approx 1.11$), 
which is the critical value below which the volatility diverges and the model becomes non-stationary. 
As emphasized in \cite{borland2005multi}, markets seem to operate close to criticality 
(this was also noted in different contexts, see \cite{bouchaud2004fluctuations,bouchaud2006random, bacry2012non,toth2011anomalous, filimonov2012quantifying, parisi2012universality} for example). 
The smallness of $\alpha - 1$ has several important consequences: first, this leads to long-memory in the volatility; 
second, the average square volatility is a factor 5 higher than the baseline volatility, 
in line with the excess volatility story \cite{shiller1981stock}: 
most of the market volatility appears to be endogenous and comes from self-reflexive, 
feedback effects (see e.g.\ \cite{soros1994alchemy,bouchaud2011endogeneous,filimonov2012quantifying} and references therein). 
Third, somewhat paradoxically, the long memory nature of the kernel leads to \emph{small} fluctuations of the volatility. 
This is due to a self-averaging mechanism occurring in the feedback sum, that kills fat tails.
This means that the high kurtosis of the returns in ARCH models cannot be ascribed to volatility fluctuations 
but rather to leptokurtic residuals, also known as unexpected price jumps. 

Related to price jumps, we should add the following interesting remark that stresses the difference 
between endogenous jumps and exogenous jumps within the ARCH framework. 
Several studies have revealed that the volatility relaxes after a jump as a power-law, 
akin to Omori's law for earthquake aftershocks: 
$\sigma^2_{\tau}\sim \sigma^2_{0}\tau^{-\theta}$, where $t=0$ is the time of the jump. 
The value of the exponent $\theta$ seems to depend on the nature of the initial price jump. 
When the jump occurs because of an exogenous news, $\theta \approx 1$ \cite{lillo2003power,joulin2008stock}, 
whereas when the jump appears to be of endogenous origin, the value of $\theta$ falls around 
$\theta \approx \frac12$ \cite{zawadowski2006short,joulin2008stock}. 
In other words, as noted in \cite{joulin2008stock}, the volatility seems to resume its normal course faster after a news 
than when the jump seems to come from nowhere. 
A similar difference in the response to exogenous and endogenous shocks was also reported in \cite{sornette2004endogenous} for book sales. 
Now, if one simulates long histories of prices generated using an ARCH model with a diagonal kernel decaying as $g\, \tau^{-\alpha}$, 
one can measure the exponent $\theta$ by conditioning on a large price jump (which can only be endogenous, by definition!). 
One finds that $\theta$ varies continuously with the amplitude of the initial jump, 
and saturates around $\theta \approx \frac12$ for large jumps (we have not found a way to show this analytically). 
A similar behavior is found within multifractal models as well \cite{sornette2003what}. 
If on the other hand an exogenous jump is artificially introduced in the time series by imposing a very large value of $\xi_{t=0}$, 
one expects the volatility to follow the decay of the kernel and decay as $g\,\tau^{-\alpha} \xi_0^2$, leading to $\theta=\alpha \approx 1$. 
Therefore, the dichotomy between endogenous and exogenous shocks seem to be well reproduced within the ARCH framework. 

Finally, we have emphasized the fact that QARCH models are by construction backward looking, 
and predict clear Time-Reversal Invariance (TRI) violations for the volatility/square-return correlation function. 
Such violations are indeed observed empirically, although the magnitude of the effect is quite a lot smaller than predicted. 
This suggests that QARCH models, which postulate a {\it deterministic} relation between volatility and past returns, 
discard another important source of fluctuations. 
We postulate that ``the'' grand model should include both ARCH-like feedback effects and stochastic volatility effects, 
in such a way that TRI is only weakly broken. 
The stochastic volatility component could be the source of the extra kurtosis of the residuals 
noted above.%
\footnote{This discussion might be related to the interesting observation made by Virasoro in \cite{virasoro2011non}.}

This study is, to the best of our knowledge, the first attempt at unveiling the fine structure volatility feedback effects 
in autoregressive models. We believe that it is a step beyond the traditional econometric approach of postulating a convenient mathematical model,
which is then brute-force calibrated on empirical data. What we really need is to identify the {\it underlying mechanisms} that would justify, at 
a deeper level, the use of a QARCH family of model rather than any another one, for example the multifractal random walk model. 
From this point of view, we find remarkable that the influence of daily returns is so strongly singled out by our empirical results, when we 
expected that other time scales would emerge as well. The quandary lies in the unexpectedly complex structure of the off-diagonal 
feedback component, for which we have no interpretation. 

A natural extension of our study that should shed further light on a
possible behavioral interpretation of volatility feedback is to decompose daily returns into higher frequency components, 
for example overnight and intraday returns. 
Preliminary investigations in this direction point toward a clear distinction of the intraday and overnight mechanisms of volatility dynamics.%
\footnote{A detailed study of intraday and overnight effects in the feedback mechanisms of volatility dynamics will be reported elsewhere \cite{Pierre_inprep}}
In particular, the self-excitement at night is characterized by both a smaller baseline intensity and fatter-tailed residuals than during the day,
indicating that there are extreme exogenous shocks occurring overnight, which however are not propagated
since the dynamics results even more from the endogenous feedback at night than intraday.

The feedback mechanism of the volatility at the intraday scale (e.g.\ 5 minute returns) might be related to the self-excitement of order arrivals in the book \cite{bacry2012non}.
Many other remaining questions should be addressed empirically, for example 
the dependence of the feedback effects on market capitalization, average volatility, etc. We have chosen to scale out the market-wide volatility, 
but other choices would be possible, such as a double regression on the past returns of stocks and of the index. Finally, other financial assets, such 
as currencies or commodities, should be studied as well. Stocks, however, offer the advantage that the data is much more abundant, specially if 
one chooses to invoke some structural universality, and to treat all stocks as different realizations of the same underlying process.

%%% blinding
%\paragraph{Acknowledgements}
%We thank R. Allez, P. Blanc, M. Potters and M. Virasoro for useful discussions. 

\chapter{Copulas in time series}\label{chap:copulas_time}
\minitoc
  \newcommand{\Xq}{X^{(q)}}
\newcommand{\Xqq}{X^{(1\!-\!q)}}
\newcommand{\Xth}{X^{\mathrm{\scriptscriptstyle{(th)}}}}
\newcommand{\Xm}{X^{\mathrm{\scriptscriptstyle{(-)}}}}
\newcommand{\Xp}{X^{\mathrm{\scriptscriptstyle{(+)}}}}
\newcommand{\Xpm}{X^{\mathrm{\scriptscriptstyle{(\pm)}}}}

% \begin{abstract}
% We gather ideas on temporal dependences and recurrences in discrete time series 
% from several areas of natural and social sciences. 
% We revisit existing studies and redefine the relevant observables in the language of copulas.
% We argue that copulas is the correct framework to study non-linear time dependences and
% related concepts --- like aftershocks, Omori law, recurrences, waiting times ---
% and claim that previous phenomenological attempts involving only a long-ranged autocorrelation function
% lack complexity in that they are essentially mono-scale. IMPORTANT DE CITER \cite{perello2004multiple}.
% We show how these concepts can be dealt with in the context of continuous processes.
% \end{abstract}

\section{Introduction}

As a tool to study the --- possibly highly non-linear --- correlations between returns, ``copulas'' have long been used in actuarial 
sciences and finance to describe and model cross-dependences of assets, often in a risk management perspective \cite{embrechts2003modelling_art,embrechts2002correlation_art,malevergne2006extreme}.
Although the widespread use of simple analytical copulas to model multivariate dependences is more and more criticized \cite{mikosch2006copulas,chicheportiche2012joint}, 
copulas remain useful as a tool to investigate empirical properties of multivariate data \cite{chicheportiche2012joint}.

More recently, copulas have also been studied in the context of self-dependent univariate time series, 
where they find yet another application range: just as Pearson's $\rho$ coefficient is commonly used to 
measure both linear cross-dependences and temporal correlations, 
copulas are well-designed to assess non-linear dependences both across assets or in time 
\cite{beare2010copulas,ibragimov2008copulas,patton2009copula_art} --- we will speak of ``self-copulas'' in the latter case. 

\subsection*{Notations}
We consider a time series of length $T$, 
that can be seen as one realization of a corresponding discrete stochastic process $\{X_t\}_{t=1\ldots T}$.
The joint cumulative distribution function (CDF) of $n$ occurrences ($1\leq t_1<\ldots<t_n<T$) of the process is 
\begin{equation}\label{eq:def_multiF}
    \mathcal{F}_{t_1,\ldots,t_n}(\mathbf{x})=\pr{X_{t_1}<x_{t_1},\ldots,X_{t_n}<x_{t_n}}.
\end{equation}
We assume that the generating process $\{X_t\}$ of the series has a strong-stationary marginal CDF $F$,
and let the joint distribution $\mathcal{F}$ have long-ranged dependences. %, as is typically the case e.g.\ for seismic and financial data.
Strictly speaking, most of the discussion that follows only relies on \emph{time-homogeneity} of the copula,
i.e.\ translational invariance of the dependence structure, and does not need %in addition 
stationarity of the marginal.
However, in this class of problems much related to empirical data, it is more convenient 
(although sometimes not justified) to assume at least a weak form of stationarity, 
in order to give sense to the notion of empirical CDF.
In the following, we will always assume stationarity of the marginal and homogeneity of the copula.

A realization of $X_t$ at date $t$ will be called an ``event'' when its value exceeds a threshold $\Xpm$:
``negative event'' when $X_t<\Xm$, and ``positive event'' when $X_t>\Xp$.
%In the following we will consider a threshold defined as the upper \mbox{$p_+$-quantile} of the marginal distribution: $\Xp=F^{-1}(1\!-\!p_+)$.
The probability $p_-$ of such a `negative event' is $F(\Xm)$, and similarly, 
the probability that $X_t$ is above a threshold $\Xp$ is the tail probability $p_+=1-F(\Xp)$.

If a unique threshold $\Xp=\Xm$ is chosen, then obviously $p_+=1-p_-$.
This is appropriate when the distribution is one-sided, typically for positive only signals,
and one wishes to distinguish between two regimes: extreme events (above the unique threshold), and regular events (below the threshold).
This case is illustrated schematically in Fig.~\ref{fig:schema_a}.
When the distribution is two-sided, it is more convenient to define, $\Xp$ as the $q$-th quantile of $F$, 
and $\Xm$ as the $(1\!-\!q)$-th quantile, for a given $q\in[\tfrac{1}{2},1]$,
so that $p_+=p_-=1-q$. 
This allows to investigate persistence and reversion effects in \emph{signed} extreme events,
while excluding a neutral zone of regular events between $\Xm$ and $\Xp$, see Fig.~\ref{fig:schema_b}

%\textcolor{red}{
When the threshold for the recurrence is defined in terms of quantiles like above (a \emph{relative} threshold), 
stationarity is not needed theoretically but much wanted empirically as already said
(otherwise the height of the threshold might change every time). 
In contrast, when the threshold is set as a number (an \emph{absolute} threshold),
 there's no issue on the empirical side,
but the theoretical discussion makes sense only under stationary marginal.
%}

\begin{figure}
\center
    \subfigure[ \, $F(\Xp)=1-F(\Xm)=q$]{\label{fig:schema_b}\includegraphics[scale=0.9,trim=0   0 220 0,clip]{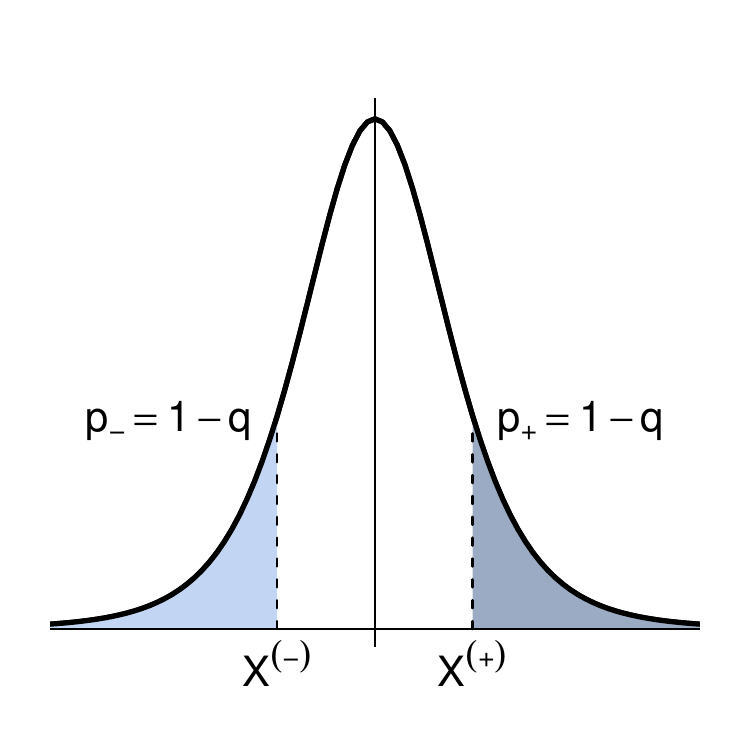}}\hfill
    \subfigure[ \, $F(\Xp)=  F(\Xm)=q$]{\label{fig:schema_a}\includegraphics[scale=0.9,trim=220 0   0 0,clip]{self-copulas/schema}}
    \caption{Two possible definitions of events: 
            either $p_-$ and $p_+$ are probabilities of extremes (negative and positive, respectively),
            or only $p_+$ is a probability of extreme and $p_-=1-p_+$.}
	\label{fig:schema}
\end{figure}

\section{Discrete time}
In this section, we consider the case where the discrete times $t_n$ 
in the definition \eqref{eq:def_multiF} are equidistant (``regularly sampled'').
A short discussion of the case with arbitrary (continuous) time stamps is provided in the next section.

\subsection{Two-points dependence measures}
Typical measures of dependences in stationary processes are two-points expectations
that only involve one parameter: the lag $\ell$ separating the points in time.
For example, the usefulness of the linear correlation function
\begin{align}\label{eq:def_rho_2points}
    \rho(\ell)&=\esp{X_tX_{t+\ell}}-\esp{X_t}\,\esp{X_{t+\ell}}%\\
    %\zeta(\ell)&=\esp{|X_t||X_{t+\ell}|}
\end{align}
is rooted in the analysis of Gaussian processes, as
those are completely characterized by their covariances, and
multi-linear correlations are reducible to all combinations of $2$-points expectations.
Some non-linear dependences, like the tail-dependence for example (defined in part~\ref{part:partI}), 
are however not expressed in terms of simple correlations, but involve the whole bivariate copula:
\begin{equation}
    \cop[\ell](u,v)=\mathcal{F}_{t,t+\ell}(F^{-1}(u),F^{-1}(v)),
\end{equation}
where $(u,v)\in[0,1]^2$.
${C}_\ell$ contains the full information on bivariate dependence that is invariant under increasing transformations of the marginal.
For example, the conditional probability
\[
    p_{++}^{(\ell)}=\pr{X_{t+\ell}>\Xp|X_t>\Xp},
\]
which is a measure of \emph{persistence} of the ``positive'' events, %and
can be written in terms of copulas, 
together with all three other cases of conditioning
\begin{subequations}\label{eq:allp}
\begin{align}
     p_{++}^{(\ell)}&=[{2p_+-1+\cop[\ell](1\!-\!p_+,1\!-\!p_+)}]/{p_+},\\
     p_{--}^{(\ell)}&= {\cop[\ell](p_-,p_-)}/{p_-},\\
     p_{-+}^{(\ell)}&=[{p_--\cop[\ell](p_-,1\!-\!p_+)}]/{p_-},\\
     p_{+-}^{(\ell)}&=[{p_--\cop[\ell](1\!-\!p_+,p_-)}]/{p_+}.
\end{align}
\end{subequations}
When $\Xp=\Xm=0$ and $\ell=1$, this is exactly the definition of Bog\`una and Masoliver in Ref.~\citet{boguna2004conditional}, 
with accordingly $p_-=p_+=F(0)$, see Fig.~\ref{fig:schema}.
Note also that $p_{\pm\pm}^{(\ell)}$ and $p_{\pm\mp}^{(\ell)}$ are straightforwardly related to 
the tail dependence coefficients through the relations~\eqref{eq:copdiff} of page~\pageref{eq:copdiff}.

In addition to caring for frequencies of conditional events, 
one can try to characterize their magnitude.
For a single random variable with cdf $F$, this can be quantified for example by the 
Expected Shortfall (or Tail Conditional Expectation), i.e.\ the average loss when conditioning on large events:
\[
    \text{ES}(p_-)=\esp{X_t|X_t<\Xm }=\frac{1}{p_-}\int_{-\infty}^{F^{-1}(p_-)}x\,\d{F(x)}
                                     =\frac{1}{p_-}\int_0^{p_-}F^{-1}(p)\,\d{p}
\]
In the same spirit, for bivariate distributions, the mean return conditionally on preceding return `sign' is defined:
\begin{subequations}\label{eq:condES}
\begin{align}
    \vev{X}^{(\ell)}_+&= \esp{X_t|X_{t-\ell}>\Xp }\\
    \vev{X}^{(\ell)}_-&= \esp{X_t|X_{t-\ell}<\Xm }.%\\
%    \operatorname{med}(X)_\pm&=F^{-1}(\alpha_\pm),
\end{align}
\end{subequations}
%where
%\[
%    \left\{\begin{array}{rl}
%        \cop[\ell](p_-,\alpha_-)&=p_-/2\\
%        \alpha_+-\cop[\ell](1-p_+,\alpha_+)&=p_+/2
%    \end{array}\right.
%\]
%
As an example, consider the Gaussian bivariate copula of the pair $(X_t,X_{t+\ell})$,
whose whole $\ell$-dependence is in the linear correlation coefficient $\rho(\ell)$.
Fig.~\ref{fig:condprobGauss} illustrates the conditional probabilities \eqref{eq:allp} as a function of the threshold when $p_+=p_-=1-q$,
and Fig.~\ref{fig:condESGauss} shows the conditional Expected Shortfall that can be computed exactly from Eqs.~\eqref{eq:condES},
and is proportional to the inverse Mill ratio:
\[
    \vev{X}_\pm=\pm\rho(\ell)\frac{\Phi'(\Xpm)}{p_\pm},
   %\vev{X}_\pm=\pm\frac{\rho(\ell)}{\sqrt{1-\rho(\ell)^2}}\frac{\Phi'(\Xpm)}{p_\pm},
               %=\pm\frac{\rho}{\sqrt{1-\rho^2}}\frac{d\ln \Phi}{dx}(\Xpm)
               %=\pm\frac{\rho}{\sqrt{1-\rho^2}}\frac{1}{p_\pm \cdot (\Phi^{-1})'(p_{\pm})}
\]
where $\Phi$ denotes the CDF of the univariate standard normal distribution.
%%% code Mathematica: 1/(p*D[InverseCDF[NormalDistribution[0,1],p]])
%%% code Mathematica: D[Log[CDF[NormalDistribution[0,1],x]]]

\begin{figure}
\center
    \subfigure[Gaussian copula]                  {\label{fig:condprobGauss} \includegraphics[scale=.55,trim=30 0 25 0,clip]{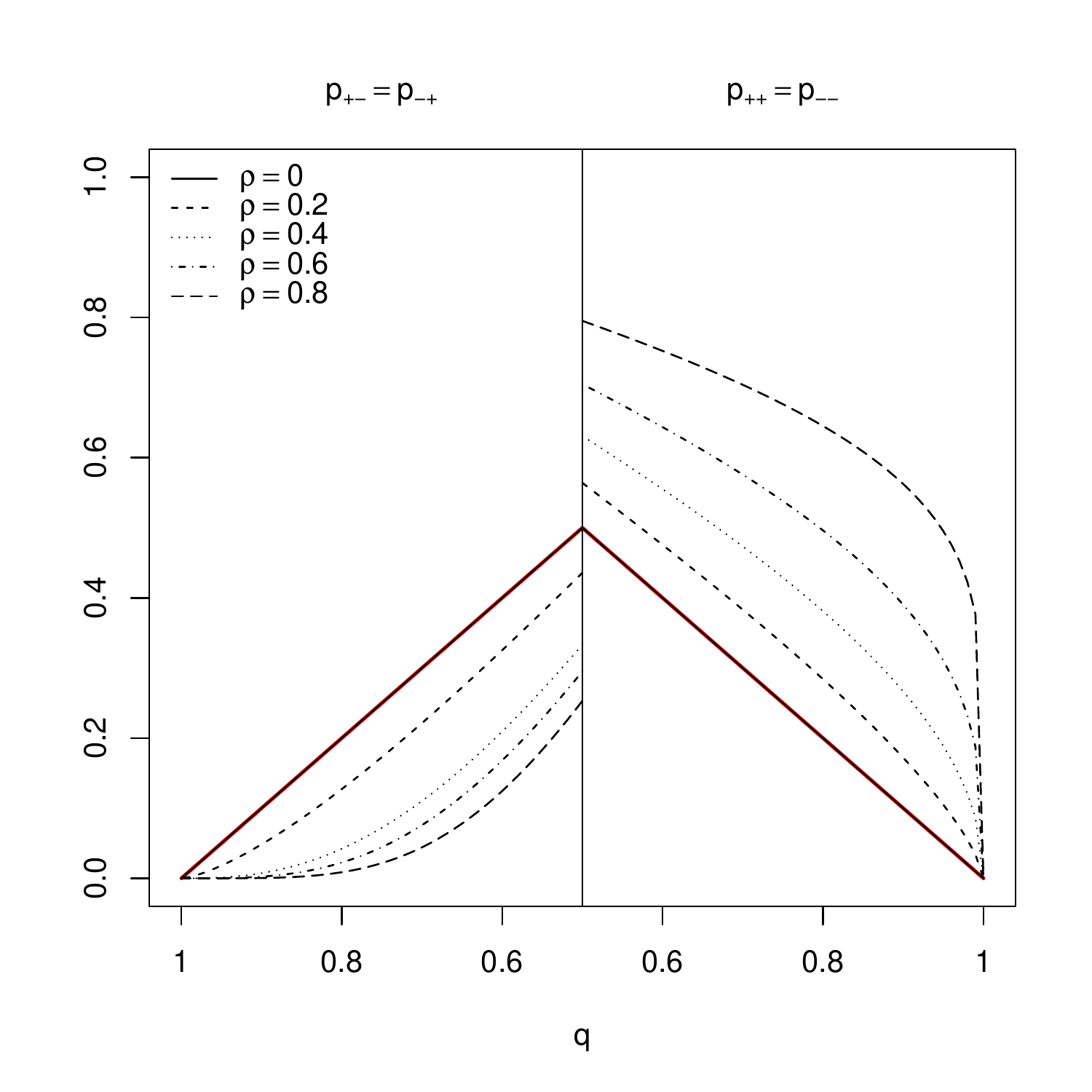}}\hfill
    \subfigure[Student copula with $\nu=5$ 
    (see also Fig.~\ref{fig:dev_gauss_stud} 
           on page~\pageref{fig:dev_gauss_stud})]{\label{fig:condprobStud5} \includegraphics[scale=.55,trim=30 0 25 0,clip]{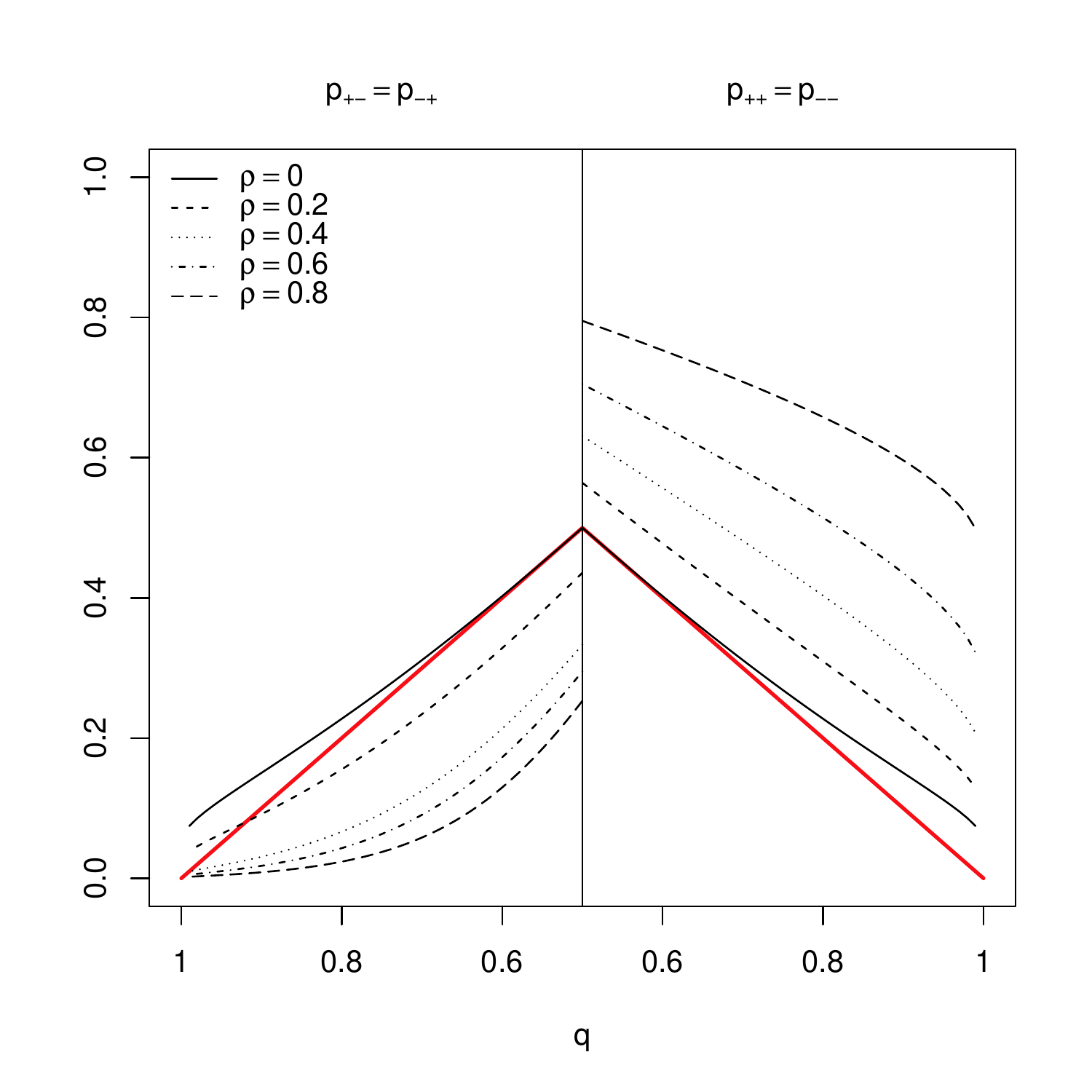}}
    \caption{Conditional probabilities $p_{\pm\mp}^{(\ell)}$ and $p_{\pm\pm}^{(\ell)}$ for different values of $\rho(\ell)$ with thresholds at $p_+=p_-=1-q$.
             The value at $q=0.5$ is $\tfrac{1}{2}+\tfrac{1}{\pi}\arcsin \rho(\ell)$.}\label{fig:condprob}
    \vfill
    \includegraphics[scale=.6]{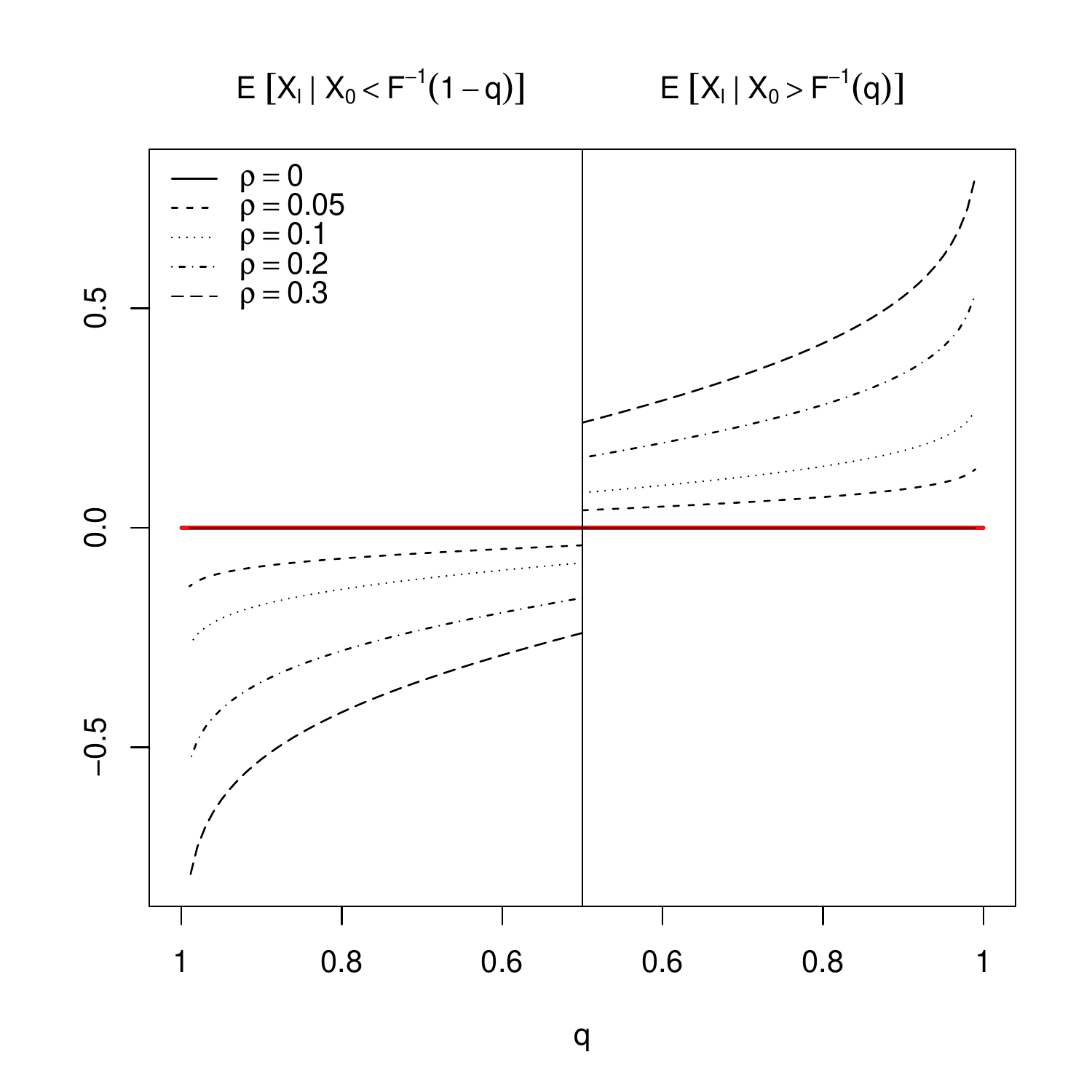}
    \caption{Conditional Expected Shortfall of a Gaussian pair $(X_0,X_\ell)$ for different values of $\rho(\ell)$.
             The value at $q=0.5$ is $\sqrt{\tfrac{2}{\pi}}\,\rho(\ell)$.}
    \label{fig:condESGauss}
\end{figure}

\subsubsection*{Aftershocks}

Omori's law characterizes the $\ell$ dependence of $p_{++}^{(\ell)}$,
i.e.\ the average frequency of events occurring $\ell$ time steps after a main event.
It has been stated first in the context of earthquakes occurrences \cite{omori1895}, where this time dependence is power-law:
\begin{equation}\label{eq:omori}
    p_{++}^{(\ell)}= \lambda \cdot \ell^{-\alpha}.
\end{equation}
Notice that any dependence on the threshold must be hidden in $\lambda$ according to this description.
%The average \emph{cumulated} number $N_\ell$ of events having occurred $\ell$ time steps after a main event is thus 
The average \emph{cumulated} number $N_\ell$ of these aftershocks until $\ell$ is thus 
\[
    \vev{N_\ell}_+=\lambda \cdot {\ell^{1-\alpha}}/({1\!-\!\alpha}),
\]
with in fact $\lambda\equiv p_+$ since, when $\alpha\to 0$, $N_{\ell}$ has no time-dependence i.e.\ it counts independent events (white noise),
and $p_{++}^{(\ell)}$ must thus tend to the unconditional probability.

%\citet{lillo2003power} describe the joint probability as $p_{++}^{(\ell)}\sim 1-\phi(\Xp/\sigma(\ell))$, which
%they justify by a \emph{non-stationary} model of aftershock amplitudes with decaying scale $\sigma(\ell)$ and independent stochastic amplitude $r_\ell$ with CDF $\phi$,
%Finance \cite{PhysRevE.76.016109,PhysRevE.82.036114}.
In order to give a phenomenological grounding to this empirical law also later observed in finance \cite{PhysRevE.76.016109,PhysRevE.82.036114},
the authors of Ref.~\citet{lillo2003power} model the aftershock volatilities in financial markets as a decaying scale $\sigma(\ell)$ times an independent stochastic amplitude $r_\ell$ with CDF $\phi$.
As a consequence, $p_{++}^{(\ell)}\sim 1-\phi(\Xp/\sigma(\ell))$ and the power-law behavior of Omori's law results from 
(i) power-law marginal $\phi(r)\sim r^{-\gamma}$, {and}
(ii) scale decaying as power-law $\sigma(t)\sim t^{-\beta}$, 
so that relation \eqref{eq:omori} is recovered with $\alpha=\beta\gamma$.
%But this must be understood in a \emph{conditional} sense:
%the marginal distribution \emph{is} stationary, 
%but the conditioning on the previous event, $\ell$ lags earlier, introduces a time dependence in the scale !
The non-stationarity described by $\sigma$ is only introduced in a \emph{conditional} sense,
and might be appropriate for aging systems or financial markets, but we believe that
%Although the non-stationarity described by $\sigma$ is only introduced in a \emph{conditional} sense,
%we believe that this construction lacks intuitive micro-foundations, and that 
Omori's law  can be accounted for in a stationary setting and without necessarily having power-law distributed amplitudes.

{The scaling of $p_{++}^{(\ell)}$ with the magnitude of the main shock is encoded in the prefactor $\lambda\equiv p_+$,
which, for example, accounts for the exponentially distributed magnitudes of earthquakes (Gutenberg-Richter law \cite{Gutenberg21021936}).}
The linear dependence of $p_{++}^{(\ell)}$ on $p_+$ shall be reflected in the diagonal of the underlying copula:
\[
    \cop[\ell](p,p)=p^2 \,\ell^{-\alpha},
\]
a prediction that can be tested empirically.

Note that Omori's law is a measure involving only the two-points probability.
In the next subsection, we show what additional information many-points probability can reflect.

\subsection{Multi-points dependence measures}
Although the $n$-points expectations of Gaussian processes reduce to all combinations of $2$-points expectations \eqref{eq:def_rho_2points},
% \[
    % \esp{X_tX_{t'}}\equiv\rho(|t-t'|),
% \]
{their full dependence structure is not reducible to the bivariate distribution}, 
unless the process is also Markovian (i.e.\ only in the particular case of exponential correlation).
Furthermore, when the process is not Gaussian, even the multi-linear correlations are irreducible.
In the general case, the whole multivariate CDF is needed, but many measures of dependence that we introduce below only 
involve the diagonal $n$-points copula:%
\footnote{We use a calligraphic $\mathcal{C}$ in order to make it clearly distinct from the bivariate copula discussed in the previous section.}
\begin{equation}\label{eq:def_cop_n}
    \CopN[n]{p}=\mathcal{F}_{t+\libracket 1,n\ribracket}(F^{-1}(p),\ldots,F^{-1}(p)),
\end{equation}
which measures the joint probability that all $n\geq 1$ consecutive variables $X_{t+1},\ldots,X_{t+n}$ are below 
the upper \mbox{$p$-th} quantile of the stationary distribution 
($p\in[0,1]$, and $t+\libracket 1,n\ribracket$ is a shorthand for $\{t\!+\!1,\ldots,t\!+\!n\}$).
Clearly, $\CopN[1]{p}=p$ and we set by convention $\CopN[0]{p}\equiv 1$.

As an example, the Gaussian diagonal copula is
\begin{equation}\label{eq:GaussCop}
    \CopN[n]{p}=\Phi_\rho\big(\Phi^{-1}(p),\ldots,\Phi^{-1}(p)\big)
\end{equation}
where $\Phi^{-1}$ is the univariate inverse CDF, and $\Phi_\rho$
denotes the multivariate CDF with $(n\times n)$ covariance matrix $\rho$,
which is Toeplitz with symmetric entries
\[
    \rho_{tt'}=\rho(|t-t'|),\quad t,t'=1,\ldots,n.
\]
The White Noise (WN)\nomenclature{WN}{White noise} product copula
 $\CopN[n]{p}=p^n$ is recovered in the limit of vanishing correlations $\rho(\ell)=0\, \forall\ell$, and
other examples include the exponentially correlated Markovian Gaussian noise, 
the power-law correlated (thus scale-free) fractional Gaussian noise, %with Hurst exponent $H$ ,
and the logarithmically correlated multi-fractal Gaussian noise.

Empirically, the $n$-points probabilities are very hard to measure due to the large noise associated
with such rare joint occurrences.
However, there exist observables that embed many-points properties and are more easily
measured, such as the length of sequences (clusters) of thresholded events,
and the recurrence times of such events, that we study next.

\subsubsection*{Recurrence intervals}
The probability $\pi(\tau)$ of observing a recurrence interval $\tau$ between two events 
is the conditional probability of observing a sequence of $\tau-1$ ``non-events'' bordered by two events:
\[
    \pi(\tau)=\pr{X_{\tau}>\Xp, X_{\libracket 1;\tau\libracket}<\Xp|X_{0}>\Xp}.
\]
(We focus on positive events, but the recurrence of negative events can be studied with the substitution $X\to -X$,
and the case of recurrence in amplitudes with the substitution $X\to |X|$).
After a simple algebraic transformation flipping all `$>$' signs to `$<$', 
it can be written in the language of copulas as:
\begin{align}\label{eq:distrecint}\nonumber
        \pi(\tau)&=\frac{\pr{X_{\tau}>\Xp, X_{\libracket 1;\tau\libracket}<\Xp; X_{0}>\Xp}}{\pr{X_{0}>\Xp}}\\\nonumber
                 &=\frac{\pr{X_{\libracket 1;\tau\libracket}<\Xp}}{p_+}
                  -\frac{\pr{X_{\tau}<\Xp, X_{\libracket 1;\tau\libracket}<\Xp}}{p_+}\\\nonumber
                 &-\frac{\pr{X_{\libracket 1;\tau\libracket}<\Xp; X_{0}<\Xp}}{p_+}
                  +\frac{\pr{X_{\tau}<\Xp, X_{\libracket 1;\tau\libracket}<\Xp; X_{0}<\Xp}}{p_+}\\
        \pi(\tau)&=\frac{\CopN[\tau-1]{1\!-\! p_+}-2\,\CopN[\tau]{1\!-\! p_+}+\CopN[\tau+1]{1\!-\! p_+}}{p_+}.
\end{align}
The cumulative distribution
\begin{equation}\label{eq:cumul_Pi}
    \Pi(\tau)=\sum_{n=1}^\tau \pi(n)=1-\frac{\CopN[\tau]{1\!-\! p_+}-\CopN[\tau+1]{1\!-\! p_+}}{p_+}
\end{equation}
is more appropriate for empirical purposes, being less sensitive to noise.
These exact expressions make clear --- almost straight from the definition --- that 
  (i)~the distribution of recurrence times \emph{depends only on the copula} of the underlying process 
      and not on the stationary law, in particular its domain or its tails 
      (this is because we take a relative definition of the threshold as a quantile); 
 (ii)~\emph{non-linear} dependences are highly relevant in the statistics of recurrences, so that 
linear correlations can in the general case by no means explain alone the properties of $\pi(\tau)$ \cite{altmann2005recurrence}; and 
(iii)~recurrence intervals have a \emph{long memory} revealed by the $(\tau\!+\!1)$-points copula being involved, 
so that only when the underlying process $X_t$ is Markovian will the recurrences themselves be memoryless.%
\footnote{It may be mentioned that in a non-stationary context, renewal processes are also able to produce independent consecutive recurrences \cite{PhysRevE.78.051113,Sazuka20092839}.}
 Hence, when the copula is known (Eq.~\eqref{eq:GaussCop} below for Gaussian processes), 
 the distribution of recurrence times is exactly characterized by the analytical expression in Eq.~\eqref{eq:distrecint}.

The average recurrence time is found straightforwardly by summing the series
\begin{equation}\label{eq:avrecint}
{
    \mu_\pi=\vev{\tau}=\sum_{\tau=1}^{\infty}\tau\,\pi(\tau)=\frac{1}{p_+},
}
\end{equation}
and is \emph{universal} whatever the dependence structure.
This result was first stated and proven by Kac in a similar fashion~\cite{kac1947notion}. 
It is intuitive as, for a given threshold, the whole time series is the succession of a fixed number $p_+T$ of recurrences 
whose lengths $\tau_i$ necessarily add up to the total size $T$, so that $\vev{\tau}=\sum_i\tau_i/(p_+T)=1/p_+$.
Note that Eq.~\eqref{eq:avrecint} assumes an infinite range for the possible lags $\tau$, which is achieved either 
by having an infinitely long time series, or more practically when the translational-invariant copula is periodic at the boundaries of the time series, as
is typically the case for artificial data which are simulated using numerical Fourier Transform methods.
Introducing the copula allows to emphasize the validity of the statement even in the presence of non-linear long-term dependences,
as Eq.~\eqref{eq:avrecint} means that the average recurrence interval is \emph{copula}-independent.

More generally, the $m$-th moment can be computed as well by summing $\tau^m \pi(\tau)$ over $\tau$:
\begin{equation}
    \vev{\tau^m}=\frac{1}{p_+}\left[1+\sum\limits_{\tau=1}^\infty\left(|\tau\!+\!1|^m-2\tau^m+|\tau\!-\!1|^m\right)\copN[\tau](1\!-\!p_+)\right].
\end{equation}
In particular, the variance of the distribution is
\begin{equation}\label{eq:secmon}
    \sigma_\pi^2\equiv\vev{\tau^2}-\mu_\pi^2=\frac{2}{p_+}\sum_{\tau=1}^{\infty}\CopN[\tau]{1\!-\!p_+}-\frac{1\!-\!p_+}{p_+^2},
\end{equation}
It is not universal, in contrast with the mean, 
and can be related to the average unconditional waiting time, see below.
Notice that in the independent case the variance $\sigma_\pi^2=(1-p_+)/p_+^2$ is not equal to the mean $\mu_\pi=1/p_+$
 (as would be the case for a continuous-time Poisson process) because of discreteness effects.

\subsubsection*{Waiting times}

The conditional mean residual time to next event, when sitting $\tau$ time steps after a (positive) event, is
\begin{equation}\label{eq:residual_time}
    \vev{w|\tau}=\sum_{w=1}^\infty w\, \pi(\tau\!+\!w)=\frac{1}{p_+}\copN[\tau](1\!-\!p_+).
\end{equation}

One is often more concerned with unconditional waiting times,
which is equivalent to asking what the size of a sequence of $w$ `non-events' starting now will be,
\emph{regardless of what happened previously}. 
The distribution of these waiting times is
\[
    \phi(w)={\copN[w](1\!-\!p_+)-\copN[w+1](1\!-\!p_+)},%{p_+},
\]
and its expected value is %similar to the average sequence length
\[
    \mu_{\phi}=\vev{w}=\sum_{w=1}^{\infty}\copN[w](1\!-\!p_+),%\frac{1}{p_+}
\]
consistently to what would be obtained by averaging $\vev{w|\tau}$  over
all possible elapsed times in Eq.~\eqref{eq:residual_time}.
%
%With $p_-=1-p_+=q$, we get the following relationship between the variance of the distribution $\pi_+$
%of recurrence intervals, and the mean sequence size:
%\[
%    \vev{\tau^2}_+-\vev{\tau}_+^2=\frac{q}{1-q}\left(2\vev{n}_--\vev{\tau}_+\right)
%\]
From Eq.~\eqref{eq:secmon}, we have the following relationship between the variance of the distribution $\pi$
of recurrence intervals, and the mean waiting time:
\begin{equation}
   %\vev{\tau^2}_+-\vev{\tau}_+^2=2\vev{w}_+ - \frac{1-p_+}{p_+^2}
   %\vev{\tau^2}_+-\vev{\tau}_+^2=\frac{1}{p_+}\left(2\vev{w}_+ +1 - \frac{1}{p_+}\right)
   %\vev{\tau^2}_+=\vev{\tau}_+\,\big[2\vev{w}_+ +1 \big]
    \sigma_\pi^2=\mu_\pi\,\big[2\mu_{\phi} +1 \big]-\mu_\pi^2
\end{equation}

\subsubsection*{Sequences lengths}
%Clustering: the distribution of the sequence sizes reveals the correlation of recurrence times.
The dependence in the process is also revealed by the distribution of sequences sizes.
The probability that a sequence of consecutive negative events%
\footnote{We consider the case of ``negative'' events, i.e.\ those with $X_t<\Xm$
because it expresses simply in terms of diagonal copulas. 
The mirror case with ``positive'' events has the exact same expression but $\copN[n]$ must be inverted around the median.
For a symmetric $F$, this distinction is irrelevant.}, starting just after a `non-event',
will have a size $n$ is
\begin{align}\nonumber%\label{eq:seq_size}
    \psi(n)%&=\pr{X_{t+2}<\Xm,\ldots,X_{t+n}<\Xm,X_{t+n+1}>\Xm|X_{t+1}<\Xm,X_{t}>\Xm}\\
               =\frac{\copN[n](p_-)-2\,\copN[n+1](p_-)+\copN[n+2](p_-)}{p_-\,(1-p_-)}
\end{align}
and the average length of a random sequence 
\begin{equation}
{
   %\vev{n}_-=\sum_{n=1}^{T}n\,\psi_-(n)\xrightarrow{T\to\infty}\frac{1}{p_-}\sum_{n=1}^{\infty}\copN[n](p_-).
   %\vev{n}_-=\sum_{n=1}^{\infty}n\,\psi_-(n)=\frac{1}{p_-}\sum_{n=1}^{\infty}\copN[n](p_-).
    \mu_\psi=\vev{n}=\sum_{n=1}^{\infty}n\,\psi(n)=\frac{1}{1-p_-}
}
\end{equation}
is universal, just like the mean recurrence time.
This property rules out the analysis of \citet{boguna2004conditional} who claim to be able to distinguish 
the dependence in processes according to the average sequence size.
%Note that $\psi_(n)$ is also the probability of having \emph{not} reached an occurrence $n$ time steps after a \emph{positive} event.
%So, $\vev{n}_-$ can naturally be understood as the average waiting time 

\subsubsection*{Conditional recurrence intervals, clustering}
The dynamics of recurrence times is as important as their statistical properties,
and in fact impacts the empirical determination of the latter%
\footnote{Distribution testing for $\pi(\tau)$ involving Goodness-of-fit tests \cite{ren2010recurrence} %,*xie2012extreme
should be discarded because those are not designed for dependent samples and rejection of the null
 cannot be relied upon. 
 See Section~\ref{sec:GoF} of Chapter~\ref{part:partII}.\ref{chap:GOF} for an extension of GoF tests when some dependence is present. 
}.
It is now clear, both from empirical evidences and analytically from the discussion on Eq.~\eqref{eq:distrecint}, 
that recurrence intervals have a long memory.
In dynamic terms, this means that their occurrences show some clustering. The natural question is then:
``Conditionally on an observed recurrence time, what is the probability distribution of the next one?''
This probability  of observing an interval $\tau'$ immediately following an observed recurrence time $\tau$
is
\begin{equation}
    \pr{X_{\tau+\tau'}>\Xp, X_{\tau+\libracket 1;\tau'\libracket}<\Xp|X_\tau>\Xp, X_{\libracket 1;\tau\libracket}<\Xp,X_0>\Xp}.
\end{equation}
Again, flipping the `$>$' to '$<$' allows to decompose it as
\[
     \frac{\copN[\tau-1;\tau'-1]{}-\copN[\tau;\tau'-1]{}-\copN[\tau-1;\tau']{}+\copN[\tau;\tau']{}}{\copN[\tau-1]{}-2\copN[\tau]{}+\copN[\tau+1]{}}
    -\frac{\pi(\tau+\tau')}{\pi(\tau)},
\]
where the $(\tau\!+\!\tau')$-points probability
\[
   %\CopN[\tau;\tau']{p}=\pr{X_{[0;\tau[}\!<\!\Xp,X_{\tau+[1;\tau']}\!<\!\Xp}
    \CopN[\tau;\tau']{p}=\mathcal{F}_{\libracket 0;\tau\!+\!\tau'\ribracket\backslash\{\tau\}}(F^{-1}(p),\ldots,F^{-1}(p))
\]
shows up.
Of course, this exact expression has no practical use, 
again because there is no hope of empirically measuring any many-points 
probabilities of extreme events with a meaningful signal-to-noise ratio.
We rather want to stress that non-linear correlations and multi-points dependences are relevant, 
and that a characterization of clustering based on the autocorrelation 
of recurrence intervals is an oversimplified view of reality.% that could have important implications in risk management applications.
%(beside the fact that recurrence times cannot be normally distributed)

\subsubsection*{Record statistic}
We conclude this theoretical section on multi-points non-linear dependences
by mentioning that the diagonal $n$-points copula
$\copN[n]$ can alternatively be understood as the distribution of the maximum of 
$n$ realizations of $X$ {in a row}, since %. Indeed
\[
    \Pr{\max_{\tau\leq n}\{X_{t+\tau}\}<F^{-1}(p)}=\Pr{X_{t+\libracket 1,\tau\ribracket}<F^{-1}(p)}%=\pr{F(X_t)<p,\forall t\leq n}
                                                  =\copN[n](p).
\]
Studying the statistics of such ``local'' maximas in short sequences (see e.g.\ \cite{PhysRevE.73.016130}) 
can thus provide information on the multi-points properties of the underlying process.
The CDF of the {running} maximum, or \emph{record}, is $\copN[t](F(x))$ 
and the probability that $t>1$ will be a record-breaking time is the joint probability
\[
	R(t)=\Pr{\max_{\tau< t}\{X_{\tau}\}<X_t},
\]
which is \emph{irrespective of the marginal law} !

%%%sim.record <- function(T=1e5){
%%%     x <- rnorm(T)
%%%     m <- matrix(c(0,-Inf),1,2)
%%%     for (i in 1:length(x)) {if (x[i]>m[nrow(m),2]) {m <- rbind(m,c(i,x[i]))}}
%%%}
%%%lapply(rep(1e5,100),sim.record)
%%%ne fait pas ce que je veux: ici c'est CONDITIONNE au précédent max,
%%%alors que la formule est inconditionnelle (et donc assez compliquée à comprendre).

\begin{table}
\center
\begin{tabular}{||c|c||c|c||c|c||c||}
\hline
  $\pi_+(\tau)$              &$\vev{\tau}_+$     &$\phi_+(w)$       &$\vev{w}_+$     &$\psi_-(n)$              &$\vev{n}_-$    &$R(t)$\\\hline\hline
  $ (1\!-\!q)\,q^{\tau-1}$   &${1}/({1\!-\!q})$  &$(1\!-\!q)\,q^w$  &$q/(1\!-\!q)$   &$ q\,(1\!-\!q)^{n-1}$    &$ {1}/{q}$     &$1/t$
\end{tabular}
\caption{Different probabilities introduced, with thresholds defined as $F(\Xp)=q=1-F(\Xm)$,
for the White Noise process.}
\end{table}

\clearpage
\section{Continuous time}
The $n$-points copula of a continuous-time process $X_t$
is the counting associated with the corresponding arrival process,
whose events are defined like Bernoulli variables with a latent condition
$\1{F(X_t)\leq u_t}$.
Indeed,
\[
    \esp{\1{F(X_{t_1})\leq u_{t_1}},\ldots,\1{F(X_{t_n})\leq u_{t_n}}}
\]
is a straightforward generalization of the discrete $n$-points copula 
\[
    \mathcal{F}_{\libracket 1,n\ribracket}(F^{-1}(u_{t_1}),\ldots,F^{-1}(u_{t_n})),
\]
when the time-stamps $t_i$ are not equidistant.

More generally, one can thus study either directly the arrivals times $t_i$ of events to be defined,
or more fundamentally the latent physical process $X_t$ underlying the firing of the events.
The latter description can be modeled by diffusions or any kind of continuous processes (see Sect.~\ref{ssec:cont_time} below), 
while the former is achieved using point processes and is studied next.

\subsection{Arrival times: the duration process}
In the context of arrivals of a point process, a recurrence interval $R_t$ is larger than 
$\tau$ if the number of events $N_{t+\tau}$ has not increased since last event time $t$:
\[%begin{equation}\label{eq:P_continu}
    \pr{R_t>\tau}=\pr{N_{t+\tau}\!-\!N_t=0}.
\]%end{equation}
When the arrivals are independent and occur with a deterministic intensity $\lambda(t)$ per unit of time,
then the process is Poisson
\begin{equation}\label{eq:P_lambda}
    \pr{N_{t+\tau}\!-\!N_t=0|\lambda}=\e^{-\int_t^{t\!+\!\tau}\lambda(s)\d{s}},
\end{equation}
and the associated durations have a distribution
$
    \pi(\tau|t)\equiv-{\d\pr{R_t>\tau}}/{\d\tau}
$
equal to
\begin{equation}\nonumber%\label{eq:lambda}
    \pi(\tau|t)=\lambda(t\!+\!\tau)\,\e^{-\int_t^{t\!+\!\tau} \lambda(s)\d{s}}.
\end{equation}
When the underlying process is stationary, the intensity $\lambda$ is constant and
in fact equal to the unconditional probability of events.
The usual exponential distribution of recurrence times for independent arrivals is retrieved,
and one verifies in particular that the average duration is equal to the inverse intensity:
\[
    \vev{\tau}=1/\lambda.
\]
When $\lambda(\cdot)$ is a stochastic process, Eq.~\eqref{eq:P_lambda} must be understood conditionally,
and the correct probabilities are found averaging over all possible $\Lambda_{t,t+\tau}=\int_t^{t+\tau}\lambda(s)\d{s}$.
%Typically an independent Ornstein-Uhlenbeck intensity 
%\[
%    \d{\Lambda(t)}=\frac{\theta}{2}(\overline\Lambda-\lambda(t))\d{t} + \sigma \d{z(t)}
%\]
%yields a corrected density
%\[
%    \pi(\tau)=\overline\lambda\,\e^{-\overline\lambda\tau}\,\exp\!\left(\frac{\sigma^2}{2\theta}(1-\e^{-\theta\tau})\right).
%\]
Such doubly stochastic Poisson processes were initially introduced by \citet{cox1955some}.
Of course, $\Lambda$ being an ``arrival rate'' it must always remain positive, 
but for the sake of the example, let us examine the case of Gaussian processes, and in which limits this holds.
As is know from exponentials of Gaussian variables, the correction to the probability \eqref{eq:P_lambda} is %
\[
    \pr{R_t>\tau}=\e^{-\esp{\Lambda_{t,t+\tau}}+\frac{1}{2}\var{\Lambda_{t,t+\tau}}}.
\]
For example, with the integrated intensity following an Ornstein-Uhlenbeck process 
\begin{equation}\label{eq:OU}
   %\d{\Lambda_{t,t+\tau}}={\theta}\,(\overline\lambda\tau-\Lambda_{t,t+\tau}))\d{t} + \sigma \d{z(t)}
    \Lambda_{t,t+\tau}=\overline\lambda\tau+\sigma\int_{0}^{t+\tau}\e^{-\theta(\tau+t-s)}\d{z_s},
\end{equation}
the drift is $\esp{\Lambda_{t,t+\tau}}=\overline\lambda\tau$ 
and the variance $\var{\Lambda_{t,t+\tau}}=\frac{\sigma^2}{2\theta}\left(1-\e^{-\theta(2t+\tau)}\right)$.
Since $\Lambda$ must always remain positive, the Gaussian case can only be a valid approximation
as long as $\esp{\Lambda}\gg\var{\lambda}$, i.e.\ $\tau\gg\sigma^2/(2\theta\overline\lambda)$ in the stationary OU\ example.

The intensity process \eqref{eq:OU} is Gaussian and Markovian, but can be generalized 
to incorporate long-memory by replacing the exponential integration kernel by one with a slow decay,
and/or choosing positive-only innovations in place of the Wiener $\d{z_t}$ \cite{Basu2002297}. % (provided a stationary distribution settles down).
The translational invariance is broken at short times $t\to\infty$, if the intensity depends upon the past realizations of the counting process $N_t$,
but at long times, stationarity is able to settle down.
Such so-called Hawkes processes are obtained replacing the independent innovation $\d{z_t}$ by the counting process itself $\d{N_t}$,
and are thus self-exciting.
They have been recently the focus of intense research in earthquakes studies --- see also the related Epidemic-Type Aftershock Sequence (ETAS) model of triggered seismicity \cite{saichev2007theory,sornette2008solution} ---
and in finance \cite{pomponio2012,bacry2012non,filimonov2012quantifying}. 
The latter studies point toward a strong self-exciting mechanism that brings the system close to criticality,
what was also the conclusion of feedback QARCH model of Chapter~\ref{chap:QARCH}.

\subsection{Continuous processes}\label{ssec:cont_time}
The recurrence times problem for a continuous process can be formulated as a first-passage problem \cite{bray2013persistence}:
indeed, the probability density $$\pi(\tau|t)\d{\tau}=\d{\pr{X_\tau=\Xp, X_{t'}<\Xp\,\forall t'<\tau|X_t\geq \Xp}}$$
can be simply written in terms of the transition probability $f(\Xp,\tau;x,t)$ of a particle 
starting at $X_t=x$ at the initial time $t$ and ending at $X_\tau=\Xp$ without ever hitting the wall at $\Xp$ inbetween.
Of course, since there are no jumps in this continuous setting, the initial value must be conditioned to 
be $X_t=\Xp$ for all the subsequent values to be lower than $\Xp$, and thus
\[
    \pi(\tau|t)=f(\Xp,\tau;\Xp,t).
\]
Once the dynamics of the process is specified (by a Stochastic Differential Equation, a Langevin equation or a non-Markovian generalization thereof),
the transition density $f$ can be found by path-integral methods, or solving Fokker-Planck-like 
equations with appropriate boundary conditions: $f(\Xp,t'|\Xp,t)=0\,\forall t'<\tau$.

Importantly, specific methods of non-Markovian Fokker-Planck equations need to be used in case
the process has memory: f.ex.\ convolution of usual Fokker-Planck with a memory kernel \cite{sokolov2002solutions},
 fractional Fokker-Planck (see \cite{metzler1999anomalous} and many later works of the same authors), continuous limits of Master equations, etc.
 
 In the case of the Brownian Motion, the moments of the distribution $f(x,\tau-t;x_0,0)$ of 
 first passages satisfy a recursive relation \cite{ebeling2005statistical}.

\clearpage
\section{Financial self-copulas}
\begin{table}% [b]
    \center
    \begin{tabular}{|lr||c|c|c|}
        Stock Index &Country        &From   &To     &$T$\\\hline\hline
        S\&P-500    &USA            &Jan.\ 02, 1970 &Dec.\ 23, 2011& 10\,615\\
        KOSPI-200   &South Korea    &Jan.\ 03, 1990 &Dec.\ 26, 2011&  5\,843\\
        CAC-40      &France         &Jul.\ 09, 1987 &Dec.\ 23, 2011&  6\,182\\
        DAX-30      &Germany        &Jan.\ 02, 1970 &Dec.\ 23, 2011& 10\,564\\
        SMI-20      &Switzerland    &Jan.\ 07, 1988 &Dec.\ 23, 2011&  5\,902
    \end{tabular}
    \caption{Description of the dataset used: time series of daily returns of international stock indices.}
    \label{tab:index_data}
\end{table}
We illustrate some of the quantities introduced above on series of daily index returns. 
The properties of the time series used are summarized in Tab.~\ref{tab:index_data}.
\subsection{Conditional probabilities and 2-points dependences}
\begin{figure}[p]
    \center
    \subfigure[SP500]{\label{fig:SP500}\includegraphics[scale=.5,trim=0 0 0 25,clip]{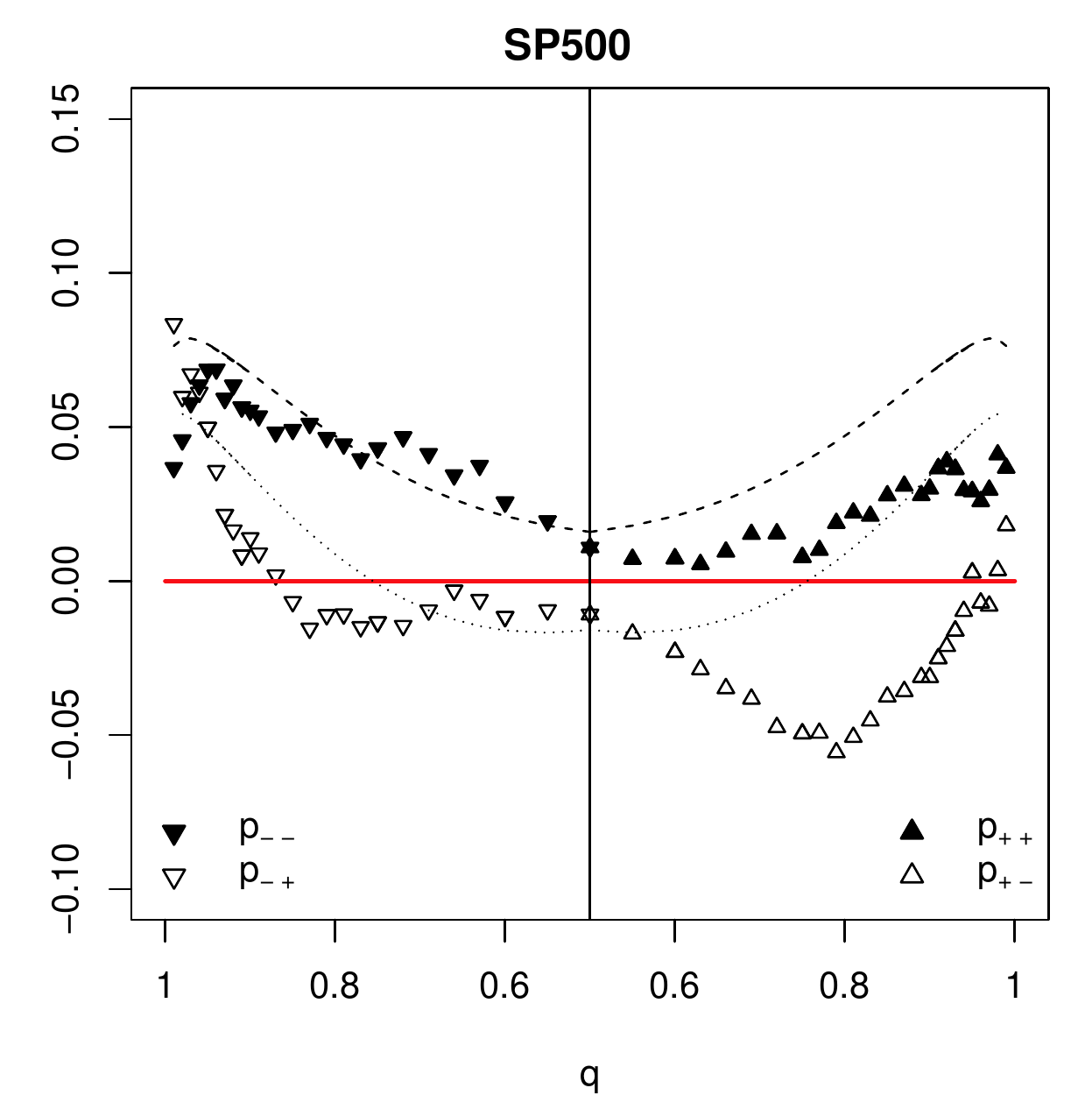}}
    \subfigure[KOSPI]{\label{fig:KOSPI}\includegraphics[scale=.5,trim=0 0 0 25,clip]{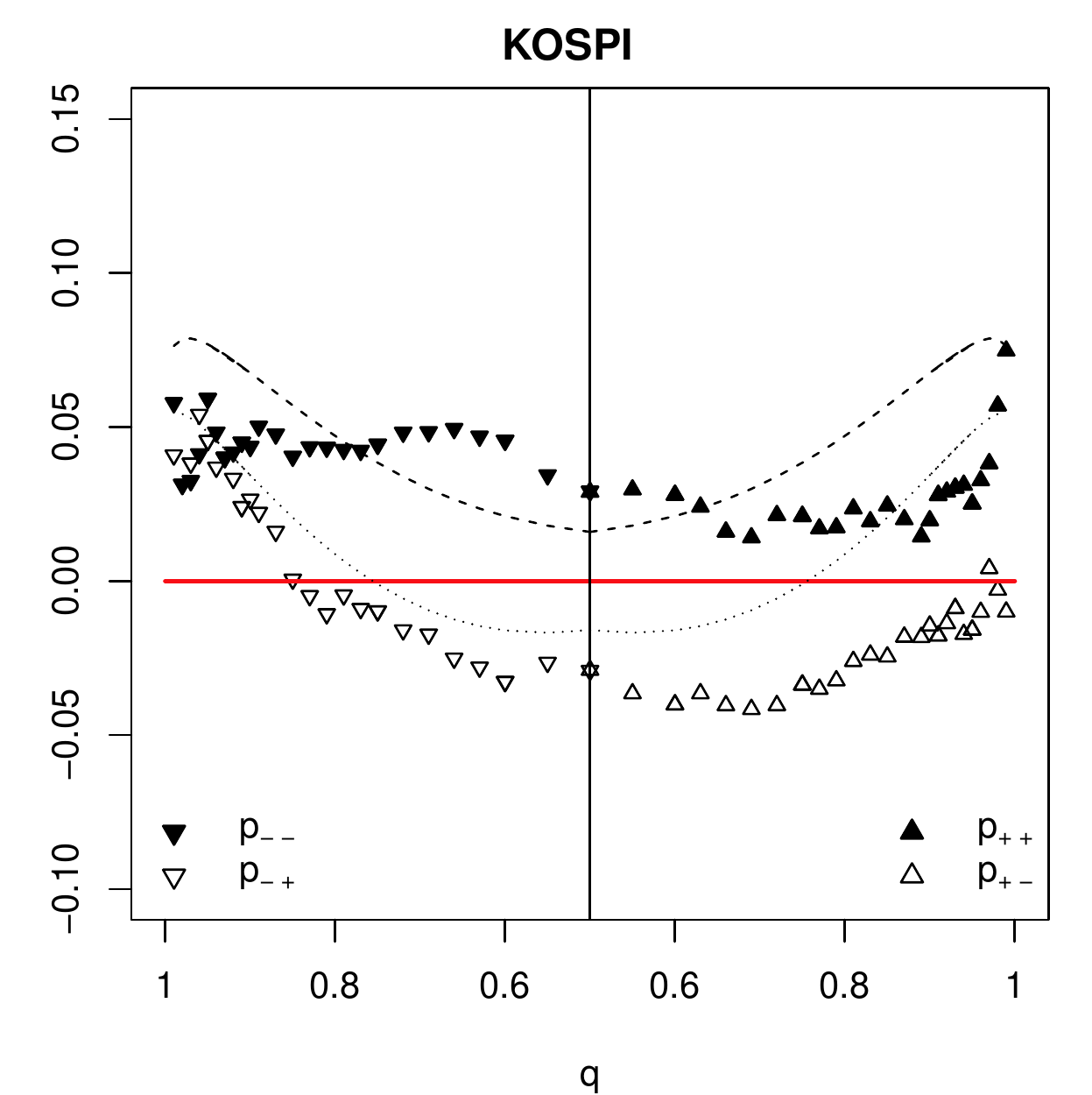}}\\
    \subfigure[CAC  ]{\label{fig:CAC}  \includegraphics[scale=.5,trim=0 0 0 25,clip]{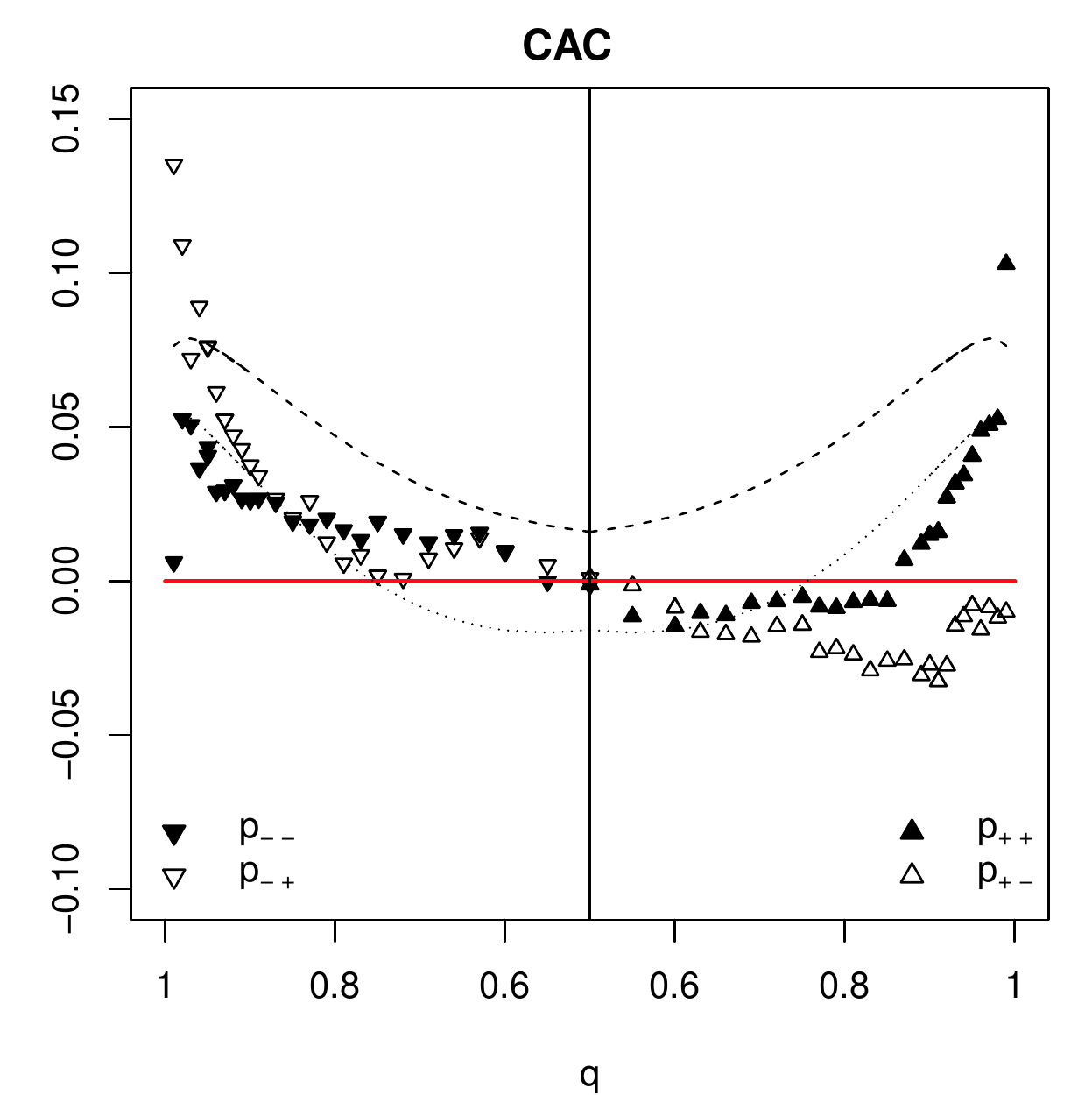}}
    \subfigure[DAX  ]{\label{fig:DAX}  \includegraphics[scale=.5,trim=0 0 0 25,clip]{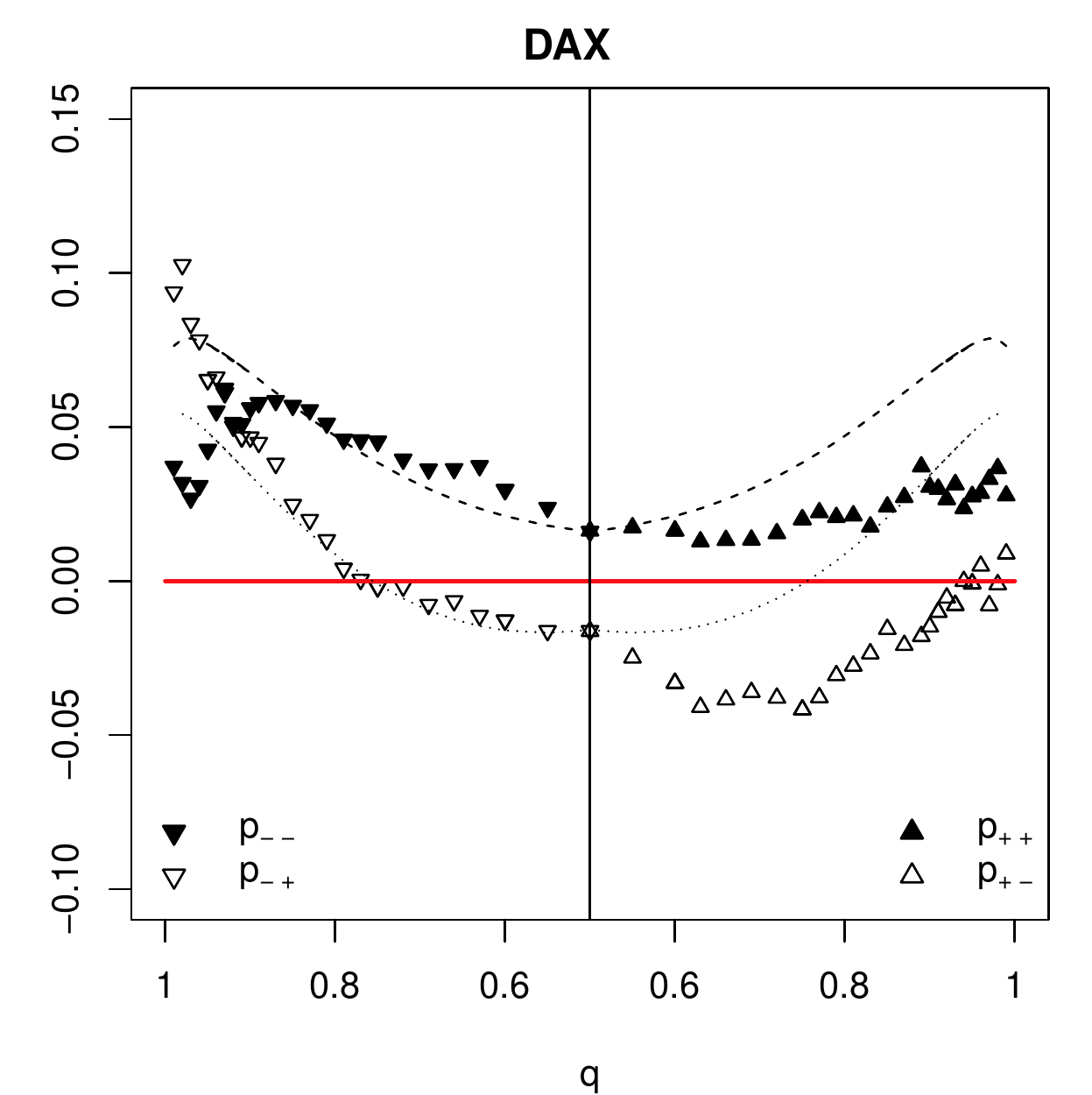}}\\
    \subfigure[SMI  ]{\label{fig:SMI}  \includegraphics[scale=.5,trim=0 0 0 25,clip]{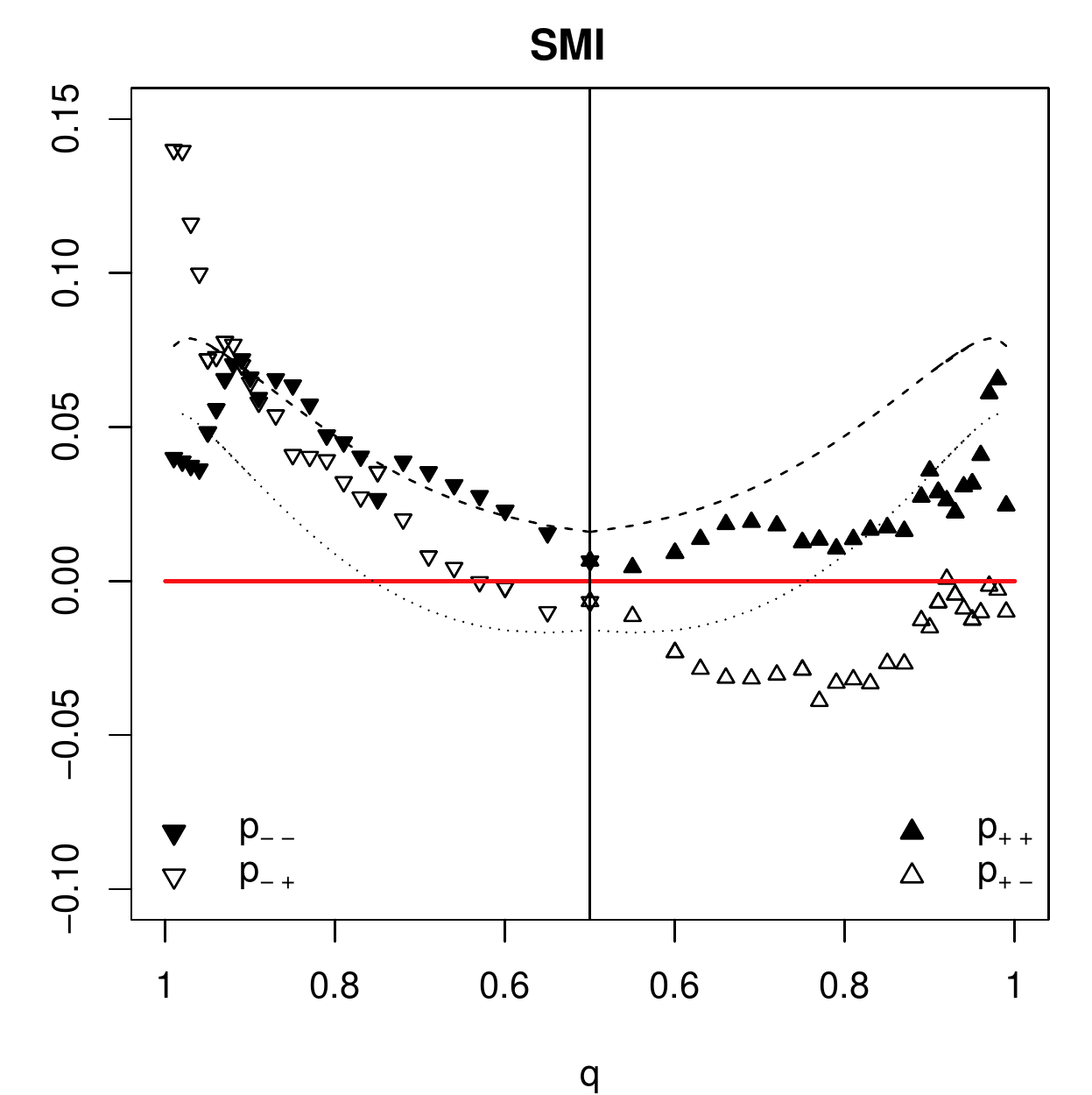}}
    \caption{Conditional extreme probabilities at $\ell=1$ (the WN contribution has been subtracted).
             Filled symbols are for persistence, and empty symbols for reversion.
             Upward pointing triangles are conditioned on positive jumps, and downward pointing triangles are conditioned on negative jumps.}
    \label{fig:condprob_idx}
\end{figure}
We reproduce the study of Ref.~\cite{boguna2004conditional} on the statistic of price changes conditionally on previous return sign,
and extend the analysis to any threshold $|\Xpm|\geq 0$ and to remote lags. 
We first illustrate on Fig.~\ref{fig:condprob_idx} the conditional probabilities $p_{\pm\pm}^{(\ell)}$ (filled symbols)
and $p_{\pm\mp}^{(\ell)}$ (empty symbols) with varying threshold $q=F(\Xp)=1-F(\Xm)$, for $\ell=1$.
To study the departure from the independent case,
 it is more convenient to subtract the White Noise contribution, to get the corresponding \emph{excess} probabilities.
Several features can be immediately observed:
    positive events (upward triangles)   trigger more subsequent positive (filled) than negative (empty) events;
    negative events (downward triangles) trigger more subsequent negative (filled) than positive (empty) events, 
    except in the far tails $q\gtrsim 0.9$ where reversion is stronger than persistence after a negative event.
    Both these effects dominate the WN benchmark, but the latter effect is however much stronger.
This overall behavior is similar for all time series.
    
The shapes of $p_{\pm\pm}$ and $p_{\pm\mp}$ versus $q$ are not compatible with the Student copula benchmarks 
(correlation $\rho=0.05$ and d.o.f.\ $\nu=5$)
shown in dashed and dotted lines, respectively.
Notice that, due to its non-trivial tail-correlations, see part~\ref{part:partI}, the Student copula 
does generate increased persistence with respect to WN, lower reversion in the core and higher reversion in the tails.
But empirically the reversion is asymmetric and typically stronger when conditioning on 
negative events rather than on positive events, a property reminiscent of the 
leverage effect which cannot be accounted for by a pure Student copulas.
Some of the indices exhibit more pronounced reversion and persistence effects. 
Interestingly, the CAC-40 returns have a regime $0.5\leq q\lesssim 0.9$ close to a white noise 
(with in particular a value of $p_{\pm\pm}^{(1)}=p_{\pm\mp}^{(1)}$ very close to the 0 at $q=0.5$,
revealing an inefficient conditioning, 
i.e.\ as many positive and negative returns immediately following positive or negative returns),
but the extreme positive events $q\gtrsim 0.9$ show a very strong persistence,
and the extreme negative events a very strong reversion.

In the next chapter, we study in detail the $p$ and $\ell$ dependence of 
$[\cop[\ell](p,p)-p^2]$ and $[\cop[\ell](p,1\!-\!p)-p\,(1\!-\!p)]$ 
--- which are straightforwardly related to $p_{\pm\pm}^{(\ell)}$ and $p_{\pm\mp}^{(\ell)}$, respectively --- 
and find that the self-copula
of stock returns can be modeled with a high accuracy by a log-normal volatility with log-decaying correlation, 
in agreement with multifractal volatility models. % \textcolor{red}{[Bacry-Muzy, etc.]}.
To anticipate over the discussion there, we give a preview of the results in Fig.~\ref{fig:previewCop}, 
for the aggregated copula of all stocks in the S\&P500 in 2000--2004.
We are able to show precisely how every kind of dependence present in the underlying process
materializes itself in $p_{++}^{(\ell)}$ for different $q$'s:
short ranged linear anti-correlation accounts for the central part ($p\approx 0.5$) departing from the WN prediction,
long-ranged amplitude clustering is responsible for the ``M'' and ``W'' shapes that reveal excess persistence and suppressed reversion, 
and the leverage effect can be observed in the asymmetric heights of the ``M'' and ``W''.
\begin{figure}
\center
	\includegraphics[scale=.74,trim=0    0 1620 0,clip]{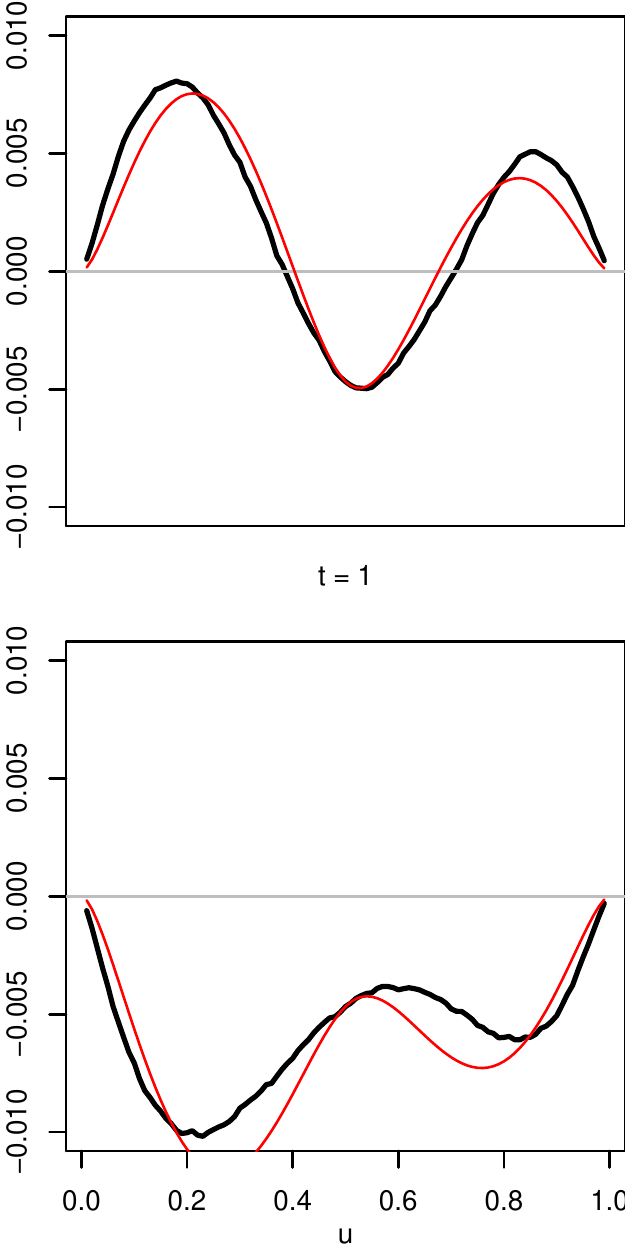}
	\includegraphics[scale=.74,trim=1275 0  350 0,clip]{GoFdependent/fig4}
	\caption{Diagonal (top) and anti-diagonal (bottom) of the self-copula for lags $\ell=1$ and $\ell=128$;
	the product copula has been subtracted. 
    The copula determined empirically on stock returns is in bold black, and 
    a fit with the model of Chapter~\ref{chap:cop_fin} is shown in thin red. }
    \label{fig:previewCop}
\end{figure}

Fig.~\ref{fig:condES} displays the behavior of $\vev{X}_\pm$ versus $q$ (we also show the median $\operatorname{med}(X)_\pm$) at 
lags corresponding to one day ($\ell=1$), one week ($\ell=5$) and one month ($\ell=20$).
The conditional amplitudes $\vev{X}_\pm$ measure ``how large'' a realization will be on average
after an event at a given threshold, 
whereas the conditional probabilities $p_{\pm\pm}$ and $p_{\pm\mp}$ quantify ``how often''
repeated such events occur.
Mind the \emph{unconditional} mean and median values, both above zero and distinct from each other.
At $\ell=1$, the reversion of \emph{extreme} events is revealed again by the change of monotonicity from $q\approx 0.8$ on, 
and more strongly for $q>0.9$ where $\vev{X}_-$ has an opposite sign than the preceding return;
this corroborates the observation made on conditional probabilities above.
Beyond the next day, the general picture is that dependences tend to vanish 
and the empirical measurements get more concentrated around the WN prediction.
However, tail effects are strongly present, with unexpectedly a typical behavior opposite
to that of $\ell=1$: weekly monthly reversion of extreme positive jumps. 
See the caption for a detailed discussion of the specificities of each stock index at every lag $\ell$.
\begin{figure}[ph]
    \center
    \subfigure[SP500]{\label{fig:SP500}\includegraphics[scale=.575]{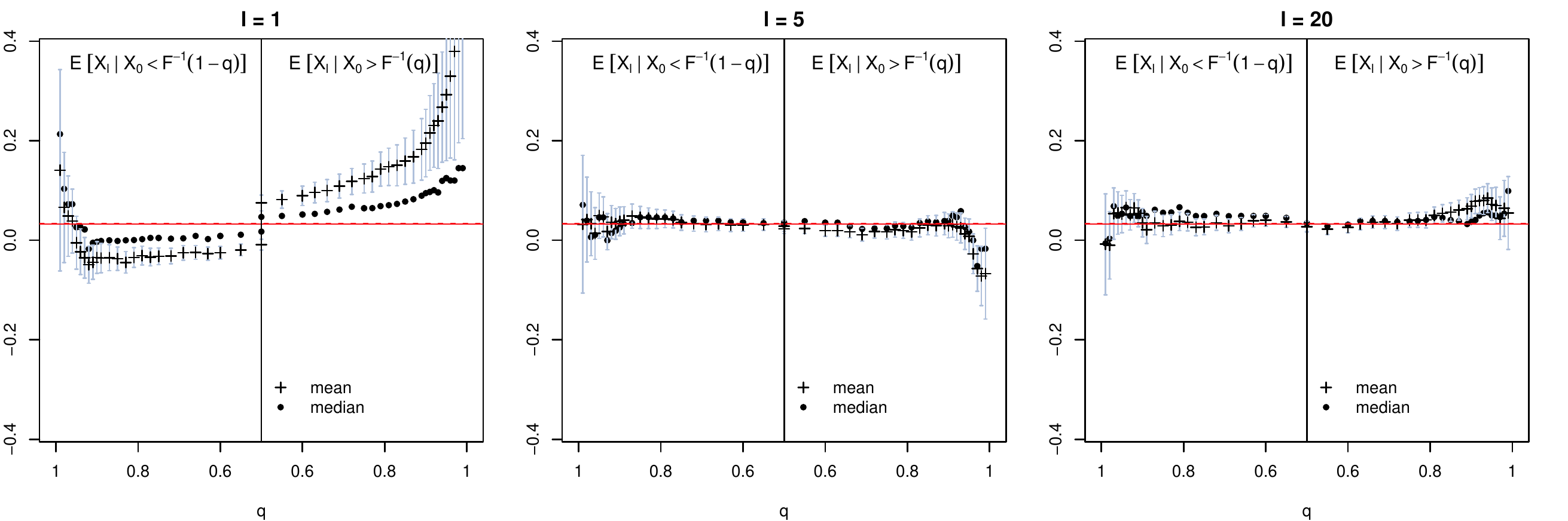}}
    \subfigure[KOSPI]{\label{fig:KOSPI}\includegraphics[scale=.575]{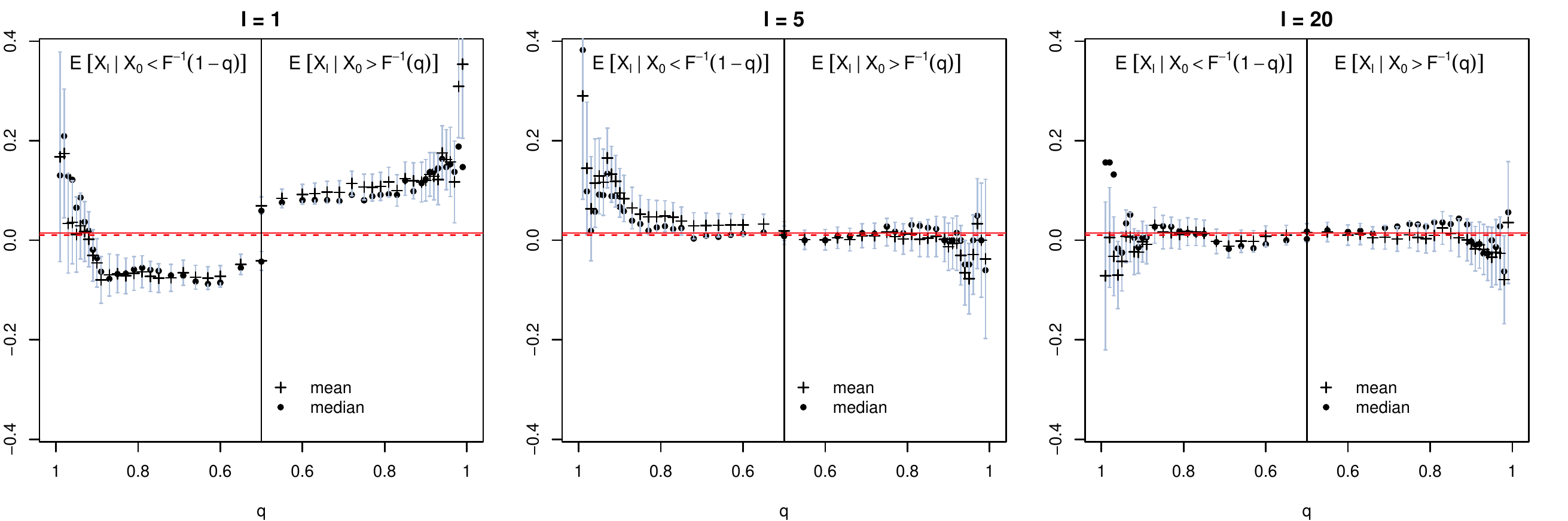}}
    \caption{Conditional extreme amplitudes, at lags $\ell=1, 5, 20$.
             The upper-right and lower-left quadrants express persistence,
             while the upper-left and lower-right quadrants reveal reversion.
             For a scale-free dependence structure, one would expect the magnitudes to decrease with the lag $\ell$
             but the global shape to be conserved.
             What we instead observe is important changes of configuration at different lags:
             For example, the strong reversion of negative tail events at $\ell=1$ vanishes
             at farther lags, and even turns into strong persistence for the CAC and DAX indices.
             That is to say that these indices tend to mean-revert after a negative event at the daily frequency,
             but to trend on the weekly scale.
             Similarly, the strong persistence of positive events at $\ell=1$ converts to a strong reversion 
             in the tails at $\ell=20$ for the European indices (CAC, DAX, SMI); 
             a weaker reversion is observed at intermediate scale ($\ell=5$) for most indices (including US and Korean).
             }
    \label{fig:condES}
\end{figure}
\begin{figure}[p]
    \addtocounter{figure}{-1}
    \addtocounter{subfigure}{2}
    \center
    \subfigure[CAC  ]{\label{fig:CAC}  \includegraphics[scale=.575]{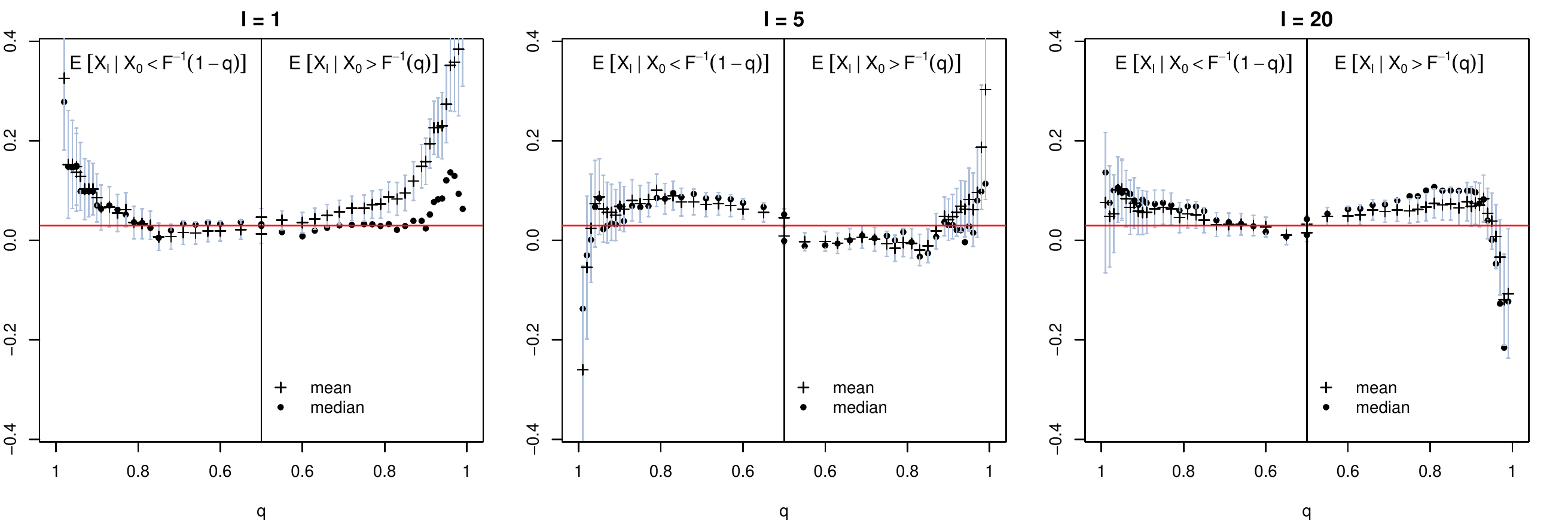}}
    \subfigure[DAX  ]{\label{fig:DAX}  \includegraphics[scale=.575]{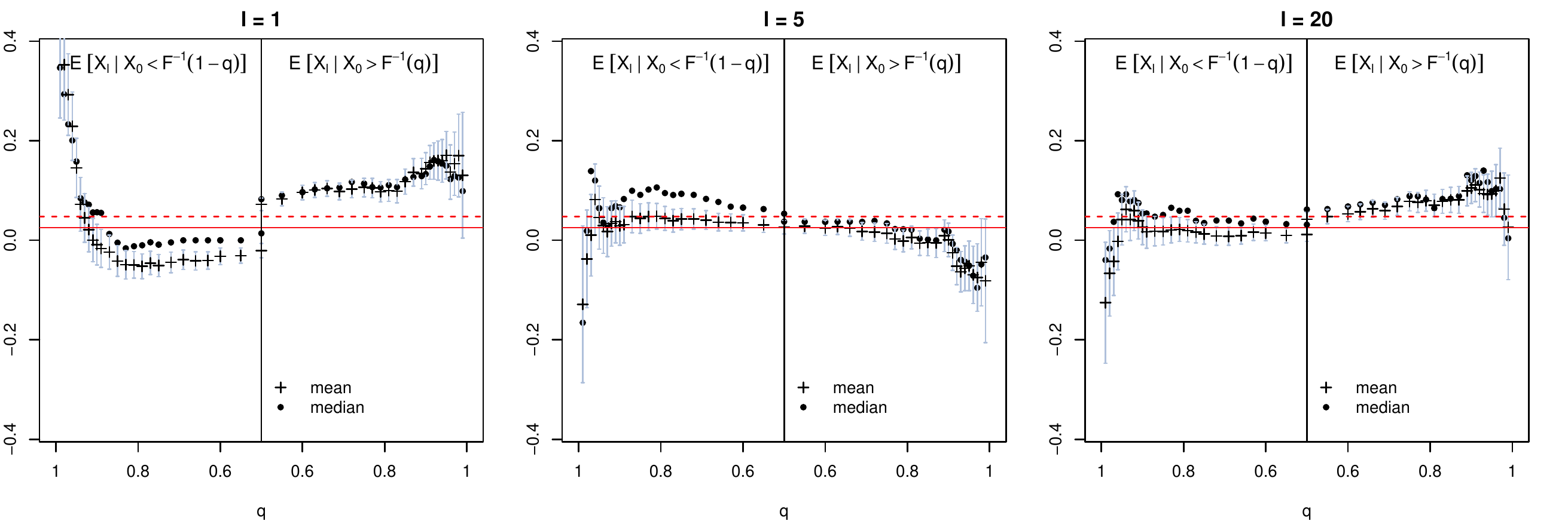}}
    \subfigure[SMI  ]{\label{fig:SMI}  \includegraphics[scale=.575]{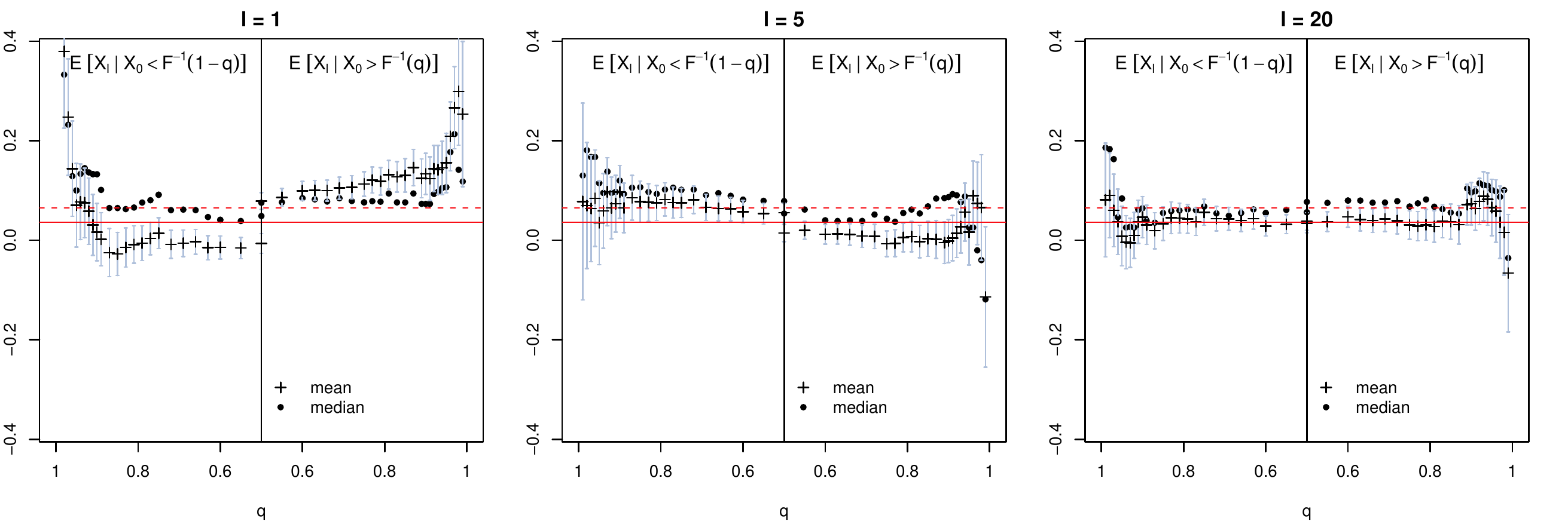}}
    \caption{(continued)}
\end{figure}

\subsection{Recurrence intervals and many-points dependences}

\begin{figure}%[th]
    \center
    \includegraphics[scale=.475,trim=0 0 0 40,clip]{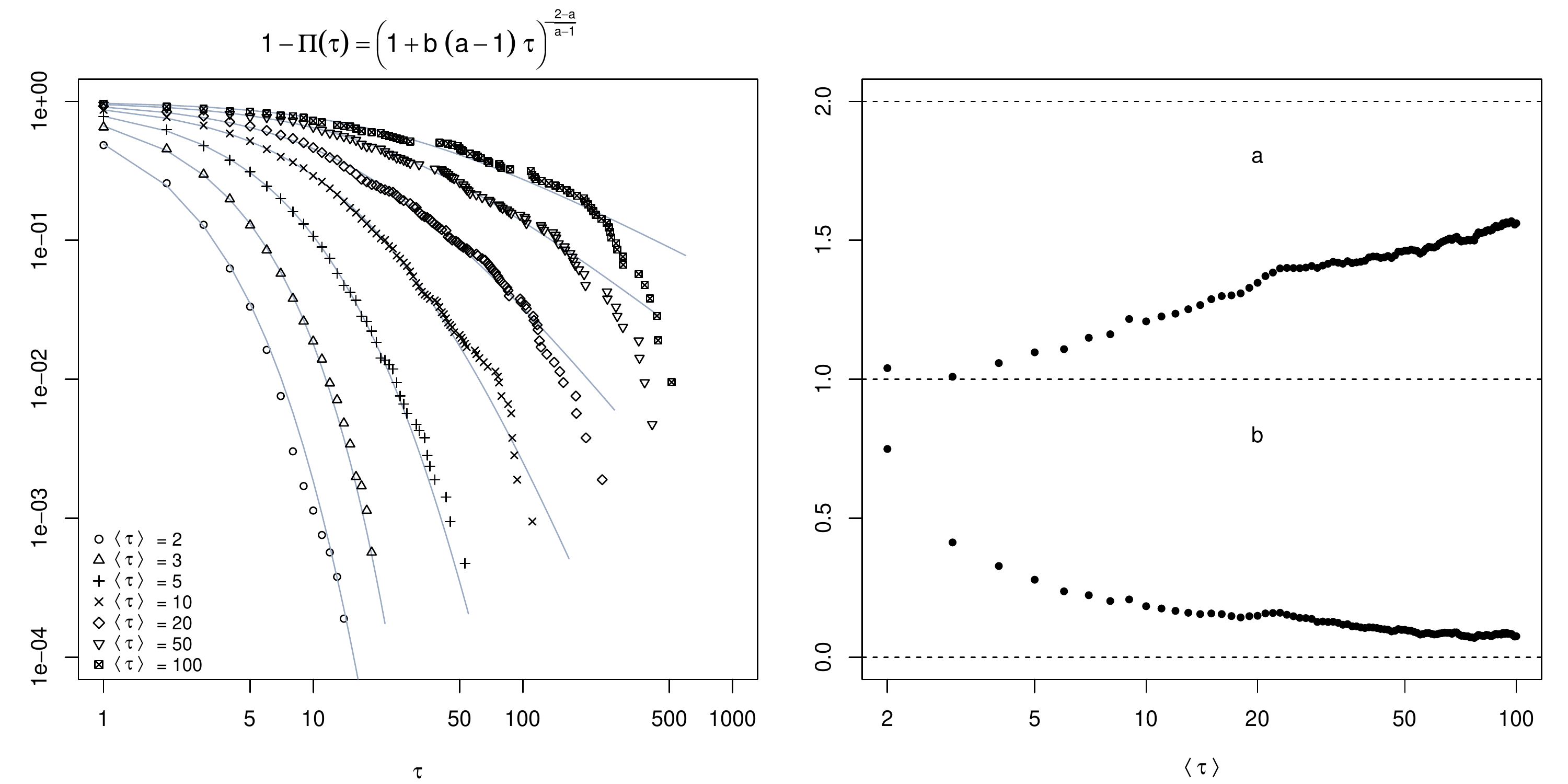}
   %\subfigure[SP500]{\label{fig:SP500_RI}\includegraphics[scale=.425,trim=0 0 425 40,clip]{self-copulas/SP500_RI}}
   %\subfigure[KOSPI]{\label{fig:KOSPI_RI}\includegraphics[scale=.425,trim=0 0 425 40,clip]{self-copulas/KOSPI_RI}}
   %\subfigure[CAC  ]{\label{fig:CAC_RI}  \includegraphics[scale=.425,trim=0 0 425 40,clip]{self-copulas/CAC_RI}}
   %\subfigure[DAX  ]{\label{fig:DAX_RI}  \includegraphics[scale=.425,trim=0 0 425 40,clip]{self-copulas/DAX_RI}}
   %\subfigure[SMI  ]{\label{fig:SMI_RI}  \includegraphics[scale=.425,trim=0 0 425 40,clip]{self-copulas/SMI_RI}}
    \caption{DAX index returns. 
    \textbf{Left:} 
    tail probability $1-\Pi(\tau)$ of the recurrence intervals, at several thresholds $p_+=1/\vev{\tau}$,
    in log-log scale. 
    Grey curves are best fits to $[1+b\,(a\!-\!1)\,\tau]^{(2\!-\!a)/(a\!-\!1)}$ suggested in Ref.~\cite{ludescher2011universal}.
    \textbf{Right:}
    estimated parameters $a$ and $b$ of the best fit.}
    \label{fig:RI_tCDF}
\end{figure}

Even the simple, two-points measures of self-dependence studied up to now show that
non-linearities and multi-scaling are two ingredients that must be taken into account
when attempting to describe financial time series;
we now examine their many-points properties.
As an example,
we compute the distribution of recurrence times of returns above a threshold $\Xp=F^{-1}(1\!-\!p_+)$.

Fig.~\ref{fig:RI_tCDF} shows the tail cumulative distribution 
 $1-\Pi(\tau)$ of the recurrence intervals of DAX returns, at several thresholds $p_+=1/\vev{\tau}$ 
 --- the distribution for other indices is very similar.
In the log-log representation used, an exponential distribution (corresponding to independent returns) 
would be concave and rapidly decreasing, while a power-law would decay linearly.
The empirical distributions fit neither of those, and the authors of Ref.~\cite{ludescher2011universal}
suggested a parametric fit of the form 
\[
    1-\Pi(\tau)=[1+b\,(a\!-\!1)\,\tau]^{(2\!-\!a)/(a\!-\!1)}.
\]
However, important deviations are present in the tail regions for thresholds at $\Xp\gtrsim F^{-1}(0.9)$, 
i.e.\ $\vev{\tau}\gtrsim 1/(1-0.9) =10$: as a consequence, there is no hope that 
the curves for different threshold collapse onto a single curve after a proper rescaling \cite{chicheportiche2013recurrences}, 
as is the case e.g.\ for seismic data.
A more fundamental determination of the form of $\Pi(\tau)$ should rely on Eq.~\eqref{eq:cumul_Pi}
and a characterization of  the $\tau$-points copula.

Similarly to the statistic of the recurrence intervals, their dynamics must be studied carefully. 
We have shown that the conditional distribution of recurrence intervals after a previous recurrence 
is very complex and involves long-ranged non-linear dependences, 
so that a simple assessment of recurrence times auto-correlation may not be informative enough for a
deep understanding of the mechanisms at stake.

%\subsection{Recurrences in transactions times in a LOB}

\section{Conclusion}
We report several properties of recurrence times and the statistic of other 
observables (waiting times, cluster sizes, records, aftershocks) in light of their decription in terms of the diagonal copula,
and hope that these studies can shed light on the $n$-points properties of the process by 
assessing the statistics of simple variables
rather than positing an {\it a priori} dependence structure.
 
The exact universality of the mean recurrence interval imposes a natural scale in the system.
A scaling relation in the distribution of such recurrences is only possible in absence of any other characteristic time.
When such additional characteristic times are present (typically in the non-linear correlations), 
% the rescaling must be performed with a scale-dependent renormalization, in the spirit of %Eq.~\eqref{eq:scaling2}. 
% \begin{equation}\label{eq:scaling2}
    % \pi(\tau)=\frac{1}{\sigma_\pi}g\!\left(\frac{\tau-\mu_\pi}{\sigma_\pi}\right),
% \end{equation}
% where $\sigma_\pi$ is defined in Eq.~\eqref{eq:secmon} and \emph{does} contain information about
% the dependence structure.
% Then in some cases where $\sigma_\pi$ embeds all the dependence of the underlying process,
 % the scaling \eqref{eq:scaling2} might hold, at least approximately, even in the presence of several scales, 
% as the recurrence times corresponding to the different regimes would be renormalized accordingly.
no such scaling is expected, in contrast with time series of earthquake magnitudes.

We also stress that recurrences are intrinsically multi-points objects 
related to the non-linear dependences in the underlying time-series.
As such, their autocorrelation is not a reliable measure of their dynamics,
for their conditional occurence probability is much history dependent.

In a continuous-time setting, the role of the $n$-points copula is played by a counting process.
The events to be counted can either be triggered by an underlying continuous process crossing a threshold, 
or more directly be modeled as a self-exciting point process, like a Hawkes process.

Ultimately, recurrences may be used to characterize risk in a new fashion.
Instead of --- or in addition to --- caring for the amplitude and probability of adverse events at a given horizon,
one should be able to characterize the risk in a dynamical point of view.
In this sense, an asset $A_1$ could be said to be ``more risky'' than another asset $A_2$
if its distribution of recurrence of adverse events has such and such ``bad'' properties that $A_1$ does not share.
This amounts to characterizing the disutility  by ``When ?''  shocks are expected to happen,
in addition to the usual ``How often ?'' and ``How large ?''.

% Of course the average recurrence time is not a good indicator because it is universal and cannot discriminate between
% blabla. Furthermore, resolution is low around the mean because blabla
% But higher moments of the distribution or other quantities like its tail index could be successful at 
% providing a .......

The generalization of the $n$-points copula for continuous-time processes has been rapidly introduced,
and it would be interesting to study many-points dependences in continuous processes,
for example transaction times in a Limit Order Book.

The next chapter uses the copula approach to address the long-range nature of volatility-dependences 
and incorporate multi-scaling and non-linearities.
The theory of Goodness-of-fit tests for dependent observations developped in part~\ref{part:partI} is put in practice 
to characterize the unconditional distribution of stock returns.

\chapter{The long-ranged nature of volatility dependences}\label{chap:cop_fin}
\minitoc

%%%%%%%%%%%%%%%%%%%%%%%%%%%%%%%%%%%%%%%%%%%%%%%%%%%%%%%%%%%%%%%%%%%%%%%%
When trying to extend GoF tests to dependent variables, self-copulas appear naturally, 
as we carefully showed in Section~\ref{sec:GoF} of Chapter~\ref{part:partI}.\ref{chap:GOF}. 
It is thus very important to estimate and/or model precisely the underlying dependence structure 
in order to apply the theory presented there.
This chapter proposes two applications of the GoF tests with dependent observations.

In Section~\ref{sec:example} we go through a detailed example when the dependences are described by a pseudo-elliptical copula
exhibiting log-amplitudes correlations.
We first investigate the Markovian case of weak, short-range dependences, 
and then relax the second assumption to investigate long-memory; 
we also add linear correlations and leverage effect in order to be as close as possible to financial copulas.
Indeed, Section~\ref{sec:application} is dedicated to an application of the theory to the case of financial data:
after defining our data set, we perform an empirical study of dependences in series of stock returns, 
and interpret the results in terms of the ``self-copula'';
implication of the dependences on GoF tests are illustrated for this special case using Monte-Carlo simulations.
In our empirical study of financial self-copulas, we rely on a non-parametric estimation rather than imposing, 
for example, a Markovian structure of the underlying process, as in e.g.~\cite{darsow1992copulas,ibragimov2008copulas}.

\section{An explicit example: the log-normal volatility}\label{sec:example}

In order to illustrate the above general formalism, 
we focus on the specific example of the product random variable $X=\sigma\xi$, 
with i.i.d.\ Gaussian residuals $\xi$ and log-normal stochastic amplitude $\sigma=\e^{\omega}$
(we denote generically by $F$ the cdf of $X$).
Such models are common in finance to describe stock returns, as will be discussed in the next section.
For the time being, let us consider the case where the $\omega$'s are of zero mean, and covariance given by:
\begin{equation}\label{eq:modelLogNorm}
	{\rm Cov}(\omega_n \omega_{n+t}) = \Sigma^2\, f\!\left(\frac{t}{T}\right), \qquad (t > 0).
\end{equation}
where $T$ is a typical cutoff scale.

The pairwise copulas can be explicitly written in the limit of weak correlations, $\Sigma^2 \to 0$.
One finds:
\begin{align}
	\cop[t](u,v)-uv&=\Sigma^2 f\!\left(\frac{t}{T}\right) \tilde{A}(u)\tilde{A}(v) + \mathcal{O}(\Sigma^4)\\\label{eq:def_A}
	\textrm{with}\quad \tilde{A}(u)&=\int\limits_{-\infty}^{\infty}\varphi(\omega)\varphi'(\frac{F^{-1}(u)}{\e^{\omega}})\d\omega
\end{align}
where here and in the following $\varphi(\cdot)$ denotes the univariate Gaussian pdf, and $\Phi(\cdot)$ the Gaussian cdf.
The spectrum of $A(u,v)=\tilde{A}(u)\tilde{A}(v)$ consists in a single non-degenerate eigenvalue 
$\lambda^A=\tr A=\int_0^1\tilde{A}(u)^2\, \d{u}$, 
and an infinitely degenerate null eigenspace. 
Assuming $f_\infty = \sum_{r=1}^\infty f(r) < +\infty$ (the memory might be long-ranged, yet must decay sufficiently fast),
the covariance kernel reads:
\[
	H(u,v)=I(u,v)+ 2\,T\,\Sigma^2 f_\infty \, A(u,v).
\]
Depending on the value of the parameters, the first term or the second term may be dominant. 
Note that one can be in the case of weak correlations ($\Sigma^2 \to 0$) but long range memory $T \gg 1$, 
such that the product $T\Sigma^2$ can be large (this is the case of financial time series, see below). 
If $T\Sigma^2$ is small, one can use perturbation theory around the Brownian bridge limit
(note that $\tr I\approx 10 \tr A$, see Tab.~\ref{tab:traces} in appendix),
whereas if $T\Sigma^2$ is large, it is rather the Brownian term $I(u,v)$ that can be treated as a perturbation. 
Elements of the algebra necessary to set up these perturbation theories are given in Appendix~\ref{chap:apx_algebra}.

We now turn to a numerical illustration of our theory, in this simple case where only volatility correlations are present.
%(i.e.\ $\beta_t =\rho_t =0$ in Equation~(\ref{eq:cop_perturb1}) above).

\subsection{Short-range memory}
In order to remain in the framework of the theory exposed in Chapter~\ref{part:partI}.\ref{chap:GOF}
and satisfy in particular the conditions for the validity of the CLT
under weak dependences, we assume a Markovian dynamics for the log-volatilities:
\begin{equation}\label{eq:markov_logvol}
X_n=\xi_n\,\e^{\omega_{n}-\var{\omega}}, \quad\textrm{with}\quad\omega_{n+1} = g \omega_n + \Sigma \eta_n,
\end{equation}
where $g < 1$ and $\eta_n$ are i.i.d.\ Gaussian variables of zero mean and unit variance. In this case, 
\begin{equation}
	\alpha_t=\,\,{\rm Cov}(\omega_n \omega_{n+t}) = \frac{\Sigma^2}{1-g^2}\,g^t.
\end{equation}
In the limit where $\Sigma^2 \ll 1$, the weak dependence expansion holds and one finds explicitly:
\begin{equation}\label{eq:weak_dep_kernel}
	H(u,v)=I(u,v)+ {2}\frac{g \Sigma^2}{(1-g)^2(1+g)} \, A(u,v).
\end{equation}
%%% OK PAREIL QUE MOI, AVEC LOG(x)~x-1, x>1
In order to find the limit distribution of the test statistics, we proceed by Monte-Carlo simulations.
The range $[0,1]^2$ of the copula is discretized on a regular lattice of size $(M\times M)$.
The limit process is described as a vector with $M$ components and built from Equation~(\ref{eq:y_sum_z}), page~\pageref{eq:y_sum_z}, as  
$\mathbf{\tilde{y}}=U\Lambda^{\frac{1}{2}}\mathbf{z}$ 
where the diagonal elements of $\Lambda$ are the eigenvalues of $H$ (in decreasing order),
and the columns of $U$ are the corresponding eigenvectors. 
Clearly, ${\rm Cov}(\mathbf{\tilde{y}},\mathbf{\tilde{y}})=U\Lambda U^{\dagger}=H$.

For each Monte-Carlo trial, $M$ independent random values are drawn 
from a standard Gaussian distribution and collected in $\mathbf{z}$.
Then $\mathbf{y}$ is computed using the above representation. This allows one to determine the two relevant statistics:
\begin{align*}
	K\!S&=\max_{u=1\ldots M}|\tilde{y}_u|\\
	C\!M&=\frac{1}{M}\sum_{u=1}^M\tilde{y}_u^2=\frac{1}{M}\mathbf{\tilde{y}}^{\dagger}\mathbf{\tilde{y}}=\frac{1}{M}\mathbf{z}^{\dagger}\Lambda\mathbf{z}.
\end{align*}
The empirical cumulative distribution functions of the statistics for a large number of trials are shown in Figure~\ref{fig:ecdf_sim} 
together with the usual Kolmogorov-Smirnov and Cram\' er-von-Mises limit distributions 
corresponding to the case of independent variables.

\begin{figure}[p]
    \centering
	\includegraphics[scale=0.45]{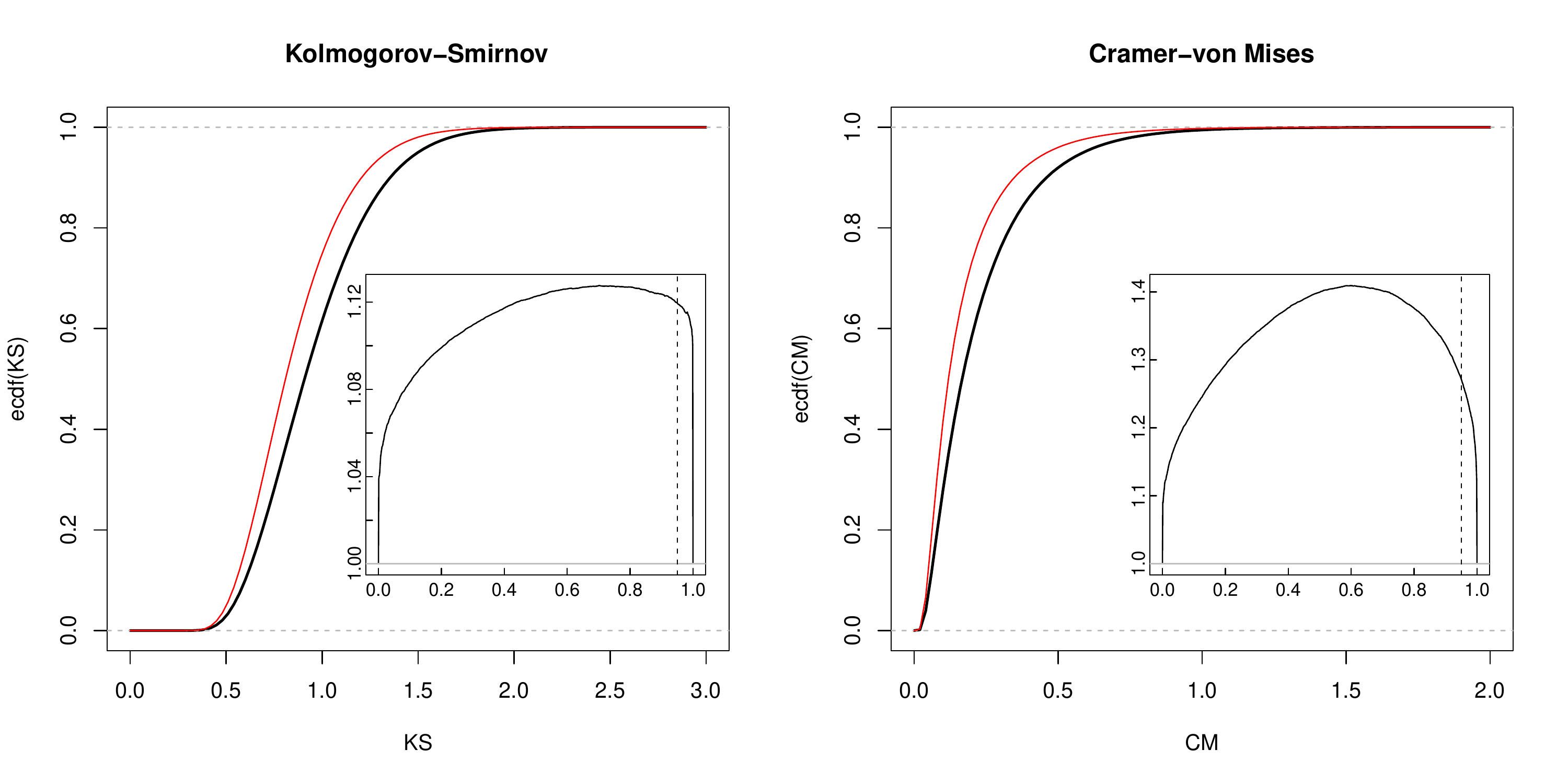}
	\caption{Markovian model. \textbf{Left:} Cumulative distribution function of the supremum of $\tilde{y}(u)$.
	         \textbf{Right:} Cumulative distribution of the norm-2 of $\tilde{y}(u)$.
	         The cases of independent draws (thin red) and dependent draws (bold black) are compared. 
	         The dependent observations are drawn according to the weak-dependence kernel (\ref{eq:weak_dep_kernel}) with parameters $g=0.88, \Sigma^2=0.05$. 
	         \textbf{Insets:} The effective reduction ratio $\sqrt{\frac{N}{N_{\mbox{\scriptsize{eff}}}(u)}}=\frac{\textrm{ecdf}^{-1}(u)}{\textrm{cdf}_{\mbox{\scriptsize{L}}}^{-1}(u)}$ 
	                          where $\textrm{L}=\textrm{KS,CM}$.
	                          The dashed vertical line is located at the 95-th centile and thus indicates the reduction ratio corresponding to the \mbox{p-value} $p=0.05$ (as the test is unilateral).}
	\label{fig:ecdf_sim}
\end{figure}

\begin{figure}[p]
    \centering
	\includegraphics[scale=0.45]{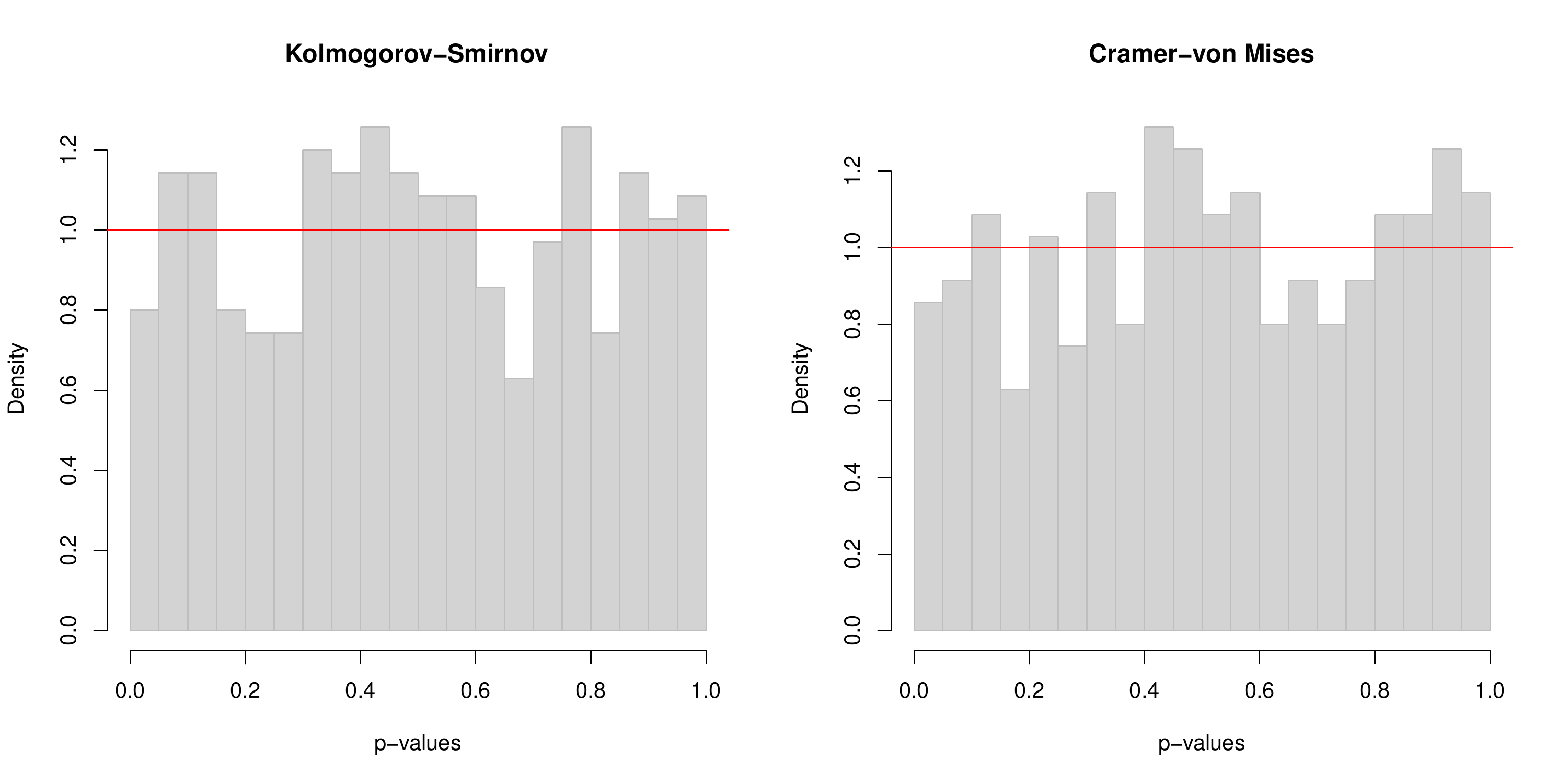}
	\caption{Histogram of the \mbox{p-values} in the GoF test on simulated data,
	according to Equation~(\ref{eq:markov_logvol}).
	Uniform distribution of the \mbox{p-values} of a test indicates that the correct law of the statistic is used.}
	\label{fig:KS_pvals_sim}
\end{figure}

In order to check the accuracy of the obtained limit distribution, 
we generate $350$ series of $N=2500$ dates according to Equation~(\ref{eq:markov_logvol}).
For each such series, we perform two GoF tests, namely KS and CM, and calculate the corresponding \mbox{p-values}. 
By construction, the \mbox{p-values} of a test should be uniformly distributed 
if the underlying distribution of the simulated data is indeed the same as the hypothesized distribution.
In our case, when using the usual KS and CM distributions for independent data,
the \mbox{p-values} are much too small and their histogram is statistically not compatible with the uniform distribution.
Instead, when using the appropriate limit distribution found by Monte-Carlo and corresponding to the correlation kernel (\ref{eq:weak_dep_kernel}),
the calculated \mbox{p-values} are uniformly distributed, as can be visually seen on Figure~\ref{fig:KS_pvals_sim}, 
and as revealed statistically by a KS test (on the KS test!), comparing the 350 \mbox{p-values} to ${\rm H}_0: p\sim\mathcal{U}[0,1]$.

\subsection{Long-range memory}
If, instead of the AR(1) (Markovian) prescription (\ref{eq:markov_logvol}), the dynamics of the $\omega_n$ is given by
a Fractional Gaussian Noise (i.e.\ the increments of a fractional Brownian motion) \cite{mandelbrot1968fractional}
with Hurst index $\frac{2-\nu}{2}>\frac{1}{2}$, the log-volatility has a long ranged autocovariance
\begin{equation}\label{eq:alpha_FBM}
	\alpha_t={\rm Cov}(\omega_n,\omega_{n+t})=\frac{\Sigma^2}{2}\left(|t+1|^{2-\nu}-2t^{2-\nu}+|t-1|^{2-\nu}\right),\quad t\geq 0
\end{equation}
that decays as a power law $\propto (2-3\nu+\nu^2)t^{-\nu}$ as $t\to\infty$, corresponding to long-memory, 
therefore violating the hypothesis under which the above theory is correct.
Still, we want to illustrate that the above methodology leads to a meaningful improvement of the test, even in the case of long-ranged dependences.
The corresponding covariance kernel of the $X$s, 
\begin{equation}\label{eq:kernel_FBM}
	H(u,v)=I(u,v)+2\Sigma^2A(u,v)\sum_{t=1}^{N}\left(1-\frac{t}{N}\right)\alpha_t,
\end{equation}
is used in a Monte-Carlo simulation like in the previous case in order to find the appropriate 
distribution of the test statistics $KS$ and $CM$ (shown in Figure~\ref{fig:FGN}, see caption for the choice of parameters).
We again apply the GoF tests to simulated series, and compute the \mbox{p-values} according to the theory above.
As stated above, our theory is designed for short-ranged dependences whereas 
the fGn\nomenclature{FGN, fGn}{Fractional Gaussian noise} process is long-ranged. 
The \mbox{p-values} are therefore not expected to be uniformly distributed. 
Nevertheless, the distribution of the \mbox{p-values} is significantly corrected 
toward the uniform distribution, see Figure~\ref{fig:ecdf_FGN}:
with the naive CM distribution (middle), the obtained \mbox{p-values} are localized around zero, 
suggesting that the test is strongly negative. 
If instead we use our prediction for short-range dependences (right), we find a clear improvement, 
as the \mbox{p-values} are more widely spread on $[0,1]$ (but still not uniformly).

\begin{figure}[p]
    \centering
    \subfigure[Cumulative distribution function of the norm-2 of $\tilde{y}(u)$, see Fig.~\ref{fig:ecdf_sim} for full caption.]{
              \includegraphics[scale=0.6,trim=425 0 0 0,clip]{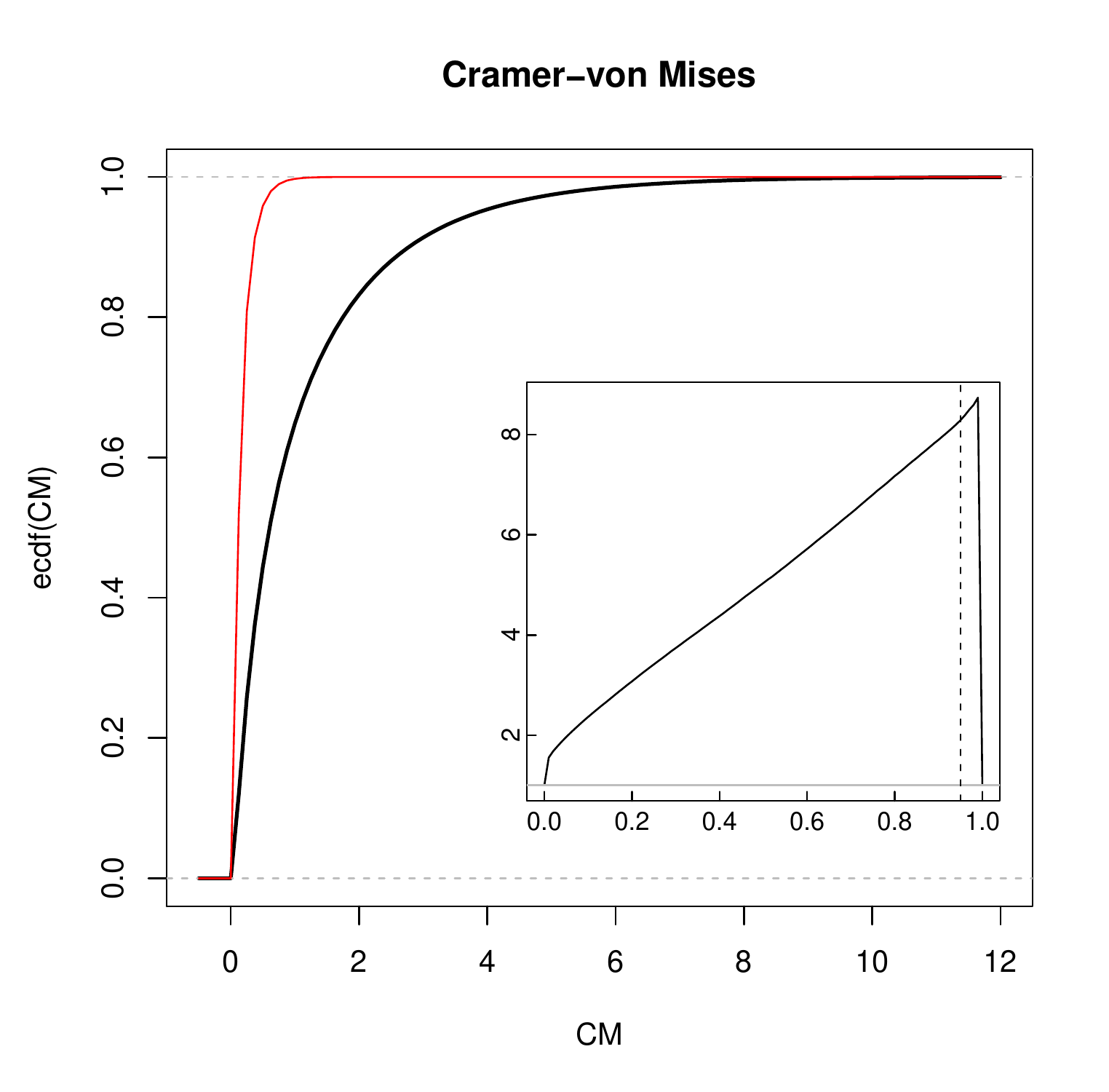}
              \label{fig:ecdf_FGN}}
	\subfigure[Histogram of the \mbox{p-values} in the CM test on simulated data, using the i.i.d.\ CM distribution (left) and the 
	         prediction obtained assuming short range correlations (right). %; 
	         %when the naive CM distribution is used (middle), the obtained \mbox{p-values} are localized indicating that the wrong law is used.
	         %When we use our prediction for short-range dependences (right), we find a clear improvement, as the
	         %\mbox{p-values} are more widely spread on $[0,1]$. However the \mbox{p-values} are still not statistically consistent with the uniform distribution.
	         The dependent observations are drawn according to Eq.~(\ref{eq:alpha_FBM}) with parameters $\nu=\frac{2}{5}$, $\Sigma^2=1$, $N = 1500$.]{
             \includegraphics[scale=0.7,trim=280 0 0 325,clip]{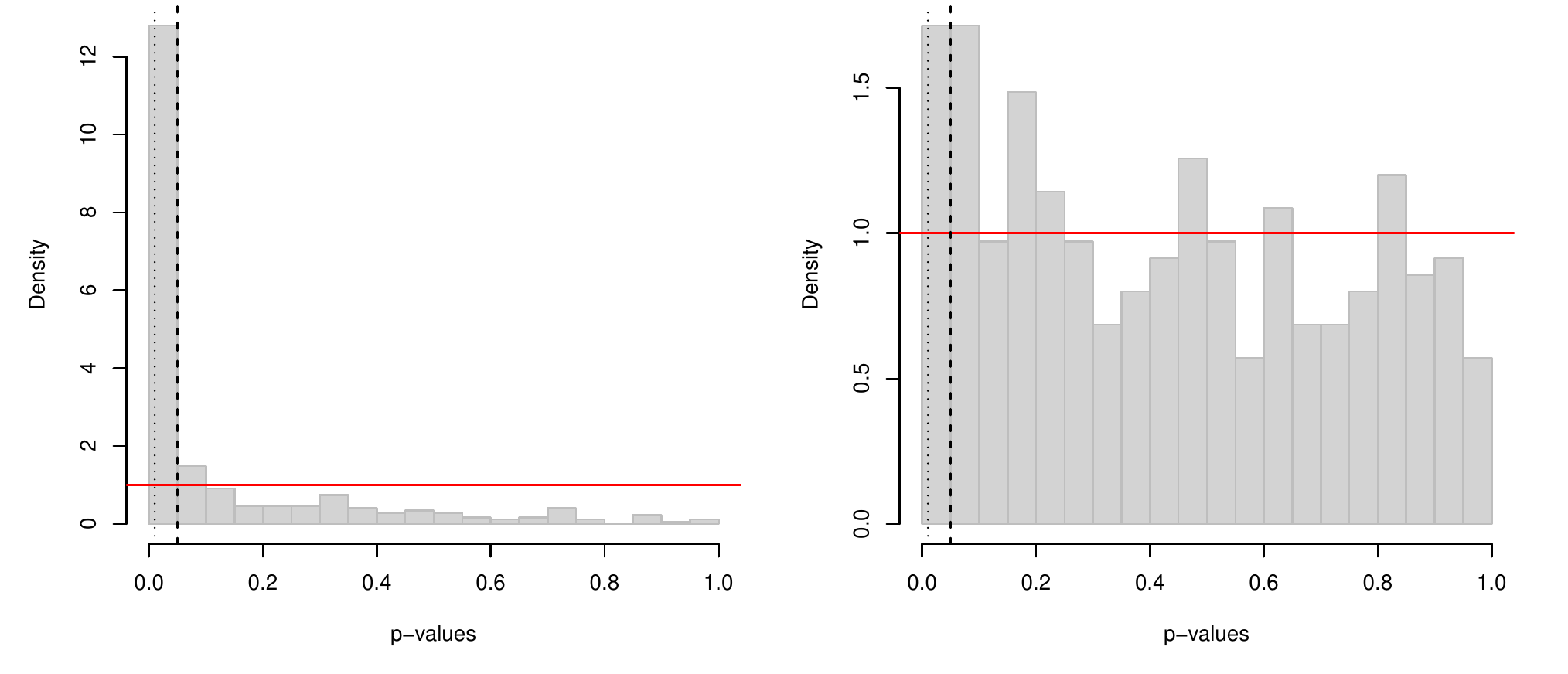}
             \label{fig:pvals_FGN}}
	\caption{Fractional Brownian Motion.}\label{fig:FGN}
\end{figure}

\subsection{Accounting for linear correlation and time reversal asymmetry}
It is interesting to generalize the above model to account for weak dependence between the residuals $\xi$ and between the residual and the volatility,
without spoiling the log-normal structure of the model. We therefore write:
\[
	X_0=\xi_0\e^{\omega_0};\quad X_t=\xi_t\e^{\alpha_t\omega_0+\beta_t\xi_0+\sqrt{1-\alpha_t^2-\beta_t^2}\omega_t}
	\quad\textrm{with}\quad
	\esp{\xi_0\xi_t}=r_t
\]
where all the variables are $\mathcal{N}(0,1)$, so that in particular 
\begin{align*}
	\rho_t={\rm Corr}(X_0,X_t)&=r_t(1+\beta_t^2)\e^{\alpha_t-1}\\
	       {\rm Corr}(X_0^2,X_t^2)&=\frac{\left(1+2r_t^2(1+10\beta_t^2+8\beta_t^4)+4\beta_t^2\right)\e^{4\alpha_t}-1}{3\e^4-1}\\
	       {\rm Corr}(X_0,X_t^2)&=2\beta_t\frac{\left(1+2r_t^2(1+2\beta_t^2)\right)\e^{2\alpha_t-\frac{1}{2}}}{\sqrt{\e^4-1}}
\end{align*}
The univariate marginal distributions of $X_0$ and $X_t$ are identical and their cdf is given by the integral
\begin{equation}\label{eq:CDF_vol_lognorm}
	F(x)=\int_{-\infty}^{\infty}\varphi(\omega)\Phi(\frac{x}{\e^{\omega}}) \d\omega.
\end{equation}
Expanding the bivariate cdf (or the copula) in the small dependence parameters $\alpha_t,\beta_t,\rho_t$ around $(0,0,0)$, we get
%\numparts\label{eq:cop_perturb}
\begin{align}\label{eq:cop_perturb1}
	\cop[t](u,v)-uv  &\approx \alpha_t A(u,v) - \beta_t B(u,v) + \rho_t R(u,v)\\\nonumber\label{eq:cop_perturb2}
	               &\approx \alpha_t\tilde{A}(u)\tilde{A}(v)-\beta_t\tilde{R}(u)\tilde{A}(v)+\rho_t\tilde{R}(u)\tilde{R}(v)
\end{align}
%\endnumparts
where $\tilde{A}(u)$ was defined above in Equation~(\ref{eq:def_A}), and 
\begin{align}
	\tilde{R}(u)&=\int_{-\infty}^{\infty}\varphi(\omega)\varphi(\frac{F^{-1}(u)}{\e^\omega})\d\omega =\tilde{R}(1-u).
\end{align}
The contributions of $A(u,v), B(u,v)$ and $R(u,v)$ on the diagonal are illustrated in Figure~\ref{fig:correctionsToIndep}.
Notice that the term $B(u,v)$ (coming from cross-correlations between $\xi_0$ and $\omega_t$, i.e.\ the so-called leverage effect, see below) breaks the symmetry $\cop[t](u,v) \neq \cop[t](v,u)$.

\begin{figure}
    \centering
	\includegraphics[scale=0.4]{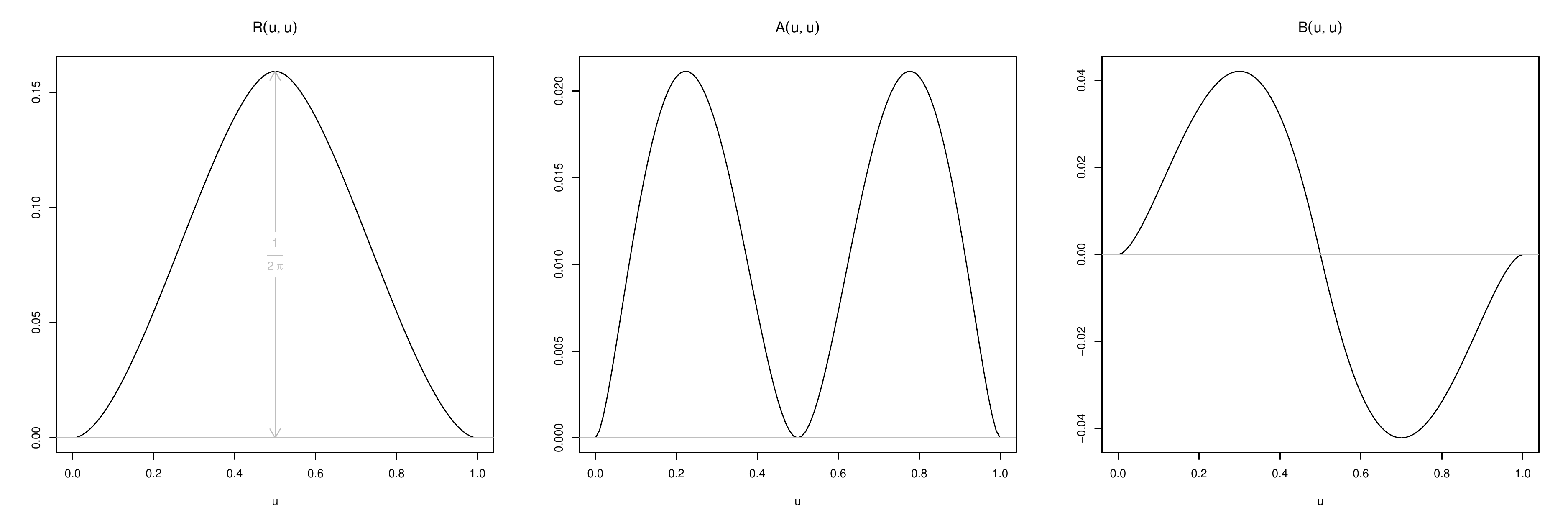}
	\caption{Copula diagonal of the log-normal volatility model: linear corrections to independence.
	         \textbf{Left:  } correction $R(u,u)$ due to correlation of the residuals (vertical axis in multiples of $\rho$)
	         \textbf{Middle:} correction $A(u,u)$ due to correlation of the log-vols (vertical axis in multiples of $\alpha$)
	         \textbf{Right: } correction $B(u,u)$ due to leverage effect (vertical axis in multiples of $-\beta$)}
	\label{fig:correctionsToIndep}
\end{figure}

%\newpage
\section{Application to financial time series}\label{sec:application}

\subsection{Stylized facts of daily stock returns}

One of the contexts where long-ranged persistence is present is time series of financial asset returns. 
At the same time, the empirical determination of the distribution of these returns is of utmost importance, 
in particular for risk control and derivative pricing. 
As we will see, the volatility correlations are so long-ranged that the number of effectively independent observations is strongly reduced, 
in such a way that the GoF tests are not very tight, 
even with time series containing thousands of raw observations. 

It is well-known that stock returns exhibit dependences of different kinds:
\begin{itemize}
\item at relatively high frequencies (up to a few days), returns show weak, but significant negative linear auto-correlations (see e.g.\ \cite{avramov2006liquidity});
\item the absolute returns show persistence over very long periods, an effect called multi-scale volatility clustering and
for which several interesting models have been proposed in the last ten years \cite{pasquini1999multiscaling,calvetfisher,bacrymuzy,zumbach2001heterogeneous,lynch2003market,borland2005dynamics};
\item past negative returns favor increased future volatility , an effect that goes under the name of ``leverage correlations'' in the
literature \cite{bouchaud2001more,perello2004multiple,pochart2002skewed,eisler2004multifractal,ahlgren2007frustration,reigneron2011principal}.
\end{itemize}

Our aim here is neither to investigate the origin of these effects 
and their possible explanations in terms of behavioral economics, nor to propose a new 
family of models to describe them. We rather want to propose a new way to characterize and 
measure the full structure of the temporal dependences of returns based on copulas, and
extract from this knowledge the quantities needed to apply  GoF tests to financial times series.

Throughout this section, the empirical results are based on a data set consisting of
the daily returns of the stock price of listed large cap US companies.
More precisely we keep only the 376 names present in the S\&P-500 index constantly over the five years period 2000--2004, corresponding to $N=1256$ days.
The individual series are standardized, but this does not change the determination of copulas, 
that are invariant under increasing and continuous transformations of the marginals.

\subsection{Empirical self-copulas}

For each $(u,v)$ on a lattice, we determine the lag dependent ``self-copula'' $\cop[t](u,v)$ by assuming stationarity, 
i.e.\ that the pairwise copula $\cop[nm](u,v)$ only depends on the time lag $t=m-n$. 
We also assume that all stocks are characterized by the same self-copula, 
and therefore an average over all the stock names in the universe is done in order to remove noise. 
Both these assumptions are questionable and we give at the end of this section an insight on
how non-stationarities as well as systematic effects of market cap, liquidity, tick size, etc, 
can be accounted for.

The self-copulas are estimated non-parametrically with a bias correction%
\footnote{Details on the copula estimator and the bias issue are given in appendix.}, 
then fitted to the parametric family of log-normal copulas introduced in the previous section. 
We assume (and check a posteriori) that the weak dependence expansion holds, leaving us with three functions of time, 
$\alpha_t$, $\beta_t$ and $\rho_t$, to be determined. 
We fit for each $t$ the copula diagonal $\cop[t](u,u)$ to Equation~(\ref{eq:cop_perturb1}) above, 
determine  $\alpha_t$, $\beta_t$ and $\rho_t$, and test for consistency on the anti-diagonal $\cop[t](u,1-u)$. 
Alternatively, we could determine these coefficients to best fit $\cop[t](u,v)$ in the whole $(u,v)$ plane, 
but the final results are not very different. 
The results are shown in Figure~\ref{fig:diag-adiag-tau} for lags $t=1,8,32,256$ days. 
Fits of similar quality are obtained up to $t=512$.  

\begin{figure}
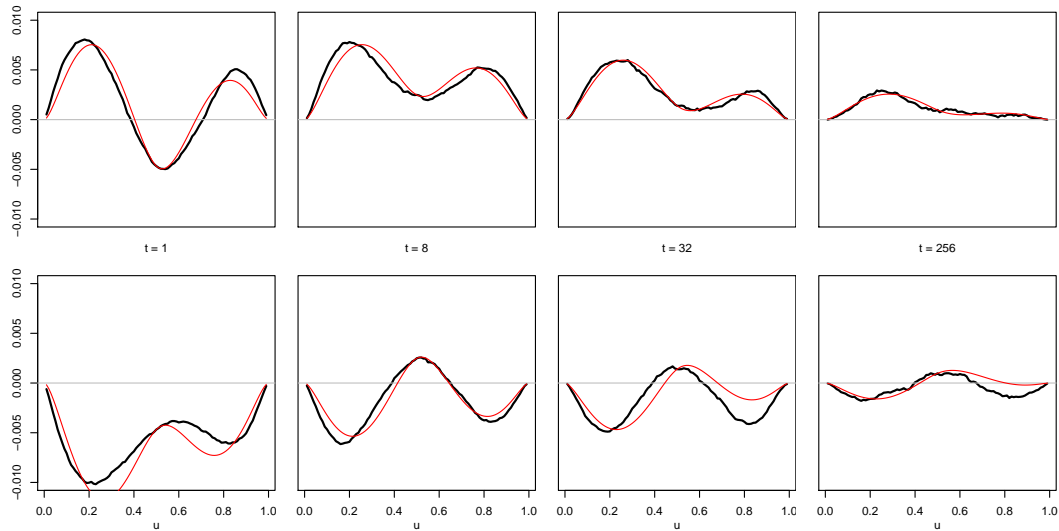

    \centering
	\includegraphics[scale=0.55,trim=   0 0 1620 0,clip]{GoFdependent/fig4}
	\includegraphics[scale=0.55,trim= 550 0 1075 0,clip]{GoFdependent/fig4}
	\includegraphics[scale=0.55,trim= 915 0  720 0,clip]{GoFdependent/fig4}
	\includegraphics[scale=0.55,trim=1450 0  175 0,clip]{GoFdependent/fig4}
	\caption{Diagonal (top) and anti-diagonal (bottom) of the self-copula for different lags;
	the product copula has been subtracted. A fit with Equation~(\ref{eq:cop_perturb1}) is shown in thin red. 
    Note that the $y$ scale is small, confirming that the weak dependence expansion is justified. 
    The dependence is still significant even for $t \sim 500$ days!}
	\label{fig:diag-adiag-tau}
\end{figure}

Before discussing the time dependence of the fitted coefficients  $\alpha_t$, $\beta_t$ and $\rho_t$, 
let us describe how the different effects show up in the plots of the diagonal and anti-diagonal copulas. 
The contribution of the linear auto-correlation can be directly observed at the central point $\cop[t](\frac{1}{2},\frac{1}{2})$ of the copula.
It is indeed known \cite{chicheportiche2012joint} that for any pseudo-elliptical model (including the present log-normal framework) one has:
\[
	\cop[t](\tfrac{1}{2},\tfrac{1}{2})=\frac{1}{4}+\frac{1}{2\pi}\arcsin \rho_t.
\]
Note that this relation holds beyond the weak dependence regime. 
If $\betaB_t=\cop[t](\frac{1}{2},\frac{1}{2})-\frac{1}{4}$ 
--- this is in fact Blomqvist's beta coefficient \cite{blomqvist1950measure} --- 
the auto-correlation is measured by $\rho_t=\sin(2\pi \betaB_t)$.

The volatility clustering effect can be visualized in terms of the diagonals of the self-copula;
indeed, the excess (unconditional) probability of large events following previous large events of the same sign is 
$\big[\cop[t](u,u)-u^2\big]$ with $u<\frac{1}{2}$ for negative returns, and $u>\frac{1}{2}$ for positive ones.
On the anti-diagonal, the excess (unconditional) probability of large positive events following large negative ones is, 
for small $u<\frac{1}{2}$, the upper-left volume
$\big[\cop[t](u,1)-u\cdot 1\big]-\big[\cop[t](u,1\!-\!u)-u(1\!-\!u)\big]=u(1\!-\!u)-\cop[t](u,1\!-\!u)$ 
and similarly the excess probability of large negative events following large positive ones is 
the same expression for large $u>\frac{1}{2}$ (lower-right volume).
As illustrated on Figure~\ref{fig:diag-adiag-tau}, these four quadrants exceed the independent case prediction, 
suggesting a genuine clustering of the amplitudes, measured by $\alpha_t$.
Finally, an asymmetry is clearly present: 
the effect of large negative events on future amplitudes is stronger than the effect of previous positive events.
This is an evidence for the leverage effect: negative returns cause a large volatility, 
which in turn makes future events (positive or negative) to be likely larger. This effect is captured by the coefficient $\beta_t$.

The evolution of the coefficients  $\alpha_t$, $\beta_t$ and $\rho_t$ for different lags reveals the following properties:
i) the linear auto-correlation $\rho_t$ is short-ranged (a few days), and negative;
ii) the leverage parameter $\beta_t$ is short-ranged and, as is well known, negative, revealing the asymmetry discussed above;
iii) the correlation of volatility is {\it long-ranged} and of relatively large positive amplitude (see Figure~\ref{fig:alpha_t}), in line with the known long range 
volatility clustering. 
More quantitatively, we find that the parameter $\alpha_t$ for lags ranging from $1$ to $768$ days
is consistent with an effective relation 
well fitted by the ``multifractal''  \cite{muzy2001multifractal,calvetfisher,lux,bacrymuzy} prediction for the 
volatility autocorrelations: $\alpha_t=-\Sigma^2\ln\frac{t}{T}$, with an amplitude $\Sigma^2=0.046$ and a horizon $T=1467$ days 
consistent, in order of magnitude, with previous determinations.

\begin{figure}
    \centering
	\includegraphics[scale=0.4]{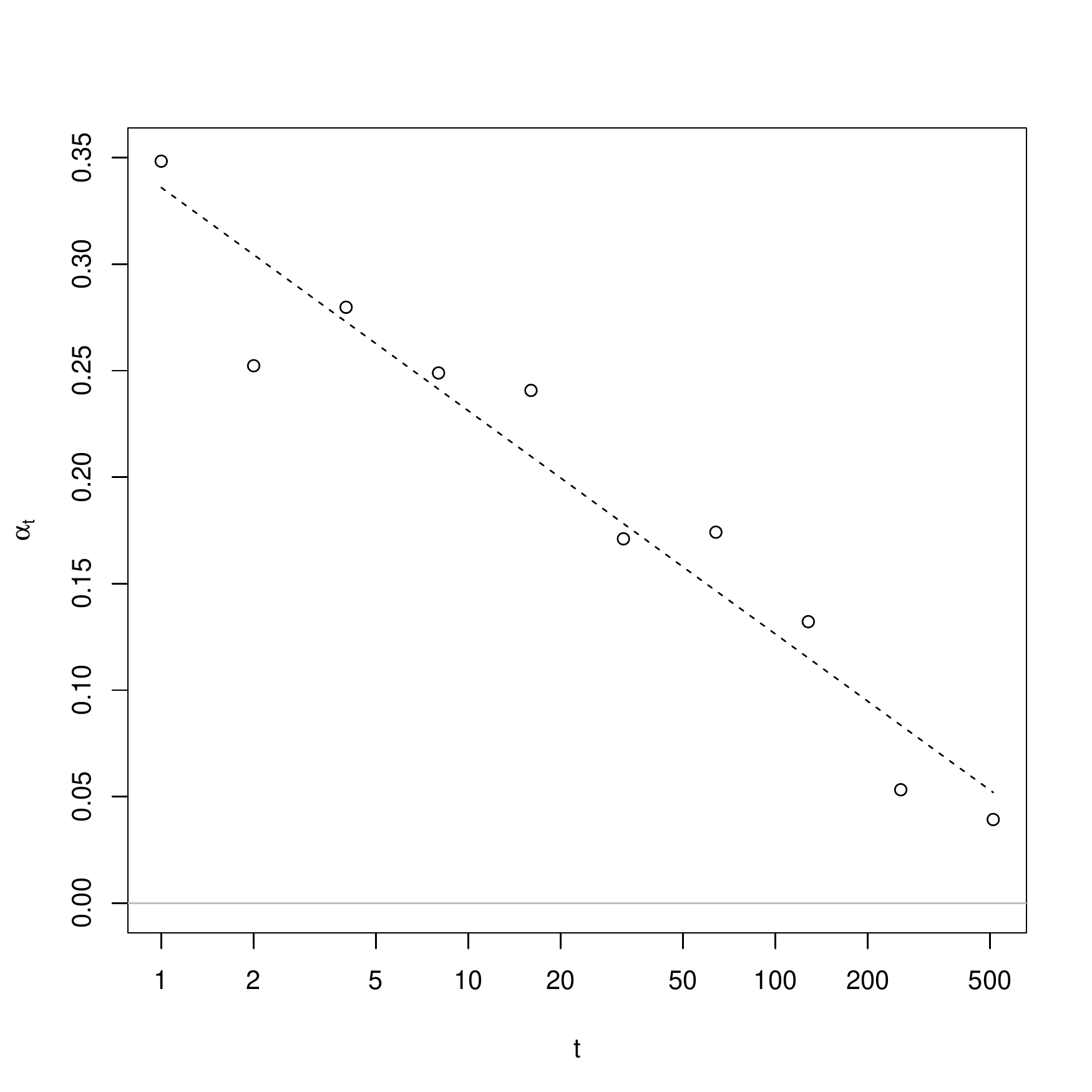}
	\caption{Auto-correlation of the volatilities, for lags  ranging from $1$ to $768$ days. 
		Each point represents the value of $\alpha_t$ extracted from a fit of the empirical copula diagonal 
		at a given lag to the relation (\ref{eq:cop_perturb1}). We also show the fit to a multifractal model, 
		$\alpha_t=-\Sigma^2\ln\frac{t}{T}$, with $\Sigma^2=0.046$ and $T=1467$ days.}
	\label{fig:alpha_t}
\end{figure}

The remarkable point, already noticed in previous empirical works on multifractals \cite{muzy2001multifractal,duchon2010forecasting}, 
is that the horizon $T$, beyond which the volatility correlations vanish, is found to be extremely long. 
In fact, the extrapolated horizon $T$ is larger than the number of points of our sample $N$! 
This long correlation time has two consequences: 
first, the parameter ${2}T\Sigma^2 f_\infty$ that appears in the kernel $H(u,v)$ is large, $\approx {135}$. 
This means that the dependence part $T\Sigma^2 f_\infty \, A(u,v)$ is dominant over the independent Brownian bridge part $I(u,v)$. 
This is illustrated in Figure~\ref{fig:eigen_vects}, where we show the first eigenvector of $H(u,v)$, 
which we compare to the non-zero eigenmode of $A(u,v)$, and to the first eigenvector of $I(u,v)$. Second, the hypothesis of a stationary process, which 
requires that $N \ll T$, is not met here, so we expect important pre-asymptotic corrections to the above theoretical results.

\begin{figure}
    \centering
	\includegraphics[scale=0.5]{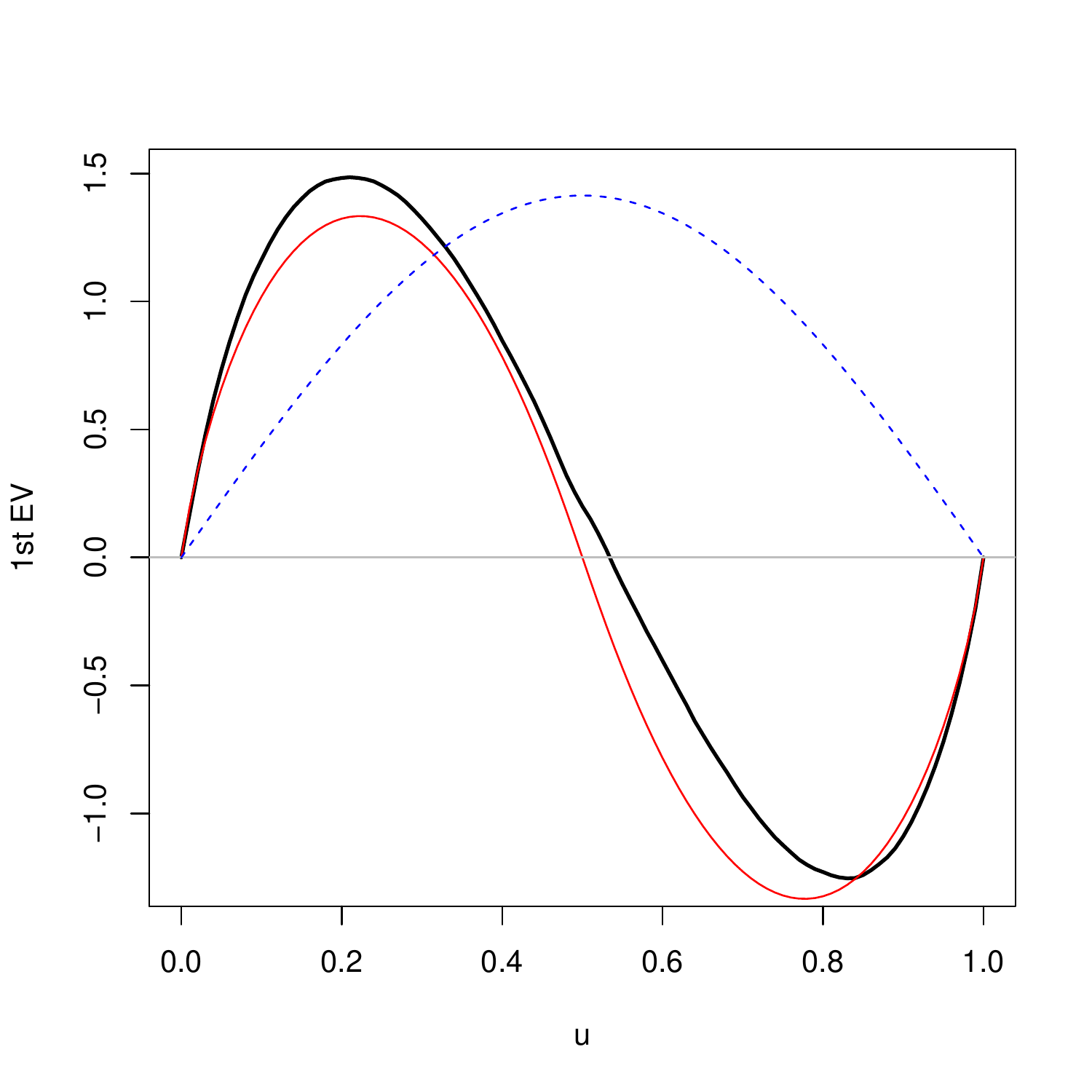}
	\caption{\textbf{Bold black:}
		 The first eigenvector of the empirical kernel $H(u,v)=I(u,v) \left[1 +\Psi_N(u,v)\right]$.
	         \textbf{Plain red:}
	         The function $\tilde{A}(u)$ (normalized), corresponding to the pure effect of volatility clustering in a log-normal model, in the
	         limit where the Brownian bridge contribution $I(u,v)$ becomes negligible. 
	         \textbf{Dashed blue:}
	         The largest eigenmode $\ket{1}=\sqrt{2}\sin(\pi u)$ of the independent kernel $I(u,v)$.
	         }
	\label{fig:eigen_vects}
\end{figure}

\subsection{Monte-Carlo estimation of the limit distributions}

Since $H(u,v)$ is copula-dependent, and considering the poor analytical progress made 
about the limit distributions of $K\!S$ and $C\!M$ in cases other than independence, 
the asymptotic laws will be computed numerically by Monte-Carlo simulations (like in the example of Section~\ref{sec:example})
with the empirically determined $H(u,v)$. 

The empirical cumulative distribution functions of the statistics for a large number of trials are shown in Figure~\ref{fig:ecdf_stats} 
together with the usual Kolmogorov-Smirnov and Cram\' er-von-Mises limit distributions 
corresponding to the case of independent variables.
One sees that the statistics adapted to account for dependences are stretched to the right,
meaning that they accept higher values of $K\!S$ or $C\!M$ (i.e.\ measures of the difference between the true and the empirical distributions).
In other words, the outcome of a test based on the usual $K\!S$ or $C\!M$ distributions is much more likely to be negative, 
as it will consider ``high'' values (like 2--3) as extremely improbable, whereas a test that accounts for the strong dependence in the time series 
would still accept the null-hypothesis for such values.

\begin{figure}
    \centering
	\includegraphics[scale=0.45]{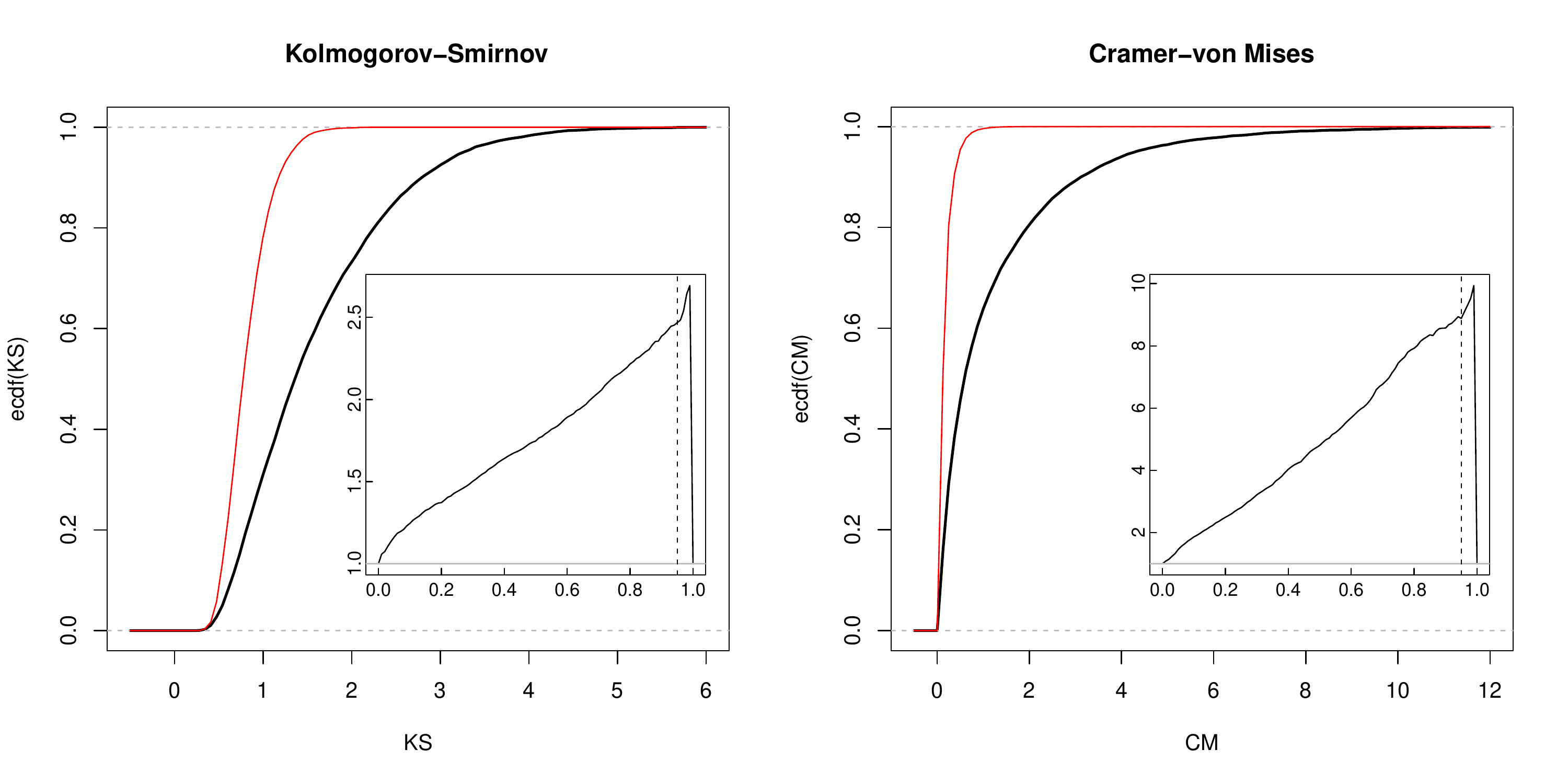}
	\caption{\textbf{Left:} Cumulative distribution function of the supremum of $\tilde{y}(u)$.
	         \textbf{Right:} Cumulative distribution of the norm-2 of $\tilde{y}(u)$.
	         The cases of independent draws (thin red) and dependent draws (bold black) are compared. 
	         The dependent observations are drawn according to the empirical average self-copula of US stock returns in 2000-2004. 
	         \textbf{Insets:} The effective reduction ratio $\sqrt{\frac{N}{N_{\mbox{\scriptsize{eff}}}(u)}}=\frac{\textrm{ecdf}^{-1}(u)}{\textrm{cdf}_{\mbox{\scriptsize{L}}}^{-1}(u)}$ 
	                          where $\textrm{L}=\textrm{KS,CM}$.}
	\label{fig:ecdf_stats}
\end{figure}

As an illustration, we apply the test of Cram\'er-von~Mises to our dataset, comparing the empirical univariate distributions of stock returns 
to a simple model of log-normal stochastic volatility, i.e.\ that the null-hypothesis cdf is similar to Eq.~\eqref{eq:CDF_vol_lognorm}:
\begin{equation}
	F(x)=\int_{-\infty}^{\infty}\varphi(\omega)\Phi(\frac{x}{\e^{s\omega-s^2}}) \d\omega.
\end{equation}
% \begin{equation}\label{eq:lognorm_vol}
%	 X=\e^{s\omega-s^2}\xi\qquad\textrm{where}\quad\xi,\omega\stackrel{\textrm{\tiny id}}{\sim}\mathcal{N}(0,1).
% \end{equation}
The volatility of volatility parameter $s$ can be calibrated from the time series $\{x_t\}_t$ as
\begin{equation}\label{eq:volvol}
	s^2=\ln\!\left(\frac{2}{\pi}\frac{\vev{x_t^2}_t}{\vev{x_t}_t^2}\right).
\end{equation}
We want to test the hypothesis that the log-normal model with a unique value of $s$ {\it for all stocks} is compatible with the data. 
In practice, for each stock $i$, $s_i$ is chosen as the average of (\ref{eq:volvol}) over all \emph{other} stocks
in order to avoid endogeneity issues and biases in the calculations of the \mbox{p-values} of the test. 
$s_i$ is found to be $\approx 0.5$ and indeed almost identical for all stocks. 
Then the GoF statistic $C\!M$ is computed for each stock $i$ and the corresponding \mbox{p-value} is calculated.%
\footnote{Another source of endogeneity is the estimated copula (either parametric or parameter-free) used in the asymptotic test law,
and is similar to the bias induced when applying usual univariate GoF test with an estimated null.
This issue can be dealt with by either estimating the copula on a subset and apply the test on another subset, 
or by using a parametric form for the copula, where the parameters are not estimated but scanned over given ranges.
A p-value is obtained for each combination of parameters, and optimal parameters can be determined among the accepted tests.
This of course can be very burdensome when the parameter space is large, and in addition it relies on a model for the copula.}

Figure~\ref{fig:KS_pvals} shows the distribution of the \mbox{p-values}, as obtained by using the usual asymptotic Cram\'er-von~Mises distribution 
for independent samples (left) and the modified version allowing for dependence (right).
We clearly observe that the standard Cram\'er-von~Mises test strongly rejects the hypothesis of a common log-normal model, 
as the corresponding \mbox{p-values} are concentrated around zero, which leads to an excessively high rejection rate. 
The use of the generalized Cram\'er-von~Mises test for dependent variables greatly improves the situation, with in fact now too many high values of $p$. 
Therefore, {\it the hypothesis of a common log-normal model for all stocks cannot be rejected} when the long-memory of volatility is taken into account. 
The overabundant large values of $p$ may be due to the fact that all stocks are in fact exposed to a common volatility factor 
(the ``market mode''), which makes the estimation of $s$ somewhat endogenous and generates an important bias. 
Another reason is that the hypothesis that the size of the sample $N$ is much larger than the correlation time $T$ does not hold for our sample, 
and corrections to our theoretical results are expected in that case.%
\footnote{Note that in practice,
we have estimated $\Psi_N(u,v)$ by summing the empirically determined copulas up to $t_{\max} = 512$, 
which clearly underestimates the contribution of large lags.}
It would be actually quite interesting to extend the above formalism to the long-memory case, where $T \gg N \gg 1$.

\begin{figure}
    \centering
	\includegraphics[scale=0.5,trim=350 0 0 400,clip]{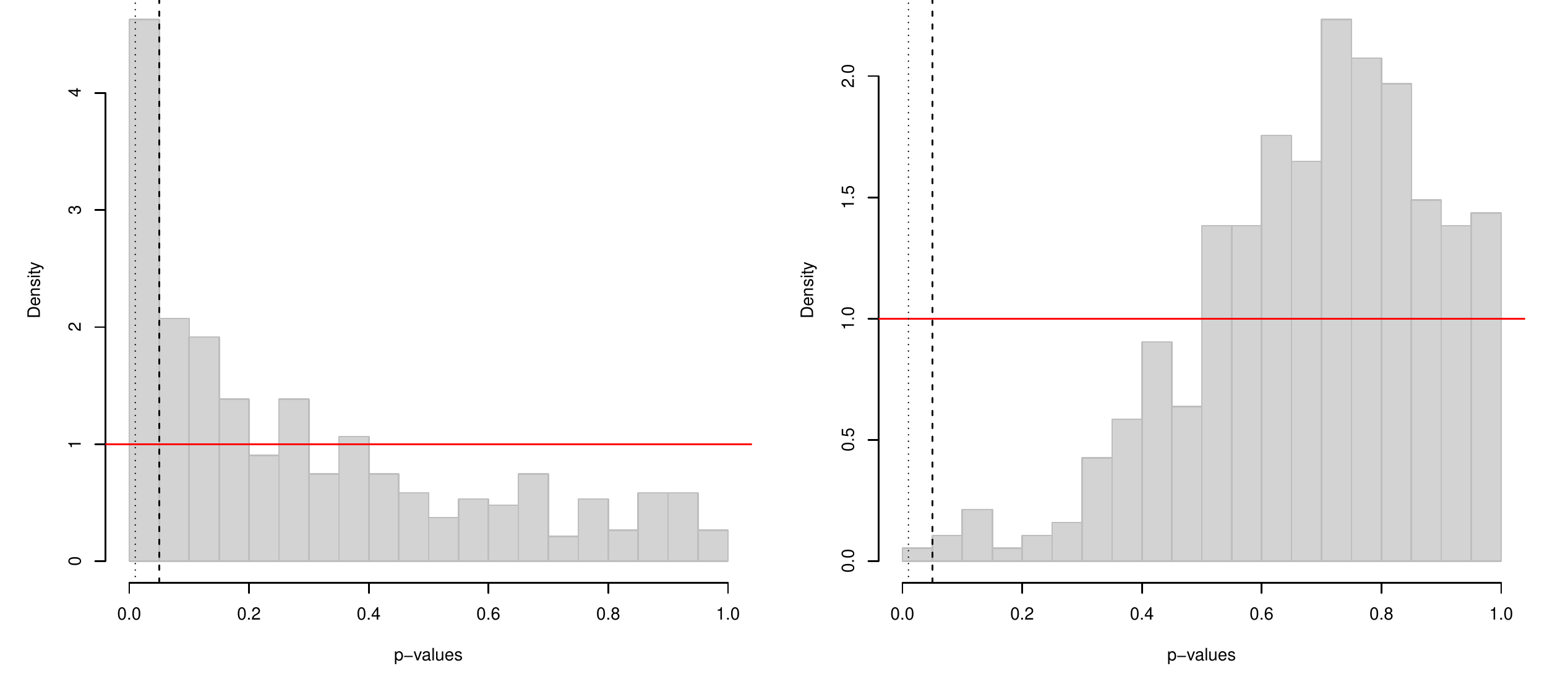}
	\caption{Histogram of the \mbox{p-values} in a Cram\'er-von~Mises-like test,
	         see the text for the test design.
	         \textbf{Left: } when using the standard Cram\'er-von~Mises law, the obtained \mbox{p-values} are
	         far from uniformly distributed and strongly localized under the threshold $p=0.05$ 
	         (dashed vertical line) occasioning numerous spurious rejections.
	         \textbf{Right: } when using the modified law taking dependences into account, 
	         the test rejects the hypothesis of an identical distribution for all stocks much less often.
	         }
	\label{fig:KS_pvals}
\end{figure}

\subsection{{Beyond stationarity and universality}}

The assumption that financial time series are stationary is far from granted.
In fact, several observations suggest that financial markets operate 
at the frontier between stationarity and non-stationarity. 
The multifractal model, for example, assumes that correlations are decaying 
marginally slowly (as a logarithm), which technically corresponds to the boundary 
between stationary models (when decay is faster) and non-stationary models. 
Furthermore, as stated above, the horizon $T$ that appears in the model is 
empirically found to be very long (on the order of several years) and therefore 
very difficult to measure precisely. 
Other models, like the ``multi-scale'' GARCH \cite{zumbach2001heterogeneous,borland2005multi,chicheportiche2012fine}, 
point in the same direction: when these models are calibrated on financial data, 
the parameters are always found to be very close to the limit of 
stability of the model (see \cite{borland2005multi} for a discussion of this point). 
Several very recent studies on high-frequency data point in the same direction, 
and even consider non-stationary extensions of the ``critically stationary'' multifractal or self-exciting models \cite{morales2012non,muzy2013random,hardiman2013critical}.

Furthermore, what is relevant in the present context is strong stationarity, 
i.e.\ the stationarity of the full self-copula. 
Hence, testing for stationarity amounts to comparing copulas, 
which amounts to a GoF for two-dimensional variables \ldots 
for which general statistical tools are missing, even in the absence of strong dependence! 
So in a sense this problem is like a snake chasing its own tail. 
A simple (and weak) argument is to perform the above study in different periods. 
When fitting the multifractal model to the self-copulas, we find the following values 
for $(\Sigma^2,\ln T)$: $(0.032,6.39)$ in 1995-1999, $(0.046,7.29)$ in 2000-2004, $(0.066,8.12)$ in 2000-2009 and $(0.078,7.86)$ in 2005-2009.
These numbers vary quite a bit, but this is exactly what is expected with the multifractal model itself 
when the size of the time series $N$ is not much larger than the horizon $T$! 
(see \cite{muzy2006extreme,bacry2006asset} for a detailed discussion of the finite $N$ properties of the model). 

As of the universality of the self-copula across the stocks, it is indeed a reasonable assumption that allows one to average over the stock ensemble.
We compared the averaged self-copula on two subsets of stocks obtained by randomly dividing the ensemble, and found the resulting copulas to be hardly distinguishable.
This can be understood if we consider that all stocks are, to a first approximation, driven by a common force --- the ``market'' --- and subject to idiosyncratic fluctuations.
We know that this picture is oversimplified, and that systematic effects of sector, market cap, etc. are expected and could in fact be treated separately 
inside the framework that we propose. Such issues are related to the cross-sectional non-linear dependence structure of stocks, and are far beyond the 
objectives of this article.

%\newpage
\section{Conclusion}
From the empirical estimation of the self-copula of US stock returns, 
long-ranged volatility clustering with multifractal properties is observed
as the dominant contribution to self-dependence, in line with previous studies. 
However, sub-dominant modes are present as well and a precise understanding of those involves an 
in-depth study of the spectral properties of the correlation kernel $H$.

One of the remarkable consequences of the long-memory nature of the volatility is 
that the number of effectively independent observations is significantly reduced, 
as both the Kolmogorov-Smirnov and Cram\'er-von~Mises tests accept much larger values of the deviations (see Figure~\ref{fig:ecdf_stats}). 
As a consequence, it is much more difficult to reject the adequation between any reasonable statistical model and empirical data. 
We suspect that many GoF tests used in the literature to test models of financial returns are fundamentally flawed 
because of the long-ranged volatility correlations. 
In intuitive terms, the long-memory nature of the volatility can be thought of as a sequence of volatility regime shifts, 
each with a different lifetime, and with a broad distribution of these lifetimes. 
It is clear that in order to fully sample the unconditional distribution of returns, 
all the regimes must be encountered several times. 
In the presence of volatility persistence, therefore, the GoF tests are much less stringent, 
because there is always a possibility that one of these regimes was not, or only partially, sampled. 
The multifractal model is an explicit case where this scenario occurs.

In terms of financial developments, we believe that an empirical exploration of the self-copulas 
for series of diverse asset returns and at diverse frequencies 
is of primordial importance in order to grasp the complexity of the non-linear time-dependences.
In particular, expanding the concept of the self-copula to pairs of assets is likely to reveal subtle dependence patterns.
From a practitioner's point of view, a multivariate generalization of the self-copula could lead to important progresses on 
such issues as causality, lead-lag effects and the accuracy of multivariate prediction.

\part*{Conclusion}
\clearpage\chapter{General conclusion and outlooks}\label{chap:Conclusion}

Partial conclusions were given at the end of every chapter,
summarizing the latter's content, discussing its conclusions in light of recent developments in the literature,
 and suggesting directions for possible extensions.

Here I want to give a general overview of the topics studied during the last three years
--- without reviewing in detail the original contributions in the thesis%
\footnote{See, for this, the abstracts in the Introduction, pages~\pageref{abstracts:begin}--\pageref{abstracts:end}.} ---
and situate them in the current state of the research.

\subsubsection*{Stochastic processes and applications}
The link between the theory of Goodness-of-fit tests and stochastic processes is a beautiful one, 
yet a poorly known one possibly due to the distance between the two concerned communities of statisticians and physicists.
Understanding the law of the Kolmogorov-Smirnov statistics as the survival probability of 
a randomly walking particle in an expanding absorbing cage allows to use techniques from statistical physics 
to solve problems in probability theory, which otherwise necessitate lengthy proofs.
More importantly, the analogy turns particularly useful in variants of the initial problem:
weighting the Kolmogorov-Smirnov statistics in order to increase the test resolution in specific ranges of the domain
in fact amounts to switching on a confining potential inside the cage.
In higher dimensions (i.e.\ when designing tests for multivariate distributions), the link to
random walkers is made nonintuitive by the trajectories being parameterized by more than one time index.
The lack of an obvious meaning to Markovianity or even to `causality' in such settings, 
is responsible for seemingly simple associated problems to remain unsolved: 
for example, the propagator of the 2D Brownian sheet (pinned or not) constrained between floor and ceiling is still unknown.
However, the discretization of one time direction suggests new perspectives to addressing the 2D GoF tests
via the survival of a random walker in an infinite-dimensional absorbing cage, see Appendix~\ref{chap:apx_KS2D}.

Other properties of stochastic processes like first-passage times can be encountered in financial applications. % \cite{chicheportiche2013fpt}.
Typically, problems related to optimal selling times of an asset, exercising of an American option, 
barriers and stop-loss criteria, etc.\ explicitly involve the statistics of first-passages and/or recurrence intervals 
in threshold-crossing processes.
Less obvious is the link between optimal control (e.g.\ in trading systems) and first-passage problems: 
it turns out that the stationary solution of the Bellman equation can sometimes be addressed, in a continuous limit, 
as a Kolmogorov or Fokker-Planck equation with vanishing Dirichlet boundary conditions. \cite{chicheportiche2013fpt}

\subsubsection*{Multivariate statistics and time series analysis}

The study of multivariate data analysis typically begins with the structure of linear correlations.
Disentangling signal and noise in estimated covariance matrices is an old yet always actual problem.
Factor models, spectral analysis, and Random Matrix Theory are some of the tools that were used in 
this thesis to address the linear properties of transversal dependences. 
There is a growing interest in the ``sparse'' symmetric matrices $\Sigma=\frac{1}{T}X^\dagger X$, 
with the $(T\times N)$ random entries $X_{ti}$ distributed according to an asymptotically power-law distribution,
or even as Bernoulli variables with very low probability $p$ \emph{not scaling with $N$}.
As opposed to the GOE and GUE ensembles, such matrices have localized eigenvectors with rich properties.
As a consequence of this non-invariance under rotations, methods relying on free matrices and R-transforms to 
find the spectrum of $\Sigma$ in the limit of large matrices are not applicable. 
The replica trick provides partial results, but new methods from statistical condensed-matter physics like the `cavity approach'
for interacting particles on graphs has proven successful in obtaining precise results \cite{rogers2008cavity}: 
yet another bridge between cutting edge research in probability theory and statistical physics.

This thesis dissertation presented many more theoretical and applied topics in connection to 
multivariate time series modeling, including extreme value statistics and tail dependences, 
ellipticity, optimal portfolio design, non-linear dependences as measured by Blomqvist's, Kendall's or Spearman's coefficients.
In this context, the fashionable `copula' has proved a useful unifying theoretical framework, 
though issues remain on the estimation and interpretation of this somewhat complex object.

Other modern ways of addressing inter-dependences are attracting interest: 
graphs and complex networks are an alternative to simple undirected pairwise links.
The higher and higher connectedness of the world has obvious illustrations in the Web 2.0 and the global financial crisis, to name only two.
Social networks, systemic risk, inter-bank lending structure, shareholding capital and voting rights webs, 
corporate client-contractor relationships, genetic and proteomic bipartite interactions, etc.\ can all be studied and modeled on graphs.
The emergence of huge databases and the trendy field of ``Big Data'' will certainly allow to complement 
theoretical constructions by empirical evidences and help the global pictures rely on microscopical foundations,
be it in the fields of finance, economics, sociology, medicine or others.

%\subsubsection*{Time series analysis}

\subsubsection*{Finance}
Beside the distributional properties of financial returns, 
many of the most important research subjects of quantitative finance deal with dependences:
the cross-sectional dependences must be well understood and modeled for optimal portfolio design and risk management, 
and the temporal dependences are exploited for signal prediction and optimal trading execution.

My first attempt at describing cross-sectional non-linear correlations was with hierarchical models, 
where each stock name is the leaf of a tree sharing branches with other leaves. 
The dependences between the leaves is hence caused by a ``common factor'' mechanism.
Although such a construction allowed for a powerful interpretation of the dependence structure and reproduced 
many required properties of dependences among stock return series, the fine structure of dependences could not
be well described. 
The alternative that was then finally chosen is a multi-factor model with a flat structure.
This construction, although less transparent from the point of view of economic interpretation, 
has several advantages: first it builds upon the traditional Principal Components Analysis, what provides a convenient benchmark;
then it does not involve such a hidden structures like a ``tree'' (although the statistical factors do not have an obvious meaning either);
and more importantly it conveniently reproduces many linear and non-linear properties of financial datasets, as studied at length in Part~\ref{part:partII}.

As of the temporal dependence, all the studies in the literature point toward the same conclusion:
financial time series have a long memory despite their being linearly uncorrelated!
In my own studies (Part~\ref{part:partIII}) this showed up in many places:
(i)   in the long-memory kernel of the auto-regressive mechanism of the volatility dynamics;
(ii)  in the waiting time and recurrence intervals statistics; and 
(iii) in the estimation of the self-copula in connection with the effective sample size diminution in GoF tests.
These many different approaches can be complemented by at least two others: the multifractal volatility models and non-stationary recent extensions thereof, 
and the Hawkes self-exciting mechanism. 
A convergence of (some of) these models is expected in the near future: 
first, ARCH and Hawkes description are conceptually similar, and in fact are equally able to generate power-law correlations with a power-law decreasing feedback kernel, see Appendix~\ref{QARCHapx:B} and Refs.~\cite{bacry2012non,filimonov2013effective};
and secondly auto-regressive deterministic constructions must be augmented by a stochastic volatility component {\it \`a la} 
multifractal in order to reconciliate the Time Reversal Asymmetry with empirical low (but finite) values.

Present and future developments in econophysics and financial data analysis include 
understanding the same properties of temporal and cross-sectional dependences but at higher frequency.
In this respect, I am involved in collaborations in view of assessing the finer structure of volatility dynamics 
at a scale lower than the day and develop intraday (Q)ARCH models \cite{Pierre_inprep}, 
and also apply GoF tests on high frequency data to address the distributional properties of order books \cite{gould2013power}.

%\backmatter

\appendix

\cleardoublepage
\phantomsection %if hyperref is used
\addcontentsline{toc}{part}{Appendices}
\part*{Appendices} 
\renewcommand{\cdf}[1][]{F_{\ifthenelse{\isempty{#1}}{}{#1}}}
\renewcommand{\qdf}[1][]{F_{\ifthenelse{\isempty{#1}}{}{#1}}^{-1}} %{\mathcal{Q}_{#1}}
%%% je dois changer de notations car trop confus avec des X,Y qui trainent partout, 
%%% et si je remplace par X_1,X_2 je ne sais plus comment indicer les realisations du couple ...

\chapter{Non-parametric copula estimator}
\label{chap:appendix1}

The copula $\cop(u,v)$ of a random pair $(X,Y)$ is
\[
	\cop(u,v)=\pr{\cdf[X](X)\leq u, \cdf[Y](Y)\leq v}=\Esp{\1{X\leq \qdf[X](u)}\1{Y\leq \qdf[Y](v)}}.
\]
If the univariate marginals $\cdf[X],\cdf[Y]$ are known, the usual empirical counterpart to the expectation operator can be used to define the empirical copula
over a sample of $N$ i.d. realizations $(X_n,Y_n)$ of the random pair:
\begin{equation}\label{eq:copest1}
	\widehat{\cop{}}(u,v)=\frac{1}{N}\sum_{n=1}^N\1{X_n\leq \qdf[X](u)}\1{Y_n\leq \qdf[Y](v)}
\end{equation}
which is clearly unbiased for any $N$.
But if the marginals are unknown, they have to be themselves estimated, for example by their usual empirical counterpart. 
Since $\cdf[X](x)=\pr{X\leq x}=\esp{\1{X\leq x}}$, an unbiased ecdf is obtained as
$\widehat{\cdf[X]}(x)=\frac{1}{N}\sum_{i=1}^N\1{X_i\leq x}$, but now the expected value of
\begin{equation}\label{eq:copest2}
	\widehat{\cop{}}(u,v)=\frac{1}{N}\sum_{i=1}^N\1{\widehat{\cdf[X]}(X_i)\leq u}\1{\widehat{\cdf[Y]}(Y_i)\leq v}
\end{equation}
is not $\cop(u,v)$ anymore (only asymptotically is $\widehat{\cop{}}(u,v)$ unbiased), 
but rather
\begin{align*}
	\esp{\widehat{\cop{}}(u,v)}%&=\frac{1}{N}\sum_{i=1}^N\pr{\widehat{F}_X(X_i)\leq u,\widehat{F}_Y(Y_i)\leq v}\\
	                       %&=\frac{1}{N}\sum_{i=1}^N\Pr{1+\sum_{j\neq i}\1{X_j\leq X_i}\leq Nu,1+\sum_{j\neq i}\1{Y_j\leq Y_i}\leq Nv}\\
	                       %&=\frac{1}{N}\sum_{i=1}^N\iint dxdy\:\Pr{B_i(x)\leq Nu\!-\!1,B_i(y)\leq Nv\!-\!1}\pr{X_i=x,Y_i=y}\\
	                        &=\int\! \d{\cdf[XY](x,y)}\:\Pr{B(x)\leq Nu\!-\!1,B(y)\leq Nv\!-\!1},
\end{align*}
where $B_X(x)=\sum_{n<N}\1{X_n\leq x}$ has a binomial distribution $\mathcal{B}(p,N\!-\!1)$ with $p=\cdf[X](x)$
and is not independent of $B_Y(y)$.
As an example, the expected value for the estimator of the independence (product) copula $\cop(u,v)=\Pi(u,v)=uv$ is 
\begin{align*}
	\esp{\widehat{\Pi}(u,v)}%&=\left(\int dF(x)\pr{B(x)\leq Nu\!-\!1}\right)\left(\int dF(y)\pr{B(y)\leq Nv\!-\!1}\right)\\
	                       %&=\left(\int dp\sum_{k=0}^{\lfloor Nu\rfloor-1} \binom{N-1}{k}p^k(1-p)^{N-1-k}\right)\left(\ldots\right)\\
	                       %&=\left(\sum_{k=0}^{\lfloor Nu\rfloor-1}\binom{N-1}{k}\frac{k!(N-k-1)!}{N!}\right)\left(\ldots\right)\\
	                        &=\frac{\lfloor Nu\rfloor}{N}\frac{\lfloor Nv\rfloor}{N}%\equiv b(u,v) uv
\end{align*}
resulting in a relative bias $b(u,v)=\frac{\lfloor Nu\rfloor}{Nu}\frac{\lfloor Nv\rfloor}{Nv}-1$ vanishing only asymptotically.

As $\esp{\widehat{\cop{}}(u,v)}$ may not be computable in the general case, we define as 
\[
	\frac{1}{b(u,v)+1}\,\widehat{\cop{}}(u,v)=\frac{Nu}{\lfloor Nu\rfloor}\frac{Nv}{\lfloor Nv\rfloor}\,\widehat{\cop{}}(u,v)
\]
our non-parametric estimator of the copula with bias correction, even when $\cop$ is not the independence copula.
Therefore, our estimator is technically biased at finite $N$ but with a good bias correction, and asymptotically unbiased.

Notice that the copula estimator $\widehat{\cop{}}$ as defined in Eq.~\eqref{eq:copest1} 
\emph{is a copula} as long as $\qdf[X]$ and $\qdf[Y]$ are continuous, 
whereas $\widehat{\cop{}}$ as defined in Eq.~\eqref{eq:copest2} \emph{is not a copula} on $[0,1]^2$ for $N<\infty$.
Only on the discrete grid $\{\tfrac{1}{N},\tfrac{2}{N},\ldots,1\}^2$ does Eq.~\eqref{eq:copest2} define a copula,
provided the empirical inverse marginals define the ranks (or inverse ``order statistic'') \cite{deheuvels1979fonction,deheuvels1980};
in this case the bias $b(u,v)$ is zero.
Continuous interpolations are not trivial, see Sect.~5.1 in \cite{malevergne2006extreme}.
In practice the copula is numerically estimated on a grid with a resolution typically coarser than $\frac{1}{N}$.

\vfill

\begin{figure}[b!h!t]
    \center
	\includegraphics[scale=0.68,trim=0 0 0 15,clip]{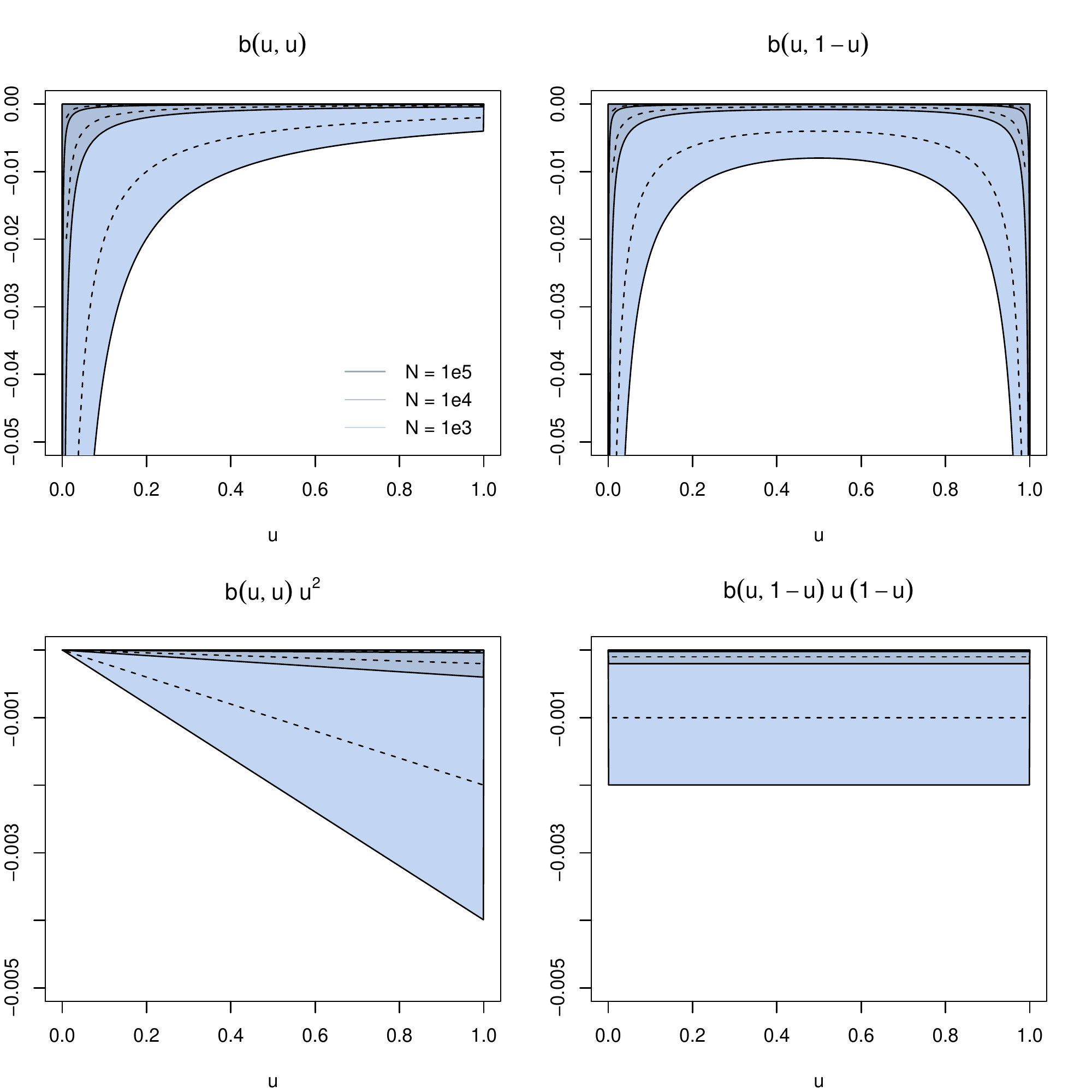}
	\caption{Bias of the non-parametric naive estimator of the independence copula with unknown marginals, 
	         for different sample sizes (see legend).
	\textbf{Top:} Relative bias.
	\textbf{Bottom:} Absolute bias.
    For a given sample size $N$, $\lfloor Nu\rfloor$ satisfies the following relations
    $
        Nu \geq {\lfloor Nu\rfloor} \approx Nu-1 \geq Nu-2,
    $
    where the equalities hold for some values of $u\in[0,1]$ !
    The dashed lines represent the approximate value, and the plain lines show the bounding envelope.
    }
	\label{fig:bias}
\end{figure}

\chapter[Gaussian sheets and bivariate distributions testing]{Gaussian sheets\\ and bivariate distributions testing}\label{chap:apx_KS2D}
Goodness-of-fit testing for univariate probability distributions has been addressed in
Chapter \ref{chap:GOF} of Part~\ref{part:partI}, both for independent and dependent samples.
%\textcolor{red}{Mentioner \cite{deheuvels2005weighted} qui generalise \cite{anderson1952asymptotic} en plusieurs dimensions (en trichant un peu)}
%
%
%%%%%%%%%%%%%%%%%%%%%%%%%%%%%%%%%%%%%%%%%%%%%%%%%%%%%%%%%%%%%%%%%%%%%%%%
We address in this appendix the generalization to multivariate samples.
%\cite{kole2007selecting}
%Saunders-Laud \cite{saunders1980mutltidimensional} (r\'efut\'e in \cite{leurgans1983distribution}).

We assume that the marginals were previously tested, and can work on the ranks only.

\subsubsection*{Empirical cumulative distribution and its fluctuations}
Consider a bivariate sample of $N$ joint observations $(U_n,V_n)$ with uniform marginals.
The empirical estimate of its copula,
\begin{equation}\label{eq:C_N}
    \cop[N](u,v)=\frac{1}{N}\sum_{n=1}^N\1{U_n\leq u}\1{V_n\leq v},
\end{equation}
converges to the true copula $\cop$ as the sample size $N$ tends to infinity,
but is biased for finite $N$.
% but for finite $N$, the expected value and fluctuations of $\cop[N](u,v)$ are%
% \begin{align}\label{eq:cop_mean}
	% \esp{\cop[N](u,v)}&=\cop(u,v),\\\nonumber%\label{eq:cop_fluct}
    % \mathrm{Cov}(\cop[N](u,v),\cop[N](u',v'))&=\frac{1}{N}\left[\cop(\min(u,u'),\min(v,v'))-\cop(u,v)\cop(u',v')\right].
% \end{align}
In the following, we only care for asymptotic and pre-asymptotic (large $N\gg 1$) properties of $\cop[N]$ so that the bias is not relevant,
but the bias-corrected estimator discussed in Appendix~\ref{chap:appendix1} should be considered when 
turning to practical applications with finite samples, i.e.\ when the ranks $U_n$ and $V_n$ are estimated rather than known.
In particular, the bias correction is exact for the product copula $\cop=\Pi$, 
what may be of practical interest for the independence test worked out in Section~\ref{sec:testofInd} below.

The rescaled empirical copula
\begin{equation}
    Y_N(u,v)=\sqrt{N}\,\left[\cop[N](u,v)-\cop(u,v)\right]
\end{equation}
measures, for a given $(u,v)\in[0,1]^2$, the difference between the empirically determined copula 
and the theoretical one, evaluated at the $(u,v)$-th quantiles.
It does not shrink to zero as $N\to\infty$, and is therefore 
an appropriate quantity on which to build a statistic for bivariate GoF testing.

\subsubsection*{Limit properties}
One now defines the process $Y(u,v)$ as the limit of $Y_N(u,v)$ when $N \to \infty$.
According to the Central Limit Theorem (CLT), 
it is Gaussian and its covariance function is given by:
\begin{align}\label{eq:HtheoC}
	H_{\cop{}}(u,v;u',v')&=\cop(\min(u,u'),\min(v,v'))-\cop(u,v)\cop(u',v')\\\nonumber
                         &=\begin{cases}
                            \cop(u ,v )\left(1-\cop(u',v')\right)   &, u \leq u', v \leq v'\\
                            \cop(u ,v')-\cop(u,v)\cop(u',v')        &, u \leq u', v'\leq v \\
                            \cop(u',v )-\cop(u,v)\cop(u',v')        &, u'\leq u , v \leq v'\\
                            \cop(u',v')\left(1-\cop(u,v)\right)     &, u'\leq u , v'\leq v
                           \end{cases},
\end{align}
and characterizes a pinned Gaussian sheet, i.e.\  a two-times Gaussian process $Y(u,v)$ such that 
$Y(u\!=\!0,v)=Y(u,v\!=\!0)=Y(u\!=\!1,v\!=\!1)=0$. 
More precisely, $H_{\cop{}}$ is the covariance function of the so-called \emph{pinned $\cop$-Brownian sheet},
and characterizes the statistical fluctuations of the copula estimate $\cop[N]$ on an infinitely large sample.
Importantly, Eq.~(\ref{eq:HtheoC}) makes clear the fact that these asymptotic fluctuations still depend on the null copula $\cop$ 
as opposed to the univariate case, where GoF tests are universal.

A global, scalar measure of distance between $\cop[N]$ and $\cop$ is provided for example by the 
 Cram\'er-von~Mises statistic
\begin{subequations}
\begin{equation}\label{eq:CM2D}
    CM=\iint_0^1\psi(u,v)\,Y(u,v)^2\d{u}\d{v},
\end{equation}
or by the Kolmogorov-Smirnov statistic
\begin{equation}\label{eq:KS2D}
    KS=\sup_{(u,v)\in[0,1]^2}\sqrt{\psi(u,v)}\,|Y(u,v)|,
\end{equation}
\end{subequations}
where $\psi(u,v)$ is a weight function accounting for possible inhomogeneities to be flattened in the domain.

Finding the laws of the functionals~\eqref{eq:CM2D} and~\eqref{eq:KS2D} (and many others) of the pinned $\cop$-Brownian sheet 
 has been the focus of intense research, but their exact distributions are still not known, 
 even in the simplest cases when $\psi(u,v)=1$ and/or $\cop=\Pi$.
 
 To overcome the intrinsic difficulty --- related in large part to the covariance \eqref{eq:HtheoC} being not factorizable in $u,v$ ---
 two main approaches have been followed:
 \begin{itemize}
 \item Either changing the functional of $\widetilde{Y}$: 
 computing other functionals \cite{cabana1994goodness}
 or investigating separate subdomains of $[0,1]^2$ \cite{fasano1987multidimensional}.
 In this respect, there have been attemps at building single-time transformations, 
 along a path $(u,\vartheta(u))$, or introducing more complex processes like Kendall's process \cite{genest1993statistical,fermanian2012overview};
 \item Or alternatively use another specification of the test that leads not to the pinned Brownian sheets,
but rather to more convenient processes. 
This means take into account not only the global spreads of the copula $|\cop[N]-\cop|$ like above, 
but also integrate information of the marginals.
This approach is used by Deheuvels in Ref.~\cite{deheuvels2005weighted} to design a Cramer-von~Mis\'es test
for bivariate copulas, and is in the spirit of what was done by Justel, Pe\~na, Zamra \cite{justel1997multivariate} 
who build upon the Probability Integral Transformation suggested by Rosenblatt~\cite{rosenblatt1952remarks}, see also \cite{genest2006goodness,fermanian2012overview}.
\end{itemize}

We choose to remain in the traditional setup of the Brownian sheets \eqref{eq:HtheoC},
and to take advantage of the weighting possibilities offered by $\psi(u,v)$
to compute the Kolmogorov-smirnov like statistic \eqref{eq:KS2D}.
We focus on the KS case because of the link with first-passages in stochastic processes, 
see Chapter~\ref{part:partI}.\ref{chap:GOF} and Ref.~\cite{chicheportiche2013fpt}.
Indeed, the law of the KS statistic can be written as the survival probability of a constrained process
$\widetilde{Y}(u,v)=\sqrt{\psi(u,v)}\,Y(u,v)$:
\begin{align}\label{eq:KS2D_law}
    S(k)&=\Pr{\sup_{(u,v)\in[0,1]^2}|\widetilde{Y}(u,v)|\leq k}\\\nonumber
        &=\Pr{-k\leq \widetilde{Y}(u,v)\leq k, \forall (u,v)\in[0,1]^2}.
\end{align}

\newpage
\section{Test of independence}\label{sec:testofInd}
Consider first the case where the null-hypothesis copula is the product copula
$\cop(u,v)=\Pi(u,v)=uv$, which is appropriate when one wants to test for independence of the variables.
As a particular case of Eq.~\eqref{eq:HtheoC}, $Y(u,v)$ has the covariance function
of a standard pinned Brownian sheet (pBs)\nomenclature{pBs}{Pinned Brownian sheet}:
\[
    I(u,v;u',v')\equiv H_\Pi(u,v;u',v')=\min(u,u')\min(v,v')-uu'vv'.
\]

The unweighted Kolmogorov-Smirnov-like statistic $\sup|Y(u,v)|$
has an unknown distribution
\[
    S(k)=\Pr{\sup_{(u,v)\in[0,1]^2}|\sqrt{\psi(u,v)} Y(u,v)|\leq k}, \qquad\psi(u,v)=1.
\]
Still, there exist bounds and approximations
 \cite{adler1986tail,adler1987tail}, in particular:
 %Attention dans les citations ne pas confondre le supabs du Bs et du pBs :
%ce n est pas un probleme, personne ne cherche le supabs du BS sur $[0,\infty]^2$, mais sur $[0,1]^2$ oui !
\begin{align}
    1-S(k)&\geq (1+2k^2)\,\e^{-2k^2}\\
    1-S(k)&\xrightarrow{k\gg \infty} 4\ln(2)\,k^2\,\e^{-2k^2},
\end{align}
and similarly, for the $\cop$-pBS only approximations and numeric solutions are available
\cite{adler1990introduction,aldous1989probability,greenberg2006bivariate,Fermanian2005119}.

% \textcolor{red}{
% s'acharner a trouver ces statistiques asymptotiques: 
    % soit brute force avec des methodes de processus stochastiques a deux temps 
    % (propriete de Markov a discuter ?, quelle equivalent a FP ?);
    % soit avec un peu de creativite, par des methodes de discretisation ramenant le probleme a des systemes physiques en dimension infinie
% }

\subsection{Pinning the Brownian sheet ($\psi(u,v)=1$)}\label{ssec:pinning}
The pBs is not the simplest 2-times process one could imagine.
Indeed, like any pinned $\cop$-Brownian sheet it has strong constraints at the borders
\[
    Y(u,0) = Y(0,v)=Y(1,1)=0.
\]
However, it can be simply written as $Y(u,v)=W(u,v)-uv\,W(1,1)$, where
$W(u,v)$ is a standard free Brownian sheet (two-times Wiener process) with covariance
\begin{equation}\label{eq:cov_BS}
    \esp{W(u,v)W(u',v')}=\min(u,u')\min(v,v'),
\end{equation}
and $W(u,0)=W(0,v)=0$. 
The possible trajectories (or surfaces) of $Y(u,v)$ are thus a subset of all possible paths gone by $W(u,v)$,
whose free propagator is 
\[
    \mathcal{P}_{(u,v)}(w;\infty)=\frac{1}{\sqrt{2\pi uv}}\Exp{-\frac{1}{2}\frac{w^2}{uv}}.
\]
But of course the survival probability \eqref{eq:KS2D_law} involves constrained rather than free processes:
\begin{equation}\label{eq:Bs2pBs}
    S(k)=\Pr{-k\leq Y(u,v)\leq k, \forall (u,v)\in[0,1]^2}=
    \frac{\mathcal{P}_{(1,1)}(0;k)}{\mathcal{P}_{(1,1)}(0;\infty)},
\end{equation}
where
\[
%   \mathcal{P}_{(u,v)}(w;k)\d{w}=\Pr{|W(u',v')|\leq k, \forall u<u',v<v'\ \text{ and }\ W(u,v)\in[w,w\!+\!\d{w}]}
    \mathcal{P}_{(u,v)}(w;k)     =\Pr{|W(u',v')|\leq k, \forall u<u',v<v'\ \text{ and }\ W(u,v)=w}.
\]
Equation~\eqref{eq:Bs2pBs} means that the law $S(k)$ of the supabs of the standard unweighted pinned Brownian sheet
is equal, up to a constant $\sqrt{2\pi}$, to the propagator of the Brownian sheet constrained in a stripe $[-k,k]$.

Unfortunately, we are not aware of a relevant Markov property that would make it possible to write (and eventually solve)
a Fokker-Planck-like equation with two times.
Instead, we present next an original way out, by discretizing one of the time directions.

\subsubsection{Law of the supabs of the Brownian sheet on $[0,1]^2$:\\ discretizing the time}\label{sec:sup.brownian_sheet}
%The Brownian bridge $B(u,v)=W(u,v)-uvW(1,1)$ on $[0,1]^2$. 

$W(u,v)$ can be seen as a continuous accumulation of innovations along the two time directions.
Discretizing the time in one of the dimensions, f.ex.\ forcing $v$ to take values in the sequence $(v_j)_{j\leq N}$ of length $N$, 
% it is possible to write, in the limit of infinitesimal increments, 
% \[
	% W(u,v)=\lim_{N\to\infty}\sum_{n=1}^{\lfloor Nv\rfloor} Z_n(u)
% \]
% where each $Z_n$ is an independent Gaussian innovation (White Noise) of variance $\sigma^2u/N$,
% such that
% \[
	% \var{W(u,v)}=\lim_{N\to\infty}\sum_{n,m}^{\lfloor Nv\rfloor}\esp{Z_n(u)Z_m(u)}=\lim_{N\to\infty}\sum_{n,m}^{\lfloor Nv\rfloor}\delta_{nm}\sigma^2 u/N=\sigma^2 u \lfloor Nv\rfloor/N=\sigma^2 uv
% \]
%as imposed by Eq.~\eqref{eq:cov_BS}.
% The joint dynamics of the $N$-dimensional vector 
% \[
    % \vect{W}(u)=\Big(Z_1(u),Z_1(u)+Z_2(u),\ldots,\sum\limits_{n=1}^NZ_n(u)\Big)
% \] 
%is described by vector SDE 
it is possible to write $W(u,v_j)=W_j(u)$, and collect all the components in a $N$-dimensional vector $\vect{W}$
whose dynamics is driven by the joint SDE
\[
    \d{\vect{W}}(u)=\rho^{\frac{1}{2}}\,\d{\widetilde{\vect{W}}}(u),
\]
where $\widetilde{\vect{W}}$ is a vector of independent Brownian motions, 
and $\rho$ is the tensor of diffusion, with elements%, or covariance matrix of the increments,
$$\rho_{jj'}=\min(v_j,v_j').$$%\frac{\sigma^2}{N}\min(j,j')$, and %we can write 
Then, the transition density $f_u(\vect{w};k)$ of the vector process $\vect{W}$ with all components constrained in $[-k,k]$
is solution of the multivariate Fokker-Planck equation 
\[
	\frac{\partial f_u(\vect{w};k)}{\partial u}=\frac{1}{2}\nabla^{\dagger} \rho\nabla f_u(\vect{w};k),\quad 0\leq u\leq 1
    %\qquad\text{and}\qquad -k\leq w_j\leq k,\quad\forall j=1\ldots N
\]
with the condition that $f_u(\vect{w};k)$ vanishes at the borders of the square box
\[
    \mathcal{D}_N=\big\{\vect{w} : |w_j|\leq k, \forall j\in\libracket 1,N\ribracket \big\}.
\]
Finally, the constrained two-times propagator writes
\begin{align*}
	\mathcal{P}_{(1,1)}(w;k)&=\lim_{N\to\infty}\Pr{\forall u\in[0,1], \forall j=1\ldots N, |W_j(u)|\leq k\ \text{ and }\ W_N(1)=w}\\
                            &=\lim_{N\to\infty}\int_{-k}^k\d{w_1}\cdots\int_{-k}^k\d{w_{N}}\,f_u(\vect{w};k)\,\delta(w_N-w).
\end{align*}
The problem is thus equivalent to the survival of a particle diffusing in an $N$-dimensional cage 
with absorbing walls at $\pm k$, and anisotropic diffusion tensor.
It can alternatively be stated as an isotropic diffusion toward the absorbing walls of an irregular box, 
by appropriately changing the system of coordinates.

\paragraph{Changing base: the isotropic problem.}
We first diagonalize $\rho$ in order to obtain an isotropic and rescaled diffusion, albeit in a cage with a new geometry. 
The corresponding eigenvalues are collected on the diagonal of $\Lambda$ and the eigenvectors in the columns of $\mat{U}$.
Changing base according to $\mat{U}\Lambda^{-1/2}$, any point writes $\hat{\vect{w}}=(\mat{U}\Lambda^{-1/2})^{\dagger}\vect{w}$ in the new coordinates.
When the discretized time stamps are equidistant ($v_j=j/N$), the diffusion coefficients are $\rho_{jj'}=\frac{1}{N}\min(j,j')$ and 
$\mat{Q}\equiv\frac{1}{N}\rho^{-1}=\frac{1}{N}\mat{U}\Lambda^{-1}\mat{U}^{\dagger}$ turns out to be 
the tridiagonal matrix corresponding to the discrete $N$-dimensional second-order finite difference operator:
$2$ on the main diagonal, $-1$ on both secondary diagonals, $1$ at position $(N,N)$, and $0$ elsewhere.
The eigenvectors and eigenvalues of $\rho$ are easily found to be
\begin{subequations}
\begin{align}
	\mat{U}_{ja}&=\sqrt{\frac{2}{N}}\sin\!\left(\frac{(2a-1)\pi j}{2N}\right)\\\label{eq:evals_rho}
	\lambda_a\equiv\Lambda_{aa}&=\frac{4N}{(2a-1)^2\pi^2}
\end{align}
\end{subequations}
In this new basis, the diffusion is isotropic
\begin{equation}
	\frac{\partial f_u({\hat{\vect{w}}};k)}{\partial u}=\frac{1}{2}\nabla^2 f_u(\hat{\vect{w}};k),\quad 0\leq u\leq 1
\end{equation}
but the box has a new reshaped geometry
\begin{equation}\label{eq:new_shape}
    -k\sqrt{\lambda_j}\leq \sum_a \mat{U}_{aj}\hat{w}_a\leq k\sqrt{\lambda_j},\quad\forall j=1\ldots N.
\end{equation}

The point on the boundary hypersurface that is the closest to the center is found minimizing the distance 
 $\hat{d}^2=\hat{\vect{w}}\cdot\hat{\vect{w}}=\vect{w}^{\dagger}\rho^{-1}\vect{w}$ to the center,
 under the constraint that at least one of the $|w_j|$ is equal to $k$
 (or equivalently maximizing the effective directional diffusion rate).
This special proximal point is found to have coordinates $w_j=k\,j/N=k\,\rho_{iN}$ and be at a distance $\hat{d}=k\sqrt{N}$ %/\sigma^2$. 
proportional to the square root of the number of faces (but recall that the box expands as $\sim \sqrt{N}$ too, see Eqs.~\eqref{eq:new_shape} and~\eqref{eq:evals_rho}).
\[
    \begin{array}{c||c|c||}
                                            & \text{Original box} & \text{Deformed geometry}\\\hline\hline
  \vect{w}_{\min} / \hat{\vect{w}}_{\min}   & k\,(1/N,2/N,\ldots,1)        & k\, (\mat{U}_{N1}/\sqrt{\lambda_1},\ldots,\mat{U}_{NN}/\sqrt{\lambda_N})\\
         d_{\min} /        \hat{d}_{\min}   & \frac{k}{N}\sqrt{\frac{N^3}{3}+\frac{N^2}{2}+\frac{N}{6}}                                                 & k\sqrt{N}\\
    \sigma_{\min} /     \hat\sigma_{\min}   & \sqrt{\frac{w_{\min}^{\dagger}w_{\min}}{w_{\min}^{\dagger}\rho^{-1}w_{\min}}}\approx\sqrt{\frac{N}{3}}    & 1\\
    \end{array}
\]

In fact, it is possible to calculate the distance from the center to every face. 
Looking for the vector $w^{(n)}$ such that
\begin{align*}
	w^{(n)}&=\argmin\{\hat{w}^{\dagger}    \hat{w} \ \text{ s.t. }\ w_n=k\}
	        =\argmin\{     w^{ \dagger}\rho^{-1}w  \ \text{ s.t. }\ w_n=k\}\\
	       &=\argmin\left\{k^2\mat{Q}_{nn}+2k\sum_{m\neq n}\mat{Q}_{nm}w_m+\sum_{m,m'\neq n}w_m'\mat{Q}_{m'm}w_m\right\}\\
	       &=\arg\left\{\sum_{m\neq n}\mat{Q}_{m'm}w_m=-k\,\mat{Q}_{m'n}, \ m'\neq n\right\},
\end{align*}
one finds
\[
        w^{(n)}_j=k\cdot\begin{cases}\frac{j}{n} &, j\leq n \leq N\\ 1&,n\leq j \leq N\end{cases},
\]
meaning that the closest point to \emph{any} face is on an edge, since  $w^{(n)}_j$ touches $N-n+1$ faces.
In particular, the closest exit in any direction is on the $N$-th face.
We already know that $\hat{d}_{\min}=\hat{d}^{(N)}=k\sqrt{N}$, and we now find more generally that
$\hat{d}^{(n)}=\sqrt{\hat{w}^{(n)\dagger}\hat{w}^{(n)}}=kN/\sqrt{n}$.%, voir Fig.~\ref{fig:closest_face} ! 

\newpage
\section{2D Kolmogorov-smirnov test of Goodness-of-fit}

The strategy:
\begin{itemize}
\item{Unweighted:}  Look for the time change that transforms any $\cop$-pBs into the standard pBs. 
                    This is achieved in Sect.~\ref{sec:time_change} just below.
                    Then, the 2D-Kolmogorov-Smirnov test is essentially related to the law of the supabs of the pBs, Sect.~\ref{sec:testofInd}.
\item{Weighted:}    Change of time and variable to recover nice processes: 
                    2D Brownian Motion (Brownian sheet) or 2D Ornstein-Uhlenbeck, or mixed Brownian Motion/Ornstein-Uhlenbeck.
\end{itemize}

\subsection{Flattening an arbitrary copula}\label{sec:time_change}
Let $\cop(u,v)$ be a two-dimensional copula function. We consider the following problem:
find a parametrization $s(u,v)$ increasing in $u$ for all $v$, and $t(u,v)$ increasing in $v$ for all $u$, 
that brings the copula $\cop(u,v)$ to the independence copula $\Pi(s,t)=st$.
A solution to this problem in \emph{any dimension} with no assumption but absolute continuity of the joint cdf
 is known as Rosenblatt's transform \cite{rosenblatt1952remarks},
but it relies on an arbitrary ordering of the variables.
We instead propose a solution for the 2D case with no particular order of the variables.
This is better stated combining integral and differential representations, and thus necessitates the appropriate differentiability conditions:
%Stated otherwise, we look for the change of variables that satisfies
\begin{equation}
  \begin{dcases}
	\cop(u,v)            &=\Pi(s(u,v),t(u,v))  =s(u,v) t(u,v)	\\
	   c(u,v)\,\d{u}\d{v}&=\pi(s,t)\,\d{s}\d{t}=\left|\frac{\partial(s,t)}{\partial(u,v)}\right|\d{u}\d{v}
  \end{dcases}
\end{equation}
with $c$ and $\pi\equiv 1$ the densities of $\cop$ and $\Pi$ respectively: we want to stretch the probability space in a particular fashion 
so that the copula density in that space is uniform.
Eliminating $t(u,v)$ from the system above, the solution satisfies the Monge-Amp\`ere partial differential equation
\begin{equation}\label{eq:monge-ampere}
	c(u,v)=\left|\frac{1}{s}\frac{\partial s}{\partial u}\frac{\partial \cop}{\partial v}-\frac{1}{s}\frac{\partial s}{\partial v}\frac{\partial \cop}{\partial u}\right|
          =||\nabla \ln s(u,v) \wedge \nabla \cop(u,v)||.
\end{equation}
This equation for $\ln s$ can be solved by the ``method of characteristics'': 
first, define the parametric trajectory --- the ``characteristic curve'' --- $\vect{U}(\gamma)\equiv(u(\gamma),v(\gamma))$ satisfying
\begin{equation}\tag{\ref{eq:monge-ampere}'}\label{eq:monge-ampere2}
	\begin{cases}
\displaystyle		\frac{\d{u}}{\d{\gamma}}&\displaystyle=\phantom{-}\frac{\partial \cop(u,v)}{\partial v}\\
\displaystyle		\frac{\d{v}}{\d{\gamma}}&\displaystyle=         - \frac{\partial \cop(u,v)}{\partial u}
	\end{cases}
%\]
\qquad
\text{along which}
\qquad
%\[
%   \frac{\d{v}}{\d{u}}=\dfrac{\dfrac{\d{v}}{\d{z}}}{\dfrac{\d{u}}{\d{z}}}=-\frac{\dfrac{\partial \cop(u,v)}{\partial u}}{\dfrac{\partial \cop(u,v)}{\partial v}}.
    c(u(\gamma),v(\gamma))=\left|\frac{\d{s(u(\gamma),v(\gamma))}}{\d\gamma}\right|
\end{equation}
The algebraic interpretation is the following: %denoting $\vect{U}=(u,v)$,
\begin{align*}
	\d{\cop}(u(\gamma),v(\gamma))=\frac{\partial \cop}{\partial u}(u(z),v(z))\,\d{u(z)}+\frac{\partial \cop}{\partial v}(u(z),v(z))\,\d{v(z)}=&0\\
	\nabla \cop(\vect{U}(\gamma)) \cdot \d{\vect{U}(\gamma)} =&0
\end{align*}
so $\cop(u(\gamma),v(\gamma))=\mathcal{C}=\const$, what defines the function $v_{\mathcal{C}}(u')$
on $u'\in[\mathcal{C},1]$ as imposed by the Frechet bounds, see page~\pageref{eq:frechet}. 
Notice that $v_\mathcal{C}(u')$ is decreasing in $u'$ along a characteristic of fixed $\mathcal{C}$, 
increasing in $\mathcal{C}$ for a given $u'>0$, and that
$
    \lim_{\mathcal{C}\to 1}v_\mathcal{C}(u')=1%,\quad \forall u'\in[0,1]
$, 
but the limit as $\mathcal{C}\to 0$ is degenerate since the copula is zero at any point on the edges $(u,0)$ and $(0,v)$.
Then, the solution of Eq.~\eqref{eq:monge-ampere2} \emph{along a given characteristic} is found integrating
\[
	\d{\ln s}=c(u(\gamma),v(\gamma))\d{\gamma}=c(u',v_{\mathcal{C}}(u'))\frac{\d{\gamma}}{\d{u'}}\,\d{u'},
\]
yielding
\begin{subequations}
\begin{equation}\label{eq:sC_solution}
	\ln s_{\mathcal{C}}(u)-\ln s_{\mathcal{C}}(u^*)=\int_{u*}^{u} {c(u',v_{\mathcal{C}}(u'))}\left[{\dfrac{\partial \cop}{\partial v}(u',v_{\mathcal{C}}(u'))}\right]^{-1}\,\d{u'}.
\end{equation}
Since the density $c(u,v)$ and the first derivative $\partial\cop/\partial v$ are always positive, 
$s_{\mathcal{C}}(u)$ is increasing in $u$ at given $\mathcal{C}$.
Because the problem is symmetric in $s\leftrightarrow t$, one could as well eliminate $s$ 
from the initial system of equations, and solve it to obtain
\begin{equation}
%t(u,v)=\Exp{-\int_{v}^1 {c(u_{\cop(u,v)}(v'),v')}\left[{\dfrac{\partial \cop}{\partial u}(u_{\cop(u,v)}(v'),v')}\right]^{-1}\,\d{v'}}.
	\ln t_{\mathcal{C}}(v)-\ln t_{\mathcal{C}}(v^*)=\int_{v*}^{v} {c(u_{\mathcal{C}}(v'),v')}\left[{\dfrac{\partial \cop}{\partial u}(u_{\mathcal{C}}(v'),v')}\right]^{-1}\,\d{v'}.
\end{equation}
\end{subequations}
%what provides a determination of $t$ alternative to Eq.~\eqref{eq:st_final}.
The final solution of Eq.~\eqref{eq:monge-ampere} is
\begin{subequations}
\begin{align}\label{eq:s_final_solution}
        s(u,v)&=s_{\cop(u,v)}(u)\\\label{eq:t_final_solution}
        t(u,v)&=t_{\cop(u,v)}(v)=\cop(u,v)/s(u,v),
\end{align}
\end{subequations}
and it is necessary to check that it satisfies the conditions set at the begining,
namely that $s(u,v)$ be increasing in $u$ and $t(u,v)$ be increasing in $v$.

The integration constants $s_{\mathcal{C}}(u^*)$ and $t_{\mathcal{C}}(v^*)$ are determined as follows. 
Evaluating simultaneously Eq.~\eqref{eq:s_final_solution} at a point $(u^*,v_\mathcal{C}(u^*))$  
                      and Eq.~\eqref{eq:t_final_solution} at a point $(u_\mathcal{C}(v^*),v^*)$, both on the characteristic curve,
and taking the product, one gets
\[
    s(u^*,v_\mathcal{C}(u^*))\,t(u_\mathcal{C}(v^*),v^*)=\mathcal{C}.
\]
But since one also has generically $s(u^*,v^*)\,t(u^*,v^*)=\mathcal{C}$,
the identification yields 
\begin{equation}
%\begin{cases}
		u_\mathcal{C}(v^*)=u^*\qquad\text{and}\qquad%\\
		v_\mathcal{C}(u^*)=v^*
%\end{cases}.
\end{equation}
For a symmetric copula, one finds that $s_\mathcal{C}(u^*)=t_\mathcal{C}(v^*)=\sqrt{\mathcal{C}}$
for $u^*=v^*$ solution of $C(u^*,u^*)=\mathcal{C}$.
In this case, Eq.~\eqref{eq:t_final_solution} makes also  clear that $t(u,v)=s(v,u)$.

\begin{figure}[b!]
    \includegraphics[scale=.6,trim=360 0 0 0,clip]{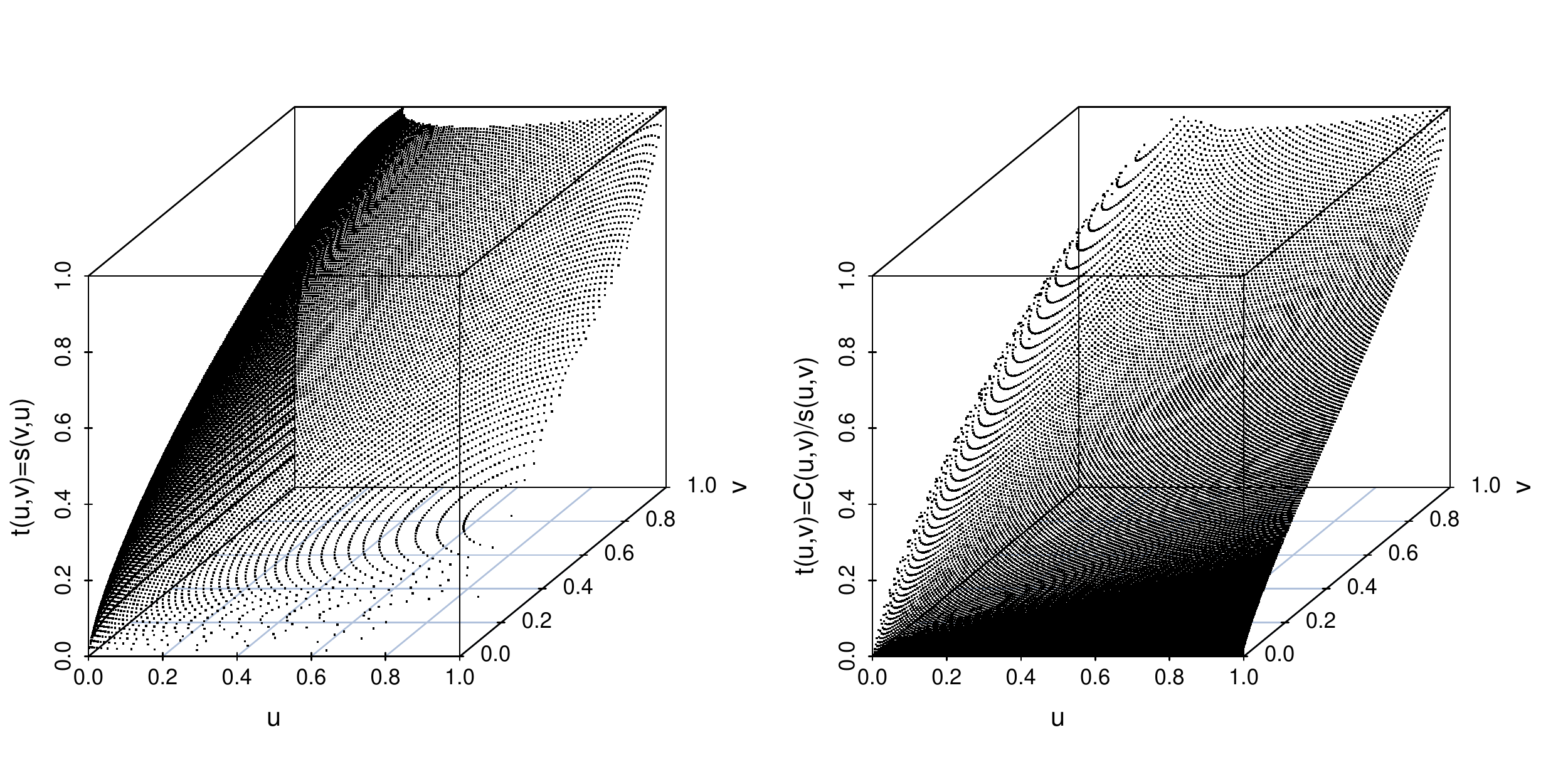}
    \caption{Flattening the Gaussian copula, example with a correlation coefficient $\rho=0.5$.
    The time change $s(u,v)$ is computed from Eqs.~(\ref{eq:sC_solution},\ref{eq:s_final_solution}),
    then $t(u,v)$ is checked for consistency:
    \textbf{Left:}  $t(u,v)=s(v,u)$;
    \textbf{Right:} $t(u,v)=C(u,v)/s(u,v)$.}
\end{figure}
% \subsubsection*{A workable example: the Gaussian copula}
% \[
    % \cop[\text{G}](u,v)=\Phi_\rho(\Phi^{-1}(u),\Phi^{-1}(v))
    % \qquad
    % c_{\text{G}}(u,v) =\frac{\varphi_\rho(\Phi^{-1}(u),\Phi^{-1}(v))}{\varphi(\Phi^{-1}(u))\,\varphi(\Phi^{-1}(v))}
        %%&=\frac{1}{1-\rho^2}\Exp{-\frac{\rho^2 [\Phi^{-1}(u)]^2 -2\rho\Phi^{-1}(u)\Phi^{-1}(v) +\rho^2 [\Phi^{-1}(v)]^2}{2\,(1-\rho^2)}}
% \]
% \[
    % \frac{\partial\cop[\text{G}](u,v)}{\partial v}=\Phi\!\left(\frac{\Phi^{-1}(u)-\rho\Phi^{-1}(v)}{\sqrt{1-\rho^2}}\right)
% \]

% \subsubsection*{\textcolor{red}{A workable example: Archimedean copula}}
% \[
    % \cop(u,v)=\phi^{-1}\!\left(\phi(u)+\phi(v)\right)
% \]
% \[
    % \frac{\partial\cop(u,v)}{\partial v}=\frac{\phi'(v)}{\phi'\!\left(\cop(u,v)\right)}
% \]
% \[
    % c(u,v)=\frac{\phi'(u)\,\phi'(v)\,\phi''\!\left(\cop(u,v)\right)}{\phi'\!\left(C(u,v)\right)^3}
% \]

\clearpage
\subsection{Variance-weighting the $\cop$-pBs ($\psi(u,v)=1/H_{\cop{}}(u,v;u,v)$)}\label{ssec:var-weighting}

Let $Y(u,v)$ be the pinned $\cop$-Brownian sheet on $(u,v)\in[\epsilon,1-\epsilon]^2$,
so that, in particular, 
$
    \var{Y(u,v)}=H_{\cop{}}(u,v;u,v)
$
as defined in Eq.~\eqref{eq:HtheoC}.
The time-changed rescaled process
\[
    Z(s,t)=Y(f(s,t),g(s,t))/\sqrt{\var{Y(f(s,t),g(s,t))}}.
\] 
is a 2D Ornstein-Uhlenbeck on $(s,t)\in[0,T_\epsilon]^2$, if
\[
\left\{
\begin{array}{rl}
    \cop(f(s,t),g(s,t)) &\displaystyle =\frac{\sqrt{\frac{\cop(\epsilon,\epsilon)}{1-\cop(\epsilon,\epsilon)}}\,\e^{s+t}}{\left[\sqrt{\frac{\cop(\epsilon,\epsilon)}{1-\cop(\epsilon,\epsilon)}}\,\e^{s+t}\right]^{-1}+\sqrt{\frac{\cop(\epsilon,\epsilon)}{1-\cop(\epsilon,\epsilon)}}\,\e^{s+t}}\\
    f(s,t)\leq f(s',t') &\Leftrightarrow g(s,t)\leq g(s',t')
\end{array}
\right. ,
\]
since one shows that the correlations are
\[
    \esp{Z(s,t)Z(s',t')}=\e^{-|s-s'|}\,\e^{-|t-t'|},
\]
and decay exponentially with the time lags, as expected for a \emph{stationary} OU process
with distribution $\mathcal{N}(0,1)$ (as long as $0<\epsilon<1)$.

Obviously, since the pBs $Y$ is not Markovian, the description in terms of the Markovian process $Z$ 
is only possible in a limit sense, where the corner $(u,v)=(1,1)$ is never reached,
thanks to a singular change of time:    
$(u,v)=(\epsilon,\epsilon)$     maps onto $(s,t)=(0,0)$, and 
$(u,v)=(1-\epsilon,1-\epsilon)$ maps onto $(s,t)=(T_\epsilon,T_\epsilon)$ where
\[
    T_\epsilon=\ln\sqrt[4]{\frac{\cop(1\!-\!\epsilon,1\!-\!\epsilon)}{1-\cop(1\!-\!\epsilon,1\!-\!\epsilon)}\frac{1-\cop(\epsilon,\epsilon)}{\cop(\epsilon,\epsilon)}}\xrightarrow{\epsilon\to0}{\infty}.
   %T_\epsilon=\ln\sqrt[4]{\frac{(1-\epsilon)^3(1+\epsilon)}{\epsilon^3(2-\epsilon)}}\xrightarrow{\epsilon\to0}{\infty}.
\]
The diagonal and anti-diagonal copula can be written in terms of tail dependence coefficients:
\[
    T_\epsilon= -4\,\ln\!\left[\epsilon^2\,\tLL(1\!-\!\epsilon)\,\big(2-\tUU(1\!-\!\epsilon)\big)\right]\xrightarrow{\epsilon\to0}{\infty},
\]
and the whole space $[0,1]^2$ tends to be covered only when $T_\epsilon$ diverges to $\infty$.

In the notations of page~\pageref{eq:KS2D}, $u=f(s,t)$, $v=g(s,t)$ and the weights are
\[
    %\tilde\psi(s,t)=
    \sqrt{\psi(f(s,t),g(s,t))}=\left[\sqrt{\frac{\epsilon}{1-\epsilon}}\,\e^{s+t}\right]^{-1}+\sqrt{\frac{\epsilon}{1-\epsilon}}\,\e^{s+t}.
\]
The law of the corresponding KS statistic \eqref{eq:KS2D} writes now
\begin{align*}
    S(k)&=\lim_{\epsilon\to 0}\Pr{-k\leq Z(s,t)\leq k, \forall (s,t)\in[0,T_\epsilon]^2}\\
        &=\int_{-k}^k\d{z}\int_{-k}^k\d{z_0}\,\frac{1}{\sqrt{2\pi}}\e^{-z_0^2/2}\,\mathcal{P}_{(T_\epsilon,T_\epsilon)}(z-z_0;k)
\end{align*}
and is the survival probability of the 2D mean-reverting process.
Note that we have averaged over the initial position $z_0$ since $Z(0)=Y(\epsilon,\epsilon)/[\epsilon\,(1-\epsilon)]\sim\mathcal{N}(0,1)$.
Unfortunately, although the free propagator of the OU process $Z(u,v)$ is known to be
\[
    \mathcal{P}_{(s,t)}(z;\infty)=\frac{1}{\sqrt{2\pi st\,(1-st)}}\Exp{-\frac{1}{2}\frac{z^2}{st\,(1-st)}},
\]
the propagator of the constrained process is unknown.
Again, there is no known Partial Derivatives Equation (PDE) of the Fokker-Planck type for two times,
that would describe the solution with arbitrary boundary conditions,
and whose free solution would be $\mathcal{P}_{(s,t)}(z;\infty)$.

% \textcolor{red}{Attention, le probleme en une dimension est asymptotiquement trivial (quand $\epsilon\to 0$) car la probabilite de survie est nulle.
% Pre-asymptotiquement ($0<\epsilon\ll 1$), la solution est donnee dans [Chicheportiche-Bouchaud-2012].}

\subsubsection{Discretizing one time direction}
Define $Z_n(t)=Z(s_n,t)$ for a $\mathds{R}$-valued increasing sequence $(s_n)_{n\in\libracket 0,N\ribracket}$.
The vector $\vect{Z}(t)$ represents a collection of ($N+1$) OU processes with 
% correlation matrix
% \[
    % \rho_{nm}=\Esp{\frac{Z_n(t)}{\sqrt{\var{Z_n}}}\frac{Z_m(t)}{\sqrt{\var{Z_m}}}}=\e^{-|s_n-s_m|}.
% \]
covariance matrix
\[
    \rho_{nm}=\esp{Z_n(t)Z_m(t)}=\e^{-|s_n-s_m|}.
\]

The vectorial SDE writes $\d{\vect{Z}(t)}=-\vect{Z}(t)\,\d{t}+\rho^{\frac{1}{2}}\,\d{\vect{W}(t)}$
with the initial condition $\vect{Z}(0)\sim\mathcal{N}(0,\rho)$, and the FPE is   
\[
    \partial_tf_t(\vect{z};k)=f_t(\vect{z};k)+\vect{z}\cdot\nabla f_t(\vect{z};k)+\frac{1}{2}{\nabla^\dagger\rho\nabla f_t(\vect{z};k)},
\]
where $\vect{z}$ lives in the interior of the $(N+1)$-dimensional square box
\[
    \mathcal{D}_N=\big\{\vect{z} :  |z_n|\leq k, \forall n\in\libracket 0,N\ribracket \big\},
\]
and the density $f_t(\vect{z};k)$ vanishes at the border $\partial\mathcal{D}_N$.

If the sequence $(s_n)$ covers the entire interval $[0,T_\epsilon]$ as $N\to\infty$, then
\begin{align*}
    S(k)&=\lim_{\epsilon\to 0}\lim_{N\to\infty}\Pr{-k\leq Z(s_n,t)\leq k, \forall n\in\libracket 0,N\ribracket, \forall t \in[0,T_\epsilon]}\\
        &=\lim_{\epsilon\to 0}\lim_{N\to\infty}\int_{\mathcal{D}_N}\d{\vect{z}}\int_{-k}^k\d{\vect{z}_0}\,\frac{1}{\sqrt{(2\pi)^N \det\rho}}\e^{-\vect{z}_0\rho^{-1}\vect{z}_0/2}\,f_{T_\epsilon}(\vect{z}-\vect{z}_0;k)
\end{align*}
The order of the limits matters, because one is likely to be interested in the preasymptotic regime where $\epsilon>0$ !

\subsection{Unpinning the $\cop$-pBs}\label{sec:sup.brownian_bridge}
We saw above that it is possible to perform a change of time and variable in order to
write the unweighted pBs as a Brownian sheet that is eventually pinned at (1,1) (Section~\ref{ssec:pinning}),
and that the variance-weighted pBs could as well be rewritten as an unpinned process, 
namely the 2D Ornstein-Uhlenbeck (Section~\ref{ssec:var-weighting}).
Moreover, by discretizing one time direction, the problems reduced to those of
a collection of $N$ processes (Brownian or OU, respectively).
We now intend to write the $\cop$-pBs as a collection of Brownian Motions by a change of time.

We consider a pinned $\cop$-Brownian sheet,
and look for a mapping between the coordinates $(u,v)\in[\epsilon_s,1\!-\!\epsilon_s]\times[\epsilon_t,1\!-\!\epsilon_t]$ 
and $(s,t)\in [0,T_{\epsilon_s}]\times[0,T_{\epsilon_t}]$.
We write this mapping as %dans la Section~3.1:
\[
    u=f(s,t)\quad v=g(s,t),
\]
and define simultaneously $X(s,t)=\widetilde{Y}(f(s,t),g(s,t))$.
For every $s$, the covariance of $X(s,t)$ is that of a Brownian motion
\begin{equation}\label{eq:covX}
    \esp{X(s,t)X(s,t')}=\gamma(s)\,\min(t,t'),
\end{equation}
if there exist characteristic times $\tau(s)>0$ and increasing mappings $f,g$, such that 
$$
    \psi(f(s,t),g(s,t))=\gamma(s)\left(\tau(s)+2t+t^2/\tau(s)\right)=\frac{\gamma(s)}{\tau(s)}\,\left(\tau(s)+t\right)^2
$$
and
\begin{equation}\label{eq:contraintes}
    \cop(f(s,t),g(s,t))=\frac{t}{\tau(s)+t},\quad\forall s.
\end{equation}
% et donc
% \[
   % \frac{\esp{X(s,t)X(s',t')}}{\sqrt{\gamma(s)\gamma(s')}}=\frac{1+t/\tau(s)}{1/\sqrt{\tau(s)}}\frac{1+t'/\tau(s')}{1/\sqrt{\tau(s')}}\min(f(s,t),f(s',t'))\min(g(s,t),g(s',t'))-\frac{tt'}{\sqrt{\tau(s)\tau(s')}}.
% \]
Along two trajectories  $s\neq s'$ indexed by the same time $t'=t$, we find 
\begin{align*}
   %\frac{\esp{X(s,t)X(s',t)}}{\sqrt{\gamma(s)\gamma(s')}}&=t\,(1+t/\tau(\max(s,s')))\,\sqrt{\frac{\tau(\max(s,s'))}{\tau(\min(s,s'))}}-\frac{t^2}{\sqrt{\tau(s)\tau(s')}}\\
   %                                                      &=\rho(s,s')\cdot t,
    \esp{X(s,t)X(s',t')}&=\sqrt{\gamma(s)\gamma(s')}\,\rho(s,s')\,\min(t,t')
\end{align*}
where
$
    \rho(s,s')=\sqrt{{\tau_{\max}}/{\tau_{\min}}}.
   %\rho(s,s')=\sqrt{\frac{\tau_{\max}}{\tau_{\min}}}=\Exp{\frac{\log\tau_{\max(s,s')}-\log\tau_{\min(s,s')}}{2}}=\Exp{-\frac{|\log\tau_{s}-\log\tau_{s'}|}{2}}.
$

\paragraph{Example:}
By choosing $\gamma(s)=1$ and $\tau(s)=\tau_0\,\e^{-2s}$, the equal-time correlation has the exponential form $\rho(s,s')=\e^{-|s-s'|}$, 
so $X(s,t)$ is a Gauss-Markov process in $s$ at every $t$ !
We get back a case similar to that treated in Sect.~\ref{ssec:pinning}: 
a Brownian motion in $t$ and an Ornstein-Uhlenbeck in $s$, 
although there we considered a regular unpinned Brownian sheet on $[0,1]^2$, 
and we discretized one time direction uniformly.
According to Eq.~\eqref{eq:contraintes}, we can also write inversely
\begin{align}
        \e^{2s}\,t/\tau_0&=\frac{\cop(u,v)}{1-\cop(u,v)}\\
               T_{\epsilon_s}+\ln\sqrt{\frac{T_{\epsilon_t}}{\tau_0}}&=\ln\sqrt{\frac{\cop(1\!-\!\epsilon_s,1\!-\!\epsilon_t)}{1-\cop(1\!-\!\epsilon_s,1\!-\!\epsilon_t)}}
\end{align}
Since the times $s$ and $t$ are of different ``nature'' (one is exponential, the other is linear),
it is convenient to choose the horizon such that
% \textcolor{red}{
% Vu que les temps $s$ et $t$ n'on pas la meme nature (l'un exponentiel, l'autre lineaire),
% il est plus intelligent de choisir l'horizon de sorte que
 \[
     T_{\epsilon_t}=\tau_0\,\e^{2T_{\epsilon_s}}=\tau_0\,\sqrt{\frac{\cop(1\!-\!\epsilon,1\!-\!\epsilon)}{1-\cop(1\!-\!\epsilon,1\!-\!\epsilon)}}\equiv T_\epsilon,
 \]
% ou on choisit simplement $\epsilon_s=\epsilon_t\equiv\epsilon$.
% }
where we simply impose $\epsilon_s=\epsilon_t\equiv\epsilon$.

\paragraph{Example (a path in the plane):}
By choosing $\tau(s)=\tau_0$, the equal-time correlation is simply $\rho(s,s')=1$, 
meaning that all the $X(s,\cdot)$ are equal: this is a ``Gaussian front''.
We can take arbitrarily $s=t$, which is equivalent to considering the joint mappings
 $(u,v)\equiv(u,\vartheta(u))$, s.t.\ $\vartheta(1\!-\!\epsilon_s)=1\!-\!\epsilon_t\equiv1\!-\!\epsilon$, i.e.\ 
$$u\equiv f(t,t) \quad\text{and}\quad v\equiv g(t,t)\equiv (\vartheta \circ f)(t,t).$$
From Eq.~\eqref{eq:contraintes}, we can also write inversely
\begin{align}
     \psi(u,\vartheta(u))&=\tau_0/[1-\cop(u,\vartheta(u))]^2\\
                 t/\tau_0&=\frac{\cop(u,\vartheta(u))}{1-\cop(u,\vartheta(u))}\\
               T_\epsilon&=\tau_0\frac{\cop(1\!-\!\epsilon,\vartheta(1\!-\!\epsilon))}{1-\cop(1\!-\!\epsilon,\vartheta(1\!-\!\epsilon))}
\end{align}
This is very similar to the unidimensional case, since we investigate only a path in the space $[0,1]^2$.

\subsubsection{Special case $\vartheta(u)=u^\alpha$ for independence copula}
We now look more carfully at the above example where $v=\vartheta(u)$,
in the particular case where the copula is $\cop(u,v)=\Pi(u,v)=uv$.
The weights are
\begin{subequations}
\label{eq:weights_specialcase}
\begin{align}
     \psi(u,u^\alpha)&=\tau_0/[1-u^{\alpha+1}]^2\\
     \psi(f(t,t),f(t,t)^\alpha)&=\tau_0\,[1+t/\tau_0]^2
\end{align}
\end{subequations}
and the time changes $f,g$ are
\[
    f_\alpha(t)=\left(\frac{t}{t+\tau_0}\right)^{\frac{1}{1+\alpha}}\quad\text{et}\quad 
    g_\alpha(t)=\left(\frac{t}{t+\tau_0}\right)^{\frac{\alpha}{1+\alpha}}
\]
whatever the exponent $\alpha\neq -1$, where $t\leq T_\epsilon(\alpha)\equiv(1-\epsilon)^{\alpha+1}/[1-(1-\epsilon)^{\alpha+1}]$. 
We exhibit a sequence $\{\alpha(n;N)\}_{-N< n< N}$ such that the whole space is covered in the limit $N\to\infty$,
and such that the functions $f_n,g_n$ have the required properties:
\begin{equation*}
\alpha(n;N) = 
  \begin{dcases}
    {\frac{N-n}{N}}  & \quad \text{si $  0\leq n< N$}\\
    {\frac{N}{N-|n|}}& \quad \text{si $ -N <n\leq 0$}\\
  \end{dcases} 
\end{equation*}

%%% CODE R
%alpha <- function(N,n) ifelse(0<=n,(N-n)/N,N/(N+n)) 
%
%NN=100
%plot(-NN:NN,alpha(NN+1,-NN:NN))

One notices that there is no condition on $\tau_0$, and it is possible to take $\tau_0=1$ for all $n$ !
The time change specific to each  $n$ is performed through $\alpha_n$, and 
the characteristic time scale is not important (as long as the change of time remains singular).
Finally we get $\psi(n,t)=(1+t)^2$, and
\[
    f(n,t;N)=\left(\frac{t}{t+1}\right)^{\frac{1}{1+\alpha(n;N)}}\quad\text{et}\quad 
    g(n,t;N)=\left(\frac{t}{t+1}\right)^{\frac{\alpha(n;N)}{1+\alpha(n;N)}}.
\]

%Je refais le calcul de (\ref{eq:covX}) en supposant $\gamma_n \equiv 1$ et $t< t'$:
%\begin{equation*}
%\esp{X_n(t)X_n(t')}= \sigma_n(t) \sigma_n(t') f_n(t) g_n(t) \left(1-f_n(t') g_n(t')\right) ,
%\end{equation*}
%et on impose \`a cette expression d'\^{e}tre \'egale \`a $t$. On en d\'eduit que l'expression suivante est ind\'ependante de $t,t'$:
%\begin{equation*}
%\sigma_n(t')(1-f_n(t') g_n(t')) = \frac{t}{\sigma_n(t) f_n(t) g_n(t)} =:\tau_n^{1/2} .
%\end{equation*}
%D'o\`u (\ref{eq:contraintes}). OK

% \textcolor{red}{
% Is there a function $F$ decreasing to $0$ such that, a.s.:
% \[
% \left|\sup_{u,v\in [\epsilon,1-\epsilon]^2} |Y_{uv}| - \max_{-N< n< N} \sup_{t\in [0,T_\epsilon(n;N)]} \frac{1}{1+t} |X_n(t)| \right|\leq F(N)  ?
% \]
% }

It is important to notice that the weights \ref{eq:weights_specialcase} do not depend at all
on the function $\vartheta(\cdot)$ chosen to parameterize the second time.
They define a curve, and not a surface in $[0,1]^2$.
Nevertheless, by following several curves $\vartheta(u)=u^\alpha(n;N)$ (with $-N<n<N$), 
it is possible to cover $[0,1]^2$ in the limit $N\to\infty$, but with different weights $\psi$ for every $n$.

\section{Perspectives}

The diagram on Fig.~\ref{fig:diagramKS2D} summarizes the possible strategies outlined in this chapter to address
the problem of finding the sup of a pinned $\cop$-Brownian bridge.
Appropriate time changes allow to transform the problem and handle stochastic processes that are more common
and whose properties are better known, like Brownian motions (BM), 1D or 2D Ornstein-Uhlenbeck (OU) processes, Brownian sheets (Bs).
We have provided the explicit transformations needed (change of time and change of variable, weighting scheme)
for every strategy.

The most important original contribution of this chapter is however to suggest a discretization of
one time direction, that ultimately leads to the problem of finding the density of presence of a stochastic particle 
in an infinite-dimensional cage with absorbing walls.

This opens a new way of addressing stochastic processes indexed by two times, 
and we hope that future efforts can lead to a better characterization of the law of the supremum.

\begin{figure}[bp]
\center
    \includegraphics[scale=0.75,trim=0 20 0 10,clip]{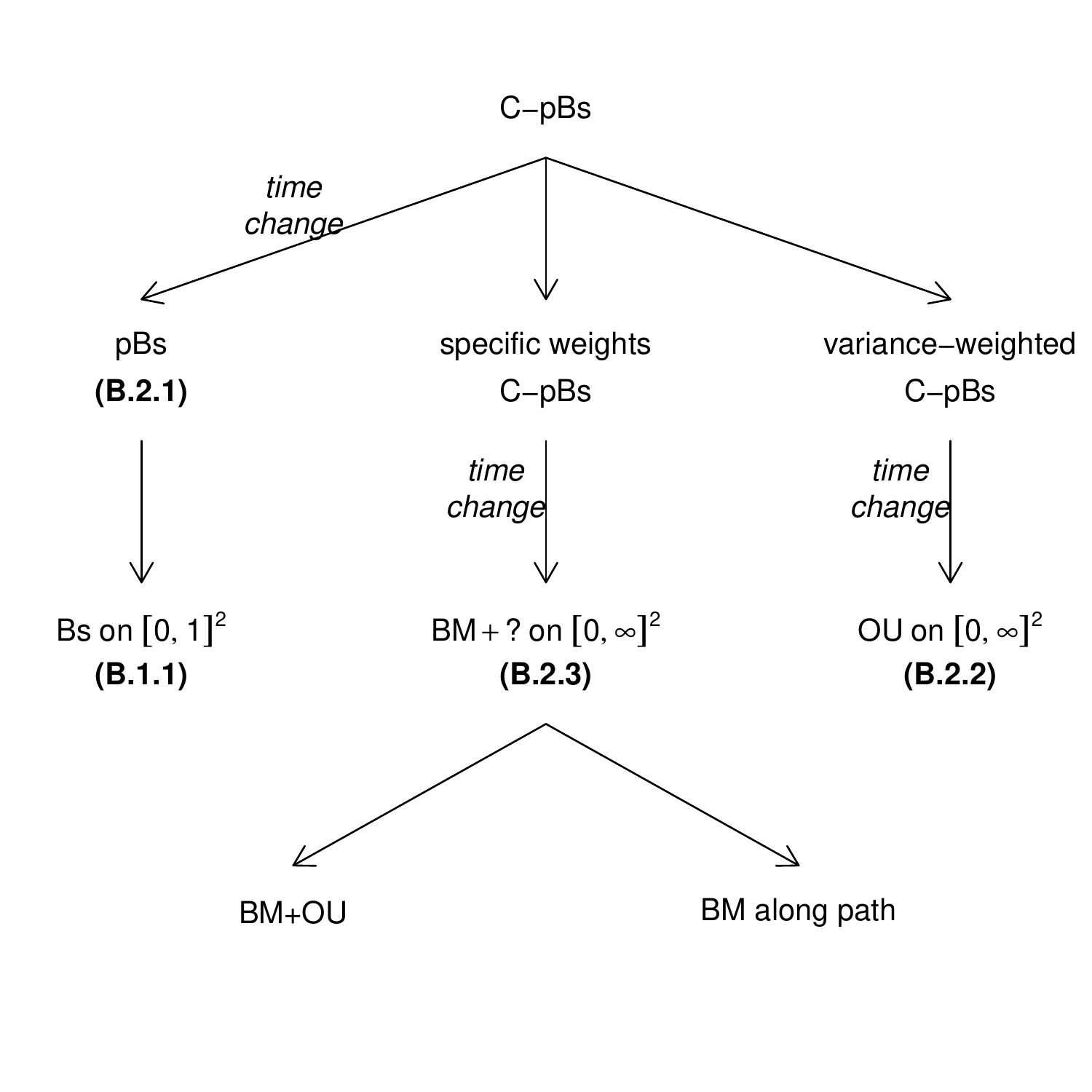}
    \caption{The different strategies to address the problem of finding the sup of the 
    pinned $\cop$-Brownian bridge.}
    \label{fig:diagramKS2D}
\end{figure}
\cleardoublepage

\chapter{One-factor linear model}\label{chap:APXonefactor}
\section{Diagonally perturbed rank-1 operator}\label{sec:APXrank1}
          \newcommand{\C}{\mathfrak{C}}

Let $\ket{\beta}$ be a vector whose $N$ components are required to satisfy the following $N(N+1)/2$ equations
\begin{align}
	\C_{ij}&=\epsilon\delta_{ij}+(1-\epsilon\delta_{ij})\,\beta_i\beta_j,
\end{align}
or in matrix form (and using Dirac's bra-ket notation):
\begin{align}\label{eq:qiqj_corrige}
	\C     &=\epsilon\,\big[\mathds{1}-\operatorname{diag}(\beta^2)\big]+\ket{\beta}\bra{\beta}
\end{align}
where $\C$ is a given symmetric positive definite matrix.
When $\epsilon=1$, this amounts to trying to explain the correlation structure
of $N$ random variables with only $N$ parameters. 
This case is typically encountered in factor models with linear exposition of every individual to a common hidden variable.

Of course this problem does not have a solution in the general case, 
due to the number of variables being much lower than the number of constraints.
It is then better stated in a ``best quadratic fit'' (least-squares estimation) sense:
 minimize the loss function (Hamiltonian) given by the Mean Squared Error (MSE)
\[
	\ket{\beta^*}=\argmin_\beta\left\{\frac{1}{N(N-1)}\sum_{i,j\neq i}\left(\C_{ij}-\beta_i\beta_j\right)^2\right\}.
\]
Differentiation yields the set of equations for the optimal solution (we drop the $^*$):
\begin{equation}\label{eq:RMS}
	\sum_{i=1}^N\C_{ki}\beta_i=\big[\braket{\beta}{\beta}+\epsilon\,(1-\beta_k^2)\big]\,\beta_k.
\end{equation}
% or in matrix form,
% \begin{equation}\label{eq:RMS}
	% \C\ket{\beta}=\big(\braket{\beta}{\beta}+\epsilon\big)\ket{\beta}-\epsilon\ket{\beta^3}
% \end{equation}
$\ket{\beta}$ can be decomposed onto contributions longitudinal and orthogonal to the first (normalized) eigenvector $\ket{V}$ of $\C$ (corresponding to the largest eigenvalue $\lambda$): 
\begin{align}\label{eq:ketq}
	\ket{\beta}&=v\ket{V}+\epsilon\left[A_\|\ket{V}+\ket{V^\perp}\right].
\end{align}
The correction is explicitly written in multiples of $\epsilon$, since to zero-th order one indeed
recovers the eigenvalue equation $\C\ket{\beta}=\braket{\beta}{\beta}\ket{\beta}$ with solution $\ket{\beta}=\sqrt{\lambda}\,\ket{V}$.% and $v^2={\braket{V}{V}}=\lambda\sim N$.
\footnote{In fact any couple of eigenvalue/eigenvector of $\C$ is solution, but the projector $\lambda\ket{V}\bra{V}$ associated to the largest eigenvalue has the largest information content in the class of rank-one operators.}
At this order, $\braket{\beta}{\beta}=\lambda\sim N$, so it dominates any $\epsilon$-correction.

In full generality, the squared norm of $\ket{\beta}$ is 
\[
    \braket{\beta}{\beta}=(v+\epsilon A_\|)^2+\epsilon^2\braket{V^\perp}{V^\perp}=v^2+2v\epsilon A_\|+\epsilon^2 [A_\|^2+\braket{V^\perp}{V^\perp}]
\]
and, to first order in $\epsilon/N$, Eq.~\eqref{eq:RMS} writes
\[
	A_\|\lambda V_k+\sum_{i}\C_{ki}V^\perp_k=\left(\lambda A_\|+2\lambda A_\|+\sqrt{\lambda}(1-\lambda V^2_k)\right)V_k+\lambda V^\perp_k
\]
Now, decomposing
\[
	V^3_k=\braket{1}{V^4}V_k+W^\perp_k
\]
longitudinal and transverse components must satisfy
\begin{align*}
                      2\lambda A_\|+\sqrt{\lambda}&=\lambda^{3/2}\braket{1}{V^4}\\
	\left(\lambda\mathds{1}-\C\right)\ket{V^\perp}&=\lambda^{3/2}\ket{W^\perp},
\end{align*}
i.e.
\begin{align}
	A_\|&=\frac{1}{2}\left(\braket{1}{V^4}{\lambda}^{1/2}-{{\lambda}}^{-1/2}\right)&&\sim 1/\sqrt{N}\\
	\ket{V^\perp}&=\lambda^{{3/2}}\left(\lambda\mathds{1}-\C\right)^{-1}\ket{W^\perp}\\
	             &\approx\sqrt{\lambda}\left(\ket{V^3}-\braket{1}{V^4}\ket{V}\right)&&\sim 1
\end{align}
so that the final solution is
\[
\boxed{
    \ket{\beta}/\sqrt{\lambda}=\left[1-\frac{\epsilon}{2}\left(\braket{1}{V^4}+1/\lambda\right)\right]\ket{V}
                          +\epsilon\ket{V^3}
}
\]

\clearpage\newcommand{\Xo}{r^{\scriptscriptstyle (0)}}
\newcommand{\hXo}{\hat{r}^{\scriptscriptstyle (0)}}
\newcommand{\Xun}{r^{\scriptscriptstyle (1)}}
\newcommand{\hXun}{\hat{r}^{\scriptscriptstyle (1)}}
\newcommand{\Xde}{r^{\scriptscriptstyle (2)}}
\newcommand{\hXde}{\hat{r}^{\scriptscriptstyle (2)}}

\newcommand{\Q}{\mat{Q}}

%\begin{abstract}
%We present an original way of constructing asset dependences away from the traditional copula descriptions.
%We rely on ideas from factor models and develop a hierarchical structure allowing for cluster classification.
%Consequences of the model in terms of correlations and non-linear dependences are assessed and confronted to empirical measurements.
%A calibration strategy is sketched.
%We insist on the role of both the explained factor contribution and the residuals !
%
%\end{abstract}

%%%%%%%%%%%%%%%%%%%%%%%%%%%%%%%%%%%%%%%%%%%%%%%%%%%%%%%%%%%%%%%%%%%%%%%%
\section{Nested single-factor linear model}\label{chap:APXnested}

\subsubsection*{One-factor model} 
Let $\vect{r}$ be an $S$-variate random vector $\vect{r}$, 
and $\C=\esp{\vect{r}\vect{r}^\dagger}$ its correlation matrix --- everything is normalized.
The linear correlation structure of the $S$ random variables $r_s$ is 
fully determined by their exposure to a common factor $\Xo$ of zero mean and unit variance:
\begin{equation}\label{eq:rel1}
	r_s = \beta_s\Xo+e_s.
\end{equation}
The residuals are assumed uncorrelated\footnote{They need not be independent. If they are, the whole non-linear dependence structure is fully determined by Eq.~\eqref{eq:rel1}.}
with variances $\var{e_s}=1-\beta_s^2$.

\subsection*{Estimating the model}
We address in this section the question of the estimation of the parameters $\vect{\beta}$. 
This simple setup is intended to be the building block of a more complex auto-similar description
of hierarchically organized variables.

\subsubsection*{When the factor $\Xo$ is known}
In this case, Eq.~\eqref{eq:rel1} is seen as a regression, and the OLS estimates of $\vect{\beta}$ are
\begin{equation}\label{eq:est1}
	\hat{\beta}_s=\vev{r_s\Xo},
\end{equation}
where $\vev{\cdot}$ means sample average over \emph{realizations} of the random variables.

The theoretical bias is $\esp{\hat{\beta}_s}-\beta_s=\esp{e_s\Xo}$, so that the OLS estimator is
unbiased only when the correlation of the residual with the explanatory variable is nil.
When the residuals are independent Gaussian, the OLS has optimal properties (minimal asymptotic variance).

\subsubsection*{When the factor $\Xo$ is not known}
In this case another set of model predictions should be used to estimate the exposure parameters,
for example the correlation structure
\begin{equation}\label{eq:corr1}
	\esp{r_sr_u}=\beta_s\beta_u+\beta_s\esp{e_u\Xo}+\beta_u\esp{e_s\Xo}+\esp{e_se_u}.
\end{equation}
With uncorrelated residuals, we can build an estimator
which leaves \eqref{eq:corr1} unbiased under the same conditions as \eqref{eq:est1}:
\begin{equation}\label{eq:est2}
	\hat{\vect{\beta}}=\argmin_{\vect{\beta}}\mathcal{L}_{\beta}(\vect{\beta}|\vect{r})
\end{equation}
where we have introduced the loss function 
%\[
	$\mathcal{L}_{\beta}(\vect{\beta}|\vect{r})=\left|\left|\C-\vect{\beta}\vect{\beta}^{\dagger}\right|\right|$ 
%\]
with an appropriate ``norm'' taking care of the peculiarities of the correlation matrix, 
f.ex. the RMS over off-diagonal elements of a symmetric matrix:
\[
	||\text{A}||=\sqrt{\frac{1}{S(S-1)}\sum_{j<i}\text{A}_{\ij}^2}.
\]
In this case, the perturbative method of Sect.~\ref{sec:APXrank1} can be used, 
otherwise a numerical solution must be looked for.

\paragraph{Reconstructing the factor}
For some applications, it might be necessary to have in hands the series of the residuals, 
or even that of the factor itself.
The factor has then to be reconstructed from the individual variables with some weights $w_s$: $\hXo(\vect{w})$.
It is clear that any such attempt will unavoidably introduce correlations of the form that bias the estimates of $\vect{\beta}$.
But this in turn precisely provides an optimality criterion for the choice of the weights:
\begin{equation}\label{eq:opt_w}
	\hat{\vect{w}}=\argmin_{\vect{w}}\underbrace{\left|\left|\vev{\vect{e}\hXo(\vect{w})^\dagger}\right|\right|}_{\mathcal{L}_w(\vect{w}|\vect{r},\vect{\beta})}=\argmin_{\vect{w}}\sqrt{\left(\vev{\vect{r}\hXo(\vect{w})}-\vect{\beta}\right)^2}
\end{equation}
For consistency, the residual correlation $\mathcal{L}_w(\hat{\vect{w}}|\vect{r},\hat{\vect{\beta}})$ has to be very small !

A simple and intuitive expression for the factor variable is a linear combination of its constituents:
	$\hXo(\vect{w})=\vect{r}^{\dagger}\vect{w}/(\vect{w}^{\dagger}\C \vect{w})^{\frac{1}{2}}$.
In this case, the loss function writes
\[
	\mathcal{L}_w(\vect{w}|\vect{r},\hat{\vect{\beta}})^2=\frac{\vect{w}^{\dagger}\C\C\vect{w}}{\vect{w}^{\dagger}\C \vect{w}}-2\frac{\hat{\vect{\beta}}^{\dagger}\C \vect{w}}{\sqrt{\vect{w}^{\dagger}\C \vect{w}}}+\hat{\vect{\beta}}^{\dagger}\hat{\vect{\beta}}.
\]
For computational convenience, the normalization of the weights can even be performed \emph{ex post} so that 
the minimization program \eqref{eq:opt_w} is quadratic in $\vect{w}$.

\paragraph{}Alternatively to the two-steps procedure exposed above,
a single optimization program can be run to estimate both sets of parameters $\vect{\beta}$ and $\vect{w}$ in a self-consistent way:
considering the system
\[
	\begin{cases}
		r_s = \beta_s\hXo(\vect{w})+e_s,\qquad\var{e_s}=1-\beta_s^2\\
		\hXo(\vect{w})=\sum_s w_s r_s\\
		\vect{w}^{\dagger}\C \vect{w} =1
	\end{cases}
\]
the joint estimator is
\[
	(\hat{\vect{\beta}},\hat{\vect{w}})=\argmin_{(\vect{\beta},\vect{w})}\left\{\mathcal{L}(\vect{\beta},\vect{w}|\vect{r})\left|\vect{w}^{\dagger}\C \vect{w} =1\right.\right\}
\]
with
\[
	\mathcal{L}(\vect{\beta},\vect{w}|\vect{r})^2=\frac{1}{S(S-1)}\sum_{u<s}\left(\C_{su}-\beta_s\beta_u-\frac{\beta_sw_u(1-\beta_u^2)+\beta_uw_s(1-\beta_s^2)}{1-\sum\limits_{k}\beta_kw_k}\right)^2
\]

%In practice however $N$ is of the order of $350$ and the 700-dimensional non-quadratic optimization problem might converge extremely slowly,
%so that the 2-steps procedure may be better suited although less consistent.
%   MAIS S EST ENVIRON 10-30

\subsection*{Nested one-factor model} 
We now consider a model where two levels of the previous structure are nested:
\begin{subequations}
\begin{align}
	\Xun_s &= \beta_s\Xo+e_s,\qquad\var{e_s}=1-\beta_s^2\\\label{eq:x2i}
	\Xde_i &= \gamma_i\Xun_{s}+e_i,\qquad\var{e_i}=1-\gamma_i^2, \qquad i\in s
\end{align}
\end{subequations}
In each branch $s$ there are $N_s$ leaves, and the trunk/branch indices are constructed as linear combinations of their respective constituents:
%\begin{subequations}
\begin{align}
	\hXo(\vect{w})         = \sum_s w_s \Xun_s\quad\text{and}\quad%\\
	\hXun_s(\vect{z}_{[s]}) = \sum_{i\in s}z_i \Xde_i
\end{align}
%\end{subequations}

\subsubsection*{Estimation algorithms when all sectorial and global factors are unobservable}
All the estimation algorithms start with a prior classification, typically Bloomberg sectors.
\paragraph{Upward estimation algorithm:}
\begin{itemize}
	\item In each sector $s$, apply the one-factor calibration procedure (either 2-steps or self-consistent)
	\begin{itemize}
		\item compute the leaves' correlation: $\C_{[s]}^{\scriptscriptstyle (2)}=\vev{\vect{r}_{[s]}^{\scriptscriptstyle (2)}\vect{r}_{[s]}^{{\scriptscriptstyle (2)}\dagger}}$
		\item estimate $\vect{\gamma}$ and the weights $\vect{z}_{[s]}$
		\item define the branch index $\hXun_s(\hat{\vect{z}}_{[s]})$ and compute the leaves' residuals \mbox{$\vect{e}^{\scriptscriptstyle (2)}_{[s]}=\vect{r}^{\scriptscriptstyle (2)}_{[s]} - \hat{\vect{\gamma}}_{[s]}\Xun_s, i\in s$}
	\end{itemize}
	\item Apply the one-factor calibration procedure to the root (either 2-steps or self-consistent)
	\begin{itemize}
		\item compute the branches' correlation $\C^{\scriptscriptstyle (1)}=\vev{\vect{r}^{\scriptscriptstyle (1)}\vect{r}^{{\scriptscriptstyle (1)}\dagger}}$
		\item estimate $\vect{\beta}$ and the weights $\vect{w}$
		\item define the root index $\hXo(\hat{\vect{w}})$ and compute the branches' residuals \mbox{$\vect{e}^{\scriptscriptstyle (1)}=\vect{r}^{\scriptscriptstyle (1)} - \hat{\vect{\beta}}r^{\scriptscriptstyle (0)}$}
	\end{itemize}	
\end{itemize}
\emph{Pluses}: auto-similar recursive procedure (appropriate for many levels); quadratic minimization problems very efficient and quick. 
\emph{Minuses}: biases are propagated (and possibly amplified) from leaves to root; trouble with very small sectors (1 to 3 individuals).%; different sizes of sectors lead to different estimation errors for corresponding parameters

\paragraph{Integrated estimation algorithm:}
\begin{itemize}
	\item compute all leaves' correlations $\C^{\scriptscriptstyle (2)}=\vev{\vect{r}^{\scriptscriptstyle (2)}\vect{r}^{{\scriptscriptstyle (2)}\dagger}}$
	\item the corresponding theoretical prediction is $$\Q_{ij}=\delta_{ij}+(1-\delta_{ij})\gamma_i\gamma_j\left[\delta_{s_is_j}+(1-\delta_{s_is_j})\beta_{s_i}\beta_{s_j}\right]$$
	\item estimate all exposure parameters jointly: $(\hat{\vect{\gamma}},\hat{\vect{\beta}})=\argmin||\C^{\scriptscriptstyle (2)}-\Q||$
	\item In each sector $s$,
	\begin{itemize}
		\item estimate the weights $\vect{z}_{[s]}$
		\item define the branch index $\hXun_s(\hat{\vect{z}})$ and compute the leaves' residuals \mbox{$\vect{e}^{\scriptscriptstyle (2)}_{[s]}=\vect{r}^{\scriptscriptstyle (2)}_{[s]} - \hat{\vect{\gamma}}_{[s]}\Xun_s, i\in s$}
	\end{itemize}
	\item compute the branches' correlation 
            $\C^{\scriptscriptstyle (1)}=\vev{\vect{r}^{\scriptscriptstyle (1)}\vect{r}^{{\scriptscriptstyle (1)}\dagger}}$ 
         and check that 
            $\left|\left|\C^{\scriptscriptstyle (1)}-\hat{\vect{\beta}}\hat{\vect{\beta}}^{\dagger}\right|\right|\approx 0$
	\item At root level,
	\begin{itemize}
		\item estimate the weights $\vect{w}$
		\item define the root index $\hXo(\hat{\vect{w}})$ and compute the branches' residuals $\vect{e}^{\scriptscriptstyle (1)}=\vect{r}^{\scriptscriptstyle (1)} - \hat{\vect{\beta}}r^{\scriptscriptstyle (0)}$
	\end{itemize}		
\end{itemize}
\emph{Pluses}: dissociation of the estimation of the exposure parameters and that of the weights; one single optimization program for all levels simultaneously; a posterior check of the correctness of intermediate correlations.
\emph{Minuses}: non-quadratic minimization with a large number of parameters !

\paragraph{Mixed estimation algorithm:}
\begin{itemize}
	\item compute all leaves' correlations $\C^{\scriptscriptstyle (2)}=\vev{\vect{r}^{\scriptscriptstyle (2)}\vect{r}^{{\scriptscriptstyle (2)}\dagger}}$
	\item In each sector $s$, apply the one-factor calibration procedure (either 2-steps or self-consistent)
	\begin{itemize}
		\item estimate $\vect{\gamma}_{[s]}$ and the weights $\vect{z}_{[s]}$, using intra-sector correlations: $\C_{[s]}$
		\item define the branch index $\hXun_s(\hat{\vect{z}}_{[s]})$ and compute the leaves' residuals \mbox{$\vect{e}^{\scriptscriptstyle (2)}_{[s]}=\vect{r}^{\scriptscriptstyle (2)}_{[s]} - \hat{\vect{\gamma}}_{[s]}\Xun_s, i\in s$}
	\end{itemize}
	\item Use inter-sector correlations to calibrate $\vect{\beta}$:
	\begin{itemize}
		\item Semi-theoretical prediction: $\Q_{ij}=\hat{\gamma}_i\hat{\gamma}_j\beta_{s_i}\beta_{s_j}$ for $s_i\neq s_j$
		\item $\hat{\vect{\beta}}=\argmin_{\vect{\beta}}||\C^{\scriptscriptstyle (2)}-\Q||$
	\end{itemize}	
		\item define the root index $\hXo(\hat{\vect{w}})$ and compute the branches' residuals $\vect{e}^{\scriptscriptstyle (1)}=\vect{r}^{\scriptscriptstyle (1)} - \hat{\vect{\beta}}r^{\scriptscriptstyle (0)}$
	\item compute the branches' correlation 
            $\C^{\scriptscriptstyle (1)}=\vev{\vect{r}^{\scriptscriptstyle (1)}\vect{r}^{{\scriptscriptstyle (1)}\dagger}}$ 
          and check that 
            $\left|\left|\C^{\scriptscriptstyle (1)}-\hat{\vect{\beta}}\hat{\vect{\beta}}^{\dagger}\right|\right|\approx 0$
\end{itemize}
\emph{Pluses}: only quadratic minimization problems; dissociation of the estimation of the exposure parameters and that of the weights; a posterior check of the correctness of intermediate correlations.
\emph{Minuses}: trouble with very small sectors (1 to 3 individuals).

\subsubsection*{Improving over the initial classification}
\paragraph{Allocating each leave to its correct branch}
The correlation of a leave variable $\Xde_i$ with a sector residual $e^{\scriptscriptstyle (1)}_s$ is
\begin{subequations}\label{eq:cor_res}
\begin{equation}
	\esp{\Xde_ie^{\scriptscriptstyle (1)}_s}=	\begin{cases}
				\gamma_i\left(1+\frac{\beta_{s}w_s}{1-\sum\limits_u\beta_uw_u}\right)(1-\beta_s^2)&, i\in s\\
				\gamma_i\frac{\beta_{s_i}w_s}{1-\sum\limits_u\beta_uw_u}&, i\notin s
				\end{cases}
\end{equation}
which means that, up to noise and the small bias introduced by endogenizing the factor variables, 
we should observe
\begin{equation}
	\vev{\Xde_ie^{\scriptscriptstyle (1)}_s}\approx	
	\begin{cases}\gamma_i\left(1-\beta_s^2\right)&, i\in s\\0&, i\notin s	\end{cases}
\end{equation}
\end{subequations}
Significant departure from this prediction indicates a bad initial classification \ldots and suggests the correct one.

\paragraph{Exposure to several factors}

It happens that projecting the leave $i$ onto the residuals of every branch $s$
reveals several non trivial correlations, ruling out the pure 1-factor model and the corresponding prediction \eqref{eq:cor_res}.
In particular, two distinct factors (according to the initial classification) can be simultaneously influential over given variables.
The relation \eqref{eq:x2i} needs to be corrected accordingly:
\[
	\Xde_i = \gamma_{i;a}\Xun_{a}+\gamma_{i;b}\Xun_{b}+e_i,
	\qquad\var{e_i}=1-\gamma_{i;a}^2-\gamma_{i;b}^2-2\gamma_{i;a}\gamma_{i;b}\beta_a\beta_b, \qquad i\in a\cap b
\]
and the corresponding prediction in terms of correlations updated: for $i\in a\cap b, j\in s$,
\[
	\esp{\Xde_i\Xde_j}=\gamma_{i;a}\gamma_j\left[\delta_{as}+(1-\delta_{as})\beta_a\beta_s\right]
	                  +\gamma_{i;b}\gamma_j\left[\delta_{bs}+(1-\delta_{bs})\beta_b\beta_s\right]
\]

%\newpage
\subsection*{Another norm: overlap distance}
Starting from a prior $\Q^0$ obtained e.g.\ by the previous estimation method,
we want to find the optimal corrections $\Q^1$ that minimize the overlap distance of 
the eigenvectors $\Psi$ of $\Q=\Q^0+\Q^1$ with the eigenvectors $\Phi$ of  $\C$.
%
%Let $\Phi$ be the eigenvectors of $\C$ and $\Psi$ those of $\Q$. 
The $M$-overlap distance is 
\[
	d_M(\C,\Q)=-\frac{1}{M}\ln|G_{M|}|
\]
where $G=\Phi^{\dagger}\Psi$ and 
the subscript denotes the restriction to the first $M$ columns and lines.
Minimizing the distance over the parameters $\Q^1$ necessitates, 
for a higher convergence rate, the gradient. 
We compute it explicitly to first order in perturbation theory, 
\begin{align*}
	\ket{\psi_n}&=\ket{\psi_n^0}+\sum_{m\neq n}\ket{\psi_m^0}\frac{\bra{\psi_m^0}\Q^1\ket{\psi_n^0}}{\lambda_n^0-\lambda_m^0}\\
	            &=\ket{\psi_n^0}  +\Psi_{-\!n}^0(\lambda_n^0-\Lambda_{-\!n}^0)^{-1}\Psi_{-\!n}^{0\dagger}\Q^1\ket{\psi_n^0}, \text{ or}\\
	    %\Psi   &=\Big(\mathds{1} +\Psi_{-\!n}^0(\lambda_n^0-\Lambda_{-\!n}^0)^{-1}\Psi_{-\!n}^{0\dagger}\Q^1\Big)\Psi^0\\
	     G      &=\underbrace{\Phi^{\dagger}\Psi^0}_{G^0} +G^1\\
	     d_p(\C,\Q)&=-\frac{1}{M}\ln|G_{M|}^0|-\tr\!\left[G_{M|}^1(G_{M|}^0)^{-1}\right]\\
	     \partial d_p(\C,\Q)&=-\tr\!\left[(\partial G_{M|}^1)(G_{M|}^0)^{-1}\right]
\end{align*}

\chapter[Expansion of the pseudo-elliptical copula]{Pseudo-elliptical copula:\\ expansion around independence} 
\label{chap:apx_algebra}

\begin{table}[t]
        \center
        \begin{tabular}{@{}cccc} 
                $I$&$A$&$R$\\\hline
                16.667&1.176&7.806\\ 
        \end{tabular} 
        \begin{tabular}{@{}cccc} 
                $I^2$&$IA$&$IR$\\\hline 
                111.139&2.948&79.067\\ 
        \end{tabular} 
        \caption{Traces of the operators appearing in the covariance functions (multiples of $10^{-2}$). 
                 Traces of the powers of the rank-one operators $A$ and $R$ equal powers of their traces.
                 The trace of $B+B^{\dagger}$ is zero.} 
        \label{tab:traces} 
\end{table} 

We compute here the spectrum and eigenvectors of the kernel $H(u,v)$ in the case of pseudo-elliptical copula with weak dependences, 
starting from the expansion (\ref{eq:cop_perturb1}) on page~\pageref{eq:cop_perturb1}. 

The situation is better understood in terms of operators acting in the Hilbert space 
of continuous functions on $[0,1]$ vanishing in the border. 
Using Dirac's bra-ket notations, 
$A=\ket{\tilde{A}}\bra{\tilde{A}}$, $B=\ket{\tilde{R}}\bra{\tilde{A}}$, $R=\ket{\tilde{R}}\bra{\tilde{R}}$. 
The sine functions $\ket{j}=\sqrt{2}\,\sin(j\pi u)$ build a basis of this Hilbert space, 
and interestingly they are the eigenvectors of the independent kernel $I(u,v)$ 
($I$ stands for `\emph{I}ndependence' and is the covariance matrix of the Brownian bridge, see page~\pageref{eq:Itheo}: 
$I=M-\Pi$ where $M$ denotes the bivariate upper Fr\'echet-Hoeffding copula and $\Pi$ the bivariate product copula defined in page~\pageref{pg:cop_examples}). 

It is then easy to find the spectra: rank-one operators have at most one non-null eigenvalue. 
Using the parities of $\tilde{A}(u)$ and $\tilde{R}(u)$ with respect to $\frac{1}{2}$ and imposing orthonormality of the eigenvectors, 
we can sketch the following table of the non zero eigenvalues and eigenvectors of the different operators:

\[
\begin{array}{r@{\,=\,}l@{\qquad\qquad}r@{\,=\,}l} 
        \lambda^I_j      &(j\pi)^{-2}                          &  U_j^I(u)     &\ket{j}\\ 
        \lambda^R        &\braket{\tilde{R}}{\tilde{R}}=\tr R  & \ket{U_0^R}   &\ket{\tilde{R}}/\sqrt{\tr R}\\ 
        \lambda^A        &\braket{\tilde{A}}{\tilde{A}}=\tr A  & \ket{U_0^A}   &\ket{\tilde{A}}/\sqrt{\tr A}.
\end{array} 
\] 

For the pseudo-elliptical copula with weak dependence, $H$ has the following general form: 
\begin{equation}
	H = I + \tilde\rho R + \tilde\alpha A - \frac{\tilde\beta}{2} (B+B^{\dagger}). 
\end{equation}
The operator $B+B^{\dagger}$ has two non zero eigenvalues $\pm \sqrt{\lambda^R\lambda^A}$, with eigenvectors $[\ket{U_0^R} \pm \ket{U_0^A}]/\sqrt{2}$. 
In order to approximately diagonalize $H$, it is useful to notice that in the present context $A$ and $R$ are close to commuting with $I$. 
More precisely, it turns out that $\ket{U_0^A}$ is very close to $\ket{2}$, and $\ket{U_0^R}$ even closer to $\ket{1}$. 
Indeed, \mbox{$a_2=\braket{U_0^A}{2}\approx 0.9934$} and \mbox{$r_1=\braket{U_0^R}{1}\approx 0.9998$}. Using the symmetry of $A$ and $R$, we can therefore 
write: 
\begin{align*} 
        \ket{U_0^A}&=a_2\ket{2}+\epsilon_{a}\ket{2_{\perp}}\quad\textrm{with}\quad\braket{2}{2_{\perp}}=\braket{2j\!-\!1}{2_{\perp}}=0,\forall j\geq 1\\ 
        \ket{U_0^R}&=r_1\ket{1}+\epsilon_{r}\ket{1_{\perp}}\quad\textrm{with}\quad\braket{1}{1_{\perp}}=\braket{2j}{1_{\perp}}=0,\forall j\geq1 
\end{align*} 
where $\epsilon_a = \sqrt{1-a_2^2} \ll 1$ and $\epsilon_r = \sqrt{1-r_1^2} \ll 1$. 
The components of $\ket{2_{\perp}}$ on the even eigenvectors of $I$ are determined as the projection: 
\[ 
\braket{2_{\perp}}{2j} = \frac{\braket{U_0^A}{2j}}{\epsilon_a} \qquad j \geq 2, 
\] 
and similarly: 
\[ 
\braket{1_{\perp}}{2j\!-\!1} = \frac{\braket{U_0^R}{2j\!-\!1}}{\epsilon_r} \qquad j \geq 2. 
\] 

Using the definition of the coefficients $\alpha_t$, $\beta_t$ and $\rho_t$ given in section~\ref{sec:example} of chapter~\ref{part:partIII}.\ref{chap:cop_fin} (page~\pageref{eq:cop_perturb1}),
 we introduce the following notations: 
\begin{align*} 
        \tilde\alpha&=2\,\tr A\lim_{N\to\infty}\sum_{t=1}^{N-1}\left(1-\frac{t}{N}\right)\alpha_t\\ 
        \tilde\rho  &=2\,\tr R\lim_{N\to\infty}\sum_{t=1}^{N-1}\left(1-\frac{t}{N}\right)\rho_t\\ 
        \tilde\beta &=2\,\sqrt{\tr A\,\tr R}\lim_{N\to\infty}\sum_{t=1}^{N-1}\left(1-\frac{t}{N}\right)\beta_t 
\end{align*} 
so that $H$ writes: 
\begin{align*} 
        H=\phantom{I_0}I &+\tilde\alpha\,\ket{U_0^A}\bra{U_0^A}+\tilde\rho\,\ket{U_0^R}\bra{U_0^R}-\tilde\beta\,\overleftrightarrow{\ket{U_0^R}\bra{U_0^A}}\\ 
         =H_0&
          +\epsilon_a \left( \tilde\alpha\, a_2\, \overleftrightarrow{\ket{2}\bra{2_{\perp}}}-\tilde\beta\, r_1\, a_{\perp}\overleftrightarrow{\ket{1}\bra{2_{\perp}}}\right)\\ 
         &+\epsilon_r \left( \tilde\rho\, r_1\, \overleftrightarrow{\ket{1_{\perp}}\bra{1}}-\tilde\beta\, a_2\,\overleftrightarrow{\ket{1_{\perp}}\bra{2}}\right)\\ 
         &+\left(\tilde\alpha\, \epsilon_a^2\,\ket{2_{\perp}}\bra{2_{\perp}}+\tilde\rho\, \epsilon_r^2\, \ket{1_{\perp}}\bra{1_{\perp}}-\tilde\beta\, \epsilon_a\, \epsilon_r\,\overleftrightarrow{\ket{2_{\perp}}\bra{1_{\perp}}}\right) ,
\end{align*} 
where $\overleftrightarrow{\ket{\psi_1}\bra{\psi_2}}=\frac{1}{2}\big[\ket{\psi_1}\bra{\psi_2}+\ket{\psi_2}\bra{\psi_1} \big]$
and $H_0$ is the unperturbed operator (0-th order in both $\epsilon$s)
\[
	H_0=\sum_{j\geq 3}\lambda_j^I\,\ket{j}\bra{j}+(\lambda_2^I+\tilde\alpha\, a_2^2)\,\ket{2}\bra{2}+(\lambda_1^I+\tilde\rho\, r_1^2)\,\ket{1}\bra{1}-\tilde\beta\, r_1\,a_2\,\overleftrightarrow{\ket{2}\bra{1}}.
\]
The spectrum of the latter is easy to determine as: 
\[ 
\begin{array}{r@{=}lp{10pt}r@{=}l} 
        \lambda_1^{H_0}&\lambda_-\xrightarrow{\tilde\rho,\tilde\beta\to 0}{\lambda_1^I}&& 
        \ket{U_1^{H_0}}&-\displaystyle\frac{\ket{-}}{\sqrt{\braket{-}{-}}} \xrightarrow{\tilde\rho,\tilde\beta\to 0}{\ket{1}}\\ 
        \lambda_2^{H_0}&\lambda_+\xrightarrow{\tilde\rho,\tilde\beta\to 0}{\lambda_2^I+\tilde\alpha a_2^2}&& 
        \ket{U_2^{H_0}}&\phantom{-}\displaystyle \frac{\ket{+}}{\sqrt{\braket{+}{+}}} \xrightarrow{\tilde\rho,\tilde\beta\to 0}{\ket{2}}\\ 
        \lambda_j^{H_0}&\lambda_j^I&& 
        \ket{U_j^{H_0}}& \ket{j} \qquad (j \geq 3) 
\end{array} 
\] 
where 
\[ 
\lambda_{\pm}=\frac{\lambda_1^I+\tilde\rho\, r_1^2+\lambda_2^I+\tilde\alpha\, a_2^2\,\pm\sqrt{(\lambda_1^I+\tilde\rho\, r_1^2-\lambda_2^I-\tilde\alpha\, a_2^2)^2+4\,(\tilde\beta\, r_1\,a_2)^2}}{2} 
\] 
and $\ket{\pm}$
%\[ 
%\ket{\pm}=\frac{\lambda_1^I+\tilde\rho r_1^2-\lambda_2^I-\tilde\alpha a_2^2\pm\sqrt{(\lambda_1^I+\tilde\rho r_1^2-\lambda_2^I-\tilde\alpha a_2^2)^2+4(\tilde\beta %r_1a_2)^2}}{2}\ket{1}+\tilde\beta r_1a_2\ket{2}
%\] 
the corresponding eigenvectors, which are linear combination of $\ket{1}$ and $\ket{2}$ only. 
Therefore, $\braket{1_{\perp}}{\pm}=\braket{2_{\perp}}{\pm}=0$, which implies that there is no corrections to the eigenvalues of $H_0$ to first order in the $\epsilon$s. 
At the next order, instead, some corrections appear. We call: 
\begin{align*}
	V_{i,j}&=\big[\tilde\rho\,   r_1\,\braket{1}{U_i^{H_0}}-\tfrac12\tilde\beta\, a_2\braket{2}{U_i^{H_0}}\big]\,\braket{j}{1_{\perp}}\,\epsilon_r\\
	       &+\big[\tilde\alpha\, a_2\,\braket{2}{U_i^{H_0}}-\tfrac12\tilde\beta\, r_1\braket{1}{U_i^{H_0}}\big]\,\braket{j}{2_{\perp}}\,\epsilon_a 
\end{align*} 
the matrix elements of the first order perturbation of $H$, whence 
\begin{align*} 
        \lambda_1^{H}&=\lambda_1^{H_0}+\sum_{j\geq 3}\frac{V_{1,j}^2}{\lambda_1^{H_0}-\lambda_j^{H_0}}\\ 
        \lambda_2^{H}&=\lambda_2^{H_0}+\sum_{j\geq 3}\frac{V_{2,j}^2}{\lambda_2^{H_0}-\lambda_j^{H_0}}\\ 
        \lambda_j^{H}&=\lambda_j^{H_0}+\sum_{i=1,2  }\frac{V_{i,j}^2}{\lambda_j^{H_0}-\lambda_i^{H_0}}\\
                     &+\tilde\alpha\, \epsilon_a^2\,\braket{j}{2_{\perp}}^2+\tilde\rho\, \epsilon_r^2\,\braket{j}{1_{\perp}}^2-\tilde\beta\, \epsilon_a\, \epsilon_r\, \braket{j}{1_{\perp}}\braket{j}{2_{\perp}}.
\end{align*} 
As of the eigenvectors, it is enough to go to first order in $\epsilon$s to get a non-trivial perturbative correction: 
\begin{align*} 
        \ket{U_1^H}&=\ket{U_1^{H_0}}+\sum_{j\geq 3}\frac{V_{1,j}}{\lambda_1^{H_0}-\lambda_j^{H_0}}\ket{j}\\ 
        \ket{U_2^H}&=\ket{U_2^{H_0}}+\sum_{j\geq 3}\frac{V_{2,j}}{\lambda_2^{H_0}-\lambda_j^{H_0}}\ket{j}\\ 
        \ket{U_j^H}&=\ket{j}        +\sum_{i=1,2  }\frac{V_{i,j}}{\lambda_j^{H_0}-\lambda_i^{H_0}}\ket{U_i^{H_0}} .
\end{align*}

\paragraph{}
The special case treated numerically in chapter~\ref{part:partIII}.\ref{chap:cop_fin} corresponds to $\tilde\rho =\tilde\beta = 0$, such that the above expressions simplify 
considerably, since in that case $V_{1,j} \equiv 0$ and $V_{2,2j-1}=0$, while $V_{2,2j} = \tilde\alpha\, a_2 \braket{U_0^A}{2j}$.
To first order in the $\epsilon$s, the spectrum is not perturbed and calls $\lambda_i^{H}=\lambda_i^{H_0}=\lambda_i^{I}+\tilde\alpha\, a_2^2\,\delta_{i2}$, so
that the characteristic function of the modified CM distribution is, according to Eq.~(\ref{eq:charctCM}),
\[
	\phi(t)=\prod_{j}\left(1-2\imath t/(j\pi)^2\right)^{-\frac{1}{2}} \times \sqrt{\frac{1-2\imath t\lambda_2^{I}}{1-2\imath t\lambda_2^{H_0}}}.
	       %\equiv \phi_I(t) \times \sqrt{\frac{1-2\imath t\lambda_2^{I}}{1-2\imath t\lambda_2^{H_0}}}.
\]
Its pdf is thus the convolution of the Fourier transform of $\phi_I(t)$ (characteristic function associated to the usual CM distribution $\mathcal{P}_I$ \cite{anderson1952asymptotic})
and the Fourier transform (FT) \nomenclature[aFT]{$\ft$}{Fourier transform} of the correction $\phi_c(t)=\sqrt{1-2\imath t\lambda_2^{I}}/\sqrt{1-2\imath t\lambda_2^{H_0}}$.
Noting that \mbox{$(1-2\imath \sigma^2t)^{-\frac{1}{2}}$} is the characteristic function of the chi-2 distribution, it can be shown that for $k > 0$, and
with $\mu\equiv\lambda_2^{H_0}$ for the sake of readability:%
\begin{align*}
	\frac{1}{\sqrt{2\pi}}\ft[\phi_c]&=\delta(k)-\int_{\lambda_2^{I}}^{\mu}\d{\lambda}\frac{\partial}{\partial k} \left(\chi^2(k;\mu)*\chi^2(k;\lambda)\right)\\
                                    &=\delta(k)-\int_{\lambda_2^{I}}^{\mu}\d{\lambda}\frac{\e^{-\frac{\lambda+\mu}{4\lambda\mu}k}}{8(\lambda\mu)^{\frac{3}{2}}}
                                      \left((\mu-\lambda)\,I_1(\frac{\mu-\lambda}{4\lambda\mu}k)-(\mu+\lambda)\,I_0(\frac{\mu-\lambda}{4\lambda\mu}k)\right)\\
                                    &\approx\delta(k)+\e^{-{k}/{2\lambda}}\,\frac{\tilde\alpha a_2^2}{4\lambda^2}\, I_0(\frac{\tilde\alpha a_2^2}{4\lambda^2}k),
\end{align*}
where $\chi^2(k;\sigma^2)=(2\pi\sigma^2k\,\e^{k/\sigma^2})^{-\frac{1}{2}}$ is the pdf of the chi-2 distribution, 
$I_{n}$ are the modified Bessel functions of the first kind, and $*$ denotes the convolution operation.
The approximation on the last line holds as long as $\tilde\alpha\ll\lambda_2^I=(2\pi)^{-2}$ and 
in this regime we obtain finally the distribution of the CM statistic with asymptotic kernel $H$, $\mathcal{P}_H(k)=\d\pr{C\!M\leq k}/\d{k}$:
\begin{align*}
	\mathcal{P}_H(k)&=\sqrt{2\pi}\,\ft[\phi](k)=(\ft[\phi_I]*\ft[\phi_c])(k)\\
	               %&=\sqrt{2\pi}IFT(\phi)(-k)=(IFT(\phi_I)*IFT(\phi_c))(-k)\\
	                &=\mathcal{P}_I(k)+{4\tilde\alpha a_2^2\pi^4}\int_{0}^{k}\mathcal{P}_I(z)\e^{-2\pi^2{(k-z)}} I_0(4\tilde\alpha a_2^2\pi^4(k-z))\,\d{z}\\
	                &=\mathcal{P}_I(k)+{4\tilde\alpha a_2^2\pi^4}\int_{0}^k\mathcal{P}_I(k-z)\e^{-2\pi^2{z}} I_0(4\tilde\alpha a_2^2\pi^4z)\,\d{z}.
\end{align*}
The second term characterizes the perturbative correction due to the departure of $H_0$ from the independent kernel $I$.

\chapter{Appendices to Chapter~\ref{part:partIII}.\ref{chap:QARCH}}
          %\chapter{The exact spectrum of the \mbox{Borland-Bouchaud} model}
 \section{The exact spectrum of the \mbox{Borland-Bouchaud} model}
 \label{QARCHapx:A}

In this appendix to chapter~\ref{chap:QARCH} of part~\ref{part:partIII}, we compute the spectral properties of a
continuous time kernel with the BB structure of Ref.~\cite{borland2005multi}:
\[
	K(\tau',\tau'')=k(\max(\tau',\tau'')),\qquad\varepsilon\leq\tau',\tau''\leq q.
\]
The eigenvalue equation
\[
	\int_{\varepsilon}^q K(\tau,\tau')v(\tau')\,\d{\tau'}=\lambda v(\tau)
\]
is differentiated to obtain a second order linear differential equation with appropriate boundary conditions:
\begin{equation}\label{QARCHeq:schrodinger}
	\left\{
	\begin{array}{rl}
		\lambda V''(\tau)       &=k'(\tau)V(\tau)\\
		\lambda V_{\lambda}'( q)&= k(q)V_{\lambda}(q)\\
		V(\varepsilon)&=0
	\end{array}
	\right.
\end{equation}
where $V(q)\equiv\int_{\varepsilon}^qv(\tau)\d{\tau}$. 
The resolution of this differential problem depends on the choice of $k(\tau)$,
and we investigate below three particular cases: a linearly, exponentially and power-law decreasing kernel.

For the problem in continuous time, $\varepsilon$ can be taken as 0, but in case this 
differential problem is seen as an approximation to a discrete time problem, 
it is important to keep $\varepsilon=1$. 
For the sake of simplicity of the solutions, we set $\varepsilon=0$ in the following, 
but more precise numerical results for the first eigenmodes are obtained with $\varepsilon=1$ (see Fig.~\ref{QARCHfig:eigenvectsBB}), 
although for higher modes and at large $q$ the choice of $\varepsilon$ is hardly relevant.

\subsection*{Linearly decreasing}
When $k(\tau)=g\cdot(1-\tau/q)\1{\tau\leq q}$, the general solution is a superposition of the two linearly independent solutions: 
\[
	V^s_{\lambda}(\tau)=\sin\!\left(\sqrt{\frac{gq}{\lambda}}\frac{\tau}{q}\right)\quad\text{and}\quad 
	V^c_{\lambda}(\tau)=\cos\!\left(\sqrt{\frac{gq}{\lambda}}\frac{\tau}{q}\right).
\]
The boundary condition $V(0)=0$ disqualifies $V^c$, and $V'(q)=0$ selects only those values of $\lambda$ which satisfy
\[
	\lambda_n=\frac{gq}{\pi^2}\cdot\frac{1}{(n-\tfrac{1}{2})^2},\quad{}n\in\mathds{N}\backslash\{0\},
\]
so that finally 
\[
	v_{n}(\tau)=V'_{\lambda_n}(\tau)\propto \cos\!\left((n-\tfrac{1}{2})\pi\frac{\tau}{q}\right).
\]
In fact, we see straight from Eq.~\eqref{QARCHeq:schrodinger} that when $k'$ is constant, the problem for $V$ amounts to that of a 
free quantum particle in a box with absorbing left wall and reflecting right wall.

\subsection*{Exponentially decreasing}
When $k(\tau)=g\,\e^{-\alpha\tau}$, the general solution is a superposition of the stretched Bessel functions 
\[
	Z_0\!\left(\pm\gamma\,\e^{-\alpha\tau/2}\right),
\]
where $\gamma=2\sqrt{{g}{(\lambda\alpha)}}$.
Bessel functions with negative argument are complex, so we keep only $+\gamma>0$.
The Bessel functions $J_{\nu}$ and $Y_{\nu}$ of the first and second kind respectively, 
are linearly independent, but
the lower boundary condition imposes the coefficients of the combination:
\[
	V(\tau)=Y_0(\gamma)\,J_0\!\left(\gamma\,\e^{-\alpha\tau/2}\right)-
	        J_0(\gamma)\,Y_0\!\left(\gamma\,\e^{-\alpha\tau/2}\right).
\]
Using recursion formulas for Bessel functions \cite{gradshteyn1980table}, the upper boundary condition becomes
\[
	Y_0(\gamma_n)\,J_{2}\!\left(\gamma_n\e^{-\alpha q/2}\right)=
	J_0(\gamma_n)\,Y_{2}\!\left(\gamma_n\e^{-\alpha q/2}\right),
\]
and the eigenvalues $\lambda_n$ are quantized according to the corresponding zeros.

\subsection*{Power-law decreasing}
Taking $k(\tau)=g\,\tau^{-\alpha}$, 
%the differential equation becomes
%\[
%	\tau^2v''(\tau)+2(1-\beta)\tau v'(\tau)+\frac{g}{\lambda}(1-2\beta)\tau^{2\beta}v(\tau)=0
%\]
%and its 
the solutions are given in terms of rescaled Bessel functions $Z_{\nu}$ (formula~8.491.12 in~\cite{gradshteyn1980table}): 
\begin{equation}\label{QARCHeq:sol_bessel}
	V^{\pm}(\tau)=\tau^{\frac{1}{2}}\,Z_{\nu}\!\left(\pm\gamma\tau^{\frac{1-\alpha}{2}}\right),
\end{equation}
with $\gamma=\frac{2}{|1-\alpha|}\sqrt{{g\alpha}/{\lambda}}$ and $\nu={1}/{(1-\alpha)}$.
Bessel functions with negative argument are complex, so we keep only $+\gamma>0$.
In the cases we consider, $\alpha>1$ (see discussion in page~\pageref{QARCHssec:second_moment}) so that \mbox{$\nu<0$}.
Although the Bessel function of the first kind $J_{\nu}(x)$ with negative non-integer $\nu$ diverges at the origin (as $\sim x^{\nu}$),
$V(\tau)$ vanishes linearly in $0^+$ since $V(\tau)\to\sqrt{\tau}(\tau^{\frac{1}{2\nu}})^{\nu}=\tau$.
The first boundary condition is thus satisfied for $Z_{\nu}=J_{\nu}$.
%Differentiating $V^{+}$ we find, $v^+(\tau)=\gamma\frac{1-\alpha}{2}\tau^{-\frac{\alpha}{2}}J_{\nu-1}\left(\gamma\tau^{\frac{1-\alpha}{2}}\right)$.
Using recursion formulas for Bessel functions, the boundary condition becomes
\[
	J_{\nu-2}\!\left(\gamma_nq^{\frac{1-\alpha}{2}}\right)=0,
\]
and the eigenvalues $\lambda_n$ are quantized according to the corresponding zeros.

\begin{figure}[p]
	\center
    \subfigure[Linear decay ($g=0.1$)]{
    \centering
	\includegraphics[scale=0.45,trim=-70 0 70 0,clip]{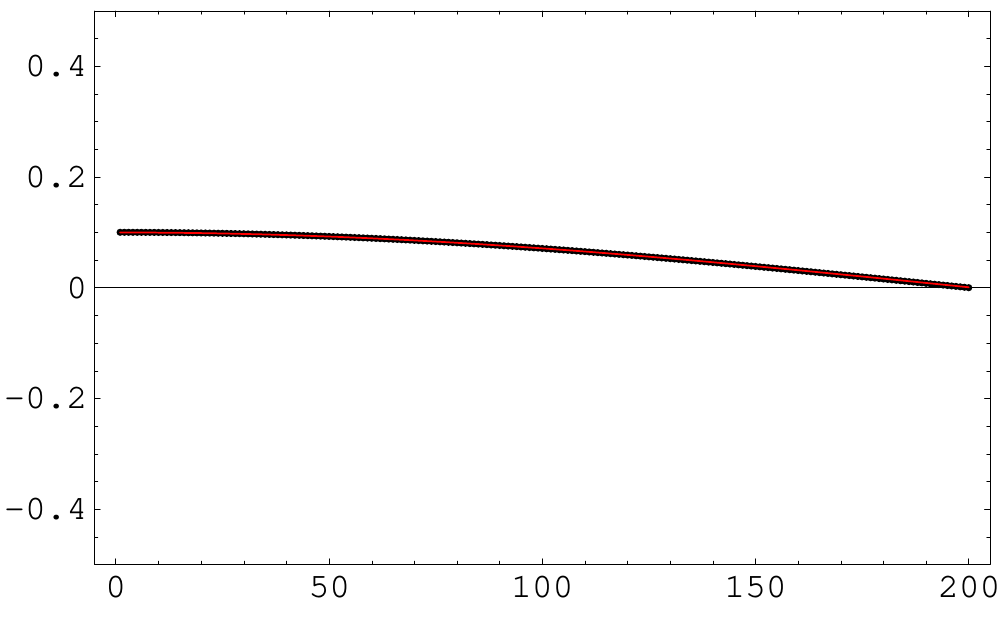}
	\includegraphics[scale=0.45,trim=-70 0 70 0,clip]{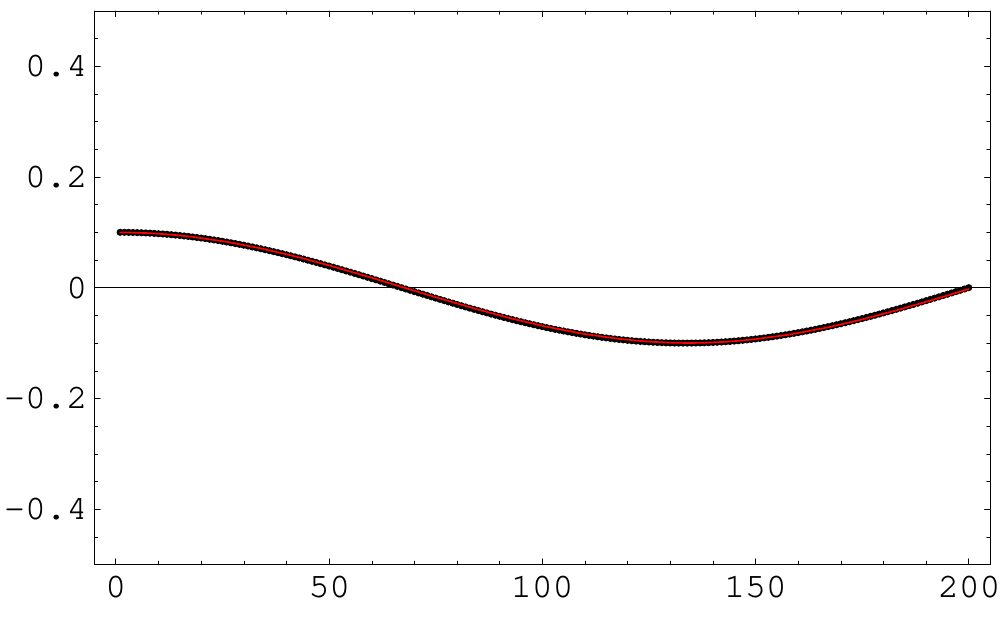}
	\includegraphics[scale=0.45,trim=-70 0 70 0,clip]{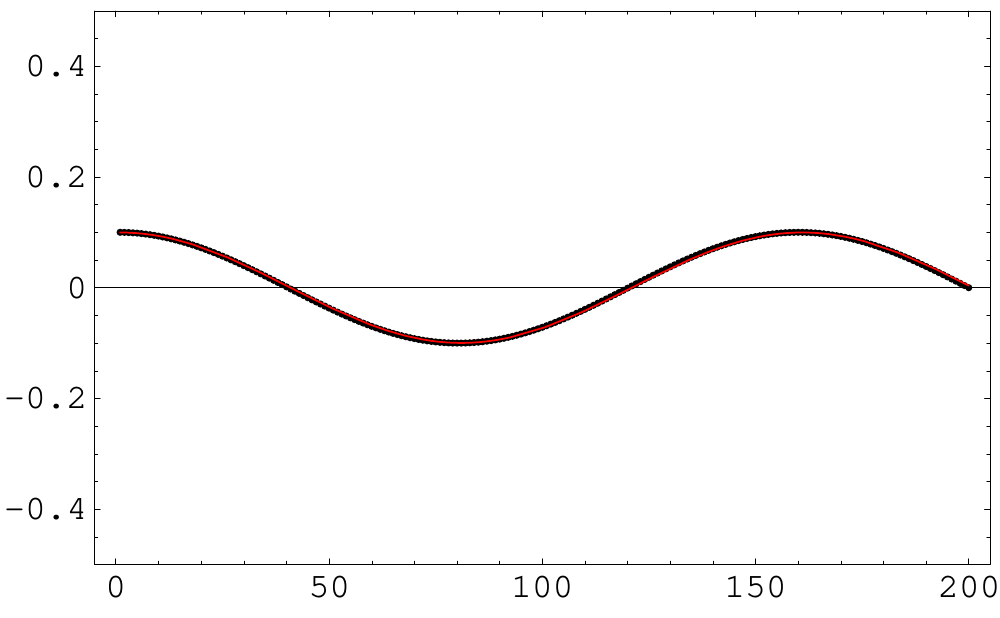}
    }
	
    \subfigure[Exponential decay ($g=0.1, \alpha=0.15$)]{
    \centering
	\includegraphics[scale=0.45,trim=-70 0 70 0,clip]{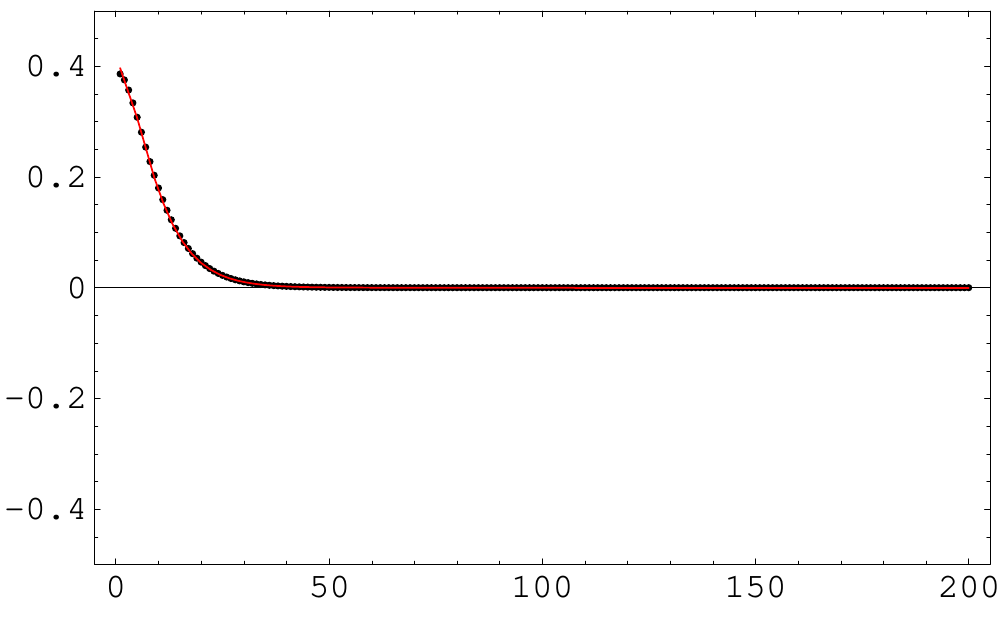}
	\includegraphics[scale=0.45,trim=-70 0 70 0,clip]{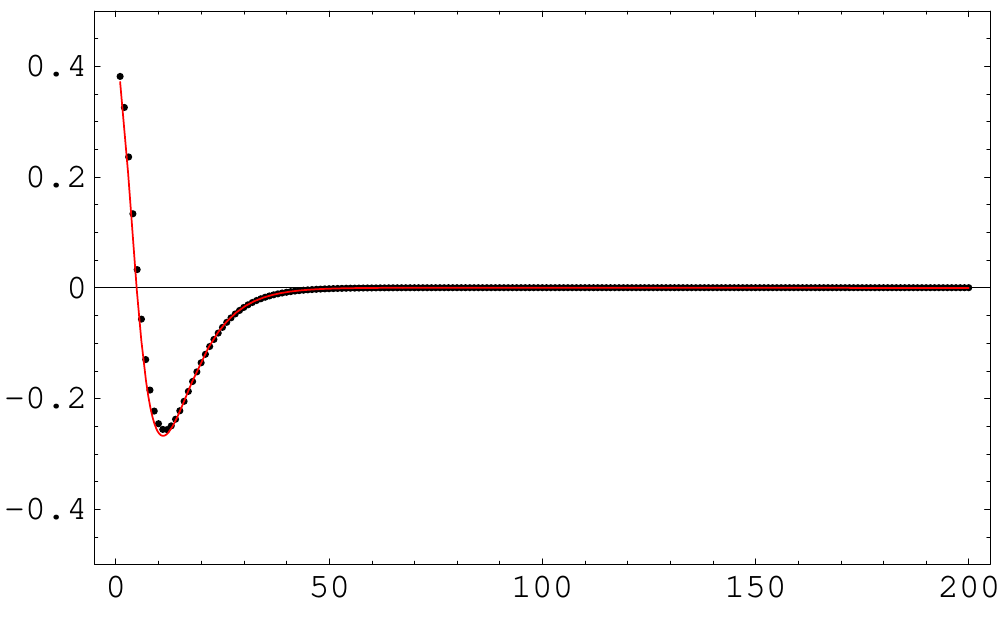}
	\includegraphics[scale=0.45,trim=-70 0 70 0,clip]{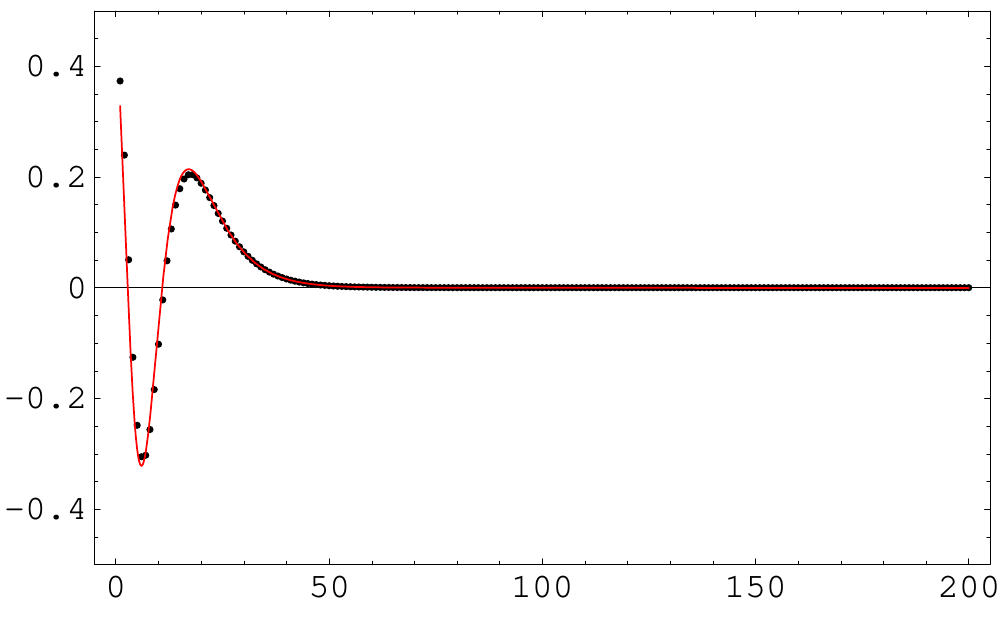}
    }
	
    \subfigure[Power-law decay ($g=0.1, \alpha=1.15$)]{
    \centering
	\includegraphics[scale=0.45,trim=-70 0 70 0,clip]{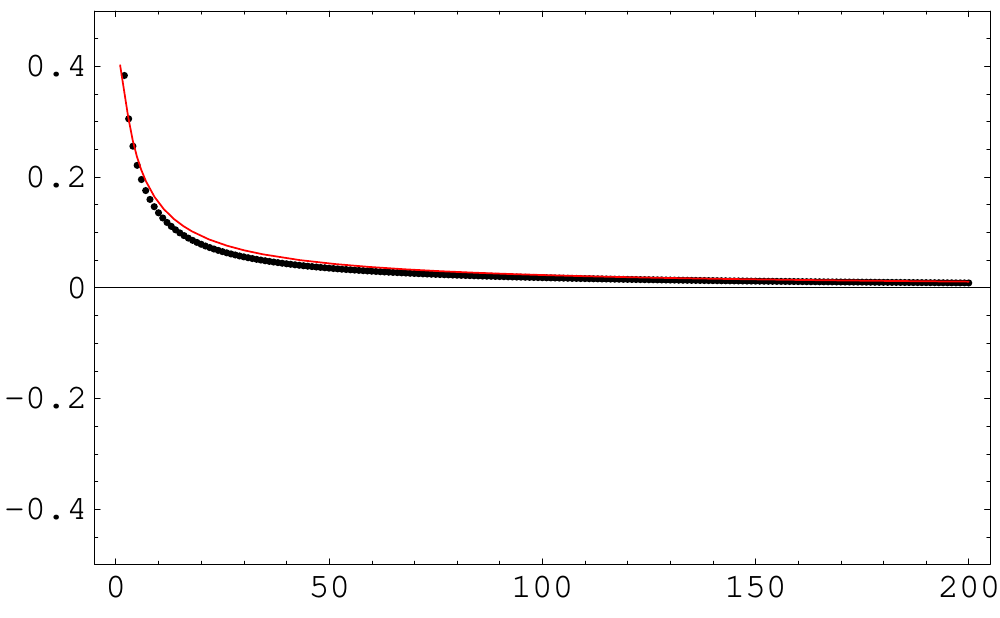}
	\includegraphics[scale=0.45,trim=-70 0 70 0,clip]{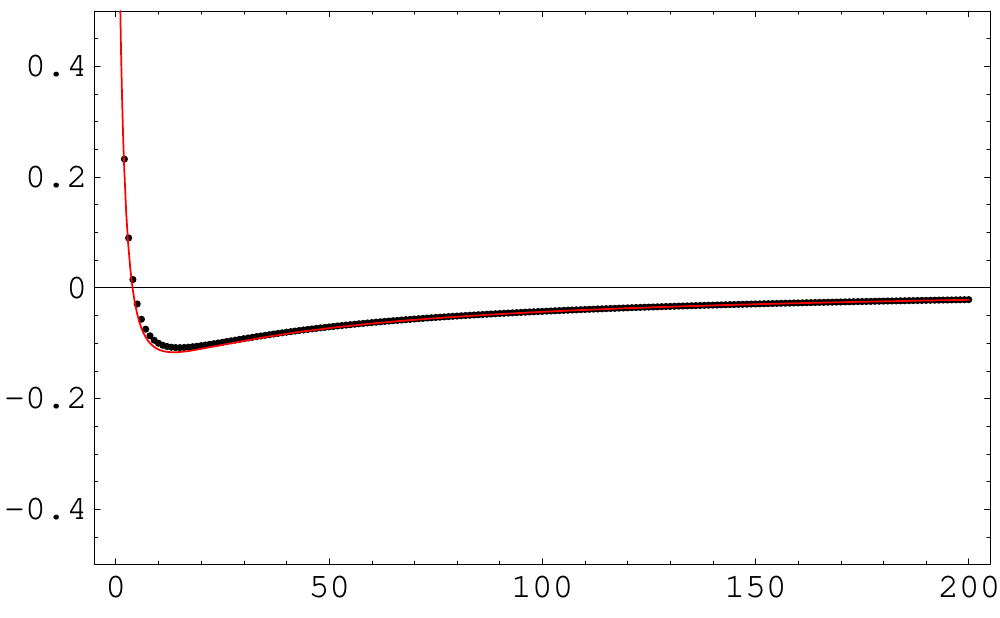}
	\includegraphics[scale=0.45,trim=-70 0 70 0,clip]{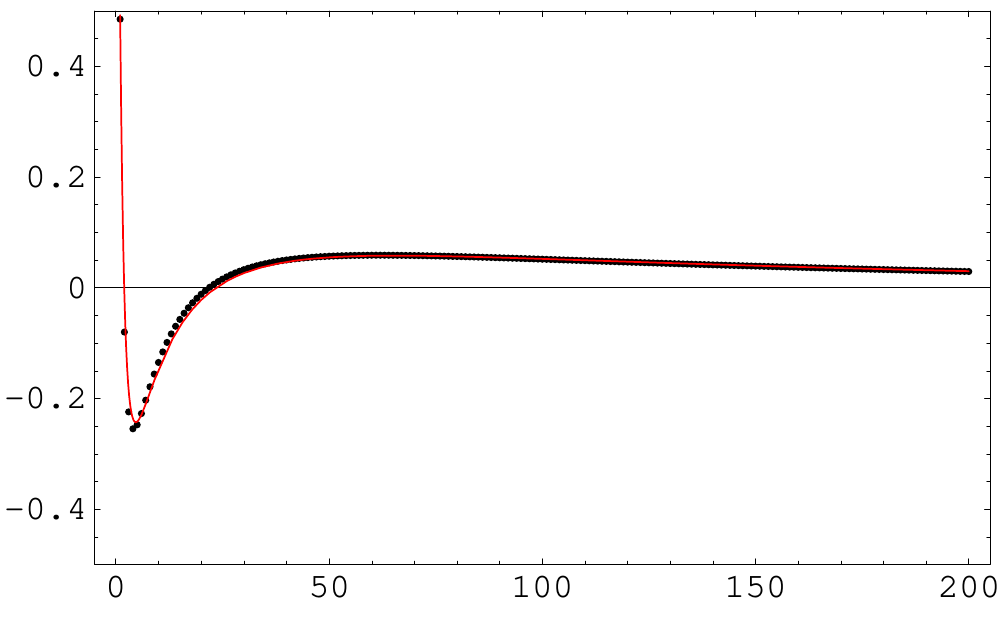}
    }
	
	\caption{First three eigenvectors of the BB Kernel ($q=200$), from left to right.
	The theoretical solutions (red line) are hardly distinguishable from 
    the results of the numerical diagonalization (black dots).}
	\label{QARCHfig:eigenvectsBB}
\end{figure}

\clearpage%\chapter{Power-law volatility correlations in FIGARCH}
 \section{Power-law volatility correlations in FIGARCH}
 \label{QARCHapx:B}
In Chapter~\ref{chap:QARCH} of Part~\ref{part:partIII}, we studied theoretical and empirical properties of Quadratic ARCH models.
In particular, 2nd and 4th order stationarity conditions were summarized, for the power-law decreasing diagonal case,
 in the parameter space of Fig.~\ref{QARCHfig:sig2crit}, page~\pageref{QARCHfig:sig2crit},
 by numerically solving Eq.~\eqref{QARCHeq:vol4}.
 
We want here to get a more precise insight on the behavior of 
the quadratic correlation $\mathcal{C}^{(2)}(\tau)$ when the input kernel $k(\tau)$ is long-ranged asymptotically power-law:
\[
    k(\tau)\xrightarrow{\tau\to\infty}A/\tau^{1+\epsilon}, \quad 0<\epsilon,
\]
where the bound on $\epsilon$ ensures the integrability of the kernel $\int_0^\infty k(\tau)\d{\tau}=1-s^2$.

In order to address this question analytically, we assume that the feedback kernel is infinitely-ranged, 
and consider the continuous-time approximation.
The sum in Eq.~\eqref{QARCHeq:vol4} is approximated by an integral, and decomposed as follows
\begin{align*}
    \mathcal{C}^{(2)}(\tau)=\mathcal{C}^{(2)}(-\tau)&=             \int_0^\infty    k(\tau')\mathcal{C}^{(2)}(\tau-\tau')\,\d{\tau'}\\
                    &= \underbrace{\int_0^\tau      k(\tau')\mathcal{C}^{(2)}(\tau-\tau')\,\d{\tau'}}_{\mathcal{C}^{(2)}_-(\tau)}+
                       \underbrace{\int_\tau^\infty k(\tau')\mathcal{C}^{(2)}(\tau'-\tau)\,\d{\tau'}}_{\mathcal{C}^{(2)}_+(\tau)}
\end{align*}

The behavior at large $\tau\to\infty$ is studied by taking the Laplace transform 
and investigating $\omega\to0$ while keeping a non-diverging product $\omega\tau$:
%\textcolor{red}{
%\begin{align*}
%    \widehat{k}(\omega)&=\int_0^\infty \e^{-\omega\tau}k(\tau)d\tau=\int_0^\infty k(x/\omega)\e^{-x}dx/\omega\\
%                  %&\xrightarrow{\omega\to0}\widehat{k}(0)+\int_0^\infty A\left(\frac{\omega}{x}\right)^{1+\epsilon}\e^{-x}\frac{dx}{\omega}=\widehat{k}(0)+A\,\Gamma(-\epsilon)\omega^{\epsilon}\\
%                   &\xrightarrow{\omega\to0}\begin{cases}
%                                \widehat{k}(0)+\omega\widehat{k}'(\omega) &, 0<\epsilon<1\\
%                                \widehat{k}(0)+\omega\widehat{k}'(0)+ \omega^2\widehat{k}''(\omega)&, 1<\epsilon<2\\
%                                \text{etc.}&
%                                \end{cases}
%\end{align*}
%On consid\`ere dans la suite le cas $\epsilon<1$:
%\begin{align*}
%    \widehat{k}(\omega)\xrightarrow{\omega\to0}&\widehat{k}(0)-\omega\int_0^\infty \e^{-\omega\tau}\tau k(\tau)d\tau \\
%                   =& \widehat{k}(0)-\omega\int_0^\infty A\left(\frac{\omega}{x}\right)^{1+\epsilon}\frac{x}{\omega}\,\e^{-x}\frac{dx}{\omega}
%                    = \widehat{k}(0)-A\,\Gamma(1-\epsilon)\omega^\epsilon
%\end{align*}
%}
%Autre mani\`ere de faire, ajouter et soustraire 1:
\begin{align*}
    \widehat{k}(\omega)&=\int_0^\infty 1-\left(1-\e^{-\omega\tau}\right)k(\tau)\,\d{\tau}
                    =\widehat{k}(0)-\int_0^\infty \left(1-\e^{-x}\right)k(x/\omega)\,\d{x}/\omega\\
                   &=\widehat{k}(0)-\int_0^\infty \left(1-\e^{-x}\right)A\left(\frac{\omega}{x}\right)^{1+\epsilon}\,\d{x}/\omega\\
                   &=\widehat{k}(0)-A\omega^{\epsilon}\int_0^\infty \left(1-\e^{-x}\right)x^{-1-\epsilon}\,\d{x}
                    =\widehat{k}(0)+A\omega^{\epsilon}\int_0^\infty \e^{-x}\,\frac{x^{-\epsilon}}{-\epsilon}\,\d{x} \\
                   &=\widehat{k}(0)+A\omega^{\epsilon}\frac{\Gamma(1-\epsilon)}{-\epsilon}=
                     \widehat{k}(0)+A\,\Gamma(-\epsilon)\omega^{\epsilon}
\end{align*}

\clearpage
\subsubsection*{Power-law resulting correlation}
Empirical observations motivate the following Ansatz for the quadratic correlation:
\[
    \mathcal{C}^{(2)}(\tau)\xrightarrow{\tau\to\infty}B/\tau^\beta,\quad 0<\beta<1.
\]
We hope to be able to reconciliate the ``fast'' decay of $k$ (since integrable)
with a very slow asymptotic decay of the solution, by finding a 
relationship between $\beta$ and $\epsilon$.
In Laplace space, we have
\begin{align*}
    \widehat{\mathcal{C}^{(2)}}(\omega)&=\int_0^\infty \e^{-\omega\tau}\,\mathcal{C}^{(2)}(\tau)\,\d{\tau}
                                        =\int_0^\infty \mathcal{C}^{(2)}(x/\omega)\,\e^{-x}\,\d{x}/\omega\\
                   &\xrightarrow{\omega\to0}\int_0^\infty B\left(\frac{\omega}{x}\right)^{\beta}\e^{-x}\frac{\d{x}}{\omega}
                                        =B\,\Gamma(1-\beta)\,\omega^{\beta-1}\\
%\end{align*}
%\begin{align*}
    \widehat{\mathcal{C}^{(2)}_-}(\omega)&=\widehat{k}(\omega)\, \widehat{\mathcal{C}^{(2)}}(\omega)\\
                     &\xrightarrow{\omega\to0}B\,\Gamma(1-\beta)\,\omega^{\beta-1} \left[(1-s^2)-A\frac{\Gamma(1-\epsilon)}{\epsilon}\omega^{\epsilon}\right]\\
%\end{align*}
%\begin{align*}
        {\mathcal{C}^{(2)}_+}(\tau)  &=\int_0^\infty k(\tau+u)\,\mathcal{C}^{(2)}(u)\,\d{u}\qquad\text{dominated by the large $u\sim\tau$}\\
                     &\xrightarrow{\tau\to\infty}\int_0^\infty \frac{B}{u^\beta}\frac{A}{(\tau+u)^{1+\epsilon}}\,\d{u}
                                                =\tau^{-\beta-(1+\epsilon)+1}\int_0^\infty \frac{B}{x^\beta}\frac{A}{(1+x)^{1+\epsilon}}\d{x}\\
                                               &=\tau^{-(\beta+\epsilon)}AB\int_0^1 \left(\frac{y}{1-y}\right)^\beta y^{1+\epsilon}\frac{\d{y}}{y^2}
                                                =\tau^{-(\beta+\epsilon)}AB\frac{\Gamma(\epsilon+\beta)\,\Gamma(1-\beta)}{\Gamma(1+\epsilon)}\\
    \widehat{\mathcal{C}^{(2)}_+}(\omega)&=\int_0^\infty\e^{-\omega\tau}\,\mathcal{C}^{(2)}_+(\tau)\,\d{\tau}=\int_0^\infty \mathcal{C}^{(2)}_+(x/\omega)\,\e^{-x}\,\d{x}/\omega\\
                     &\xrightarrow{\omega\to0}AB\frac{\Gamma(\epsilon+\beta)\,\Gamma(1-\beta)}{\epsilon\Gamma(\epsilon)}\,\omega^{\epsilon+\beta-1}\,\Gamma(1-\epsilon-\beta)
\end{align*}
Collecting all the terms, we finally get
\begin{align*}
    B\,\Gamma(1-\beta)\,\omega^{\beta-1}&=B\,\Gamma(1-\beta)\,\omega^{\beta-1} \left[(1-s^2)-A\frac{\Gamma(1-\epsilon)}{\epsilon}\omega^{\epsilon}\right]\\
                                    &+AB\,\frac{\Gamma(\epsilon+\beta)\,\Gamma(1-\beta)}{\epsilon\Gamma(\epsilon)}\,\omega^{\epsilon+\beta-1}\,\Gamma(1-\epsilon-\beta)\\
    s^2B\,\Gamma(1-\beta)\,\omega^{\beta-1}&=AB\,\omega^{\epsilon+\beta-1}\,\Gamma(1-\beta) \left[\frac{\Gamma(\epsilon+\beta)\,\Gamma(1-\epsilon-\beta)}{\epsilon\,\Gamma(\epsilon)}-\frac{\Gamma(1-\epsilon)}{\epsilon}\right]
\end{align*}
Very surprisingly, the functional dependence is perfectly equalized in the limit $s^2\to0$: 
the different terms have compatible powers of $\omega$.
For the relationship to be an equality, it is necessary that $s^2\to0$ 
(i.e.\ the model is at the critical limit of quadratic non-stationarity),
and simultaneously the RHS term vanishes:
\begin{align*}
    \Gamma(\epsilon+\beta)\,\Gamma(1-\epsilon-\beta)&=\Gamma(1-\epsilon)\,\Gamma(\epsilon)\\
    \sin\left(\pi(\epsilon+\beta)\right)&=\sin\left(\pi\epsilon\right)\\
    (\epsilon+\beta)-\epsilon=2n&\quad\text{or}\quad(\epsilon+\beta)+\epsilon=1+2n,\qquad n\in\mathds{Z}
\end{align*}
For $\widehat{k}(0)=1-s^2\to 1$ to hold, $\epsilon$ must be close to $0$, 
and in this case there is only one solution 
\[
    \boxed{\beta=1-2\epsilon, \qquad 0<\epsilon<\tfrac{1}{2}}.
\]
% \[
    % \beta=\begin{cases}
            % 1-2\epsilon&,0<\epsilon<\frac{1}{2}\\
            % 3-2\epsilon&,1<\epsilon<\frac{3}{2}\\
            % 5-2\epsilon&,2<\epsilon<\frac{5}{2}\\
            % \text{etc.}&
          % \end{cases}
% \]
%Il y a des ``trous'' dans les valeurs permises de $\epsilon$ !!!

\clearpage
\subsubsection*{Exponential resulting correlation}
We now check whether there might exist a crossover to an exponentially vanishing asymptotic correlation:
\[
    \mathcal{C}^{(2)}(\tau)\xrightarrow{\tau\to\infty}B\,\e^{-\beta\tau},\quad 0<\beta.
\]
Again, the Laplace transform are computed asymptotically: 
\begin{align*}
    \widehat{\mathcal{C}^{(2)}}(\omega)  &=\frac{B}{\omega+\beta}\xrightarrow{\omega\to0}\frac{B}{\beta}\left(1-\frac{\omega}{\beta}\right)\\
    \widehat{\mathcal{C}^{(2)}_-}(\omega)&\xrightarrow{\omega\to0}\frac{B}{\omega+\beta} \left[(1-s^2)+A\,\Gamma(-\epsilon)\omega^{\epsilon}\right]\\
%\end{align*}
%\begin{align*}
        {\mathcal{C}^{(2)}_+}(\tau)  &=\int_0^\infty k(\tau+u)\,\mathcal{C}^{(2)}(u)\,\d{u}\qquad\text{dominated by the large $u\sim\tau$}\\
                     &\xrightarrow{\tau\to\infty}\int_0^\infty B\,\e^{-\beta u}\frac{A}{(\tau+u)^{1+\epsilon}}\,\d{u}
                                                =\tau^{-(1+\epsilon)+1}\int_0^\infty B\,\e^{-\beta\tau u}\frac{A}{(1+x)^{1+\epsilon}}\,\d{x}\\
                                               &=\tau^{-\epsilon}AB\,\e^{\beta\tau}\int_1^\infty y^{-1-\epsilon}\,\e^{-\beta\tau y}\,\d{y}\\
    \widehat{\mathcal{C}^{(2)}_+}(\omega)&=\int_0^\infty\e^{-\omega\tau}\mathcal{C}^{(2)}_+(\tau)\,\d{\tau}=\int_0^\infty \mathcal{C}^{(2)}_+(x/\omega)\e^{-x}\,\d{x}/\omega\\
                     &\xrightarrow{\omega\to0}\omega^{-1}\int_0^\infty \d{x}\, \e^{-x} \left(\frac{x}{\omega}\right)^{-\epsilon}AB\int_1^\infty \d{y}\,y^{-1-\epsilon}\e^{-\beta\frac{x}{\omega}y}\\
                     &                       =\omega^{\epsilon-1}AB\int_1^\infty \d{y}\,y^{-1-\epsilon}\int_0^\infty \d{x}\, x^{-\epsilon}\e^{-x(1+\frac{\beta}{\omega}y)}\\
                     &                       =\omega^{\epsilon-1}AB\int_1^\infty \d{y}\,y^{-1-\epsilon} \left(1+\beta\frac{y}{\omega}\right)^{-(1-\epsilon)}\Gamma(1-\epsilon),\quad 1-\epsilon>0\\
                     &                       =AB\,\Gamma(1-\epsilon)\int_1^\infty \d{y}\,y^{-1-\epsilon} (\omega+\beta y)^{\epsilon-1},\quad 1-\epsilon>0\\
                     &                       =AB\,\Gamma(1-\epsilon)\,\beta^\epsilon\int_\beta^\infty \d{z}\,\frac{(\omega+z)^{\epsilon-1}}{z^{\epsilon+1}}\\
                     &                       =AB\,\Gamma(1-\epsilon)\,\frac{(\beta+\omega)^\epsilon-\beta^\epsilon}{\omega\epsilon}
\end{align*}
Bringing everything together,
\[
    \frac{\omega^\epsilon}{\omega+\beta}=\frac{(\omega+\beta)^\epsilon-\beta^\epsilon}{\omega},
\]
which only works if $\beta=0$.

\nocite{dempster2002risk}
\nocite{tilak2012study,klein2005capital}
\bibliographystyle{Thesis_Template/ThesisStyle}
\bibliography{../paper/biblio_all}

%\printindex        % remy test
%\printglossary     % remy test

\end{document}